\documentclass[a4paper,12pt]{amsart}

\usepackage{latexsym}
\usepackage{amsmath}
\usepackage{amsthm}
\usepackage{amssymb}
\usepackage{enumerate}
\usepackage{multicol}
\usepackage{graphicx}
\usepackage{scalefnt}
\usepackage{fancyhdr}
\usepackage{longtable}

\allowdisplaybreaks[4]

\renewcommand{\headrulewidth}{0.0pt}

\theoremstyle{plain}
\newtheorem{theorem}{Theorem}[section]
\newtheorem{proposition}[theorem]{Proposition}
\newtheorem{lemma}[theorem]{Lemma}
\newtheorem{corollary}[theorem]{Corollary}

\theoremstyle{definition}

\newtheorem{remark}[theorem]{Remark}
\numberwithin{equation}{section}

\def\cB{{\mathcal{B}}}

\def\C{{\mathbb{C}}}

\def\cD{{\mathcal{D}}}

\def\be{{\mathbf{e}}}
\def\cE{{\mathcal{E}}}

\def\cH{{\mathcal{H}}}
\def\sH{{\mathsf{H}}}

\def\cI{{\mathcal{I}}}

\def\bk{{\mathbf{k}}}

\def\cM{{\mathcal{M}}}

\def\N{{\mathbb{N}}}

\def\cP{{\mathcal{P}}}
\def\bp{{\mathbf{p}}}

\def\bq{{\mathbf{q}}}

\def\R{{\mathbb{R}}}
\def\cR{{\mathcal{R}}}
\def\fr{{\mathrm{r}}}

\def\S{{\mathbb{S}}}
\def\cS{{\mathcal{S}}}
\def\bs{{\mathbf{s}}}

\def\T{{\mathbb{T}}}
\def\bt{{\mathbf{t}}}

\def\bu{{\mathbf{u}}}
\def\bU{{\mathbf{U}}}

\def\cV{{\mathcal{V}}}
\def\sV{{\mathsf{V}}}
\def\bv{{\mathbf{v}}}

\def\fw{{\mathrm{w}}}
\def\bW{{\mathbf{W}}}

\def\bx{{\mathbf{x}}}
\def\bX{{\mathbf{X}}}

\def\by{{\mathbf{y}}}
\def\bY{{\mathbf{Y}}}

\def\Z{{\mathbb{Z}}}
\def\bz{{\mathbf{z}}}
\def\bZ{{\mathbf{Z}}}

\def\opsi{{\overline{\psi}}}

\def\spin{{\{\uparrow,\downarrow\}}}
\def\ua{{\uparrow}}
\def\da{{\downarrow}}

\def\O{{\Omega}}
\def\o{{\omega}}

\def\eps{{\varepsilon}}
\def\g{{\gamma}}
\def\G{{\Gamma}}
\def\s{{\sigma}}

\def\D{{\Delta}}

\def\<{{\langle}}
\def\>{{\rangle}}

\def\esssup{\mathop{\mathrm{esssup}}}

\def\Tr{\mathop{\mathrm{Tr}}}

\def\Mat{\mathop{\mathrm{Mat}}}
\def\dis{\mathop{\mathrm{dis}}\nolimits}

\def\sgn{\mathop{\mathrm{sgn}}\nolimits}

\def\supp{\mathop{\mathrm{supp}}\nolimits}

\def\Im{\mathop{\mathrm{Im}}}
\def\Re{\mathop{\mathrm{Re}}}
\def\Ref{\mathop{\mathrm{Ref}}\nolimits}

\def\b0{{\mathbf{0}}}

\def\frah{{\left(\frac{1}{h}\right)}}

\begin{document}
\scalefont{1.1}

\thispagestyle{fancy}
\lhead{}
\chead{}
\rhead{}
\cfoot{\thepage}
\renewcommand{\headrulewidth}{0.0pt}

\begin{center}\Large\bf
Renormalization Group Analysis\\
of Multi-Band Many-Electron Systems\\
 at Half-Filling
\end{center}
\bigskip 

\begin{center}\large
Yohei Kashima \medskip \\
Graduate School of Mathematical Sciences, University of Tokyo,\\
Komaba, Tokyo, 153-8914, Japan\\ 
kashima@ms.u-tokyo.ac.jp
\end{center}
\bigskip

\begin{quotation}
\small {\bf  Abstract. }
Renormalization group analysis for multi-band many-electron systems at
 half-filling at positive temperature is presented. The analysis
 includes the Matsubara ultra-violet integration and the infrared
 integration around the zero set of the dispersion relation.  The
 multi-scale integration schemes are implemented in a finite-dimensional
 Grassmann algebra indexed by discrete position-time variables. In
 order that the multi-scale integrations are justified inductively, various
 scale-dependent estimates on Grassmann polynomials are established. 
We apply these theories in practice to prove that for the half-filled
 Hubbard model with nearest-neighbor hopping on a square
 lattice the infinite-volume, zero-temperature limit of the free energy
 density exists as an analytic function of the coupling constant in a
 neighborhood of the origin if the system contains the
 magnetic flux $\pi$ (mod $2\pi$) per plaquette and $0$ (mod $2\pi$)
 through the large circles around the periodic lattice. 
Combined with Lieb's result on the flux phase problem
 ([Lieb, E. H., Phys. Rev. Lett. {\bf 73} (1994), 2158]), this theorem
 implies that the minimum free energy density of the flux
 phase problem converges to an analytic function of the coupling
 constant in the infinite-volume, zero-temperature limit. The proof of
 the theorem is based on a four-band formulation of the model
 Hamiltonian and an extension of Giuliani-Mastropietro's
 renormalization designed for the half-filled Hubbard model on
 the honeycomb lattice ([Giuliani, A. and V. Mastropietro,
 Commun. Math. Phys. {\bf 293} (2010), 301--346]).
\end{quotation}
\medskip

\begin{quotation}\small 
\end{quotation}

\tableofcontents

\section{Introduction}\label{sec_introduction}

\subsection{Introduction}
It is becoming clear that many-electron lattice systems at positive
temperature can be constructed rigorously within the framework of the
finite-dimensional Grassmann integrals and various physical quantities
defined in the system can be analyzed by solid calculus on the
finite-dimensional Grassmann algebra. One analytical technique at the
core of this research field is the multi-scale integration. Since
its iterative operation with decomposed covariances formally obeys a
semi-group property, the multi-scale integration is also called the
renormalization group (RG) method.  For instance, the existence of
infinite-volume limit of thermodynamic physical quantities in
many-electron systems and their analyticity with respect to the coupling constant
can be proved by carrying out a multi-scale
integration over the Matsubara frequency. This type of multi-scale
integration is called the Matsubara ultra-violet (UV)
integration. Nowadays, however, it is known that a wide class of
many-electron systems can be controlled independently of the volume
factor by a simple single-scale analysis thanks to the development of
volume-independent determinant bounds on the covariances by Pedra and
Salmhofer (\cite{PS}). Though the Matsubara UV integration or the
single-scale integration based on Pedra-Salmhofer's determinant bound
proves the analyticity of physical quantities in the infinite-volume
limit with the coupling constant,
these methods do not improve the temperature-dependency of the domain in
which such analytic statements can be made. Without any further
treatment, the allowed magnitude of the coupling typically shrinks in a
power order of temperature. As a consequence, the theory gives little
insight into physics caused by interacting electrons in low
temperatures. A multi-scale integration designed to ease the
temperature-dependency of the maximal magnitude of interaction is called
the infrared (IR) integration. Proper implementation of the IR
integration is believed to guarantee the analyticity of physical
quantities down to exponentially small temperatures or even to the
absolute zero-temperature. Since there are demands for rigorous tools
which enable us to treat many-electron models in wide parameter regions, 
the RG methods need to be
systematically investigated from various view points as a hopeful
candidate for such anticipated mathematical methods.

This paper has two purposes. One is to construct necessary estimates for
the multi-scale integrations on a finite-dimensional Grassmann algebra
to ensure the convergence of infinite-volume, zero-temperature limit of
thermodynamic physical quantities in half-filled multi-band
many-electron systems. The other is to apply these general estimates in
practice to a specific many-electron model and reach rigorous conclusions
in low temperatures. More precise explanation of the second purpose is
the following. We prove that
for the half-filled Hubbard model on a square lattice
there exists an analytic function of the complex coupling constants
on a multi-disk around the origin such that the free energy density
is equal to the restriction of the analytic function on the real axis
and the analytic function uniformly converges in the infinite-volume,
zero-temperature limit, if the nearest-neighbor hopping parameter of the
Hubbard model contains the magnetic flux $\pi$ (mod $2\pi$) per
plaquette and 0 (mod $2\pi$) through the large circles around the
periodic lattice. The Hubbard model with
this constraint on the magnetic flux is rarely seen in the study of
mathematical RG so far. However, it is not irrelevant in mathematical physics. In fact this model defines
the minimum free energy in the flux phase problem, which seeks a
configuration of the arguments of the complex-valued hopping parameter 
in the half-filled Hubbard model in order that the free
energy of the system is minimum. Lieb (\cite{L}) essentially gave a sufficient condition for
the arguments to attain the minimum, which is the above condition on
the magnetic flux.  The sufficiency of this condition was emphasized
 by Macris and Nachtergaele in \cite{MN}. Since our model is the
minimizer, the same analytic and convergent properties hold for  
the minimum free energy density in the flux phase problem. These results
are officially stated in Subsection \ref{subsec_model}.

Let us explain the motive for this work by reviewing recent developments
in the multi-scale analysis concerning the 2-dimensional
Hubbard models at positive temperature, especially by focusing on the
temperature-dependency of the possible magnitude of the coupling constant.
In the series \cite{Ri}, \cite{AMRi1},
\cite{AMRi2} the half-filled Hubbard model on a square lattice was
studied. These multi-scale analysis suggest that the correlation functions
in the system are analytic with respect to the coupling constant $U$ in
the domain $|U|<c |\log T|^{-2}$, where $T$ is the temperature 
and $c$ is a generic positive constant. In the doctoral thesis \cite{P} Pedra
characterized the 2-point correlation function in the Hubbard model away
from half-filling on a square lattice under the constraint $|U|\le
c|\log T|^{-1}$ and concluded that the system in this domain of the
coupling is a Fermi liquid.  The RG analysis by Benfatto, Giuliani and
Mastropietro \cite{BGM} also showed that the behavior of the 2-point
correlation function in the Hubbard model away from half-filling on a
square lattice corresponds to a Fermi liquid if the couping constant $U$
obeys the condition $|U|\le c |\log T|^{-1}$. One remarkable achievement was made by Giuliani and
Mastropietro in \cite{GM}.  They developed an infrared integration
technique for the half-filled Hubbard model on the honeycomb lattice and
concluded that the free energy density and the correlation functions in the
infinite-volume limit are analytic in the temperature-independent
domain $|U|<c$. Giuliani, Mastropietro and Porta continued their RG
analysis for the same model in the following article \cite{GMP}. Despite
the conceptual importance of the 2-d Hubbard models in
condensed matter physics, complete implementation of RG methods leading
to rigorous conclusions on the model in low temperatures is still
scarce. It is necessary to clarify the applicability of rigorous
versions of RG to the 2-d Hubbard models. This paper is aimed
at achieving this goal by presenting another example of analytic control
of the 2-d Hubbard model down to the absolute
zero-temperature together with a general framework constructed in a
self-contained style.

Let us look into more details of related research articles to understand
new aspects of this paper from a technical view point. It was shown in
\cite{K1} that the partition function in many-electron systems can be
formulated into a time-continuum limit of
the finite-dimensional Grassmann Gaussian integral, whose derivation is
based on a discretization of the Riemann integral with respect to the time variable
inside the perturbative expansion of the partition function. In this
formulation the basis of Grassmann algebra is indexed by the discrete
space-time points. The following papers \cite{K2}, \cite{K3} adopted the
same formulation and proved exponential decay properties of the 
finite-temperature correlation functions in the Hubbard models by a
single-scale analysis based on Pedra-Salmhofer's determinant bound and a
multi-scale integration over the Matsubara frequency respectively.
 The Matsubara UV integration in \cite{K3} was inductively constructed
 as a transform on the finite-dimensional Grassmann algebra. Since no
 infrared integration is performed in \cite{K1}, \cite{K2}, \cite{K3}, the results
 in these papers are restricted within a domain of the coupling
 constant depending on temperature as significantly as $|U|<c T^n$ with
 some $n\in \N$. So the next step is to analyze a many-electron system
 in low temperatures by means of an IR integration in the same
 finite-dimensional Grassmann algebra as in \cite{K1},
 \cite{K2}, \cite{K3} and to justify the IR integration process by the
 mathematical induction with the scale index in the same manner as in
 the Matsubara UV integration of \cite{K3}.

A key idea of the IR analysis in \cite{P}, \cite{BGM}, \cite{GM} is the
modification of the covariance at each integration step by the insertion
of the
kernel of the quadratic term produced by the previous
integration. Because of the symmetries of Grassmann polynomials preserved
during the multi-scale integration process, this modification does not
qualitatively change the shape of the zero set of the denominator of the
covariance, and thus the IR integration approaching the zero points of the
denominator is guaranteed to continue. The IR integration in this paper
uses this adaptive modification method introduced in \cite{P},
\cite{BGM}, \cite{GM}. Though this renormalization
technique explicitly plays a role only when we solve the model problem in
Section \ref{sec_IR_model}, keeping it in mind, we prepare general
estimates in Subsection \ref{subsec_IR_general} and Subsection
\ref{subsec_IR_general_difference} by giving Grassmann polynomials whose
degrees are at least 4 as the input to the integrations.

The main reasons why the physical quantities in the half-filled
Hubbard model on the honeycomb lattice are proved to be analytic
independently of temperature in \cite{GM} are the following. 
The zero set of the free dispersion relation degenerates into 2
distinct points, which remain to be the zero points of the
denominator of the effective covariance, 
and consequently the integral of each effective 
interaction term of order $\ge 4$ is bounded from above by a negative power of
the support size of IR cut-off at the scale, in other words,
effective interaction terms of order $\ge 4$ are irrelevant under 
the iterative IR integrations.  In fact the
invariance of the 2 Fermi-points is a remarkable discovery made by
Giuliani and Mastropietro in \cite{GM}. In this paper we formulate the
half-filled Hubbard model with the flux $\pi$ condition on a square
lattice into a 4-band many-electron model, in which the zero set of the
free dispersion relation consists of a single point. Then, we prove that 
 the zero point of the free dispersion relation essentially continues to be a zero point of
the denominator of the effective covariance during RG process by extending 
Giuliani-Mastropietro's renormalization method originally developed for
the 2-band half-filled Hubbard model on the honeycomb lattice. 
More precisely speaking, our effective covariance in momentum space 
is the inverse of a $4\times 4$ matrix. What we prove is that each
element of the effective $4\times 4$
matrix becomes negligibly small when either the momentum variable is
close to the zero point of the free dispersion relation or the Matsubara
frequency is close to zero and thus the point where the effective matrix
is not invertible is the same as in the free case. 
The non-corresponding property of the free covariance at equal
space-time points erases the quadratic term of the interaction in the
Grassmann integral formulation adopted in \cite{GM}, while the quadratic
term remains if we formulate the model by using the Grassmann Gaussian
integral proposed in \cite{K1}, \cite{K2}, \cite{K3}. The quadratic term in the
interaction breaks one of the invariances called `inversion' in 
\cite[\mbox{Lemma 1}]{GM}, which was used especially to prove that
the diagonal elements of the effective $2\times 2$ matrix vanish as
the Matsubara frequency approaches zero in \cite{GM}. 
Because of a lack of necessary invariances, the argument of \cite{GM} to confirm that certain elements including diagonal ones
of the effective matrix vanish in the IR limit does not immediately fit in our formulation. In this paper, therefore, we start from reforming the
formulation built in the same manner as in \cite{K1}, \cite{K2}, \cite{K3}
into a more convenient form having desirable symmetries for the IR
integration.

The proof of validity of RG in this paper is based on the mathematical
induction with respect to the integration scale, which assumes a
scale-dependent norm bound on the input as the induction hypothesis 
 and shows the succeeding norm
bound on the output of the single-scale integration, while the
Gallavotti-Nicol\`o tree spreading over all the scales is the main tool
to organize the multi-scale integration process in \cite{GM} as well as
in \cite{BGM}. For this reason the major part of this paper is devoted
to establish norm bounds on Grassmann polynomials produced by the
single-scale integrations, especially by the tree expansions for
derivatives of logarithm of the Grassmann Gaussian integral. Norm
estimations on finite-dimensional Grassmann algebra were rigorously
summarized with the aim of validating RG by induction by Feldman,
Kn\"orrer and Trubowitz in the book \cite{FKT2}, in
which, however, a representation theorem for the Schwinger functional
developed in \cite{FKT1}, rather than the tree formula,  underlay the Fermionic expansion. The concepts
of \cite{FKT2} were extended into the RG analysis on
infinite-dimensional Grassmann algebra for interacting Fermions in
\cite{FKT3}, \cite{FKT4}. This paper intends to keep the finite-dimensionality of
Grassmann algebra and shows the existence of infinite-volume,
zero-temperature limit as a result of calculus on the finite-dimensional
vector space. In summary what this paper newly presents apart from the
statements of the main theorem and its corollary in Subsection
\ref{subsec_model} are
\begin{enumerate}[(i)]
\item Inductive construction of the multi-scale integrations, which lead to
      the zero-temperature limit of the free energy density, on the finite-dimensional Grassmann
      algebra indexed by discrete space-time points.
\item An extension of Giuliani-Mastropietro's renormalization
      to a 4-band many-electron system.
\end{enumerate}

Before closing the introductory remarks we should also argue possible
limitations of our framework for IR analysis. Our IR multi-scale
integration procedure is based on a general proposition, namely
Proposition \ref{prop_infrared_integration_general}, which concerns
scale-dependent bound properties of the output of a single-scale
integration generalizing a real IR integration step. The validity of the
proposition is due to the structure that with respect to the scale-dependent
norm and semi-norm set in the proposition, any Grassmann monomial of
order $\ge 4$ with the bound of scale $l+1$ automatically admits the
bound of scale $l$. In more details the norm bound on a monomial of
order $\ge 4$ at scale $l+1$ amounts to requiring the integral of the
monomial to be bounded from above by a negative power of the factor
$M^{l+1}$, where $M(>1)$ is a parameter to control the support size of
IR cut-off. Since the negative power of $M^{l+1}$ is smaller than that
of $M^l$, the monomial satisfies the norm bound of scale $l$ as
well. When we solve the model problem in Section \ref{sec_IR_model}, for
example, the power of the factor $M^l$ for a monomial of order $m$ is
$-m+7/2$, which is negative for $m\ge 4$. As long as we go
through Proposition \ref{prop_infrared_integration_general}, therefore,
our constructive theory is such that effective interaction terms of
order $\ge 4$ are irrelevant at every step of IR integrations. 
In this paper we do not have a rigorous a priori criterion of to which
model the proposition does or does not apply. The proposition is built
upon an assumption, namely \eqref{eq_parameters_assumption_IR_general},
determined by exponents in the determinant and $L^1$ bounds on
an effective covariance. A heuristic argument in Remark
\ref{rem_renormalizability} suggests that the assumption of the
proposition is unlikely to be realized in a $d$-dimensional
many-electron model where
the $d-1$-dimensional Hausdorff measure of the zero set of the free
dispersion relation is non-zero such as in the $1$-dimensional
Hubbard models with free Fermi points or the $2$-dimensional Hubbard
models whose free Fermi curve does not degenerate into finite points. For
this reason we expect that these usual many-electron models cannot be
analyzed at zero-temperature by an immediate application of our framework.

The contents of this paper after this section are outlined as follows. In
Section \ref{sec_formulation} we introduce the Hamiltonian operator in a
generalized setting and formulate the free energy density as a
time-continuum limit of
logarithm of the finite-dimensional Grassmann Gaussian integral. In Section
\ref{sec_general_estimation} we present norm estimates on single-scale
integrations without assuming quantitative upper bounds on the
covariances. In Section \ref{sec_general_estimation_temperature} we
establish norm estimates on the difference between single-scale
integrations at 2 different temperatures without assuming quantitative
upper bounds on the covariances. In Section
\ref{sec_multi_scale_analysis} we apply the general norm estimates
developed in Section \ref{sec_general_estimation} and Section
\ref{sec_general_estimation_temperature} to construct both the UV integration
process and the IR integration process as well as to measure the difference
between Grassmann polynomials produced by these integrations at 2
different temperatures in a model-independent general setting. In
Section \ref{sec_UV} we complete the UV integration over the Matsubara
frequency by showing that the covariance with UV cut-off actually satisfies the bound
properties assumed in Section \ref{sec_multi_scale_analysis}. In
Section \ref{sec_IR_model} we apply the estimations prepared in Section
\ref{sec_multi_scale_analysis} for the IR integration to the model
Hamiltonian and prove the main theorem of this paper. In Appendix
\ref{app_flux_phase} we restate Lieb's result on the flux phase problem
 with some supplementary arguments
concerning the repeated reflection. In Appendix
\ref{app_grassmann_L1_theory} we establish $L^1$-norm bounds on kernels
of Grassmann polynomials, which are necessary for the proof of the convergence
of the symmetric Grassmann integral formulation to the free energy density in
Section \ref{sec_formulation}. 
In Appendix \ref{app_gevrey} we summarize
basic estimates on functions of Gevrey-class, to which our cut-off functions
belong. In Appendix \ref{app_h_L_limit} we prove that the
time-continuum, infinite-volume limit of derivatives of logarithm of the
Grassmann Gaussian integral exists at the origin. These convergence
properties are used in the proof of the
existence of infinite-volume limit of the free energy density in Section \ref{sec_IR_model}. Finally in Appendix \ref{app_partition_function} some
lemmas concerning the free energy density are directly proved
without going through the Grassmann integral formulation. The flow chart
of our construction is shown in Figure \ref{fig_flow_chart}.

\begin{figure}
\begin{center}
\begin{picture}(320,260)(0,0)

\put(40,240){\line(1,0){120}}
\put(40,260){\line(1,0){120}}
\put(40,240){\line(0,1){20}}
\put(160,240){\line(0,1){20}}

\put(70,245){Section \ref{sec_introduction}}
\put(40,250){\line(-1,0){40}}
\put(0,250){\line(0,-1){245}}
\put(0,5){\vector(1,0){40}}

\put(40,200){\line(1,0){120}}
\put(40,220){\line(1,0){120}}
\put(40,200){\line(0,1){20}}
\put(160,200){\line(0,1){20}}

\put(70,205){Section \ref{sec_formulation}}

\put(100,200){\vector(0,-1){20}}
\put(40,210){\line(-1,0){20}}
\put(20,210){\line(0,-1){195}}
\put(20,15){\vector(1,0){20}}

\put(40,160){\line(1,0){120}}
\put(40,180){\line(1,0){120}}
\put(40,160){\line(0,1){20}}
\put(160,160){\line(0,1){20}}

\put(70,165){Section \ref{sec_general_estimation}}

\put(100,160){\vector(0,-1){20}}

\put(40,120){\line(1,0){120}}
\put(40,140){\line(1,0){120}}
\put(40,120){\line(0,1){20}}
\put(160,120){\line(0,1){20}}

\put(70,125){Section \ref{sec_general_estimation_temperature}}

\put(100,120){\vector(0,-1){20}}

\put(40,80){\line(1,0){120}}
\put(40,100){\line(1,0){120}}
\put(40,80){\line(0,1){20}}
\put(160,80){\line(0,1){20}}

\put(70,85){Section \ref{sec_multi_scale_analysis}}

\put(100,80){\vector(0,-1){20}}

\put(40,40){\line(1,0){120}}
\put(40,60){\line(1,0){120}}
\put(40,40){\line(0,1){20}}
\put(160,40){\line(0,1){20}}

\put(70,45){Section \ref{sec_UV}}

\put(100,40){\vector(0,-1){20}}

\put(40,0){\line(1,0){120}}
\put(40,20){\line(1,0){120}}
\put(40,0){\line(0,1){20}}
\put(160,0){\line(0,1){20}}

\put(70,5){Section \ref{sec_IR_model}}

\put(200,240){\line(1,0){120}}
\put(200,260){\line(1,0){120}}
\put(200,240){\line(0,1){20}}
\put(320,240){\line(0,1){20}}

\put(222,245){Appendix \ref{app_flux_phase}}

\put(200,250){\vector(-1,0){40}}

\put(200,200){\line(1,0){120}}
\put(200,220){\line(1,0){120}}
\put(200,200){\line(0,1){20}}
\put(320,200){\line(0,1){20}}

\put(222,205){Appendix \ref{app_grassmann_L1_theory}}

\put(200,210){\vector(-1,0){40}}

\put(200,80){\line(1,0){120}}
\put(200,100){\line(1,0){120}}
\put(200,80){\line(0,1){20}}
\put(320,80){\line(0,1){20}}

\put(222,85){Appendix \ref{app_gevrey}}

\put(200,90){\vector(-1,-1){40}}

\put(200,90){\vector(-1,-2){40}}

\put(200,40){\line(1,0){120}}
\put(200,60){\line(1,0){120}}
\put(200,40){\line(0,1){20}}
\put(320,40){\line(0,1){20}}

\put(222,45){Appendix \ref{app_h_L_limit}}

\put(200,50){\vector(-1,-1){40}}

\put(200,0){\line(1,0){120}}
\put(200,20){\line(1,0){120}}
\put(200,0){\line(0,1){20}}
\put(320,0){\line(0,1){20}}

\put(222,5){Appendix \ref{app_partition_function}}

\put(200,10){\vector(-1,0){40}}

\end{picture}
 \caption{Flow chart of our construction, where the arrows mean major dependency.}\label{fig_flow_chart}
\end{center}
\end{figure}
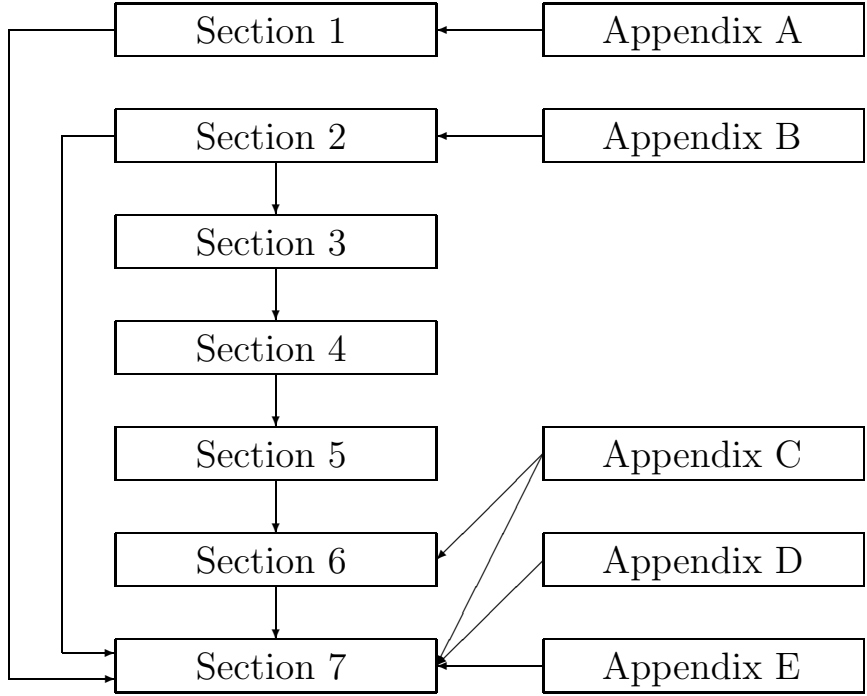

\subsection{The model and the main results}\label{subsec_model}
Here we introduce the model Hamiltonian and state the main results of
this paper. For $L\in \N$ we define the spatial lattice $\G(2L)$ by 
$\G(2L):=\{0,1,\cdots,2L-1\}^2$. For $(\bx,\s)\in\G(2L)\times\spin$ let $\psi_{\bx\s}$ denote the
annihilation operator of the Fermionic Fock space
$F_f(L^2(\G(2L)\times\spin))$ and $\psi_{\bx\s}^*$ denote its adjoint
operator, which is called the creation
operator. For $\bx\in \Z^2$ we define
$\psi_{\bx\s},\psi_{\bx\s}^*$ by identifying $\bx$ with the
corresponding site of $\G(2L)$ which is equal to $\bx$ in $(\Z/2L\Z)^2$.

Let $\be_1:=(1,0)$, $\be_2:=(0,1)\in \Z^2$. We define the amplitude
$t(\cdot,\cdot):\Z^2\times\Z^2\to \R_{\ge 0}$ of the hopping matrix elements as
follows. With parameters
$t_{h,e},t_{h,o},t_{v,e},t_{v,o}\in\R_{>0}$,
\begin{align*}
&t(\bx,\by):=\left\{\begin{array}{ll} t_{h,e} & \text{if
	     }\bx-\by=\be_1,-\be_1\text{ in }(\Z/2L\Z)^2\text{ and
	      }x_2=0\text{ in }\Z/2\Z,\\
t_{h,o} & \text{if
	     }\bx-\by=\be_1,-\be_1\text{ in }(\Z/2L\Z)^2\text{ and
	      }x_2=1\text{ in }\Z/2\Z,\\
t_{v,e} & \text{if
	     }\bx-\by=\be_2,-\be_2\text{ in }(\Z/2L\Z)^2\text{ and
	      }x_1=0\text{ in }\Z/2\Z,\\
t_{v,o} & \text{if
	     }\bx-\by=\be_2,-\be_2\text{ in }(\Z/2L\Z)^2\text{ and
	      }x_1=1\text{ in }\Z/2\Z,\\
0 &\text{otherwise,}
\end{array}
\right.\\
&(\forall \bx=(x_1,x_2),\by\in\Z^2). 
\end{align*}
We allow the hopping matrix elements to be
complex. Assume that the argument
$\theta_L(\cdot,\cdot):\Z^2\times\Z^2\to\R$ satisfies
\begin{align}
&\theta_L(\bx,\by)=-\theta_L(\by,\bx)\text{ in }\R/2\pi
 \Z,\label{eq_phase_condition}\\
&\theta_L(\bx+2mL\be_1+2nL\be_2,\by)=\theta_L(\bx,\by)\text{ in }\R/2\pi
 \Z,\notag\\
&(\forall \bx,\by\in \Z^2,m,n\in \Z)\notag
\end{align}
and
\begin{equation}\label{eq_flux_pi}
\begin{split}
&\theta_L(\bx+\be_1,\bx)+\theta_L(\bx+\be_1+\be_2,\bx+\be_1)\\
&+\theta_L(\bx+\be_2,\bx+\be_1+\be_2)+\theta_L(\bx,\bx+\be_2)=\pi\text{
 in }\R/2\pi
 \Z,\ (\forall \bx\in \Z^2),\\
&\sum_{j=0}^{2L-1}\theta_L((j+1,x),(j,x))=\sum_{j=0}^{2L-1}\theta_L((x,j+1),(x,j))\\
&=0\text{
 in }\R/2\pi
 \Z,\ (\forall x\in \Z).
\end{split}
\end{equation}

The kinetic part $\sH_0$ of the Hamiltonian is defined by
\begin{align}
\sH_0:=&\sum_{\bx\in\G(2L)}\sum_{\s\in
 \spin}\sum_{j=1}^2(t(\bx,\bx+\be_j)e^{i\theta_L(\bx,\bx+\be_j)}\psi_{\bx\s}^*\psi_{\bx+\be_j\s}\label{eq_free_part_original}\\
&+t(\bx,\bx-\be_j)e^{i\theta_L(\bx,\bx-\be_j)}\psi_{\bx\s}^*\psi_{\bx-\be_j\s}).\notag
\end{align}
One can see that $\sH_0^*=\sH_0$.

The condition \eqref{eq_flux_pi} is interpreted as having the magnetic
flux $\pi$ (mod $2\pi$) per plaquette and $0$ (mod $2\pi$) through the
circles winding around the periodic lattice, because the sum in \eqref{eq_flux_pi} is the value of the line integral of the magnetic
vector potential around the corresponding contour, if we adopt the Peierls substitution. One simple example of such $\theta_L$ is that 
\begin{align}
&\theta_L(\bx,\by)=\left\{\begin{array}{ll} \pi&\text{if
		   }\bx-\by=\be_2,-\be_2\text{ in }(\Z/2L\Z)^2\text{ and
	      }x_1=1\text{ in }\Z/2\Z,\\
0&\text{otherwise,}\end{array}
\right. \label{eq_phase_example}\\
&(\forall \bx=(x_1,x_2),\by\in \Z^2).\notag
\end{align}
In this case the nearest-neighbor hopping is pictured as in Figure \ref{fig_single_band_example}.
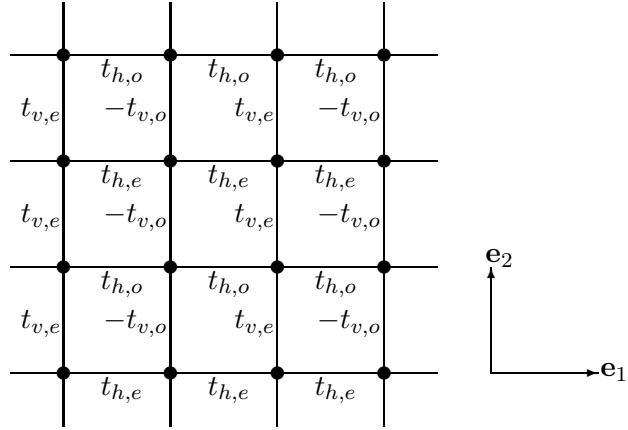
\begin{figure}
\begin{center}
\begin{picture}(230,160)(0,0)
\put(20,20){\circle*{5}}
\put(60,20){\circle*{5}}
\put(100,20){\circle*{5}}
\put(140,20){\circle*{5}}
\put(20,60){\circle*{5}}
\put(60,60){\circle*{5}}
\put(100,60){\circle*{5}}
\put(140,60){\circle*{5}}
\put(20,100){\circle*{5}}
\put(60,100){\circle*{5}}
\put(100,100){\circle*{5}}
\put(140,100){\circle*{5}}
\put(20,140){\circle*{5}}
\put(60,140){\circle*{5}}
\put(100,140){\circle*{5}}
\put(140,140){\circle*{5}}

\put(0,20){\line(1,0){160}}
\put(0,60){\line(1,0){160}}
\put(0,100){\line(1,0){160}}
\put(0,140){\line(1,0){160}}
\put(20,0){\line(0,1){160}}
\put(60,0){\line(0,1){160}}
\put(100,0){\line(0,1){160}}
\put(140,0){\line(0,1){160}}

\put(34,12){\small$t_{h,e}$}
\put(74,12){\small$t_{h,e}$}
\put(114,12){\small$t_{h,e}$}

\put(34,92){\small$t_{h,e}$}
\put(74,92){\small$t_{h,e}$}
\put(114,92){\small$t_{h,e}$}

\put(34,52){\small$t_{h,o}$}
\put(74,52){\small$t_{h,o}$}
\put(114,52){\small$t_{h,o}$}

\put(34,132){\small$t_{h,o}$}
\put(74,132){\small$t_{h,o}$}
\put(114,132){\small$t_{h,o}$}

\put(4,37){\small$t_{v,e}$}
\put(4,77){\small$t_{v,e}$}
\put(4,117){\small$t_{v,e}$}

\put(84,37){\small$t_{v,e}$}
\put(84,77){\small$t_{v,e}$}
\put(84,117){\small$t_{v,e}$}

\put(35,37){\small$-t_{v,o}$}
\put(35,77){\small$-t_{v,o}$}
\put(35,117){\small$-t_{v,o}$}

\put(115,37){\small$-t_{v,o}$}
\put(115,77){\small$-t_{v,o}$}
\put(115,117){\small$-t_{v,o}$}

\put(180,20){\vector(1,0){40}}
\put(180,20){\vector(0,1){40}}
\put(221,19){\small$\be_1$}
\put(178,62){\small$\be_2$}
\end{picture}
 \caption{The nearest-neighbor hopping with the phase $\theta_L$ defined
 by \eqref{eq_phase_example}.}\label{fig_single_band_example}
\end{center}
\end{figure}

To define the interacting part of the Hamiltonian, we assume that the
magnitude of the on-site interaction may depend on sites. More
specifically, with parameters $U_{e,e},U_{o,e},U_{e,o},U_{o,o}\in\R$ we
define $U(\cdot):\Z^2\to \R$ by 
\begin{align*}
U(\bx):=\left\{\begin{array}{ll} U_{e,e} &\text{if }\bx=(0,0)\text{ in
	 }(\Z/2\Z)^2,\\
U_{o,e} &\text{if }\bx=(1,0)\text{ in
	 }(\Z/2\Z)^2,\\
U_{e,o} &\text{if }\bx=(0,1)\text{ in
	 }(\Z/2\Z)^2,\\
U_{o,o} &\text{if }\bx=(1,1)\text{ in
	 }(\Z/2\Z)^2,\\
\end{array}
\right.(\forall \bx\in \Z^2).
\end{align*}
With this $U(\cdot)$, define the interacting part $\sV$ by
\begin{align}
\sV:=\sum_{\bx\in
 \G(2L)}U(\bx)\Bigg(\psi_{\bx\ua}^*\psi_{\bx\da}^*\psi_{\bx\da}\psi_{\bx\ua}-\frac{1}{2}\sum_{\s\in
 \spin}\psi_{\bx\s}^*\psi_{\bx\s}\Bigg).\label{eq_interacting_part_original}
\end{align}

The Hamiltonian $\sH$ is defined by $\sH:=\sH_0+\sV$, which is a self-adjoint
operator on $F_f(L^2(\G(2L)\times \spin))$. Including the quadratic
term in the interacting part as above makes the system half-filled. This
fact can be confirmed by a well-known argument. We provide the
proof in Remark \ref{rem_half_filled} below for completeness. With the
inverse temperature $\beta\in\R_{>0}$, the free energy density of the
system is given by 
$$
-\frac{1}{\beta(2L)^2}\log(\Tr e^{-\beta\sH}).
$$

To shorten formulas, we set $\bt:=(t_{h,e},t_{h,o},t_{v,e},t_{v,o})$
$(\in \R_{>0}^4)$,
\begin{align}
f_{\bt}:=\frac{\min\{t_{h,e}t_{h,o},t_{v,e}t_{v,o}\}}{(\max\{t_{h,e},t_{h,o},t_{v,e},t_{v,o}\})^{\frac{3}{2}}}\cdot\min\left\{\frac{t_{h,o}}{t_{h,e}},\frac{t_{h,e}}{t_{h,o}},\frac{t_{v,o}}{t_{v,e}},\frac{t_{v,e}}{t_{v,o}}\right\}\label{eq_hopping_amplitude_function}
\end{align}
and 
\begin{align*}
D_{\bt}(c):=\{z\in\C\ |\
 |z|< c f_{\bt}^2\}
\end{align*}
for $c\in\R_{>0}$. The goal of this paper is to prove the
following theorem.

\begin{theorem}\label{thm_main_theorem}
Set $\bU:=(U_{e,e},U_{o,e},U_{e,o},U_{o,o})$. 
There exists a constant $c\in\R_{>0}$ independent of any parameter such that the following statements hold true.
\begin{enumerate}
\item\label{item_theorem_analytic_continuation}
There exists a function $F_{\beta,L}(\cdot):\overline{D_{\bt}(c)}^4\to \C$ parameterized by $\beta
     \in \R_{>0}$ and $L\in\N$ satisfying $L\ge
     \max\{t_{h,e},t_{h,o},t_{v,e},t_{v,o}\}\beta$ 
such that $F_{\beta,L}(\cdot)$ is
     continuous in $\overline{D_{\bt}(c)}^4$, analytic in $D_{\bt}(c)^4$ and
\begin{align*}
&F_{\beta,L}(\bU)=-\frac{1}{\beta(2L)^2}\log(\Tr e^{-\beta
 \sH}),\\
&(\forall \bU\in \overline{D_{\bt}(c)}^4\cap \R^4,\beta\in
 \R_{>0},\\
&\quad L\in\N \text{ with } L\ge
     \max\{t_{h,e},t_{h,o},t_{v,e},t_{v,o}\}\beta).
\end{align*}
\item\label{item_theorem_L_limit}
There exists a function $F_{\beta}(\cdot):\overline{D_{\bt}(c)}^4\to \C$
     parameterized by $\beta \in \R_{>0}$, independent of $L\in\N$ such that 
\begin{align*}
\lim_{L\to \infty\atop L\in \N}\sup_{\bz \in
 \overline{D_{\bt}(c)}^4}|F_{\beta,L}(\bz)-F_{\beta}(\bz)|=0,\ (\forall \beta\in\R_{>0}).
\end{align*}
\item\label{item_theorem_beta_limit}
There exists a function $F(\cdot):\overline{D_{\bt}(c)}^4\to \C$
     independent of $\beta\in\R_{>0}$ such that 
\begin{align*}
\lim_{\beta\to \infty\atop \beta\in \R_{>0}}\sup_{\bz \in
 \overline{D_{\bt}(c)}^4}|F_{\beta}(\bz)-F(\bz)|=0.
\end{align*}
\end{enumerate}
\end{theorem}

If we impose additional conditions on
$t_{h,e},t_{h,o},t_{v,e},t_{v,o}$, $U_{e,e},U_{o,e}$, $U_{e,o}$, 
$U_{o,o}$ and $L$, we can relate the free
energy density considered in Theorem \ref{thm_main_theorem} to the
minimum free energy in the flux phase problem, which seeks a phase
$\phi:\Z^2\times\Z^2\to \R$ of the hopping parameter 
minimizing the free energy. Lieb (\cite{L}) essentially gave a sufficient condition
for a phase to be a minimizer of the flux phase problem. 
The sufficient condition was also claimed by Macris and Nachtergaele in
\cite{MN}. That is the condition 
\eqref{eq_flux_pi} if $L\in 2\N+1$. For readers who are not familiar with the flux phase
problem, we restate Lieb's result in Appendix \ref{app_flux_phase} with
some supplementary arguments which were not explicit in the letter \cite{L}. In
mathematical terms, the flux phase problem is to find
$\tilde{\phi}:\Z^2\times\Z^2\to \R$ satisfying
\eqref{eq_phase_condition} such that 
\begin{equation}\label{eq_mathematical_flux_phase}
\begin{split}
&-\frac{1}{\beta(2L)^2}\log(\Tr e^{-\beta \sH(\tilde{\phi})})\\
& =\min\left\{-\frac{1}{\beta(2L)^2}\log(\Tr e^{-\beta \sH(\phi)})\
 \Big|\ \phi:\Z^2\times\Z^2\to \R\text{ satisfying }\eqref{eq_phase_condition}\right\},
\end{split}
\end{equation}
where
\begin{align*}
\sH(\phi):=&\sum_{(\bx,\s)\in\G(2L)\times\spin}\sum_{j=1}^2(t(\bx,\bx+\be_j)e^{i\phi(\bx,\bx+\be_j)}\psi_{\bx\s}^*\psi_{\bx+\be_j\s}\\
&+t(\bx,\bx-\be_j)e^{i\phi(\bx,\bx-\be_j)}\psi_{\bx\s}^*\psi_{\bx-\be_j\s})+\sV.\end{align*}
Since this is equivalent to a minimization problem of a continuous function defined on
the compact set $[0,2\pi]^{2(2L)^2}$, a minimizer exists. Under the
additional conditions that $t_{h,e}=t_{h,o}$, $t_{v,e}=t_{v,o}$, 
$U_{e,e}=U_{o,e}=U_{e,o}=U_{o,o}$ and $L\in 2\N+1$, Theorem \ref{thm_Lieb_theorem} in
Appendix \ref{app_flux_phase} ensures that any phase $\theta_L$ satisfying \eqref{eq_phase_condition}
and \eqref{eq_flux_pi} is a minimizer of the flux phase problem. Thus,
we have the following corollary.

\begin{corollary}\label{cor_application_flux_phase}
Assume that $t_{h,e}=t_{h,o}$, $t_{v,e}=t_{v,o}$, 
$U_{e,e}=U_{o,e}=U_{e,o}=U_{o,o}=U\in \R$. Let
 $G_{\beta,L}(U)$ denote the right-hand side of
 \eqref{eq_mathematical_flux_phase}. Then,
 there exists a constant $c\in\R_{>0}$ independent of any parameter such that
 the following statements hold true.
 \begin{enumerate}
\item
There exists a function $F_{\beta,L}(\cdot):\overline{D_{\bt}(c)}\to \C$ parameterized by $\beta
     \in \R_{>0}$ and $L\in2\N+1$ satisfying $L\ge \max\{t_{h,e},t_{v,e}\}\beta$ such that $F_{\beta,L}(\cdot)$ is
     continuous in $\overline{D_{\bt}(c)}$, analytic in $D_{\bt}(c)$ and
\begin{align*}
&F_{\beta,L}(U)=G_{\beta,L}(U),\\ 
&(\forall U\in \overline{D_{\bt}(c)}\cap \R,\beta\in \R_{>0},\\
&\quad L\in 2\N+1\text{ with }L\ge  \max\{t_{h,e},t_{v,e}\}\beta).
\end{align*}
\item
There exists a function $F_{\beta}(\cdot):\overline{D_{\bt}(c)}\to \C$
     parameterized by $\beta \in \R_{>0}$, independent of $L\in2\N+1$ such that 
\begin{align*}
\lim_{L\to \infty\atop L\in 2\N+1}\sup_{z \in
 \overline{D_{\bt}(c)}}|F_{\beta,L}(z)-F_{\beta}(z)|=0,\ (\forall \beta\in\R_{>0}).
\end{align*}
\item
There exists a function $F(\cdot):\overline{D_{\bt}(c)}\to \C$
     independent of $\beta \in\R_{>0}$ such that 
\begin{align*}
\lim_{\beta\to \infty\atop \beta\in \R_{>0}}\sup_{z \in
 \overline{D_{\bt}(c)}}|F_{\beta}(z)-F(z)|=0.
\end{align*}
\end{enumerate}
\end{corollary}

\begin{remark}
Theorem \ref{thm_main_theorem} implies the analyticity of the
 infinite-volume, zero-temperature limit of the free energy density in
 the following sense. There exists a function
 $F(\cdot):\overline{D_{\bt}(c)}^4\to \C$ independent of
 $\beta\in\R_{>0}$, $L\in\N$ such that $F(\cdot)$ is continuous in
 $\overline{D_{\bt}(c)}^4$, analytic in $D_{\bt}(c)^4$ and 
\begin{align*}
\lim_{\beta\to \infty\atop \beta\in \R_{>0}}\lim_{L\to \infty\atop L\in \N}
\sup_{\bU \in
 \overline{D_{\bt}(c)}^4\cap\R^4}\left|-\frac{1}{\beta(2L)^2}\log(\Tr e^{-\beta
 \sH})-F(\bU)\right|=0.
\end{align*}
\end{remark}

\begin{remark}\label{rem_half_filled}
The system is half-filled. To confirm this, let us define the operator
 $A$ on $F_f(L^2(\G(2L)\times\spin))$ by
\begin{align*}
&A(\alpha \O_{2L}):= \overline{\alpha}\prod_{\bx\in
 \G(2L)}(\psi_{\bx\ua}^*\psi_{\bx\da}^*)\O_{2L},\\
&A(\alpha\psi_{(x_{1,1},x_{1,2})\s_1}^*\psi_{(x_{2,1},x_{2,2})\s_2}^*\cdots\psi_{(x_{n,1},x_{n,2})\s_n}^*\O_{2L})\\
&:=(-1)^{\sum_{j=1}^n(x_{j,1}+x_{j,2})}\overline{\alpha}\psi_{(x_{1,1},x_{1,2})\s_1}\psi_{(x_{2,1},x_{2,2})\s_2}\cdots\psi_{(x_{n,1},x_{n,2})\s_n}\\
&\qquad\cdot\prod_{\bx\in
 \G(2L)}(\psi_{\bx\ua}^*\psi_{\bx\da}^*)\O_{2L},\\
&(\forall \alpha\in\C,(x_{j,1},x_{j,2})\in \G(2L),\s_j\in\spin\
 (j=1,2,\cdots,n) )
\end{align*}
and by linearity, where $\O_{2L}$ denotes the vacuum of
 $F_f(L^2(\G(2L)\times$\\
$\spin))$. One can check that $A$ is unitary, $AHA^*=H$ and
\begin{align*}
&A\psi_{\bx \s}^*\psi_{\bx\s}A^*=id_{2L}-\psi_{\bx \s}^*\psi_{\bx\s},\
 (\forall (\bx,\s)\in \G(2L)\times \spin),
\end{align*}
where $id_{2L}$ denotes the identity map on $F_f(L^2(\G(2L)\times\spin))$.
Thus,
\begin{align*}
\Tr(e^{-\beta \sH}\psi_{\bx\s}^*\psi_{\bx\s})&=\Tr(e^{-\beta
 A\sH A^*}A\psi_{\bx\s}^*\psi_{\bx\s}A^*)\\
&=\Tr e^{-\beta
 \sH}-\Tr(e^{-\beta \sH}\psi_{\bx\s}^*\psi_{\bx\s}),
\end{align*}
which implies that
$$
\frac{\Tr(e^{-\beta \sH}\psi_{\bx\s}^*\psi_{\bx\s})}{\Tr e^{-\beta
 \sH}}=\frac{1}{2},\ (\forall (\bx,\s)\in \G(2L)\times\spin). 
$$
\end{remark}

\begin{remark}\label{rem_phase_freedom}
In Theorem \ref{thm_main_theorem} we have freedom to choose a phase
 $\theta_L$ satisfying \eqref{eq_phase_condition} and
 \eqref{eq_flux_pi}. However, the free energy density is
 independent of the choice of $\theta_L$. Let $\theta_L$, $\theta_L'$ be
 phases satisfying \eqref{eq_phase_condition} and \eqref{eq_flux_pi} and
 $\sH(\theta_L)$, $\sH(\theta_L')$ be the Hamiltonian having the phase
 $\theta_L$, $\theta_L'$ respectively. Then, Lemma
 \ref{lem_gauge_invariance} given in Appendix \ref{app_flux_phase}
 implies that $\Tr
 e^{-\beta \sH(\theta_L)}=\Tr
 e^{-\beta \sH(\theta_L')}$. In brief, this equality is due to the fact 
that the flux of $\theta_L$
 through any circuit in the periodic lattice $(\Z/2L\Z)^2$ is the
 same as that of $\theta_L'$.
\end{remark}

\begin{remark}
The proof of Theorem \ref{thm_Lieb_theorem} requires that the hopping
 amplitude and the magnitude of on-site interaction are invariant under
 vertical and horizontal reflections. To meet this requirement, we
 need to assume that $t_{h,e}=t_{h,o}$, $t_{v,e}=t_{v,o}$, 
$U_{e,e}=U_{o,e}=U_{e,o}=U_{o,o}$. Moreover, on the assumption $L\in
 2\N+1$, having the flux $\pi(L-1)$ (mod $2\pi$) through the circles
 around the periodic lattice, another requirement of Theorem \ref{thm_Lieb_theorem}, is equal to having the flux $0$ (mod
 $2\pi$), which is satisfied by our model Hamiltonian. In the case
 $L\in 2\N$ we do not have the equivalence between the free energy governed by
our model Hamiltonian and the minimum free energy in the flux phase
 problem.
\end{remark}

\begin{remark}
Consider the Hamiltonian $\sH$ with the phase defined by
 \eqref{eq_phase_example}. If $t_{h,o}=t_{v,o}=U_{o,o}=0$, the
 Hamiltonian $\sH$ becomes the half-filled Hubbard model on the
 copper-oxide (CuO) lattice. Since the condition $t_{h,o},t_{v,o}\neq 0$
 is indispensable for our analysis, we cannot treat the half-filled CuO
 Hubbard model itself in this paper. As an operator on the
 finite-dimensional space our Hamiltonian can be
 arbitrarily close to the half-filled CuO Hubbard model as
 $t_{h,o},t_{v,o}\searrow 0$. For such an approximate model with small
 but non-zero $t_{h,o},t_{v,o}$, Theorem \ref{thm_main_theorem}
 guarantees the existence of infinite-volume, zero-temperature limit of
 the free energy density and its analyticity with the coupling constants
$U_{e,e}, U_{o,e}, U_{e,o}$. However, since the domain $D_{\bt}(c)$
 shrinks as $t_{h,o},t_{v,o}\searrow 0$, we cannot extract any
 information on the free energy
 density defined in the half-filled CuO Hubbard model from
Theorem \ref{thm_main_theorem}.\end{remark}

\begin{remark}\label{rem_no_IR_analysis}
Later in Lemma \ref{lem_temperature_dependency_covariance_bound} in
 Section \ref{sec_IR_model} we will see that the integral of modulus of
 the free covariance is bounded by a constant times $\beta$ from above
 and below if $L$ is sufficiently large. This also implies that the free
 covariance with the Matsubara UV cut-off has the same bound property in
 low temperature, since the integral of modulus of the free covariance
 with the large Matsubara frequency, the difference between the free
 covariance and that with the Matsubara UV cut-off, is bounded from above
 independently of $\beta$. These facts tell us that we cannot prove the
 analyticity of the free energy density in the infinite-volume limit in
 a domain larger than $\{U\in \C\ |\ |U|<c\beta^{-1}\}^4$ by means of a
 single-scale analysis based on Pedra-Salmhofer's determinant bound as
 in \cite{K1} or a multi-scale analysis over the Matsubara frequency as
 in \cite{K3}, since the inverse of the $L^1$-bound of the free
 covariance with or without UV cut-off determines the maximal magnitude
 of the coupling in these theories. Therefore, we are led to perform an
 IR analysis in order to reach the infinite-volume,
 zero-temperature limit in this model of interacting electrons.
\end{remark}

\section{Formulation}\label{sec_formulation}

In this section, first we define the multi-band Hamiltonian $H$
consisting of the free part $H_0$ and the interacting part $V$ in a generalized
setting. In Subsection \ref{subsec_model} we introduced the single-band
Hamiltonian $\sH$. We will prove Theorem \ref{thm_main_theorem} by
formulating the Hamiltonian $\sH$ into a 4-band Hamiltonian in Section
\ref{sec_IR_model}. The Hamiltonian $H$ should be considered as a
generalization of the 4-band model Hamiltonian. Then we introduce the finite-dimensional Grassmann integral
formulation of the normalized
free energy density $-\frac{1}{\beta L^d}\log(\Tr e^{-\beta H}/\Tr
e^{-\beta H_0})$. Though the main theorem of
this paper concerns the free energy density of the form $-\frac{1}{\beta L^d}\log(\Tr e^{-\beta H})$, it is more convenient to deal with the
normalized one, since it fits in the framework of Grassmann Gaussian
integration. We can reach the conclusions on the free energy density   
from the analysis of the normalized free energy density, since the
non-interacting free energy density $-\frac{1}{\beta L^d}\log(\Tr e^{-\beta H_0})$, the difference between them, is exactly computable.  

\subsection{The multi-band Hamiltonian}\label{subsec_Hamiltonian}
Let us set up a system which we focus on until we analyze the
specific model in Section \ref{sec_IR_model}. Let $d\ (\in\N)$ denote the
spatial dimension. Take a basis
$\bu_1,\bu_2,\cdots,\bu_d$ of $\R^d$. Let $\bv_1,\bv_2,\cdots,\bv_d$ be
another basis of $\R^d$ satisfying $\<\bu_l,\bv_m\>=\delta_{l,m}$
$(\forall l,m\in \{1,2,\cdots,d\})$, where $\<\cdot,\cdot\>$ is the
standard inner product of $\R^d$. The spatial lattice $\G$ is defined by 
$$
\G:=\left\{\sum_{j=1}^dm_j\bu_j\ \Big|\ m_j\in \{0,1,\cdots, L-1\}\
(j=1,2,\cdots,d)\right\}.
$$
The momentum lattice $\G^*$ dual to $\G$ is given by 
$$
\G^*:=\left\{\frac{2\pi}{L}\sum_{j=1}^dm_j\bv_j\ \Big|\ m_j\in \{0,1,\cdots, L-1\}\
(j=1,2,\cdots,d)\right\}.
$$
With a number $b\ (\in \N)$ we assume that the crystal lattice is
modeled by the lattice $\G$ with a $b$-point basis. The integer $b$
stands for the number of atomic sites in a primitive unit cell of the
lattice $\G$. Each site of the crystal lattice is identified
with an element of the set $\{1,2,\cdots,b\}\times \G$. For conciseness 
we set $\cB:=\{1,2,\cdots,b\}$.

The Hamiltonian $H$ is defined as a self-adjoint operator on the Fermio-nic Fock space $F_f(L^2(\cB\times\G\times\spin))$. To define the
free part of the Hamiltonian $H$, we assume that in the momentum space 
the hopping matrix is represented by $E(\bk)\in
\Mat(b,\C)$ $(\bk\in\G^*)$. Moreover we assume that the domain of
$E(\cdot)$ can be extended to $\R^d$ and 
\begin{align}
&E\in C(\R^d;\Mat(b,\C)),\notag\\
&E(\bk)^*=E(\bk),\ (\forall \bk\in \R^d),\label{eq_dispersion_hermite}\\
&E(\bk+2\pi\bv_j)=E(\bk),\ (\forall \bk\in \R^d,j\in
 \{1,2,\cdots,d\})\label{eq_dispersion_periodic}.
\end{align}
We consider $\Mat(b,\C)$ as a $b^2$-dimensional complex Banach
space with the norm $\|\cdot\|_{b\times b}$ defined by
$$
\|A\|_{b\times b}:=\sup_{\bv\in \C^b\text{ with }\atop\|\bv\|_{\C^b}=1}\|A\bv\|_{\C^b},\ (A\in \Mat(b,\C)),
$$
where $\|\cdot\|_{\C^b}$ denotes the norm of $\C^b$ induced by the
standard inner product $\<\cdot,\cdot\>_{\C^b}$. 
With $E(\cdot)$ we define the free part $H_0$ by
\begin{align}
H_0:=\frac{1}{L^d}\sum_{\rho,\eta\in \cB}\sum_{\bx,\by\in\G}\sum_{\s\in
 \spin}\sum_{\bk\in\G^*}e^{i\<\bk,\bx-\by\>}E(\bk)(\rho,\eta)\psi_{\rho\bx\s}^*\psi_{\eta\by\s},\label{eq_free_hamiltonian_general}
\end{align}
where $\psi_{\rho\bx\s}$ is the annihilation operator destroying an
electron with the spin $\s$ on the site
$(\rho,\bx)$ and
$\psi_{\rho\bx\s}^*$ is its adjoint operator called the creation
operator.

The interacting part $V$ is defined by
\begin{align}
V:=\sum_{\rho\in
\cB}\sum_{\bx\in\G}U_{\rho}\Bigg(\psi_{\rho\bx\ua}^*\psi_{\rho\bx\da}^*\psi_{\rho\bx\da}\psi_{\rho\bx\ua}-\frac{1}{2}\sum_{\s\in\spin}\psi_{\rho\bx\s}^*\psi_{\rho\bx\s}\Bigg),\label{eq_interacting_hamiltonian_general}
\end{align}
with the coupling constants $U_{\rho}\in \R$ $(\rho\in \cB)$. To be more
precise, the second term of $V$ should be considered as a part
representing the on-site energy minus the chemical potential. 
Since we are going to construct a theory for the half-filled systems, 
the on-site quadratic term of this form needs to be included in $V$.
The Hamiltonian governing the multi-band many-electron system is defined
by $H:=H_0+V$$:F_f(L^2(\cB\times\G\times\spin))\to F_f(L^2(\cB\times\G\times\spin))$. By the condition \eqref{eq_dispersion_hermite}, $H$ is
self-adjoint. In the rest of this section we will introduce the
Grassmann integral formulation of the normalized free energy density
$-\frac{1}{\beta L^d}\log(\Tr e^{-\beta H}/\Tr
e^{-\beta H_0})$.

\subsection{The finite-dimensional Grassmann
  integrals}\label{subsec_grassmann_integral}
Let us summarize the notions of Grassmann integration over a 
finite-dimensional Grassmann algebra. Take a parameter $h\in
(2/\beta)\N$ and introduce the discrete analogue of the interval
$[0,\beta)$ by
$$
[0,\beta)_h:=\left\{0,\frac{1}{h},\frac{2}{h},\cdots,\beta-\frac{1}{h}\right\}.
$$
We take the parameter $h$ from $(2/\beta)\N$ rather than from
$(1/\beta)\N$ in order to refer to the basic results of
\cite[\mbox{Appendix C}]{K1} constructed with $h$ belonging to
$(2/\beta)\N$. The index sets $I_0$, $I$ are defined by
\begin{align*}
&I_0:=\cB\times \G\times\spin\times[0,\beta)_h,\\
&I:=I_0\times\{1,-1\}.
\end{align*}
Let $N$ stand for the number $4b\beta h L^d$, the cardinality
of $I$. Let $\cV$ denote the complex vector space spanned by the
abstract basis $\{\psi_X\}_{X\in I}$. Similarly for $p\in \N$ let $\cV_p$
 be the complex vector space spanned by the basis $\{\psi_X^p\}_{X\in
 I}$. For $X\in I_0$ we sometimes write $\opsi_X$, $\psi_X$ in place of
 $\psi_{(X,1)}$, $\psi_{(X,-1)}$ respectively.

For a finite-dimensional complex vector space $W$ and $n\in \N$, set  
$\bigwedge^0 W:=\C $ and let $\bigwedge^n W$ denote the $n$-fold
anti-symmetric tensor product of $W$. Moreover, set 
$$
\bigwedge W:= \bigoplus_{n=0}^{\dim W}\bigwedge^n W.
$$
For $n\in \{0,1,\cdots,\dim W\}$ let $\cP_n:\bigwedge W\to \bigwedge^n W$
denote the standard projection.

We call $\bigwedge \cV$ Grassmann algebra generated by $\{\psi_X\}_{X\in
I}$. We write an element $f$ of $\bigwedge \cV$ as $f(\psi)$ when we want to show
its Grassmann variable explicitly. We can define $f(\psi+\psi^p)\in
\bigwedge (\cV\bigoplus\cV_p)$ from $f(\psi)\in \bigwedge \cV$ by
replacing each $\psi_X$ by $\psi_X+\psi_X^p$ inside $f(\psi)$.
For $\bX=(X_1,X_2,\cdots,X_m)\in I^m$ we simply write $\psi_{\bX}$ in
place of $\psi_{X_1}\psi_{X_2}\cdots \psi_{X_m}$ and $\bX_{\s}$ in place
of $(X_{\s(1)},X_{\s(2)},\cdots,X_{\s(m)})$ for any $\s\in \S_m$, where
$\S_m$ denotes the set of all permutations over $\{1,2,\cdots,m\}$. We call a
function $f_m:I^m\to \C$ anti-symmetric if
$f_m(\bX)=\sgn(\s)f_m(\bX_{\s})$ for any $\bX\in I^m$, $\s\in \S_m$. 

For any $f(\psi)\in \bigwedge \cV$ there uniquely exist $f_0\in\C$ and
anti-symmetric functions $f_m(\cdot):I^m\to\C$ $(m\in \{1,2,\cdots,N\})$ such
that 
$$
f(\psi)=f_0+\sum_{m=1}^N\left(\frac{1}{h}\right)^m\sum_{\bX\in I^m}f_m(\bX)\psi_{\bX}.
$$
Throughout the paper we follow the notational convention that for
$f(\psi)\in\bigwedge \cV$, $f_m(\psi)$ denotes $\cP_mf(\psi)$ and $f_m(\cdot):I^m\to\C$ denotes the anti-symmetric
kernel of $f_m(\psi)$. For example, we write as follows.
$$
f(\psi)=\sum_{m=0}^Nf_m(\psi),\ f_m(\psi)=\frah^m\sum_{\bX\in
I^m}f_m(\bX)\psi_{\bX}.
$$

We can construct a norm in the complex vector
space $\bigwedge \cV$ by defining a norm in the space of anti-symmetric
functions on $I^m$ for all $m\in \{1,2,\cdots, N\}$. In this paper we
will introduce various norms in the space of anti-symmetric functions. We
will define the norms one by one when necessary rather than by listing them all
together at this stage. 

Let $a\in\N$ and $D$ be a domain of $\C^a$. Assume that
$f(\bz)(\psi)\in\bigwedge \cV$ is parameterized by
$\bz\in\overline{D}$. We say that $f(\bz)(\psi)$ is continuous with
$\bz$ in $\overline{D}$ if so is $f(\bz)_m(\bX)$ $(\forall
m\in\{0,1,\cdots,N\},\bX\in I^m)$. Similarly we say that $f(\bz)(\psi)$
is analytic with $\bz$ in $D$ if so is $f(\bz)_m(\bX)$ $(\forall
m\in\{0,1,\cdots,N\},\bX\in I^m)$. In this case we define the Grassmann polynomials 
\begin{align*}
&\prod_{j=1}^a\left(\frac{\partial}{\partial
z_j}\right)^{n_j}f(\bz)(\psi)\in\bigwedge \cV,\\
&(\bz=(z_1,z_2,\cdots,z_a)\in
D,n_j\in\N\cup\{0\}\ (j=1,2,\cdots,a))
\end{align*}
by
\begin{align*}
\prod_{j=1}^a\left(\frac{\partial}{\partial
z_j}\right)^{n_j}f(\bz)(\psi):=\sum_{m=0}^N\frah^m\sum_{\bX\in I^m}
\prod_{j=1}^a\left(\frac{\partial}{\partial
z_j}\right)^{n_j}f(\bz)_m(\bX)\psi_{\bX}.
\end{align*}
Consider a sequence $(f^n(\psi))_{n=1}^{\infty}$ of $\bigwedge \cV$. We
say that $f^n(\psi)$ converges as $n\to \infty$ if so
does $f_m^n(\bX)$ $(\forall m\in\{0,1,\cdots,N\},\bX\in I^m)$. Consider
a sequence $(f^n(\bz)(\psi))_{n=1}^{\infty}$  of $\bigwedge \cV$
parameterized by $\bz\in \overline{D}$. We say that $f^n(\bz)(\psi)$
uniformly converges with $\bz\in\overline{D}$ as $n\to \infty$ if so
does $f^n(\bz)_m(\bX)$ $(\forall m\in\{0,1,\cdots,N\},\bX\in I^m)$.
If a norm is defined in $\bigwedge \cV$, the normed space $\bigwedge\cV$
is complete, since $\dim\bigwedge \cV<\infty$. These definitions of continuity,
analyticity, derivative, convergence and uniform convergence are
equivalent to those defined in the Banach space $\bigwedge\cV$.

For $p_1,p_2,\cdots,p_n,p\in \N$ with $p_j\neq p$ $(\forall j\in
\{1,2,\cdots,n\})$ the Grassmann Gaussian integral $\int\cdot
d\mu_{C_o}(\psi^p)$ with a covariance matrix $C_o:I_0^2\to\C$ is a
linear map from $\bigwedge (\cV_{p_1}\bigoplus\cdots
\bigoplus\cV_{p_n}\bigoplus\cV_{p})$ to $\bigwedge (\cV_{p_1}\bigoplus\cdots
\bigoplus\cV_{p_n})$ defined as follows. For any $f\in \bigwedge (\cV_{p_1}\bigoplus\cdots
\bigoplus\cV_{p_n})$,
\begin{align*}
&\int fd\mu_{C_o}(\psi^p):=f,\\
&\int
 f\opsi_{X_1}^p\cdots\opsi_{X_l}^p\psi_{Y_m}^p\cdots\psi_{Y_1}^pd\mu_{C_o}(\psi^p)\\
&:=\left\{\begin{array}{ll}\det(C_o(X_i,Y_j))_{1\le i,j\le l}f& \text{if }l=m,\\
0&\text{if }l\neq m.\end{array}
\right.
\end{align*}
Then, for any $g\in \bigwedge (\cV_{p_1}\bigoplus\cdots
\bigoplus\cV_{p_n}\bigoplus \cV_p)$ the value of $\int g
d\mu_{C_o}(\psi^p)$ can be uniquely determined by linearity and anti-symmetry.
For $Y\in I$ the left derivative $\partial/\partial \psi_Y^p$ is a
linear transform on 
$\bigwedge  (\cV_{p_1}\bigoplus\cdots
\bigoplus\cV_{p_n}\bigoplus \cV_p)$ defined as follows.
\begin{align*}
&\frac{\partial}{\partial
 \psi_Y^p}(\alpha\psi_{X_1}^{q_1}\cdots\psi_{X_l}^{q_l}\psi_{Y}^{p}\psi_{X_{l+1}}^{q_{l+1}}\cdots
 \psi_{X_m}^{q_m}):=(-1)^l\alpha\psi_{X_1}^{q_1}\cdots\psi_{X_l}^{q_l}\psi_{X_{l+1}}^{q_{l+1}}\cdots
 \psi_{X_m}^{q_m},\\
&\frac{\partial}{\partial
 \psi_Y^p}(\alpha\psi_{X_1}^{q_1}\cdots\psi_{X_l}^{q_l}\psi_{X_{l+1}}^{q_{l+1}}\cdots \psi_{X_m}^{q_m}):=0,\\
&\big(\forall \alpha\in\C,
 \psi_{X_1}^{q_1},\cdots,\psi_{X_m}^{q_m}
 \in
\{\psi_X^{p_1},\cdots, \psi_X^{p_n},\psi_X^{p}\}_{X\in
I}\backslash \{\psi_Y^p\}\big).
\end{align*}
 Then, the value of $\partial
/\partial \psi_Y^p$ on any element of $\bigwedge (\cV_{p_1}\bigoplus\cdots
\bigoplus\cV_{p_n}\bigoplus \cV_p)$ can be uniquely determined by
linearity and anti-symmetry. For $\bX=(X_1,X_2,\cdots,X_m)\in I^m$ we
sometimes write $\partial/\partial\psi_{\bX}$ in place of
$\partial/\partial\psi_{X_1}\cdot\partial/\partial\psi_{X_2}\cdots\partial/\partial\psi_{X_m}$
for simplicity. 

We will frequently deal with the exponential and the
logarithm of a Grassmann polynomial. Let us recall their
definitions. For $f\in \bigwedge \cV$ the polynomial $e^f$ $(\in\bigwedge \cV)$ is defined by
$$
e^f:=e^{f_0}\sum_{n=0}^N\frac{1}{n!}(f-f_0)^n.
$$
Additionally, assume that $f_0\in\C\backslash \R_{\le 0}$. The logarithm of $f$ is defined by
$$
\log f:=\log f_0 +\sum_{n=1}^N\frac{(-1)^{n-1}}{n}\left(\frac{f-f_0}{f_0}\right)^n.
$$
For any $z\in\C\backslash \R_{\le 0}$ we define $\log z$
by the principal value $\log|z|+i\theta$, where $\theta\in(-\pi,\pi)$
satisfies $z=|z|e^{i\theta}$.

\subsection{The full covariance}\label{subsec_covariance}
The covariance in our Grassmann Gaussian integral formulation of the
free energy density is equal to the non-interacting 2-point correlation function. For $(\rho,\bx,\s,x)$, $(\eta,\by,\tau,y)\in \cB\times \G\times
\spin \times [0,\beta)$,
\begin{align}
C(\rho\bx\s x,\eta\by\tau y):=\frac{\Tr(e^{-\beta
H_0}T(\psi_{\rho\bx\s}^*(x)\psi_{\eta\by\tau}(y)))}{\Tr e^{-\beta H_0}},\label{eq_covariance_2_point_function}
\end{align}
where $\psi_{\rho\bx \s}^*(x):=e^{xH_0}\psi_{\rho\bx \s}^*e^{-xH_0}$, 
 $\psi_{\eta\by \tau}(y):=e^{yH_0}\psi_{\eta\by \tau}e^{-yH_0}$,
$$
T(\psi_{\rho\bx\s}^*(x)\psi_{\eta\by\tau}(y)):=1_{x\ge
y}\psi_{\rho\bx\s}^*(x)\psi_{\eta\by\tau}(y)-1_{x<y}\psi_{\eta\by\tau}(y)\psi_{\rho\bx\s}^*(x).
$$ 
For a proposition $P$ the value of $1_P$ is defined as follows. $1_P:=1$ if $P$ is true, $1_P:=0$ otherwise.
We use the same symbol $C$ even when its variables are restricted to be
in the finite subset $I_0^2$. 

Let $\cM$ denote the set of the Matsubara frequency $(\pi/\beta)(2\Z+1)$.
We define the $h$-dependent finite subset $\cM_h$ of the Matsubara frequency by
$$\cM_h:=\left\{\o \in \frac{\pi}{\beta}(2\Z+1)\ \Big|\ |\o|<\pi h\right\}.$$
Let $I_b$ denote the $b\times b$ unit matrix. The covariance matrix
$C:I_0^2\to \C$ is characterized as follows. 

\begin{lemma}\label{lem_characterization_covariance}
For any $(\rho,\bx,\s,x)$, $(\eta,\by,\tau,y)\in I_0$, 
\begin{align}
&C(\rho\bx\s x,\eta\by\tau y)\label{eq_characterization_covariance}\\
&=\frac{\delta_{\s,\tau}}{\beta
 L^d}\sum_{\bk\in \G^*}\sum_{\o\in
 \cM_h}e^{-i\<\bx-\by,\bk\>}e^{i(x-y)\o}h^{-1}(I_b-e^{-i\frac{\o}{h}I_b+\frac{1}{h}\overline{E(\bk)}})^{-1}(\rho,\eta).\notag
\end{align}
\end{lemma}
\begin{proof} One can complete the characterization in the
 same way as in \cite[\mbox{Appendix A}]{K3}. For readers'
 convenience we provide a sketch of the proof. 

For any $\bk\in\R^d$ let $\alpha_1(\bk),\cdots,\alpha_b(\bk)$ $(\in \R)$ be the eigen values
 of $E(\bk)$. There exists a
 unitary matrix $U(\bk)\in \Mat(b,\C)$ such that 
\begin{align}\label{eq_diagonalization_kinetic}
(U(\bk)^*E(\bk)U(\bk))(\rho,\eta)=\alpha_{\rho}(\bk)\delta_{\rho,\eta},\
 (\forall \rho,\eta\in \cB).
\end{align}
Define the matrix $W:(\cB\times \G\times \spin)^2\to\C$ by 
$$
W(\rho\bx\s,\eta\by\tau):=\frac{\delta_{\s,\tau}}{L^d}\sum_{\bk\in\G^*}e^{-i\<\bx-\by,\bk\>}\overline{U(\bk)(\rho,\eta)}.
$$
Set 
\begin{align*}
&(W\psi^*)_{\rho\bx\s}:=\sum_{(\eta,\by,\tau)\in \cB\times
 \G\times\spin}W(\rho\bx\s,\eta\by\tau)\psi_{\eta\by\tau}^*,\\
&(\forall (\rho,\bx,\s)\in \cB\times \G\times \spin).
\end{align*}
Let us define the linear transform $F$ on $F_f(L^2(\cB\times
 \G\times \spin))$ by
\begin{align*}
&F\O:=\O,\\
&F\psi^*_{X_1}\psi^*_{X_2}\cdots\psi^*_{X_n}\O:=(W\psi^*)_{X_1}(W\psi^*)_{X_2}\cdots
 (W\psi^*)_{X_n}\O,\\
&(\forall n\in \N,X_1,X_2,\cdots,X_n\in \cB\times \G\times \spin)
\end{align*}
and by linearity, where $\O$ denotes the vacuum of $F_f(L^2(\cB\times
 \G\times \spin))$. By using the unitary
 property of $U(\bk)$ and the equality \eqref{eq_diagonalization_kinetic} one can
 check that the transform $F$ is unitary and 
\begin{align*}
FH_0F^*=\sum_{\rho\in \cB}\sum_{\bx,\by\in \G}\sum_{\s\in
 \spin}\frac{1}{L^d}\sum_{\bk\in
 \G^*}e^{i\<\bx-\by,\bk\>}\alpha_{\rho}(\bk)\psi_{\rho\bx\s}^*\psi_{\rho\by\s}.
\end{align*}
For any $(\rho,\bx,\s,x)$, $(\eta,\by,\tau,y)\in \cB\times \G\times
 \spin\times[0,\beta)$ set 
\begin{align*}
&\tilde{\psi}_{\rho\bx\s}^*(x):=e^{xFH_0F^*}\psi_{\rho\bx\s}^*e^{-xFH_0F^*},\\
&\tilde{\psi}_{\rho\bx\s}(x):=e^{xFH_0F^*}\psi_{\rho\bx\s}e^{-xFH_0F^*},\\
&T(\tilde{\psi}_{\rho\bx\s}^*(x)\tilde{\psi}_{\eta\by\tau}(y)):=1_{x\ge
 y}\tilde{\psi}_{\rho\bx\s}^*(x)\tilde{\psi}_{\eta\by\tau}(y)-1_{x<
 y}\tilde{\psi}_{\eta\by\tau}(y)\tilde{\psi}_{\rho\bx\s}^*(x).
\end{align*}
Since $F$ is unitary,
\begin{align}
&C(\rho\bx\s x,\eta\by\tau y)\label{eq_pre_characterization_covariance}\\
&=\sum_{X,Y\in
 \cB\times\G\times\spin}W(\rho\bx\s,X)\overline{W(\eta\by\tau,
 Y)}\cdot\frac{\Tr(e^{-\beta F H_0
 F^*}T(\tilde{\psi}_{X}^*(x)\tilde{\psi}_{Y}(y)))}{\Tr e^{-\beta F H_0
 F^*}}.\notag
\end{align}
Since $FH_0F^*$ is diagonalized with respect to the band index, the
 2-point function $\Tr(e^{-\beta F H_0
 F^*}T(\tilde{\psi}_{X}^*(x)\tilde{\psi}_{Y}(y)))/\Tr e^{-\beta F H_0
 F^*}$ can be computed by a standard procedure (see, e.g.,
 \cite[\mbox{Lemma B.10}]{K1}). The result is that 
\begin{align}
&\frac{\Tr(e^{-\beta F H_0
 F^*}T(\tilde{\psi}_{\rho'\bx'\s'}^*(x)\tilde{\psi}_{\eta'\by'\tau'}(y)))}{\Tr e^{-\beta F H_0
 F^*}}\label{eq_time_characterization_covariance}\\
&=\frac{\delta_{\rho',\eta'}\delta_{\s',\tau'}}{L^d}\sum_{\bk\in\G^*}e^{-i\<\bx'-\by',\bk\>}e^{(x-y)\alpha_{\rho'}(\bk)}\left(\frac{1_{x\ge
 y}}{1+e^{\beta \alpha_{\rho'}(\bk)}}- \frac{1_{x<
 y}}{1+e^{-\beta
 \alpha_{\rho'}(\bk)}}\right),\notag\\
&(\forall (\rho',\bx',\s'),(\eta',\by',\tau')\in
 \cB\times\G\times\spin).\notag
\end{align}
By substituting \eqref{eq_time_characterization_covariance} into
 \eqref{eq_pre_characterization_covariance},
\begin{align}
&C(\rho\bx\s x,\eta\by\tau
 y)=\frac{\delta_{\s,\tau}}{L^d}\sum_{\bk\in\G^*}e^{-i\<\bx-\by,\bk\>}\sum_{\g\in
 \cB}\overline{U(\bk)(\rho,\g)}U(\bk)(\eta,\g)\label{eq_time_characterization_covariance_full}\\
&\qquad\qquad\qquad\qquad\cdot e^{(x-y)\alpha_{\g}(\bk)}\left(\frac{1_{x\ge y}}{1+e^{\beta
 \alpha_{\g}(\bk)}}-\frac{1_{x< y}}{1+e^{-\beta
 \alpha_{\g}(\bk)}}\right),\notag\\
&(\forall (\rho,\bx,\s,x),(\eta,\by,\tau, y)\in \cB\times \G\times
 \spin \times [0,\beta)).\notag
\end{align}
Since $h\in (2/\beta)\N$, we can apply \cite[\mbox{Lemma C.3}]{K1}
 to obtain that for any $x,y\in [0,\beta)_h$,
\begin{align}
e^{(x-y)\alpha_{\g}(\bk)}\left(\frac{1_{x\ge y}}{1+e^{\beta
 \alpha_{\g}(\bk)}}-\frac{1_{x< y}}{1+e^{-\beta
 \alpha_{\g}(\bk)}}\right)=\frac{1}{\beta}\sum_{\o\in \cM_h}\frac{e^{i(x-y)\o}}{h(1-e^{-i\frac{\o}{h}+\frac{1}{h}\alpha_{\g}(\bk)})}.
\label{eq_application_time_discretization_lemma}
\end{align}
By combining \eqref{eq_application_time_discretization_lemma} with \eqref{eq_time_characterization_covariance_full} and using
 \eqref{eq_diagonalization_kinetic} we can derive \eqref{eq_characterization_covariance}.
\end{proof}

\subsection{The Grassmann Gaussian integral
  formulation}\label{subsec_grassmann_gaussian_integral}
Here we formulate $\log(\Tr e^{-\beta H}/\Tr e^{-\beta H_0})$ into the
Grassmann Gaussian integral with the covariance $C$. Let us introduce a
counterpart of the interaction $V$ in the Grassmann algebra $\bigwedge
\cV$.
\begin{align}\label{eq_interaction_counterpart}
&V(\psi)\\
&:=\frac{1}{h}\sum_{(\rho,\bx,x)\in\cB\times \G\times [0,\beta)_h}U_{\rho}\Bigg(\opsi_{\rho\bx\ua x}\opsi_{\rho\bx\da x}\psi_{\rho\bx\da x}\psi_{\rho\bx\ua x}-\frac{1}{2}\sum_{\s\in\spin}\opsi_{\rho\bx\s
 x}\psi_{\rho\bx\s x}\Bigg).\notag
\end{align}
From now we simply write $\bU$ in place of $(U_1,U_2,\cdots,U_b)$.
\begin{lemma}\label{lem_formulation_covariance}
The following statements hold.
\begin{enumerate}
\item\label{item_real_part_positive}
For any $U_{max}\in\R_{>0}$ there exists $h_0\in \R_{>0}$ such that 
\begin{align*}
&\Re\int e^{-V(\psi)}d\mu_{C}(\psi)>0,\\
&(\forall \bU\in\R^b\text{ with }|U_{\rho}|\le U_{max}\ (\forall
 \rho \in \cB),\ \forall h\in (2/\beta)\N\text{ with }h\ge h_0).  
\end{align*}
\item\label{item_logarithm_uniform_convergence}
For any $U_{max}\in\R_{>0}$,
\begin{align*}
&\lim_{h\to \infty\atop h\in (2/\beta)\N}\sup_{\bU\in\R^b\text{ with }\atop
 |U_{\rho}|\le U_{max}(\forall \rho\in \cB)}\\
&\quad\cdot\left|\log\left(\frac{\Tr e^{-\beta H}}{\Tr e^{-\beta H_0}}\right)
-\log\left(\int e^{-V(\psi)}d\mu_{C}(\psi)\right)\right|=0.
\end{align*}
\end{enumerate}
\end{lemma}
\begin{proof}
These claims can be proved in the same way as in the proof of \cite[\mbox{Lemma
 3.4}]{K2}, \cite[\mbox{Appendix B}]{K3}. We outline the proof for self-containedness. We can rewrite the interacting
 part $V$ of the Hamiltonian $H$ as follows. 
\begin{align*}
V=&\sum_{(\rho,\bx,\s),(\eta,\by,\tau)\atop\in\cB\times\G\times\spin}w_1(\rho\bx\s,\eta\by\tau)\psi_{\rho\bx\s}^*\psi_{\eta\by\tau}\\
&+\prod_{j=1}^2\Bigg(\sum_{(\rho_j,\bx_j,\s_j),(\eta_j,\by_j,\tau_j)\atop\in\cB\times\G\times\spin}\Bigg)w_2(\rho_1\bx_1\s_1,\rho_2\bx_2\s_2,\eta_1\by_1\tau_1,\eta_2\by_2\tau_2)\\
&\quad\cdot\psi_{\rho_1\bx_1\s_1}^*\psi_{\rho_2\bx_2\s_2}^*\psi_{\eta_2\by_2\tau_2}\psi_{\eta_1\by_1\tau_1},
\end{align*}
where
\begin{align*}
&w_1(\rho\bx\s,\eta\by\tau):=-\frac{1_{(\rho,\bx,\s)=(\eta,\by,\tau)}}{2}U_{\rho},\\
&w_2(\rho_1\bx_1\s_1,\rho_2\bx_2\s_2,\eta_1\by_1\tau_1,\eta_2\by_2\tau_2)\\
&:=1_{(\rho_1,\bx_1)=(\rho_2,\bx_2)=(\eta_1,\by_1)=(\eta_2,\by_2)}1_{(\s_1,\s_2,\tau_1,\tau_2)=(\ua,\da,\ua,\da)}U_{\rho_1}.
\end{align*}
By repeating the same argument as in \cite[\mbox{Proposition 3.2}]{K2} we can
 derive the following series.
\begin{align}
\frac{\Tr e^{-\beta H}}{\Tr e^{-\beta
 H_0}}=1+&\sum_{m=1}^{\infty}\frac{(-1)^m}{m!}\prod_{k=1}^m
\Bigg(\sum_{l_k=1}^2\sum_{\bX_k,\bY_k\atop \in (\cB\times \G\times
 \spin)^{l_k}}\int_0^{\beta}dx_kw_{l_k}(\bX_k,\bY_k)\Bigg)\label{eq_continuous_series}\\
 &\quad\cdot \det(C(X_p,Y_q))_{1\le p,q\le
 \sum_{k=1}^ml_k},\notag
\end{align}
where the variables are defined by the following rule.
\begin{align}
&\bX_k:=((\rho_{k,1},\bx_{k,1},\s_{k,1}),\cdots,(\rho_{k,l_k},\bx_{k,l_k},\s_{k,l_k})),\label{eq_variables_relation}\\
&\bY_k:=((\eta_{k,1},\by_{k,1},\tau_{k,1}),\cdots,(\eta_{k,l_k},\by_{k,l_k},\tau_{k,l_k})),\notag\\
&X_p:=(\rho_{u+1,v},\bx_{u+1,v},\s_{u+1,v},x_{u+1}),\notag\\
&Y_p:=(\eta_{u+1,v},\by_{u+1,v},\tau_{u+1,v},x_{u+1}),\notag
\end{align}
for $p=\sum_{k=1}^ul_k+v$, $u\in \{0,1,\cdots,m-1\}$, $v\in
 \{1,\cdots,l_{u+1}\}$.

We define the function $\bU\mapsto P(\bU):\C^b\to \C$ by the right-hand
 side of \eqref{eq_continuous_series}. By replacing the integral
 $\int_0^{\beta}dx_k$ by the Riemann sum $\frac{1}{h}\sum_{x\in
 [0,\beta)_h}$ we introduce the discrete analogue $P_h$ of $P$ as follows.
\begin{align*}
P_h(\bU):=1+&\sum_{m=1}^{N/2}\frac{(-1)^m}{m!}\prod_{k=1}^m
\Bigg(\sum_{l_k=1}^2\sum_{\bX_k,\bY_k\atop \in (\cB\times \G\times
 \spin)^{l_k}}\frac{1}{h}\sum_{x_k\in [0,\beta)_h}w_{l_k}(\bX_k,\bY_k)\Bigg)\notag\\
 &\quad\cdot \det(C(X_p,Y_q))_{1\le p,q\le
 \sum_{k=1}^ml_k},
\end{align*}
where the variables are defined by the rule
 \eqref{eq_variables_relation}.

Define the function $\tilde{C}:(\cB\times \G\times\spin\times
 [0,\beta))^2\to \C$ by
$\tilde{C}(\rho\bx\s x,\eta \by\tau y)$ $:=C(\rho\bx\s \hat{x},\eta \by
 \tau\hat{y})$ with $\hat{x},\hat{y}\in [0,\beta)_h$ satisfying
 $x\in [\hat{x},\hat{x}+1/h)$, $y\in [\hat{y},\hat{y}+1/h)$. The
 function $P_h$ can be rewritten as follows. 
\begin{align*}
P_h(\bU)=1+&\sum_{m=1}^{N/2}\frac{(-1)^m}{m!}\prod_{k=1}^m
\Bigg(\sum_{l_k=1}^2\sum_{\bX_k,\bY_k\atop\in (\cB\times \G\times
 \spin)^{l_k}}\int_0^{\beta}dx_kw_{l_k}(\bX_k,\bY_k)\Bigg)\notag\\
 &\quad\cdot \det(\tilde{C}(X_p,Y_q))_{1\le p,q\le
 \sum_{k=1}^ml_k},
\end{align*}
Then, we have for any $U_{max}\in\R_{>0}$ that
\begin{align}
&\sup_{\bU\in\C^b\text{ with }\atop |U_{\rho}|\le U_{\max}(\forall
 \rho\in\cB)}|P(\bU)-P_h(\bU)|\label{eq_pre_uniform_inequality}\\
&\le
 \sum_{m=1}^{\infty}\frac{1}{m!}\prod_{k=1}^m\Bigg(\sum_{l_k=1}^2\sum_{\bX_k,\bY_k\atop
 \in (\cB\times \G\times
 \spin)^{l_k}}\int_0^{\beta}dx_k\sup_{\bU\in\C^b\text{ with }\atop
 |U_{\rho}|\le U_{max}(\forall \rho\in\cB)}
|w_{l_k}(\bX_k,\bY_k)|\Bigg)\notag\\
&\qquad\cdot|\det(C(X_p,Y_q))_{1\le p,q\le
 \sum_{k=1}^ml_k}-\det(\tilde{C}(X_p,Y_q))_{1\le p,q\le
 \sum_{k=1}^ml_k}|.\notag
\end{align}

Now let us confirm the fact that $C$, $\tilde{C}$ have
 $\beta,L$-dependent determinant bounds. For any
 $(\rho_j,\bx_j,\s_j,x_j)$, $(\eta_j,\by_j,\tau_j,y_j)\in \cB\times
 \G\times \spin \times [0,\beta)$ $(j=1,2,\cdots,n)$ we can choose
 operators $A_j$ $(j=1,2,\cdots,2n)$ from
 $\{e^{x_jH_0}\psi_{\rho_j\bx_j\s_j}^*e^{-x_jH_0},e^{y_jH_0}\psi_{\eta_j\by_j\tau_j}e^{-y_jH_0}\}_{j=1}^n$
 so that 
$$
|\det(C(\rho_p\bx_p\s_px_p,\eta_q\by_q\tau_qy_q))_{1\le p,q\le
 n}|=\left|\frac{\Tr(e^{-\beta H_0}A_1A_2\cdots A_{2n})}{\Tr e^{-\beta H_0}}\right|.
$$ 
Let $\<\cdot,\cdot\>_{F_f}$ denote the inner product of the Fermionic
 Fock space $F_f(L^2(\cB$ $\times \G\times \spin))$ and 
$\|\cdot\|_{F_f}$ denote the norm induced by $\<\cdot,\cdot\>_{F_f}$.
 For any linear transform
 $\zeta$ on $F_f(L^2(\cB\times \G\times \spin))$ let $\|\zeta\|_{\mathfrak{B}(F_f)}$ denote its operator norm
 defined by
$$
\|\zeta\|_{\mathfrak{B}(F_f)}:=\sup_{g\in F_f(L^2(\cB\times \G\times
 \spin))\atop \text{ with }\|g\|_{F_f}=1}\|\zeta g\|_{F_f}.
$$
Since $\|e^{xH_0}\psi_{\rho\bx\s}^{(*)}e^{-xH_0}\|_{\mathfrak{B}(F_f)}\le
 e^{2\beta \|H_0\|_{\mathfrak{B}(F_f)}}$ $(\forall (\rho,\bx,\s,x)\in
 \mathcal{B}\times \G \times \spin \times [0,\beta))$, we have that 
\begin{align}
&|\det (C(X_p,Y_q))|_{1\le p,q\le n}\le D_1\cdot
 D_2^n,\label{eq_crudest_determinant_bound}\\
&(\forall X_j,Y_j \in\cB\times \G\times\spin\times [0,\beta)\ (j=1,2,\cdots,n)),\notag
\end{align}
where 
$$D_1:=\frac{2^{2bL^d}e^{\beta \|H_0\|_{\mathfrak{B}(F_f)}}}{\Tr
 e^{-\beta H_0}},\
 D_2:=e^{4\beta \|H_0\|_{\mathfrak{B}(F_f)}}.$$

By using the determinant bound \eqref{eq_crudest_determinant_bound} we obtain the inequality 
\begin{align}
&\frac{1}{m!}\prod_{k=1}^m\Bigg(\sum_{l_k=1}^2\sum_{\bX_k,\bY_k\atop\in (\cB\times \G\times
 \spin)^{l_k}}\int_0^{\beta}dx_k\sup_{\bU\in\C^b\text{ with }\atop
 |U_{\rho}|\le U_{max}(\forall \rho\in
 \cB)}|w_{l_k}(\bX_k,\bY_k)|\Bigg)\label{eq_each_term_bound_l1}\\
&\qquad\cdot|\det(C(X_p,Y_q))_{1\le p,q\le
 \sum_{k=1}^ml_k}-\det(\tilde{C}(X_p,Y_q))_{1\le
 p,q\le \sum_{k=1}^ml_k}|\notag\\
&\le \frac{2D_1}{m!}(U_{max}bL^b\beta(D_2+D_2^2))^m.\notag
\end{align}

Since $(x,y)\mapsto C(\rho\bx\s x,\eta\by\tau y)$ is continuous a.e. in
 $[0,\beta)^2$, so is $(x_1,x_2,\cdots,x_m)\mapsto \det
 (C(X_p,Y_q))_{1\le p,q\le 
 \sum_{k=1}^ml_k}$ in $[0,\beta)^m$. Thus, 
\begin{align*}
 \lim_{h\to \infty\atop h\in (2/\beta)\N}|\det(C(X_p,Y_q))_{1\le p,q\le
 \sum_{k=1}^ml_k}-\det(\tilde{C}(X_p,Y_q))_{1\le p,q\le \sum_{k=1}^ml_k}|=0,
\end{align*}
for a.e. $(x_1,x_2,\cdots,x_n)\in [0,\beta)^m$. Therefore, the dominated
 convergence theorem for $L^1([0,\beta)^m)$ ensures that
\begin{align}
&\lim_{h\to \infty\atop h\in
 (2/\beta)\N}\frac{1}{m!}\prod_{k=1}^m\Bigg(\sum_{l_k=1}^2\sum_{\bX_k,\bY_k\atop
 \in (\cB\times \G\times
 \spin)^{l_k}}\int_0^{\beta}dx_k\sup_{\bU\in\C^b\text{ with }\atop
 |U_{\rho}|\le U_{max}(\forall \rho\in \cB)}|w_{l_k}(\bX_k,\bY_k)|\Bigg)\label{eq_each_term_convergence_l1}\\
&\qquad\cdot|\det(C(X_p,Y_q))_{1\le p,q\le
 \sum_{k=1}^ml_k}-\det(\tilde{C}(X_p,Y_q))_{1\le
 p,q\le \sum_{k=1}^ml_k}|=0.\notag
\end{align}

By \eqref{eq_each_term_bound_l1} and \eqref{eq_each_term_convergence_l1}
 we can apply the dominated convergence theorem for $l^1(\N)$
to deduce from \eqref{eq_pre_uniform_inequality} that
\begin{align}
\lim_{h\to \infty\atop h\in (2/\beta)\N}\sup_{\bU\in\C^b\text{ with }\atop
 |U_{\rho}|\le U_{max}(\forall \rho\in
 \cB)}|P(\bU)-P_h(\bU)|=0.\label{eq_primitive_uniform_convergence}
\end{align}
By replacing the determinant in $P_h$ by the Grassmann Gaussian integral
 we can derive that
\begin{align}
P_h(\bU)=\int
 e^{-V(\psi)}d\mu_{C}(\psi).\label{eq_primitive_grassmann_formulation}
\end{align}
By \eqref{eq_primitive_uniform_convergence} and
 \eqref{eq_primitive_grassmann_formulation},
\begin{align}
\lim_{h\to \infty\atop h\in (2/\beta)\N}\sup_{\bU\in\R^b\text{ with }\atop
 |U_{\rho}|\le U_{max}(\forall \rho\in
 \cB)}\left|\frac{\Tr e^{-\beta H}}{\Tr e^{-\beta H_0}}-\int
 e^{-V(\psi)}d\mu_C(\psi)\right|=0,\label{eq_partition_uniform_convergence}
\end{align}
which implies
 the claim \eqref{item_real_part_positive}. The claim
 \eqref{item_logarithm_uniform_convergence} follows from the claim
 \eqref{item_real_part_positive} and \eqref{eq_partition_uniform_convergence}.
\end{proof}

\subsection{A symmetric formulation}\label{subsec_symmetric_formulation}
It is vital for the validity of the forthcoming IR integration
that any Grassmann polynomial produced by the single-scale IR integration is invariant under certain transforms.
 The Grassmann integral formulation constructed in the previous
 subsection does not satisfy one of the necessary invariant properties
 by itself if we connect it to the IR integration process. We need to modify the
 formulation into a more suitable form for the forthcoming IR
 analysis. For this purpose we introduce a few more
 covariances. Then, we propose another formulation, which will be shown to have desired
 symmetries in Section \ref{sec_IR_model}, by using the newly introduced covariances. 
For any $(\rho,\bx,\s,x)$, $(\eta,\by,\tau,y)\in I_0$ set 
\begin{align}\label{eq_covariance_negative_index}
C_{\le 0}^+(\rho\bx\s x,\eta\by\tau y):=&\frac{\delta_{\s,\tau}}{\beta
 L^d}\sum_{\bk\in \G^*}\sum_{\o\in
 \cM_h}e^{-i\<\bx-\by,\bk\>}e^{i(x-y)\o}\\
&\cdot\chi(h|1-e^{i\frac{\o}{h}}|)h^{-1}(I_b-e^{-i\frac{\o}{h}I_b+\frac{1}{h}\overline{E(\bk)}})^{-1}(\rho,\eta),\notag
\end{align}
where $\chi(\cdot):\R\to [0,1]$ is a smooth function. We
assume that the support of $\chi(\cdot)$ is contained in the interval
$[-c_{\chi},c_{\chi}]$, where $c_{\chi}\in\R_{>0}$ is a constant. In this
section we do not need more detailed information on the cut-off function $\chi(\cdot)$.
For any $(\rho,\bx,\s,x)$, $(\eta,\by,\tau,y)\in I_0$ set
\begin{equation}\label{eq_covariance_non_negative_index}
\begin{split}
C_{> 0}^+(\rho\bx\s x,\eta\by\tau y):=&\frac{\delta_{\s,\tau}}{\beta
 L^d}\sum_{\bk\in \G^*}\sum_{\o\in
 \cM_h}e^{-i\<\bx-\by,\bk\>}e^{i(x-y)\o}\\
&\cdot(1-\chi(h|1-e^{i\frac{\o}{h}}|))h^{-1}(I_b-e^{-i\frac{\o}{h}I_b+\frac{1}{h}\overline{E(\bk)}})^{-1}(\rho,\eta),
\end{split}
\end{equation}
so that $C(X,Y)=C_{\le 0}^+(X,Y)+C_{>0}^+(X,Y)$, $(\forall X,Y\in I_0)$. Define $C_{\le 0}^{\infty}:I_0^2\to
\C$ by 
\begin{align}
&C_{\le 0}^{\infty}(\rho\bx\s x,\eta\by\tau y)\label{eq_covariance_non_positive_index_infinity}\\
&:=\frac{\delta_{\s,\tau}}{\beta
 L^d}\sum_{\bk\in \G^*}\sum_{\o\in
 \cM_h}e^{-i\<\bx-\by,\bk\>}e^{i(x-y)\o}\chi(|\o|)(i\o I_b-\overline{E(\bk)})^{-1}(\rho,\eta).\notag
\end{align}
Moreover, we define the covariance
$C_{>0}^{-}:I_0^2\to\C$ as follows. For $(\rho,\bx,\s,x)$,
$(\eta,\by,\tau,y)\in I_0$, 
\begin{align}\label{eq_covariance_non_negative_index_another}
C_{> 0}^-(\rho\bx\s x,\eta\by\tau y):=&\frac{\delta_{\s,\tau}}{\beta
 L^d}\sum_{\bk\in \G^*}\sum_{\o\in
 \cM_h}e^{-i\<\bx-\by,\bk\>}e^{i(x-y)\o}\\
&\cdot(1-\chi(h|1-e^{i\frac{\o}{h}}|))h^{-1}(e^{i\frac{\o}{h}I_b-\frac{1}{h}\overline{E(\bk)}}-I_b)^{-1}(\rho,\eta).\notag
\end{align}
We can derive the following equality from the definitions.
\begin{lemma}\label{lem_relation_various_covariance}
For any $(\rho,\bx,\s,x)$, $(\eta,\by,\tau,y)\in I_0$, 
\begin{align*}
&C_{> 0}^+(\rho\bx\s x,\eta\by\tau
 y)+\frac{1_{(\rho,\bx,\s)=(\eta,\by,\tau)}}{\beta h}\sum_{\o\in
 \cM_h}e^{i(x-y)\o}\chi(h|1-e^{i\frac{\o}{h}}|)\\
&=C_{> 0}^-(\rho\bx\s x,\eta\by\tau
 y)+1_{(\rho,\bx,\s,x)=(\eta,\by,\tau,y)}.\notag
\end{align*}
\end{lemma}
Finally we define the covariances $C_{>0}^{+(h)}$, $\cI:I_0^2\to \C$ by 
\begin{align*}
&C_{>0}^{+(h)}(\rho\bx\s x,\eta \by \tau y)\\
&:=C_{>0}^+(\rho\bx\s x,\eta \by \tau y)+\frac{1_{(\rho,\bx,\s)=(\eta,\by,\tau)}}{\beta h}\sum_{\o\in
 \cM_h}e^{i(x-y)\o}\chi(h|1-e^{i\frac{\o}{h}}|),\\
&\cI(\rho\bx\s x,\eta \by \tau
 y):=1_{(\rho,\bx,\s,x)=(\eta,\by,\tau,y)}.
\end{align*}

We will make another Grassmann integral
formulation out of these covariances. In order to prove that the Grassmann integral formulation
converges to the normalized free energy density as $h\to \infty$, we must know that these covariances have suitable 
determinant bounds.
\begin{lemma}\label{lem_h_independent_determinant_bound}
There exist $(\beta,L^d,b,\chi,E)$-dependent,
 $h$-independent constants $h_0,\ c_1\in\R_{>0}$ such that
the following inequalities hold true for any $h\in (2/\beta)\N$ with $h\ge
 h_0$.
\begin{align*}
&|\det(C_o(X_i,Y_j))_{1\le i,j\le n}|\le c_1^n,\\
&|\det(C_{>0}^{+(h)}(X_i,Y_j)-C_{>0}^{+}(X_i,Y_j))_{1\le i,j\le
 n}|\le\frac{1}{h}c_1^n,\\
&|\det(C_{\le 0}^{+}(X_i,Y_j)-C_{\le 0}^{\infty}(X_i,Y_j))_{1\le i,j\le
 n}|\le\frac{1}{h}c_1^n,\\
&(\forall n\in \N,X_j,Y_j\in I_0\ (j=1,2,\cdots,n))
\end{align*}
for $C_o=C$, $C_{\le 0}^+$, $C_{> 0}^+$, $C_{\le 0}^{\infty}$, $C_{>0}^-$, $C_{>
 0}^{+(h)}$ respectively.
\end{lemma}
 \begin{proof}
The determinant bound on $C$ has been given in
  \eqref{eq_crudest_determinant_bound}. The determinant bounds on the
  other covariances can be proved by applying Gram's inequality in the
  complex Hilbert space $\cH=L^2(\cB\times\G^*\times\spin\times\cM_h)$,
  which consists of all complex-valued functions on
  $\cB\times\G^*\times\spin\times\cM_h$ and is equipped with the inner
  product $\<\cdot,\cdot\>_{\cH}$ defined by
\begin{align*}
\<f,g\>_{\cH}:=\frac{1}{\beta L^d}\sum_{K\in
 \cB\times\G^*\times\spin\times\cM_h}\overline{f(K)}g(K),\ (\forall
 f,g\in \cH).
\end{align*}
Define the vectors $f_{\rho\bx\s x}$, $g_{\rho\bx\s x}\in\cH$
  $((\rho,\bx,\s,x)\in I_0)$ by
\begin{align*}
&f_{\rho\bx\s
 x}(\eta,\bk,\tau,\o):=\delta_{\rho,\eta}\delta_{\s,\tau}e^{i\<\bx,\bk\>}e^{-ix\o}\chi(h|1-e^{i\frac{\o}{h}}|)^{\frac{1}{2}},\\
&g_{\rho\bx\s
 x}(\eta,\bk,\tau,\o):=\delta_{\s,\tau}e^{i\<\bx,\bk\>}e^{-ix\o}\chi(h|1-e^{i\frac{\o}{h}}|)^{\frac{1}{2}}\\
&\qquad\qquad\qquad\qquad\cdot h^{-1}(I_b-e^{-i\frac{\o}{h}I_b+\frac{1}{h}\overline{E(\bk)}})^{-1}(\eta,\rho),\\
&(\forall (\eta,\bk,\tau,\o)\in \cB\times\G^*\times\spin\times\cM_h).
\end{align*}
It follows that $C_{\le
  0}^+(X,Y)=\<f_X,g_Y\>_{\cH}$, $(\forall X,Y\in I_0)$. Moreover, by using the inequality  $h|1-e^{i\frac{\o}{h}}|\ge (2/\pi)|\o|$ $(\forall \o\in \cM_h)$ we have
\begin{align*}
\|f_X\|^2_{\cH}\le \frac{1}{\beta}\sum_{\o\in
 \cM_h}\chi(h|1-e^{i\frac{\o}{h}}|)\le \frac{1}{\beta}\sum_{\o\in \cM}1_{|\o|\le \frac{\pi}{2}c_{\chi}},
\end{align*}
where $\|\cdot\|_{\cH}$ denotes the norm of $\cH$ induced by $\<\cdot,\cdot\>_{\cH}$.
Since 
\begin{align}
&h^{-1}(I_b-e^{-i\frac{\o}{h}I_b+\frac{1}{h}\overline{E(\bk)}})^{-1}\label{eq_h_dependent_integrant_expansion}\\ 
&= (i\o
 I_b-\overline{E(\bk)})^{-1}\sum_{m=0}^{\infty}\left(\sum_{n=2}^{\infty}\frac{(-1)^n}{n!}\frac{1}{h^{n-1}}(i\o I_b-\overline{E(\bk)})^{n-1}\right)^m,\notag\\
&\|h^{-1}(I_b-e^{-i\frac{\o}{h}I_b+\frac{1}{h}\overline{E(\bk)}})^{-1}\|_{b\times
 b}\label{eq_direct_bound_symmetric_formulation}\\
&\le \|(i\o I_b-\overline{E(\bk)})^{-1}\|_{b\times
 b}\sum_{m=0}^{\infty}\left(\frac{1}{h}\sum_{n=2}^{\infty}\frac{1}{n!}\|i\o
 I_b-\overline{E(\bk)}\|_{b\times b}^{n-1}\right)^m\notag\\
&\le
 \frac{\beta}{\pi}\sum_{m=0}^{\infty}\left(\frac{1}{h}e^{\frac{\pi}{2}c_{\chi}+\sup_{\bk\in\R^d}\|E(\bk)\|_{b\times
 b}}\right)^m\notag\\
&\le \beta,\ (\forall \bk\in\G^*,\o\in \cM\text{ with }|\o|\le (\pi/2)c_{\chi}),\notag
\end{align}
if $h$ is large enough. Therefore,
\begin{align*}
\|g_X\|_{\cH}^2&\le \frac{1}{\beta L^d}\sum_{(\bk,\o)\in\G^*\times
 \cM}1_{|\o|\le
 \frac{\pi}{2}c_{\chi}}\|h^{-1}(I_b-e^{-i\frac{\o}{h}I_b+\frac{1}{h}\overline{E(\bk)}})^{-1}\|_{b\times
 b}^2\\
&\le \beta \sum_{\o\in\cM}1_{|\o|\le \frac{\pi}{2}c_{\chi}}.
\end{align*}
We can, thus, apply Gram's inequality to deduce that
\begin{align}\label{eq_beta_dependent_determinant_bound}
|\det(C_{\le 0}^+(X_i,Y_j))_{1\le i,j\le n}|\le
 \prod_{j=1}^n\|f_{X_j}\|_{\cH}\|g_{Y_j}\|_{\cH}\le
 \left(\sum_{\o\in\cM}1_{|\o|\le \frac{\pi}{2}c_{\chi}}\right)^n.
\end{align}

Take any $X_1,\cdots,X_n$, $Y_1,\cdots,Y_n\in I_0$. Define
  $C_{(n)}$, $C_{>0,(n)}^+$, $C_{\le 0,(n)}^+\in \Mat(n,\C)$ by 
\begin{align*}
&C_{(n)}(i,j):=C(X_i,Y_j),\ C^+_{>0,(n)}(i,j):=C_{>0}^+(X_i,Y_j),\\
& C^+_{\le
 0,(n)}(i,j):=C_{\le 0}^+(X_i,Y_j),\ (\forall i,j\in \{1,2,\cdots,n\}).
\end{align*}
Since 
$$C^+_{>0,(n)}=C_{(n)}-C_{\le
  0,(n)}^+=(C_{(n)},I_n)\left(\begin{array}{c}I_n\\ -C_{\le
				0,(n)}^+\end{array}\right),$$   
the Cauchy-Binet formula gives that 
\begin{align}
\label{eq_application_cauchy_binet}\det(C_{>0,(n)}^+)=\sum_{\phi:\{1,2,\cdots,n\}\to
 \{1,2,\cdots,2n\}\atop\text{ with }
 \phi(1)<\phi(2)<\cdots<\phi(n)}&\det((C_{(n)},I_n)(i,\phi(j)))_{1\le
 i,j\le n}\\
&\cdot\det\left( \left(\begin{array}{c}I_n\\ -C_{\le
				0,(n)}^+\end{array}\right)(\phi(i),j)\right)_{1\le i,j\le n}.\notag
\end{align}
By using \eqref{eq_crudest_determinant_bound},
  \eqref{eq_beta_dependent_determinant_bound} 
and admitting that $\phi(0)=n$, $\phi(n+1)=n+1$
we can derive from
  \eqref{eq_application_cauchy_binet} that
\begin{align*}
&|\det(C_{>0,(n)}^{+})|\\
&\le \sum_{m=0}^n\sum_{\phi:\{1,2,\cdots,n\}\to
 \{1,2,\cdots,2n\}\atop\text{ with }\phi(1)<\phi(2)<\cdots<\phi(n)}1_{\phi(m)\le
 n<\phi(m+1)}\\
&\quad\cdot\Bigg(1_{m=0}+1_{m>0}\sup_{X_j',Y_j'\in I_0\atop
 (j=1,2,\cdots,m)}|\det(C(X_i',Y_j'))_{1\le i,j\le m}|\Bigg)\\
&\quad\cdot \Bigg(1_{m=n}+1_{m<n}\sup_{X_j',Y_j'\in I_0\atop
 (j=1,2,\cdots,n-m)}|\det(C_{\le 0}^+(X_i',Y_j'))_{1\le i,j\le n-m}|\Bigg)\\
&\le \sum_{m=0}^n\left(\begin{array}{c}n\\m\end{array}\right)
\left(\begin{array}{c}n\\n-m\end{array}\right)\max\{1,D_1\}D_2^m\Bigg(\sum_{\o\in\cM}1_{|\o|\le
 \frac{\pi}{2}c_{\chi}}\Bigg)^{n-m}\\
&\le \max\{1,D_1\}\sum_{m=0}^n\left(\begin{array}{c}2n\\2m\end{array}\right)D_2^m\Bigg(\sum_{\o\in\cM}1_{|\o|\le
 \frac{\pi}{2}c_{\chi}}\Bigg)^{n-m}\\
&\le \max\{1,D_1\}\Bigg(D_2^{\frac{1}{2}}+\Bigg(\sum_{\o\in\cM}1_{|\o|\le
 \frac{\pi}{2}c_{\chi}}\Bigg)^{\frac{1}{2}}\Bigg)^{2n}.
\end{align*}
The determinant bounds on the covariances $C_{\le 0}^{\infty}$, $C_{>
  0}^{-}$, $C_{>
  0}^{+(h)}$ can be derived in the same way as above.

Since $(C_{>0}^{+(h)}-C_{>0}^+)(X,Y)=\<\frac{1}{h}f_X,f_Y\>_{\cH}$
  $(\forall X,Y\in I_0)$, we have
\begin{align*} 
 |\det ((C_{>0}^{+(h)}-C_{>0}^+)(X_i,Y_j))_{1\le i,j\le n}|&\le
 \frac{1}{h}\prod_{j=1}^n\|f_{X_j}\|_{\cH}\|f_{Y_j}\|_{\cH}\\
&\le \frac{1}{h}\left(\frac{1}{\beta}\sum_{\o\in \cM}1_{|\o|\le
 \frac{\pi}{2}c_{\chi}}\right)^n.
\end{align*}

To prove the determinant bound on $C_{\le 0}^{+}-C_{\le 0}^{\infty}$,
  let us define the vectors $f_{\rho\bx\s x}',g_{\rho\bx\s x}'\in \cH$ ($(\rho,\bx,\s, x)\in I_0$) by
\begin{align*}
&f'_{\rho\bx\s
 x}(\eta,\bk,\tau,\o):=\delta_{\rho,\eta}\delta_{\s,\tau}e^{i\<\bx,\bk\>}e^{-ix\o}1_{|\o|\le\frac{\pi}{2}c_{\chi}},\\
&g'_{\rho\bx\s
x}(\eta,\bk,\tau,\o)\\
&:=\delta_{\s,\tau}e^{i\<\bx,\bk\>}e^{-ix\o}1_{|\o|\le\frac{\pi}{2}c_{\chi}}\big(\chi(h|1-e^{i\frac{\o}{h}}|)h^{-1}(I_b-e^{-i\frac{\o}{h}I_b+\frac{1}{h}\overline{E(\bk)}})^{-1}(\eta,\rho)\\
&\quad\qquad\qquad\qquad\qquad\qquad -\chi(|\o|)(i\o
 I_b-\overline{E(\bk)})^{-1}(\eta,\rho)\big).
\end{align*}
Since
  $\chi(h|1-e^{i\frac{\o}{h}}|)=1_{|\o|\le\frac{\pi}{2}c_{\chi}}\chi(h|1-e^{i\frac{\o}{h}}|)$,
  $\chi(|\o|)=1_{|\o|\le \frac{\pi}{2}c_{\chi}}\chi(|\o|)$ $(\forall \o\in
  \cM_h)$, $(C_{\le 0}^+-C_{\le 0}^{\infty})(X,Y)=\<f_X',g_Y'\>_{\cH}$
  $(\forall X,Y\in I_0)$. Note that 
$$\|f_X'\|_{\cH}^2\le \frac{1}{\beta}\sum_{\o\in
  \cM}1_{|\o|\le\frac{\pi}{2}c_{\chi}}.$$
By assuming that
  $h$ is sufficiently large and using
  \eqref{eq_h_dependent_integrant_expansion},
  \eqref{eq_direct_bound_symmetric_formulation} we deduce that for any
  $\o\in\cM_h$, $\bk\in\G^*$,
\begin{align*}
&\|\chi(h|1-e^{i\frac{\o}{h}}|)h^{-1}(I_b-e^{-i\frac{\o}{h}I_b+\frac{1}{h}\overline{E(\bk)}})^{-1}-\chi(|\o|)(i\o
 I_b-\overline{E(\bk)})^{-1}\|_{b\times b}\\
&\le
 |\chi(h|1-e^{i\frac{\o}{h}}|)-\chi(|\o|)|\|h^{-1}(I_b-e^{-i\frac{\o}{h}I_b+\frac{1}{h}\overline{E(\bk)}})^{-1}\|_{b\times
 b}\\
&\quad +\chi(|\o|)\|h^{-1}(I_b-e^{-i\frac{\o}{h}I_b+\frac{1}{h}\overline{E(\bk)}})^{-1}-(i\o
 I_b-\overline{E(\bk)})^{-1}\|_{b\times b}\\
& \le
 \big|h|1-e^{i\frac{\o}{h}}|-|\o|\big|\sup_{x\in\R}|\chi'(x)|1_{|\o|\le\frac{\pi}{2}c_{\chi}}\beta\\
&\quad + 1_{|\o|\le
 \frac{\pi}{2}c_{\chi}}\|(i\o
 I_b-\overline{E(\bk)})^{-1}\|_{b\times
 b}\sum_{m=1}^{\infty}\left(\frac{1}{h}\sum_{n=2}^{\infty}\frac{1}{n!}\|i\o
 I_b-\overline{E(\bk)}\|_{b\times b}^{n-1}\right)^m\\
& \le
 \frac{\beta}{h}e^{\frac{\pi}{2}c_{\chi}}\sup_{x\in\R}|\chi'(x)|
 +\frac{\beta}{h}e^{\frac{\pi}{2}c_{\chi}+\sup_{\bk\in\R^d}\|E(\bk)\|_{b\times
 b}},
\end{align*}
which implies that 
\begin{align*}
\|g_X'\|_{\cH}^2\le \left(\frac{\beta}{h}e^{\frac{\pi}{2}c_{\chi}+\sup_{\bk\in\R^d}\|E(\bk)\|_{b\times
 b}}\big(1+\sup_{x\in\R}|\chi'(x)|\big)\right)^2\frac{1}{\beta}\sum_{\o\in\cM}1_{|\o|\le\frac{\pi}{2}c_{\chi}}.
\end{align*}
Again by Gram's inequality,
\begin{align*}
&|\det(C_{\le 0}^+-C_{\le 0}^{\infty})(X_i,Y_j))_{1\le i,j\le n}|\le
 \prod_{j=1}^n \|f_{X_j}'\|_{\cH}\|g_{X_j}'\|_{\cH}\\
&\le \frac{1}{h}\left(e^{\frac{\pi}{2}c_{\chi}+\sup_{\bk\in\R^d}\|E(\bk)\|_{b\times
 b}}\left(1+\sup_{x\in\R}|\chi'(x)|\right)\sum_{\o\in\cM}1_{|\o|\le\frac{\pi}{2}c_{\chi}}\right)^n.
\end{align*}
The claimed determinant bounds have been obtained.
\end{proof}

To continue our analysis, we introduce an $L^1$-norm on functions on
$I^m$. 
For $m\in \{1,2,\cdots,N\}$ and a function
$f_m:I^m\to \C$, let $\|f_m\|_{L^1}$ be defined by 
$$\|f_m\|_{L^1}:=\left(\frac{1}{h}\right)^m\sum_{\bX\in I^m}|f_m(\bX)|.$$
To simplify arguments, we let $\|f_0\|_{L^1}$ denote $|f_0|$ for
$f_0\in\C$ as well. Necessary basic estimations with this norm are
separately prepared in Appendix \ref{app_grassmann_L1_theory}. 

Let us introduce the Grassmann polynomials $V^+(\psi)$,
$V^-(\psi)$,  $S^+(\psi)$,
$S^-(\psi)$, $S^0(\psi)$ $(\in \bigwedge \cV)$ as follows. 
\begin{align}
&V^+(\psi):=V(\psi),\quad
 V^-(\psi):=V(\psi)+\frac{1}{h}\sum_{(\rho,\bx,\s,x)\in
 I_0}U_{\rho}\opsi_{\rho\bx\s x}\psi_{\rho\bx\s
 x},\label{eq_signed_interaction_polynomial}\\
&S^{\delta}(\psi):=\int
 e^{-V^{\delta}(\psi+\psi^1)}d\mu_{C_{>0}^{\delta}}(\psi^1),\ (\delta
 \in \{+,-\}),\notag\\
&S^{0}(\psi):=\int
 e^{-V^{+}(\psi+\psi^1)}d\mu_{C_{>0}^{+(h)}}(\psi^1),\notag
\end{align}
where $V(\psi)$ is the Grassmann polynomial defined in
\eqref{eq_interaction_counterpart}. In the following `$c_1$' denotes the
$h$-independent constant appearing in Lemma \ref{lem_h_independent_determinant_bound} and the parameter $h\in
(2/\beta)\N$ is assumed to be larger than $h_0$ appearing in the same
lemma. Also, let us assume that
$\bU\in\C^b$ satisfies $|U_{\rho}|\le U_{max}$ $(\forall \rho\in \cB)$
for some $U_{max}(\in\R_{\ge 0})$.  

\begin{lemma}\label{lem_exponential_appl_bound}
The following inequalities hold.
\begin{enumerate}
\item\label{item_exponential_appl_0th}
$$
|S_0^{\delta}-1|\le e^{b\beta L^dU_{max}(c_1+c_1^2)}-1,\ (\forall\delta \in \{+,-,0\}).$$
\item\label{item_exponential_appl_sum}
For any $\alpha \in\R_{\ge 0}$,
$$
\sum_{m=0}^N\alpha^mc_1^{\frac{m}{2}}\|S_m^{\delta}\|_{L^1}\le e^{b\beta
     L^d U_{max}((\alpha+1)^2c_1+(\alpha+1)^4c_1^2)},\ (\forall\delta\in \{+,-,0\}).
$$
\item\label{item_exponential_appl_difference}
For any $\alpha \in\R_{\ge 0}$,
$$
\sum_{m=0}^N\alpha^mc_1^{\frac{m}{2}}\|S_m^{+}-S_m^0\|_{L^1}\le \frac{1}{h}(e^{b\beta L^d U_{max}((\alpha+2)^2c_1+(\alpha+2)^4c_1^2)}-1).
$$
\end{enumerate}
\end{lemma}
\begin{proof}
The anti-symmetric kernels $V_m^{\delta}(\cdot):I^m\to\C$ $(m\in
 \{2,4\})$ of $V^{\delta}(\psi)$ are characterized
 as follows. For $\delta \in\{+,-\}$,
\begin{align}
&V_2^{\delta}((\rho_1,\bx_1,\s_1,x_1,\theta_1),(\rho_2,\bx_2,\s_2,x_2,\theta_2))\label{eq_anti_symmetric_kernel_initial_V}\\
& =-\frac{\delta
 h}{4}U_{\rho_1}1_{(\rho_1,\bx_1,\s_1,x_1)=(\rho_2,\bx_2,\s_2,x_2)}(1_{(\theta_1,\theta_2)=(1,-1)}-1_{(\theta_1,\theta_2)=(-1,1)}),\notag\\
&V_4^{\delta}((\rho_1,\bx_1,\s_1,x_1,\theta_1),(\rho_2,\bx_2,\s_2,x_2,\theta_2),\notag\\
&\qquad(\rho_3,\bx_3,\s_3,x_3,\theta_3),(\rho_4,\bx_4,\s_4,x_4,\theta_4))\notag\\&=\frac{h^3}{4!}U_{\rho_1}1_{(\rho_1,\bx_1,x_1)=(\rho_2,\bx_2,x_2)=(\rho_3,\bx_3,x_3)=(\rho_4,\bx_4,x_4)}\notag\\
&\qquad\cdot\sum_{\zeta\in\S_4}\sgn(\zeta)1_{((\s_{\zeta(1)},\theta_{\zeta(1)}),(\s_{\zeta(2)},\theta_{\zeta(2)}),(\s_{\zeta(3)},\theta_{\zeta(3)}),(\s_{\zeta(4)},\theta_{\zeta(4)})
 )=((\ua,1),(\da,1),(\da,-1),(\ua,-1))}.\notag
\end{align}
From this we can see that 
$$
\|V_m^{\delta}\|_{L^1}\le b\beta L^dU_{max},\ (\forall m\in \{2,4\},\delta \in\{+,-\}).
$$
By these inequalities and Lemma
 \ref{lem_h_independent_determinant_bound} we can readily apply Lemma \ref{lem_exponential_general_bound}
 \eqref{item_exponential_general_0th},\eqref{item_exponential_general_sum},\eqref{item_exponential_general_covariance_difference} proved in Appendix \ref{app_grassmann_L1_theory} to verify the statements \eqref{item_exponential_appl_0th},\eqref{item_exponential_appl_sum},\eqref{item_exponential_appl_difference} respectively. 
\end{proof}

\begin{lemma}\label{lem_exponential_appl_bound_another}
For any $\alpha\in\R_{\ge 0}$,
\begin{align*}
&\sum_{m=0}^N\alpha^mc_1^{\frac{m}{2}}\|S_m^--S_m^0\|_{L^1}\\
&\le\frac{1}{\beta
 h}(b\beta L^dU_{max}((\alpha+1)^2c_1+(\alpha+1)^4c_1^2))^2e^{b\beta L^dU_{max}((\alpha+1)^2c_1+(\alpha+1)^4c_1^2)}.
\end{align*}
\end{lemma}

\begin{proof}
For any $(\rho,\bx,x)\in\cB\times\G\times [0,\beta)_h$ set
\begin{align}\label{eq_truncated_interaction_polynomial}
&V_{\rho\bx x}^+(\psi):=\opsi_{\rho\bx \ua x}\opsi_{\rho\bx \da
 x}\psi_{\rho\bx \da x}\psi_{\rho\bx \ua x}-\frac{1}{2}\sum_{\s\in
 \spin}\opsi_{\rho\bx\s x}\psi_{\rho\bx\s x},\\
&V_{\rho\bx x}^-(\psi):=\opsi_{\rho\bx \ua x}\opsi_{\rho\bx \da
 x}\psi_{\rho\bx \da x}\psi_{\rho\bx \ua x}+\frac{1}{2}\sum_{\s\in
 \spin}\opsi_{\rho\bx\s x}\psi_{\rho\bx\s x}.\notag
\end{align}
Recall the general property of Grassmann Gaussian integral that
for any covariances $A,B:I_0^2\to \C$ and $f\in \bigwedge \cV$, 
\begin{align}\label{eq_covariance_decomposition_integral_double}
\int f(\psi+\psi^1)d\mu_{A+B}(\psi^1)=\int\int f(\psi+\psi^1+\psi^2)d\mu_{A}(\psi^2)d\mu_{B}(\psi^1)
\end{align}
(see \cite[\mbox{Proposition I.21}]{FKT2}). Assume that
 $x_1,x_2,\cdots,x_m\in[0,\beta)_h$ satisfy $x_i\neq x_j$ if $i\neq
 j$. Using Lemma \ref{lem_relation_various_covariance} and
 \eqref{eq_covariance_decomposition_integral_double}, we observe that
\begin{align}
&\int\prod_{j=1}^mV_{\rho_j\bx_jx_j}^+(\psi+\psi^1)d\mu_{C_{>0}^{+(h)}}(\psi^1)\label{eq_polynomialwise_equality}\\
&=\int\int\prod_{j=1}^mV_{\rho_j\bx_jx_j}^+(\psi+\psi^1+\psi^2)d\mu_{\cI}(\psi^2)d\mu_{C_{>0}^{-}}(\psi^1)\notag\\
&= \int\int\prod_{j=1}^m\big(V_{\rho_j\bx_jx_j}^+(\psi+\psi^1)+V_{\rho_j\bx_jx_j}^+(\psi^2)\notag\\
&\quad +(\opsi+\opsi^1)_{\rho_j\bx_j\ua x_j}(\psi+\psi^1)_{\rho_j\bx_j\ua x_j}\opsi_{\rho_j\bx_j\da x_j}^2\psi_{\rho_j\bx_j\da x_j}^2\notag\\
&\quad +(\opsi+\opsi^1)_{\rho_j\bx_j\da x_j}(\psi+\psi^1)_{\rho_j\bx_j\da
 x_j}\opsi_{\rho_j\bx_j\ua x_j}^2\psi_{\rho_j\bx_j\ua
 x_j}^2\big)d\mu_{\cI}(\psi^2)d\mu_{C_{>0}^-}(\psi^1)\notag\\
&=\int\prod_{j=1}^m\Bigg(V_{\rho_j\bx_j
 x_j}^+(\psi+\psi^1)+\sum_{\s\in\spin}(\opsi+\opsi^1)_{\rho_j\bx_j\s
 x_j}(\psi+\psi^1)_{\rho_j\bx_j\s
 x_j}\Bigg)\notag\\
&\quad\cdot d\mu_{C_{>0}^-}(\psi^1)\notag\\
&= \int\prod_{j=1}^mV_{\rho_j\bx_j
 x_j}^-(\psi+\psi^1)d\mu_{C_{>0}^-}(\psi^1).\notag
\end{align}

Let us define $f_{V^{+}}(\psi),f_{V^{-}}(\psi)\in \bigwedge\cV$ by
\begin{align*}
f_{V^{\delta}}(\psi):=1+&\sum_{n=1}^N\frac{(-1)^n}{n!}\prod_{j=1}^n\left(\frac{1}{h}\sum_{(\rho_j,\bx_j,x_j)\in\cB\times\G\times
 [0,\beta)_h}U_{\rho_j}\right)\\
&\cdot 1_{\forall i\forall j(i\neq j\to x_i\neq x_j)}
\prod_{j=1}^nV^{\delta}_{\rho_j\bx_j x_j}(\psi),\ (\delta\in
 \{+,-\}).
\end{align*}
The equality \eqref{eq_polynomialwise_equality} implies that 
$$
\int f_{V^+}(\psi+\psi^1)d\mu_{C_{>0}^{+(h)}}(\psi^1)=\int f_{V^-}(\psi+\psi^1)d\mu_{C_{>0}^-}(\psi^1),
$$
and thus,
\begin{align*}
S^0(\psi)-S^-(\psi)=&\int(e^{-V^+(\psi+\psi^1)}-f_{V^+}(\psi+\psi^1))d\mu_{C_{>0}^{+(h)}}(\psi^1)\\
&-\int(e^{-V^-(\psi+\psi^1)}-f_{V^-}(\psi+\psi^1))d\mu_{C_{>0}^{-}}(\psi^1).
\end{align*}

Set 
$$
\tilde{S}^0(\psi):=\int(e^{-V^+(\psi+\psi^1)}-f_{V^+}(\psi+\psi^1))d\mu_{C_{>0}^{+(h)}}(\psi^1).
$$
We can see that
\begin{align*}
&\tilde{S}^0_m(\psi)\\
&=\sum_{n=2}^N\frac{(-1)^n}{n!}\prod_{j=1}^n\left(\frac{1}{h}\sum_{(\rho_j,\bx_j,x_j)\in\cB\times\G\times
 [0,\beta)_h}U_{\rho_j}\right)1_{\exists i \exists j(i\neq j\land x_i=x_j)}\\
&\quad\cdot 
\cP_m\int\prod_{l=1}^nV^{+}_{\rho_l\bx_l
 x_l}(\psi+\psi^1)d\mu_{C_{>0}^{+(h)}}(\psi^1)\\
&=\sum_{n=2}^N\frac{(-1)^n}{n!}\prod_{j=1}^n\left(\frac{1}{h}\sum_{(\rho_j,\bx_j,x_j)\in\cB\times\G\times
 [0,\beta)_h}U_{\rho_j}\right)1_{\exists i \exists j(i\neq j\land x_i=x_j)}\\
&\quad\cdot
 \prod_{l=1}^n\left(\sum_{m_l\in\{2,4\}}\left(\frac{1}{h}\right)^{m_l}\sum_{k_l=0}^{m_l}\left(\begin{array}{c}m_l\\k_l\end{array}\right)\sum_{\bX_l\in
 I^{m_l-k_l}}\sum_{\bY_l\in I^{k_l}}V_{\rho_l\bx_l
 x_l,m_l}^+(\bX_l,\bY_l)\right)\\
&\quad\cdot \eps_{\pm}1_{\sum_{l=1}^nk_l=m}\int \psi_{\bX_1}^1
 \psi_{\bX_2}^1\cdots  \psi_{\bX_n}^1d\mu_{C_{>0}^{+(h)}}(\psi^1)\psi_{\bY_1}
 \psi_{\bY_2}\cdots  \psi_{\bY_n},
\end{align*}
where the factor $\eps_{\pm}\in\{1,-1\}$ depends only on
 $(m_l)_{l=1}^n$, $(k_l)_{l=1}^n$.
The anti-symmetric kernels $V_{\rho\bx
 x,m}^+(\cdot)$ $(m=2,4)$ are characterized as follows. 
\begin{align*}
&V_{\rho\bx
 x,2}^+((\eta_1,\by_1,\tau_1,y_1,\xi_1),(\eta_2,\by_2,\tau_2,y_2,\xi_2))\\
&= -\frac{h^2}{4}1_{(\eta_1,\by_1,y_1)=(\eta_2,\by_2,y_2)=(\rho,\bx,
 x)}1_{\tau_1=\tau_2}(1_{(\xi_1,\xi_2)=(1,-1)}-1_{(\xi_1,\xi_2)=(-1,1)}),\\
&V_{\rho\bx
 x,4}^+((\eta_1,\by_1,\tau_1,y_1,\xi_1),(\eta_2,\by_2,\tau_2,y_2,\xi_2),\\
&\qquad\quad(\eta_3,\by_3,\tau_3,y_3,\xi_3),(\eta_4,\by_4,\tau_4,y_4,\xi_4))\\
&= \frac{h^4}{4!}1_{(\eta_i,\by_i,y_i)=(\rho,\bx,x),
 (\forall i\in\{1,2,3,4\})}\\
&\quad\cdot\sum_{\zeta\in \S_4}\sgn(\zeta)
1_{((\tau_{\zeta(1)},\xi_{\zeta(1)}),(\tau_{\zeta(2)},\xi_{\zeta(2)}),
 (\tau_{\zeta(3)},\xi_{\zeta(3)}),
 (\tau_{\zeta(4)},\xi_{\zeta(4)}))=((\ua,1),(\da,1),(\da,-1),(\ua,-1))},
\end{align*}
which imply that
\begin{align}\label{eq_each_interaction_kernel_L1_bound}
\|V_{\rho\bx x,m}^+\|_{L^1}\le 1,\ (\forall (\rho,\bx,x)\in\cB\times
 \G\times [0,\beta)_h,m\in\{2,4\}).
\end{align}
By using Lemma \ref{lem_h_independent_determinant_bound},
 \eqref{eq_each_interaction_kernel_L1_bound} and Lemma
 \ref{lem_L1_norm_comparison} we can
 estimate $\|\tilde{S}_m^0\|_{L^1}$ as follows.
\begin{align*}
\|\tilde{S}_m^0\|_{L^1}&\le \sum_{n=2}^N\frac{1}{n!}\prod_{j=1}^n\left(\frac{U_{max}}{h}\sum_{(\rho_j,\bx_j,x_j)\in\cB\times\G\times
 [0,\beta)_h}\right)1_{\exists i \exists j(i\neq j\land x_i=x_j)}\\
&\quad\cdot \prod_{l=1}^n\left(\sum_{m_l\in
 \{2,4\}}\sum_{k_l=0}^{m_l}\left(\begin{array}{c}m_l\\ k_l\end{array}\right)
\right)1_{\sum_{l=1}^nk_l=m}c_1^{\frac{1}{2}\sum_{l=1}^n(m_l-k_l)}\\
&\le\frac{c_1^{-\frac{m}{2}}}{\beta
 h}\sum_{n=2}^N\frac{1}{n!}\left(\begin{array}{c} n \\
				 2\end{array}\right)
(b\beta L^d
 U_{max})^n\\
&\quad\cdot\prod_{l=1}^n\left(\sum_{m_l\in\{2,4\}}c_1^{\frac{m_l}{2}}
\sum_{k_l=0}^{m_l}\left(\begin{array}{c}m_l\\
			k_l\end{array}\right)\right)
1_{\sum_{l=1}^nk_l=m}.
\end{align*}
Thus, for any $\alpha\in\R_{\ge 0}$,
\begin{align*}
&\sum_{m=0}^N\alpha^mc_1^{\frac{m}{2}}\|\tilde{S}_m^0\|_{L^1}
\le\frac{1}{2\beta h}\sum_{n=2}^N\frac{1}{(n-2)!}\left(b\beta L^d
 U_{max}\sum_{m\in\{2,4\}}(\alpha+1)^mc_1^{\frac{m}{2}}\right)^n\\
&\le \frac{1}{2\beta h}\left(b\beta L^d
 U_{max}\sum_{m\in\{2,4\}}(\alpha+1)^mc_1^{\frac{m}{2}}\right)^2e^{b\beta L^d
 U_{max}\sum_{m\in\{2,4\}}(\alpha+1)^mc_1^{\frac{m}{2}}}.
\end{align*}
We can estimate $\int
 (e^{-V^-(\psi+\psi^1)}-f_{V^-}(\psi+\psi^1))d\mu_{C_{> 0}^-}(\psi^1)$ in
 the same way as above and obtain the claimed estimation of $S^-(\psi)-S^{0}(\psi)$.
\end{proof}

\begin{lemma}\label{lem_explicit_beta_L_condition}
Take any $\alpha\in\R_{\ge 0}$, $\eps\in (0,1)$. Assume that 
\begin{align*}
U_{max}\le (b\beta
 L^d((\alpha+1)^2c_1+(\alpha+1)^4c_1^2))^{-1}\log\left(\frac{2(\eps+1)}{\eps+2}\right).
\end{align*}
Then, the following inequalities hold.
\begin{enumerate}
\item\label{item_exponential_explicit_0th}
$$
|S_0^{\delta}-1|\le\frac{\eps}{\eps+2},\ (\forall \delta\in\{+,-,0\}).
$$
\item\label{item_sup_inf_condition}
$$
\sup_{\delta\in\{+,-,0\}}\sum_{m=1}^N\alpha^mc_1^{\frac{m}{2}}\|S_m^{\delta}\|_{L^1}\le
     \eps\inf_{\delta\in\{+,-,0\}}|S_0^{\delta}|.
$$
\end{enumerate}
\end{lemma}
\begin{proof}
\eqref{item_exponential_explicit_0th}: The inequalities follow from the
 assumption and Lemma \ref{lem_exponential_appl_bound}
 \eqref{item_exponential_appl_0th}.

\eqref{item_sup_inf_condition}: By the assumption, 
\begin{align*}
e^{b\beta L^d U_{max}((\alpha+1)^2c_1+(\alpha+1)^4c_1^2)}\le
 (\eps+1)(2-e^{b\beta L^dU_{max}(c_1+c_1^2)}).
\end{align*}
Thus, by Lemma \ref{lem_exponential_appl_bound}
 \eqref{item_exponential_appl_0th},\eqref{item_exponential_appl_sum},
\begin{align*}
\sup_{\delta\in \{+,-,0\}}\sum_{m=1}^N
\alpha^mc_1^{\frac{m}{2}}\|S_m^{\delta}\|_{L^1}&\le \sup_{\delta\in
 \{+,-,0\}}\sum_{m=0}^N\alpha^mc_1^{\frac{m}{2}}\|S_m^{\delta}\|_{L^1}-\inf_{\delta\in \{+,-,0\}}|S_0^{\delta}|\\
&\le (\eps+1)(2-e^{b\beta L^d U_{max}(c_1+c_1^2)})-\inf_{\delta\in
 \{+,-,0\}}|S_0^{\delta}|\\
&\le \eps \inf_{\delta\in
 \{+,-,0\}}|S_0^{\delta}|.
\end{align*}
\end{proof}

\begin{lemma}\label{lem_logarithm_appl_bound}
Take any $\alpha\in\R_{\ge 0}$, $\eps\in(0,1)$. Assume that 
\begin{align*}
U_{max}\le (b\beta
 L^d((\alpha+2)^2c_1+(\alpha+2)^4c_1^2))^{-1}\log\left(\frac{2(\eps+1)}{\eps+2}\right).
\end{align*}
Set $R^{\delta}(\psi):=\log S^{\delta}(\psi)$, $(\delta\in
 \{+,-,0\})$. Then, the following inequalities hold for any $h\in
 (2/\beta)\N$ satisfying $h>(1/2)\max\{1,4/\beta\}$.
\begin{enumerate}
\item\label{item_logarithm_appl_0th}
$$
|R_0^{\delta}|\le \log\left(\frac{\eps+2}{2}\right),\ (\forall \delta \in \{+,-,0\}).
$$
\item\label{item_logarithm_appl_sum}
$$
\sum_{m=1}^N\alpha^mc_1^{\frac{m}{2}}\|R_m^{\delta}\|_{L^1}\le
     -\log(1-\eps),\ (\forall \delta\in \{+,-,0\}).
$$
\item\label{item_logarithm_appl_0th_difference}
\begin{align*}
|R_0^{\delta}-R_0^0|\le
 -\log\left(1-\max\left\{1,\frac{4}{\beta}\right\}\frac{1}{2h}\right),\ (\forall \delta\in\{+,-\}).
\end{align*}
\item\label{item_logarithm_appl_difference}
\begin{align*}
\sum_{m=1}^N\alpha^mc_1^{\frac{m}{2}}\|R_m^{\delta}-R_m^0\|_{L^1}\le
\max\left\{1,\frac{4}{\beta}\right\}\frac{1}{2(1-\eps)h},\ (\forall \delta\in\{+,-\}).
\end{align*}
\end{enumerate}
\end{lemma}
\begin{remark}\label{rem_real_part_remark}
By Lemma \ref{lem_explicit_beta_L_condition}
 \eqref{item_exponential_explicit_0th},
$$
\Re S_0^{\delta}\ge \frac{2}{\eps+2}>0,\ (\forall \delta\in \{+,-,0\}).
$$
Thus, $R^{\delta}(\psi)$ $(\delta\in \{+,-,0\})$ are
 well-defined. 
\end{remark}
\begin{proof}[Proof of Lemma \ref{lem_logarithm_appl_bound}]
\eqref{item_logarithm_appl_0th}: The result follows from Lemma
 \ref{lem_explicit_beta_L_condition}
 \eqref{item_exponential_explicit_0th} and Lemma
 \ref{lem_logarithm_general_bound} \eqref{item_logarithm_general_0th}.

\eqref{item_logarithm_appl_sum}: By the assumption and Lemma
 \ref{lem_explicit_beta_L_condition} \eqref{item_sup_inf_condition} we
 can apply Lemma \ref{lem_logarithm_general_bound}
 \eqref{item_logarithm_general_sum} to obtain the result.

\eqref{item_logarithm_appl_0th_difference}: By the assumption, Lemma
 \ref{lem_exponential_appl_bound}
 \eqref{item_exponential_appl_difference}, Lemma
 \ref{lem_exponential_appl_bound_another}, Lemma
 \ref{lem_explicit_beta_L_condition} \eqref{item_exponential_explicit_0th} and
 Lemma \ref{lem_logarithm_general_bound}
 \eqref{item_logarithm_general_0th_difference} we have that
\begin{align*}
|R_0^+-R_0^0|&\le
 \sum_{n=1}^{\infty}\frac{1}{n}\left|\frac{S_0^+-S_0^0}{S_0^0}\right|^n\le
 \sum_{n=1}^{\infty}\frac{1}{n}\left|\frac{\frac{1}{h}\left(\frac{2(\eps+1)}{\eps+2}-1\right)}{\frac{2}{\eps+2}}\right|^n\\
&\le -\log \left(1-\frac{1}{2h}\right),\\
|R_0^--R_0^0|&\le\sum_{n=1}^{\infty}\frac{1}{n}\left|\frac{S_0^--S_0^0}{S_0^0}\right|^n\le
\sum_{n=1}^{\infty}\frac{1}{n}\left|\frac{\frac{2(\eps+1)}{\beta
 h(\eps+2)}}{\frac{2}{\eps+2}}\right|^n\le -\log\left(1-\frac{2}{\beta h}\right).
\end{align*}
These imply the result.

\eqref{item_logarithm_appl_difference}: The assumption, Lemma
 \ref{lem_exponential_appl_bound} \eqref{item_exponential_appl_difference}, Lemma
 \ref{lem_exponential_appl_bound_another}, Lemma
 \ref{lem_explicit_beta_L_condition}
 \eqref{item_exponential_explicit_0th},\eqref{item_sup_inf_condition}
 and Lemma \ref{lem_logarithm_general_bound}
 \eqref{item_logarithm_general_difference} ensure that
\begin{align*}
&\sum_{m=1}^N\alpha^mc_1^{\frac{m}{2}}\|R_m^+-R^0_m\|_{L^1}\\
&\le\frac{1}{1-\eps}|S_0^0|^{-1}|S_0^+|^{-1}\sum_{m=0}^N\alpha^mc_1^{\frac{m}{2}}\|S_m^+\|_{L^1}\sum_{n=0}^N\alpha^nc_1^{\frac{n}{2}}\|S_n^+-S_n^0\|_{L^1}\\
&\le \frac{1}{1-\eps}\cdot\frac{\eps+2}{2}\cdot\eps\cdot\frac{1}{h}\left(\frac{2(\eps+1)}{\eps+2}-1\right)\\
&\le \frac{1}{2h(1-\eps)},\\
&\sum_{m=1}^N\alpha^mc_1^{\frac{m}{2}}\|R_m^--R^0_m\|_{L^1}\\
&\le\frac{1}{1-\eps}|S_0^0|^{-1}|S_0^-|^{-1}\sum_{m=0}^N\alpha^mc_1^{\frac{m}{2}}\|S_m^-\|_{L^1}\sum_{n=0}^N\alpha^nc_1^{\frac{n}{2}}\|S_n^--S_n^0\|_{L^1}\\
&\le
 \frac{1}{1-\eps}\cdot\frac{\eps+2}{2}\cdot\eps\cdot\frac{1}{\beta h}\cdot\frac{2(\eps+1)}{\eps+2}\\
& \le \frac{2}{\beta h(1-\eps)}.
\end{align*}
The claimed inequality follows from these inequalities.
\end{proof}

We conclude this section by proving the following lemma, which enables
us to adopt 
$$
-\frac{1}{\beta L^d}\log\left(\int
e^{\frac{1}{2}(R^+(\psi)+R^-(\psi))}d\mu_{C_{\le 0}^{\infty}}(\psi)\right)
$$
as a formulation of the normalized free energy density. Later in Section
\ref{sec_IR_model} we will see that this Grassmann integral formulation
is suited to the IR analysis since it has desirable symmetries.

\begin{lemma}\label{lem_symmetric_formulation}
There exist $(\beta,L^d,b,\chi,E)$-dependent, $h$-independent constants
 $h_0,c_2,c_3\in\R_{>0}$ such that the following statements hold
 for any $h\in (2/\beta)\N$ satisfying $h\ge h_0$.
\begin{enumerate}
\item\label{item_real_part_symmetric}
\begin{align*}
&\Re \int e^{-V(\psi)}d\mu_C(\psi)>0,\\
& \Re \int e^{\frac{1}{2}(R^+(\psi)+R^-(\psi))}d\mu_{C_{\le
 0}^{\infty}}(\psi)>0,\\
&(\forall \bU\in\C^b\text{ satisfying }|U_{\rho}|\le
 c_2\ (\forall \rho\in \cB)).
\end{align*}
\item\label{item_logarithm_final_h_estimate}
\begin{align*}
&\left|\log\left(\int e^{-V(\psi)}d\mu_{C}(\psi)\right)-\log\left(\int e^{\frac{1}{2}(R^+(\psi)+R^-(\psi))}d\mu_{C_{\le
 0}^{\infty}}(\psi)\right)\right|\le \frac{1}{h}c_3,\\
&(\forall \bU\in\C^b\text{ satisfying }|U_{\rho}|\le
 c_2\ (\forall \rho\in \cB)).
\end{align*}
 \end{enumerate}
\end{lemma}

\begin{proof}
Take any $\eps\in (0,2/5)$ and assume that
\begin{align}\label{eq_Umax_symmetric_formulation_proof}
U_{max}\le (b\beta
 L^d(4^2 c_1+4^4c_1^2))^{-1}\log\left(\frac{2(\eps+1)}{\eps+2}\right).
\end{align}
\eqref{item_real_part_symmetric}: By Lemma
 \ref{lem_exponential_general_bound}
 \eqref{item_exponential_general_0th} and Lemma
 \ref{lem_logarithm_appl_bound} \eqref{item_logarithm_appl_0th},\eqref{item_logarithm_appl_sum},
\begin{align}
&\left|\int e^{\frac{1}{2}(R^+(\psi)+R^-(\psi))}d\mu_{C_{\le
 0}^+}(\psi)-1\right|\label{eq_exponential_typical_calculation}\\
&\le \left|\int e^{\frac{1}{2}(R^+(\psi)+R^-(\psi))}d\mu_{C_{\le
 0}^+}(\psi)-e^{\frac{1}{2}(R_0^++R_0^-)}\right|+|e^{\frac{1}{2}(R_0^++R_0^-)}-1|\notag\\
&\le
 e^{\frac{1}{2}(|R_0^+|+|R_0^-|)}\left(e^{\frac{1}{2}\sum_{m=1}^Nc_1^{\frac{m}{2}}(\|R_m^+\|_{L^1}+\|R_m^-\|_{L^1})}-1\right)+e^{\frac{1}{2}(|R_0^+|+|R_0^-|)}-1\notag\\
&\le \frac{\eps+2}{2(1-\eps)}-1.\notag
\end{align}
Since $\eps\in (0,2/5)$,
\begin{align*}
\Re \int e^{\frac{1}{2}(R^+(\psi)+R^-(\psi))}d\mu_{C_{\le
 0}^+}(\psi)\ge \frac{2-5\eps}{2(1-\eps)}>0.
\end{align*}

 It follows from \eqref{eq_Umax_symmetric_formulation_proof} and the same argument as in the proof of Lemma
 \ref{lem_exponential_appl_bound} \eqref{item_exponential_appl_0th} that
\begin{align*}
&\left|\int e^{-V(\psi)}d\mu_{C}(\psi)-1\right|\le e^{b\beta
 L^dU_{max}(c_1+c_1^2)}-1\le \frac{2(\eps+1)}{\eps+2}-1.
\end{align*}
Therefore,
$$
\Re \int e^{-V(\psi)}d\mu_{C}(\psi)\ge \frac{2}{\eps+2}>0.
$$

\eqref{item_logarithm_final_h_estimate}:
The same calculation as in \eqref{eq_exponential_typical_calculation}
 yields that
\begin{align}\label{eq_preparation_lower_bound}
\left|\int e^{R^+(\psi)}d\mu_{C_{\le 0}^+}(\psi)\right|\ge \frac{2-5\eps}{2(1-\eps)}.
\end{align}
By Lemma \ref{lem_h_independent_determinant_bound} and Lemma
 \ref{lem_exponential_general_bound}
 \eqref{item_exponential_general_difference},\eqref{item_exponential_general_covariance_difference}, 
\begin{align*}
&\left|\int
 e^{\frac{1}{2}(R^+(\psi)+R^-(\psi))}d\mu_{C_{\le
 0}^{\infty}}(\psi)-\int e^{R^+(\psi)}d\mu_{C_{\le 0}^+}(\psi)\right|\\
&\le \left|\int
 e^{\frac{1}{2}(R^+(\psi)+R^-(\psi))}d\mu_{C_{\le
 0}^{\infty}}(\psi)- \int
 e^{\frac{1}{2}(R^+(\psi)+R^-(\psi))}d\mu_{C_{\le
 0}^{+}}(\psi)\right|\\
&\quad+ \left|\int
 e^{\frac{1}{2}(R^+(\psi)+R^-(\psi))}d\mu_{C_{\le
 0}^{+}}(\psi)-\int e^{R^+(\psi)}d\mu_{C_{\le 0}^+}(\psi)\right|\\
&\le \frac{1}{h}
 e^{\frac{1}{2}(|R_0^+|+|R_0^-|)}\left(e^{\frac{1}{2}\sum_{m=1}^N2^mc_1^{\frac{m}{2}}(\|R_m^+\|_{L^1}+\|R_m^-\|_{L^1})}-1\right)\\
&\quad
 +\left((e^{\frac{1}{2}|R_0^+-R_0^-|}-1)e^{|R_0^+|}+\frac{1}{2}e^{|R_0^+|}\sum_{m=1}^Nc_1^{\frac{m}{2}}\|R_m^+-R_m^-\|_{L^1}\right)\\
&\qquad\quad\cdot e^{\sup_{\delta\in
 \{+,-\}}\sum_{m=1}^Nc_1^{\frac{m}{2}}\|R_m^{\delta}\|_{L^1}}\\
&\le \frac{1}{h}
 e^{\frac{1}{2}(|R_0^+|+|R_0^-|)}\left(e^{\frac{1}{2}\sum_{m=1}^N2^mc_1^{\frac{m}{2}}(\|R_m^+\|_{L^1}+\|R_m^-\|_{L^1})}-1\right)\\
&\quad
 +\Bigg((e^{\frac{1}{2}\sum_{\delta\in\{+,-\}}|R_0^{\delta}-R_0^0|}-1)e^{|R_0^+|}+\frac{1}{2}e^{|R_0^+|}\sum_{\delta\in\{+,-\}}\sum_{m=1}^Nc_1^{\frac{m}{2}}\|R_m^{\delta}-R_m^0\|_{L^1}\Bigg)\\
&\qquad\quad\cdot e^{\sup_{\delta\in
 \{+,-\}}\sum_{m=1}^Nc_1^{\frac{m}{2}}\|R_m^{\delta}\|_{L^1}}.
\end{align*}
Set $c':=\max\{1,4/\beta\}$. By the assumption \eqref{eq_Umax_symmetric_formulation_proof} we
 can substitute the inequalities proved in Lemma
 \ref{lem_logarithm_appl_bound} for $\alpha=2$ to derive that 
\begin{align*}
&\left|\int
 e^{\frac{1}{2}(R^+(\psi)+R^-(\psi))}d\mu_{C_{\le
 0}^{\infty}}(\psi)-\int e^{R^+(\psi)}d\mu_{C_{\le 0}^+}(\psi)\right|\\
&\le \frac{\eps+2}{2h}\left(\frac{1}{1-\eps}-1\right)\\
&\quad+\Bigg(\left(\frac{1}{1-c'/(2h)}-1\right)\frac{\eps+2}{2}+\frac{\eps+2}{4}\cdot\frac{c'}{(1-\eps)h}\Bigg)\frac{1}{1-\eps}.
\end{align*}
This inequality implies that there exist $h$-independent constants
 $c'',h_0\in\R_{>0}$ such that for any $h\in (2/\beta)\N$ with $h\ge h_0$,
\begin{align}\label{eq_preparation_integral_final_bound}
\left|\int
 e^{\frac{1}{2}(R^+(\psi)+R^-(\psi))}d\mu_{C_{\le
 0}^{\infty}}(\psi)-\int e^{R^+(\psi)}d\mu_{C_{\le 0}^+}(\psi)\right|\le
 \frac{1}{h}c''.
\end{align}
By using the inequalities \eqref{eq_preparation_lower_bound},
 \eqref{eq_preparation_integral_final_bound} and taking a larger $h$ if
 necessary we have that
\begin{align}
&\left|\log\left(\int
 e^{\frac{1}{2}(R^+(\psi)+R^-(\psi))}d\mu_{C_{\le
 0}^{\infty}}(\psi)\right)-\log\left(\int e^{R^+(\psi)}d\mu_{C_{\le 0}^+}(\psi)\right)\right|
\label{eq_preparation_integral_end_bound}\\
&\le
 \sum_{n=1}^{\infty}\frac{1}{n}\left(\frac{2(1-\eps)}{2-5\eps}\cdot\frac{c''}{h}\right)^n\le \frac{2(1-\eps)c''}{2-5\eps-2(1-\eps)c''/h}\cdot\frac{1}{h}.\notag
\end{align}
We saw in Remark \ref{rem_real_part_remark} that $\Re
 S_0^+>0$. Therefore, we can apply \cite[\mbox{Lemma C.2}]{K3} to justify that
\begin{align*}
\int e^{R^+(\psi)}d\mu_{C_{\le 0}^+}(\psi)&=\int\int
 e^{-V(\psi+\psi^1)}d\mu_{C_{>0}^+}(\psi^1)d\mu_{C_{\le 0}^+}(\psi)\\
&=\int
 e^{-V(\psi)}d\mu_{C}(\psi).
\end{align*}
By combining this equality with
 \eqref{eq_preparation_integral_end_bound} we obtain the inequality
 claimed in \eqref{item_logarithm_final_h_estimate}.
\end{proof}

\section{General estimation}\label{sec_general_estimation}
In this section we establish various inequalities which will form the
basis of both the Matsubara UV integration and the IR integration around
the zero set of the free dispersion relation. We also show that a Grassmann
polynomial produced by a single-scale integration inherits symmetric
properties which the covariance and the input polynomial originally
have. Here we assume that a covariance $C_o:I_0^2\to \C$ is given and
satisfies 
\begin{align}
&|\det(\<\bp_i,\bq_j\>_{\C^r}C_o(X_i,Y_j))_{1\le i,j\le n}|\le
 D_{et}^n,\label{eq_general_determinant_bound_condition}\\
&(\forall r,n\in\N,\bp_i,\bq_i\in\C^r\text{ with }
\|\bp_i\|_{\C^r},\|\bq_i\|_{\C^r}\le 1,\notag\\
&\quad X_i,Y_i\in I_0\
 (i=1,2,\cdots,n)),\notag
\end{align}
with a constant $D_{et}\in\R_{>0}$. It is sometimes more convenient to deal with the anti-symmetric
extension $\widetilde{C_o}:I^2\to \C$ of $C_o$ than $C_o$ itself. The definition
of $\widetilde{C_o}$ is that
\begin{align}
&\widetilde{C_o}((X,\theta),(Y,\xi)):=\frac{1}{2}(1_{(\theta,\xi)=(1,-1)}C_o(X,Y)-1_{(\theta,\xi)=(-1,1)}C_o(Y,X)),\label{eq_anti_symmetric_extension_covariance}\\
&(\forall X,Y\in I_0,\theta, \xi\in \{1,-1\}).\notag
\end{align}

In order to measure sizes of Grassmann polynomials during the multi-scale
integrations, we need to define a family of norms and semi-norms on the
linear space of anti-symmetric functions on $I^m$. For any
$(\rho,\bx,\s,x,$ $\theta)$, $(\eta,\by,\tau,y,\xi)\in I$ and $j\in
\{0,1,\cdots,d\}$ set 
\begin{align*}
&d_j((\rho,\bx,\s,x,\theta),(\eta,\by,\tau,y,\xi))\\
&:=\left\{\begin{array}{ll}\frac{\beta}{2\pi}|e^{i\frac{2\pi}{\beta}x}-e^{i\frac{2\pi}{\beta}y}|&\text{if }j=0,\\
\frac{L}{2\pi}|e^{i\frac{2\pi}{L}\<\bx,\bv_j\>}-e^{i\frac{2\pi}{L}\<\by,\bv_j\>}|&\text{if }j\in\{1,2,\cdots,d\}.
\end{array}
\right. 
\end{align*}
Assume that a set of positive numbers $\{\fw(l)\}_{l\in \Z}$ is given. Fix
$\fr\in (0,1]$. For any $m\in \{2,3,\cdots,N\}$, anti-symmetric
function $f:I^m\to \C$ and $l\in\Z$, let
\begin{align}\label{eq_definition_norm_semi_norm}
\|f\|_{l,0}&:=\sup_{X\in I}\left(\frac{1}{h}\right)^{m-1}\sum_{\bY\in
 I^{m-1}}e^{\sum_{j=0}^d(\fw(l)d_j(X,Y_1))^{\fr}}|f(X,\bY)|,\\
\|f\|_{l,1}&:=\sup_{j'\in\{0,1,\cdots,d\}}\sup_{q\in\{1,2,\cdots,m-1\}}\sup_{X\in
 I}\notag\\
&\quad\cdot\left(\frac{1}{h}\right)^{m-1}\sum_{\bY\in
 I^{m-1}}d_{j'}(X,Y_q)e^{\sum_{j=0}^d(\fw(l)d_j(X,Y_1))^{\fr}}|f(X,\bY)|.\notag
\end{align}
To understand these definitions clearly, recall the notational rule that $\bY=(Y_1,Y_2,\cdots,Y_{m-1})$ for
$\bY\in I^{m-1}$. Since no Grassmann polynomial of degree 1 appears in
our multi-scale analysis,
there is no need to newly introduce a norm in the space of functions on $I$. To
organize formulas we sometimes write $\|f_0\|_{l,0}$ in place of $|f_0|$ for
$f_0\in \C$ as well. 

When we practically use the norm $\|\cdot\|_{l,0}$ and the semi-norm
$\|\cdot\|_{l,1}$ in Section \ref{sec_UV} and Section \ref{sec_IR_model}, the integer $l$ will represent an integration scale in the
multi-scale integration procedure. Moreover, in these sections the weight $\fw(l)$ and the exponent
$\fr$ will be specifically defined. However, the general theory in this
section can be completed without more detailed information on these
parameters.

\subsection{Estimation of the free
  integration}\label{subsec_general_free_integration}
Here we estimate a Grassmann polynomial produced by the free
integration. With $J(\psi)\in\bigwedge \cV$ satisfying $J_m(\psi)=0$ if $m\notin
2\N\cup \{0\}$, set
\begin{equation}\label{eq_definition_free_general}
F(\psi):=\int J(\psi+\psi^1)d\mu_{C_o}(\psi^1).
\end{equation}
By the definition of the Grassmann Gaussian integral, $F_m(\psi)=0$ if
$m\notin 2\N\cup \{0\}$.

\begin{lemma}\label{lem_general_free_bound}
The following inequalities hold.
\begin{align*}
&|F_0|\le |J_0|+\frac{N}{h}\sum_{n=1}^ND_{et}^{\frac{n}{2}}\|J_n\|_{l,0},\\
&\|F_m\|_{l,r}\le
 \|J_m\|_{l,r}+\sum_{n=m+1}^N2^nD_{et}^{\frac{n-m}{2}}\|J_n\|_{l,r},\\
&(\forall r\in\{0,1\},l\in\Z,m\in\N_{\ge 2}).
\end{align*}
\end{lemma}
\begin{proof}
By anti-symmetry, 
\begin{align*}
F_m(\psi)=J_m(\psi)+\left(\frac{1}{h}\right)^m\sum_{\bX\in I^m}
\Bigg(&\sum_{n=m+1}^N\left(\frac{1}{h}\right)^{n-m}\sum_{\bY\in
	       I^{n-m}}\left(\begin{array}{c}n\\m\end{array}
\right)\\
&\cdot J_n(\bX,\bY)\int \psi_{\bY}^1d\mu_{C_o}(\psi^1)\Bigg)\psi_{\bX},
\end{align*}
which implies that for any $\bX\in I^m$,
\begin{align}\label{eq_characterization_free_kernel}
F_m(\bX)=&J_m(\bX)\\
         &+\sum_{n=m+1}^N\left(\frac{1}{h}\right)^{n-m}
\sum_{\bY\in I^{n-m}}\left(\begin{array}{c}n\\m\end{array}
\right)J_n(\bX,\bY)\int \psi_{\bY}^1d\mu_{C_o}(\psi^1).\notag
\end{align}
Since $\|J_n\|_{L^1}\le 1_{n=0}|J_0| + 1_{n\ge 1}\frac{N}{h}\|J_n\|_{l,0}$, 
\begin{align*}
|F_0|\le \sum_{n=0}^ND_{et}^{\frac{n}{2}}\|J_n\|_{L^1}
\le |J_0|+\frac{N}{h} \sum_{n=1}^ND_{et}^{\frac{n}{2}}\|J_n\|_{l,0}.
\end{align*}
By using the inequality 
$$\left(\begin{array}{c}n \\ m\end{array}\right)\le 2^n,$$
we can derive the claimed upper bound on $\|F_m\|_{l,r}$ from \eqref{eq_characterization_free_kernel}.   
\end{proof}

\subsection{Estimation of the tree
  expansion}\label{subsec_general_tree_expansion}
Take $J(\psi)\in\bigwedge \cV$ satisfying $J_m(\psi)=0$ if $m\notin 2\N\cup \{0\}$.
We can see from definition that there exists a domain $O$ of $\C$
containing 0 such that
$$
\log\left(\int
e^{zJ(\psi+\psi^1)}d\mu_{C_o}(\psi^1)\right)
$$
is analytic with $z$ in $O$. It is known that for any
$n\in\N_{\ge 2}$,
$$
\left(\frac{d}{dz}\right)^n\log\left.\left(\int e^{z J(\psi+\psi^1)}d\mu_{C_o}(\psi^1)\right)\right|_{z=0}
$$
can be characterized as a sum over trees with $n$ vertices. We adopt one
version of such formulas clearly proved in \cite{SW}. The formula
\cite[\mbox{Theorem 3}]{SW} states that for any $n\in \N_{\ge 2}$ and
$J(\psi)\in\bigwedge \cV$ satisfying $J_m(\psi)=0$ if $m\notin 2\N\cup \{0\}$,  \begin{align}
&\frac{1}{n!}\left(\frac{d}{dz}\right)^n\log\left.\left(\int e^{z
 J(\psi+\psi^1)}d\mu_{C_o}(\psi^1)\right)\right|_{z=0}\label{eq_explicit_tree_formula}\\
&= \frac{1}{n!}\sum_{T\in \T_n}\prod_{\{p,q\}\in
 T}(\D_{p,q}(C_o)+\D_{q,p}(C_o))\int_{[0,1]^{n-1}}d\bs
 \sum_{\xi\in\S_n(T)}\varphi(T,\xi,\bs)\notag\\
&\quad\cdot
 e^{\sum_{r,s=1}^nM_{at}(T,\xi,\bs)(r,s)\D_{r,s}(C_o)}\prod_{j=1}^nJ(\psi^j+\psi)\Bigg|_{\psi^j=0\atop(\forall j\in\{1,2,\cdots,n\})},\notag
\end{align}
where $\T_n$ is the set of all trees over the vertices
$\{1,2,\cdots,n\}$, 
\begin{align}\label{eq_grassmann_laplacian}
\D_{r,s}(C_o)&:=-\sum_{X,Y\in I_0}C_o(X,Y)\frac{\partial}{\partial
 \opsi_X^r}\frac{\partial}{\partial \psi_Y^s}\\
&: \bigwedge \left(\bigoplus_{j=1}^n\cV_j\right)\to \bigwedge
 \left(\bigoplus_{j=1}^n\cV_j\right),\ (\forall r,s\in
 \{1,2,\cdots,n\}),\notag
\end{align}
$\S_n(T)$ is a $T$-dependent subset of $\S_n$, $\varphi(T,\xi,\cdot)$ is a
$(T,\xi)$-dependent continuous function from $[0,1]^{n-1}$ to $\R_{\ge 0}$
satisfying that 
\begin{align}\label{eq_property_irrelevant_function}
\int_{[0,1]^{n-1}}d\bs\sum_{\xi\in \S_n(T)}\varphi(T,\xi,\bs)=1,\ (\forall
 T\in \T_n),
\end{align}
$(M_{at}(T,\xi,\bs)(r,s))_{1\le r,s\le n}$ is a $(T,\xi,\bs)$-dependent
real symmetric \\
non-negative matrix satisfying that
$M_{at}(T,\xi,\bs)(r,r)=1$ $(\forall r\in \{1,\cdots,$ $n\})$. Moreover,
$\bs\mapsto M_{at}(T,\xi,\bs)(r,s)$ is continuous in $[0,1]^{n-1}$
$(\forall r,s\in\{1,$ $\cdots,n\})$. 

For a given polynomial $J(\psi)\in\bigwedge \cV$ satisfying $J_m(\psi)=0$ if
$m\notin 2\N\cup \{0\}$, let us define $T^{(n)}(\psi)\in \bigwedge \cV$ by the
right-hand side of \eqref{eq_explicit_tree_formula}. The goal of this
subsection is to establish norm estimates on the anti-symmetric kernels
of $T^{(n)}(\psi)$.

To facilitate our analysis, let us fix some notational conventions. For
$\bX\in I^m$, $\bY\in I^n$  we write $\bX\subset \bY$ if $m\le n$ and
there exist $j_1,j_2,\cdots,j_m\in\{1,2,\cdots,n\}$ such that
$j_1<j_2<\cdots<j_m$ and $\bX=(Y_{j_1},Y_{j_2},\cdots,Y_{j_m})$. In this
case we also define $\bY\backslash \bX\in I^{n-m}$ by $\bY\backslash
\bX:=(Y_{k_1},Y_{k_2},\cdots,Y_{k_{n-m}})$, where $1\le
k_1<k_2<\cdots<k_{n-m}\le n$ and $\{k_1,k_2,\cdots,k_{n-m}\}\cap
\{j_1,j_2,\cdots,$ $j_m\}=\emptyset$. The following abbreviation will be often
used. For any object $f(\bX)$ parameterized by the variable $\bX\in I^m$,
$$\sum_{\bX\subset\bY\atop \bX\in I^m}f(\bX)$$
 denotes 
$$
\sum_{1\le i_1<i_2<\cdots < i_m\le n}f((Y_{i_1},Y_{i_2},\cdots,Y_{i_m})).
$$
 
For any $T\in \T_n$, $p,q\in \{1,2,\cdots,n\}$ let
$n_p(T)(\in\{1,2,\cdots,n-1\})$ be the incidence number of the vertex
$p$ and $\dis_T(p,q)(\in \{0,1,\cdots,n\})$ denote the distance between
the vertex $p$ and the vertex $q$ along the unique path connecting $p$
to $q$ in $T$. Moreover, set 
$$d_T(p):=\max_{r\in
\{1,2,\cdots,n\}}\dis_T(p,r).$$
Define $L_q^p(T)(\subset T)$ by
$$
L_q^p(T):=\left\{\{q,r\}\in T\ |\ \dis_T(p,r)=\dis_T(p,q)+1\right\}.
$$
Note that 
$$\sharp L_p^p(T)=n_p(T),\ \sharp L_q^p(T)=n_q(T)-1,\ (\forall q\in \{1,2,\cdots,n\}\backslash\{p\}).$$
For $q\in \{1,2,\cdots,n\}$ with $\sharp L_q^p(T)\neq 0$ let $\zeta_q$ denote the
bijective map from $L_q^p(T)$ to $\{1,2,\cdots,\sharp L_q^p(T)\}$
satisfying that 
\begin{align*}
\zeta_q(\{q,r\})<\zeta_q(\{q,s\}),\ (\forall \{q,r\},\{q,s\}\in
 L^p_q(T)\text{ with }r<s).
\end{align*}

For non-commutative mathematical objects $f_r,f_{r+1},\cdots, f_s$
$(r,s\in \Z,$ $r<s)$ we set  
$$\prod_{n=r\atop order}^sf_n:=f_rf_{r+1}\cdots f_s.$$
This notation will help us to shorten formulas on various occasions. Also for
conciseness, let us set
\begin{align*}
ope(T,C_o)&:=\int_{[0,1]^{n-1}}d\bs\sum_{\xi\in
 \S_n(T)}\varphi(T,\xi,\bs)e^{\sum_{r,s=1}^nM_{at}(T,\xi,\bs)(r,s)\D_{r,s}(C_o)},\\
Ope(T,C_o)&:=ope(T,C_o)\prod_{\{p,q\}\in
 T}(\D_{p,q}(C_o)+\D_{q,p}(C_o)),\ (\forall T\in\T_n).
\end{align*}

We construct necessary estimates step by step. By using the assumption
\eqref{eq_general_determinant_bound_condition}, the properties of
$M_{at}(T,\xi,\bs)$ and by repeating the same argument as in
\cite[\mbox{Lemma 4.5}]{K1} one can prove the next lemma. The
reason why the covariance $C_o$ needs to be multiplied by
$\<\bp_i,\bq_j\>_{\C^r}$ in \eqref{eq_general_determinant_bound_condition} is
that $M_{at}(T,\xi,\bs)(r,s)$ can be rewritten as $\<\bp_r,\bq_s\>_{\C^n}$
with some $\bp_r,\bq_s\in\R^n$ satisfying
$\|\bp_r\|_{\C^n}=\|\bq_s\|_{\C^n}=1$ during the proof of the next lemma.

\begin{lemma}\label{lem_application_determinant_bound_tree}
For any $T\in \T_n$, $\xi\in\S_n(T)$, $\bs\in[0,1)^{n-1}$, $\bX_j\in
 I^{m_j}$ $(j=1,2,\cdots,n)$,
$$
\left|e^{\sum_{r,s=1}^nM_{at}(T,\xi,\bs)(r,s)\D_{r,s}(C_o)}\prod_{j=1\atop
 order }^n\psi_{\bX_j}^j\Bigg|_{\psi^j=0\atop(\forall
 j\in\{1,2,\cdots,n\})}\right|\le D_{et}^{\frac{1}{2}\sum_{j=1}^nm_j}.
$$
\end{lemma}

Lemma \ref{lem_application_determinant_bound_tree} will be used in the proof of the following lemma.
\begin{lemma}\label{lem_tree_line_expansion}
Take any $T\in \T_n$, $m_j\in\N$ and $\bX_j\in I^{m_j}$
 $(j=1,2,\cdots,n)$. The following inequality holds. 
\begin{align*}
&\left|Ope(T,C_o)\prod_{j=1\atop
 order}^n\psi_{\bX_j}^j\Bigg|_{\psi^j=0\atop (\forall j\in
 \{1,2,\cdots,n\})}\right|\\
&\le 1_{n_j(T)\le m_j(\forall j\in
 \{1,2,\cdots,n\})}2^{n-1}D_{et}^{\frac{1}{2}\sum_{j=1}^nm_j-n+1}\\
&\quad\cdot \sum_{\bW_1\subset \bX_1\atop \bW_1\in I^{n_1(T)}}
\sum_{\s_1\in
 \S_{n_1(T)}}\prod_{\{1,s\}\in L_1^1(T)}\Bigg(\sum_{Z_s\subset
 \bX_s\atop Z_s\in I}|\widetilde{C_o}(W_{1,\s_1\circ
 \zeta_1(\{1,s\})},Z_s)|\Bigg)\\
&\quad \cdot \prod_{u=1\atop order
 }^{d_T(1)-1}\Bigg(\prod_{j\in\{2,3,\cdots,n\}\text{ with }\atop
 \dis_T(1,j)=u,n_j(T)\neq 1}\Bigg(\sum_{\bW_j\subset \bX_j\backslash
 Z_j\atop \bW_j\in
 I^{n_j(T)-1}}\sum_{\s_j\in\S_{n_j(T)-1}}\\
&\qquad\cdot \prod_{\{j,s\}\in L_j^1(T)}\Bigg(\sum_{Z_s\subset
 \bX_s\atop Z_s\in I}|\widetilde{C_o}(W_{j,\s_j\circ\zeta_j(\{j,s\})},Z_s)|
\Bigg)
\Bigg)\Bigg).
\end{align*}
\end{lemma}
\begin{proof}
For $\{p,q\}\in T$ set $\D_{\{p,q\}}:=\D_{p,q}(C_o)+\D_{q,p}(C_o)$. Note
 that 
\begin{equation}\label{eq_laplacian_anti_symmetrization}
\D_{\{p,q\}}=-2\sum_{\bX\in I^2}\widetilde{C_o}(\bX)\frac{\partial
 }{\partial \psi_{X_1}^p}\frac{\partial
 }{\partial \psi_{X_2}^q}.   
\end{equation}
The operator $\prod_{l\in T}\D_l$ erases $n_j(T)$ elements from the
 Grassmann monomial $\psi_{\bX_j}^j$ for every $j\in
 \{1,2,\cdots,n\}$. Hence, we need the constraint $1_{n_j(T)\le m_j
 (\forall j\in \{1,2,\cdots,n\})}$. The operator
 $\prod_{l\in T}\D_l$ can be decomposed as follows.
\begin{align*}
\prod_{l\in T}\D_l=\prod_{l\in L_1^1(T)}\D_l\prod_{u=1}^{d_T(1)-1}\Bigg(
\prod_{j\in\{2,3,\cdots,n\}\text{ with }\atop
 \dis_T(1,j)=u,n_j(T)\neq 1}\prod_{l\in L_j^1(T)}\D_l\Bigg).
\end{align*}
We apply $\D_l$ ($l\in L_1^1(T)$) to the input Grassmann monomial
 first. Then, from $u=1$ up to $u=d_T(1)-1$ we let the operator 
$$
\prod_{j\in\{2,3,\cdots,n\}\text{ with }\atop
 \dis_T(1,j)=u,n_j(T)\neq 1}\prod_{l\in L_j^1(T)}\D_l
$$
act on the remaining polynomial by turns. This procedure yields that
\begin{align*}
&\left|Ope(T,C_o)\prod_{j=1\atop
 order}^n\psi_{\bX_j}^j\Bigg|_{\psi^j=0\atop (\forall j\in
 \{1,2,\cdots,n\})}\right|\\
&\le 1_{n_j(T)\le m_j(\forall j\in
 \{1,2,\cdots,n\})}\\
&\quad\cdot \sum_{\bW_1\subset \bX_1\atop
\bW_1\in I^{n_1(T)}}\sum_{\s_1\in
 \S_{n_1(T)}}\prod_{\{1,s\}\in L_1^1(T)}\Bigg(2\sum_{Z_s\subset
 \bX_s\atop Z_s\in I}|\widetilde{C_o}(W_{1,\s_1\circ
 \zeta_1(\{1,s\})},Z_s)|\Bigg)\\
&\quad \cdot \Bigg|ope(T,C_o)\prod_{u=1 }^{d_T(1)-1}
\Bigg(\prod_{j\in\{2,3,\cdots,n\}\text{ with }\atop
 \dis_T(1,j)=u,n_j(T)\neq 1}\prod_{l\in L_j^1(T)}\D_l\Bigg)\\
&\qquad\cdot\psi_{\bX_1\backslash \bW_1}^1\prod_{j=2\atop order }^n
(1_{\dis_T(1,j)\le 1}\psi_{\bX_j\backslash Z_j}^j+1_{\dis_T(1,j)>
 1}\psi_{\bX_j}^j)\Bigg|_{\psi^j=0\atop (\forall j\in
 \{1,2,\cdots,n\})}\Bigg|\\
&\le 1_{n_j(T)\le m_j(\forall j\in
 \{1,2,\cdots,n\})}\\
&\quad\cdot \sum_{\bW_1\subset \bX_1\atop
\bW_1\in I^{n_1(T)}}\sum_{\s_1\in
 \S_{n_1(T)}}\prod_{\{1,s\}\in L_1^1(T)}\Bigg(2\sum_{Z_s\subset
 \bX_s\atop Z_s\in I}|\widetilde{C_o}(W_{1,\s_1\circ
 \zeta_1(\{1,s\})},Z_s)|\Bigg)\\
&\quad\cdot \prod_{j\in\{2,3,\cdots,n\}\text{ with }\atop
 \dis_T(1,j)=1,n_j(T)\neq 1}
\Bigg(\sum_{\bW_j\subset \bX_j\backslash
 Z_j\atop \bW_j\in
 I^{n_j(T)-1}}\sum_{\s_j\in\S_{n_j(T)-1}}\\
&\qquad\cdot \prod_{\{j,s\}\in L_j^1(T)}\Bigg(2\sum_{Z_s\subset \bX_s\atop
  Z_s\in I}|\widetilde{C_o}(W_{j,\s_j\circ\zeta_j(\{j,s\})},Z_s)|
\Bigg)\Bigg)\\
&\quad \cdot \Bigg|ope(T,C_o)\prod_{u=2 }^{d_T(1)-1}
\Bigg(\prod_{j\in\{2,3,\cdots,n\}\text{ with }\atop
 \dis_T(1,j)=u,n_j(T)\neq 1}\prod_{l\in L_j^1(T)}\D_l\Bigg)\\
&\qquad\cdot\psi_{\bX_1\backslash \bW_1}^1\prod_{j=2\atop order }^n
(1_{\dis_T(1,j)\le 1\text{ and }n_j(T)\neq 1}\psi_{(\bX_j\backslash
 Z_j)\backslash \bW_j}^j\\
&\qquad\qquad\qquad\qquad+1_{\dis_T(1,j)\le 1\text{ and }n_j(T)= 1}
\psi_{\bX_j\backslash
 Z_j}^j+1_{\dis_T(1,j)=2}\psi_{\bX_j\backslash
 Z_j}^j\\
&\qquad\qquad\qquad\qquad+1_{\dis_T(1,j)>2}\psi_{\bX_j}^j)\Bigg|_{\psi^j=0\atop (\forall j\in
 \{1,2,\cdots,n\})}\Bigg|\\
&\le 1_{n_j(T)\le m_j(\forall j\in
 \{1,2,\cdots,n\})}\\
&\quad\cdot \sum_{\bW_1\subset \bX_1\atop
\bW_1\in I^{n_1(T)}}\sum_{\s_1\in
 \S_{n_1(T)}}\prod_{\{1,s\}\in L_1^1(T)}\Bigg(2\sum_{Z_s\subset
 \bX_s\atop Z_s\in I}|\widetilde{C_o}(W_{1,\s_1\circ
 \zeta_1(\{1,s\})},Z_s)|\Bigg)\\
&\quad\cdot \prod_{u=1\atop order }^v\Bigg(\prod_{j\in\{2,3,\cdots,n\}\text{ with }\atop
 \dis_T(1,j)=u,n_j(T)\neq 1}
\Bigg(\sum_{\bW_j\subset \bX_j\backslash
 Z_j\atop \bW_j\in
 I^{n_j(T)-1}}\sum_{\s_j\in\S_{n_j(T)-1}}\\
&\qquad\cdot \prod_{\{j,s\}\in L_j^1(T)}\Bigg(2\sum_{Z_s\subset \bX_s\atop
 Z_s\in I}|\widetilde{C_o}(W_{j,\s_j\circ\zeta_j(\{j,s\})},Z_s)|
\Bigg)\Bigg)\Bigg)\\
&\quad \cdot \Bigg|ope(T,C_o)\prod_{u=v+1 }^{d_T(1)-1}
\Bigg(\prod_{j\in\{2,3,\cdots,n\}\text{ with }\atop
 \dis_T(1,j)=u,n_j(T)\neq 1}\prod_{l\in L_j^1(T)}\D_l\Bigg)\\
&\qquad\cdot\psi_{\bX_1\backslash \bW_1}^1\prod_{j=2\atop order }^n
(1_{\dis_T(1,j)\le v\text{ and }n_j(T)\neq 1}\psi_{(\bX_j\backslash
 Z_j)\backslash \bW_j}^j\\
&\qquad\qquad\qquad\qquad+1_{\dis_T(1,j)\le v\text{ and }n_j(T)= 1}
\psi_{\bX_j\backslash
 Z_j}^j+1_{\dis_T(1,j)=v+1}\psi_{\bX_j\backslash
 Z_j}^j\\
&\qquad\qquad\qquad\qquad+1_{\dis_T(1,j)>v+1}\psi_{\bX_j}^j)\Bigg|_{\psi^j=0\atop (\forall j\in
 \{1,2,\cdots,n\})}\Bigg|\\
&\le 1_{n_j(T)\le m_j (\forall j\in
 \{1,2,\cdots,n\})}\\
&\quad\cdot \sum_{\bW_1\subset \bX_1\atop \bW_1\in
 I^{n_1(T)}}\sum_{\s_1\in
 \S_{n_1(T)}}\prod_{\{1,s\}\in L_1^1(T)}\Bigg(2\sum_{Z_s\subset \bX_s
\atop Z_s\in I}|\widetilde{C_o}(W_{1,\s_1\circ
 \zeta_1(\{1,s\})},Z_s)|\Bigg)\\
&\quad \cdot \prod_{u=1\atop order
 }^{d_T(1)-1}\Bigg(\prod_{j\in\{2,3,\cdots,n\}\text{ with }\atop
 \dis_T(1,j)=u,n_j(T)\neq 1}\Bigg(\sum_{\bW_j\subset \bX_j\backslash
 Z_j\atop \bW_j\in
 I^{n_j(T)-1}}\sum_{\s_j\in\S_{n_j(T)-1}}\\
&\qquad\cdot \prod_{\{j,s\}\in L_j^1(T)}\Bigg(2\sum_{Z_s\subset
 \bX_s\atop Z_s\in I}|\widetilde{C_o}(W_{j,\s_j\circ\zeta_j(\{j,s\})},Z_s)|
\Bigg)
\Bigg)\Bigg)\\
&\quad\cdot \Bigg|ope(T,C_o)\psi_{\bX_1\backslash\bW_1}^1\\
&\qquad\cdot \prod_{j=2\atop
 order}^n(1_{n_j(T)\neq 1}\psi^j_{(\bX_j\backslash
 Z_j)\backslash\bW_j}+1_{n_j(T)= 1}\psi^j_{\bX_j\backslash Z_j})
\Bigg|_{\psi^j=0\atop(\forall j\in \{1,2,\cdots,n\})}\Bigg|.
\end{align*}
Collecting the factor $2$ gives $2^{n-1}$, since the tree $T$
 has $n-1$ lines.  Then, by using \eqref{eq_property_irrelevant_function}, Lemma
 \ref{lem_application_determinant_bound_tree} and the fact that
 $\sum_{j=1}^nn_j(T)=2(n-1)$ we obtain the claimed inequality.
\end{proof}

\begin{lemma}\label{lem_tree_line_kernels}
Take any $T\in \T_n$ and $m_j\in \N_{\ge 2}$ $(j=1,2,\cdots,n)$.
Let $J_{m_j}:I^{m_j}\to \C$ $(j=1,2,\cdots,n)$ be anti-symmetric
 functions. Then, the following inequalities hold.
\begin{enumerate}
\item\label{item_tree_kernels_no_fixed} For any $X_{1,1}\in I$,
\begin{align}\label{eq_tree_kernels_no_fixed}
&\Bigg|Ope(T,C_o)\left(\frac{1}{h}\right)^{m_1-1}\sum_{(X_{1,2},X_{1,3},\cdots,X_{1,m_1})\in
 I^{m_1-1}}J_{m_1}(\bX_1)\\
&\quad\cdot\prod_{j=2}^n
\left(\left(\frac{1}{h}\right)^{m_j}\sum_{\bX_j\in
 I^{m_j}}J_{m_j}(\bX_j)\right)\prod_{k=1\atop
 order}^n\psi_{\bX_k}^k\Bigg|_{\psi^j=0\atop (\forall j\in
 \{1,2,\cdots,n\})}\Bigg|\notag\\
&\le 1_{n_j(T)\le m_j(\forall j\in
 \{1,2,\cdots,n\})}2^{n-1}\prod_{j=1}^n\left(m_j\left(\begin{array}{c}m_j-1\\ n_j(T)-1\end{array}\right)(n_j(T)-1)!\right)\notag\\
&\quad\cdot D_{et}^{\frac{1}{2}\sum_{j=1}^nm_j-n+1}\|\widetilde{C_o}\|_{l,0}^{n-1}\prod_{k=1}^n\|J_{m_k}\|_{l,0}.\notag
\end{align}
\item\label{item_tree_kernels_fixed}
In addition, assume that $k_j\in\{0,1,\cdots,m_j-1\}$ $(\forall j\in
     \{1,2,\cdots,$ $n\})$,
$p,q\in\{1,2,\cdots,n\}$, 
     $k_1,k_p,k_q\ge 1$, $r\in\{1,2,\cdots,k_q\}$,
     $j'\in\{0,1,\cdots,$ $d\}$ and $a\in\{0,1\}$. Then, for any
     $Y_{1,1}\in I$,    
\begin{align}\label{eq_tree_line_kernels}
&\Bigg|Ope(T,C_o)\left(\frac{1}{h}\right)^{m_1-1}\sum_{\bX_1\in
 I^{m_1-k_1}}\sum_{(Y_{1,2},Y_{1,3},\cdots,Y_{1,k_1})\in I^{k_1-1}}J_{m_1}(\bX_1,\bY_1)\\
&\quad\cdot\prod_{j=2}^n\Bigg(\left(\frac{1}{h}\right)^{m_j}\sum_{\bX_j\in
 I^{m_j-k_j}}\sum_{\bY_j\in I^{k_j}}J_{m_j}(\bX_j,\bY_j)\Bigg)\notag\\
&\quad\cdot d_{j'}(Y_{1,1},Y_{q,r})^ae^{\sum_{j=0}^d(\fw(l)d_j(Y_{1,1},Y_{p,1}))^{\fr}}\prod_{k=1\atop order}^n\psi_{\bX_k}^k\Bigg|_{\psi^j=0\atop (\forall j\in
 \{1,2,\cdots,n\})}\Bigg|\notag\\
&\le 1_{n_j(T)\le m_j-k_j (\forall j\in
 \{1,2,\cdots,n\})}2^{n-1}\notag\\
&\quad\cdot\prod_{i=1}^n\left((m_i-k_i)\left(\begin{array}{c}m_i-k_i-1\\
					     n_i(T)-1\end{array}\right)(n_i(T)-1)!\right)D_{et}^{\frac{1}{2}\sum_{j=1}^n(m_j-k_j)-n+1}\notag\\
&\quad\cdot\prod_{j=1}^n\Bigg(\sum_{q_j=0}^1\|J_{m_j}\|_{l,q_j}\Bigg)
\prod_{k=2}^n\Bigg(\sum_{r_k=0}^1\|\widetilde{C_o}\|_{l,r_k}\Bigg)
1_{\sum_{j=1}^nq_j+\sum_{k=2}^nr_k=a}.\notag
\end{align}
\end{enumerate}
\end{lemma}
\begin{proof}
\eqref{item_tree_kernels_no_fixed}: By Lemma
 \ref{lem_tree_line_expansion} we have that
\begin{align*}
&(\text{the left-hand side of \eqref{eq_tree_kernels_no_fixed}})\\
&\le 1_{n_j(T)\le m_j(\forall j\in
 \{1,2,\cdots,n\})}2^{n-1}D_{et}^{\frac{1}{2}\sum_{j=1}^nm_j-n+1}\notag\\
&\quad\cdot \frah^{m_1-1}\sum_{(X_{1,2},X_{1,3},\cdots,X_{1,m_1})\in I^{m_1-1}}|J_{m_1}(\bX_1)|
\sum_{\bW_1\subset \bX_1\atop \bW_1\in I^{n_1(T)}}
\sum_{\s_1\in
 \S_{n_1(T)}}\notag\\
&\quad\cdot\prod_{\{1,s\}\in L_1^1(T)}\Bigg(\frah^{m_s}\sum_{\bX_s\in
 I^{m_s}}|J_{m_s}(\bX_s)|
\sum_{Z_s\subset
 \bX_s\atop Z_s\in I}|\widetilde{C_o}(W_{1,\s_1\circ
 \zeta_1(\{1,s\})},Z_s)|\Bigg)\notag\\
&\quad \cdot \prod_{u=1\atop order
 }^{d_T(1)-1}\Bigg(\prod_{j\in\{2,3,\cdots,n\}\text{ with }\atop
 \dis_T(1,j)=u,n_j(T)\neq 1}\Bigg(\sum_{\bW_j\subset \bX_j\backslash
 Z_j\atop \bW_j\in
 I^{n_j(T)-1}}\sum_{\s_j\in\S_{n_j(T)-1}}\notag\\
&\quad\ \cdot \prod_{\{j,s\}\in L_j^1(T)}\Bigg(\frah^{m_s}\sum_{\bX_s\in
 I^{m_s}}|J_{m_s}(\bX_s)|\sum_{Z_s\subset
 \bX_s\atop Z_s\in I}|\widetilde{C_o}(W_{j,\s_j\circ\zeta_j(\{j,s\})},Z_s)|
\Bigg)
\Bigg)\Bigg).\notag
\end{align*}
Then, estimating recursively from $u=d_T(1)-1$ to $u=1$, we observe that 
\begin{align*}
&(\text{the left-hand side of \eqref{eq_tree_kernels_no_fixed}})\\
&\le 1_{n_j(T)\le m_j(\forall j\in
 \{1,2,\cdots,n\})}2^{n-1}D_{et}^{\frac{1}{2}\sum_{j=1}^nm_j-n+1}\\
&\quad\cdot \frah^{m_1-1}\sum_{(X_{1,2},X_{1,3},\cdots,X_{1,m_1})\in I^{m_1-1}}|J_{m_1}(\bX_1)|
\sum_{\bW_1\subset \bX_1\atop \bW_1\in I^{n_1(T)}}
\sum_{\s_1\in
 \S_{n_1(T)}}\\
&\quad\cdot\prod_{\{1,s\}\in L_1^1(T)}\Bigg(\frah^{m_s}\sum_{\bX_s\in
 I^{m_s}}|J_{m_s}(\bX_s)|
\sum_{Z_s\subset
 \bX_s\atop Z_s\in I}|\widetilde{C_o}(W_{1,\s_1\circ
 \zeta_1(\{1,s\})},Z_s)|\Bigg)\\
&\quad \cdot \prod_{u=1\atop order
 }^{d_T(1)-2}\Bigg(\prod_{j\in\{2,3,\cdots,n\}\text{ with }\atop
 \dis_T(1,j)=u,n_j(T)\neq 1}\Bigg(\sum_{\bW_j\subset \bX_j\backslash
 Z_j\atop \bW_j\in
 I^{n_j(T)-1}}\sum_{\s_j\in\S_{n_j(T)-1}}\\
&\qquad\cdot \prod_{\{j,s\}\in L_j^1(T)}\Bigg(\frah^{m_s}\sum_{\bX_s\in I^{m_s}}|J_{m_s}(\bX_s)|
\sum_{Z_s\subset
 \bX_s\atop Z_s\in I}|\widetilde{C_o}(W_{j,\s_j\circ\zeta_j(\{j,s\})},Z_s)|
\Bigg)
\Bigg)\Bigg)\\
&\quad\cdot\prod_{j\in\{2,3,\cdots,n\}\text{ with }\atop
 \dis_T(1,j)=d_T(1)-1,n_j(T)\neq 1}\Bigg(\left(\begin{array}{c}m_j-1\\
					       n_j(T)-1\end{array}\right)(n_j(T)-1)!\\
&\qquad\cdot\prod_{\{j,s\}\in
 L_j^1(T)}(m_s\|J_{m_s}\|_{l,0}\|\widetilde{C_o}\|_{l,0})\Bigg)\\
&\le 1_{n_j(T)\le m_j(\forall j\in
 \{1,2,\cdots,n\})}2^{n-1}D_{et}^{\frac{1}{2}\sum_{j=1}^nm_j-n+1}\\
&\quad\cdot \frah^{m_1-1}\sum_{(X_{1,2},X_{1,3},\cdots,X_{1,m_1})\in I^{m_1-1}}|J_{m_1}(\bX_1)|
\sum_{\bW_1\subset \bX_1\atop \bW_1\in I^{n_1(T)}}
\sum_{\s_1\in
 \S_{n_1(T)}}\\
&\quad\cdot\prod_{\{1,s\}\in L_1^1(T)}\Bigg(\frah^{m_s}\sum_{\bX_s\in
 I^{m_s}}|J_{m_s}(\bX_s)|
\sum_{Z_s\subset
 \bX_s\atop Z_s\in I}|\widetilde{C_o}(W_{1,\s_1\circ
 \zeta_1(\{1,s\})},Z_s)|\Bigg)\\
&\quad \cdot \prod_{u=1\atop order
 }^{v}\Bigg(\prod_{j\in\{2,3,\cdots,n\}\text{ with }\atop
 \dis_T(1,j)=u,n_j(T)\neq 1}\Bigg(\sum_{\bW_j\subset \bX_j\backslash
 Z_j\atop \bW_j\in
 I^{n_j(T)-1}}\sum_{\s_j\in\S_{n_j(T)-1}}\\
&\qquad\cdot \prod_{\{j,s\}\in L_j^1(T)}\Bigg(\frah^{m_s}\sum_{\bX_s\in I^{m_s}}|J_{m_s}(\bX_s)|
\sum_{Z_s\subset
 \bX_s\atop Z_s\in I}|\widetilde{C_o}(W_{j,\s_j\circ\zeta_j(\{j,s\})},Z_s)|
\Bigg)
\Bigg)\Bigg)\\
&\quad\cdot \prod_{w=v+1}^{d_T(1)-1}\Bigg(\prod_{j\in\{2,3,\cdots,n\}\text{ with }\atop
 \dis_T(1,j)=w,n_j(T)\neq 1}\Bigg(\left(\begin{array}{c}m_j-1\\
					       n_j(T)-1\end{array}\right)(n_j(T)-1)!\\
&\qquad\cdot\prod_{\{j,s\}\in
 L_j^1(T)}(m_s\|J_{m_s}\|_{l,0}\|\widetilde{C_o}\|_{l,0})\Bigg)\Bigg)\\
&\le 1_{n_j(T)\le m_j(\forall j\in
 \{1,2,\cdots,n\})}2^{n-1}D_{et}^{\frac{1}{2}\sum_{j=1}^nm_j-n+1}\\
&\quad\cdot \left(\begin{array}{c} m_1\\ n_1(T)
		  \end{array}\right)n_1(T)!\|J_{m_1}\|_{l,0}
\prod_{\{1,s\}\in
 L_1^1(T)}(m_s\|J_{m_s}\|_{l,0}\|\widetilde{C_o}\|_{l,0})\\
&\quad \cdot \prod_{u=1}^{d_T(1)-1}\Bigg(\prod_{j\in\{2,3,\cdots,n\}\text{ with }\atop
 \dis_T(1,j)=u,n_j(T)\neq 1}\Bigg(\left(\begin{array}{c}m_j-1\\
					       n_j(T)-1\end{array}\right)(n_j(T)-1)!\\
&\qquad\cdot \prod_{\{j,s\}\in
 L_j^1(T)}(m_s\|J_{m_s}\|_{l,0}\|\widetilde{C_o}\|_{l,0})\Bigg)\Bigg),
\end{align*}
which is equal to the right-hand side of \eqref{eq_tree_kernels_no_fixed}.

\eqref{item_tree_kernels_fixed}: Using Lemma \ref{lem_tree_line_expansion}
 and the anti-symmetric property of the functions $J_{m_j}(\cdot)$
 $(j=1,2,\cdots,n)$, we see that 
\begin{align}
&(\text{the left-hand side of \eqref{eq_tree_line_kernels}})\label{eq_tree_line_kernel_long}\\
&\le 1_{n_j(T)\le m_j-k_j (\forall j\in
 \{1,2,\cdots,n\})}2^{n-1}D_{et}^{\frac{1}{2}\sum_{j=1}^n(m_j-k_j)-n+1}\notag\\
&\quad\cdot \left(\frac{1}{h}\right)^{m_1-1}\sum_{\bX_1\in
 I^{m_1-k_1}}\sum_{(Y_{1,2},Y_{1,3},\cdots,Y_{1,k_1})\in
 I^{k_1-1}}|J_{m_1}(\bX_1,\bY_1)|\notag\\
&\quad\cdot\sum_{\bW_1\subset \bX_1\atop \bW_1\in
 I^{n_1(T)}}\sum_{\s_1\in
 \S_{n_1(T)}}\prod_{\{1,s\}\in
 L_1^1(T)}\Bigg(\left(\frac{1}{h}\right)^{m_s}\sum_{\bX_s\in
 I^{m_s-k_s}}\sum_{\bY_s\in I^{k_s}}|J_{m_s}(\bX_s,\bY_s)|\notag\\
&\qquad\qquad\qquad\cdot \sum_{Z_s\subset \bX_s\atop Z_s\in I}|\widetilde{C_o}(W_{1,\s_1\circ \zeta_1(\{1,s\})},Z_s)|\Bigg)\notag\\
&\quad \cdot \prod_{u=1\atop order
 }^{d_T(1)-1}\Bigg(\prod_{j\in\{2,3,\cdots,n\}\text{ with }\atop
 \dis_T(1,j)=u,n_j(T)\neq 1}\Bigg(\sum_{\bW_j\subset \bX_j\backslash
 Z_j\atop \bW_j\in
 I^{n_j(T)-1}}\sum_{\s_j\in\S_{n_j(T)-1}}\notag\\
&\qquad\cdot \prod_{\{j,s\}\in L_j^1(T)}\Bigg(\left(\frac{1}{h}\right)^{m_s}
\sum_{\bX_s\in I^{m_s-k_s}}\sum_{\bY_s\in I^{k_s}}|J_{m_s}(\bX_s,\bY_s)|\notag
\\
&\qquad\qquad\qquad\cdot\sum_{Z_s\subset \bX_s\atop Z_s\in I}|\widetilde{C_o}(W_{j,\s_j\circ\zeta_j(\{j,s\})},Z_s)|
\Bigg)
\Bigg)\Bigg)\notag\\
&\quad \cdot d_{j'}(Y_{1,1},Y_{q,r})^ae^{\sum_{j=0}^d(\fw(l)d_j(Y_{1,1},Y_{p,1}))^{\fr}}\notag\\
&= 1_{n_j(T)\le m_j-k_j (\forall j\in
 \{1,2,\cdots,n\})}2^{n-1}D_{et}^{\frac{1}{2}\sum_{j=1}^n(m_j-k_j)-n+1}\notag\\
&\quad\cdot \left(\frac{1}{h}\right)^{m_1-1}\sum_{\bX_1\in
 I^{m_1-k_1-n_1(T)}}\sum_{\bW_1\in I^{n_1(T)}}
\sum_{(Y_{1,2},Y_{1,3},\cdots,Y_{1,k_1})\in
 I^{k_1-1}}\notag\\
&\quad\cdot |J_{m_1}(\bX_1,\bW_1,\bY_1)| \left(\begin{array}{c}m_1-k_1\\
					    n_1(T)\end{array}\right) n_1(T)!\notag\\
&\quad\cdot\prod_{\{1,s\}\in
 L_1^1(T)}\Bigg((m_s-k_s)
\left(\frac{1}{h}\right)^{m_s}\sum_{\bX_s\in
 I^{m_s-k_s-n_s(T)}}\sum_{\bW_s\in I^{n_s(T)-1}}\sum_{Z_s\in
 I}\sum_{\bY_s\in I^{k_s}}\notag\\
&\qquad\cdot |J_{m_s}(\bX_s,\bW_s,Z_s,\bY_s)||\widetilde{C_o}(W_{1,\zeta_1(\{1,s\})},Z_s)|\Bigg)\notag\\
&\quad \cdot \prod_{u=1\atop order
 }^{d_T(1)-1}\Bigg(\prod_{j\in\{2,3,\cdots,n\}\text{ with }\atop
 \dis_T(1,j)=u,n_j(T)\neq 1}\Bigg(\left(\begin{array}{c}m_j-k_j-1\\
					n_j(T)-1\end{array}\right)(n_j(T)-1)!\notag\\
&\qquad\cdot \prod_{\{j,s\}\in L_j^1(T)}\Bigg((m_s-k_s)
\left(\frac{1}{h}\right)^{m_s}
\sum_{\bX_s\in I^{m_s-k_s-n_s(T)}}\sum_{\bW_s\in
 I^{n_s(T)-1}}\sum_{Z_s\in I}\sum_{\bY_s\in
 I^{k_s}}\notag\\
&\qquad\quad\cdot |J_{m_s}(\bX_s,\bW_s,Z_s,\bY_s)|
|\widetilde{C_o}(W_{j,\zeta_j(\{j,s\})},Z_s)|
\Bigg)
\Bigg)\Bigg)\notag\\
&\quad \cdot d_{j'}(Y_{1,1},Y_{q,r})^ae^{\sum_{j=0}^d(\fw(l)d_j(Y_{1,1},Y_{p,1}))^{\fr}}\notag\\
&=1_{n_j(T)\le m_j-k_j (\forall j\in
 \{1,2,\cdots,n\})}2^{n-1}D_{et}^{\frac{1}{2}\sum_{j=1}^n(m_j-k_j)-n+1}\notag\\
&\quad\cdot
 \prod_{j=1}^n\Bigg((m_j-k_j)\left(\begin{array}{c}m_j-k_j-1\\
				   n_j(T)-1\end{array}\right)(n_j(T)-1)!\Bigg)\notag\\
&\quad\cdot \left(\frac{1}{h}\right)^{m_1-1}\sum_{\bX_1\in
 I^{m_1-k_1-n_1(T)}}\sum_{\bW_1\in I^{n_1(T)}}
\sum_{(Y_{1,2},Y_{1,3},\cdots,Y_{1,k_1})\in
 I^{k_1-1}}|J_{m_1}(\bX_1,\bW_1,\bY_1)|\notag\\
&\quad\cdot\prod_{\{1,s\}\in
 L_1^1(T)}\Bigg(\left(\frac{1}{h}\right)^{m_s}\sum_{\bX_s\in
 I^{m_s-k_s-n_s(T)}}\sum_{\bW_s\in I^{n_s(T)-1}}\sum_{Z_s\in
 I}\sum_{\bY_s\in I^{k_s}}\notag\\
&\qquad\quad\cdot |J_{m_s}(\bX_s,\bW_s,Z_s,\bY_s)||\widetilde{C_o}(W_{1,\zeta_1(\{1,s\})},Z_s)|\Bigg)\notag\\
&\quad \cdot \prod_{u=1\atop order
 }^{d_T(1)-1}\Bigg(\prod_{j\in\{2,3,\cdots,n\}\text{ with }\atop
 \dis_T(1,j)=u,n_j(T)\neq 1}\notag\\
&\qquad\cdot \prod_{\{j,s\}\in L_j^1(T)}\Bigg(\left(\frac{1}{h}\right)^{m_s}
\sum_{\bX_s\in I^{m_s-k_s-n_s(T)}}\sum_{\bW_s\in
 I^{n_s(T)-1}}\sum_{Z_s\in I}\sum_{\bY_s\in
 I^{k_s}}\notag\\
&\qquad\quad\cdot |J_{m_s}(\bX_s,\bW_s,Z_s,\bY_s)|
|\widetilde{C_o}(W_{j,\zeta_j(\{j,s\})},Z_s)|
\Bigg)
\Bigg)\notag\\
&\quad \cdot
 d_{j'}(Y_{1,1},Y_{q,r})^ae^{\sum_{j=0}^d(\fw(l)d_j(Y_{1,1},Y_{p,1}))^{\fr}}.\notag
\end{align}
Then, we can apply Lemma \ref{lem_tree_line_estimate_induction}, which
 will be proved next, to derive the claimed inequality.
\end{proof}

\begin{lemma}\label{lem_tree_line_estimate_induction}
On the same assumption as in Lemma \ref{lem_tree_line_kernels}
 \eqref{item_tree_kernels_fixed} plus that $n_j(T)\le m_j-k_j$
 $(\forall j\in \{1,2,\cdots,n\})$, the following inequality holds.
\begin{align}
&\left(\frac{1}{h}\right)^{m_1-1}\sum_{\bX_1\in
 I^{m_1-k_1-n_1(T)}}\sum_{\bW_1\in I^{n_1(T)}}
\sum_{(Y_{1,2},Y_{1,3},\cdots,Y_{1,k_1})\in
 I^{k_1-1}}|J_{m_1}(\bX_1,\bW_1,\bY_1)|\label{eq_tree_line_estimate_induction}\\
&\quad\cdot\prod_{\{1,s\}\in
 L_1^1(T)}\Bigg(\left(\frac{1}{h}\right)^{m_s}\sum_{\bX_s\in
 I^{m_s-k_s-n_s(T)}}\sum_{\bW_s\in I^{n_s(T)-1}}\sum_{Z_s\in
 I}\sum_{\bY_s\in I^{k_s}}\notag\\
&\qquad\cdot |J_{m_s}(\bX_s,\bW_s,Z_s,\bY_s)||\widetilde{C_o}(W_{1,\zeta_1(\{1,s\})},Z_s)|\Bigg)\notag\\
&\quad \cdot \prod_{u=1\atop order
 }^{d_T(1)-1}\Bigg(\prod_{j\in\{2,3,\cdots,n\}\text{ with }\atop
 \dis_T(1,j)=u,n_j(T)\neq 1}\notag\\
&\qquad\cdot \prod_{\{j,s\}\in L_j^1(T)}\Bigg(\left(\frac{1}{h}\right)^{m_s}
\sum_{\bX_s\in I^{m_s-k_s-n_s(T)}}\sum_{\bW_s\in
 I^{n_s(T)-1}}\sum_{Z_s\in I}\sum_{\bY_s\in
 I^{k_s}}\notag\\
&\qquad\quad\cdot |J_{m_s}(\bX_s,\bW_s,Z_s,\bY_s)|
|\widetilde{C_o}(W_{j,\zeta_j(\{j,s\})},Z_s)|
\Bigg)
\Bigg)\notag\\
&\quad \cdot
 d_{j'}(Y_{1,1},Y_{q,r})^ae^{\sum_{j=0}^d(\fw(l)d_j(Y_{1,1},Y_{p,1}))^{\fr}}\notag\\
&\le
 \prod_{j=1}^n\left(\sum_{q_j=0}^1\|J_{m_j}\|_{l,q_j}\right)\prod_{k=2}^n\left(
\sum_{r_k=0}^1\|\widetilde{C_o}\|_{l,r_k}\right)1_{\sum_{j=1}^nq_j+\sum_{k=2}^nr_k=a}.\notag
\end{align}
\end{lemma}
\begin{proof}
We prove the claimed inequality by induction with $n$. In the following
 we will repeatedly use the inequality $(x+y)^{\fr}\le x^{\fr}+y^{\fr}$,
 $(\forall x,y\in\R_{\ge 0})$. If $n=2$,
\begin{align*}
&(\text{the left-hand side of
 \eqref{eq_tree_line_estimate_induction}})\\
&= \left(\frac{1}{h}\right)^{m_1-1}\sum_{\bX_1\in
 I^{m_1-k_1-1}}\sum_{W_1\in I}\sum_{(Y_{1,2},Y_{1,3},\cdots,Y_{1,k_1})\in
 I^{k_1-1}}|J_{m_1}(\bX_1,W_1,\bY_1)|\\
&\quad\cdot \left(\frac{1}{h}\right)^{m_2}\sum_{\bX_2\in
 I^{m_2-k_2-1}}\sum_{Z_2\in I}\sum_{\bY_2\in
 I^{k_2}}|J_{m_2}(\bX_2,Z_2,\bY_2)||\widetilde{C_o}(W_{1},Z_2)|\\
&\quad \cdot
 d_{j'}(Y_{1,1},Y_{q,r})^ae^{\sum_{j=0}^d(\fw(l)d_j(Y_{1,1},Y_{p,1}))^{\fr}}\\
&\le \left(\frac{1}{h}\right)^{m_1-1}\sum_{\bX_1\in
 I^{m_1-k_1-1}}\sum_{W_1\in I}\sum_{(Y_{1,2},Y_{1,3},\cdots,Y_{1,k_1})\in
 I^{k_1-1}}|J_{m_1}(\bX_1,W_1,\bY_1)|\\
&\quad\cdot \left(\frac{1}{h}\right)^{m_2}\sum_{\bX_2\in
 I^{m_2-k_2-1}}\sum_{Z_2\in I}\sum_{\bY_2\in
 I^{k_2}}|J_{m_2}(\bX_2,Z_2,\bY_2)||\widetilde{C_o}(W_{1},Z_2)|\\
&\quad\cdot \Bigg(
 1_{q=1}\sum_{q_1=0}^1d_{j'}(Y_{1,1},Y_{1,r})^{q_1}1_{q_1=a}\\
&\qquad+1_{q=2}\sum_{q_1=0}^1d_{j'}(Y_{1,1},W_{1})^{q_1}\sum_{r_2=0}^1d_{j'}(W_1,Z_2)^{r_2}\\
&\qquad\qquad\quad\cdot \sum_{q_2=0}^1d_{j'}(Z_2,Y_{2,r})^{q_2}1_{q_1+q_2+r_2=a}\Bigg)\\
&\quad\cdot
 \Bigg(1_{p=1}+1_{p=2}
e^{\sum_{j=0}^d(\fw(l)d_j(Y_{1,1},W_1))^{\fr}}
e^{\sum_{j=0}^d(\fw(l)d_j(W_1,Z_2))^{\fr}}
e^{\sum_{j=0}^d(\fw(l)d_j(Z_2,Y_{2,1}))^{\fr}}\Bigg)\\
&\le
 1_{q=1}\sum_{q_1=0}^1\|J_{m_1}\|_{l,q_1}\|\widetilde{C_0}\|_{l,0}\|J_{m_2}\|_{l,0}1_{q_1=a}\\
&\quad +
 1_{q=2}\sum_{q_1=0}^1\|J_{m_1}\|_{l,q_1}\sum_{r_2=0}^1\|\widetilde{C_0}\|_{l,r_2}\sum_{q_2=0}^1\|J_{m_2}\|_{l,q_2}1_{q_1+q_2+r_2=a},
\end{align*}
which is less than or equal to the right-hand side of
 \eqref{eq_tree_line_estimate_induction} for $n=2$.

Assume that the claim holds for some $n\in\N_{\ge 2}$. Let us estimate
 the left-hand side of \eqref{eq_tree_line_estimate_induction} for
 $n+1$. Take a vertex $\hat{s}\in \{1,2,\cdots,n+1\}$ satisfying 
 $\dis_T(1,\hat{s})=d_T(1)$. Take $\s\in\S_{n+1}$
 satisfying $\s(1)=1$, $\s(\hat{s})=n+1$. Then, define the tree
 $T'\in\T_{n+1}$ by $T':=\{\{\s(j),\s(s)\}\ |\ \{j,s\}\in T\}$. 
Note that $\dis_{T'}(1,n+1)=d_{T'}(1)$. Setting $m_j':=m_{\s^{-1}(j)}$, $k_j':=k_{\s^{-1}(j)}$
 $(j=1,2,\cdots,n+1)$, we see that $n_j(T')\le m_j'-k_j'$ $(\forall
 j\in\{1,2,\cdots,n+1\})$ and 
\begin{align*}
&(\text{the left-hand side of
 \eqref{eq_tree_line_estimate_induction}})\\
&=\left(\frac{1}{h}\right)^{m_1'-1}\sum_{\bX_1\in
 I^{m_1'-k_1'-n_1(T')}}\sum_{\bW_1\in I^{n_1(T')}}
\sum_{(Y_{1,2},Y_{1,3},\cdots,Y_{1,k_1'})\in
 I^{k_1'-1}}|J_{m_1'}(\bX_1,\bW_1,\bY_1)|\\
&\quad\cdot\prod_{\{1,s\}\in
 L_1^1(T')}\Bigg(\left(\frac{1}{h}\right)^{m_s'}\sum_{\bX_s\in
 I^{m_s'-k_s'-n_s(T')}}\sum_{\bW_s\in I^{n_s(T')-1}}\sum_{Z_s\in
 I}\sum_{\bY_s\in I^{k_s'}}\\
&\qquad\cdot |J_{m_s'}(\bX_s,\bW_s,Z_s,\bY_s)||\widetilde{C_o}(W_{1,\zeta_{\s^{-1}(1)}(\{\s^{-1}(1),\s^{-1}(s)\})},Z_s)|\Bigg)\notag\\
&\quad \cdot \prod_{u=1\atop order
 }^{d_{T'}(1)-1}\Bigg(\prod_{j\in\{2,3,\cdots,n\}\text{ with }\atop
 \dis_{T'}(1,j)=u,n_j(T')\neq 1}\notag\\
&\qquad\cdot \prod_{\{j,s\}\in L_j^1(T')}\Bigg(\left(\frac{1}{h}\right)^{m_s'}
\sum_{\bX_s\in I^{m_s'-k_s'-n_s(T')}}\sum_{\bW_s\in
 I^{n_s(T')-1}}\sum_{Z_s\in I}\sum_{\bY_s\in
 I^{k_s'}}\notag\\
&\qquad\quad\cdot |J_{m_s'}(\bX_s,\bW_s,Z_s,\bY_s)|
|\widetilde{C_o}(W_{j,\zeta_{\s^{-1}(j)}(\{\s^{-1}(j),\s^{-1}(s)\})},Z_s)|
\Bigg)
\Bigg)\notag\\
&\quad \cdot
 d_{j'}(Y_{1,1},Y_{\s(q),r})^ae^{\sum_{j=0}^d(\fw(l)d_j(Y_{1,1},Y_{\s(p),1}))^{\fr}}.
\end{align*}
 Let $\zeta_j'$ denote the bijective
 map from $L_j^1(T')$ to $\{1,2,\cdots,\sharp L_j^1(T')\}$ satisfying
 that $\zeta_{j}'(\{j,r\})<\zeta_j'(\{j,s\})$, $(\forall \{j,r\},\{j,s\}\in
 L_j^1(T')\text{ with }r<s)$. By the anti-symmetry of
 $J_{m_j'}(\bX_j,\bW_j,Z_j,\bY_j)$ with respect to the variable $\bW_j$ we can
replace $\zeta_{\s^{-1}(j)}(\{\s^{-1}(j),\s^{-1}(s)\})$ by
 $\zeta_j'(\{j,s\})$ in the right-hand side of the above equality. This argument implies that we may assume
 $\dis_T(1,n+1)=d_T(1)$ in the left-hand side of
 \eqref{eq_tree_line_estimate_induction} without losing
 generality. Thus, we assume so in the following.

If $\{\tilde{j},n+1\}\in T$, 
\begin{align}
&d_{j'}(Y_{1,1},Y_{q,r})^ae^{\sum_{j=0}^d(\fw(l)d_j(Y_{1,1},Y_{p,1}))^{\fr}}\label{eq_weight_decomposition_2_vertices}\\
&\le  \Bigg(
 1_{q\le n}d_{j'}(Y_{1,1},Y_{q,r})^{a}\notag\\
&\qquad+1_{q=n+1}\sum_{q_{\tilde{j}}=0}^1d_{j'}(Y_{1,1},W_{\tilde{j},\zeta_{\tilde{j}}(\{\tilde{j},n+1\})})^{q_{\tilde{j}}}\sum_{r_{n+1}=0}^1d_{j'}(W_{\tilde{j},\zeta_{\tilde{j}}(\{\tilde{j},n+1\})},Z_{n+1})^{r_{n+1}}\notag\\
&\qquad\qquad\qquad\cdot\sum_{q_{n+1}=0}^1d_{j'}(Z_{n+1},Y_{n+1,r})^{q_{n+1}}1_{q_{\tilde{j}}+q_{n+1}+r_{n+1}=a}\Bigg)\notag\\
&\quad\cdot\Big(1_{p\le
 n}e^{\sum_{j=0}^d(\fw(l)d_j(Y_{1,1},Y_{p,1}))^{\fr}}\notag\\
&\qquad
 +1_{p=n+1}e^{\sum_{j=0}^d(\fw(l)d_j(Y_{1,1},W_{\tilde{j},\zeta_{\tilde{j}}(\{\tilde{j},n+1\})}))^{\fr}}e^{\sum_{j=0}^d(\fw(l)d_j(W_{\tilde{j},\zeta_{\tilde{j}}(\{\tilde{j},n+1\})}, Z_{n+1}))^{\fr}}\notag\\
&\qquad\qquad\qquad\cdot e^{\sum_{j=0}^d(\fw(l)d_j(Z_{n+1},Y_{n+1,1}))^{\fr}}\Big).\notag
\end{align}
Set $\tilde{T}:=T\backslash\{\{\tilde{j},n+1\}\}$, $\tilde{k}_j:=k_j$
 $(\forall j\in\{1,2,\cdots,n\}\backslash \{\tilde{j}\})$,
 $\tilde{k}_{\tilde{j}}:=k_{\tilde{j}}+1$. It follows that
 $\tilde{k}_j\in\{0,1,\cdots,m_j-1\}$, $n_j(\tilde{T})\le
 m_j-\tilde{k}_j$, $(\forall j\in\{1,2,\cdots,n\})$. Substitution of
 \eqref{eq_weight_decomposition_2_vertices} yields that 
\begin{align}
&(\text{the left-hand side of
 \eqref{eq_tree_line_estimate_induction}})\label{eq_tree_line_estimate_induction_final}\\
&\le \left(\frac{1}{h}\right)^{m_1-1}\sum_{\bX_1\in
 I^{m_1-k_1-n_1(T)}}\sum_{\bW_1\in I^{n_1(T)}}
\sum_{(Y_{1,2},Y_{1,3},\cdots,Y_{1,k_1})\in
 I^{k_1-1}}|J_{m_1}(\bX_1,\bW_1,\bY_1)|\notag\\
&\quad\cdot\prod_{\{1,s\}\in
 L_1^1(T)}\Bigg(\left(\frac{1}{h}\right)^{m_s}\sum_{\bX_s\in
 I^{m_s-k_s-n_s(T)}}\sum_{\bW_s\in I^{n_s(T)-1}}\sum_{Z_s\in
 I}\sum_{\bY_s\in I^{k_s}}\notag\\
&\qquad\cdot |J_{m_s}(\bX_s,\bW_s,Z_s,\bY_s)||\widetilde{C_o}(W_{1,\zeta_1(\{1,s\})},Z_s)|\Bigg)\notag\\
&\quad \cdot \prod_{u=1\atop order
 }^{d_T(1)-1}\Bigg(\prod_{j\in\{2,3,\cdots,n\}\text{ with }\atop
 \dis_T(1,j)=u,n_j(T)\neq 1}\notag\\
&\qquad\cdot \prod_{\{j,s\}\in L_j^1(T)\backslash\{\{\tilde{j},n+1\}\}}\Bigg(\left(\frac{1}{h}\right)^{m_s}
\sum_{\bX_s\in I^{m_s-k_s-n_s(T)}}\sum_{\bW_s\in
 I^{n_s(T)-1}}\sum_{Z_s\in I}\sum_{\bY_s\in
 I^{k_s}}\notag\\
&\qquad\quad\cdot |J_{m_s}(\bX_s,\bW_s,Z_s,\bY_s)|
|\widetilde{C_o}(W_{j,\zeta_j(\{j,s\})},Z_s)|
\Bigg)
\Bigg)\notag\\
&\quad \cdot\Bigg(1_{q\le n}1_{p\le
 n}d_{j'}(Y_{1,1},Y_{q,r})^{a}e^{\sum_{j=0}^d(\fw(l)d_j(Y_{1,1},Y_{p,1}))^{\fr}}\|\widetilde{C_o}\|_{l,0}\|J_{m_{n+1}}\|_{l,0}\notag\\
&\qquad + 1_{q\le n}1_{p=
 n+1}d_{j'}(Y_{1,1},Y_{q,r})^{a}e^{\sum_{j=0}^d(\fw(l)d_j(Y_{1,1},W_{\tilde{j},\zeta_{\tilde{j}}(\{\tilde{j},n+1\})}))^{\fr}}\|\widetilde{C_o}\|_{l,0}\|J_{m_{n+1}}\|_{l,0}\notag\\
&\qquad +1_{q= n+1}1_{p\le
 n}\sum_{q_{\tilde{j}}=0}^1d_{j'}(Y_{1,1},W_{\tilde{j},\zeta_{\tilde{j}}(\{\tilde{j},n+1\})})^{q_{\tilde{j}}}e^{\sum_{j=0}^d(\fw(l)d_j(Y_{1,1},Y_{p,1}))^{\fr}}\notag\\
&\qquad\quad\cdot\sum_{r_{n+1}=0}^1\|\widetilde{C_o}\|_{l,r_{n+1}}\sum_{q_{n+1}=0}^1\|J_{m_{n+1}}\|_{l,q_{n+1}}1_{q_{\tilde{j}}+q_{n+1}+r_{n+1}=a}\notag\\
&\qquad +1_{q=
 n+1}1_{p=n+1}\sum_{q_{\tilde{j}}=0}^1d_{j'}(Y_{1,1},W_{\tilde{j},\zeta_{\tilde{j}}(\{\tilde{j},n+1\})})^{q_{\tilde{j}}}e^{\sum_{j=0}^d(\fw(l)d_j(Y_{1,1},W_{\tilde{j},\zeta_{\tilde{j}}(\{\tilde{j},n+1\})}))^{\fr}}\notag\\
&\qquad\quad\cdot\sum_{r_{n+1}=0}^1\|\widetilde{C_o}\|_{l,r_{n+1}}\sum_{q_{n+1}=0}^1\|J_{m_{n+1}}\|_{l,q_{n+1}}1_{q_{\tilde{j}}+q_{n+1}+r_{n+1}=a}\Bigg)\notag\\
&= \left(\frac{1}{h}\right)^{m_1-1}\sum_{\bX_1\in
 I^{m_1-\tilde{k}_1-n_1(\tilde{T})}}\sum_{\bW_1\in I^{n_1(\tilde{T})}}
\sum_{(Y_{1,2},Y_{1,3},\cdots,Y_{1,\tilde{k}_1})\in
 I^{\tilde{k}_1-1}}|J_{m_1}(\bX_1,\bW_1,\bY_1)|\notag\\
&\quad\cdot\prod_{\{1,s\}\in
 L_1^1(\tilde{T})}\Bigg(\left(\frac{1}{h}\right)^{m_s}\sum_{\bX_s\in
 I^{m_s-\tilde{k}_s-n_s(\tilde{T})}}\sum_{\bW_s\in I^{n_s(\tilde{T})-1}}\sum_{Z_s\in
 I}\sum_{\bY_s\in I^{\tilde{k}_s}}\notag\\
&\qquad\cdot |J_{m_s}(\bX_s,\bW_s,Z_s,\bY_s)||\widetilde{C_o}(W_{1,\zeta_1(\{1,s\})},Z_s)|\Bigg)\notag\\
&\quad \cdot \prod_{u=1\atop order
 }^{d_{\tilde{T}}(1)-1}\Bigg(\prod_{j\in\{2,3,\cdots,n\}\text{ with }\atop
 \dis_{\tilde{T}}(1,j)=u,n_j(\tilde{T})\neq 1}\notag\\
&\qquad\cdot \prod_{\{j,s\}\in L_j^1(\tilde{T})}\Bigg(\left(\frac{1}{h}\right)^{m_s}
\sum_{\bX_s\in I^{m_s-\tilde{k}_s-n_s(\tilde{T})}}\sum_{\bW_s\in
 I^{n_s(\tilde{T})-1}}\sum_{Z_s\in I}\sum_{\bY_s\in
 I^{\tilde{k}_s}}\notag\\
&\qquad\cdot |J_{m_s}(\bX_s,\bW_s,Z_s,\bY_s)|
|\widetilde{C_o}(W_{j,\zeta_j(\{j,s\})},Z_s)|
\Bigg)
\Bigg)\notag\\
&\quad \cdot\Bigg(1_{q\le n}1_{p\le
 n}d_{j'}(Y_{1,1},Y_{q,r})^{a}e^{\sum_{j=0}^d(\fw(l)d_j(Y_{1,1},Y_{p,1}))^{\fr}}\|\widetilde{C_o}\|_{l,0}\|J_{m_{n+1}}\|_{l,0}\notag\\
&\qquad + 1_{q\le n}1_{p=
 n+1}d_{j'}(Y_{1,1},Y_{q,r'})^{a}e^{\sum_{j=0}^d(\fw(l)d_j(Y_{1,1},Y_{\tilde{j},1}))^{\fr}}\|\widetilde{C_o}\|_{l,0}\|J_{m_{n+1}}\|_{l,0}\notag\\
&\qquad +1_{q= n+1}1_{p\le
 n}\sum_{q_{\tilde{j}}=0}^1d_{j'}(Y_{1,1},Y_{\tilde{j},r''})^{q_{\tilde{j}}}e^{\sum_{j=0}^d(\fw(l)d_j(Y_{1,1},Y_{p,1}))^{\fr}}\notag\\
&\qquad\quad\cdot\sum_{r_{n+1}=0}^1\|\widetilde{C_o}\|_{l,r_{n+1}}\sum_{q_{n+1}=0}^1\|J_{m_{n+1}}\|_{l,q_{n+1}}1_{q_{\tilde{j}}+q_{n+1}+r_{n+1}=a}\notag\\
&\qquad +1_{q=
 n+1}1_{p=n+1}\sum_{q_{\tilde{j}}=0}^1d_{j'}(Y_{1,1},Y_{\tilde{j},1})^{q_{\tilde{j}}}e^{\sum_{j=0}^d(\fw(l)d_j(Y_{1,1},Y_{\tilde{j},1}))^{\fr}}\notag\\
&\qquad\quad\cdot\sum_{r_{n+1}=0}^1\|\widetilde{C_o}\|_{l,r_{n+1}}\sum_{q_{n+1}=0}^1\|J_{m_{n+1}}\|_{l,q_{n+1}}1_{q_{\tilde{j}}+q_{n+1}+r_{n+1}=a}\Bigg),\notag
\end{align}
where we used the anti-symmetry of $J_{m_{\tilde{j}}}(\cdot)$ to shift
 the variable \\
$W_{\tilde{j},\zeta_{\tilde{j}}(\{\tilde{j},n+1\})}$ to be
 in front of $\bY_{\tilde{j}}$ (or behind $\bY_{\tilde{j}}$) and replaced \\
 $(W_{\tilde{j},\zeta_{\tilde{j}}(\{\tilde{j},n+1\})},\bY_{\tilde{j}})$
 (or $(\bY_{\tilde{j}},W_{\tilde{j},\zeta_{\tilde{j}}(\{\tilde{j},n+1\})})$)
 by $\bY_{\tilde{j}}\in I^{\tilde{k}_{\tilde{j}}}$. Because of this
 change of the variable, the component inside $d_{j'}(\cdot,\cdot)$ may
 be changed from the original one. We used the numbers
 $r'\in\{1,2,\cdots,\tilde{k}_q\}$,
 $r''\in\{1,2,\cdots,\tilde{k}_{\tilde{j}}\}$ to represent the possible
 new components. Now, we can apply the hypothesis of induction to derive from
 \eqref{eq_tree_line_estimate_induction_final} that 
\begin{align*}
&(\text{the left-hand side of
 \eqref{eq_tree_line_estimate_induction}})\\  
&\le
 \prod_{j=1}^n\left(\sum_{q_j=0}^1\|J_{m_j}\|_{l,q_j}\right)\prod_{k=2}^n\left(
\sum_{r_k=0}^1\|\widetilde{C_o}\|_{l,r_k}\right)\\
&\quad\cdot\Bigg(1_{q\le
 n}1_{\sum_{j=1}^nq_j+\sum_{k=2}^nr_k=a}\|\widetilde{C_o}\|_{l,0}\|J_{m_{n+1}}\|_{l,0}\\
&\qquad+1_{q=n+1}\sum_{q_{\tilde{j}}=0}^11_{\sum_{j=1}^nq_j+\sum_{k=2}^nr_k=q_{\tilde{j}}}\\
&\qquad\quad\cdot\sum_{r_{n+1}=0}^1\|\widetilde{C_o}\|_{l,r_{n+1}}\sum_{q_{n+1}=0}^1
\|J_{m_{n+1}}\|_{l,q_{n+1}}1_{q_{\tilde{j}}+q_{n+1}+r_{n+1}=a}\Bigg),
\end{align*}
which is less than or equal to the right-hand side of
 \eqref{eq_tree_line_estimate_induction} for $n+1$. Thus, the induction
 concludes that the
 inequality \eqref{eq_tree_line_estimate_induction} holds for all $n\in
 \N_{\ge 2}$.  
\end{proof}

In order to deal with combinatorial factors in the tree expansion, we
use the following concise estimate, though it is not
quantitatively optimal.
\begin{lemma}\label{lem_tree_combinatorial_factor}
For any $m_j\in \N$ $(j=1,2,\cdots,n)$ the following inequality holds.
\begin{align}
&\frac{2^{n-1}}{n!}\sum_{T\in \T_n} 1_{n_j(T)\le m_j (\forall j\in
 \{1,2,\cdots,n\})}\prod_{j=1}^n\left(m_j\left(\begin{array}{c}m_j-1\\
					       n_j(T)-1\end{array}\right)(n_j(T)-1)!\right)\label{eq_tree_combinatorial_factor}\\
&\le 2^{2\sum_{j=1}^nm_j}.\notag
\end{align}
\end{lemma}
\begin{proof}
By Cayley's theorem on the number of trees with fixed incidence
 numbers we can replace the sum over $T\in \T_n$ by the sum over
 possible incidence numbers. As the result, we have that
\begin{align*}
&\text{ (the left-hand side of \eqref{eq_tree_combinatorial_factor}) }\\
&
 =\frac{2^{n-1}}{n!}\prod_{i=1}^n\left(\sum_{l_i=1}^{m_i}\right)1_{\sum_{i=1}^nl_i=2(n-1)}\frac{(n-2)!}{\prod_{k=1}^n(l_k-1)!}\\
&\quad\cdot\prod_{j=1}^n\left(m_j\left(\begin{array}{c}m_j-1\\l_j-1\end{array}\right)(l_j-1)!\right)\\
&\le \frac{2^{n-1}}{n(n-1)}\prod_{j=1}^n\left(m_j\sum_{l_j=1}^{m_j}\left(\begin{array}{c}m_j-1\\l_j-1\end{array}\right)\right)\le  2^{2\sum_{j=1}^nm_j}.
\end{align*}
\end{proof}

\begin{lemma}\label{lem_tree_sum_kernels}
Take any $m_j\in \N_{\ge 2}$ $(j=1,2,\cdots,n)$. Let $J_{m_j}:I^{m_j}\to
 \C$ $(j=1,2,\cdots,n)$ be anti-symmetric functions. Then, the following
 inequalities hold.
\begin{enumerate}
\item\label{item_tree_sum_kernels_no_fixed} For any $X_{1,1}\in I$,
\begin{align*}
&\Bigg|\frac{1}{n!}\sum_{T\in
 \T_n}Ope(T,C_o)\left(\frac{1}{h}\right)^{m_1-1}\sum_{(X_{1,2},X_{1,3},\cdots,X_{1,m_1})\in
 I^{m_1-1}}J_{m_1}(\bX_1)\\
&\quad\cdot\prod_{j=2}^n
\Bigg(\left(\frac{1}{h}\right)^{m_j}\sum_{\bX_j\in
 I^{m_j}}J_{m_j}(\bX_j)\Bigg)\prod_{k=1\atop
 order}^n\psi_{\bX_k}^k\Bigg|_{\psi^j=0\atop (\forall j\in
 \{1,2,\cdots,n\})}\Bigg|\\
&\le 2^{2\sum_{j=1}^nm_j}D_{et}^{\frac{1}{2}\sum_{j=1}^nm_j-n+1}
\|\widetilde{C_o}\|_{l,0}^{n-1}\prod_{j=1}^n\|J_{m_j}\|_{l,0}.
\end{align*}
\item\label{item_tree_sum_kernels_fixed}
In addition, assume that $k_j\in\{0,1,\cdots,m_j-1\}$ $(\forall
     j\in\{1,2,\cdots,$ $n\})$, $p,q\in\{1,2,\cdots,n\}$,
      $k_1,k_p,k_q\ge 1$,
     $r\in\{1,2,\cdots,k_q\}$, $j'\in\{0,1,\cdots,$ $d\}$ and $a\in\{0,1\}$. Then, for any $Y_{1,1}\in I$,
\begin{align*}
&\Bigg|\frac{1}{n!}\sum_{T\in \T_n}Ope(T,C_o)\left(\frac{1}{h}\right)^{m_1-1}\sum_{\bX_1\in
 I^{m_1-k_1}}\sum_{(Y_{1,2},Y_{1,3},\cdots,Y_{1,k_1})\in I^{k_1-1}}J_{m_1}(\bX_1,\bY_1)\\
&\quad\cdot\prod_{j=2}^n\Bigg(\left(\frac{1}{h}\right)^{m_j}\sum_{\bX_j\in
 I^{m_j-k_j}}\sum_{\bY_j\in I^{k_j}}J_{m_j}(\bX_j,\bY_j)\Bigg)\\
&\quad\cdot d_{j'}(Y_{1,1},Y_{q,r})^ae^{\sum_{j=0}^d(\fw(l)d_j(Y_{1,1},Y_{p,1}))^{\fr}}\prod_{k=1\atop order}^n\psi_{\bX_k}^k\Bigg|_{\psi^j=0\atop (\forall j\in
 \{1,2,\cdots,n\})}\Bigg|\\
&\le
 2^{2\sum_{j=1}^n(m_j-k_j)}D_{et}^{\frac{1}{2}\sum_{j=1}^n(m_j-k_j)-n+1}\\
&\quad\cdot \prod_{j=1}^n\Bigg(\sum_{q_j=0}^1\|J_{m_j}\|_{l,q_j}\Bigg)
\prod_{k=2}^n\Bigg(\sum_{r_k=0}^1\|\widetilde{C_o}\|_{l,r_k}\Bigg)
1_{\sum_{j=1}^nq_j+\sum_{k=2}^nr_k=a}.
\end{align*}
\end{enumerate}
\end{lemma}
\begin{proof}
We only need to sum the right-hand sides of
 \eqref{eq_tree_kernels_no_fixed} and \eqref{eq_tree_line_kernels} over
 trees. The claimed upper bounds follow from Lemma
 \ref{lem_tree_combinatorial_factor}.
\end{proof}

Here let us recall the definition of $T^{(n)}(\psi)$ $(n\in\N_{\ge
2})$. With $J(\psi)\in \bigwedge \cV$ satisfying $J_m(\psi)=0$ if $m\notin
2\N\cup \{0\}$,  
\begin{align*}
T^{(n)}(\psi)=\frac{1}{n!}\sum_{T\in \T_n}Ope(T,C_o)\prod_{j=1}^nJ(\psi^j+\psi)\Bigg|_{\psi^j=0\atop (\forall j\in
 \{1,2,\cdots,n\})}.
\end{align*}
We conclude this subsection by proving the next lemma.
\begin{lemma}\label{lem_tree_formula_general_bound}
 The following inequalities hold for any $n\in \N_{\ge 2}$.
\begin{enumerate}
\item\label{item_tree_formula_bound_0th}
\begin{align*}
|T_0^{(n)}|\le\frac{N}{h} D_{et}^{-n+1}\|\widetilde{C_o}\|_{l,0}^{n-1}\left(
\sum_{m=2}^N2^{2m}D_{et}^{\frac{m}{2}}\|J_m\|_{l,0}\right)^n.
\end{align*}
\item\label{item_tree_formula_bound_higher}
For any $m\in \{2,3,\cdots,N\}$ and $a\in\{0,1\}$,
\begin{align*}
\|T_m^{(n)}\|_{l,a}\le& 2^{-2m} D_{et}^{-\frac{m}{2}-n+1}
\prod_{i=1}^n\Bigg(\sum_{q_i=0}^1\Bigg)\prod_{j=2}^n\Bigg(\sum_{r_j=0}^1\Bigg)
1_{\sum_{i=1}^nq_i+\sum_{j=2}^nr_j=a}\\
&\cdot\prod_{k=2}^n\|\widetilde{C_o}\|_{l,r_k}\prod_{p=1}^n\Bigg(\sum_{m_p=2}^N2^{3m_p} D_{et}^{\frac{m_p}{2}}\|J_{m_p}\|_{l,q_p}\Bigg)1_{\sum_{j=1}^nm_j-2n+2\ge
 m}.
\end{align*}
\end{enumerate}
\end{lemma}
\begin{proof}
First note that the constant part $J_0$ of the input $J(\psi)$ does not
 affect the result since the operator $\prod_{l\in T}\D_l$ erases
 it. By using the anti-symmetric property of the kernels we have that
\begin{align}
T_m^{(n)}(\psi)&=\cP_m\frac{1}{n!}\sum_{T\in \T_n}Ope(T,C_o)\label{eq_tree_formula_explicit_polynomial}\\
&\quad\cdot\prod_{j=1}^n\Bigg(\sum_{m_j=2}^N\left(\frac{1}{h}\right)^{m_j}\sum_{\bX_j\in
 I^{m_j}}J_{m_j}(\bX_j)(\psi^j+\psi)_{\bX_j}\Bigg)\Bigg|_{\psi^j=0\atop (\forall
 j\in \{1,2,\cdots,n\})}\notag\\
&=\prod_{i=1}^n\Bigg(\sum_{m_i=2}^N\sum_{k_i=0}^{m_i-1}\left(\begin{array}{c}m_i\\
						     k_i\end{array}\right)\left(\frac{1}{h}\right)^{k_i}\sum_{\bY_i\in I^{k_i}}\Bigg)1_{\sum_{j=1}^nm_j-2n+2\ge m}1_{\sum_{j=1}^nk_j= m}\notag\\
&\quad\cdot \frac{\eps_{\pm}}{n!}\sum_{T\in \T_n}Ope(T,C_o)
\prod_{j=1}^n\Bigg(\left(\frac{1}{h}\right)^{m_j-k_j}\sum_{\bX_j\in
 I^{m_j-k_j}}J_{m_j}(\bX_j,\bY_j)\Bigg)\notag\\
&\quad\cdot\prod_{k=1\atop order}^n\psi_{\bX_k}^k
\Bigg|_{\psi^j=0\atop (\forall j\in \{1,2,\cdots,n\})}\prod_{p=1\atop
 order}^n\psi_{\bY_p},\notag
\end{align}
where the factor $\eps_{\pm}\in\{1,-1\}$ depends only on
 $(m_i)_{i=1}^n$, $(k_i)_{i=1}^n$ and the constraint
 $\sum_{j=1}^nm_j-2n+2\ge m$ is due to the fact that $\prod_{l\in
 T}\D_l$ erases $2n-2$ Grassmann variables. 
By the uniqueness of anti-symmetric kernels we can characterize the
 kernel of $T_m^{(n)}(\psi)$ as follows. For any $\bY\in I^m$,
\begin{align}
T_m^{(n)}(\bY)&=\prod_{i=1}^n\Bigg(\sum_{m_i=2}^N\sum_{k_i=0}^{m_i-1}\left(
\begin{array}{c} m_i\\ k_i\end{array}\right)\sum_{\bY_i\in
 I^{k_i}}\Bigg)1_{\sum_{j=1}^nm_j-2n+2\ge m}1_{\sum_{j=1}^nk_j= m}\label{eq_tree_formula_explicit_kernel}\\
&\quad\cdot\frac{1}{m!}\sum_{\s\in \S_m}\sgn(\s)1_{\bY_{\s}=(\bY_1,\bY_2,\cdots,\bY_n)}
\frac{\eps_{\pm}}{n!}\sum_{T\in \T_n}Ope(T,C_o)\notag\\
&\quad\cdot\prod_{j=1}^n\Bigg(\left(\frac{1}{h}\right)^{m_j-k_j}\sum_{\bX_j\in
 I^{m_j-k_j}}J_{m_j}(\bX_j,\bY_j)\Bigg)\prod_{k=1\atop order}^n\psi_{\bX_k}^k
\Bigg|_{\psi^j=0\atop (\forall j\in \{1,2,\cdots,n\})}.\notag
\end{align}
If $m\ge 2$, by changing the numbering if necessary we can apply Lemma \ref{lem_tree_sum_kernels}
 \eqref{item_tree_sum_kernels_fixed} to
 \eqref{eq_tree_formula_explicit_kernel} to deduce that 
\begin{align*}
\|T_m^{(n)}\|_{l,a}&\le
 \prod_{i=1}^n\Bigg(\sum_{m_i=2}^N\sum_{k_i=0}^{m_i-1}
\left(\begin{array}{c} m_i\\
      k_i\end{array}\right)\Bigg)1_{\sum_{j=1}^nm_j-2n+2\ge
 m}1_{\sum_{j=1}^nk_j= m}\\
&\quad\cdot
 2^{2\sum_{j=1}^nm_j-2m}D_{et}^{\frac{1}{2}\sum_{j=1}^nm_j-\frac{m}{2}-n+1}\\
&\quad\cdot \prod_{j=1}^n\Bigg(\sum_{q_j=0}^1\|J_{m_j}\|_{l,q_j}\Bigg)
\prod_{k=2}^n\Bigg(\sum_{r_k=0}^1\|\widetilde{C_o}\|_{l,r_k}\Bigg)
1_{\sum_{j=1}^nq_j+\sum_{k=2}^nr_k=a}.
\end{align*}
Then, by substituting the inequality 
$$
\prod_{i=1}^n\Bigg(\sum_{k_i=0}^{m_i-1}\left(\begin{array}{c} m_i\\
      k_i\end{array}\right)\Bigg)1_{\sum_{j=1}^nk_j=m}\le 2^{\sum_{j=1}^nm_j},
$$
we obtain the inequality claimed in
 \eqref{item_tree_formula_bound_higher}. By applying Lemma
 \ref{lem_tree_sum_kernels} \eqref{item_tree_sum_kernels_no_fixed} to
 \eqref{eq_tree_formula_explicit_kernel} we can derive the
 inequality claimed in \eqref{item_tree_formula_bound_0th}.
\end{proof}

\subsection{Invariance of Grassmann polynomials}\label{subsec_invariance_general}
Here we show that Grassmann polynomials produced by the free integration
or the tree expansion inherit symmetric properties from the covariance
$C_o$ and the input polynomial $J(\psi)$. The
general results summarized in this subsection will have practical applications in
Section \ref{sec_IR_model}. 

Let $S$ be a bijective map from $I$ to $I$ and $Q$ be a map from $I$ to
$\R$. For $m\in \N$, define $S_m:I^m\to I^m$, $Q_m:I^m\to \R$ by 
\begin{align*}
&S_m(X_1,X_2,\cdots,X_m):=(S(X_1),S(X_2),\cdots,S(X_m)),\\ 
&Q_m(X_1,X_2,\cdots,X_m):=\sum_{j=1}^mQ(X_j).
\end{align*}
For $X \in I$, set $(\cR\psi)_{X}:=e^{iQ(S(X))}\psi_{S(X)}$. For
$f(\psi)\in \bigwedge \cV$ we define $f(\cR\psi)\in \bigwedge \cV$ by
replacing each $\psi_X$ by $(\cR\psi)_X$ $(X\in I)$ inside $f(\psi)$. More 
precisely, for
$f(\psi)=\sum_{m=0}^N\left(\frac{1}{h}\right)^m\sum_{\bX\in
I^m}f_m(\bX)\psi_{\bX}$, 
\begin{align*}
f(\cR\psi):=\sum_{m=0}^N\left(\frac{1}{h}\right)^m\sum_{\bX\in
I^m}f_m(\bX)e^{iQ_m(S_m(\bX))}\psi_{S_m(\bX)}.
\end{align*}
For $f(\psi)\in\bigwedge \cV$ we define $\overline{f}(\psi)\in \bigwedge \cV$ by
$$
\overline{f}(\psi):=\sum_{m=0}^N\left(\frac{1}{h}\right)^m\sum_{\bX\in
I^m}\overline{f_m(\bX)}\psi_{\bX}.
$$  
\begin{lemma}\label{lem_free_tree_invariance_general}
Let $F(\psi)$, $T^{(n)}(\psi)\in \bigwedge \cV$ $(n\in \N_{\ge 2})$ be
 defined by \eqref{eq_definition_free_general} and the right-hand side of
 \eqref{eq_explicit_tree_formula} respectively with the covariance
 $C_o:I_0^2\to\C$ and the input $J(\psi)$
 satisfying $J_m(\psi)=0$ if $m\notin 2\N\cup\{0\}$. Let
 $\widetilde{C_o}:I^2\to \C$ be the anti-symmetric extension of $C_o$
 defined by \eqref{eq_anti_symmetric_extension_covariance}. Then, the
 following statements hold true.
\begin{enumerate}
\item\label{item_invariance_simple}
If 
\begin{align*}
J(\cR\psi)=J(\psi),\
 \widetilde{C_o}(\bX)=e^{iQ_2(S_2(\bX))}\widetilde{C_o}(S_2(\bX)),\
 (\forall \bX\in I^2), 
\end{align*}
then,
\begin{align*}
&F(\cR\psi)=F(\psi),\\
&T^{(n)}(\cR\psi)=T^{(n)}(\psi),\ (\forall n\in \N_{\ge 2}).
\end{align*}
\item\label{item_invariance_variables}
Let $a\in \N$ and $D$ be a domain of $\C^a$ satisfying that 
$\overline{\bz}\in \overline{D}$ $(\forall \bz\in \overline{D})$,
where $\overline{D}$ denotes the closure of $D$. 
Additionally assume that $J(\psi)$ and $C_o$ are parameterized by $\bz\in
     \overline{D}$ and write $J(\bz)(\psi)$, $C_o(\bz)$,
     $\widetilde{C_o}(\bz)$, $F(\bz)(\psi)$, $T^{(n)}(\bz)(\psi)$ in place of $J(\psi)$, $C_o$,
     $\widetilde{C_o}$, $F(\psi)$, $T^{(n)}(\psi)$ respectively. If  
     \begin{align*}
     &\overline{J(\overline{\bz})}(\cR\psi)=J(\bz)(\psi),\
      \widetilde{C_o}(\bz)(\bX)=e^{-iQ_2(S_2(\bX))}\overline{\widetilde{C_o}(\overline{\bz})(S_2(\bX))},\\
&(\forall \bz\in\overline{D},\bX\in I^2),
     \end{align*}
then,
\begin{align*}
&\overline{F(\overline{\bz})}(\cR\psi)=F(\bz)(\psi),\\
&\overline{T^{(n)}(\overline{\bz})}(\cR\psi)=T^{(n)}(\bz)(\psi),\
 (\forall \bz\in \overline{D},n\in \N_{\ge 2}).
\end{align*}
\end{enumerate}
\end{lemma}
\begin{proof}
We provide the proof for the claim
 \eqref{item_invariance_variables} first. The claim
 \eqref{item_invariance_simple} can be proved similarly. 

\eqref{item_invariance_variables}: Let us show the invariance of 
$F(\bz)(\psi)$. Define
 $\widetilde{C_o}(\bz)':I^2\to \C$ by 
$$
\widetilde{C_o}(\bz)'(\bX):=e^{iQ_2(\bX)}\widetilde{C_o}(\bz)(S_2^{-1}(\bX)).
$$
It follows from the assumption that
\begin{align}\label{eq_invariance_covariance_general}
\overline{\widetilde{C_o}(\overline{\bz})'(\bX)}=\widetilde{C_o}(\bz)(\bX),\
 (\forall \bz\in \overline{D},\bX\in I^2).
\end{align}
Recalling the definition of the Grassmann Gaussian integral, we observe
 that for any $\bX\in I^n$,
\begin{align}
\int \psi^1_{\bX}d\mu_{C_o(\bz)}(\psi^1)&=e^{-\sum_{Y,Z\in
 I_0}C_o(\bz)(Y,Z)\frac{\partial}{\partial \opsi_Y^1}\frac{\partial}{\partial \psi_Z^1}}\psi^1_{\bX}\Big|_{\psi^1=0}\label{eq_invariance_free_pre_pre}\\
&= e^{-\sum_{\bY\in
 I^2}\widetilde{C_o}(\bz)(\bY)\frac{\partial}{\partial
 \psi_{\bY}^1}}\psi^1_{\bX}\Big|_{\psi^1=0}\notag\\
&=e^{-\sum_{\bY\in
 I^2}\widetilde{C_o}(\bz)'(\bY)\frac{\partial}{\partial
 \psi_{\bY}^1}}e^{-iQ_n(S_n(\bX))}\psi^1_{S_n(\bX)}\Big|_{\psi^1=0}.\notag
\end{align}
Moreover, by \eqref{eq_invariance_covariance_general},
\begin{align}
\int\psi_{\bX}^1d\mu_{\overline{C_o(\overline{\bz})}}(\psi^1)&=e^{-\sum_{\bY\in
 I^2}\widetilde{C_o}(\bz)(\bY)\frac{\partial}{\partial
 \psi_{\bY}^1}}e^{iQ_n(S_n(\bX))}\psi^1_{S_n(\bX)}\Big|_{\psi^1=0}\label{eq_invariance_free_preparation}\\
&=\int
 e^{iQ_n(S_n(\bX))}\psi^1_{S_n(\bX)}d\mu_{C_o(\bz)}(\psi^1).\notag
\end{align}
By using anti-symmetry we can characterize $F(\bz)(\psi)$ as follows.
\begin{align}
F(\bz)(\psi)=J(\bz)_0+&\sum_{m=1}^N\left(\frac{1}{h}\right)^m\sum_{\bX\in
 I^m}\sum_{n=m}^N\left(\begin{array}{c}n \\
		       m\end{array}\right)\label{eq_free_explicit_characterization}\\
&\cdot\left(\frac{1}{h}\right)^{n-m}\sum_{\bY\in I^{n-m}}J(\bz)_n(\bX,\bY)
\int \psi_{\bY}^1d\mu_{C_o(\bz)}(\psi^1)\psi_{\bX}.\notag
\end{align}
By the invariance of $J(\bz)(\psi)$ and
 \eqref{eq_invariance_free_preparation} we have that
\begin{align*}
&\overline{F(\overline{\bz})}(\cR\psi)\\
&=\overline{J(\overline{\bz})_0}+\sum_{m=1}^N\left(\frac{1}{h}\right)^m\sum_{\bX\in
 I^m}\sum_{n=m}^N\left(\begin{array}{c}n \\
		       m\end{array}\right)\left(\frac{1}{h}\right)^{n-m}\sum_{\bY\in I^{n-m}}\overline{J(\overline{\bz})_n(\bX,\bY)}\\
&\quad\qquad\qquad\cdot\int e^{iQ_{n-m}(S_{n-m}(\bY))}
 \psi_{S_{n-m}(\bY)}^1d\mu_{C_o(\bz)}(\psi^1)e^{iQ_m(S_m(\bX))}\psi_{S_m(\bX)}\\&=\int\overline{J(\overline{\bz})}(\cR\psi+\cR\psi^1)d\mu_{C_o(\bz)}(\psi^1)=\int J(\bz)(\psi+\psi^1)d\mu_{C_o(\bz)}(\psi^1)\\
&=F(\bz)(\psi).
\end{align*}

Next let us prove the invariance of $T^{(n)}(\bz)(\psi)$. Since $M_{at}(T,\xi,\bs)^t=M_{at}(T,\xi,\bs)$,
\begin{align}
&\sum_{r,s=1}^nM_{at}(T,\xi,\bs)(r,s)\D_{r,s}(C_o(\bz))\label{eq_laplacian_anti_symmetrization_another}\\
&=-\sum_{r,s=1}^nM_{at}(T,\xi,\bs)(r,s)\sum_{(X_1,X_2)\in
 I^2}\widetilde{C_o}(\bz)(X_1,X_2)\frac{\partial}{\partial
 \psi_{X_1}^r}\frac{\partial}{\partial \psi_{X_2}^s}.\notag
\end{align}
For an anti-symmetric function $A:I^2\to \C$ and $T\in \T_n$ we define
 the operator $\widetilde{Ope}(T,A)$ by
\begin{align*}
\widetilde{Ope}(T,A):=&\prod_{\{p,q\}\in T}\left(-2\sum_{(X_1,X_2)\in
 I^2}A(X_1,X_2)\frac{\partial}{\partial
 \psi_{X_1}^p}\frac{\partial}{\partial \psi_{X_2}^q}\right)\\
&\cdot\int_{[0,1]^{n-1}}d\bs\sum_{\xi\in
 \S_n(T)}\varphi(T,\xi,\bs)\\
&\cdot e^{-\sum_{r,s=1}^nM_{at}(T,\xi,\bs)(r,s)\sum_{(Y_1,Y_2)\in
 I^2}A(Y_1,Y_2)\frac{\partial}{\partial
 \psi_{Y_1}^r}\frac{\partial}{\partial \psi_{Y_2}^s}}.
\end{align*}
The equalities \eqref{eq_laplacian_anti_symmetrization},
 \eqref{eq_laplacian_anti_symmetrization_another} ensure that
\begin{align}\label{eq_equivalence_tree_operators}
\widetilde{Ope}\left(T,\widetilde{C_o}(\bz)\right)=Ope(T,C_o(\bz)).
\end{align}
By the same argument as in \eqref{eq_invariance_free_pre_pre} we have
 for any $m_j\in \{1,2,\cdots,N\}$, 
$\bX_j\in I^{m_j}$ $(j=1,2,\cdots,n)$ that 
\begin{align*}
&\widetilde{Ope}\left(T,\widetilde{C_o}(\bz)\right)\prod_{j=1\atop
 order}^n\psi_{\bX_j}^j\Bigg|_{\psi^j=0\atop(\forall
 j\in\{1,2,\cdots,n\})}\\
&= \widetilde{Ope}\left(T,\widetilde{C_o}(\bz)'\right)\prod_{j=1\atop
 order}^n\left(e^{-iQ_{m_j}(S_{m_j}(\bX_j))}\psi_{S_{m_j}(\bX_j)}^j\right)\Bigg|_{\psi^j=0\atop(\forall
 j\in\{1,2,\cdots,n\})}.
\end{align*}
Thus, by \eqref{eq_invariance_covariance_general},
\begin{align}\label{eq_invariance_tree_preparation}
&\widetilde{Ope}\left(T,\overline{\widetilde{C_o}(\overline{\bz})}\right)\prod_{j=1\atop
 order}^n\psi_{\bX_j}^j\Bigg|_{\psi^j=0\atop(\forall
 j\in\{1,2,\cdots,n\})}\\
&= \widetilde{Ope}\left(T,\widetilde{C_o}(\bz)\right)\prod_{j=1\atop
 order}^n\left(e^{iQ_{m_j}(S_{m_j}(\bX_j))}\psi_{S_{m_j}(\bX_j)}^j\right)\Bigg|_{\psi^j=0\atop(\forall
 j\in\{1,2,\cdots,n\})}.\notag
\end{align}
By using the invariance of $J(\bz)(\psi)$,
 \eqref{eq_equivalence_tree_operators} and
 \eqref{eq_invariance_tree_preparation} we can deduce from
 \eqref{eq_tree_formula_explicit_polynomial} that for any $m\in
 \{0,1,\cdots,N\}$,
\begin{align*}
&\overline{T^{(n)}(\overline{\bz})_m}(\cR\psi)\\
&=\prod_{i=1}^n\Bigg(\sum_{m_i=2}^N\sum_{k_i=0}^{m_i-1}\left(\begin{array}{c}m_i\\
						     k_i\end{array}\right)\left(\frac{1}{h}\right)^{k_i}\sum_{\bY_i\in I^{k_i}}\Bigg)1_{\sum_{j=1}^nm_j-2n+2\ge m}1_{\sum_{j=1}^nk_j= m}\\
&\quad\cdot \frac{\eps_{\pm}}{n!}\sum_{T\in \T_n}\widetilde{Ope}\left(T,\widetilde{C_o}(\bz)\right)
\prod_{j=1}^n\Bigg(\left(\frac{1}{h}\right)^{m_j-k_j}\sum_{\bX_j\in
 I^{m_j-k_j}}\overline{J(\overline{\bz})_{m_j}(\bX_j,\bY_j)}\Bigg)\\
&\quad\cdot e^{i\sum_{j=1}^nQ_{m_j-k_j}(S_{m_j-k_j}(\bX_j))}
\prod_{p=1\atop order}^n\psi_{S_{m_p-k_p}(\bX_p)}^p
\Bigg|_{\psi^j=0\atop (\forall j\in \{1,2,\cdots,n\})}\\
&\quad\cdot e^{i\sum_{j=1}^nQ_{k_j}(S_{k_j}(\bY_j))}\prod_{q=1\atop
 order}^n\psi_{S_{k_q}(\bY_q)}\\
&=\cP_m\frac{1}{n!}\sum_{T\in \T_n}\widetilde{Ope}\left(T,\widetilde{C_o}(\bz)\right)\prod_{j=1}^n\overline{J(\overline{\bz})}(\cR\psi+\cR\psi^j)
\Bigg|_{\psi^j=0\atop (\forall j\in \{1,2,\cdots,n\})}\\
&=\cP_m\frac{1}{n!}\sum_{T\in
 \T_n}Ope(T,C_o(\bz))\prod_{j=1}^nJ(\bz)(\psi+\psi^j)\Bigg|_{\psi^j=0\atop (\forall j\in \{1,2,\cdots,n\})}\\
&=T^{(n)}(\bz)_m(\psi),
\end{align*}
which implies that
 $\overline{T^{(n)}(\overline{\bz})}(\cR\psi)=T^{(n)}(\bz)(\psi)$. 

\eqref{item_invariance_simple}: It follows from the invariance
 of $\widetilde{C_o}$ that for any $\bX\in I^n$,
\begin{align}
e^{-\sum_{\bY\in
 I^2}\widetilde{C_o}(\bY)\frac{\partial}{\partial
 \psi_{\bY}^1}}\psi^1_{\bX}\Big|_{\psi^1=0}=e^{-\sum_{\bY\in
 I^2}\widetilde{C_o}(\bY)\frac{\partial}{\partial
 \psi_{\bY}^1}}e^{iQ_n(S_n(\bX))}\psi^1_{S_n(\bX)}\Big|_{\psi^1=0}.\label{eq_key_equality_invariance}
\end{align}
We can see from the definition of the Grassmann Gaussian integral and
 \eqref{eq_key_equality_invariance} that for any $\bX\in I^n$,
\begin{align}
\int\psi_{\bX}^1d\mu_{C_o}(\psi^1)=\int
 e^{iQ_n(S_n(\bX))}\psi^1_{S_n(\bX)}d\mu_{C_o}(\psi^1).\label{eq_free_part_necessary_invariance}
\end{align}
By substituting \eqref{eq_free_part_necessary_invariance} into 
\eqref{eq_free_explicit_characterization} and using the invariance
 of $J(\psi)$ we obtain that $F(\psi)=F(\cR\psi)$. 

Using \eqref{eq_equivalence_tree_operators} and
 \eqref{eq_key_equality_invariance}, we can prove that for any
 $m_j\in\{1,2,\cdots,N\}$, $\bX_j\in I^{m_j}$ $(j=1,2,\cdots,n)$,
\begin{align}\label{eq_tree_part_necessary_invariance}
&Ope(T,C_o)\prod_{j=1\atop
 order}^n\psi_{\bX_j}^j\Bigg|_{\psi^j=0\atop(\forall
 j\in\{1,2,\cdots,n\})}\\
&= Ope(T,C_o)\prod_{j=1\atop
 order}^n\left(e^{iQ_{m_j}(S_{m_j}(\bX_j))}\psi_{S_{m_j}(\bX_j)}^j\right)\Bigg|_{\psi^j=0\atop(\forall
 j\in\{1,2,\cdots,n\})}.\notag
\end{align}
Then, by inserting \eqref{eq_tree_part_necessary_invariance} into
 \eqref{eq_tree_formula_explicit_polynomial} and using the 
 invariance of $J(\psi)$ we can confirm that
 $T_m^{(n)}(\psi)=T_m^{(n)}(\cR \psi)$ $(\forall m\in
 \{0,1,\cdots,N\})$, which implies that $T^{(n)}(\psi)=T^{(n)}(\cR\psi)$.
\end{proof}
 
\section{General estimation at different temperatures}
\label{sec_general_estimation_temperature}
In this section we estimate differences between 2 Grassmann
polynomials produced by a single-scale integration at 2 different
temperatures. One can prove that the free energy density is analytic
with the coupling constants in a $\beta$-independent domain around the
origin without measuring the differences between Grassmann
polynomials created at different temperatures. However, in order to
prove the existence of zero-temperature limit of the free energy
density, we need the temperature-dependent estimates constructed in this section.

Let us set up notations which we start using from this section. Since we
consider the problems at 2 different temperatures, we sometimes add the
notation $(\beta)$ to the right of a $\beta$-dependent object. For
example, we write $I_0(\beta)$, $I(\beta)$
instead of the index sets $I_0$, $I$ when we want to indicate with which
$\beta$ these sets are defined. 

Let us introduce the extended index sets
$I_{0,\infty}$, $I_{\infty}$ by
$$
I_{0,\infty}:=\cB\times\G\times\spin\times\frac{1}{h}\Z,\ I_{\infty}:=I_{0,\infty}\times\{1,-1\}.
$$
For any $x\in(1/h)\Z$ there uniquely exist $n_{\beta}(x)\in\Z$ and
$r_{\beta}(x)\in [0,\beta)_h$ such that $x=n_{\beta}(x)\beta
+r_{\beta}(x)$. For any
$$\bX=((\rho_1,\bx_1,\s_1,x_1),(\rho_2,\bx_2,\s_2,x_2),\cdots,(\rho_m,\bx_m,\s_m,x_m))\in
I_{0,\infty}^m$$ 
we define $R_{\beta}(\bX)\in I_0^m$, $N_{\beta}(\bX)\in
\Z$ by 
\begin{align*}
&R_{\beta}(\bX):=((\rho_1,\bx_1,\s_1,r_{\beta}(x_1)),(\rho_2,\bx_2,\s_2,r_{\beta}(x_2)),\cdots,\\
&\qquad\qquad\quad(\rho_m,\bx_m,\s_m,r_{\beta}(x_m))),\\
&N_{\beta}(\bX):=\sum_{j=1}^mn_{\beta}(x_j).
\end{align*}
Moreover, for any $x\in (1/h)\Z$ let $\bX+x\in I_{0,\infty}^m$ be
defined by 
$$
\bX+x:=((\rho_1,\bx_1,\s_1,x_1+x),(\rho_2,\bx_2,\s_2,x_2+x),\cdots,(\rho_m,\bx_m,\s_m,x_m+x)).
$$
Similarly for any
$$\bX=((\rho_1,\bx_1,\s_1,x_1,\theta_1),(\rho_2,\bx_2,\s_2,x_2,\theta_2),\cdots,(\rho_m,\bx_m,\s_m,x_m,\theta_m))\in
I_{\infty}^m$$ 
we define $R_{\beta}(\bX)\in I^m$, $N_{\beta}(\bX)\in
\Z$ by 
\begin{align*}
&R_{\beta}(\bX):=((\rho_1,\bx_1,\s_1,r_{\beta}(x_1),\theta_1),(\rho_2,\bx_2,\s_2,r_{\beta}(x_2),\theta_2),\cdots,\\
&\qquad\qquad\quad(\rho_m,\bx_m,\s_m,r_{\beta}(x_m),\theta_m)),\\
&N_{\beta}(\bX):=\sum_{j=1}^mn_{\beta}(x_j),
\end{align*}
though we must admit that these are abuse of notation.
Also, for $x\in (1/h)\Z$ let
\begin{align*}
\bX+x:=(&(\rho_1,\bx_1,\s_1,x_1+x,\theta_1),(\rho_2,\bx_2,\s_2,x_2+x,\theta_2),\cdots,\\
&(\rho_m,\bx_m,\s_m,x_m+x,\theta_m))\in I_{\infty}^m.
\end{align*}
Set
\begin{align*}
\G_{\infty}:=\left\{\sum_{j=1}^dm_j\bu_j\ \Big|\ m_j\in\Z\ (j=1,2,\cdots,d)\right\}.\end{align*}
Define the map $r_L$ from $\G_{\infty}$ to $\G$ by 
$$
r_L\left(\sum_{j=1}^dm_j\bu_j\right):=\sum_{j=1}^dm_j'\bu_j,
$$
where $m_j'\in\{0,1,\cdots,L-1\}$ and $m_j=m_j'$ in $\Z/L\Z$ $(\forall
j\in \{1,2,\cdots,d\})$. In fact the notations $\G_{\infty}$,
$r_L(\cdot)$ are used only in Section \ref{sec_IR_model} and Appendix
\ref{app_h_L_limit}. However, it is systematic to introduce them at
this stage together with other notations.

In this section we treat Grassmann polynomials whose anti-symmetric
kernels $f_m:I(\beta)^m\to \C$ $(m\in \{1,2,\cdots,N\})$ satisfy that
\begin{align}\label{eq_temperature_translation}
f_m(\bX)=(-1)^{N_{\beta}(\bX+x)}f_m(R_{\beta}(\bX+x)),\ (\forall \bX\in I(\beta)^m, x\in (1/h)\Z).
\end{align}
In our practical multi-scale integrations all relevant Grassmann
polynomials will be proved to have the kernels satisfying 
\eqref{eq_temperature_translation}.

From now until the end of this section we assume that
\begin{align}\label{eq_beta_h_assumption}
\beta_1,\beta_2\in \N,\ \beta_2\ge\beta_1,\ h\in 4\N.
\end{align}
It will eventually turn out that the condition
\eqref{eq_beta_h_assumption} can be naturally imposed during the proof of the
main theorem about the existence of zero-temperature limit of the free
energy density in Section \ref{sec_IR_model}. 

We introduce discrete versions of the intervals 
$[-\beta_1/4,\beta_1/4)$, \\
$[\beta_1/4,\beta_a-\beta_1/4)$ $(a=1,2)$ by
\begin{align*}
&\left[-\frac{\beta_1}{4},\frac{\beta_1}{4}\right)_h:=\left\{-\frac{\beta_1}{4},
-\frac{\beta_1}{4}+\frac{1}{h},\cdots,\frac{\beta_1}{4}-\frac{1}{h}\right\},\\
&\left[\frac{\beta_1}{4},\beta_a-\frac{\beta_1}{4}\right)_h:=\left\{\frac{\beta_1}{4},
\frac{\beta_1}{4}+\frac{1}{h},\cdots,\beta_a-\frac{\beta_1}{4}-\frac{1}{h}\right\},\
 (a=1,2).
\end{align*}
Note that $0\in [-\beta_1/4,\beta_1/4)_h$ by the assumption $h\in 4\N$.

Define the index sets $\hat{I}_0$, $\hat{I}$, $I_0^0$, $I^0$ by
\begin{align*}
&\hat{I}_0:=\cB\times\G\times\spin\times\left[-\frac{\beta_1}{4},\frac{\beta_1}{4}\right)_h,\quad \hat{I}:=\hat{I}_0\times\{1,-1\},\\
&I_0^0:=\cB\times\G\times\spin\times\{0\},\quad I^0:=I_0^0\times\{1,-1\}.
\end{align*}

We assume that covariances $C_o(\beta_a):I_0(\beta_a)^2\to\C$ $(a=1,2)$ are
given and, as in Section \ref{sec_general_estimation}, there exists a
constant $D_{et}\in\R_{\ge 0}$ such that $C_o(\beta_a)$ $(a=1,2)$ satisfy the
determinant bound \eqref{eq_general_determinant_bound_condition} with
$D_{et}$. Moreover, assume that there is a
$\beta_1,\beta_2$-dependent constant $D(\beta_1,\beta_2)\in\R_{\ge 0}$ such that
 \begin{align}
&|\det(\<\bp_i,\bq_j\>_{\C^r}C_o(\beta_1)(R_{\beta_1}(X_i,Y_j)))_{1\le
  i,j\le n}\label{eq_general_determinant_bound_difference}\\
&-
\det(\<\bp_i,\bq_j\>_{\C^r}C_o(\beta_2)(R_{\beta_2}(X_i,Y_j)))_{1\le i,j\le n}|\le
 D(\beta_1,\beta_2)\cdot D_{et}^n,\notag\\
&(\forall r,n\in\N,\bp_i,\bq_i\in\C^r\text{ with }
\|\bp_i\|_{\C^r},\|\bq_i\|_{\C^r}\le 1,\notag\\
&\quad X_i,Y_i\in \hat{I}_0\
 (i=1,2,\cdots,n)).\notag
\end{align}
Furthermore, the covariances $C_o(\beta_a)$ $(a=1,2)$ are assumed to
satisfy that 
\begin{align}
&C_o(\beta_a)(\bX)=(-1)^{N_{\beta_a}(\bX+x)}C_o(\beta_a)(R_{\beta_a}(\bX+x)),\label{eq_temperature_translation_covariance_general}\\
& (\forall \bX\in I_0(\beta_a)^2, x\in (1/h)\Z).\notag
\end{align}

For any
$(\rho,\bx,\s,x,\theta)$, $(\eta,\by,\tau,y,\xi)\in \hat{I}$, $j\in
\{0,1,\cdots,d\}$, set
\begin{align*}
&\hat{d}_j((\rho,\bx,\s,x,\theta),(\eta,\by,\tau,y,\xi))\\
&:=\left\{\begin{array}{ll}|x-y|&\text{if }j=0,\\
\frac{L}{2\pi}|e^{i\frac{2\pi}{L}\<\bx,\bv_j\>}-e^{i\frac{2\pi}{L}\<\by,\bv_j\>}|&\text{if }j\in\{1,2,\cdots,d\},
\end{array}\right.
\end{align*}
which is an analogue of $d_j(\cdot,\cdot)$ introduced in Section
\ref{sec_general_estimation}. 
Let $m\in \N_{\ge 2}$. For anti-symmetric functions
$f_m(\beta_a):I(\beta_a)^m\to \C$ $(a=1,2)$ we estimate the difference
between them by the quantities $|f_m(\beta_1)-f_m(\beta_2)|_{l}$
$(l\in\Z)$ defined as follow.
\begin{align*}
|f_m(\beta_1)-f_m(\beta_2)|_l:=&\sup_{X\in
 I^0}\left(\frac{1}{h}\right)^{m-1}\sum_{\bY\in\hat{I}^{m-1}}e^{\sum_{j=0}^d(\frac{1}{\pi}\fw(l)\hat{d}_j(X,Y_1))^{\fr}}\\
&\cdot|f_m(\beta_1)(R_{\beta_1}(X,\bY))-f_m(\beta_2)(R_{\beta_2}(X,\bY))|,
\end{align*}
where $\{\fw(l)\}_{l\in\Z}$ $\subset \R_{>0}$ and $\fr\in (0,1]$ are the
same parameters as those used in the definitions of $\|\cdot\|_{l,0}$, $\|\cdot\|_{l,1}$ in
Section \ref{sec_general_estimation}.

By the assumption \eqref{eq_beta_h_assumption}, $N(\beta_1)(=\sharp I(\beta_1))\le N(\beta_2)(=\sharp
I(\beta_2))$.  It
will be convenient to write
$$
f(\beta_1)(\psi)=\sum_{m=0}^{N(\beta_2)}f_m(\beta_1)(\psi)
$$ 
for any $f(\beta_1)(\psi)\in \bigwedge \cV(\beta_1)$ 
by admitting that 
$f_m(\beta_1)(\psi)=0$ $(\forall m\in \{N(\beta_1)+1,N(\beta_1)+2,\cdots,N(\beta_2)\})$.

\subsection{Estimation of the free integration at different temperatures}\label{subsec_general_free_integration_temperature}
As in Subsection \ref{subsec_general_free_integration} we set for $a=1,2$, 
\begin{equation*}
F(\beta_a)(\psi):=\int
 J(\beta_a)(\psi+\psi^1)d\mu_{C_o(\beta_a)}(\psi^1)\ \left(\in\bigwedge \cV(\beta_a)\right)
\end{equation*}
with $J(\beta_a)(\psi)\in\bigwedge \cV(\beta_a)$ satisfying that $J_m(\beta_a)(\psi)=0$ if $m\notin
2\N\cup \{0\}$ and having the anti-symmetric kernels satisfying \eqref{eq_temperature_translation}. It follows from
definition that $F_m(\beta_a)(\psi)=0$ if $m\notin
2\N\cup \{0\}$. In this subsection we measure differences between
$F(\beta_1)(\psi)$ and $F(\beta_2)(\psi)$. The result is the following.
 
\begin{lemma}\label{lem_general_free_bound_difference}
\begin{enumerate}
\item \label{item_general_free_bound_difference_0th} For any $l\in \Z$,
\begin{align*}
&\left|\frac{h}{N(\beta_1)}F_0(\beta_1)-\frac{h}{N(\beta_2)}F_0(\beta_2)\right|\\&\le
 \left|\frac{h}{N(\beta_1)}J_0(\beta_1)-\frac{h}{N(\beta_2)}J_0(\beta_2)\right|\\
&\quad+\sum_{n=2}^{N(\beta_2)}2^nD_{et}^{\frac{n}{2}}\Bigg(|J_n(\beta_1)-J_n(\beta_2)|_l+D(\beta_1,\beta_2)\sum_{a=1}^2\|J_n(\beta_a)\|_{l,0}\\
&\qquad\qquad\qquad\qquad+\frac{2\pi}{\beta_1}\sum_{a=1}^2\|J_n(\beta_a)\|_{l,1}\Bigg).
\end{align*}
\item\label{item_general_free_bound_difference_higher} For any $l\in
      \Z$, $m\in \N_{\ge 2}$,
\begin{align*}
&|F_m(\beta_1)-F_m(\beta_2)|_l\\
&\le
 |J_m(\beta_1)-J_m(\beta_2)|_l\\
&\quad+\sum_{n=m+1}^{N(\beta_2)}2^{2n}D_{et}^{\frac{n-m}{2}}\Bigg(|J_n(\beta_1)-J_n(\beta_2)|_l+D(\beta_1,\beta_2)\sum_{a=1}^2\|J_n(\beta_a)\|_{l,0}\\
&\qquad\qquad\qquad\qquad\qquad+\frac{2\pi}{\beta_1}\sum_{a=1}^2\|J_n(\beta_a)\|_{l,1}\Bigg).
\end{align*}
\end{enumerate}
\end{lemma}
\begin{proof}
\eqref{item_general_free_bound_difference_0th}: By the property
 \eqref{eq_temperature_translation_covariance_general} and the
 definition of the Grassmann Gaussian integral we see that
\begin{align}\label{eq_application_temperature_translation_covariance}
&\int\psi_{\bX}d\mu_{C_o(\beta_a)}(\psi)=(-1)^{N_{\beta_a}(\bX+x)}\int\psi_{R_{\beta_a}(\bX+x)}d\mu_{C_o(\beta_a)}(\psi),\\
&(\forall \bX\in I(\beta_a)^m,x\in
 (1/h)\Z,a\in\{1,2\}).\notag
\end{align}
It follows from \eqref{eq_characterization_free_kernel},
 \eqref{eq_temperature_translation} and
 \eqref{eq_application_temperature_translation_covariance} that
\begin{align*}
&F_0(\beta_a)\\
&=J_0(\beta_a)+\sum_{n=2}^{N(\beta_2)}\left(\frac{1}{h}\right)^n\sum_{x\in
 [0,\beta_a)_h}\sum_{X\in I^0}\sum_{\bY\in
 I(\beta_a)^{n-1}}J_n(\beta_a)(X,R_{\beta_a}(\bY-x))\\
&\qquad\qquad\qquad\cdot\int\psi_{X}^1\psi_{R_{\beta_a}(\bY-x)}^1d\mu_{C_o(\beta_a)}(\psi^1)\\
&=J_0(\beta_a)+\beta_a\sum_{n=2}^{N(\beta_2)}\left(\frac{1}{h}\right)^{n-1}\sum_{X\in I^0}\sum_{\bY\in
 I(\beta_a)^{n-1}}J_n(\beta_a)(X,\bY)\\
&\quad\qquad\qquad\qquad\cdot\int\psi_{(X,\bY)}^1d\mu_{C_o(\beta_a)}(\psi^1)\\
&=J_0(\beta_a)+\beta_a\sum_{n=2}^{N(\beta_2)}\left(\frac{1}{h}\right)^{n-1}\sum_{X\in I^0}\sum_{\bY\in \hat{I}^{n-1}}J_n(\beta_a)(R_{\beta_a}(X,\bY))\\
&\quad\qquad\qquad\qquad\cdot\int\psi_{R_{\beta_a}(X,\bY)}^1d\mu_{C_o(\beta_a)}(\psi^1)\\
&\quad+\beta_a\sum_{n=2}^{N(\beta_2)}\left(\frac{1}{h}\right)^{n-1}\sum_{X\in
 I^0}\sum_{\bY\in I(\beta_a)^{n-1}}\\
&\quad\qquad\cdot 1_{\exists
 (\rho,\bx,\s,x,\theta)\in I(\beta_a)\text{ s.t. }
(\rho,\bx,\s,x,\theta)\subset\bY\text{ and }x\in
 [\frac{\beta_1}{4},\beta_a-\frac{\beta_1}{4})_h}\\
&\quad\qquad\cdot J_n(\beta_a)(X,\bY)\int\psi_{(X,\bY)}^1d\mu_{C_o(\beta_a)}(\psi^1).
\end{align*}
Then, by using the determinant bounds
 \eqref{eq_general_determinant_bound_condition}, \eqref{eq_general_determinant_bound_difference} we have
\begin{align*}
&\left|\frac{h}{N(\beta_1)}F_0(\beta_1)-\frac{h}{N(\beta_2)}F_0(\beta_2)\right|\\
&\le
 \left|\frac{h}{N(\beta_1)}J_0(\beta_1)-\frac{h}{N(\beta_2)}J_0(\beta_2)\right|\\
&\quad+\sum_{n=2}^{N(\beta_2)}\big(|J_n(\beta_1)-J_n(\beta_2)|_lD_{et}^{\frac{n}{2}}+\|J_n(\beta_2)\|_{l,0}D(\beta_1,\beta_2)D_{et}^{\frac{n}{2}}\big)\\
&\quad+\frac{2\pi}{\beta_1}\sum_{n=2}^{N(\beta_2)}(n-1)\sum_{a=1}^2\|J_n(\beta_a)\|_{l,1}D_{et}^{\frac{n}{2}},
\end{align*}
where we also used the inequalities that
\begin{align}\label{eq_basic_temperature_lower_bound}
&\left|\frac{\beta_a}{2\pi}(e^{i\frac{2\pi}{\beta_a}x}-1)\right|\ge\frac{\beta_1}{2\pi},\
 (\forall x\in
 [\beta_1/4,\beta_a-\beta_1/4),a\in \{1,2\})
\end{align}
and 
$$
1\le e^{\sum_{j=0}^d(\fw(l)d_j(X,Y_1))^{\fr}},\ 1\le e^{\sum_{j=0}^d(\frac{1}{\pi}\fw(l)\hat{d}_j(X,Y_1))^{\fr}}.
$$

\eqref{item_general_free_bound_difference_higher}: We can see from
 \eqref{eq_characterization_free_kernel} that for any $X\in I^0$,
 $\bY\in \hat{I}^{m-1}$,
\begin{align*}
&F_m(\beta_a)(R_{\beta_a}(X,\bY))\\
&=J_m(\beta_a)(R_{\beta_a}(X,\bY))\\
&\quad+\sum_{n=m+1}^{N(\beta_2)}\left(\frac{1}{h}\right)^{n-m}\sum_{\bZ\in
 \hat{I}^{n-m}}\left(\begin{array}{c}n\\ m\end{array}\right)
J_n(\beta_a)(R_{\beta_a}(X,\bY,\bZ))\\
&\qquad\quad\cdot\int\psi_{R_{\beta_a}(\bZ)}^1d\mu_{C_o(\beta_a)}(\psi^1)\\
&\quad+\sum_{n=m+1}^{N(\beta_2)}\left(\frac{1}{h}\right)^{n-m}\sum_{\bZ\in I(\beta_a)^{n-m}}1_{\exists
 (\rho,\bx,\s,x,\theta)\in I(\beta_a)\text{ s.t. }(\rho,\bx,\s,x,\theta)
\subset\bZ\text{ and }x\in
 [\frac{\beta_1}{4},\beta_a-\frac{\beta_1}{4})_h}\\
&\quad\qquad\cdot\left(\begin{array}{c}n\\m\end{array}\right) J_n(\beta_a)(R_{\beta_a}(X,\bY),\bZ)\int\psi_{\bZ}^1d\mu_{C_o(\beta_a)}(\psi^1).
\end{align*}
By using \eqref{eq_general_determinant_bound_condition},
 \eqref{eq_general_determinant_bound_difference}, \eqref{eq_basic_temperature_lower_bound} and the inequality that
\begin{align}\label{eq_basic_temperature_bound}
&\frac{1}{\pi}|x-y|\le \left|\frac{\beta_a}{2\pi}(e^{i\frac{2\pi}{\beta_a}r_{\beta_a}(x)}-e^{i\frac{2\pi}{\beta_a}r_{\beta_a}(y)})\right|,\\
& (\forall x,y\in
 [-\beta_1/4,\beta_1/4),a\in \{1,2\}),\notag
\end{align}
we obtain 
\begin{align*}
|F_m(\beta_1)-F_m(\beta_2)|_l
&\le
 |J_m(\beta_1)-J_m(\beta_2)|_l\\
&\quad+\sum_{n=m+1}^{N(\beta_2)}\left(\begin{array}{c}n\\m\end{array}\right)
\Bigg(|J_n(\beta_1)-J_n(\beta_2)|_lD_{et}^{\frac{n-m}{2}}\\
&\qquad\qquad+\|J_n(\beta_2)\|_{l,0}D(\beta_1,\beta_2) D_{et}^{\frac{n-m}{2}}\Bigg)\\
&\quad+\frac{2\pi}{\beta_1}\sum_{n=m+1}^{N(\beta_2)}(n-m)\left(\begin{array}{c}n\\m\end{array}\right)\sum_{a=1}^2\|J_n(\beta_a)\|_{l,1}
D_{et}^{\frac{n-m}{2}}.\end{align*}
Finally, by substituting the inequalities 
$$\left(\begin{array}{c}n\\m\end{array}\right),\ (n-m)\left(\begin{array}{c}n\\m\end{array}\right)\le 2^{2n},$$
we reach the claimed inequality.
\end{proof}

\subsection{Estimation of the tree expansion at different temperatures}
\label{subsec_general_tree_expansion_temperature}
Here we estimate the differences between Grassmann polynomials produced
by the tree expansion at 2 different temperatures. In the same style as
in Subsection \ref{subsec_general_tree_expansion} we prepare necessary lemmas
step by step. Our strategy is to decompose $T(\beta_a)(\psi)$ into 2
parts. One is a polynomial which integrates at least one time-variable
away from $0$ and $\beta_a$. The other is a
polynomial which integrates only the time-variables close to $0$ or
$\beta_a$. 
We will find an upper bound on the first polynomial. We will measure the differences between the second
polynomial at $\beta_1$ and that at $\beta_2$. The next lemma is necessary to
bound the first polynomials.

\begin{lemma}\label{lem_tree_line_kernels_large}
Fix $a\in \{1,2\}$. Take any $m_j\in \N_{\ge 2}$ $(j=1,2,\cdots,n)$.
Let $J_{m_j}(\beta_a):I(\beta_a)^{m_j}\to \C$ $(j=1,2,\cdots,n)$ be anti-symmetric
 functions. Then, the following inequalities hold.
\begin{enumerate}
\item\label{item_tree_kernels_no_fixed_large}
 For any $X_{1,1}\in I^0$,
\begin{align}
&\Bigg|\frac{1}{n!}\sum_{T\in
 \T_n}Ope(T,C_o(\beta_a))\left(\frac{1}{h}\right)^{m_1-1}\sum_{(X_{1,2},X_{1,3},\cdots,X_{1,m_1})\in
 I(\beta_a)^{m_1-1}}J_{m_1}(\beta_a)(\bX_1)\label{eq_tree_kernels_no_fixed_large}\\
&\quad\cdot\prod_{j=2}^n
\Bigg(\left(\frac{1}{h}\right)^{m_j}\sum_{\bX_j\in
 I(\beta_a)^{m_j}}J_{m_j}(\beta_a)(\bX_j)\Bigg)\notag\\
&\quad\cdot 1_{\exists
 (\rho,\bx,\s,x,\theta)\in I(\beta_a)\text{ s.t. }(\rho,\bx,\s,x,\theta)\subset(\bX_1,\bX_2,\cdots,\bX_n)\text{ and }x\in [\frac{\beta_1}{4},\beta_a-\frac{\beta_1}{4})_h} \notag\\
&\quad\cdot\prod_{k=1\atop
 order}^n\psi_{\bX_k}^k\Bigg|_{\psi^j=0\atop (\forall j\in
 \{1,2,\cdots,n\})}\Bigg|\notag\\
&\le
 \frac{2\pi}{\beta_1}2^{3\sum_{j=1}^nm_j}D_{et}^{\frac{1}{2}\sum_{j=1}^nm_j-n+1}\notag\\
&\quad\cdot\prod_{j=1}^n\left(\sum_{q_j=0}^1\|J_{m_j}(\beta_a)\|_{l,q_j}\right)
\prod_{k=2}^n\left(\sum_{r_k=0}^1\|\widetilde{C_o}(\beta_a)\|_{l,r_k}\right)1_{\sum_{j=1}^nq_j+\sum_{k=2}^nr_k=1}.\notag
\end{align}
\item\label{item_tree_kernels_fixed_large}
In addition, assume that $k_j\in\{0,1,\cdots,m_j-1\}$ $(\forall j\in
     \{1,2,\cdots,$ $n\})$, $p\in \{1,2,\cdots,n\}$ and $k_1,k_p\ge
     1$. Then, for any $Y_{1,1}\in I^0$,
\begin{align}
&\Bigg|\frac{1}{n!}\sum_{T\in \T_n}Ope(T,C_o(\beta_a))\label{eq_tree_kernels_fixed_large}\\
&\quad\cdot\left(\frac{1}{h}\right)^{m_1-1}\sum_{\bX_1\in
 I(\beta_a)^{m_1-k_1}}\sum_{(Y_{1,2},Y_{1,3},\cdots,Y_{1,k_1})\in {\hat{I}}^{k_1-1}}J_{m_1}(\beta_a)(\bX_1,R_{\beta_a}(\bY_1))\notag\\
&\quad\cdot\prod_{j=2}^n\Bigg(\left(\frac{1}{h}\right)^{m_j}\sum_{\bX_j\in
 I(\beta_a)^{m_j-k_j}}\sum_{\bY_j\in {\hat{I}}^{k_j}}J_{m_j}(\beta_a)(\bX_j,R_{\beta_a}(\bY_j))\Bigg)\notag\\
&\quad\cdot e^{\sum_{j=0}^d(\frac{1}{\pi}\fw(l)\hat{d}_j(Y_{1,1},Y_{p,1}))^{\fr}}\notag\\
&\quad\cdot 1_{\exists
 (\rho,\bx,\s,x,\theta)\in I(\beta_a)\text{ s.t. }(\rho,\bx,\s,x,\theta)
\subset (\bX_1,\bX_2,\cdots,\bX_n)\text{ and }x\in [\frac{\beta_1}{4},\beta_a-\frac{\beta_1}{4})_h}\notag\\ 
&\quad\cdot \prod_{k=1\atop order}^n\psi_{\bX_k}^k\Bigg|_{\psi^j=0\atop (\forall j\in
 \{1,2,\cdots,n\})}\Bigg|\notag\\
&\le
\frac{2\pi}{\beta_1} 2^{3\sum_{j=1}^n(m_j-k_j)}D_{et}^{\frac{1}{2}\sum_{j=1}^n(m_j-k_j)-n+1}\notag\\
&\quad\cdot \prod_{j=1}^n\Bigg(\sum_{q_j=0}^1\|J_{m_j}(\beta_a)\|_{l,q_j}\Bigg)
\prod_{k=2}^n\Bigg(\sum_{r_k=0}^1\|\widetilde{C_o}(\beta_a)\|_{l,r_k}\Bigg)
1_{\sum_{j=1}^nq_j+\sum_{k=2}^nr_k=1}.\notag
\end{align}
 \end{enumerate}
\end{lemma}

\begin{proof}
We show the claim \eqref{item_tree_kernels_fixed_large} first. The proof for the claim \eqref{item_tree_kernels_no_fixed_large}
 is parallel to the proof of \eqref{item_tree_kernels_fixed_large}.

\eqref{item_tree_kernels_fixed_large}: Using
 \eqref{eq_basic_temperature_lower_bound}, we can prove that
\begin{align}
&1_{\exists (\rho,\bx,\s,x,\theta)\in I(\beta_a)\text{
 s.t. }(\rho,\bx,\s,x,\theta)\subset (\bX_1,\bX_2,\cdots,\bX_n)\text{
 and }x\in [\frac{\beta_1}{4},\beta_a-\frac{\beta_1}{4})_h}\label{eq_beta_dominant_bound}\\
&\le
 \frac{2\pi}{\beta_1}\sum_{q=1}^n\sum_{r=1}^{m_q-k_q}d_0(Y_{1,1},X_{q,r}).\notag
\end{align}
We can deduce from Lemma \ref{lem_application_determinant_bound_tree},
Lemma
 \ref{lem_tree_line_expansion} and \eqref{eq_beta_dominant_bound} that 
\begin{align}
&(\text{the left-hand side of \eqref{eq_tree_kernels_fixed_large}})\label{eq_tree_kernels_fixed_large_pre}\\
&\le \frac{2\pi}{\beta_1}\sum_{q=1}^n\sum_{r=1}^{m_q-k_q}\frac{1}{n!}\sum_{T\in\T_n} 1_{n_j(T)\le m_j-k_j (\forall j\in
 \{1,2,\cdots,n\})}2^{n-1}D_{et}^{\frac{1}{2}\sum_{j=1}^n(m_j-k_j)-n+1}\notag\\
&\quad\cdot \left(\frac{1}{h}\right)^{m_1-1}\sum_{\bX_1\in
 I(\beta_a)^{m_1-k_1}}\sum_{(Y_{1,2},Y_{1,3},\cdots,Y_{1,k_1})\in
 \hat{I}^{k_1-1}}|J_{m_1}(\beta_a)(\bX_1,R_{\beta_a}(\bY_1))|\notag\\
&\quad\cdot \sum_{\bW_1\subset \bX_1\atop \bW_1\in 
 I(\beta_a)^{n_1(T)}}\sum_{\s_1\in
 \S_{n_1(T)}}\notag\\
&\quad\cdot\prod_{\{1,s\}\in
 L_1^1(T)}\Bigg(\left(\frac{1}{h}\right)^{m_s}\sum_{\bX_s\in
 I(\beta_a)^{m_s-k_s}}\sum_{\bY_s\in \hat{I}^{k_s}}|J_{m_s}(\beta_a)(\bX_s,R_{\beta_a}(\bY_s))|\notag\\
&\qquad\qquad\qquad\cdot \sum_{Z_s\subset \bX_s\atop Z_s\in I(\beta_a)}|\widetilde{C_o}(\beta_a)(W_{1,\s_1\circ
 \zeta_1(\{1,s\})},Z_s)|\Bigg)\notag\\
&\quad \cdot \prod_{u=1\atop order
 }^{d_T(1)-1}\Bigg(\prod_{j\in\{2,3,\cdots,n\}\text{ with }\atop
 \dis_T(1,j)=u,n_j(T)\neq 1}\Bigg(\sum_{\bW_j\subset \bX_j\backslash
  Z_j\atop \bW_j\in
 I(\beta_a)^{n_j(T)-1}}\sum_{\s_j\in\S_{n_j(T)-1}}\notag\\
&\qquad\cdot \prod_{\{j,s\}\in L_j^1(T)}\Bigg(\left(\frac{1}{h}\right)^{m_s}
\sum_{\bX_s\in I(\beta_a)^{m_s-k_s}}\sum_{\bY_s\in \hat{I}^{k_s}}|J_{m_s}(\beta_a)(\bX_s,R_{\beta_a}(\bY_s))|\notag
\\
&\qquad\qquad\qquad\cdot\sum_{Z_s\subset
 \bX_s\atop 
Z_s\in I(\beta_a)}|\widetilde{C_o}(\beta_a)(W_{j,\s_j\circ\zeta_j(\{j,s\})},Z_s)|
\Bigg)
\Bigg)\Bigg)\notag\\
&\quad\cdot d_0(Y_{1,1},X_{q,r})e^{\sum_{j=0}^d(\fw(l)d_j(Y_{1,1},R_{\beta_a}(Y_{p,1})))^{\fr}},\notag
\end{align}
where we also used \eqref{eq_basic_temperature_bound} to justify the inequality
$$
e^{\sum_{j=0}^d(\frac{1}{\pi}\fw(l)\hat{d}_j(Y,Y_{p,1}))^{\fr}}\le
e^{\sum_{j=0}^d(\fw(l)d_j(Y,R_{\beta_a}(Y_{p,1})))^{\fr}}.
$$
We can estimate the right-hand side of
 \eqref{eq_tree_kernels_fixed_large_pre} by straightforwardly following
 the argument in Subsection \ref{subsec_general_tree_expansion} leading
 to Lemma \ref{lem_tree_sum_kernels}
 \eqref{item_tree_sum_kernels_fixed}.
This procedure is summarized as follows.
\begin{align}
&(\text{The right-hand side of \eqref{eq_tree_kernels_fixed_large_pre}})\label{eq_tree_kernels_fixed_large_next}\\
&\le \frac{2\pi}{\beta_1}\sum_{q=1}^n(m_q-k_q)\frac{1}{n!}\sum_{T\in \T_n}
1_{n_j(T)\le m_j-k_j (\forall j\in
 \{1,2,\cdots,n\})}2^{n-1}D_{et}^{\frac{1}{2}\sum_{j=1}^n(m_j-k_j)-n+1}\notag\\
&\quad\cdot\prod_{i=1}^n\left((m_i-k_i)\left(\begin{array}{c}m_i-k_i-1\\
					     n_i(T)-1\end{array}\right)(n_i(T)-1)!\right)\notag\\
&\quad\cdot\prod_{j=1}^n\Bigg(\sum_{q_j=0}^1\|J_{m_j}(\beta_a)\|_{l,q_j}\Bigg)
\prod_{k=2}^n\Bigg(\sum_{r_k=0}^1\|\widetilde{C_o}(\beta_a)\|_{l,r_k}\Bigg)
1_{\sum_{j=1}^nq_j+\sum_{k=2}^nr_k=1}\notag\\
&\le \frac{2\pi}{\beta_1} 2^{3\sum_{j=1}^n(m_j-k_j)}D_{et}^{\frac{1}{2}\sum_{j=1}^n(m_j-k_j)-n+1}\notag\\
&\quad\cdot\prod_{j=1}^n\Bigg(\sum_{q_j=0}^1\|J_{m_j}(\beta_a)\|_{l,q_j}\Bigg)
\prod_{k=2}^n\Bigg(\sum_{r_k=0}^1\|\widetilde{C_o}(\beta_a)\|_{l,r_k}\Bigg)
1_{\sum_{j=1}^nq_j+\sum_{k=2}^nr_k=1}.\notag
\end{align}
To derive the first inequality we followed the proof of Lemma
 \ref{lem_tree_line_kernels} \eqref{item_tree_kernels_fixed}. 
Then, we applied Lemma \ref{lem_tree_combinatorial_factor} to derive the
 second inequality.

\eqref{item_tree_kernels_no_fixed_large}: We let $X_{1,1}\in I^0$ play the same
 role as $Y_{1,1}\in I^{0}$ did in the proof of
 \eqref{item_tree_kernels_fixed_large}. Application of Lemma
 \ref{lem_tree_line_expansion}, \eqref{eq_beta_dominant_bound} and repetition of the argument in
 Subsection \ref{subsec_general_tree_expansion} leading to Lemma
 \ref{lem_tree_sum_kernels} \eqref{item_tree_sum_kernels_fixed} ensure that
\begin{align*}
&(\text{the left-hand side of \eqref{eq_tree_kernels_no_fixed_large}})\\
&\le
 \frac{2\pi}{\beta_1}\sum_{q=1}^n\sum_{r=1}^{m_q}\frac{1}{n!}\sum_{T\in \T_n}
 1_{n_j(T)\le m_j (\forall j\in
 \{1,2,\cdots,n\})}2^{n-1}D_{et}^{\frac{1}{2}\sum_{j=1}^nm_j-n+1}\\
&\quad\cdot
 \left(\frac{1}{h}\right)^{m_1-1}\sum_{(X_{1,2},X_{1,3},\cdots,X_{1,m_1})\in I(\beta_a)^{m_1-1}}|J_{m_1}(\beta_a)(\bX_1)|\sum_{\bW_1\subset \bX_1\atop
\bW_1\in I(\beta_a)^{n_1(T)}}\sum_{\s_1\in
 \S_{n_1(T)}}\\
&\quad\cdot\prod_{\{1,s\}\in
 L_1^1(T)}\Bigg(\left(\frac{1}{h}\right)^{m_s}\sum_{\bX_s\in
 I(\beta_a)^{m_s}}|J_{m_s}(\beta_a)(\bX_s)|\\
&\qquad\qquad\qquad\cdot \sum_{Z_s\subset \bX_s\atop
Z_s\in I(\beta_a)}|\widetilde{C_o}(\beta_a)(W_{1,\s_1\circ
 \zeta_1(\{1,s\})},Z_s)|\Bigg)\\
&\quad \cdot \prod_{u=1\atop order
 }^{d_T(1)-1}\Bigg(\prod_{j\in\{2,3,\cdots,n\}\text{ with }\atop
 \dis_T(1,j)=u,n_j(T)\neq 1}\Bigg(\sum_{\bW_j\subset \bX_j\backslash
 Z_j\atop \bW_j\in
 I(\beta_a)^{n_j(T)-1}}\sum_{\s_j\in\S_{n_j(T)-1}}\\
&\qquad\cdot \prod_{\{j,s\}\in L_j^1(T)}\Bigg(\left(\frac{1}{h}\right)^{m_s}
\sum_{\bX_s\in I(\beta_a)^{m_s}}|J_{m_s}(\beta_a)(\bX_s)|\\
&\quad\qquad\qquad\qquad\cdot\sum_{Z_s\subset
 \bX_s\atop Z_s\in I(\beta_a)}|\widetilde{C_o}(\beta_a)(W_{j,\s_j\circ\zeta_j(\{j,s\})},Z_s)|
\Bigg)
\Bigg)\Bigg)\\
&\quad\cdot d_0(X_{1,1},X_{q,r})\\
&\le \frac{2\pi}{\beta_1}\sum_{q=1}^nm_q\frac{1}{n!}\sum_{T\in \T_n}
1_{n_j(T)\le m_j(\forall j\in
 \{1,2,\cdots,n\})}2^{n-1}D_{et}^{\frac{1}{2}\sum_{j=1}^nm_j-n+1}\\
&\quad\cdot\prod_{i=1}^n\left(m_i\left(\begin{array}{c}m_i-1\\
					     n_i(T)-1\end{array}\right)(n_i(T)-1)!\right)\\
&\quad\cdot\prod_{j=1}^n\Bigg(\sum_{q_j=0}^1\|J_{m_j}(\beta_a)\|_{l,q_j}\Bigg)
\prod_{k=2}^n\Bigg(\sum_{r_k=0}^1\|\widetilde{C_o}(\beta_a)\|_{l,r_k}\Bigg)
1_{\sum_{j=1}^nq_j+\sum_{k=2}^nr_k=1},
\end{align*}
which is proved to be less than or equal to the right-hand side of
 \eqref{eq_tree_kernels_no_fixed_large} by Lemma \ref{lem_tree_combinatorial_factor}. 
\end{proof}

Next we construct necessary lemmas to measure the differences between the
polynomial containing only the time-integrals close to 0 or $\beta_1$
and that containing only the time-integrals close to 0 or $\beta_2$.

\begin{lemma}\label{lem_application_determinant_bound_tree_difference}
Take any $T\in \T_n$, $m_j\in\N$ and $\bX_j\in \hat{I}^{m_j}$
 $(j=1,2,\cdots,n)$. The following inequality holds. 
\begin{align*}
&\Bigg|ope(T,C_o(\beta_1))\prod_{j=1\atop
 order}^n\psi_{R_{\beta_1}(\bX_j)}^j\Bigg|_{\psi^j=0\atop (\forall j\in
 \{1,2,\cdots,n\})}\\
&\quad-ope(T,C_o(\beta_2))\prod_{j=1\atop
 order}^n\psi_{R_{\beta_2}(\bX_j)}^j\Bigg|_{\psi^j=0\atop (\forall j\in
 \{1,2,\cdots,n\})}\Bigg|\\
&\le D(\beta_1,\beta_2)D_{et}^{\frac{1}{2}\sum_{j=1}^nm_j}.
\end{align*}
\end{lemma}
\begin{proof}
The result follows from the equality
 \eqref{eq_property_irrelevant_function}, the properties of the matrix
 $M_{at}(T,\xi,\bs)$ and the assumption
 \eqref{eq_general_determinant_bound_difference}.
\end{proof}

Here let us introduce a couple of notations to organize  
formulas. Let $m\in \{2,3,\cdots,N(\beta_2)\}$. 
For anti-symmetric functions $J_m(\beta_a):I(\beta_a)^m\to \C$
$(a=1,2)$, $\bX\in \hat{I}^m$, $i,j\in
\N$ and $l\in \Z$, set
\begin{align*}
&J_m^{(i,j)}(\bX):=\left\{\begin{array}{ll}
		     J_m(\beta_1)(R_{\beta_1}(\bX))-J_m(\beta_2)(R_{\beta_2}(\bX))&\text{if
		      }i=j,\\
J_m(\beta_1)(R_{\beta_1}(\bX))&\text{if }i<j,\\
J_m(\beta_2)(R_{\beta_2}(\bX))&\text{if }i>j,
\end{array}
\right.\\
&A(J_m(\beta_1),J_m(\beta_2))_l^{(i,j)}\\
&:=\left\{\begin{array}{ll}
		     |J_m(\beta_1)-J_m(\beta_2)|_l+\frac{2\pi}{\beta_1}(m-1)\sum_{a=1}^2\|J_m(\beta_a)\|_{l,1}&\text{if }i=j,\\
\sum_{a=1}^2\|J_m(\beta_a)\|_{l,0}&\text{if }i\neq j.
\end{array}
\right.
\end{align*}

\begin{lemma}\label{lem_tree_line_expansion_difference}
Take any $T\in \T_n$, $m_j\in\N$ and $\bX_j\in \hat{I}^{m_j}$
 $(j=1,2,\cdots,n)$. The following inequality holds. 
\begin{align}
&\Bigg|Ope(T,C_o(\beta_1))\prod_{j=1\atop
 order}^n\psi_{R_{\beta_1}(\bX_j)}^j\Bigg|_{\psi^j=0\atop (\forall j\in
 \{1,2,\cdots,n\})}\label{eq_tree_line_expansion_difference}\\
&\quad-Ope(T,C_o(\beta_2))\prod_{j=1\atop
 order}^n\psi_{R_{\beta_2}(\bX_j)}^j\Bigg|_{\psi^j=0\atop (\forall j\in
 \{1,2,\cdots,n\})}\Bigg|\notag\\
&\le 1_{n_j(T)\le m_j (\forall j\in
 \{1,2,\cdots,n\})}2^{n-1}D_{et}^{\frac{1}{2}\sum_{j=1}^nm_j-n+1}\sum_{v=1}^n(1_{v=1}D(\beta_1,\beta_2)+
1_{v\ge 2})\notag\\
&\quad\cdot \sum_{\bW_1\subset \bX_1\atop 
\bW_1\in \hat{I}^{n_1(T)}}\sum_{\s_1\in
 \S_{n_1(T)}}\prod_{\{1,s\}\in L_1^1(T)}\Bigg(\sum_{Z_s\subset \bX_s\atop Z_s\in \hat{I}}|\widetilde{C_o}^{(v,s)}(W_{1,\s_1\circ
 \zeta_1(\{1,s\})},Z_s)|\Bigg)\notag\\
&\quad \cdot \prod_{u=1\atop order
 }^{d_T(1)-1}\Bigg(\prod_{j\in\{2,3,\cdots,n\}\text{ with }\atop
 \dis_T(1,j)=u,n_j(T)\neq 1}\Bigg(\sum_{\bW_j\subset \bX_j\backslash
 Z_j\atop \bW_j\in
 \hat{I}^{n_j(T)-1}}\sum_{\s_j\in\S_{n_j(T)-1}}\notag\\
&\qquad\cdot \prod_{\{j,s\}\in L_j^1(T)}\Bigg(\sum_{Z_s\subset \bX_s
\atop Z_s\in \hat{I}}|\widetilde{C_o}^{(v,s)}(W_{j,\s_j\circ\zeta_j(\{j,s\})},Z_s)|
\Bigg)
\Bigg)\Bigg).\notag
\end{align}
\end{lemma}
\begin{proof}
Letting the operators act on the input Grassmann monomials in the
 same way as in the proof of Lemma \ref{lem_tree_line_expansion} and
 using Lemma \ref{lem_application_determinant_bound_tree_difference}
 yield that
\begin{align*}
&(\text{the left-hand side of
 \eqref{eq_tree_line_expansion_difference}})\\
&\le 1_{n_j(T)\le m_j(\forall j\in
 \{1,2,\cdots,n\})}2^{n-1}\\
&\quad\cdot\sum_{\bW_1\subset \bX_1\atop
\bW_1\in
 \hat{I}^{n_1(T)}}\sum_{\s_1\in
 \S_{n_1(T)}}\prod_{\{1,s\}\in L_1^1(T)}\Bigg(\sum_{Z_s\subset
 \bX_s\atop Z_s\in \hat{I}}|\widetilde{C_o}(\beta_1)(R_{\beta_1}(W_{1,\s_1\circ
 \zeta_1(\{1,s\})},Z_s))|\Bigg)\\
&\quad \cdot \prod_{u=1\atop order
 }^{d_T(1)-1}\Bigg(\prod_{j\in\{2,3,\cdots,n\}\text{ with }\atop
 \dis_T(1,j)=u, n_j(T)\neq 1}\Bigg(\sum_{\bW_j\subset \bX_j\backslash
 Z_s\atop \bW_j\in
 \hat{I}^{n_j(T)-1}}\sum_{\s_j\in\S_{n_j(T)-1}}\\
&\qquad\cdot \prod_{\{j,s\}\in L_j^1(T)}\Bigg(\sum_{Z_s\subset \bX_s\atop
Z_s\in \hat{I}}|\widetilde{C_o}(\beta_1)(R_{\beta_1}(W_{j,\s_j\circ\zeta_j(\{j,s\})},Z_s))|\Bigg)
\Bigg)\Bigg)\\
&\quad\cdot \Bigg|
 ope(T,C_o(\beta_1))\psi^1_{R_{\beta_1}(\bX_1\backslash\bW_1)}\\
&\qquad\cdot\prod_{k=2\atop
 order}^n(1_{n_k(T)\neq 1}\psi_{R_{\beta_1}((\bX_k\backslash Z_k)\backslash \bW_k)}^k+1_{n_k(T)= 1}\psi_{R_{\beta_1}(\bX_k\backslash Z_k)}^k)\Bigg|_{\psi^j=0\atop (\forall j\in
 \{1,2,\cdots,n\})}\\
&\qquad -
 ope(T,C_o(\beta_2))\psi^1_{R_{\beta_2}(\bX_1\backslash\bW_1)}\\
&\qquad\quad\cdot\prod_{k=2\atop
 order}^n(1_{n_k(T)\neq 1}\psi_{R_{\beta_2}((\bX_k\backslash Z_k)\backslash \bW_k)}^k+1_{n_k(T)= 1}\psi_{R_{\beta_2}(\bX_k\backslash Z_k)}^k)\Bigg|_{\psi^j=0\atop (\forall j\in
 \{1,2,\cdots,n\})}\Bigg|\\
&\quad+1_{n_j(T)\le m_j (\forall j\in
 \{1,2,\cdots,n\})}2^{n-1}\\
&\quad\cdot \sum_{\bW_1\subset \bX_1\atop 
\bW_1\in
 \hat{I}^{n_1(T)}}\sum_{\s_1\in
 \S_{n_1(T)}}\prod_{\{1,s\}\in L_1^1(T)}\Bigg(\sum_{Z_s\subset \bX_s\atop
Z_s\in \hat{I}}\Bigg)\\
&\quad \cdot \prod_{u=1\atop order
 }^{d_T(1)-1}\Bigg(\prod_{j\in\{2,3,\cdots,n\}\text{ with }\atop
 \dis_T(1,j)=u,n_j(T)\neq 1}\Bigg(\sum_{\bW_j\subset \bX_j\backslash
 Z_s\atop \bW_j\in
 \hat{I}^{n_j(T)-1}}\sum_{\s_j\in\S_{n_j(T)-1}}\prod_{\{j,s\}\in L_j^1(T)}\Bigg(\sum_{Z_s\subset \bX_s\atop
Z_s\in \hat{I}}\Bigg)
\Bigg)\Bigg)\\
&\quad\cdot \Bigg|\prod_{j\in\{1,2,\cdots,n\}\atop\text{with }
 n_j(T)\neq 1}\prod_{\{j,s\}\in L_j^1(T)}
\widetilde{C_o}(\beta_1)(R_{\beta_1}(W_{j,\s_j\circ\zeta_j(\{j,s\})},Z_s))\\
&\qquad - \prod_{j\in\{1,2,\cdots,n\}\atop\text{with }
 n_j(T)\neq 1}\prod_{\{j,s\}\in L_j^1(T)}
\widetilde{C_o}(\beta_2)(R_{\beta_2}(W_{j,\s_j\circ\zeta_j(\{j,s\})},Z_s))\Bigg|\\
&\quad\cdot \Bigg|
 ope(T,C_o(\beta_2))\psi^1_{R_{\beta_2}(\bX_1\backslash\bW_1)}\\
&\qquad\cdot \prod_{k=2\atop
 order}^n(1_{n_k(T)\neq 1}\psi_{R_{\beta_2}((\bX_k\backslash Z_k)\backslash \bW_k)}^k+1_{n_k(T)= 1}\psi_{R_{\beta_2}(\bX_k\backslash Z_k)}^k)\Bigg|_{\psi^j=0\atop (\forall j\in
 \{1,2,\cdots,n\})}\Bigg|\\
&\le 1_{n_j(T)\le m_j(\forall j\in
 \{1,2,\cdots,n\})}2^{n-1}D_{et}^{\frac{1}{2}\sum_{j=1}^nm_j-n+1}D(\beta_1,\beta_2)\\
&\quad\cdot \sum_{\bW_1\subset \bX_1\atop
\bW_1\in
 \hat{I}^{n_1(T)}}\sum_{\s_1\in
 \S_{n_1(T)}}\prod_{\{1,s\}\in L_1^1(T)}\Bigg(\sum_{Z_s\subset \bX_s
\atop Z_s\in \hat{I}}|\widetilde{C_o}(\beta_1)(R_{\beta_1}(W_{1,\s_1\circ
 \zeta_1(\{1,s\})},Z_s))|\Bigg)\\
&\quad \cdot \prod_{u=1\atop order
 }^{d_T(1)-1}\Bigg(\prod_{j\in\{2,3,\cdots,n\}\text{ with }\atop
 \dis_T(1,j)=u,n_j(T)\neq 1}\Bigg(\sum_{\bW_j\subset \bX_j\backslash
 Z_j\atop \bW_j\in  \hat{I}^{n_j(T)-1}}\sum_{\s_j\in\S_{n_j(T)-1}}\\
&\qquad\cdot \prod_{\{j,s\}\in L_j^1(T)}\Bigg(\sum_{Z_s\subset
 \bX_s\atop 
Z_s\in \hat{I}}|\widetilde{C_o}(\beta_1)(R_{\beta_1}(W_{j,\s_j\circ\zeta_j(\{j,s\})},Z_s))|\Bigg)
\Bigg)\Bigg)\\
&\quad +1_{n_j(T)\le m_j(\forall j\in
 \{1,2,\cdots,n\})}2^{n-1}D_{et}^{\frac{1}{2}\sum_{j=1}^nm_j-n+1}\sum_{v=2}^n\\
&\qquad\cdot\sum_{\bW_1\subset \bX_1\atop
\bW_1\in
 \hat{I}^{n_1(T)}}\sum_{\s_1\in
 \S_{n_1(T)}}\prod_{\{1,s\}\in L_1^1(T)}\Bigg(\sum_{Z_s\subset \bX_s
\atop Z_s\in \hat{I}}|\widetilde{C_o}^{(v,s)}(W_{1,\s_1\circ
 \zeta_1(\{1,s\})},Z_s)|\Bigg)\\
&\qquad \cdot \prod_{u=1\atop order
 }^{d_T(1)-1}\Bigg(\prod_{j\in\{2,3,\cdots,n\}\text{ with }\atop
 \dis_T(1,j)=u,n_j(T)\neq 1}\Bigg(\sum_{\bW_j\subset \bX_j\backslash
 Z_j\atop \bW_j\in  \hat{I}^{n_j(T)-1}}\sum_{\s_j\in\S_{n_j(T)-1}}\\
&\qquad\quad\cdot \prod_{\{j,s\}\in L_j^1(T)}\Bigg(\sum_{Z_s\subset
 \bX_s\atop 
Z_s\in \hat{I}}|\widetilde{C_o}^{(v,s)}(W_{j,\s_j\circ\zeta_j(\{j,s\})},Z_s)|\Bigg)
\Bigg)\Bigg),
\end{align*}
which is equal to the right-hand side of \eqref{eq_tree_line_expansion_difference}.
\end{proof}

\begin{lemma}\label{lem_tree_sum_kernels_difference}
Take any $m_j\in\{2,3,\cdots,N(\beta_2)\}$ $(j=1,2,\cdots,n)$. Let $J_{m_j}(\beta_a):I(\beta_a)^{m_j}\to\C$
 $(j=1,2,\cdots,n,\ a=1,2)$ be anti-symmetric functions. Then, the following
 inequalities hold. 
\begin{enumerate}
\item\label{item_tree_sum_no_fixed_difference} For any $X_{1,1}\in I^0$,
\begin{align*}
&\Bigg|\frac{1}{n!}\sum_{T\in
 \T_n}Ope(T,C_o(\beta_1))\left(\frac{1}{h}\right)^{m_1-1}\sum_{(X_{1,2},X_{1,3},\cdots,X_{1,m_1})\in
 \hat{I}^{m_1-1}}J_{m_1}(\beta_1)(R_{\beta_1}(\bX_1))\\
&\quad\cdot\prod_{j=2}^n\Bigg(\left(\frac{1}{h}\right)^{m_j}\sum_{\bX_j\in \hat{I}^{m_j}}J_{m_j}(\beta_1)(R_{\beta_1}(\bX_j))\Bigg)\prod_{k=1\atop
 order}^n\psi_{R_{\beta_1}(\bX_k)}^k\Bigg|_{\psi^j=0\atop (\forall j\in
 \{1,2,\cdots,n\})}\\
&-\frac{1}{n!}\sum_{T\in
 \T_n}Ope(T,C_o(\beta_2))\\
&\quad\cdot\left(\frac{1}{h}\right)^{m_1-1}\sum_{(X_{1,2},X_{1,3},\cdots,X_{1,m_1})\in
 \hat{I}^{m_1-1}}J_{m_1}(\beta_2)(R_{\beta_2}(\bX_1))\\
&\quad\cdot\prod_{j=2}^n\Bigg(\left(\frac{1}{h}\right)^{m_j}\sum_{\bX_j\in \hat{I}^{m_j}}J_{m_j}(\beta_2)(R_{\beta_2}(\bX_j))\Bigg)\prod_{k=1\atop
 order}^n\psi_{R_{\beta_2}(\bX_k)}^k\Bigg|_{\psi^j=0\atop (\forall j\in
 \{1,2,\cdots,n\})}\Bigg|\\
&\le
 2^{2\sum_{j=1}^nm_j}D_{et}^{\frac{1}{2}\sum_{j=1}^nm_j-n+1}\Bigg(\|\widetilde{C_o}(\beta_1)\|_{l,0}^{n-1}\sum_{v=1}^n\prod_{j=1}^nA(J_{m_j}(\beta_1),J_{m_j}(\beta_2))_l^{(v,j)}\\
&\qquad
 +\sum_{v=2}^n\prod_{j=2}^nA(\widetilde{C_o}(\beta_1),\widetilde{C_o}(\beta_2))_l^{(v,j)}\prod_{k=1}^n\|J_{m_k}(\beta_2)\|_{l,0}\\
&\qquad+D(\beta_1,\beta_2)\sum_{a=1}^2\|\widetilde{C_o}(\beta_a)\|_{l,0}^{n-1}\prod_{j=1}^n\|J_{m_j}(\beta_2)\|_{l,0}\Bigg).
\end{align*}
\item\label{item_tree_sum_fixed_difference}
In addition, assume that $k_j\in\{0,1,\cdots,m_j-1\}$ $(\forall j\in
     \{1,2,\cdots,$ $n\})$,
     $p\in \{1,2,\cdots,n\}$ and $k_1,k_p\ge 1$. Then, for any $Y_{1,1}\in I^0$,
\begin{align}
&\Bigg|\frac{1}{n!}\sum_{T\in \T_n}Ope(T,C_o(\beta_1))\label{eq_difference_tree_pre_external}\\
&\quad\cdot\left(\frac{1}{h}\right)^{m_1-1}\sum_{\bX_1\in
 \hat{I}^{m_1-k_1}}\sum_{(Y_{1,2},Y_{1,3},\cdots,Y_{1,k_1})\in {\hat{I}}^{k_1-1}}J_{m_1}(\beta_1)(R_{\beta_1}(\bX_1,\bY_1))\notag\\
&\quad\cdot\prod_{j=2}^n\Bigg(\left(\frac{1}{h}\right)^{m_j}\sum_{\bX_j\in
 \hat{I}^{m_j-k_j}}\sum_{\bY_j\in {\hat{I}}^{k_j}}J_{m_j}(\beta_1)(R_{\beta_1}(\bX_j,\bY_j))\Bigg)\notag\\
&\quad\cdot e^{\sum_{j=0}^d(\frac{1}{\pi}\fw(l)\hat{d}_j(Y_{1,1},Y_{p,1}))^{\fr}}\prod_{k=1\atop order}^n\psi_{R_{\beta_1}(\bX_k)}^k\Bigg|_{\psi^j=0\atop (\forall j\in
 \{1,2,\cdots,n\})}\notag\\
&-\frac{1}{n!}\sum_{T\in \T_n}Ope(T,C_o(\beta_2))\notag\\
&\quad\cdot\left(\frac{1}{h}\right)^{m_1-1}\sum_{\bX_1\in
 \hat{I}^{m_1-k_1}}\sum_{(Y_{1,2},Y_{1,3},\cdots,Y_{1,k_1})\in {\hat{I}}^{k_1-1}}J_{m_1}(\beta_2)(R_{\beta_2}(\bX_1,\bY_1))\notag\\
&\quad\cdot\prod_{j=2}^n\Bigg(\left(\frac{1}{h}\right)^{m_j}\sum_{\bX_j\in
 \hat{I}^{m_j-k_j}}\sum_{\bY_j\in {\hat{I}}^{k_j}}J_{m_j}(\beta_2)(R_{\beta_2}(\bX_j,\bY_j))\Bigg)\notag\\
&\quad\cdot e^{\sum_{j=0}^d(\frac{1}{\pi}\fw(l)\hat{d}_j(Y_{1,1},Y_{p,1}))^{\fr}}\prod_{k=1\atop order}^n\psi_{R_{\beta_2}(\bX_k)}^k\Bigg|_{\psi^j=0\atop (\forall j\in
 \{1,2,\cdots,n\})}\Bigg|\notag\\
&\le
 2^{2\sum_{j=1}^n(m_j-k_j)}D_{et}^{\frac{1}{2}\sum_{j=1}^n(m_j-k_j)-n+1}\notag\\
&\quad\cdot\Bigg(\|\widetilde{C_o}(\beta_1)\|_{l,0}^{n-1}\sum_{v=1}^n\prod_{j=1}^nA(J_{m_j}(\beta_1),J_{m_j}(\beta_2))_l^{(v,j)}\notag\\
&\qquad
 +\sum_{v=2}^n\prod_{j=2}^nA(\widetilde{C_o}(\beta_1),\widetilde{C_o}(\beta_2))_l^{(v,j)}\prod_{k=1}^n\|J_{m_k}(\beta_2)\|_{l,0}\notag\\
&\qquad+D(\beta_1,\beta_2)\sum_{a=1}^2\|\widetilde{C_o}(\beta_a)\|_{l,0}^{n-1}\prod_{j=1}^n\|J_{m_j}(\beta_2)\|_{l,0}\Bigg).\notag
\end{align}
\end{enumerate}
\end{lemma}
\begin{proof} We present the proof of
 \eqref{item_tree_sum_fixed_difference} first. The claim \eqref{item_tree_sum_no_fixed_difference} can be proved
 similarly.

\eqref{item_tree_sum_fixed_difference}: Note that
\begin{align}
&(\text{the left-hand side of \eqref{eq_difference_tree_pre_external}})\label{eq_tree_fixed_difference_decomposed}\\
&\le \sum_{v=1}^n \Bigg|\frac{1}{n!}\sum_{T\in\T_n}Ope(T,C_o(\beta_1))
 \notag\\
&\quad\cdot\left(\frac{1}{h}\right)^{m_1-1}\sum_{\bX_1\in
 \hat{I}^{m_1-k_1}}\sum_{(Y_{1,2},Y_{1,3},\cdots,Y_{1,k_1})\in
 {\hat{I}}^{k_1-1}} J_{m_1}^{(v,1)}(\bX_1,\bY_1)\notag\\
&\quad\cdot\prod_{j=2}^n\Bigg(\left(\frac{1}{h}\right)^{m_j}\sum_{\bX_j\in
 \hat{I}^{m_j-k_j}}\sum_{\bY_j\in
 {\hat{I}}^{k_j}}J_{m_j}^{(v,j)}(\bX_j,\bY_j)\Bigg)e^{\sum_{j=0}^d(\frac{1}{\pi}\fw(l)\hat{d}_j(Y_{1,1},Y_{p,1}))^{\fr}}\notag\\
&\quad\cdot \prod_{k=1\atop order}^n\psi_{R_{\beta_1}(\bX_k)}^k\Bigg|_{\psi^j=0\atop (\forall j\in
 \{1,2,\cdots,n\})}\Bigg|\notag\\
&\quad+\Bigg|\frac{1}{n!}\sum_{T\in\T_n}\notag\\ 
&\qquad\cdot \left(\frac{1}{h}\right)^{m_1-1}\sum_{\bX_1\in
 \hat{I}^{m_1-k_1}}
\sum_{(Y_{1,2},Y_{1,3},\cdots,Y_{1,k_1})\in {\hat{I}}^{k_1-1}}J_{m_1}(\beta_2)(R_{\beta_2}(\bX_1,\bY_1))\notag\\
&\qquad\cdot\prod_{j=2}^n\Bigg(\left(\frac{1}{h}\right)^{m_j}\sum_{\bX_j\in
 \hat{I}^{m_j-k_j}}\sum_{\bY_j\in
 {\hat{I}}^{k_j}}J_{m_j}(\beta_2)(R_{\beta_2}(\bX_j,\bY_j))\Bigg)\notag\\
&\qquad\cdot e^{\sum_{j=0}^d(\frac{1}{\pi}\fw(l)\hat{d}_j(Y_{1,1},Y_{p,1}))^{\fr}}\Bigg(Ope(T,C_o(\beta_1))\prod_{k=1\atop order}^n\psi_{R_{\beta_1}(\bX_k)}^k\Bigg|_{\psi^j=0\atop (\forall j\in
 \{1,2,\cdots,n\})}\notag\\
&\qquad\qquad\qquad\qquad\qquad\qquad-Ope(T,C_o(\beta_2))\prod_{k=1\atop order}^n\psi_{R_{\beta_2}(\bX_k)}^k\Bigg|_{\psi^j=0\atop (\forall j\in
 \{1,2,\cdots,n\})}\Bigg)\Bigg|.\notag
\end{align}

By the same procedure leading to the proof of Lemma
 \ref{lem_tree_sum_kernels} \eqref{item_tree_sum_kernels_fixed} we have that
\begin{align}
&(\text{the first term in the right-hand side of
 \eqref{eq_tree_fixed_difference_decomposed}})\label{eq_tree_fixed_difference_kernel_decomposition}\\
&\le
 2^{2\sum_{j=1}^n(m_j-k_j)}D_{et}^{\frac{1}{2}\sum_{j=1}^n(m_j-k_j)-n+1}\notag\\&\quad\cdot \Bigg(\sup_{X\in\hat{I}}\frac{1}{h}\sum_{Y\in \hat{I}} e^{\sum_{j=0}^d(\frac{1}{\pi}\fw(l)\hat{d}_j(X,Y))^{\fr}}
|\widetilde{C_o}(\beta_1)(R_{\beta_1}(X,Y))|\Bigg)^{n-1}\notag\\
&\quad\cdot\sum_{v=1}^n\prod_{i=1}^n\Bigg(
\sup_{X\in \hat{I}}\frah^{m_i-1}\sum_{\bY\in \hat{I}^{m_i-1}}
e^{\sum_{j=0}^d(\frac{1}{\pi}\fw(l)\hat{d}_j(X,Y_{1}))^{\fr}}
|J_{m_i}^{(v,i)}(X,\bY)|\Bigg).\notag
\end{align}
It follows from \eqref{eq_basic_temperature_bound} that
\begin{align}
&\sup_{X\in\hat{I}}\frac{1}{h}\sum_{Y\in \hat{I}} e^{\sum_{j=0}^d(\frac{1}{\pi}\fw(l)\hat{d}_j(X,Y))^{\fr}}
|\widetilde{C_o}(\beta_1)(R_{\beta_1}(X,Y))|\le
 \|\widetilde{C_o}(\beta_1)\|_{l,0},\label{eq_separate_estimation_covariance_kernel}\\
&\sup_{X\in \hat{I}}\frah^{m_i-1}\sum_{\bY\in \hat{I}^{m_i-1}}
e^{\sum_{j=0}^d(\frac{1}{\pi}\fw(l)\hat{d}_j(X,Y_{1}))^{\fr}}
|J_{m_i}^{(v,i)}(X,\bY)|\notag\\
&\le
 \sum_{a=1}^2\|J_{m_i}(\beta_a)\|_{l,0}=A(J_{m_i}(\beta_1),J_{m_i}(\beta_2))_l^{(v,i)},\notag
\end{align}
if $v\neq i$. Note that for any
 $x,y\in [-\beta_1/4,\beta_1/4)_h$,
\begin{align*}
&r_{\beta_a}(r_{\beta_a}(y)-x)=r_{\beta_a}(y-x),\ (\forall
 a\in\{1,2\}),\\
&n_{\beta_1}(r_{\beta_1}(y)-x)=n_{\beta_2}(r_{\beta_2}(y)-x).
\end{align*}
By these equalities, \eqref{eq_temperature_translation}, \eqref{eq_basic_temperature_lower_bound} and
 \eqref{eq_basic_temperature_bound}, for any
 $X=(\rho,\bx,\s,x,\theta)\in \hat{I}$,
\begin{align}
&\frah^{m_v-1}\sum_{\bY\in\hat{I}^{m_v-1}}e^{\sum_{j=0}^d(\frac{1}{\pi}\fw(l)\hat{d}_j(X,Y_{1}))^{\fr}}|J_{m_v}^{(v,v)}(X,\bY)|\label{eq_convenient_arrangement}\\
&=\frah^{m_v-1}\sum_{Y=(\eta,\by,\tau,y,\xi)\in\hat{I}}\sum_{\bW\in\hat{I}^{m_v-2}}e^{\sum_{j=0}^d(\frac{1}{\pi}\fw(l)\hat{d}_j(X,Y))^{\fr}}\notag\\
&\quad\cdot|J_{m_v}(\beta_1)(R_{\beta_1}(X,Y,\bW))-J_{m_v}(\beta_2)(R_{\beta_2}(X,Y,\bW))|\notag\\
&=\frah^{m_v-1}\sum_{Y=(\eta,\by,\tau,y,\xi)\in\hat{I}}\sum_{\bW\in\hat{I}^{m_v-2}}e^{\sum_{j=0}^d(\frac{1}{\pi}\fw(l)\hat{d}_j(X,Y))^{\fr}}\notag\\
&\quad\cdot|J_{m_v}(\beta_1)(X-x,R_{\beta_1}((Y,\bW)-x))\notag\\
&\qquad-J_{m_v}(\beta_2)(X-x,R_{\beta_2}((Y,\bW)-x))|\notag\\
&\quad\cdot(1_{(Y,\bW)-x\in \hat{I}^{m_v-1}}+1_{(Y,\bW)-x\notin
 \hat{I}^{m_v-1}})\notag\\
&\le |J_{m_v}(\beta_1)-J_{m_v}(\beta_2)|_l\notag\\
&\quad +\sum_{a=1}^2\frah^{m_v-1}\sum_{Y=(\eta,\by,\tau,y,\xi)\in\hat{I}}\sum_{\bW\in\hat{I}^{m_v-2}}e^{\sum_{j=0}^d(\fw(l)d_j(\beta_a)(X-x,R_{\beta_a}(Y-x)))^{\fr}}\notag\\
&\qquad\cdot|J_{m_v}(\beta_a)(X-x,R_{\beta_a}((Y,\bW)-x))|1_{(Y,\bW)-x\notin
 \hat{I}^{m_v-1}}\notag\\
&\le
 |J_{m_v}(\beta_1)-J_{m_v}(\beta_2)|_l+\frac{2\pi}{\beta_1}(m_v-1)\sum_{a=1}^2\|J_{m_v}(\beta_a)\|_{l,1}\notag\\
&=A(J_{m_v}(\beta_1),J_{m_v}(\beta_2))_l^{(v,v)}.\notag
\end{align}
Substitution of \eqref{eq_separate_estimation_covariance_kernel},
\eqref{eq_convenient_arrangement} into 
\eqref{eq_tree_fixed_difference_kernel_decomposition} gives 
\begin{align}
&(\text{the first term in the right-hand side of
 \eqref{eq_tree_fixed_difference_decomposed}})\label{eq_tree_fixed_difference_decomposed_one_result}\\
&\le 2^{2\sum_{j=1}^n(m_j-k_j)}D_{et}^{\frac{1}{2}\sum_{j=1}^n(m_j-k_j)-n+1}
\|\widetilde{C_o}(\beta_1)\|_{l,0}^{n-1}\notag\\
&\quad\cdot\sum_{v=1}^n\prod_{i=1}^nA(J_{m_i}(\beta_1),J_{m_i}(\beta_2))_l^{(v,i)}.\notag
\end{align}

On the other hand, by applying Lemma
 \ref{lem_tree_line_expansion_difference} we have that 
\begin{align*}
&(\text{the second term in the right-hand side of
 \eqref{eq_tree_fixed_difference_decomposed}})\\
&\le \sum_{v=1}^n(1_{v=1}D(\beta_1,\beta_2)+1_{v\ge 2})
\frac{1}{n!}\sum_{T\in\T_n}
1_{n_j(T)\le m_j-k_j (\forall j\in
 \{1,2,\cdots,n\})}\\
&\quad\cdot 2^{n-1}D_{et}^{\frac{1}{2}\sum_{j=1}^n(m_j-k_j)-n+1}\notag\\
&\quad\cdot \left(\frac{1}{h}\right)^{m_1-1}\sum_{\bX_1\in
 \hat{I}^{m_1-k_1}}\sum_{(Y_{1,2},Y_{1,3},\cdots,Y_{1,k_1})\in
 \hat{I}^{k_1-1}}|J_{m_1}(\beta_2)(R_{\beta_2}(\bX_1,\bY_1))|\\
&\quad\cdot \sum_{\bW_1\subset
 \bX_1\atop \bW_1\in \hat{I}^{n_1(T)}}\sum_{\s_1\in
 \S_{n_1(T)}}\notag\\
&\quad\cdot\prod_{\{1,s\}\in
 L_1^1(T)}\Bigg(\left(\frac{1}{h}\right)^{m_s}\sum_{\bX_s\in
 \hat{I}^{m_s-k_s}}\sum_{\bY_s\in \hat{I}^{k_s}}|J_{m_s}(\beta_2)(R_{\beta_2}(\bX_s,\bY_s))|\notag\\
&\qquad\qquad\qquad\cdot \sum_{Z_s\subset \bX_s\atop 
Z_s\in \hat{I}}|\widetilde{C_o}^{(v,s)}(W_{1,\s_1\circ
 \zeta_1(\{1,s\})},Z_s)|\Bigg)\notag\\
&\quad \cdot \prod_{u=1\atop order
 }^{d_T(1)-1}\Bigg(\prod_{j\in\{2,3,\cdots,n\}\text{ with }\atop
 \dis_T(1,j)=u,n_j(T)\neq 1}\Bigg(\sum_{\bW_j\subset \bX_j\backslash
 Z_j\atop \bW_j\in
 \hat{I}^{n_j(T)-1}}\sum_{\s_j\in\S_{n_j(T)-1}}\notag\\
&\qquad\cdot \prod_{\{j,s\}\in L_j^1(T)}\Bigg(\left(\frac{1}{h}\right)^{m_s}
\sum_{\bX_s\in \hat{I}^{m_s-k_s}}\sum_{\bY_s\in \hat{I}^{k_s}}|J_{m_s}(\beta_2)(R_{\beta_2}(\bX_s,\bY_s))|\notag
\\
&\quad\qquad\qquad\qquad\cdot\sum_{Z_s\subset \bX_s\atop
Z_s\in \hat{I}}|\widetilde{C_o}^{(v,s)}(W_{j,\s_j\circ\zeta_j(\{j,s\})},Z_s)|\Bigg)
\Bigg)\Bigg)\notag\\
&\quad \cdot
e^{\sum_{j=0}^d(\frac{1}{\pi}\fw(l)\hat{d}_j(Y_{1,1},Y_{p,1}))^{\fr}}.\notag
\end{align*}
Moreover, by following the argument leading to Lemma
 \ref{lem_tree_sum_kernels} \eqref{item_tree_sum_kernels_fixed} we can
 deduce that 
\begin{align}
&(\text{the second term in the right-hand side of
 \eqref{eq_tree_fixed_difference_decomposed}})\label{eq_tree_fixed_difference_covariance_decomposition}\\
&\le \sum_{v=1}^n(1_{v=1}D(\beta_1,\beta_2)+1_{v\ge 2})2^{2\sum_{j=1}^n(m_j-k_j)}
D_{et}^{\frac{1}{2}\sum_{j=1}^n(m_j-k_j)-n+1}
\notag\\
&\quad\cdot
 \prod_{i=2}^n\Bigg(\sup_{X\in\hat{I}}\frac{1}{h}\sum_{Y\in\hat{I}}
e^{\sum_{j=0}^d(\frac{1}{\pi}\fw(l)\hat{d}_j(X,Y))^{\fr}}
|\widetilde{C_o}^{(v,i)}(X,Y)|\Bigg)\notag\\
&\quad\cdot
 \prod_{k=1}^n\Bigg(\sup_{X\in\hat{I}}\frah^{m_k-1}\sum_{\bY\in\hat{I}^{m_k-1}}
e^{\sum_{j=0}^d(\frac{1}{\pi}\fw(l)\hat{d}_j(X,Y_1))^{\fr}}
|J_{m_k}(\beta_2)(R_{\beta_2}(\bX,\bY))|\Bigg).\notag
\end{align}
The inequality \eqref{eq_basic_temperature_bound} guarantees that 
\begin{align*}
&\sup_{X\in\hat{I}}\frah^{m_k-1}\sum_{\bY\in\hat{I}^{m_k-1}}
e^{\sum_{j=0}^d(\frac{1}{\pi}\fw(l)\hat{d}_j(X,Y_1))^{\fr}}
|J_{m_k}(\beta_2)(R_{\beta_2}(\bX,\bY))|\\
&\le \|J_{m_k}(\beta_2)\|_{l,0},\\
&\sup_{X\in\hat{I}}\frac{1}{h}\sum_{Y\in\hat{I}}
e^{\sum_{j=0}^d(\frac{1}{\pi}\fw(l)\hat{d}_j(X,Y))^{\fr}}
|\widetilde{C_o}^{(v,i)}(X,Y)|\\
&\le
 \sum_{a=1}^2\|\widetilde{C_o}(\beta_a)\|_{l,0}=A(\widetilde{C_o}(\beta_1),\widetilde{C_o}(\beta_2))_l^{(v,i)},
\end{align*}
if $v\neq i$. It follows from \eqref{eq_convenient_arrangement} that 
\begin{align*}
\sup_{X\in\hat{I}}\frac{1}{h}\sum_{Y\in\hat{I}}
e^{\sum_{j=0}^d(\frac{1}{\pi}\fw(l)\hat{d}_j(X,Y))^{\fr}}
|\widetilde{C_o}^{(v,v)}(X,Y)|\le A(\widetilde{C_o}(\beta_1),\widetilde{C_o}(\beta_2))_l^{(v,v)}.
\end{align*}
By inserting these inequalities into \eqref{eq_tree_fixed_difference_covariance_decomposition} we obtain
\begin{align}
&(\text{the second term in the right-hand side of
 \eqref{eq_tree_fixed_difference_decomposed}})\label{eq_tree_fixed_difference_decomposed_another_result}\\
&\le 2^{2\sum_{j=1}^n(m_j-k_j)}D_{et}^{\frac{1}{2}\sum_{j=1}^n(m_j-k_j)-n+1}\notag\\
&\quad\cdot\sum_{v=1}^n(1_{v=1}D(\beta_1,\beta_2)+1_{v\ge 2})\prod_{i=2}^nA(\widetilde{C_o}(\beta_1),\widetilde{C_o}(\beta_2))_l^{(v,i)}\prod_{k=1}^n\|J_{m_k}(\beta_2)\|_{l,0}.\notag
\end{align}

By combining \eqref{eq_tree_fixed_difference_decomposed_one_result},
 \eqref{eq_tree_fixed_difference_decomposed_another_result} with
 \eqref{eq_tree_fixed_difference_decomposed} we reach the inequality claimed in \eqref{item_tree_sum_fixed_difference}.

\eqref{item_tree_sum_no_fixed_difference}: By considering the fixed variable $X_{1,1}\in
 I^0$ as $Y_{1,1}$ we can straightforwardly transform the proof of
 \eqref{item_tree_sum_fixed_difference} to derive the
 claimed inequality.
\end{proof}

For $a=1,2$ set
\begin{align*}
T^{(n)}(\beta_a)(\psi):=\frac{1}{n!}\sum_{T\in
 \T_n}Ope(T,C_o(\beta_a))\prod_{j=1}^nJ(\beta_a)(\psi+\psi^j)\Bigg|_{\psi^j=0\atop (\forall j\in
 \{1,2,\cdots,n\})}
\end{align*}
with $J(\beta_a)$ $(\in\bigwedge \cV(\beta_a))$ satisfying that
$J_m(\beta_a)(\psi)=0$ if $m\notin 2\N\cup \{0\}$ and having the
anti-symmetric kernels satisfying 
\eqref{eq_temperature_translation}. By putting the preceding lemmas
together we can prove the following lemma, which is the goal of this subsection.

\begin{lemma}\label{lem_tree_formula_general_bound_difference} 
The following inequalities hold true.
\begin{enumerate}
\item\label{item_tree_formula_bound_0th_difference} For any $n\in
 \N_{\ge 2}$ and $l\in \Z$, 
\begin{align*}
&\left|\frac{h}{N(\beta_1)}T_0^{(n)}(\beta_1)-\frac{h}{N(\beta_2)}T_0^{(n)}(\beta_2)\right|\\
&\le 2n
 D_{et}^{-n+1}\Bigg(\sum_{a=1}^2\|\widetilde{C_o}(\beta_a)\|_{l,0}\Bigg)^{n-2}
\Bigg(\sum_{m=2}^{N(\beta_2)}2^{3m}D_{et}^{\frac{m}{2}}\sum_{a=1}^2\|J_m(\beta_a)\|_{l,0}\Bigg)^{n-1}\\
&\quad
 \cdot\sum_{m=2}^{N(\beta_2)}2^{3m}D_{et}^{\frac{m}{2}}\Bigg(\frac{2\pi}{\beta_1}\sum_{a=1}^2\|\widetilde{C_o}(\beta_a)\|_{l,0}\sum_{c=1}^2\|J_m(\beta_c)\|_{l,1}\\
&\qquad+
\frac{2\pi}{\beta_1}\sum_{a=1}^2\|\widetilde{C_o}(\beta_a)\|_{l,1}\sum_{c=1}^2\|J_m(\beta_c)\|_{l,0}\\
&\qquad+\sum_{a=1}^2\|\widetilde{C_o}(\beta_a)\|_{l,0}|J_m(\beta_1)-J_m(\beta_2)|_{l}\\
&\qquad+|\widetilde{C_o}(\beta_1)-\widetilde{C_o}(\beta_2)|_l\sum_{c=1}^2\|J_m(\beta_c)\|_{l,0}\\
&\qquad+
D(\beta_1,\beta_2)\sum_{a=1}^2\|\widetilde{C_o}(\beta_a)\|_{l,0}\sum_{c=1}^2\|J_m(\beta_c)\|_{l,0}\Bigg).
\end{align*}
\item\label{item_tree_formula_bound_higher_difference}
For any $n\in \N_{\ge 2}$, $l\in\Z$ and $m\in \{2,3,\cdots,N(\beta_2)\}$,
\begin{align*}
&|T_m^{(n)}(\beta_1)-T_m^{(n)}(\beta_2)|_l\\
&\le 2n\cdot 2^{-2m}D_{et}^{-\frac{m}{2}-n+1}
 \Bigg(\sum_{a=1}^2\|\widetilde{C_o}(\beta_a)\|_{l,0}\Bigg)^{n-2}\\
&\quad\cdot \prod_{j=2}^n\Bigg(
\sum_{m_j=2}^{N(\beta_2)}2^{4m_j}D_{et}^{\frac{m_j}{2}}\sum_{a=1}^2\|J_{m_j}(\beta_a)\|_{l,0}\Bigg)\\
&\quad
 \cdot\sum_{m_1=2}^{N(\beta_2)}2^{4m_1}D_{et}^{\frac{m_1}{2}}\Bigg(
\frac{2\pi}{\beta_1}\sum_{a=1}^2\|\widetilde{C_o}(\beta_a)\|_{l,0}\sum_{c=1}^2\|J_{m_1}(\beta_c)\|_{l,1}\\
&\qquad+
\frac{2\pi}{\beta_1}\sum_{a=1}^2\|\widetilde{C_o}(\beta_a)\|_{l,1}\sum_{c=1}^2\|J_{m_1}(\beta_c)\|_{l,0}\\
&\qquad+\sum_{a=1}^2\|\widetilde{C_o}(\beta_a)\|_{l,0}|J_{m_1}(\beta_1)-J_{m_1}(\beta_2)|_{l}\\
&\qquad+
|\widetilde{C_o}(\beta_1)-\widetilde{C_o}(\beta_2)|_l\sum_{c=1}^2\|J_{m_1}(\beta_c)\|_{l,0}\\
&\qquad+
D(\beta_1,\beta_2)\sum_{a=1}^2\|\widetilde{C_o}(\beta_a)\|_{l,0}\sum_{c=1}^2\|J_{m_1}(\beta_c)\|_{l,0}\Bigg)1_{\sum_{j=1}^nm_j-2n+2\ge
 m}.
\end{align*}
\end{enumerate}
\end{lemma}
\begin{proof}
\eqref{item_tree_formula_bound_0th_difference}: It follows from
 \eqref{eq_temperature_translation_covariance_general} that
\begin{align}
&Ope(T,C_o(\beta_a))\prod_{j=1\atop order}^n\psi_{\bX_j}^j\Bigg|_{\psi^j=0\atop (\forall j\in
 \{1,2,\cdots,n\})}\label{eq_application_temperature_translation_covariance_tree}\\&= Ope(T,C_o(\beta_a))\prod_{j=1\atop order}^n((-1)^{N_{\beta_a}(\bX_j+x)}\psi_{R_{\beta_a}(\bX_j+x)}^j)\Bigg|_{\psi^j=0\atop (\forall j\in
 \{1,2,\cdots,n\})},\notag\\
&(\forall a\in \{1,2\},\bX_j\in I(\beta_a)^{m_j}\ (j=1,2,\cdots,n),x\in(1/h)\Z).\notag
\end{align}
By using \eqref{eq_temperature_translation} and
 \eqref{eq_application_temperature_translation_covariance_tree} we can
 transform $T_0^{(n)}(\beta_a)$ as follows.
\begin{align*}
&T_0^{(n)}(\beta_a)\\
&=\prod_{i=1}^n\Bigg(\sum_{m_i=2}^{N(\beta_2)}\Bigg)
\frac{1}{n!}\sum_{T\in
 \T_n}Ope(T,C_o(\beta_a))\\
&\quad\cdot\prod_{j=1}^n
\Bigg(\left(\frac{1}{h}\right)^{m_j}\sum_{\bX_j\in
 I(\beta_a)^{m_j}}J_{m_j}(\beta_a)(\bX_j)\Bigg) \prod_{k=1\atop
 order}^n\psi_{\bX_k}^k\Bigg|_{\psi^j=0\atop (\forall j\in
 \{1,2,\cdots,n\})}\\
&=\prod_{i=1}^n\Bigg(\sum_{m_i=2}^{N(\beta_2)}\Bigg)
\frac{1}{n!}\sum_{T\in
 \T_n}Ope(T,C_o(\beta_a))\\
&\quad\cdot \frah^{m_1}\sum_{X_{1,1}\in I^0}\sum_{x\in
 [0,\beta_a)_h}\sum_{(X_{1,2},X_{1,3},\cdots,X_{1,m_1})\in
 I(\beta_a)^{m_1-1}}\\
&\quad\cdot(-1)^{N_{\beta_a}((X_{1,2},X_{1,3},\cdots,X_{1,m_1})-x)}\\
&\quad\cdot J_{m_1}(\beta_a)(X_{1,1},R_{\beta_a}((X_{1,2},X_{1,3},\cdots,X_{1,m_1})-x))\\
&\quad\cdot \prod_{j=2}^n
\Bigg(\left(\frac{1}{h}\right)^{m_j}\sum_{\bX_j\in
 I(\beta_a)^{m_j}}(-1)^{N_{\beta_a}(\bX_j-x)}
J_{m_j}(\beta_a)(R_{\beta_a}(\bX_j-x))\Bigg) \\
&\quad\cdot
 (-1)^{N_{\beta_1}((X_{1,2},X_{1,3},\cdots,X_{1,m_1})-x)}\psi_{(X_{1,1},R_{\beta_a}((X_{1,2},X_{1,3},\cdots,X_{1,m_1})-x))}^1\\
&\quad\cdot\prod_{k=2\atop
 order}^n((-1)^{N_{\beta_a}(\bX_k-x)}\psi_{R_{\beta_a}(\bX_k-x)}^k)\Bigg|_{\psi^j=0\atop (\forall j\in
 \{1,2,\cdots,n\})}\\
&=\beta_a\prod_{i=1}^n\Bigg(\sum_{m_i=2}^{N(\beta_2)}\Bigg)
\frac{1}{n!}\sum_{T\in
 \T_n}Ope(T,C_o(\beta_a))\\
&\quad\cdot \frah^{m_1-1}\sum_{X_{1,1}\in I^0}\sum_{(X_{1,2},X_{1,3},\cdots,X_{1,m_1})\in I(\beta_a)^{m_1-1}}J_{m_1}(\beta_a)(\bX_1)\\
&\quad\cdot\prod_{j=2}^n
\Bigg(\left(\frac{1}{h}\right)^{m_j}\sum_{\bX_j\in
 I(\beta_a)^{m_j}}J_{m_j}(\beta_a)(\bX_j)\Bigg)\prod_{k=1\atop
 order}^n\psi_{\bX_k}^k\Bigg|_{\psi^j=0\atop (\forall j\in
 \{1,2,\cdots,n\})}.
\end{align*}
Then, we decompose $T_0^{(n)}(\beta_a)$ as follows.
 $$T_0^{(n)}(\beta_a)=S_0^{(n)}(\beta_a)+U_0^{(n)}(\beta_a),$$ 
where
\begin{align*}
&S_0^{(n)}(\beta_a)\\
&:=\beta_a\prod_{i=1}^n\Bigg(\sum_{m_i=2}^{N(\beta_2)}\Bigg)
\frac{1}{n!}\sum_{T\in
 \T_n}Ope(T,C_o(\beta_a))\\
&\quad\cdot \frah^{m_1-1}\sum_{X_{1,1}\in I^0}\sum_{(X_{1,2},X_{1,3},\cdots,X_{1,m_1})\in I(\beta_a)^{m_1-1}}J_{m_1}(\beta_a)(\bX_1)\\
&\quad\cdot\prod_{j=2}^n
\Bigg(\left(\frac{1}{h}\right)^{m_j}\sum_{\bX_j\in
 I(\beta_a)^{m_j}}J_{m_j}(\beta_a)(\bX_j)\Bigg)\prod_{k=1\atop
 order}^n\psi_{\bX_k}^k\Bigg|_{\psi^j=0\atop (\forall j\in
 \{1,2,\cdots,n\})}\\
&\quad\cdot 1_{\exists (\rho,\bx,\s,x,\theta)\in I(\beta_a)\text{ s.t. }
(\rho,\bx,\s,x,\theta)\subset (\bX_1,\bX_2,\cdots,\bX_n)\text{ and }x\in
 [\frac{\beta_1}{4},\beta_a-\frac{\beta_1}{4})_h},\\
&U_0^{(n)}(\beta_a)\\
&:=\beta_a\prod_{i=1}^n\Bigg(\sum_{m_i=2}^{N(\beta_2)}\Bigg)
\frac{1}{n!}\sum_{T\in
 \T_n}Ope(T,C_o(\beta_a))\\
&\quad\cdot \frah^{m_1-1}\sum_{X_{1,1}\in I^0}\sum_{(X_{1,2},X_{1,3},\cdots,X_{1,m_1})\in \hat{I}^{m_1-1}}J_{m_1}(\beta_a)(R_{\beta_a}(\bX_1))\\
&\quad\cdot\prod_{j=2}^n
\Bigg(\left(\frac{1}{h}\right)^{m_j}\sum_{\bX_j\in
 \hat{I}^{m_j}}J_{m_j}(\beta_a)(R_{\beta_a}(\bX_j))\Bigg)\prod_{k=1\atop
 order}^n\psi_{R_{\beta_a}(\bX_k)}^k\Bigg|_{\psi^j=0\atop (\forall j\in
 \{1,2,\cdots,n\})}.
\end{align*}

By applying Lemma \ref{lem_tree_line_kernels_large}
 \eqref{item_tree_kernels_no_fixed_large} we have that
\begin{align}
&|S_0^{(n)}(\beta_a)|\label{eq_tree_large_part_0th}\\
&\le \frac{N(\beta_a)}{h}\frac{2\pi}{\beta_1}D_{et}^{-n+1}\prod_{j=1}^n\Bigg(\sum_{q_j=0}^1\sum_{m_j=2}^{N(\beta_2)}2^{3m_j}D_{et}^{\frac{m_j}{2}}\|J_{m_j}(\beta_a)\|_{l,q_j}\Bigg)\notag\\
&\quad\cdot
\prod_{k=2}^n\left(\sum_{r_k=0}^1\|\widetilde{C_o}(\beta_a)\|_{l,r_k}\right)1_{\sum_{j=1}^nq_j+\sum_{k=2}^nr_k=1}\notag\\
&=\frac{N(\beta_a)}{h}\frac{2\pi}{\beta_1}D_{et}^{-n+1}
\Bigg(\sum_{m=2}^{N(\beta_2)}2^{3m}D_{et}^{\frac{m}{2}}\|J_{m}(\beta_a)\|_{l,0}\Bigg)^{n-1}
 \|\widetilde{C_o}(\beta_a)\|_{l,0}^{n-2}\notag\\
&\quad\cdot\sum_{m=2}^{N(\beta_2)}2^{3m}D_{et}^{\frac{m}{2}}
 ((n-1)\|\widetilde{C_o}(\beta_a)\|_{l,1}\|J_{m}(\beta_a)\|_{l,0}\notag\\
&\qquad+
n\|\widetilde{C_o}(\beta_a)\|_{l,0}\|J_{m}(\beta_a)\|_{l,1}).\notag
\end{align}

On the other hand, Lemma \ref{lem_tree_sum_kernels_difference}
 \eqref{item_tree_sum_no_fixed_difference} ensures that
\begin{align}
&\left|\frac{1}{\beta_1}U_0^{(n)}(\beta_1)-\frac{1}{\beta_2}U_0^{(n)}(\beta_2)\right|\label{eq_tree_small_part_0th}\\
&\le\sharp I^0 \prod_{i=1}^n\Bigg(\sum_{m_i=2}^{N(\beta_2)}\Bigg)
 2^{2\sum_{j=1}^nm_j}D_{et}^{\frac{1}{2}\sum_{j=1}^nm_j-n+1}\notag\\
&\quad\cdot\Bigg(\|\widetilde{C_o}(\beta_1)\|_{l,0}^{n-1}\sum_{v=1}^n\prod_{j=1}^nA(J_{m_j}(\beta_1),J_{m_j}(\beta_2))_l^{(v,j)}\notag\\
&\qquad
 +\sum_{v=2}^n\prod_{j=2}^nA(\widetilde{C_o}(\beta_1),\widetilde{C_o}(\beta_2))_l^{(v,j)}\prod_{k=1}^n\|J_{m_k}(\beta_2)\|_{l,0}\notag\\
&\qquad+D(\beta_1,\beta_2)\sum_{a=1}^2\|\widetilde{C_o}(\beta_a)\|_{l,0}^{n-1}\prod_{j=1}^n\|J_{m_j}(\beta_2)\|_{l,0}\Bigg)\notag\\
&\le \sharp I^0 D_{et}^{-n+1}\Bigg(\sum_{m=2}^{N(\beta_2)}2^{2m} D_{et}^{\frac{m}{2}}\sum_{a=1}^2\|J_{m}(\beta_a)\|_{l,0}\Bigg)^{n-1}
\Bigg(\sum_{a=1}^2 \|\widetilde{C_o}(\beta_a)\|_{l,0}\Bigg)^{n-2}\notag\\
&\quad\cdot\sum_{m=2}^{N(\beta_2)}2^{2m}D_{et}^{\frac{m}{2}}
 \Bigg(\notag\\
&\qquad n\sum_{a=1}^2\|\widetilde{C_o}(\beta_a)\|_{l,0}\Bigg(|J_m(\beta_1)-J_m(\beta_2)|_l+\frac{2\pi}{\beta_1}(m-1)\sum_{c=1}^2\|J_{m}(\beta_c)\|_{l,1}\Bigg)\notag\\
&\qquad+(n-1)\Bigg(|\widetilde{C_o}(\beta_1)-\widetilde{C_o}(\beta_2)|_l+\frac{2\pi}{\beta_1}\sum_{a=1}^2\|\widetilde{C_o}(\beta_a)\|_{l,1}\Bigg)\sum_{c=1}^2\|J_{m}(\beta_c)\|_{l,0}\notag\\
&\qquad +D(\beta_1,\beta_2)\sum_{a=1}^2\|\widetilde{C_o}(\beta_a)\|_{l,0}
\sum_{c=1}^2\|J_{m}(\beta_c)\|_{l,0}\Bigg).\notag
\end{align}

Substitution of \eqref{eq_tree_large_part_0th},
 \eqref{eq_tree_small_part_0th} into the inequality 
\begin{align*}
&\left|\frac{h}{N(\beta_1)}T_0^{(n)}(\beta_1)-\frac{h}{N(\beta_2)}T_0^{(n)}(\beta_2)\right|\\
&\le
 \sum_{a=1}^2\frac{h}{N(\beta_a)}|S_0^{(n)}(\beta_a)|+\left|\frac{h}{N(\beta_1)}U_0^{(n)}(\beta_1)-\frac{h}{N(\beta_2)}U_0^{(n)}(\beta_2)\right|
\end{align*}
gives the inequality claimed in \eqref{item_tree_formula_bound_0th_difference}.

\eqref{item_tree_formula_bound_higher_difference}: The anti-symmetric
 kernel $T_m^{(n)}(\beta_a)(\cdot)$ characterized in
 \eqref{eq_tree_formula_explicit_kernel} can be decomposed as
 follows. For any $\bY\in I(\beta_a)^m$,
$$
T_m^{(n)}(\beta_a)(\bY)=S_m^{(n)}(\beta_a)(\bY)+U_m^{(n)}(\beta_a)(\bY),
$$
where
\begin{align*}
&S_m^{(n)}(\beta_a)(\bY)\\
&:=\prod_{i=1}^n\Bigg(\sum_{m_i=2}^{N(\beta_2)}\sum_{k_i=0}^{m_i-1}\left(
\begin{array}{c} m_i\\ k_i\end{array}\right)\sum_{\bY_i\in
 I(\beta_a)^{k_i}}\Bigg)1_{\sum_{j=1}^nm_j-2n+2\ge m}1_{\sum_{j=1}^nk_j= m}\\
&\quad\cdot\frac{1}{m!}\sum_{\s\in \S_m}\sgn(\s)1_{\bY_{\s}=(\bY_1,\bY_2,\cdots,\bY_n)}
\frac{\eps_{\pm}}{n!}\sum_{T\in \T_n}Ope(T,C_o(\beta_a))\\
&\quad\cdot\prod_{j=1}^n\Bigg(\left(\frac{1}{h}\right)^{m_j-k_j}\sum_{\bX_j\in
 I(\beta_a)^{m_j-k_j}}J_{m_j}(\beta_a)(\bX_j,\bY_j)\Bigg)\prod_{k=1\atop order}^n\psi_{\bX_k}^k
\Bigg|_{\psi^j=0\atop (\forall j\in \{1,2,\cdots,n\})}\\
&\quad\cdot 1_{\exists (\rho,\bx,\s,x,\theta)\in I(\beta_a)\text{ s.t. }
(\rho,\bx,\s,x,\theta)\subset
 (\bX_1,\bX_2,\cdots,\bX_n)\text{ and }x\in
 [\frac{\beta_1}{4},\beta_a-\frac{\beta_1}{4})_h},\\
&U_m^{(n)}(\beta_a)(\bY)\\
&:=\prod_{i=1}^n\Bigg(\sum_{m_i=2}^{N(\beta_2)}\sum_{k_i=0}^{m_i-1}\left(
\begin{array}{c} m_i\\ k_i\end{array}\right)\sum_{\bY_i\in
 I(\beta_a)^{k_i}}\Bigg)1_{\sum_{j=1}^nm_j-2n+2\ge m}1_{\sum_{j=1}^nk_j= m}\\
&\quad\cdot\frac{1}{m!}\sum_{\s\in \S_m}\sgn(\s)1_{\bY_{\s}=(\bY_1,\bY_2,\cdots,\bY_n)}
\frac{\eps_{\pm}}{n!}\sum_{T\in \T_n}Ope(T,C_o(\beta_a))\\
&\quad\cdot\prod_{j=1}^n\Bigg(\left(\frac{1}{h}\right)^{m_j-k_j}\sum_{\bX_j\in
 \hat{I}^{m_j-k_j}}J_{m_j}(\beta_a)(R_{\beta_a}(\bX_j),\bY_j)\Bigg)\\
&\quad\cdot\prod_{k=1\atop order}^n\psi_{R_{\beta_a}(\bX_k)}^k
\Bigg|_{\psi^j=0\atop (\forall j\in \{1,2,\cdots,n\})}.
\end{align*}
 
Application of Lemma \ref{lem_tree_line_kernels_large}
 \eqref{item_tree_kernels_fixed_large} yields that
\begin{align}
&|S_m^{(n)}(\beta_1)-S_m^{(n)}(\beta_2)|_l\label{eq_tree_large_part_higher}\\
&\le \prod_{i=1}^n\Bigg(\sum_{m_i=2}^{N(\beta_2)}\sum_{k_i=0}^{m_i-1}\left(
\begin{array}{c} m_i\\ k_i\end{array}\right)\Bigg)1_{\sum_{j=1}^nm_j-2n+2\ge m}1_{\sum_{j=1}^nk_j=
 m}\notag\\
&\quad \cdot \frac{2\pi}{\beta_1} 2^{3\sum_{j=1}^n(m_j-k_j)}D_{et}^{\frac{1}{2}\sum_{j=1}^n(m_j-k_j)-n+1}\notag\\
&\quad\cdot \sum_{a=1}^2\prod_{j=1}^n\Bigg(\sum_{q_j=0}^1\|J_{m_j}(\beta_a)\|_{l,q_j}\Bigg)
\prod_{k=2}^n\Bigg(\sum_{r_k=0}^1\|\widetilde{C_o}(\beta_a)\|_{l,r_k}\Bigg)
 1_{\sum_{j=1}^nq_j+\sum_{k=2}^nr_k=1}\notag\\
&\le \frac{2\pi}{\beta_1} 2^{-3m}D_{et}^{-\frac{m}{2}-n+1}\notag\\
&\quad\cdot\prod_{j=1}^n\Bigg(\sum_{m_j=2}^{N(\beta_2)}
2^{4m_j}D_{et}^{\frac{m_j}{2}}\sum_{q_j=0}^1\sum_{a=1}^2\|J_{m_j}(\beta_a)\|_{l,q_j}\Bigg)
\prod_{k=2}^n\Bigg(\sum_{r_k=0}^1\sum_{a=1}^2\|\widetilde{C_o}(\beta_a)\|_{l,r_k}\Bigg)\notag\\
&\quad\cdot
 1_{\sum_{j=1}^nq_j+\sum_{k=2}^nr_k=1}1_{\sum_{j=1}^nm_j-2n+2\ge
 m}\notag\\
&\le n\frac{2\pi}{\beta_1} 2^{-3m}D_{et}^{-\frac{m}{2}-n+1}\Bigg(\sum_{a=1}^2\|\widetilde{C_o}(\beta_a)\|_{l,0}\Bigg)^{n-2}\notag\\
&\quad\cdot\prod_{j=2}^n\Bigg(\sum_{m_j=2}^{N(\beta_2)}2^{4m_j}D_{et}^{\frac{m_j}{2}}\sum_{a=1}^2\|J_{m_j}(\beta_a)\|_{l,0}\Bigg)\notag\\
&\quad\cdot\sum_{m_1=2}^{N(\beta_2)}2^{4m_1}D_{et}^{\frac{m_1}{2}}\Bigg(\sum_{a=1}^2\|\widetilde{C_o}(\beta_a)\|_{l,1}\sum_{c=1}^2\|J_{m_1}(\beta_c)\|_{l,0}\notag\\
&\qquad+
\sum_{a=1}^2\|\widetilde{C_o}(\beta_a)\|_{l,0}\sum_{c=1}^2\|J_{m_1}(\beta_c)\|_{l,1}\Bigg)1_{\sum_{j=1}^nm_j-2n+2\ge
 m}\notag.
\end{align}

On the other hand, Lemma \ref{lem_tree_sum_kernels_difference}
 \eqref{item_tree_sum_fixed_difference} implies that
\begin{align}
&|U_m^{(n)}(\beta_1)-U_m^{(n)}(\beta_2)|_l\label{eq_tree_small_part_higher}\\
&\le \prod_{i=1}^n\Bigg(\sum_{m_i=2}^{N(\beta_2)}\sum_{k_i=0}^{m_i-1}\left(
\begin{array}{c} m_i\\ k_i\end{array}\right)\Bigg)1_{\sum_{j=1}^nm_j-2n+2\ge m}1_{\sum_{j=1}^nk_j=
 m}\notag\\
&\quad\cdot
 2^{2\sum_{j=1}^n(m_j-k_j)}D_{et}^{\frac{1}{2}\sum_{j=1}^n(m_j-k_j)-n+1}\notag\\
&\quad\cdot\Bigg(\|\widetilde{C_o}(\beta_1)\|_{l,0}^{n-1}\sum_{v=1}^n\prod_{j=1}^nA(J_{m_j}(\beta_1),J_{m_j}(\beta_2))_l^{(v,j)}\notag\\
&\qquad
 +\sum_{v=2}^n\prod_{j=2}^nA(\widetilde{C_o}(\beta_1),\widetilde{C_o}(\beta_2))_l^{(v,j)}\prod_{k=1}^n\|J_{m_k}(\beta_2)\|_{l,0}\notag\\
&\qquad+D(\beta_1,\beta_2)\sum_{a=1}^2\|\widetilde{C_o}(\beta_a)\|_{l,0}^{n-1}\prod_{j=1}^n\|J_{m_j}(\beta_2)\|_{l,0}\Bigg)\notag\\
&\le
 n2^{-2m}D_{et}^{-\frac{1}{2}m-n+1}\Bigg(\sum_{a=1}^2\|\widetilde{C_o}(\beta_a)\|_{l,0}\Bigg)^{n-2}\notag\\
&\quad\cdot \prod_{j=2}^n\Bigg(
\sum_{m_j=2}^{N(\beta_2)}2^{3m_j} D_{et}^{\frac{m_j}{2}}\sum_{a=1}^2\|J_{m_j}(\beta_a)\|_{l,0}\Bigg)\notag\\
&\quad\cdot\sum_{m_1=2}^{N(\beta_2)}2^{4m_1} D_{et}^{\frac{m_1}{2}}
\Bigg(\sum_{a=1}^2\|\widetilde{C_o}(\beta_a)\|_{l,0}|J_{m_1}(\beta_1)-J_{m_1}(\beta_2)|_{l}\notag\\
&\qquad + \frac{2\pi}{\beta_1}\sum_{a=1}^2\|\widetilde{C_o}(\beta_a)\|_{l,0}\sum_{c=1}^2\|J_{m_1}(\beta_c)\|_{l,1}\notag\\
&\qquad+
 |\widetilde{C_o}(\beta_1)-\widetilde{C_o}(\beta_2)|_l\sum_{a=1}^2\|J_{m_1}(\beta_a)\|_{l,0}\notag\\
&\qquad +
 \frac{2\pi}{\beta_1}\sum_{a=1}^2\|\widetilde{C_o}(\beta_a)\|_{l,1}\sum_{c=1}^2\|J_{m_1}(\beta_c)\|_{l,0}\notag\\
&\qquad + D(\beta_1,\beta_2)\sum_{a=1}^2\|\widetilde{C_o}(\beta_a)\|_{l,0}\sum_{c=1}^2\|J_{m_1}(\beta_c)\|_{l,0}\Bigg)1_{\sum_{j=1}^nm_j-2n+2\ge
 m}\notag.
\end{align}

Finally, by combining the inequalities
 \eqref{eq_tree_large_part_higher}, \eqref{eq_tree_small_part_higher}
 with the inequality
\begin{align*}
|T_m^{(n)}(\beta_1)-T_m^{(n)}(\beta_2)|_l\le
|S_m^{(n)}(\beta_1)-S_m^{(n)}(\beta_2)|_l+
|U_m^{(n)}(\beta_1)-U_m^{(n)}(\beta_2)|_l,
\end{align*}
we obtain the inequality claimed in
 \eqref{item_tree_formula_bound_higher_difference}.
\end{proof}

\section{Generalized multi-scale integrations}
\label{sec_multi_scale_analysis}

In this section we present multi-scale integrations, assuming 
that a family of covariances is given and each covariance
belonging to the family has certain scale-dependent
upper bounds. We inductively define a family of Grassmann
polynomials by means of the free integration and the tree expansion with
the covariance at one scale. Then, we establish
scale-dependent estimates on the Grassmann polynomials by applying the
general lemmas prepared in Section \ref{sec_general_estimation} and
Section \ref{sec_general_estimation_temperature}. The analysis of this
section can be seen as a generalization of the multi-scale integration
over the Matsubara frequency and that around the zero set of the
dispersion relation in the momentum space. The results obtained in this section will underlie
more concrete, model-dependent analysis in Section \ref{sec_UV} and
Section \ref{sec_IR_model}. 

From this section we use the symbol `$c$' to represent a real positive constant independent
of any parameter. When we construct inequalities, we will frequently
replace the generic constant $c$ by a larger constant with the same symbol in the
following lines without acknowledging the replacement. However, it must 
be clear from the context that such replacement does not change
what the arguments conclude in the last line.

\subsection{The generalized ultra-violet integration}
\label{subsec_UV_general}
Let $N_+\in\N$ be a fixed number. Assume that a family of
covariances $\{C_{o,l}\}_{l=1}^{N_+}$ is given and it satisfies the following
properties with constants $M$, $c_0$, $c_0'\in \R_{\ge 1}$, a weight
$\fw(0)\in \R_{>0}$ and an exponent $\fr\in(0,1]$.
\begin{align}
&C_{o,l}(\rho\bx\ua x, \eta\by\ua y)=C_{o,l}(\rho\bx\da x, \eta\by\da
 y),\label{eq_spin_symmetry_covariance_general}\\
&C_{o,l}(\rho\bx\ua x, \eta\by\da y)=C_{o,l}(\rho\bx\da x, \eta\by\ua
 y)=0,\notag\\
&(\forall (\rho,\bx,x),(\eta,\by,y)\in\cB\times\G\times[0,\beta)_h,l\in\{1,2,\cdots,N_{+}\}),\notag
\end{align}

\begin{align}
&C_{o,l}(\rho\bx\s x, \eta\bx\tau x)=C_{o,l}(\rho\b0\s 0, \eta\b0\tau 0),\label{eq_translation_invariance_covariance_general}\\
&(\forall (\rho,\bx,\s, x)\in I_0,\eta\in\cB,\tau\in \spin,l\in\{1,2,\cdots,N_{+}\}),\notag
\end{align}
 
\begin{align}
&|\det(\<\bp_i,\bq_j\>_{\C^r}C_{o,l}(X_i,Y_j))_{1\le i,j\le n}|\le
 c_0^n,\label{eq_determinant_bound_UV_general}\\
&(\forall r,n\in\N,\bp_i,\bq_i\in\C^r\text{ with }
\|\bp_i\|_{\C^r},\|\bq_i\|_{\C^r}\le 1,\notag\\
&\quad X_i,Y_i\in I_0\
 (i=1,2,\cdots,n),l\in\{1,2,\cdots,N_{+}\}),\notag
\end{align}

\begin{align}
\|\widetilde{C_{o,l}}\|_{0,j}\le c_0M^{-l},\ (\forall j\in
 \{0,1\},l\in\{1,2,\cdots,N_{+}\}),\label{eq_decay_bound_UV_general}
\end{align}
where $\widetilde{C_{o,l}}:I^2\to\C$ is the anti-symmetric extension of
$C_{o,l}$ defined as in \eqref{eq_anti_symmetric_extension_covariance},

\begin{align}
\sum_{l=1}^{N_+}\max_{\rho\in \cB}|C_{o,l}(\rho\b0\ua 0, \rho\b0\ua
 0)|\le c_0'.\label{eq_tadpole_bound_UV_general}
\end{align}
These are the conditions typically satisfied by an actual covariance
with the Matsubara UV cut-off. In fact the parameters $\fw(0)$, $\fr$ do
not play any explicit role here. We need these parameters only to
introduce the norm $\|\cdot\|_{0,0}$ and the semi-norm $\|\cdot\|_{0,1}$.

Using the covariances $\{C_{o,l}\}_{l=1}^{N_+}$, we inductively define a
family of Grassmann polynomials as follows. With parameters $U_{\rho}\in
\C$ $(\rho\in \cB)$, $\delta \in \{1,-1\}$, define
$F^{N_+}(\psi)$, $T^{N_+,(n)}(\psi)\ (n\in\N_{\ge 2})$, $T^{N_+}(\psi)$,
$J^{N_+}(\psi)\in \bigwedge\cV$ by
\begin{align}
&F^{N_+}(\psi):=\frac{\delta}{2h}\sum_{(\rho,\bx,\s,x)\in
 I_0}U_{\rho}\opsi_{\rho\bx\s x}\psi_{\rho\bx\s
 x}\label{eq_initial_polynomial_UV_general}\\
&\qquad\qquad\quad
 -\frac{1}{h}\sum_{(\rho,\bx,x)\in\cB\times\G\times[0,\beta)_h}U_{\rho}
\opsi_{\rho\bx\ua x}\opsi_{\rho\bx\da x}\psi_{\rho\bx\da
 x}\psi_{\rho\bx\ua x},\notag\\
&T^{N_+,(n)}(\psi):=0,\ (\forall n\in\N_{\ge 2}),\
 T^{N_+}(\psi):=0,\notag\\
&J^{N_+}(\psi):=F^{N_+}(\psi)+T^{N_+}(\psi).\notag
\end{align}
Assume that $l\in \{0,1,\cdots,N_+-1\}$ and $J^{l+1}(\psi)\in
\bigwedge\cV$ is given. Define $F^{l}(\psi)$, $T^{l,(n)}(\psi)\
(n\in\N_{\ge 2})$, $T^{l}(\psi)$, $J^{l}(\psi)\in \bigwedge\cV$ by
\begin{align}
&F^{l}(\psi):=\int
 J^{l+1}(\psi+\psi^1)d\mu_{C_{o,l+1}}(\psi^1),\label{eq_inductive_polynomial_UV_general}\\
&T^{l,(n)}(\psi):=\frac{1}{n!}\sum_{T\in
 \T_n}Ope(T,C_{o,l+1})\prod_{j=1\atop order}^nJ^{l+1}(\psi^j+\psi)\Bigg|_{\psi^j=0\atop
 (\forall j\in \{1,2,\cdots,n\})},\notag\\
&\ (\forall n\in\N_{\ge 2}),\notag\\
&T^l(\psi):=\sum_{n=2}^{\infty}T^{l,(n)}(\psi),\notag\\
&J^l(\psi):=F^l(\psi)+T^l(\psi),\notag
\end{align}
on the assumption that $\sum_{n=2}^{\infty}T^{l,(n)}(\psi)$ converges.

Though one can directly see from definition, let us prove the following
lemma by applying Lemma \ref{lem_free_tree_invariance_general}.
\begin{lemma}\label{lem_even_term_survive_UV}
Assume that $J^l(\psi)$ $(l=0,1,\cdots,N_+)$ are well-defined. Then, if
 $m\notin 2\N\cup \{0\}$, 
\begin{align}
T_m^{l,(n)}(\psi)=F_m^l(\psi)=0,\ (\forall l\in \{0,1,\cdots,N_+\},
 n\in\N_{\ge 2}).\label{eq_odd_term_vanish_induction}
\end{align}
\end{lemma}
\begin{proof}
Apparently the equalities \eqref{eq_odd_term_vanish_induction} hold for
 $l=N_+$. Assume that $J_m^l(\psi)=0$ if $m\notin 2\N\cup\{0\}$ for some
 $l\in \{1,2,\cdots,N_+\}$. 

Let $S:I\to I$, $Q:I\to \R$ be defined by $S(X):=X$, $Q(X):=\pi$,
 $(\forall X\in I)$. Using the notations introduced in Subsection
 \ref{subsec_invariance_general}, we see that
\begin{align*}
J^{l}(\cR\psi)=J^{l}(\psi),\ \widetilde{C_{o,l}}(\bX)=e^{iQ_2(S_2(\bX))}\widetilde{C_{o,l}}(S_2(\bX)),\
 (\forall \bX\in I^2).
\end{align*}
Thus, we can apply Lemma
 \ref{lem_free_tree_invariance_general} \eqref{item_invariance_simple}
 to deduce that
\begin{align*}
&F^{l-1}(\cR\psi)=F^{l-1}(\psi),\\
&T^{l-1,(n)}(\cR\psi)=T^{l-1,(n)}(\psi),\  (\forall n\in\N_{\ge 2}).
\end{align*}
This implies \eqref{eq_odd_term_vanish_induction} for $l-1$. By
 induction the claim holds true.
\end{proof}

The following proposition is a generalization of the multi-scale integration
over the Matsubara frequency.
\begin{proposition}\label{prop_UV_integration_general}
There exists a constant $c\in\R_{>0}$ independent of any parameter such that if
 the parameters $M$, $\alpha\in \R_{\ge 1}$, $U_{\rho}\in\C$ $(\rho\in\cB)$
 satisfy
\begin{align}
M\ge c,\ \alpha^2\ge cM,\ \sup_{\rho\in\cB}|U_{\rho}|\le
 \frac{1}{c(c_0+c_0')^2\alpha^4},\label{eq_parameters_assumption_UV_general}
\end{align}
the following inequalities hold. For any $l\in\{0,1,\cdots,N_+\}$,
 $r\in\{0,1\}$,
\begin{align}
&\frac{h}{N}\Bigg(|F_0^l|+\sum_{n=2}^{\infty}|T_0^{l,(n)}|\Bigg)\le
 \alpha^{-4},\label{eq_UV_general_result_0th}\\ 
&c_0\alpha^2\Bigg(\|F_2^l\|_{0,r}+\sum_{n=2}^{\infty}\|T_2^{l,(n)}\|_{0,r}\Bigg)\le
 1,\label{eq_UV_general_result_2nd}\\
&M^{-2l}\sum_{m=2}^Nc_0^{\frac{m}{2}}M^{\frac{l}{2}m}\alpha^m \Bigg(\|F_m^l\|_{0,r}+\sum_{n=2}^{\infty}\|T_m^{l,(n)}\|_{0,r}\Bigg)\le
 1.\label{eq_UV_general_result_higher}
\end{align}
\end{proposition}
\begin{proof}
Let the symbol $U_{max}$ denote $\sup_{\rho\in\cB}|U_{\rho}|$ during the proof.

\eqref{eq_UV_general_result_2nd},\eqref{eq_UV_general_result_higher}:
First let us prove the inequalities \eqref{eq_UV_general_result_2nd},
 \eqref{eq_UV_general_result_higher} by induction with $l$. This part is close to the proof of \cite[\mbox{Proposition
 4.1}]{K3}. However, we present the full argument for self-containedness of the paper.
 Note that 
$$F_m^{N_+}(\bX)=-V_m^{\delta}(\bX),\ (\forall \bX\in
 I^m,m\in\{2,4\}),$$
where $V_m^{\delta}(\cdot):I^m\to\C$ $(m=2,4)$ are the anti-symmetric functions  characterized in \eqref{eq_anti_symmetric_kernel_initial_V}. From this
 we see that
$$\|F_m^{N_+}\|_{0,r}\le U_{max},\ (\forall m\in\{2,4\},
 r\in \{0,1\}).$$
Therefore, if $U_{max}\le (2(c_0+c_0')^2\alpha^4)^{-1}$,
\begin{align*}
&c_0\alpha^2\|F_2^{N_+}\|_{0,r}\le 1,\\
&M^{-2N_+}\sum_{m\in\{2,4\}}c_0^{\frac{m}{2}}M^{\frac{N_+}{2}m}\alpha^m\|F_m^{N_+}\|_{0,r}\le
 c_0\alpha^2U_{max}+c_0^2\alpha^4U_{max}\le 1.
\end{align*}
Thus, the inequalities \eqref{eq_UV_general_result_2nd},
 \eqref{eq_UV_general_result_higher} for $l=N_+$ hold.

Let us fix $l\in \{0,1,\cdots,N_+-1\}$ and assume that
 \eqref{eq_UV_general_result_2nd}, \eqref{eq_UV_general_result_higher}
 hold for all $l'\in \{l+1,l+2,\cdots,N_+\}$. Fix $r\in
 \{0,1\}$. By combining \eqref{eq_determinant_bound_UV_general},
 \eqref{eq_decay_bound_UV_general} with Lemma
 \ref{lem_tree_formula_general_bound}
 \eqref{item_tree_formula_bound_higher} we have that for any $m\in \{2,3,\cdots,N\}$,
\begin{align}
\|T_m^{l,(n)}\|_{0,r}&\le 2^{-2m} c_0^{-\frac{m}{2}-n+1}
\prod_{i=1}^n\Bigg(\sum_{q_i=0}^1\Bigg)\prod_{j=2}^n\Bigg(\sum_{r_j=0}^1\Bigg)
1_{\sum_{i=1}^nq_i+\sum_{j=2}^nr_j=r}\label{eq_application_tree_formula_bound_higher}\\
&\quad\cdot(c_0M^{-(l+1)})^{n-1}\prod_{k=1}^n\Bigg(\sum_{m_k=2}^N2^{3m_k}c_0^{\frac{m_k}{2}}\|J_{m_k}^{l+1}\|_{0,q_k}\Bigg)1_{\sum_{j=1}^nm_j-2n+2\ge
 m}\notag\\
&=2^{-2m} c_0^{-\frac{m}{2}}M^{-(l+1)(n-1)}
\prod_{i=1}^n\Bigg(\sum_{q_i=0}^1\Bigg)\prod_{j=2}^n\Bigg(\sum_{r_j=0}^1\Bigg)
1_{\sum_{i=1}^nq_i+\sum_{j=2}^nr_j=r}\notag\\
&\quad\cdot\prod_{k=1}^n\Bigg(\sum_{m_k=2}^N2^{3m_k} c_0^{\frac{m_k}{2}}\|J_{m_k}^{l+1}\|_{0,q_k}\Bigg)1_{\sum_{j=1}^nm_j-2n+2\ge
 m}.\notag
\end{align}
By the assumption of induction, 
\begin{align}
\sum_{m=2}^N2^{3m} c_0^{\frac{m}{2}}\|J_{m}^{l+1}\|_{0,r}
&=2^6c_0\|J_2^{l+1}\|_{0,r}+\sum_{m=4}^N2^{3m} c_0^{\frac{m}{2}}\|J_m^{l+1}\|_{0,r}\label{eq_input_upper_bound_general}\\
&\le c\alpha^{-2}+c\alpha^{-4}\le c\alpha^{-2}.\notag
\end{align}
Substitution of \eqref{eq_input_upper_bound_general} into
 \eqref{eq_application_tree_formula_bound_higher} yields that
$$
\|T_m^{l,(n)}\|_{0,r}\le
 c^n c_0^{-\frac{m}{2}}M^{-(l+1)(n-1)}\alpha^{-2n},\ (\forall m\in \{2,4\}).
$$
Therefore, if $\alpha \ge c$,
\begin{align}
\sum_{n=2}^{\infty}\|T_m^{l,(n)}\|_{0,r}\le
 c c_0^{-\frac{m}{2}}M^{-l-1}\alpha^{-4},\ (\forall m\in
 \{2,4\}).\label{eq_tree_2nd_4th_kernel_bound}
\end{align}

It follows from \eqref{eq_application_tree_formula_bound_higher}
 and \eqref{eq_UV_general_result_higher} for $l+1$ that if $M\ge c$, 
\begin{align*}
&M^{-2l}\sum_{m=2}^N c_0^{\frac{m}{2}}M^{\frac{l}{2}m}\alpha^m\|T_m^{l,(n)}\|_{0,r}\\
&\le c M^{-2l}\cdot M^{-(l+1)(n-1)}\prod_{i=1}^n\Bigg(\sum_{q_i=0}^1\Bigg)\prod_{j=2}^n\Bigg(\sum_{r_j=0}^1\Bigg)
1_{\sum_{i=1}^nq_i+\sum_{j=2}^nr_j=r}\\
&\quad\cdot\prod_{k=1}^n\Bigg(\sum_{m_k=2}^N2^{3m_k} c_0^{\frac{m_k}{2}}\|J_{m_k}^{l+1}\|_{0,q_k}\Bigg)\notag\\
&\quad\cdot M^{\frac{l}{2}(\sum_{j=1}^nm_j-2n+2)}
\alpha^{\sum_{j=1}^nm_j-2n+2}2^{-2(\sum_{j=1}^nm_j-2n+2)}\\
&\le c M\alpha^2
\prod_{i=1}^n\Bigg(\sum_{q_i=0}^1\Bigg)\prod_{j=2}^n\Bigg(\sum_{r_j=0}^1\Bigg)
1_{\sum_{i=1}^nq_i+\sum_{j=2}^nr_j=r}\\
&\quad\cdot\prod_{k=1}^n\Bigg(2^4M^{-2l-1}\alpha^{-2}\sum_{m_k=2}^N2^{m_k}c_0^{\frac{m_k}{2}}M^{\frac{l}{2}m_k}\alpha^{m_k}\|J_{m_k}^{l+1}\|_{0,q_k}\Bigg)\\
&\le c M\alpha^2
\prod_{i=1}^n\Bigg(\sum_{q_i=0}^1\Bigg)\prod_{j=2}^n\Bigg(\sum_{r_j=0}^1\Bigg)
1_{\sum_{i=1}^nq_i+\sum_{j=2}^nr_j=r}(c\alpha^{-2})^n\\
&\le c M\alpha^2(c\alpha^{-2})^n,
\end{align*}
where we especially used the inequality
\begin{align*}
\sum_{m=2}^N2^m c_0^{\frac{m}{2}}M^{\frac{l}{2}m}\alpha^m\|J_m^{l+1}\|_{0,q}
\le 2^2M^{2l+1},\ (\forall q\in\{0,1\}).
\end{align*}
Therefore, on the assumption $\alpha\ge c$,
\begin{align}
M^{-2l}\sum_{m=2}^N c_0^{\frac{m}{2}}M^{\frac{l}{2}m}\alpha^m\sum_{n=2}^{\infty}\|T_m^{l,(n)}\|_{0,r}\le
 c M \alpha^{-2}.\label{eq_application_tree_formula_bound_sum}
\end{align}

One implication of Lemma \ref{lem_general_free_bound} is that for any
 $m\in \{6,7,\cdots,N\}$,
$$
\|F_m^l\|_{0,r}\le\sum_{n=m}^N2^n c_0^{\frac{n-m}{2}}\|J_n^{l+1}\|_{0,r}.
$$
Thus, by the assumption $M$, $\alpha\ge c$,
\begin{align}
M^{-2l}\sum_{m=6}^N c_0^{\frac{m}{2}}M^{\frac{l}{2}m}\alpha^m\|F_m^{l}\|_{0,r}&\le
 M^{-2l}\sum_{n=6}^N\sum_{m=6}^n2^n c_0^{\frac{n}{2}}M^{\frac{l}{2}m}\alpha^m
\|J_n^{l+1}\|_{0,r}\label{eq_application_free_bound_sum}\\
&\le
c
 M^{-2l}\sum_{n=6}^N2^n c_0^{\frac{n}{2}}M^{\frac{l}{2}n}\alpha^n\|J_n^{l+1}\|_{0,r}\notag\\
&\le
c
 M^{-2l-3}\sum_{n=6}^N c_0^{\frac{n}{2}}M^{\frac{l+1}{2}n}\alpha^n\|J_n^{l+1}\|_{0,r}\notag\\
&\le cM^{-1}.\notag
\end{align}

It remains to bound $\|F_m^l\|_{0,r}$ $(m\in\{2,4\})$. Set
$$
\hat{F}_4^l(\psi):=F_4^l(\psi)-F_4^{N_+}(\psi).
$$
Note that
\begin{align*}
\hat{F}_4^l(\psi)=\hat{F}_4^{l+1}(\psi)+T_4^{l+1}(\psi)+\cP_4\int\sum_{m=6}^NJ_m^{l+1}(\psi+\psi^1)d\mu_{C_{o,l+1}}(\psi^1).
\end{align*}
By using Lemma \ref{lem_general_free_bound},
  \eqref{eq_UV_general_result_higher}, \eqref{eq_tree_2nd_4th_kernel_bound} for $l'\in \{l+1,l+2,\cdots,N_+\}$
 and the assumption $M$, $\alpha\ge c$ we deduce that
\begin{align}
\|\hat{F}_4^l\|_{0,r}&\le\|\hat{F}_4^{l+1}\|_{0,r}+\|T_4^{l+1}\|_{0,r}+\sum_{m=6}^N2^m c_0^{\frac{m-4}{2}}\|J_m^{l+1}\|_{0,r}\label{eq_free_modified_4th_general_bound}\\
&\le
 \|\hat{F}_4^{l+1}\|_{0,r}+\|T_4^{l+1}\|_{0,r}+cM^{-3(l+1)}\alpha^{-6}\sum_{m=6}^N c_0^{\frac{m-4}{2}}M^{\frac{l+1}{2}m}\alpha^m\|J_m^{l+1}\|_{0,r}\notag\\
&\le  \|\hat{F}_4^{l+1}\|_{0,r}+c c_0^{-2}M^{-l-1}\alpha^{-4}\le c c_0^{-2}\sum_{j=l}^{N_+-1}M^{-j-1}\alpha^{-4}\notag\\
&\le c c_0^{-2}M^{-l-1}\alpha^{-4},\notag
\end{align}
which implies that
\begin{align}
c_0^2\alpha^4\|F_4^{l}\|_{0,r}\le c_0^2\alpha^4U_{max}+c M^{-1}.\label{eq_free_4th_general_bound} 
\end{align}

Next let us bound $\|F_2^{l}\|_{0,r}$. By definition,
\begin{align}
&F_2^l(\psi)\label{eq_flow_equation_free_2nd}\\
&=F_2^{l+1}(\psi)+T_2^{l+1}(\psi)+\cP_2\int\hat{F}_4^{l+1}(\psi+\psi^1)d\mu_{C_{o,l+1}}(\psi^1)\notag\\
&\quad+\cP_2\int{F}_4^{N_{+}}(\psi+\psi^1)d\mu_{C_{o,l+1}}(\psi^1)+\cP_2\int
 T_4^{l+1}(\psi+\psi^1)d\mu_{C_{o,l+1}}(\psi^1)\notag\\
&\quad+\cP_2\int\sum_{m=6}^NJ_m^{l+1}(\psi+\psi^1)d\mu_{C_{o,l+1}}(\psi^1)\notag.
\end{align}
Application of Lemma \ref{lem_general_free_bound},
\eqref{eq_spin_symmetry_covariance_general},
 \eqref{eq_translation_invariance_covariance_general},
\eqref{eq_determinant_bound_UV_general},
 \eqref{eq_tadpole_bound_UV_general},
 \eqref{eq_UV_general_result_higher},
 \eqref{eq_tree_2nd_4th_kernel_bound},
\eqref{eq_free_modified_4th_general_bound} for $l'\in
 \{l+1,l+2,\cdots,N_+\}$ and the assumption $M\ge c$ gives that   
\begin{align*}
\|F_2^l\|_{0,r}
&\le
 \|F_2^{l+1}\|_{0,r}+\|T_2^{l+1}\|_{0,r}+2^4c_0\|\hat{F}_4^{l+1}\|_{0,r}\\
&\quad+U_{max}\max_{\rho\in\cB}|C_{o,l+1}(\rho\b0\ua 0,\rho\b0\ua
 0)|\\
&\quad+2^4c_0\|T_4^{l+1}\|_{0,r}+\sum_{m=6}^N2^m c_0^{\frac{m-2}{2}}\|J_m^{l+1}\|_{0,r}\\
&\le
 \|F_2^{l+1}\|_{0,r}+ U_{max}\max_{\rho\in\cB}|C_{o,l+1}(\rho\b0\ua
 0,\rho\b0\ua 0)|+c c_0^{-1}M^{-l-1}\alpha^{-4}\\
&\le
 \|F_2^{N_+}\|_{0,r}+U_{max}\sum_{j=l}^{N_+-1}\max_{\rho\in\cB}|C_{o,j+1}(\rho\b0\ua
 0,\rho\b0\ua 0)|\\
&\quad+c c_0^{-1}\sum_{j=l}^{N_+-1}M^{-j-1}\alpha^{-4}\\
&\le c c_0'U_{max}+c c_0^{-1}M^{-l-1}\alpha^{-4},
\end{align*}
or
\begin{align}
c_0\alpha^2 \|F_2^l\|_{0,r}\le  c c_0 c_0'\alpha^2
 U_{max}+cM^{-1}\alpha^{-2}.
\label{eq_free_2nd_general_bound}
\end{align}
The inequalities \eqref{eq_tree_2nd_4th_kernel_bound}, 
\eqref{eq_application_tree_formula_bound_sum},
\eqref{eq_application_free_bound_sum},
\eqref{eq_free_4th_general_bound}, \eqref{eq_free_2nd_general_bound}
 ensure that
\begin{align}
& c_0\alpha^2
 \Bigg(\|F_2^l\|_{0,r}+\sum_{n=2}^{\infty}\|T_2^{l,(n)}\|_{0,r}\Bigg)\le
 c c_0 c_0'\alpha^2
 U_{max}+cM^{-1}\alpha^{-2}.\label{eq_resulting_bound_2nd_UV}\\
&M^{-2l}\sum_{m=2}^Nc_0^{\frac{m}{2}}M^{\frac{l}{2}m}\alpha^m
 \Bigg(\|F_m^l\|_{0,r}+\sum_{n=2}^{\infty}\|T_m^{l,(n)}\|_{0,r}\Bigg)\label{eq_resulting_bound_sum_UV}\\
&\le\sum_{m\in\{2,4\}}c_0^{\frac{m}{2}}\alpha^m\|F_m^l\|_{0,r}+M^{-2l}\sum_{m=6}^Nc_0^{\frac{m}{2}}M^{\frac{l}{2}m}\alpha^m\|F_m^l\|_{0,r}\notag\\
&\quad +M^{-2l}\sum_{m=2}^N c_0^{\frac{m}{2}}M^{\frac{l}{2}m}\alpha^m
 \sum_{n=2}^{\infty}\|T_m^{l,(n)}\|_{0,r}\notag\\
&\le c {c_0}c_0'\alpha^2
 U_{max}+ c_0^2\alpha^4 U_{max}+c M^{-1}+cM\alpha^{-2}.\notag
\end{align}
On the assumption \eqref{eq_parameters_assumption_UV_general} the
 right-hand sides of \eqref{eq_resulting_bound_2nd_UV} and
 \eqref{eq_resulting_bound_sum_UV} are less than 1. Thus, the induction
 concludes that the inequalities \eqref{eq_UV_general_result_2nd} and
 \eqref{eq_UV_general_result_higher} hold for all $l\in
 \{0,1,\cdots,N_+\}$ and $r\in \{0,1\}$.

\eqref{eq_UV_general_result_0th}: Let us prove the inequality
 \eqref{eq_UV_general_result_0th}, assuming that the inequalities
 \eqref{eq_UV_general_result_2nd}, \eqref{eq_UV_general_result_higher}
 are valid for all $l\in\{0,1,\cdots,N_{+}\}$. It follows from Lemma \ref{lem_tree_formula_general_bound}
 \eqref{item_tree_formula_bound_0th},
 \eqref{eq_determinant_bound_UV_general}, 
\eqref{eq_decay_bound_UV_general}
and
 \eqref{eq_input_upper_bound_general} that
\begin{align*}
\frac{h}{N}|T_0^{l,(n)}|\le
 c_0^{-n+1}\cdot c_0^{n-1}M^{-(l+1)(n-1)}(c\alpha^{-2})^n=M^{l+1}(cM^{-l-1}\alpha^{-2})^n.
\end{align*}
Thus, if $\alpha\ge c$,
\begin{align}
\frac{h}{N}\sum_{n=2}^{\infty}|T_0^{l,(n)}|\le
 cM^{-l-1}\alpha^{-4}.\label{eq_tree_0th_general_bound}
\end{align}

Define $\hat{F}_2^l(\psi)\in\bigwedge \cV$ $(l\in \{0,1,\cdots,N_+\})$
by 
\begin{align*}
&\hat{F}_2^{N_+}(\psi):=0,\\
&\hat{F}_2^{l}(\psi):=F_2^l(\psi)-F_2^{N_+}(\psi)-\sum_{j=l+1}^{N_+}\cP_2\int
 F_4^{N_+}(\psi+\psi^1)d\mu_{C_{o,j}}(\psi^1),\\
&(\forall l\in \{0,1,\cdots,N_+-1\}).
\end{align*}
Note that for any $l\in \{0,1,\cdots,N_+-1\}$,
\begin{align}
\hat{F}_2^l(\psi)&=\hat{F}_2^{l+1}(\psi)+T_2^{l+1}(\psi)+\cP_2\int
 \hat{F}_4^{l+1}(\psi+\psi^1)d\mu_{C_{o,l+1}}(\psi^1)\label{eq_flow_equation_free_2nd_modified}\\
&\quad + \cP_2\int T_4^{l+1}(\psi+\psi^1)d\mu_{C_{o,l+1}}(\psi^1)\notag\\
&\quad + \cP_2\int
 \sum_{m=6}^NJ_m^{l+1}(\psi+\psi^1)d\mu_{C_{o,l+1}}(\psi^1).\notag
\end{align}
By estimating in the same manner as in Lemma
 \ref{lem_general_free_bound} and using \eqref{eq_determinant_bound_UV_general} we can derive from
 \eqref{eq_flow_equation_free_2nd_modified} that for $r\in\{0,1\}$,
\begin{align*}
\|\hat{F}_2^l\|_{0,r}&\le\|\hat{F}_2^{l+1}\|_{0,r}+\|T_2^{l+1}\|_{0,r}+2^4c_0^{\frac{4-2}{2}}
 \|\hat{F}_4^{l+1}\|_{0,r}\\
&\quad +2^4 c_0^{\frac{4-2}{2}}
 \|T_4^{l+1}\|_{0,r}+\sum_{m=6}^N2^m c_0^{\frac{m-2}{2}}\|J_m^{l+1}\|_{0,r}.
\end{align*}
By \eqref{eq_UV_general_result_higher},
 \eqref{eq_tree_2nd_4th_kernel_bound},
 \eqref{eq_free_modified_4th_general_bound} and the assumption $M\ge c$ we have that
\begin{align}
\|\hat{F}_2^l\|_{0,r}&\le \|\hat{F}_2^{l+1}\|_{0,r}+c
 c_0^{-1}M^{-l-1}\alpha^{-4}\le c
 c_0^{-1}\sum_{j=l}^{N_+-1}M^{-j-1}\alpha^{-4}\label{eq_free_modified_2nd_general_bound}\\
&\le c c_0^{-1}M^{-l-1}\alpha^{-4},\ (\forall r\in\{0,1\}).\notag
\end{align}
Remark that for any $l\in \{0,1,\cdots,N_+-1\}$,
\begin{align}
F_0^l&=F_0^{l+1}+T_0^{l+1}+\int\hat{F}_2^{l+1}(\psi)d\mu_{C_{o,l+1}}(\psi)
+\int F_2^{N_+}(\psi)d\mu_{C_{o,l+1}}(\psi)\label{eq_flow_equation_free_0th}\\
&\quad + 1_{l\le N_+-2}\int\Bigg(\sum_{j=l+2}^{N_+}\cP_2\int
 F_4^{N_+}(\psi+\psi^1)d\mu_{C_{o,j}}(\psi^1)\Bigg)d\mu_{C_{o,l+1}}(\psi)\notag\\
&\quad+\int \hat{F}_4^{l+1}(\psi)d\mu_{C_{o,l+1}}(\psi) +\int
 F_4^{N_+}(\psi)d\mu_{C_{o,l+1}}(\psi)\notag\\
&\quad+\sum_{m\in\{2,4\}}\int
 T_m^{l+1}(\psi)d\mu_{C_{o,l+1}}(\psi) +\sum_{m=6}^N\int
 J_m^{l+1}(\psi)d\mu_{C_{o,l+1}}(\psi)\notag\\
&=F_0^{l+1}+T_0^{l+1}+\frac{\delta
 N}{4bh}\sum_{\rho\in\cB}U_{\rho}C_{o,l+1}(\rho\b0\ua 0,\rho\b0\ua
 0)\notag\\
&\quad -1_{l\le
 N_+-2}\frac{N}{2bh}\sum_{\rho\in\cB}U_{\rho}C_{o,l+1}(\rho\b0\ua
 0,\rho\b0\ua 0)\sum_{j=l+2}^{N_+}C_{o,j}(\rho\b0\ua 0,\rho\b0\ua
 0)\notag\\
&\quad -\frac{N}{4b h}\sum_{\rho\in\cB}U_{\rho}C_{o,l+1}(\rho\b0\ua
 0,\rho\b0\ua
 0)^2+\sum_{m\in\{2,4\}}\int\hat{F}_m^{l+1}(\psi)d\mu_{C_{o,l+1}}(\psi)\notag\\
&\quad +\sum_{m\in\{2,4\}}\int T_m^{l+1}(\psi)d\mu_{C_{o,l+1}}(\psi)
+ \sum_{m=6}^N\int J_m^{l+1}(\psi)d\mu_{C_{o,l+1}}(\psi).\notag
\end{align}
The equality
\eqref{eq_flow_equation_free_0th} and 
the inequalities \eqref{eq_determinant_bound_UV_general},
\eqref{eq_tadpole_bound_UV_general},
 \eqref{eq_UV_general_result_higher},
\eqref{eq_tree_2nd_4th_kernel_bound},\\
\eqref{eq_free_modified_4th_general_bound},
  \eqref{eq_tree_0th_general_bound},
 \eqref{eq_free_modified_2nd_general_bound} imply that
\begin{align}
&\frac{h}{N}|F_0^l|\label{eq_free_0th_general_bound}\\
&\le \frac{h}{N}|F_0^{l+1}|+ \frac{h}{N}|T_0^{l+1}| +U_{max}\max_{\rho\in \cB}|C_{o,l+1}(\rho\b0 \ua 0,\rho\b0 \ua
 0)|\notag\\
&\quad +1_{l\le N_+-2}U_{max}\max_{\rho\in \cB}|C_{o,l+1}(\rho\b0 \ua
 0,\rho\b0 \ua 0)|\sum_{j=l+2}^{N_+}\max_{\eta\in \cB}|C_{o,j}(\eta\b0 \ua
 0,\eta\b0 \ua 0)|\notag\\
&\quad + U_{max}\max_{\rho\in \cB}|C_{o,l+1}(\rho\b0 \ua
 0,\rho\b0 \ua 0)|^2 +
 \sum_{m\in\{2,4\}}c_0^{\frac{m}{2}}\|\hat{F}_m^{l+1}\|_{0,0}\notag\\
&\quad+\sum_{m\in\{2,4\}}c_0^{\frac{m}{2}}\|T_m^{l+1}\|_{0,0}
+ \sum_{m=6}^N c_0^{\frac{m}{2}}\|J_m^{l+1}\|_{0,0}\notag\\
&\le
 \frac{h}{N}|F_0^{l+1}|+cM^{-l-1}\alpha^{-4}+c(c_0+c_0')U_{max}\max_{\rho\in \cB}|C_{o,l+1}(\rho\b0\ua 0,\rho\b0\ua 0)|\notag\\
&\le c \sum_{j=l}^{N_+-1}M^{-j-1}\alpha^{-4}+c(c_0+c_0')U_{max}\sum_{j=l}^{N_+-1}\max_{\rho\in \cB}|C_{o,j+1}(\rho\b0 \ua
 0,\rho\b0 \ua 0)|\notag\\
&\le  c M^{-l-1}\alpha^{-4}+c(c_0+{c_0'})c_0'U_{max}.\notag
\end{align}

By \eqref{eq_tree_0th_general_bound} and
 \eqref{eq_free_0th_general_bound},
\begin{align}
\frac{h}{N}\Bigg(|F_0^l|+\sum_{n=2}^{\infty}|T_0^{l,(n)}|\Bigg)\le
 cM^{-1}\alpha^{-4}+
 c(c_0+{c_0'})c_0'U_{max}.\label{eq_0th_general_bound_pre}
\end{align}
On the assumption \eqref{eq_parameters_assumption_UV_general} the
 right-hand side of \eqref{eq_0th_general_bound_pre} is bounded by
 $\alpha^{-4}$ from above. The proof is complete.
\end{proof}

\subsection{The generalized ultra-violet integration at different temperatures}\label{subsec_UV_general_difference}
Here we estimate the differences between Grassmann polynomials created by
the multi-scale integration described in the previous subsection at 2
different temperatures. The analysis in this subsection is based on 
the inequalities developed in Section
\ref{sec_general_estimation_temperature}. We assume the condition
\eqref{eq_beta_h_assumption} and that 2 families of covariances
$\{C_{o,l}(\beta_a)\}_{l=1}^{N_+}$ $(a=1,2)$ are given and they satisfy
\eqref{eq_spin_symmetry_covariance_general}, 
\eqref{eq_translation_invariance_covariance_general},
\eqref{eq_determinant_bound_UV_general},
\eqref{eq_decay_bound_UV_general},
\eqref{eq_tadpole_bound_UV_general} as well as the following.
\begin{align}\label{eq_temperature_translation_covariance_general_again}
&C_{o,l}(\beta_a)(\bX)=(-1)^{N_{\beta_a}(\bX+x)}C_{o,l}(\beta_a)(R_{\beta_a}(\bX+x)),\\
&(\forall \bX\in I_0(\beta_a)^2, x\in
 (1/h)\Z, a\in\{1,2\},l\in\{1,2,\cdots,N_+\}),\notag
\end{align}

 \begin{align}
&|\det(\<\bp_i,\bq_j\>_{\C^r}C_{o,l}(\beta_1)(R_{\beta_1}(X_i,Y_j)))_{1\le
  i,j\le n}\label{eq_determinant_bound_UV_difference_general}\\
&\quad-
\det(\<\bp_i,\bq_j\>_{\C^r}C_{o,l}(\beta_2)(R_{\beta_2}(X_i,Y_j)))_{1\le
  i,j\le n}|\le \beta_1^{-\frac{1}{2}} c_0^n,\notag\\
&(\forall r,n\in\N,\bp_i,\bq_i\in\C^r\text{ with }
\|\bp_i\|_{\C^r},\|\bq_i\|_{\C^r}\le 1,\notag\\
&\quad X_i,Y_i\in \hat{I}_0\
 (i=1,2,\cdots,n),l\in\{1,2,\cdots,N_+\}),\notag
\end{align}

\begin{align}
|\widetilde{C_{o,l}}(\beta_1)-\widetilde{C_{o,l}}(\beta_2)|_0\le
 \beta_1^{-\frac{1}{2}}{c_0}M^{-l},\ (\forall l\in\{1,2,\cdots,N_+\}),\label{eq_decay_bound_UV_difference_general}
\end{align}
where $\widetilde{C_{o,l}}(\beta_a):I(\beta_a)^2\to \C$ is the
anti-symmetric extension of $C_{o,l}(\beta_a)$ $(a=1,2)$ defined as in
\eqref{eq_anti_symmetric_extension_covariance},
\begin{align}
\sum_{l=1}^{N_+}\max_{\rho\in\cB}|C_{o,l}(\beta_1)(\rho\b0\ua
 0,\rho\b0\ua 0)-C_{o,l}(\beta_2)(\rho\b0\ua
 0,\rho\b0\ua 0)|\le
 \beta_1^{-\frac{1}{2}}c_0'.\label{eq_tadpole_bound_UV_difference_general}
\end{align}
Here the parameters $M\in \R_{\ge 1}$, $\fw(0)\in\R_{>0}$, $\fr\in
(0,1]$ and the constants $c_0$, $c_0'\in\R_{\ge 1}$ are the same as
those in \eqref{eq_spin_symmetry_covariance_general}, 
\eqref{eq_translation_invariance_covariance_general},
\eqref{eq_determinant_bound_UV_general},
\eqref{eq_decay_bound_UV_general},
\eqref{eq_tadpole_bound_UV_general}. Remind us that the parameters $\fw(0)$, $\fr$ are
also used in the measurement $|\cdot-\cdot|_0$. 

With the covariances $\{C_{o,l}(\beta_a)\}_{l=1}^{N_+}$ let
$F^l(\beta_a)(\psi)$, $T^{l,(n)}(\beta_a)(\psi)$ $(n\in \N_{\ge 2})$,
$T^{l}(\beta_a)(\psi)$, $J^l(\beta_a)(\psi)\ (\in
\bigwedge\cV(\beta_a))$ $(l=0,1,\cdots,N_+)$ be defined by
\eqref{eq_initial_polynomial_UV_general},
\eqref{eq_inductive_polynomial_UV_general} for $a=1,2$ respectively.

One requirement of the analysis in Section
\ref{sec_general_estimation_temperature} was that the kernels of
Grassmann polynomials must satisfy the invariance
\eqref{eq_temperature_translation}. First let us confirm that this
requirement is fulfilled in this situation.
\begin{lemma}\label{lem_temperature_translation_UV}
Assume that $J^l(\beta_a)(\psi)$ $(l\in\{0,1,\cdots,N_+\},a\in\{1,2\})$ are
 well-defined. Then,
\begin{align}\label{eq_temperature_translation_UV_general}
&F_m^l(\beta_a)(\bX)=(-1)^{N_{\beta_a}(\bX+x)}F_m^l(\beta_a)(R_{\beta_a}(\bX+x)),\\
&T_m^{l,(n)}(\beta_a)(\bX)=(-1)^{N_{\beta_a}(\bX+x)}T_m^{l,(n)}(\beta_a)(R_{\beta_a}(\bX+x)),\notag\\
&(\forall m\in \{1,2,\cdots,N(\beta_2)\},\bX\in I(\beta_a)^m,x\in
 (1/h)\Z,n\in\N_{\ge 2},\notag\\
&\quad l\in \{0,1,\cdots,N_+\}, a\in\{1,2\}).\notag
\end{align} 
\end{lemma}
\begin{proof}
Fix $a\in \{1,2\}$ and $x\in (1/h)\Z$. Let us define
 $S:I(\beta_a)\to I(\beta_a)$, $Q:I(\beta_a)\to \R$ by 
\begin{align*}
S(\bX):=R_{\beta_a}(\bX+x),\ Q(\bX):=\pi N_{\beta_a}(S^{-1}(\bX)+x),\
 (\forall \bX\in I(\beta_a)).
\end{align*}
It follows from \eqref{eq_temperature_translation_covariance_general_again} and
 the definition of $F^{N_+}(\beta_a)(\psi)$ that
\begin{align*}
&\widetilde{C_{o,l}}(\beta_a)(\bX)=e^{iQ_2(S_2(\bX))}
\widetilde{C_{o,l}}(\beta_a)(S_2(\bX)),\\
&(\forall \bX\in I(\beta_a)^2,l\in \{1,2,\cdots,N_+\}),\\
&F^{N_+}(\beta_a)(\psi)=F^{N_+}(\beta_a)(\cR\psi),
\end{align*}
where we used the notations defined in Subsection
 \ref{subsec_invariance_general}. Thus, recursive application of Lemma
 \ref{lem_free_tree_invariance_general} \eqref{item_invariance_simple}
 with respect to $l$ shows that
\begin{align*}
&F^l(\beta_a)(\psi)=F^l(\beta_a)(\cR\psi),\
 T^{l,(n)}(\beta_a)(\psi)=T^{l,(n)}(\beta_a)(\cR\psi),\\
&(\forall l\in \{0,1,\cdots,N_+\},n\in \N_{\ge 2}).
\end{align*}
By comparing the right-hand side of the equality 
\begin{align*}
F_m^l(\beta_a)(\psi)=\left(\frac{1}{h}\right)^m\sum_{\bX\in
 I(\beta_a)^m}F_m^l(\beta_a)(R_{\beta_a}(\bX+x))\psi_{R_{\beta_a}(\bX+x)}
\end{align*}
with that of the equality
\begin{align*}
F_m^l(\beta_a)(\cR\psi)=\left(\frac{1}{h}\right)^m\sum_{\bX\in
 I(\beta_a)^m}(-1)^{N_{\beta_a}(\bX+x)}F_m^l(\beta_a)(\bX)\psi_{R_{\beta_a}(\bX+x)}
\end{align*}
and by the uniqueness of anti-symmetric kernels we conclude that 
\begin{align*}
&F_m^l(\beta_a)(R_{\beta_a}(\bX+x))=(-1)^{N_{\beta_a}(\bX+x)}F_m^l(\beta_a)(\bX),\\
&(\forall m\in \{2,3,\cdots,N(\beta_2)\},\bX\in I(\beta_a)^m,x\in (1/h)\Z).
\end{align*}
The claimed equality concerning the kernels of $T^{l,(n)}(\beta_a)(\psi)$ can
 be derived in the same way. 
\end{proof}

The purpose of this subsection is to prove the following proposition.
\begin{proposition}\label{prop_UV_integration_difference_general}
There exists a constant $c\in\R_{>0}$ independent of any parameter such that if
 the condition \eqref{eq_parameters_assumption_UV_general} holds with $c$,
the following inequalities hold true. For any $l\in\{0,1,\cdots,N_+\}$,
 $r\in\{0,1\}$,
\begin{align}
&\left|\frac{h}{N(\beta_1)}F_0^l(\beta_1)-\frac{h}{N(\beta_2)}F_0^l(\beta_2)\right|\label{eq_UV_general_result_0th_difference}\\
&+\sum_{n=2}^{\infty}\left|\frac{h}{N(\beta_1)}T_0^{l,(n)}(\beta_1)-\frac{h}{N(\beta_2)}T_0^{l,(n)}(\beta_2)\right|\le \beta_1^{-\frac{1}{2}} \alpha^{-4},\notag\\ 
&c_0\alpha^2\Bigg(|F_2^l(\beta_1)-F_2^l(\beta_2)|_{0}+\sum_{n=2}^{\infty}|T_2^{l,(n)}(\beta_1)-T_2^{l,(n)}(\beta_2)|_{0}\Bigg)\le \beta_1^{-\frac{1}{2}},\label{eq_UV_general_result_2nd_difference}\\
&M^{-2l}\sum_{m=2}^{N(\beta_2)}c_0^{\frac{m}{2}}M^{\frac{l}{2}m}\alpha^m
 \Bigg(|F_m^l(\beta_1)-F_m^l(\beta_2)|_{0}\label{eq_UV_general_result_higher_difference}\\
&\qquad\quad+\sum_{n=2}^{\infty}|T_m^{l,(n)}(\beta_1)-T_m^{l,(n)}(\beta_2)|_{0}\Bigg)\le \beta_1^{-\frac{1}{2}}.\notag
\end{align}
\end{proposition}
\begin{proof} Set $U_{max}:=\sup_{\rho\in\cB}|U_{\rho}|$. We assume
 the condition \eqref{eq_parameters_assumption_UV_general} so that the
 inequalities \eqref{eq_UV_general_result_0th},
 \eqref{eq_UV_general_result_2nd}, \eqref{eq_UV_general_result_higher}
 hold for $\beta_1$, $\beta_2$. 

\eqref{eq_UV_general_result_2nd_difference},\eqref{eq_UV_general_result_higher_difference}: First let us prove
 \eqref{eq_UV_general_result_2nd_difference} and
 \eqref{eq_UV_general_result_higher_difference}. The proof is made by
 induction with respect to $l$. We can see from
 \eqref{eq_anti_symmetric_kernel_initial_V} that
 $|F_m^{N_+}(\beta_1)-F_m^{N_+}(\beta_2)|_0=0$ $(\forall m\in
 \{2,4\})$. Thus, the inequalities \eqref{eq_UV_general_result_2nd_difference},
 \eqref{eq_UV_general_result_higher_difference} for $l=N_+$ hold true.

Let us fix $l\in \{0,1,\cdots,N_+-1\}$ and assume that
 \eqref{eq_UV_general_result_2nd_difference},
 \eqref{eq_UV_general_result_higher_difference} hold for all $l'\in
 \{l+1,l+2,\cdots,N_+\}$. By substituting
 \eqref{eq_determinant_bound_UV_general}, 
 \eqref{eq_decay_bound_UV_general},
 \eqref{eq_determinant_bound_UV_difference_general},  
 \eqref{eq_decay_bound_UV_difference_general} into the inequality in
 Lemma \ref{lem_tree_formula_general_bound_difference}
 \eqref{item_tree_formula_bound_higher_difference} we have 
\begin{align}
&|T_m^{l,(n)}(\beta_1)- T_m^{l,(n)}(\beta_2)|_0\label{eq_application_tree_formula_difference_higher}\\
&\le c  n
 2^{-2m}c_0^{-\frac{m}{2}-n+1}(2c_0M^{-l-1})^{n-2}\prod_{j=2}^n\Bigg(\sum_{m_j=2}^{N(\beta_2)}2^{4m_j}c_0^{\frac{m_j}{2}}\sum_{a=1}^2\|J_{m_j}^{l+1}(\beta_a)\|_{0,0}\Bigg)\notag\\
&\quad\cdot
 \sum_{m_1=2}^{N(\beta_2)}2^{4m_1}c_0^{\frac{m_1}{2}}\Bigg(c_0M^{-l-1}|J_{m_1}^{l+1}(\beta_1)-
 J_{m_1}^{l+1}(\beta_2)|_0\notag\\
&\qquad
 +\beta_1^{-1}c_0M^{-l-1}\sum_{a=1}^2\|J_{m_1}^{l+1}(\beta_a)\|_{0,1}
+ \beta_1^{-\frac{1}{2}} c_0 M^{-l-1}\sum_{a=1}^2\|J_{m_1}^{l+1}(\beta_a)\|_{0,0}
\Bigg)\notag\\
&\quad\cdot 1_{\sum_{j=1}^nm_j-2n+2\ge
 m}\notag\\
&\le c^n
 2^{-2m}c_0^{-\frac{m}{2}}M^{-(l+1)(n-1)}\prod_{j=2}^n\Bigg(\sum_{m_j=2}^{N(\beta_2)}2^{4m_j}c_0^{\frac{m_j}{2}}\sum_{a=1}^2\|J_{m_j}^{l+1}(\beta_a)\|_{0,0}\Bigg)\notag\\
&\quad\cdot
 \sum_{m_1=2}^{N(\beta_2)}2^{4m_1}c_0^{\frac{m_1}{2}}\Bigg(|J_{m_1}^{l+1}(\beta_1)- J_{m_1}^{l+1}(\beta_2)|_0
 +\beta_1^{-\frac{1}{2}}\sum_{r=0}^1\sum_{a=1}^2\|J_{m_1}^{l+1}(\beta_a)\|_{0,r}\Bigg)\notag\\
&\quad\cdot1_{\sum_{j=1}^nm_j-2n+2\ge m}.\notag
\end{align}
The hypothesis of induction implies that
\begin{align}
\sum_{m=2}^{N(\beta_2)}2^{4m}c_0^{\frac{m}{2}}|J_{m}^{l+1}(\beta_1)-
 J_{m}^{l+1}(\beta_2)|_0\le c
 \beta_1^{-\frac{1}{2}}\alpha^{-2}.\label{eq_input_upper_bound_general_difference}
\end{align}
Since the inequalities \eqref{eq_UV_general_result_2nd},
 \eqref{eq_UV_general_result_higher} are available, we can also claim
 that
\begin{align}
\sum_{m=2}^{N(\beta_2)}2^{4m}c_0^{\frac{m}{2}}\sum_{a=1}^2\|J_{m}^{l+1}(\beta_a)\|_{0,r}\le
 c\alpha^{-2},\ (\forall r\in \{0,1\}).\label{eq_input_upper_bound_general_another}
\end{align}
Using \eqref{eq_input_upper_bound_general_difference},
 \eqref{eq_input_upper_bound_general_another}, we obtain from
 \eqref{eq_application_tree_formula_difference_higher} that for $m\in
 \{2,4\}$,
\begin{align*}
|T_m^{l,(n)}(\beta_1)- T_m^{l,(n)}(\beta_2)|_0\le c^n
\beta_1^{-\frac{1}{2}} c_0^{-\frac{m}{2}}M^{-(l+1)(n-1)}\alpha^{-2n}.\end{align*}
Moreover, by the assumption $\alpha\ge c$,
\begin{align}
\sum_{n=2}^{\infty}|T_m^{l,(n)}(\beta_1)- T_m^{l,(n)}(\beta_2)|_0\le c
 \beta_1^{-\frac{1}{2}} c_0^{-\frac{m}{2}}M^{-l-1}\alpha^{-4},\
 (\forall m\in\{2,4\}).\label{eq_tree_2nd_4th_kernel_difference}
\end{align}

By using \eqref{eq_UV_general_result_higher}, 
\eqref{eq_UV_general_result_higher_difference} for $l+1$ we can derive
 from \eqref{eq_application_tree_formula_difference_higher} that
\begin{align*}
&M^{-2l}\sum_{m=2}^{N(\beta_2)} c_0^{\frac{m}{2}}M^{\frac{l}{2}m}\alpha^m
|T_m^{l,(n)}(\beta_1)- T_m^{l,(n)}(\beta_2)|_0\\
&\le c^nM^{-2l}\cdot M^{-(l+1)(n-1)}\prod_{j=2}^n\Bigg(\sum_{m_j=2}^{N(\beta_2)}2^{4m_j}c_0^{\frac{m_j}{2}}\sum_{a=1}^2\|J_{m_j}^{l+1}(\beta_a)\|_{0,0}\Bigg)\\
&\quad\cdot
 \sum_{m_1=2}^{N(\beta_2)}2^{4m_1} c_0^{\frac{m_1}{2}}\Bigg(|J_{m_1}^{l+1}(\beta_1)- J_{m_1}^{l+1}(\beta_2)|_0 +\beta_1^{-\frac{1}{2}}\sum_{r=0}^1\sum_{a=1}^2\|J_{m_1}^{l+1}(\beta_a)\|_{0,r}
\Bigg)\\
&\quad\cdot
 M^{\frac{l}{2}(\sum_{j=1}^nm_j-2n+2)}\alpha^{\sum_{j=1}^nm_j-2n+2}
2^{-2(\sum_{j=1}^nm_j-2n+2)}\\
&\le c^nM^{-2l}\prod_{j=2}^n\Bigg(M^{-2l-1}\alpha^{-2}
\sum_{m_j=2}^{N(\beta_2)}2^{2m_j} c_0^{\frac{m_j}{2}}M^{\frac{l}{2}m_j}\alpha^{m_j}
\sum_{a=1}^2\|J_{m_j}^{l+1}(\beta_a)\|_{0,0}\Bigg)\\
&\quad\cdot
 \sum_{m_1=2}^{N(\beta_2)}2^{2m_1} c_0^{\frac{m_1}{2}}M^{\frac{l}{2}m_1}\alpha^{m_1}\\
&\quad\cdot\Bigg(|J_{m_1}^{l+1}(\beta_1)-J_{m_1}^{l+1}(\beta_2)|_0+ \beta_1^{-\frac{1}{2}}\sum_{r=0}^1\sum_{a=1}^2\|J_{m_1}^{l+1}(\beta_a)\|_{0,r}\Bigg)\\
&\le c \beta_1^{-\frac{1}{2}}M(c\alpha^{-2})^{n-1}.
\end{align*}
Thus, on the assumption $\alpha\ge c$,
\begin{align}
M^{-2l}\sum_{m=2}^{N(\beta_2)}c_0^{\frac{m}{2}}M^{\frac{l}{2}m}\alpha^m\sum_{n=2}^{
\infty}|T_m^{l,(n)}(\beta_1)- T_m^{l,(n)}(\beta_2)|_0\le
 c \beta_1^{-\frac{1}{2}}M\alpha^{-2}.\label{eq_application_tree_formula_bound_difference}
\end{align}

On the other hand, Lemma
 \ref{lem_general_free_bound_difference}
 \eqref{item_general_free_bound_difference_higher},
 \eqref{eq_determinant_bound_UV_general} and \eqref{eq_determinant_bound_UV_difference_general} imply that for $m\in
 \{6,7,\cdots,N(\beta_2)\}$,
\begin{align*}
|F_m^l(\beta_1)-F_m^l(\beta_2)|_0
&\le c\sum_{n=m}^{N(\beta_2)}2^{2n}c_0^{\frac{n-m}{2}}\Bigg(
|J_n^{l+1}(\beta_1)-J_n^{l+1}(\beta_2)|_0\\
&\quad+\beta_1^{-\frac{1}{2}}\sum_{a=1}^2\|J_n^{l+1}(\beta_a)\|_{0,0}+\beta_1^{-1}\sum_{a=1}^2\|J_n^{l+1}(\beta_a)\|_{0,1}\Bigg).
\end{align*}
Thus, by \eqref{eq_UV_general_result_higher},
 \eqref{eq_UV_general_result_higher_difference} for $l+1$,
\begin{align}
& M^{-2l}\sum_{m=6}^{N(\beta_2)}c_0^{\frac{m}{2}}M^{\frac{l}{2}m}\alpha^m|F_m^l(\beta_1)-F_m^l(\beta_2)|_0\label{eq_application_free_bound_sum_difference}\\
&\le
 M^{-2l}\sum_{n=6}^{N(\beta_2)}\sum_{m=6}^{n}2^{2n}c_0^{\frac{n}{2}}M^{\frac{l}{2}m}\alpha^m\notag\\
&\quad\cdot
 \Bigg(|J_n^{l+1}(\beta_1)-J_n^{l+1}(\beta_2)|_0+\beta_1^{-\frac{1}{2}}\sum_{r=0}^1\sum_{a=1}^2\|J_n^{l+1}(\beta_a)\|_{0,r}\Bigg)\notag\\
&\le c
 M^{-2l-3}\sum_{n=6}^{N(\beta_2)}c_0^{\frac{n}{2}}M^{\frac{l+1}{2}n}\alpha^n\notag\\
&\quad\cdot\Bigg(|J_n^{l+1}(\beta_1)-J_n^{l+1}(\beta_2)|_0+
\beta_1^{-\frac{1}{2}}\sum_{r=0}^1\sum_{a=1}^2\|J_n^{l+1}(\beta_a)\|_{0,r}\Bigg)\notag\\
&\le c \beta_1^{-\frac{1}{2}}M^{-1}.\notag
\end{align}

Let us find an upper bound on $|F_m^l(\beta_1)-F_m^l(\beta_2)|_0$
 $(m=2,4)$. Lemma \ref{lem_general_free_bound_difference}
 \eqref{item_general_free_bound_difference_higher}, 
\eqref{eq_determinant_bound_UV_general}, \eqref{eq_determinant_bound_UV_difference_general}
and the inequalities
 \eqref{eq_UV_general_result_higher},
 \eqref{eq_UV_general_result_higher_difference}, 
\eqref{eq_tree_2nd_4th_kernel_difference} for $l'\in
 \{l+1,l+2,\cdots,N_+\}$ guarantee that
\begin{align}
&|F_4^l(\beta_1)-F_4^l(\beta_2)|_0\label{eq_free_4th_general_bound_difference}\\
&\le
 |F_4^{l+1}(\beta_1)-F_4^{l+1}(\beta_2)|_0+\sum_{n=2}^{\infty}|T_4^{l+1,(n)}(\beta_1)-T_4^{l+1,(n)}(\beta_2)|_0\notag\\
&\quad +c
 \sum_{n=6}^{N(\beta_2)}2^{2n}c_0^{\frac{n-4}{2}}\Bigg(|J_n^{l+1}(\beta_1)-J_n^{l+1}(\beta_2)|_0\notag\\
&\qquad\qquad+\beta_1^{-\frac{1}{2}}\sum_{a=1}^2\|J_n^{l+1}(\beta_a)\|_{0,0}
+\beta_1^{-1}\sum_{a=1}^2\|J_n^{l+1}(\beta_a)\|_{0,1}\Bigg)\notag\\
&\le  |F_4^{l+1}(\beta_1)-F_4^{l+1}(\beta_2)|_0+
 c \beta_1^{-\frac{1}{2}} c_0^{-2}M^{-l-1}\alpha^{-4}\notag\\
&\le  |F_4^{N_+}(\beta_1)-F_4^{N_+}(\beta_2)|_0 + c \beta_1^{-\frac{1}{2}} c_0^{-2}\sum_{j=l}^{N_+-1}M^{-j-1}\alpha^{-4}\notag\\
&\le c \beta_1^{-\frac{1}{2}} c_0^{-2}M^{-l-1}\alpha^{-4},\notag
\end{align}
which implies that 
\begin{align}
c_0^{2}\alpha^{4}|F_4^l(\beta_1)-F_4^l(\beta_2)|_0\le
 c \beta_1^{-\frac{1}{2}}M^{-1}.\label{eq_free_4th_general_bound_coefficient}
\end{align}

Remark that 
\begin{align*}
\cP_2\int
 J_4^{N_+}(\beta_a)(\psi+\psi^1)d\mu_{C_{o,l+1}(\beta_a)}(\psi^1)=\frah^2\sum_{\bX\in
 I(\beta_a)^2}K_2^{l+1}(\beta_a)(\bX)\psi_{\bX}
\end{align*}
with the anti-symmetric kernel
 $K_2^{l+1}(\beta_a)(\cdot):I(\beta_a)^2\to \C$ defined by 
\begin{align*}
&K_2^{l+1}(\beta_a)((\rho,\bx,\s,x,\theta),(\eta,\by,\tau,y,\xi))\\
&:=-\frac{h}{2}U_{\rho}
C_{o,l+1}(\beta_a)(\rho\b0\ua 0,\rho\b0\ua
 0)1_{(\rho,\bx,\s,x)=(\eta,\by,\tau,y)}
(1_{(\theta,\xi)=(1,-1)}-1_{(\theta,\xi)=(-1,1)}).
\end{align*}
We can see that
\begin{align}
&|K_2^{l+1}(\beta_1)-K_2^{l+1}(\beta_2)|_0\label{eq_tadpole_difference_derivation}\\
&\le
 U_{max}\sup_{\rho\in\cB}|C_{o,l+1}(\beta_1)(\rho\b0\ua 0,\rho\b0\ua
 0)-C_{o,l+1}(\beta_2)(\rho\b0\ua 0,\rho\b0\ua 0)|.\notag
\end{align}
It follows from \eqref{eq_flow_equation_free_2nd} and Lemma
 \ref{lem_general_free_bound_difference}
 \eqref{item_general_free_bound_difference_higher},
\eqref{eq_determinant_bound_UV_general},
\eqref{eq_determinant_bound_UV_difference_general}
that
\begin{align*}
&|F_2^l(\beta_1)-F_2^l(\beta_2)|_0\\
&\le
 |F_2^{l+1}(\beta_1)-F_2^{l+1}(\beta_2)|_0+|T_2^{l+1}(\beta_1)-T_2^{l+1}(\beta_2)|_0\\
&\quad+|K_2^{l+1}(\beta_1)-K_2^{l+1}(\beta_2)|_0\\
&\quad +cc_0^{\frac{4-2}{2}}\Bigg(|\hat{F}_4^{l+1}(\beta_1)-\hat{F}_4^{l+1}(\beta_2)|_0
+\beta_1^{-\frac{1}{2}}\sum_{a=1}^2\|\hat{F}_4^{l+1}(\beta_a)\|_{0,0}\\
&\qquad\qquad+ \beta_1^{-1}\sum_{a=1}^2\|\hat{F}_4^{l+1}(\beta_a)\|_{0,1}+|T_4^{l+1}(\beta_1)-T_4^{l+1}(\beta_2)|_0\\
&\qquad\qquad +
 \beta_1^{-\frac{1}{2}}\sum_{a=1}^2\|T_4^{l+1}(\beta_a)\|_{0,0}+\beta_1^{-1}\sum_{a=1}^2\|T_4^{l+1}(\beta_a)\|_{0,1}\Bigg)\\
&\quad +c \sum_{n=6}^{N(\beta_2)}2^{2n}c_0^{\frac{n-2}{2}}\Bigg(|J_n^{l+1}(\beta_1)-J_n^{l+1}(\beta_2)|_0+
 \beta_1^{-\frac{1}{2}}\sum_{a=1}^2\|J_n^{l+1}(\beta_a)\|_{0,0}\\
&\qquad\qquad+\beta_1^{-1}\sum_{a=1}^2\|J_n^{l+1}(\beta_a)\|_{0,1}\Bigg).
\end{align*}
Substitution of \eqref{eq_UV_general_result_higher},
 \eqref{eq_tree_2nd_4th_kernel_bound},
 \eqref{eq_free_modified_4th_general_bound},
\eqref{eq_tadpole_bound_UV_difference_general}, 
\eqref{eq_UV_general_result_higher_difference},
\eqref{eq_tree_2nd_4th_kernel_difference},
\eqref{eq_free_4th_general_bound_difference},
\eqref{eq_tadpole_difference_derivation}
 and the equality
\begin{align}
|F_4^{l'}(\beta_1)-F_4^{l'}(\beta_2)|_0=|\hat{F}_4^{l'}(\beta_1)-\hat{F}_4^{l'}(\beta_2)|_0\label{eq_UV_difference_4th_free_equivalence}
\end{align}
for $l'\in \{l+1,l+2,\cdots,N_+\}$
yield that
\begin{align*}
&|F_2^{l}(\beta_1)-F_2^{l}(\beta_2)|_0\\
&\le |F_2^{l+1}(\beta_1)-F_2^{l+1}(\beta_2)|_0+c
 \beta_1^{-\frac{1}{2}} c_0^{-1}M^{-l-1}\alpha^{-4}\\
&\quad+U_{max}\sup_{\rho\in\cB}|C_{o,l+1}(\beta_1)(\rho\b0\ua 0,\rho\b0\ua 0)-C_{o,l+1}(\beta_2)(\rho\b0\ua 0,\rho\b0\ua 0)|\\
&\le |F_2^{N_+}(\beta_1)-F_2^{N_+}(\beta_2)|_0+
c
 \beta_1^{-\frac{1}{2}}c_0^{-1}\sum_{j=l}^{N_+-1}M^{-j-1}\alpha^{-4}\\
&\quad +U_{max}\sum_{j=l}^{N_+-1}\sup_{\rho\in\cB}|C_{o,j+1}(\beta_1)(\rho\b0\ua 0,\rho\b0\ua 0)-C_{o,j+1}(\beta_2)(\rho\b0\ua 0,\rho\b0\ua 0)|\\
&\le c  \beta_1^{-\frac{1}{2}}c_0^{-1}M^{-l-1}\alpha^{-4}+U_{max}\beta_1^{-\frac{1}{2}}c_0',
\end{align*}
or 
\begin{align}
c_0\alpha^2|F_2^{l}(\beta_1)-F_2^{l}(\beta_2)|_0\le
 \beta_1^{-\frac{1}{2}}(c_0c_0'\alpha^2U_{max}+cM^{-1}\alpha^{-2}).\label{eq_free_2nd_general_bound_difference}
\end{align}

By \eqref{eq_tree_2nd_4th_kernel_difference},
 \eqref{eq_application_tree_formula_bound_difference},
 \eqref{eq_application_free_bound_sum_difference}, 
\eqref{eq_free_4th_general_bound_coefficient},
\eqref{eq_free_2nd_general_bound_difference} we have that
\begin{align}
&c_0\alpha^2\Bigg(|F_2^{l}(\beta_1)-F_2^{l}(\beta_2)|_0+\sum_{n=2}^{\infty}|T_2^{l,(n)}(\beta_1)-T_2^{l,(n)}(\beta_2)|_0\Bigg)\label{eq_resulting_bound_2nd_UV_difference}\\
&\le \beta_1^{-\frac{1}{2}}(c_0c_0'\alpha^2U_{max}+cM^{-1}\alpha^{-2}).\notag\\
& M^{-2l}\sum_{m=2}^{N(\beta_2)}c_0^{\frac{m}{2}}M^{\frac{l}{2}m}\alpha^m\Bigg(|F_m^l(\beta_1)-F_m^l(\beta_2)|_0+\sum_{n=2}^{\infty}|T_m^{l,(n)}(\beta_1)-T_m^{l,(n)}(\beta_2)|_0
\Bigg)\label{eq_resulting_bound_sum_UV_difference}\\
& \le
 \sum_{m\in\{2,4\}}c_0^{\frac{m}{2}}\alpha^m|F_m^l(\beta_1)-F_m^l(\beta_2)|_0\notag\\
&\quad+
 M^{-2l}\sum_{m=6}^{N(\beta_2)}c_0^{\frac{m}{2}}M^{\frac{l}{2}m}\alpha^m|F_m^l(\beta_1)-F_m^l(\beta_2)|_0\notag\\
&\quad+
 M^{-2l}\sum_{m=2}^{N(\beta_2)}c_0^{\frac{m}{2}}M^{\frac{l}{2}m}\alpha^m\sum_{n=2}^{\infty}|T_m^{l,(n)}(\beta_1)-T_m^{l,(n)}(\beta_2)|_0\notag\\
&\le
 c \beta_1^{-\frac{1}{2}}(c_0c_0'\alpha^2U_{max}+M^{-1}+M\alpha^{-2}).\notag
\end{align}
On the assumption \eqref{eq_parameters_assumption_UV_general} the right-hand
 sides of \eqref{eq_resulting_bound_2nd_UV_difference},
 \eqref{eq_resulting_bound_sum_UV_difference} are less than
 $\beta_1^{-\frac{1}{2}}$. Thus, by induction the inequalities \eqref{eq_UV_general_result_2nd_difference},
 \eqref{eq_UV_general_result_higher_difference} hold for all $l\in
 \{0,1,\cdots,N_+\}$. 

\eqref{eq_UV_general_result_0th_difference}: Let us prove the inequality
\eqref{eq_UV_general_result_0th_difference}, assuming that
 \eqref{eq_UV_general_result_2nd_difference},
 \eqref{eq_UV_general_result_higher_difference} are true for all $l\in
 \{0,1,\cdots,N_+\}$. By substituting
\eqref{eq_determinant_bound_UV_general}, 
\eqref{eq_decay_bound_UV_general}, 
\eqref{eq_input_upper_bound_general},
\eqref{eq_determinant_bound_UV_difference_general},
\eqref{eq_decay_bound_UV_difference_general},
\eqref{eq_input_upper_bound_general_difference} into 
the inequality in Lemma \ref{lem_tree_formula_general_bound_difference} 
 \eqref{item_tree_formula_bound_0th_difference} we obtain 
\begin{align*}
&\left|\frac{h}{N(\beta_1)}T_0^{l,(n)}(\beta_1)-\frac{h}{N(\beta_2)}T_0^{l,(n)}(\beta_2)\right|\\
&\le cn
 c_0^{-n+1}(2c_0 M^{-l-1})^{n-2}
\Bigg(\sum_{m=2}^{N(\beta_2)}2^{3m}c_0^{\frac{m}{2}}\sum_{a=1}^2\|J_m^{l+1}(\beta_a)\|_{0,0}\Bigg)^{n-1}\\
&\quad
 \cdot\sum_{m=2}^{N(\beta_2)}2^{3m}c_0^{\frac{m}{2}}\Bigg(\beta_1^{-1}
c_0M^{-l-1}\sum_{r=0}^1\sum_{a=1}^2\|J_m^{l+1}(\beta_a)\|_{0,r}\\
&\qquad+ \beta_1^{-\frac{1}{2}}c_0 M^{-l-1}\sum_{a=1}^2\|J_m^{l+1}(\beta_a)\|_{0,0}+ c_0 M^{-l-1}|J_m^{l+1}(\beta_1)-J_m^{l+1}(\beta_2)|_{0}\Bigg)\\
&\le \beta_1^{-\frac{1}{2}} M^{l+1}(c M^{-l-1}\alpha^{-2})^n,
\end{align*}
which leads to 
\begin{align}
\sum_{n=2}^{\infty}\left|\frac{h}{N(\beta_1)}T_0^{l,(n)}(\beta_1)-\frac{h}{N(\beta_2)}T_0^{l,(n)}(\beta_2)\right|\le c\beta_1^{-\frac{1}{2}}M^{-l-1}\alpha^{-4}.\label{eq_tree_0th_general_bound_difference}
\end{align}

It follows from 
\eqref{eq_determinant_bound_UV_general},
\eqref{eq_flow_equation_free_2nd_modified},
\eqref{eq_determinant_bound_UV_difference_general}
 and Lemma \ref{lem_general_free_bound_difference}
 \eqref{item_general_free_bound_difference_higher} that
\begin{align*}
&|\hat{F}_2^l(\beta_1)-\hat{F}_2^l(\beta_2)|_0\\
&\le |\hat{F}_2^{l+1}(\beta_1)-\hat{F}_2^{l+1}(\beta_2)|_0 +
 |T_2^{l+1}(\beta_1)-T_2^{l+1}(\beta_2)|_0\\
&\quad + c
 c_0^{\frac{4-2}{2}}\Bigg(|\hat{F}_4^{l+1}(\beta_1)-\hat{F}_4^{l+1}(\beta_2)|_0+ \beta_1^{-\frac{1}{2}}\sum_{a=1}^2\|\hat{F}_4^{l+1}(\beta_a)\|_{0,0}\\
&\qquad + \beta_1^{-1}\sum_{a=1}^2\|\hat{F}_4^{l+1}(\beta_a)\|_{0,1}
+|T_4^{l+1}(\beta_1)-T_4^{l+1}(\beta_2)|_0\\
&\qquad+\beta_1^{-\frac{1}{2}}\sum_{a=1}^2\|T_4^{l+1}(\beta_a)\|_{0,0} + \beta_1^{-1}\sum_{a=1}^2\|T_4^{l+1}(\beta_a)\|_{0,1}\Bigg)\\
&\quad +
 c\sum_{n=6}^{N(\beta_2)}2^{2n}c_0^{\frac{n-2}{2}}\Bigg(|J_n^{l+1}(\beta_1)-J_n^{l+1}(\beta_2)|_0+\beta_{1}^{-\frac{1}{2}}\sum_{a=1}^2\|J_n^{l+1}(\beta_a)\|_{0,0}\\
&\qquad+\beta_{1}^{-1}\sum_{a=1}^2\|J_n^{l+1}(\beta_a)\|_{0,1}\Bigg).
\end{align*}
Using \eqref{eq_UV_general_result_higher},
 \eqref{eq_tree_2nd_4th_kernel_bound},
 \eqref{eq_free_modified_4th_general_bound}, 
\eqref{eq_UV_general_result_higher_difference},
\eqref{eq_tree_2nd_4th_kernel_difference},
\eqref{eq_free_4th_general_bound_difference},
\eqref{eq_UV_difference_4th_free_equivalence}
 for $l'\in
 \{l+1,l+2,\cdots,N_+\}$, we can deduce that
\begin{align}
|\hat{F}_2^l(\beta_1)-\hat{F}_2^l(\beta_2)|_0&\le
 |\hat{F}_2^{l+1}(\beta_1)-\hat{F}_2^{l+1}(\beta_2)|_0 +c \beta_1^{-\frac{1}{2}}c_0^{-1}M^{-l-1}\alpha^{-4}\label{eq_free_modified_2nd_general_bound_difference}\\
&\le |\hat{F}_2^{N_+}(\beta_1)-\hat{F}_2^{N_+}(\beta_2)|_0+
c
 \beta_1^{-\frac{1}{2}} c_0^{-1}\sum_{j=l}^{N_+-1}M^{-j-1}\alpha^{-4}\notag\\
&\le c\beta_1^{-\frac{1}{2}} c_0^{-1}M^{-l-1}\alpha^{-4}.\notag
\end{align}
By combining Lemma \ref{lem_general_free_bound_difference}
 \eqref{item_general_free_bound_difference_0th} with
 \eqref{eq_flow_equation_free_0th} and inserting 
\eqref{eq_determinant_bound_UV_general},
\eqref{eq_determinant_bound_UV_difference_general}
we obtain that
\begin{align*}
&\left|\frac{h}{N(\beta_1)}F_0^{l}(\beta_1)-\frac{h}{N(\beta_2)}F_0^{l}(\beta_2)\right|\\
&\le
 \left|\frac{h}{N(\beta_1)}F_0^{l+1}(\beta_1)-\frac{h}{N(\beta_2)}F_0^{l+1}(\beta_2)\right|\notag\\
&\quad + \left|\frac{h}{N(\beta_1)}T_0^{l+1}(\beta_1)-\frac{h}{N(\beta_2)}T_0^{l+1}(\beta_2)\right|\notag\\
&\quad + U_{max}\max_{\rho\in\cB}|C_{o,l+1}(\beta_1)(\rho\b0\ua 0,
 \rho\b0\ua 0)- C_{o,l+1}(\beta_2)(\rho\b0\ua 0,
 \rho\b0\ua 0)|\notag\\
&\quad+1_{l\le N_+-2}U_{max}\max_{\rho\in \cB}|C_{o,l+1}(\beta_1)(\rho\b0\ua 0,
 \rho\b0\ua 0)- C_{o,l+1}(\beta_2)(\rho\b0\ua 0,
 \rho\b0\ua 0)|\notag\\
&\qquad\cdot\sum_{j=l+2}^{N_+}\max_{\eta\in\cB}|C_{o,j}(\beta_1)(\eta \b0
 \ua 0,\eta \b0 \ua 0)|\notag\\
&\quad + 1_{l\le N_+-2}U_{max}\max_{\rho\in
 \cB}|C_{o,l+1}(\beta_2)(\rho\b0\ua 0, \rho\b0\ua 0)|\notag\\
&\qquad\cdot\sum_{j=l+2}^{N_+}\max_{\eta\in\cB}|C_{o,j}(\beta_1)(\eta \b0
 \ua 0,\eta \b0 \ua 0)- C_{o,j}(\beta_2)(\eta \b0
 \ua 0,\eta \b0 \ua 0)|\notag\\
&\quad + U_{max}\max_{\rho\in \cB}|C_{o,l+1}(\beta_1)(\rho\b0\ua 0,
 \rho\b0\ua 0)^2- C_{o,l+1}(\beta_2)(\rho\b0\ua 0,
 \rho\b0\ua 0)^2|\notag\\
&\quad +c \sum_{m\in \{2,4\}}c_0^{\frac{m}{2}}\Bigg(
|\hat{F}_m^{l+1}(\beta_1)-\hat{F}_m^{l+1}(\beta_2)|_0+\beta_1^{-\frac{1}{2}}\sum_{a=1}^2\|\hat{F}_m^{l+1}(\beta_a)\|_{0,0}\notag\\
&\qquad+\beta_1^{-1}\sum_{a=1}^2\|\hat{F}_m^{l+1}(\beta_a)\|_{0,1}+|T_m^{l+1}(\beta_1)-T_m^{l+1}(\beta_2)|_0\notag\\
&\qquad+\beta_1^{-\frac{1}{2}}\sum_{a=1}^2\|T_m^{l+1}(\beta_a)\|_{0,0}+\beta_1^{-1}\sum_{a=1}^2\|T_m^{l+1}(\beta_a)\|_{0,1}\Bigg)\notag\\
&\quad +c \sum_{m=6}^{N(\beta_2)}2^mc_0^{\frac{m}{2}}\Bigg(
|J_m^{l+1}(\beta_1)-J_m^{l+1}(\beta_2)|_0
+\beta_1^{-\frac{1}{2}}\sum_{a=1}^2\|J_m^{l+1}(\beta_a)\|_{0,0}\notag\\
&\qquad+\beta_1^{-1}\sum_{a=1}^2\|J_m^{l+1}(\beta_a)\|_{0,1}\Bigg).\notag
\end{align*}
Moreover, substitution of \eqref{eq_determinant_bound_UV_general},
\eqref{eq_tadpole_bound_UV_general},
 \eqref{eq_UV_general_result_higher}, 
\eqref{eq_tree_2nd_4th_kernel_bound},
\eqref{eq_free_modified_4th_general_bound},
\eqref{eq_free_modified_2nd_general_bound},
\eqref{eq_tadpole_bound_UV_difference_general},
\eqref{eq_UV_general_result_higher_difference},
\eqref{eq_tree_2nd_4th_kernel_difference},
\eqref{eq_free_4th_general_bound_difference},
\eqref{eq_UV_difference_4th_free_equivalence},
\eqref{eq_tree_0th_general_bound_difference},
\eqref{eq_free_modified_2nd_general_bound_difference}
for $l'\in
 \{l+1,l+2,\cdots,N_+\}$ gives that 
\begin{align}
&\left|\frac{h}{N(\beta_1)}F_0^{l}(\beta_1)-\frac{h}{N(\beta_2)}F_0^{l}(\beta_2)\right|\label{eq_free_0th_general_bound_difference}\\
&\le
 \left|\frac{h}{N(\beta_1)}F_0^{l+1}(\beta_1)-\frac{h}{N(\beta_2)}F_0^{l+1}(\beta_2)\right| + c \beta_1^{-\frac{1}{2}}M^{-l-1}\alpha^{-4}\notag\\
&\quad + c(c_0+c_0')U_{max}\max_{\rho\in
 \cB}|C_{o,l+1}(\beta_1)(\rho\b0\ua 0,\rho\b0\ua
 0)-C_{o,l+1}(\beta_2)(\rho\b0\ua 0,\rho\b0\ua 0)|\notag\\
&\quad + c \beta_1^{-\frac{1}{2}} c_0'U_{max}\max_{\rho\in
 \cB}|C_{o,l+1}(\beta_2)(\rho\b0\ua 0,\rho\b0\ua 0)|\notag\\
&\le c
 \beta_1^{-\frac{1}{2}}\sum_{j=l}^{N_+-1}M^{-j-1}\alpha^{-4}\notag\\
&\quad + c(c_0+c_0')U_{max}\notag\\
&\qquad \cdot\sum_{j=l}^{N_+-1}\max_{\rho\in
 \cB}|C_{o,j+1}(\beta_1)(\rho\b0\ua 0,\rho\b0\ua
 0)-C_{o,j+1}(\beta_2)(\rho\b0\ua 0,\rho\b0\ua 0)|\notag\\
&\quad + c \beta_1^{-\frac{1}{2}} c_0'U_{max}\sum_{j=l}^{N_+-1}\max_{\rho\in
 \cB}|C_{o,j+1}(\beta_2)(\rho\b0\ua 0,\rho\b0\ua 0)|\notag\\
&\le c \beta_1^{-\frac{1}{2}}(M^{-1}\alpha^{-4}+(c_0+c_0')c_0'U_{max}).\notag
\end{align}

By coupling \eqref{eq_tree_0th_general_bound_difference} with
 \eqref{eq_free_0th_general_bound_difference} and using the assumption
 \eqref{eq_parameters_assumption_UV_general} we conclude
 that
\begin{align*}
&\left|\frac{h}{N(\beta_1)}F_0^{l}(\beta_1)-\frac{h}{N(\beta_2)}F_0^{l}(\beta_2)\right|\\
&\quad +\sum_{n=2}^{\infty}\left|\frac{h}{N(\beta_1)}T_0^{l,(n)}(\beta_1)-\frac{h}{N(\beta_2)}T_0^{l,(n)}(\beta_2)\right|\\
&\le c
 \beta_1^{-\frac{1}{2}}(M^{-1}\alpha^{-4}+(c_0+c_0')c_0'U_{max})\le \beta_1^{-\frac{1}{2}}\alpha^{-4}.
\end{align*}
\end{proof}

\subsection{The generalized infrared integration}\label{subsec_IR_general}
In this subsection we estimate Grassmann polynomials produced by a
single-scale integration with a covariance which has different bound
properties from those assumed in the previous
subsection. Our aim here is to summarize a power-counting procedure of
the infrared integration by giving a covariance with bound properties
typical of a real covariance with infrared cut-off. In the model-dependent
infrared integration regime in Section \ref{sec_IR_model}, we need to update the covariance by including
the kernel of the quadratic part of a Grassmann polynomial created by
the preceding integration. This means that in the IR integration, unlike in the UV
integration, we cannot a priori give covariances for all the integration steps. For this reason here we construct estimates only for
one integration step as a preliminary to the practical IR
integration.

Let $l\in \Z_{<0}$. We assume that an exponent $\fr\in (0,1]$ and weights $\fw(l)$, $\fw(l+1)$ satisfying
$0<\fw(l)\le \fw(l+1)$ are given and a covariance $C_{o,l+1}:I_0^2\to \C$
satisfies the following bound properties with constants $M$,
$c_0\in\R_{\ge 1}$,
$a_1,a_2,a_3\in\R_{\ge 0}$.

\begin{align}
&|\det(\<\bp_i,\bq_j\>_{\C^r}C_{o,l+1}(X_i,Y_j))_{1\le i,j\le n}|\le
 (c_0M^{a_1(l+1)})^n,\label{eq_determinant_bound_infrared_general}\\
&(\forall r,n\in\N,\bp_i,\bq_i\in\C^r\text{ with }
\|\bp_i\|_{\C^r},\|\bq_i\|_{\C^r}\le 1,\notag\\
&\ X_i,Y_i\in I_0\
 (i=1,2,\cdots,n)),\notag
\end{align}

\begin{align}
\left\|\widetilde{C_{o,l+1}}\right\|_{l,r}\le c_0M^{-a_2(l+1)-ra_3(l+1)},\ (\forall r\in
 \{0,1\}),\label{eq_decay_bound_infrared_general}
\end{align}
where $\widetilde{C_{o,l+1}}:I^2\to\C$ is the anti-symmetric extension of
$C_{o,l+1}$ defined as in \eqref{eq_anti_symmetric_extension_covariance}.
Recall that the parameters $(\fw(l),\fr)$, $(\fw(l+1),\fr)$ are used in the definition
of $\|\cdot\|_{l,j}$, $\|\cdot\|_{l+1,j}$ respectively.   

We assume that $J^{l+1}(\psi)$ $(\in\bigwedge \cV)$ is given
and it satisfies 
$J_m^{l+1}(\psi)=0$ if $m\notin 2\N$ or $m\in \{0,2\}$. Then, we define
$F^l(\psi)$, $T^{l,(n)}(\psi)$ $(n\in \N_{\ge 2})$, $T^l(\psi)
$ $(\in\bigwedge \cV)$ by \eqref{eq_inductive_polynomial_UV_general} with the
input $J^{l+1}(\psi)$ and the covariance $C_{o,l+1}$ on the assumption
that $\sum_{n=2}^{\infty}T^{l,(n)}(\psi)$ converges. We can prove the
following lemma by the same argument as in the proof of Lemma
\ref{lem_even_term_survive_UV}.
\begin{lemma}\label{lem_even_term_survive_infrared}
Assume that $J^l(\psi)$ is well-defined. Then, if
 $m\notin 2\N\cup \{0\}$,
$$
T_m^{l,(n)}(\psi)=F_m^l(\psi)=0,\ (\forall n\in \N_{\ge 2}).
$$
\end{lemma}

This subsection is devoted to proving the following proposition.

\begin{proposition}\label{prop_infrared_integration_general}
Assume that $a_4\in\R_{\ge 0}$ and 
\begin{align}
M^{-(a_1+a_2+a_4)(l+1)+ra_3(l+1)}\sum_{m=4}^N c_0^{\frac{m}{2}}M^{\frac{a_1(l+1)}{2}m}\alpha^m\|J_m^{l+1}\|_{l+1,r}\le
 1,\ (\forall r\in \{0,1\}).\label{eq_infrared_assumption_general}
\end{align}
Then, there exists a constant $c\in\R_{>1}$ independent of any parameter such that if
 the parameters $M$, $\alpha\in \R_{\ge 1}$
 satisfy
\begin{align}
M^{a_1-a_2-a_4}\ge c,\ \alpha \ge cM^{a_1+a_2+a_4},\label{eq_parameters_assumption_IR_general}
\end{align}
the following inequalities hold.
\begin{align}
&\frac{h}{N}\Bigg(|F_0^l|+\sum_{n=2}^{\infty}|T_0^{l,(n)}|\Bigg)\le M^{(a_1+a_2+a_4)l} \alpha^{-3},\label{eq_infrared_general_result_0th}\\ 
&M^{-(a_1+a_2+a_4)l+ra_3l}\sum_{m=2}^N c_0^{\frac{m}{2}}M^{\frac{a_1l}{2}m}\alpha^m \Bigg(\|F_m^l\|_{l,r}+\sum_{n=2}^{\infty}\|T_m^{l,(n)}\|_{l,r}\Bigg)\le
 1,\label{eq_infrared_general_result}\\
&\ (\forall r\in \{0,1\}).\notag
\end{align}
\end{proposition}
\begin{proof}
It follows from \eqref{eq_infrared_assumption_general} and the
 inequalities $\fw(l)\le \fw(l+1)$, $\alpha\ge c$, $M^{a_1}\ge c$ that 
\begin{align}
&\sum_{m=4}^N2^{3m}c_0^{\frac{m}{2}}M^{\frac{a_1(l+1)}{2}m}\|J_m^{l+1}\|_{l,r}\le
 c
 M^{(a_1+a_2+a_4)(l+1)-ra_3(l+1)}\alpha^{-4},\label{eq_input_upper_bound_infrared_general}\\
&\sum_{m=4}^N2^{2m}c_0^{\frac{m}{2}}M^{\frac{a_1 l}{2}m}\alpha^m\|J_m^{l+1}\|_{l,r}\le
 c
 M^{(a_1+a_2+a_4)(l+1)-ra_3(l+1)-2a_1},\label{eq_input_upper_bound_infrared_general_weight}\\
&\ (\forall r\in \{0,1\}).\notag
\end{align}
We will use these inequalities not only in this proof but also in the
 proof of Proposition
 \ref{prop_infrared_integration_difference_general} in the next subsection.

\eqref{eq_infrared_general_result_0th}: First let us prove the inequality
 \eqref{eq_infrared_general_result_0th}. By Lemma
 \ref{lem_general_free_bound}, \eqref{eq_determinant_bound_infrared_general} and
 \eqref{eq_input_upper_bound_infrared_general}, 
\begin{align}
\frac{h}{N}|F_0^l|\le \sum_{m=4}^N (c_0
 M^{a_1(l+1)})^{\frac{m}{2}}\|J_m^{l+1}\|_{l,0}\le c M^{(a_1+a_2+a_4)(l+1)}
\alpha^{-4}.\label{eq_free_0th_infrared_general}
\end{align}

On the other hand, by substituting \eqref{eq_determinant_bound_infrared_general},
 \eqref{eq_decay_bound_infrared_general},
 \eqref{eq_input_upper_bound_infrared_general} into the inequality
 in Lemma \ref{lem_tree_formula_general_bound}
 \eqref{item_tree_formula_bound_0th} we have that
\begin{align*}
&\frac{h}{N}|T_0^{l,(n)}|\\
&\le (c_0 M^{a_1(l+1)})^{-n+1}(c_0
 M^{-a_2(l+1)})^{n-1}\Bigg(\sum_{m=4}^N2^{2m}(c_0
 M^{a_1(l+1)})^{\frac{m}{2}}\|J_{m}^{l+1}\|_{l,0}\Bigg)^n\\
&\le M^{(a_1+a_2)(l+1)}(c M^{a_4(l+1)}\alpha^{-4})^n.
\end{align*}
Thus, by the assumption $\alpha\ge c$,
\begin{align}
\frac{h}{N}\sum_{n=2}^{\infty}|T_0^{l,(n)}|\le c
 M^{(a_1+a_2+2a_4)(l+1)}\alpha^{-8}.\label{eq_tree_0th_infrared_general}
\end{align}
By coupling \eqref{eq_tree_0th_infrared_general} with
 \eqref{eq_free_0th_infrared_general} and using the assumption $\alpha
 \ge c M^{a_1+a_2+a_4}$ we obtain
 \eqref{eq_infrared_general_result_0th}.

\eqref{eq_infrared_general_result}: Next let us show
 \eqref{eq_infrared_general_result}. It follows from Lemma
 \ref{lem_general_free_bound} and
 \eqref{eq_determinant_bound_infrared_general} that
\begin{align}\label{eq_free_part_IR_sensitive_framework}
\|F_m^l\|_{l,r}
\le &1_{m=2}\sum_{n=4}^N 2^n
 (c_0M^{a_1(l+1)})^{\frac{n-2}{2}}\|J_n^{l+1}\|_{l,r}\\
&+1_{m\ge 4}\Bigg(\|J_m^{l+1}\|_{l,r}+
\sum_{n=m+2}^N 2^n
 (c_0M^{a_1(l+1)})^{\frac{n-m}{2}}\|J_n^{l+1}\|_{l,r}\Bigg).\notag
\end{align}
Moreover, by \eqref{eq_input_upper_bound_infrared_general},
\eqref{eq_input_upper_bound_infrared_general_weight} and the
 assumption $\alpha\ge c M^{a_1/2}$,
\begin{align}
&M^{-(a_1+a_2+a_4)l+ra_3l}\sum_{m=2}^N c_0^{\frac{m}{2}}M^{\frac{a_1
 l}{2}m}\alpha^m\|F_m^{l}\|_{l,r}\label{eq_application_free_bound_sum_infrared}\\
&\le M^{-(a_1+a_2+a_4)l+ra_3l}c_0M^{a_1l}\alpha^2\sum_{n=4}^N2^n(c_0M^{a_1(l+1)})^{\frac{n-2}{2}}\|J_n^{l+1}\|_{l,r}\notag\\
&\quad + M^{-(a_1+a_2+a_4)l+ra_3l}\sum_{n=4}^N\sum_{m=4}^n c_0^{\frac{m}{2}}M^{\frac{a_1
 l}{2}m}\alpha^m2^n(c_0M^{a_1(l+1)})^{\frac{n-m}{2}}\|J_n^{l+1}\|_{l,r}\notag\\
&\le M^{-a_1}\alpha^2
 M^{-(a_1+a_2+a_4)l+ra_3l}\sum_{n=4}^N2^nc_0^{\frac{n}{2}}
M^{\frac{a_1(l+1)}{2}n}\|J_n^{l+1}\|_{l,r}\notag\\
&\quad + c M^{-(a_1+a_2+a_4)l+ra_3l}\sum_{n=4}^N2^n c_0^{\frac{n}{2}}M^{\frac{a_1
 l}{2}n}\alpha^n\|J_n^{l+1}\|_{l,r}\notag\\
&\le c M^{-a_1+a_2+a_4-ra_3}.\notag
\end{align}

On the other hand, insertion of
 \eqref{eq_determinant_bound_infrared_general},
 \eqref{eq_decay_bound_infrared_general} into
Lemma \ref{lem_tree_formula_general_bound}
 \eqref{item_tree_formula_bound_higher} yields that
\begin{align*}
&\|T_m^{l,(n)}\|_{l,r}\\
&\le 2^{-2m}(c_0 M^{a_1(l+1)})^{-\frac{m}{2}-n+1}\prod_{i=1}^n\Bigg(\sum_{q_i=0}^1\Bigg)\prod_{j=2}^n\Bigg(\sum_{r_j=0}^1\Bigg)
1_{\sum_{i=1}^nq_i+\sum_{j=2}^nr_j=r}\\
&\quad\cdot\prod_{k=2}^n(c_0M^{-a_2(l+1)-a_3r_k(l+1)})\prod_{p=1}^n\Bigg(\sum_{m_p=4}^N2^{3m_p}(c_0
 M^{a_1(l+1)})^{\frac{m_p}{2}}\|J_{m_p}^{l+1}\|_{l,q_p}\Bigg)\\
&\quad\cdot 1_{\sum_{j=1}^nm_j-2n+2\ge
 m}\notag\\
&=2^{-2m} c_0^{-\frac{m}{2}}M^{-\frac{a_1(l+1)}{2}m+(a_1+a_2)(l+1)-ra_3(l+1)}\\&\quad\cdot\prod_{i=1}^n\Bigg(\sum_{q_i=0}^1\Bigg)\prod_{j=2}^n\Bigg(\sum_{r_j=0}^1\Bigg)
1_{\sum_{i=1}^nq_i+\sum_{j=2}^nr_j=r}\\
&\quad\cdot\prod_{k=1}^n\Bigg(M^{-(a_1+a_2)(l+1)+a_3q_k(l+1)}\sum_{m_k=4}^N2^{3m_k}c_0^{\frac{m_k}{2}}
 M^{\frac{a_1(l+1)}{2}m_k}\|J_{m_k}^{l+1}\|_{l,q_k}\Bigg)\\
&\quad\cdot 1_{\sum_{j=1}^nm_j-2n+2\ge
 m}.
\end{align*}
By using the
 inequality $\alpha \ge c M^{a_1/2}$ and \eqref{eq_input_upper_bound_infrared_general_weight} we can derive from the inequality
 above that
\begin{align*}
&M^{-(a_1+a_2+a_4)l+ra_3l}\sum_{m=2}^Nc_0^{\frac{m}{2}}M^{\frac{a_1
 l}{2}m}\alpha^m\|T_m^{l,(n)}\|_{l,r}\\
&\le c M^{a_1+a_2- r a_3-a_4l}\prod_{i=1}^n\Bigg(\sum_{q_i=0}^1\Bigg)\prod_{j=2}^n\Bigg(\sum_{r_j=0}^1\Bigg)
1_{\sum_{i=1}^nq_i+\sum_{j=2}^nr_j=r}\\
&\quad\cdot\prod_{k=1}^n\Bigg(M^{-(a_1+a_2)(l+1)+a_3q_k(l+1)}\sum_{m_k=4}^N2^{3m_k}c_0^{\frac{m_k}{2}}
 M^{\frac{a_1(l+1)}{2}m_k}\|J_{m_k}^{l+1}\|_{l,q_k}\Bigg)\\
&\quad \cdot
 2^{-2(\sum_{j=1}^nm_j-2n+2)}M^{-\frac{a_1}{2}(\sum_{j=1}^nm_j-2n+2)}\alpha^{\sum_{j=1}^nm_j-2n+2}\\
&\le c  M^{a_2- r a_3-a_4l}\alpha^2\prod_{i=1}^n\Bigg(\sum_{q_i=0}^1\Bigg)\prod_{j=2}^n\Bigg(\sum_{r_j=0}^1\Bigg)
1_{\sum_{i=1}^nq_i+\sum_{j=2}^nr_j=r}\\
&\quad\cdot\prod_{k=1}^n\Bigg(c M^{a_1-(a_1+a_2)(l+1)+a_3q_k(l+1)}\alpha^{-2}
\sum_{m_k=4}^N2^{m_k}c_0^{\frac{m_k}{2}}
 M^{\frac{a_1 l}{2}m_k}\alpha^{m_k}\|J_{m_k}^{l+1}\|_{l,q_k}\Bigg)\\
&\le  M^{a_2- r a_3-a_4l}\alpha^2(c M^{-a_1+a_4(l+1)}\alpha^{-2})^n.
\end{align*}
Since $\alpha \ge c$,
\begin{align}
&M^{-(a_1+a_2+a_4)l+ra_3l}\sum_{m=2}^N c_0^{\frac{m}{2}}M^{\frac{a_1
 l}{2}m}\alpha^m\sum_{n=2}^{\infty}\|T_m^{l,(n)}\|_{l,r}\label{eq_application_tree_formula_bound_sum_infrared}\\
&\le c M^{-2a_1+a_2+2a_4-ra_3+a_4l}\alpha^{-2}.\notag
\end{align}

The inequalities \eqref{eq_application_free_bound_sum_infrared},
\eqref{eq_application_tree_formula_bound_sum_infrared} imply that
\begin{align*}
&M^{-(a_1+a_2+a_4)l+ra_3l}\sum_{m=2}^Nc_0^{\frac{m}{2}}M^{\frac{a_1
 l}{2}m}\alpha^m\Bigg(\|F_m^l\|_{l,r}+\sum_{n=2}^{\infty}\|T_m^{l,(n)}\|_{l,r}\Bigg)\\
&\le c (M^{-a_1+a_2+a_4}+M^{-2a_1+a_2+a_4}\alpha^{-2}).
\end{align*}
By the assumption \eqref{eq_parameters_assumption_IR_general} the right-hand side of the inequality above is less than 1.
\end{proof}

\begin{remark}\label{rem_renormalizability}
Since Proposition \ref{prop_infrared_integration_general} forms the
 basis of our IR integration process, it is important to know to which
 model the proposition does or does not apply. The covariance
 $C_{o,l+1}$ is a generalization of an effective covariance at the IR
 integration of scale $l+1$, which is different from a computable free
 covariance with IR cut-off of scale $l+1$. Therefore, by analyzing free
 covariances alone we cannot reach a rigorous statement on the
 applicability of the proposition. However, on the hypothesis that the
 IR singularity of an effective covariance is essentially same as that
 of a free covariance, let us try to extract at least some hints from
 calculations of free covariances. The condition
 \eqref{eq_parameters_assumption_IR_general} necessarily implies that
 $a_1-a_2>0$. Let us investigate in which model the inequality
 $a_1-a_2>0$ is unlikely to hold. Most studied many-electron models in
 constructive theories so far are single-band models having a non-empty
 free Fermi surface. So let us focus our attention on such models.

Assume that the free dispersion relation $\cE(\cdot):\R^d\to \R$ is
 Lipschitz continuous and satisfies the periodicity
 \eqref{eq_dispersion_periodic} and that $\{\bk\in \R^d\ |\
 \cE(\bk)=\mu\}\neq \emptyset$ with the chemical potential $\mu$. For
 example, in the Hubbard model with nearest-neighbor hopping, without
 magnetic field, defined on a $d$-dimensional hyper-cubic lattice, the
 free dispersion relation is given by $\cE(\bk)=\sum_{j=1}^d\cos k_j$,
 apart from a multiplication of amplitude. With a non-negative smooth
 function $\chi(\cdot):\R\to\R_{\ge 0}$ supported on the interval
 $[c_1,c_2]$, taking the value 1 on a subinterval of $[c_1,c_2]$, a free
 covariance $C_{o,l+1}:(\G\times\spin\times [0,\beta)_h)^2\to \C$ with
 IR cut-off of scale $l+1$ $(<0)$ typically takes the form that
\begin{align*}
&C_{o,l+1}(\bx\s x,\by\tau y)\\
&=\frac{\delta_{\s,\tau}}{\beta
 L^d}\sum_{(\o,\bk)\in \cM_h\times \G^*}e^{i\<\bx-\by,\bk\>}e^{i(x-y)\o}
\frac{\chi(M^{-2(l+1)}(\o^2+(\cE(\bk)-\mu)^2))}{i\o-\cE(\bk)+\mu}.
\end{align*}
On the assumption that $h,L,\beta$ are sufficiently large we can choose
 $(\o_0,\bk_0)\in\cM_h\times\G^*$ so that
 $\chi(M^{-2(l+1)}(\o_0^2+(\cE(\bk_0)-\mu)^2))=1$ and thus
\begin{align*}
\left\|\widetilde{C_{o,l+1}}\right\|_{l,0}&\ge
 \frac{1}{2h}\sum_{(\bx,x)\in \G\times [0,\beta)_h}|C_{o,l+1}(\b0\ua
 0,\bx\ua x)|\\
&\ge \left|\frac{1}{2h}\sum_{(\bx,x)\in \G\times
 [0,\beta)_h} e^{i\<\bx,\bk_0\>}e^{ix\o_0}C_{o,l+1}(\b0\ua
 0,\bx\ua x)\right|\\
&=\frac{\chi(M^{-2(l+1)}(\o_0^2+(\cE(\bk_0)-\mu)^2))}{2|i\o_0-\cE(\bk_0)+\mu|}
\ge \text{const } M^{-l-1}.
\end{align*}
This means that if $C_{o,l+1}$ satisfies
 \eqref{eq_decay_bound_infrared_general} with small $l$, 
then $a_2\ge 1$. 

Estimation of the determinant of many-electron covariances is normally
 done by applying Gram's inequality. By following this standard approach
 we eventually have that 
\begin{align*}
&(\text{the left-hand side of
 \eqref{eq_determinant_bound_infrared_general} with }C_{o,l+1})\\
&\le \Bigg(\frac{\text{const}}{\beta L^d}\sum_{(\o,\bk)\in \cM_h\times
 \G^*}\frac{\chi(M^{-2(l+1)}(\o^2+(\cE(\bk)-\mu)^2))}{|i\o-\cE(\bk)+\mu|}\Bigg)^n.
\end{align*}
So it comes down to estimating 
\begin{align}\label{eq_eventual_integration}
\frac{1}{\beta L^d}\sum_{(\o,\bk)\in \cM_h\times
 \G^*}\frac{\chi(M^{-2(l+1)}(\o^2+(\cE(\bk)-\mu)^2))}{|i\o-\cE(\bk)+\mu|}.
\end{align}
If $\cE(\bk)=\mu$ $(\forall \bk\in \R^d)$, 
\begin{align*}
(\text{the term \eqref{eq_eventual_integration}})&\ge
 \text{const } M^{-l-1}\frac{1}{\beta}\sum_{\o\in
 \cM_h}\chi(M^{-2(l+1)}\o^2)\ge\text{const}\\
&\ge \text{const } M^{l}.
\end{align*}
Let us consider the case that $\cE(\bk)\neq \mu$ for some $\bk\in \R^d$.
Set 
$$B:=\Bigg\{\sum_{j=1}^d p_j\bv_j\ \Big|\ p_j\in [0,2\pi]\ (j=1,2,\cdots,d)\Bigg\}.$$ 
It follows that $\esssup_{\bk\in B}|\nabla
 \cE(\bk)|\neq 0$. In this case we further assume that there exists an
 interval $(\alpha_1,\alpha_2)$ containing $0$ such that
\begin{align}\label{eq_Hausdorff_measure_condition}
\inf_{\eta \in (\alpha_1,\alpha_2)}\cH^{d-1}(\{\bk\in B\ |\
 \cE(\bk)-\mu=\eta\})>0,
\end{align}
where $\cH^{d-1}$ denotes the $d-1$-dimensional Hausdorff measure on
 $\R^d$. For $l\in \Z_{<0}$ satisfying $[-c_2^{1/2}M^{l+1},c_2^{1/2}M^{l+1}]\subset (\alpha_1,\alpha_2)$ we deduce by the coarea formula that
\begin{align*}
&\lim_{L\to \infty\atop L\in\N}(\text{the term \eqref{eq_eventual_integration}})\\
&\ge \text{const}\left(\esssup_{\bk\in B}|\nabla
 \cE(\bk)|\right)^{-1}M^{-l-1}\\
&\quad\cdot\frac{1}{\beta}\sum_{\o\in
 \cM_h}\int_{-\infty}^{\infty}d\eta
 \chi(M^{-2(l+1)}(\o^2+\eta^2))\cH^{d-1}(\{\bk\in B\ |\
 \cE(\bk)-\mu=\eta\})\\
&\ge \text{const}\left(\esssup_{\bk\in B}|\nabla
 \cE(\bk)|\right)^{-1} \inf_{\eta \in(\alpha_1,\alpha_2)}\cH^{d-1}(\{\bk\in B\ |\
 \cE(\bk)-\mu=\eta\})\\
&\quad\cdot M^{l+1}.
\end{align*}
Therefore, if either $\{\bk\in B\ |\ \cE(\bk)-\mu=0\}=B$ or
 \eqref{eq_Hausdorff_measure_condition} holds, an estimation based on
 Gram's inequality can hardly yield the determinant bound
 \eqref{eq_determinant_bound_infrared_general} with $a_1>1$ for small
 $l$ and large $L$. However, if $a_1\le 1$ and $a_2\ge 1$, the inequality
 $a_1-a_2>0$ cannot hold.

For example, if $\cE(\bk)=\sum_{j=1}^d\cos k_j$, $\mu\in (-d,d)$ and
 $B=[0,2\pi]^d$, the condition \eqref{eq_Hausdorff_measure_condition}
 holds for some interval $(\alpha_1,\alpha_2)$ containing 0. This suggests that an IR integration process
 based on the iteration of Proposition
 \ref{prop_infrared_integration_general} does not instantly apply to
 the corresponding many-electron models. We should also remark that the role
of $\cE(\cdot)$ above is played by the smallest band spectrum whose zero
set consists of a single point and the chemical potential $\mu$ is set
 to be zero when we analyze our 4-band model in Section
 \ref{sec_IR_model}. In this situation, 
\begin{align*}
&\{\bk\in B\ |\ \cE(\bk)=0\}\neq B,\quad \inf_{\eta \in (\alpha_1,\alpha_2)}\cH^{1}(\{\bk\in B\ |\
 \cE(\bk)=\eta\})=0,
\end{align*}
for any interval $(\alpha_1,\alpha_2)$ containing 0. Thus, we cannot
 exclude the applicability of Proposition
 \ref{prop_infrared_integration_general} by the above argument. In fact
it will turn out that we can apply the proposition with the power
 $a_1=2$, $a_2=1$, $a_3=1$, $a_4=1/2$. 
\end{remark}

\begin{remark}\label{rem_marginal} 
Let us study a possible reconstruction of Proposition
 \ref{prop_infrared_integration_general} in the case that $a_1>0$,
 $a_1-a_2-a_4=0$. Assume that \eqref{eq_infrared_assumption_general}
 holds, $\alpha\ge cM^{a_1+a_2+a_4}$ and $M^{a_1}\ge c$ for a
 sufficiently large $c$. The only part which essentially
 needed the condition $M^{a_1-a_2-a_4}\ge c$ in the proof was the
 derivation of the upper bound on $\|F_4^l\|_{l,r}$. The bounds on the
 other terms can be obtained by using the conditions $\alpha\ge
 cM^{a_1+a_2+a_4}$ and $M^{a_1}\ge c$. 
Without using the
 condition $M^{a_1-a_2-a_4}\ge c$
we can derive from \eqref{eq_infrared_assumption_general} and
 \eqref{eq_free_part_IR_sensitive_framework} that 
\begin{align*}
M^{ra_3l}c_0^2\alpha^4\|F_4^l\|_{l,r}\le
 M^{ra_3(l+1)}c_0^2\alpha^4\|J_4^{l+1}\|_{l+1,r}+c\alpha^{-2}.
\end{align*}
Then, by combining with
 \eqref{eq_application_tree_formula_bound_sum_infrared} we obtain
\begin{align*}
&M^{ra_3l}c_0^2\alpha^4\Bigg(\|F_4^l\|_{l,r}+\sum_{n=2}^{\infty}\|T_4^{l,(n)}\|_{l,r}\Bigg)\le 
 M^{ra_3(l+1)}c_0^2\alpha^4\|J_4^{l+1}\|_{l+1,r}+c\alpha^{-2}.
\end{align*}
If we assume that \eqref{eq_determinant_bound_infrared_general},
 \eqref{eq_decay_bound_infrared_general},
 \eqref{eq_infrared_assumption_general} hold for $l'\in
 \{l,l+1,\cdots,-1\}$, we can repeatedly apply the above inequality to
 deduce that 
\begin{align*}
&M^{ra_3l}c_0^2\alpha^4\Bigg(\|F_4^l\|_{l,r}+\sum_{n=2}^{\infty}\|T_4^{l,(n)}\|_{l,r}\Bigg)\le 
 c_0^2\alpha^4\|J_4^0\|_{0,r}+c|l|\alpha^{-2}.
\end{align*}
In practice the initial polynomial $J^0(\psi)$ is an output of the
 Matsubara UV integration. We see from
 \eqref{eq_tree_2nd_4th_kernel_bound} and
 \eqref{eq_free_4th_general_bound} that the term 
$c_0^2\alpha^4\|J_4^0\|_{0,r}$ can be made less than $1/2$ by a
 minor assumption on $M$ and $\sup_{\rho\in \cB}|U_{\rho}|$ and thus
\begin{align*}
&M^{ra_3l}c_0^2\alpha^4\Bigg(\|F_4^l\|_{l,r}+\sum_{n=2}^{\infty}\|T_4^{l,(n)}\|_{l,r}\Bigg)\le \frac{1}{2}+c|l|\alpha^{-2}.
\end{align*}
This inequality implies that if the term $|l|\alpha^{-2}$ is assumed to be small,
 the effective interaction $J_4^l(\psi)$ remains small and consequently
 the conclusions of Proposition \ref{prop_infrared_integration_general}
 follow. As we will see in Section \ref{sec_IR_model}, the maximum value
 of $|l|$ in the IR integrations is proportional to $|\log\beta|$.  
Thus, we expect that in many-electron models where the quadratic kernels
 are qualitatively same as the free dispersion relation in a
 neighborhood of IR singularity and the marginal condition
 $a_1-a_2-a_4=0$ plus $a_1>0$ hold, an inductive IR integration
 procedure based on a variant of Proposition
 \ref{prop_infrared_integration_general} can be justified under the
 additional assumption that $\alpha\ge c |\log \beta|^{1/2}$. The
 condition \eqref{eq_parameters_assumption_UV_general} suggests that the
 inequality $\alpha\ge c |\log \beta|^{1/2}$ eventually restricts the
 allowed magnitude of the coupling to be less than some power of $|\log
 \beta|^{-1}$ after connecting the UV integration process to the IR
 integration process. Therefore, the resulting constructive theory in
 this case would be such that the domain of analyticity shrinks
 logarithmically with temperature.
\end{remark}

\subsection{The generalized infrared integration at different temperatures}\label{subsec_IR_general_difference}
Here we establish upper bounds on the differences between Grassmann
polynomials produced by the single-scale integration introduced in the
previous subsection at 2 different temperatures. To this end
 we need to assume that $l\in \Z_{<0}$, the condition
\eqref{eq_beta_h_assumption} holds and the covariances
$C_{o,l+1}(\beta_a):I_0(\beta_a)^2\to \C$ $(a=1,2)$ satisfy
\eqref{eq_determinant_bound_infrared_general},
\eqref{eq_decay_bound_infrared_general} and 
\begin{align}\label{eq_temperature_translation_covariance_general_once_more}
&C_{o,l+1}(\beta_a)(\bX)=(-1)^{N_{\beta_a}(\bX+x)}C_{o,l+1}(\beta_a)(R_{\beta_a}(\bX+x)),\\
&(\forall \bX\in I_0(\beta_a)^2,x\in
 (1/h)\Z, a\in\{1,2\}),\notag
\end{align}

 \begin{align}
&|\det(\<\bp_i,\bq_j\>_{\C^r}C_{o,l+1}(\beta_1)(R_{\beta_1}(X_i,Y_j)))_{1\le
  i,j\le n}\label{eq_determinant_bound_infrared_difference_general}\\
&\quad-
\det(\<\bp_i,\bq_j\>_{\C^r}C_{o,l+1}(\beta_2)(R_{\beta_2}(X_i,Y_j)))_{1\le
  i,j\le n}|\notag\\
&\le \beta_1^{-\frac{1}{2}}M^{-a_3(l+1)}(c_0 M^{a_1(l+1)})^n,\notag\\
&(\forall r,n\in\N,\bp_i,\bq_i\in\C^r\text{ with }
\|\bp_i\|_{\C^r},\|\bq_i\|_{\C^r}\le 1,\notag\\
&\quad X_i,Y_i\in \hat{I}_0\
 (i=1,2,\cdots,n)),\notag
\end{align}

\begin{align}
\left|\widetilde{C_{o,l+1}}(\beta_1)-\widetilde{C_{o,l+1}}(\beta_2)\right|_l\le
 \beta_1^{-\frac{1}{2}}{c_0}M^{-(a_2+a_3)(l+1)},\label{eq_decay_bound_infrared_difference_general}
\end{align}
where $\widetilde{C_{o,l+1}}(\beta_a):I(\beta_a)^2\to \C$ is the
anti-symmetric extension of \\
$C_{o,l+1}(\beta_a)$ $(a=1,2)$ defined as in
\eqref{eq_anti_symmetric_extension_covariance}. Let us note that the
parameters $(\fw(l),\fr)$, $(\fw(l+1),\fr)$ are also used in the
definition of $|\cdot-\cdot|_l$,  $|\cdot-\cdot|_{l+1}$ respectively. 

In addition, we assume that the input $J^{l+1}(\beta_a)(\psi)$ $(\in
\bigwedge \cV(\beta_a))$ $(a=1,2)$ satisfy
 $J^{l+1}_m(\beta_a)(\psi)=0$ 
if $m\notin 2\N$ or $m\in \{0,2\}$ and their kernels have the invariant property
\eqref{eq_temperature_translation}. Let $F^l(\beta_a)(\psi)$, $T^{l,(n)}(\beta_a)(\psi)$ $(n\in \N_{\ge 2})$,
$T^{l}(\beta_a)(\psi)$, $J^l(\beta_a)(\psi)$ $(\in
\bigwedge\cV(\beta_a))$ be defined by
\eqref{eq_inductive_polynomial_UV_general} with
$J^{l+1}(\beta_a)(\psi)$, $C_{o,l+1}(\beta_a)$ for $a=1,2$ respectively. 
We prove the following.

\begin{proposition}\label{prop_infrared_integration_difference_general}
Let $a_4\in\R_{\ge 0}$. Assume that $J^{l+1}(\beta_a)(\psi)$ $(a=1,2)$
 satisfy \eqref{eq_infrared_assumption_general} and 
\begin{align}
M^{-(a_1+a_2+a_4)(l+1)+a_3(l+1)}\sum_{m=4}^{N(\beta_2)}c_0^{\frac{m}{2}}M^{\frac{a_1(l+1)}{2}m}\alpha^m|J_m^{l+1}(\beta_1)-J_m^{l+1}(\beta_2)|_{l+1}\le
 \beta_1^{-\frac{1}{2}}.\label{eq_infrared_assumption_difference_general}
\end{align}
Then, there exists a constant $c\in\R_{>1}$ independent of any parameter such that if
 the condition \eqref{eq_parameters_assumption_IR_general} holds with
 $c$, the following inequalities hold.
\begin{align}
&\left|\frac{h}{N(\beta_1)}F_0^{l}(\beta_1)-\frac{h}{N(\beta_2)}F_0^{l}(\beta_2)\right|\label{eq_infrared_general_result_0th_difference}\\
&\quad +\sum_{n=2}^{\infty}\left|\frac{h}{N(\beta_1)}T_0^{l,(n)}(\beta_1)-\frac{h}{N(\beta_2)}T_0^{l,(n)}(\beta_2)\right|\notag\\
&\le \beta_1^{-\frac{1}{2}} M^{(a_1+a_2+a_4)l-a_3l} \alpha^{-3}.\notag\\
&M^{-(a_1+a_2+a_4)l+a_3l}\sum_{m=2}^{N(\beta_2)}c_0^{\frac{m}{2}}M^{\frac{a_1l}{2}m}\alpha^m\label{eq_infrared_general_result_difference}\\
&\quad\cdot
 \Bigg(|F_m^l(\beta_1)-F_m^l(\beta_2)|_{l}+\sum_{n=2}^{\infty}|T_m^{l,(n)}(\beta_1)- T_m^{l,(n)}(\beta_2)|_{l}\Bigg)\le \beta_1^{-\frac{1}{2}}.\notag
\end{align}
\end{proposition}

\begin{proof}
Note that the inequalities
 \eqref{eq_infrared_assumption_difference_general}, $\fw(l)\le
 \fw(l+1)$, $\alpha\ge c$,
 $M^{a_1}\ge c$ imply that
\begin{align}
&\sum_{m=4}^{N(\beta_2)}2^{3m}c_0^{\frac{m}{2}}M^{\frac{a_1(l+1)}{2}m}|J_m^{l+1}(\beta_1)-J_m^{l+1}(\beta_2)|_{l}\label{eq_input_upper_bound_infrared_difference_general}\\
&\le c \beta_1^{-\frac{1}{2}}M^{(a_1+a_2+a_4)(l+1)-a_3(l+1)}\alpha^{-4},\notag
\\
&\sum_{m=4}^{N(\beta_2)}2^{2m} c_0^{\frac{m}{2}}M^{\frac{a_1
 l}{2}m}\alpha^m|J_m^{l+1}(\beta_1)-J_m^{l+1}(\beta_2)|_{l}\label{eq_input_upper_bound_infrared_difference_general_weight}\\
&\le c
 \beta_1^{-\frac{1}{2}}M^{(a_1+a_2+a_4)(l+1)-a_3(l+1)-2a_1}.\notag
\end{align}

\eqref{eq_infrared_general_result_0th_difference}: First we prove
 \eqref{eq_infrared_general_result_0th_difference}. By inserting
 \eqref{eq_determinant_bound_infrared_general}, 
\eqref{eq_input_upper_bound_infrared_general},
 \eqref{eq_determinant_bound_infrared_difference_general},
 \eqref{eq_input_upper_bound_infrared_difference_general} into the
 inequality in Lemma \ref{lem_general_free_bound_difference}
 \eqref{item_general_free_bound_difference_0th} we observe that
\begin{align}
&\left|\frac{h}{N(\beta_1)}F_0^{l}(\beta_1)-\frac{h}{N(\beta_2)}F_0^{l}(\beta_2)\right|\label{eq_free_0th_infrared_difference_general}\\
&\le c\sum_{m=4}^{N(\beta_2)}2^mc_0^{\frac{m}{2}}M^{\frac{a_1
 (l+1)}{2}m}\Bigg(|J_m^{l+1}(\beta_1)- J_m^{l+1}(\beta_2)|_l\notag\\
&\qquad +
 \beta_1^{-\frac{1}{2}}M^{-a_3(l+1)}\sum_{\delta=1}^2\|J_m^{l+1}(\beta_{\delta})\|_{l,0}+ \beta_1^{-1}\sum_{\delta=1}^2\|J_m^{l+1}(\beta_{\delta})\|_{l,1}\Bigg)\notag\\
&\le c
\beta_1^{-\frac{1}{2}}M^{(a_1+a_2+a_4)(l+1)-a_3(l+1)}\alpha^{-4}.\notag\end{align}

By using \eqref{eq_determinant_bound_infrared_general}, 
\eqref{eq_decay_bound_infrared_general},
\eqref{eq_input_upper_bound_infrared_general},
 \eqref{eq_determinant_bound_infrared_difference_general},
\eqref{eq_decay_bound_infrared_difference_general},
\eqref{eq_input_upper_bound_infrared_difference_general}
 we can deduce from Lemma
 \ref{lem_tree_formula_general_bound_difference}
 \eqref{item_tree_formula_bound_0th_difference} that
\begin{align*}
&\left|\frac{h}{N(\beta_1)}T_0^{l,(n)}(\beta_1)-\frac{h}{N(\beta_2)}T_0^{l,(n)}(\beta_2)\right|\\
&\le c n(c_0
 M^{a_1(l+1)})^{-n+1}(2c_0M^{-a_2(l+1)})^{n-2}\\
&\quad\cdot\Bigg(\sum_{m=4}^{N(\beta_2)}2^{3m}c_0^{\frac{m}{2}}M^{\frac{a_1(l+1)}{2}m}\sum_{\delta =1}^2\|J_m^{l+1}(\beta_{\delta})\|_{l,0}
\Bigg)^{n-1}
 \sum_{m=4}^{N(\beta_2)}2^{3m}(c_0M^{a_1(l+1)})^{\frac{m}{2}}\\
&\quad\cdot\Bigg(
\beta_1^{-1}c_0M^{-a_2(l+1)}\sum_{\delta=1}^2\|J_m^{l+1}(\beta_\delta)\|_{l,1}\\&\qquad+
 \beta_1^{-1}c_0M^{-(a_2+a_3)(l+1)}\sum_{\delta=1}^2\|J_m^{l+1}(\beta_\delta)\|_{l,0}\\
&\qquad +c_0M^{-a_2(l+1)}|J_m^{l+1}(\beta_1)-J_m^{l+1}(\beta_2)|_l\\
&\qquad + \beta_1^{-\frac{1}{2}}c_0M^{-(a_2+a_3)(l+1)}\sum_{\delta=1}^2\|J_m^{l+1}(\beta_\delta)\|_{l,0}\Bigg)\\
&\le c^n( M^{-(a_1+a_2)(l+1)})^{n-1}\Bigg(\sum_{m=4}^{N(\beta_2)}2^{3m}c_0^{\frac{m}{2}}M^{\frac{a_1(l+1)}{2}m}\sum_{\delta =1}^2\|J_m^{l+1}(\beta_{\delta})\|_{l,0}\Bigg)^{n-1}\\
&\quad \cdot
 \sum_{m=4}^{N(\beta_2)}2^{3m} c_0^{\frac{m}{2}}M^{\frac{a_1(l+1)}{2}m}\Bigg(
|J_m^{l+1}(\beta_1)-J_m^{l+1}(\beta_2)|_l\\
&\qquad+\beta_1^{-\frac{1}{2}}\sum_{r=0}^1
M^{a_3(r-1)(l+1)}\sum_{\delta=1}^2\|J_m^{l+1}(\beta_\delta)\|_{l,r}\Bigg)\\
&\le c^n
 (M^{-(a_1+a_2)(l+1)})^{n-1}(M^{(a_1+a_2+a_4)(l+1)}\alpha^{-4})^{n-1}\\
&\quad\cdot\beta_1^{-\frac{1}{2}} M^{(a_1+a_2+a_4)(l+1)-a_3(l+1)}\alpha^{-4}\\
&\le \beta_1^{-\frac{1}{2}} M^{(a_1+a_2-a_3)(l+1)}(c
 M^{a_4(l+1)}\alpha^{-4})^n.
\end{align*}
Thus, by assuming that $\alpha \ge c$,
\begin{align}
\sum_{n=2}^{\infty}
\left|\frac{h}{N(\beta_1)}T_0^{l,(n)}(\beta_1)-\frac{h}{N(\beta_2)}T_0^{l,(n)}(\beta_2)\right|\le c
 \beta_1^{-\frac{1}{2}}M^{(a_1+a_2+2a_4-a_3)(l+1)}\alpha^{-8}.\label{eq_tree_0th_infrared_difference_general}
\end{align}
On the assumption $\alpha\ge cM^{a_1+a_2+a_4}$ the inequalities
 \eqref{eq_free_0th_infrared_difference_general},
 \eqref{eq_tree_0th_infrared_difference_general} imply the inequality \eqref{eq_infrared_general_result_0th_difference}.

\eqref{eq_infrared_general_result_difference}: Let us prove
 \eqref{eq_infrared_general_result_difference}. By substituting
 \eqref{eq_determinant_bound_infrared_general},
 \eqref{eq_determinant_bound_infrared_difference_general} into Lemma
 \ref{lem_general_free_bound_difference}
 \eqref{item_general_free_bound_difference_higher} we obtain that
\begin{align*}
&|F_m^{l}(\beta_1)-F_m^{l}(\beta_2)|_l\\
&\le c 1_{m=2}\sum_{n=4}^{N(\beta_2)}2^{2n}(c_0M^{a_1 (l+1)})^{\frac{n-2}{2}}\Bigg(|J_n^{l+1}(\beta_1)- J_n^{l+1}(\beta_2)|_l\\
&\qquad +
 \beta_1^{-\frac{1}{2}}\sum_{r=0}^1M^{a_3(r-1)(l+1)}\sum_{\delta=1}^2\|J_n^{l+1}(\beta_{\delta})\|_{l,r}\Bigg)\\
&\quad + c 1_{m\ge 4}\sum_{n=m}^{N(\beta_2)}2^{2n}(c_0M^{a_1 (l+1)})^{\frac{n-m}{2}}\Bigg(|J_n^{l+1}(\beta_1)- J_n^{l+1}(\beta_2)|_l\\
&\qquad +
 \beta_1^{-\frac{1}{2}}\sum_{r=0}^1M^{a_3(r-1)(l+1)}\sum_{\delta=1}^2\|J_n^{l+1}(\beta_{\delta})\|_{l,r}\Bigg).
\end{align*}
Moreover, by \eqref{eq_input_upper_bound_infrared_general},
\eqref{eq_input_upper_bound_infrared_general_weight},
\eqref{eq_input_upper_bound_infrared_difference_general},
\eqref{eq_input_upper_bound_infrared_difference_general_weight} and the
 condition $\alpha \ge c M^{a_1/2}$, 
\begin{align}
&M^{-(a_1+a_2+a_4)l+a_3l}\sum_{m=2}^{N(\beta_2)}c_0^{\frac{m}{2}}M^{\frac{a_1l}{2}m}\alpha^m|F_m^{l}(\beta_1)-F_m^{l}(\beta_2)|_l\label{eq_application_free_bound_sum_infrared_difference}\\
&\le c M^{-(a_1+a_2+a_4)l+a_3l}c_0 M^{a_1 l} \alpha^2
 \sum_{n=4}^{N(\beta_2)}2^{2n}(c_0M^{a_1 (l+1)})^{\frac{n-2}{2}}\notag\\
&\quad\cdot \Bigg(|J_n^{l+1}(\beta_1)- J_n^{l+1}(\beta_2)|_l+
 \beta_1^{-\frac{1}{2}}\sum_{r=0}^1M^{a_3(r-1)(l+1)}\sum_{\delta=1}^2\|J_n^{l+1}(\beta_{\delta})\|_{l,r}\Bigg)\notag\\
&\quad+ c
 M^{-(a_1+a_2+a_4)l+a_3l}\sum_{n=4}^{N(\beta_2)}
\sum_{m=4}^{n}c_0^{\frac{m}{2}}M^{\frac{a_1l}{2}m}\alpha^m
2^{2n}(c_0M^{a_1
 (l+1)})^{\frac{n-m}{2}}\notag\\
&\qquad\cdot\Bigg(|J_n^{l+1}(\beta_1)- J_n^{l+1}(\beta_2)|_l+
 \beta_1^{-\frac{1}{2}}\sum_{r=0}^1M^{a_3(r-1)(l+1)}\sum_{\delta=1}^2\|J_n^{l+1}(\beta_{\delta})\|_{l,r}\Bigg)\notag\\
&\le c M^{-a_1}\alpha^2 M^{-(a_1+a_2+a_4)l+a_3l}
 \sum_{n=4}^{N(\beta_2)}2^{2n}c_0^{\frac{n}{2}}M^{\frac{a_1 (l+1)}{2}n}\notag\\
&\quad\cdot \Bigg(|J_n^{l+1}(\beta_1)- J_n^{l+1}(\beta_2)|_l+
 \beta_1^{-\frac{1}{2}}\sum_{r=0}^1M^{a_3(r-1)(l+1)}\sum_{\delta=1}^2\|J_n^{l+1}(\beta_{\delta})\|_{l,r}\Bigg)\notag\\
&\quad+ c
 M^{-(a_1+a_2+a_4)l+a_3l}\sum_{n=4}^{N(\beta_2)}2^{2n}c_0^{\frac{n}{2}}M^{\frac{a_1 l}{2}n}\alpha^n\notag\\
&\qquad\cdot\Bigg(|J_n^{l+1}(\beta_1)- J_n^{l+1}(\beta_2)|_l+
 \beta_1^{-\frac{1}{2}}\sum_{r=0}^1M^{a_3(r-1)(l+1)}\sum_{\delta=1}^2\|J_n^{l+1}(\beta_{\delta})\|_{l,r}\Bigg)\notag\\
&\le c\beta_1^{-\frac{1}{2}}M^{-a_1+a_2-a_3+a_4}.\notag
\end{align}

On the other hand, by substituting
 \eqref{eq_determinant_bound_infrared_general}, 
\eqref{eq_decay_bound_infrared_general},
\eqref{eq_determinant_bound_infrared_difference_general},
\eqref{eq_decay_bound_infrared_difference_general}
into the inequality in Lemma
 \ref{lem_tree_formula_general_bound_difference} 
 \eqref{item_tree_formula_bound_higher_difference} we see that
\begin{align*}
&|T_m^{l,(n)}(\beta_1)-T_m^{l,(n)}(\beta_2)|_l\\
&\le c n 2^{-2m}(c_0
 M^{a_1(l+1)})^{-\frac{m}{2}-n+1}(2c_0
 M^{-a_2(l+1)})^{n-2}\\
&\quad \cdot \prod_{j=2}^n\Bigg(\sum_{m_j=4}^{N(\beta_2)}2^{4m_j}(c_0
 M^{a_1(l+1)})^{\frac{m_j}{2}}\sum_{\delta=1}^2\|J_{m_j}^{l+1}(\beta_{\delta})\|_{l,0}\Bigg)\\
&\quad\cdot
 \sum_{m_1=4}^{N(\beta_2)}2^{4m_1}(c_0M^{a_1(l+1)})^{\frac{m_1}{2}}\\
&\quad\cdot\Bigg( 
\beta_1^{-1}c_0M^{-a_2(l+1)}\sum_{\delta=1}^2\|J_{m_1}^{l+1}(\beta_{\delta})\|_{l,1}\\
&\qquad +
 \beta_1^{-1}c_0M^{-(a_2+a_3)(l+1)}\sum_{\delta=1}^2\|J_{m_1}^{l+1}(\beta_\delta)\|_{l,0}\\
&\qquad + c_0M^{-a_2(l+1)}|J_{m_1}^{l+1}(\beta_1)-J_{m_1}^{l+1}(\beta_2)|_l\\
&\qquad +
 \beta_1^{-\frac{1}{2}}c_0M^{-(a_2+a_3)(l+1)}\sum_{\delta=1}^2\|J_{m_1}^{l+1}(\beta_\delta)\|_{l,0}\Bigg)1_{\sum_{j=1}^nm_j-2n+2\ge m}\\
&\le c^n 2^{-2m}c_0^{-\frac{m}{2}} M^{-\frac{a_1(l+1)}{2}m}(M^{-(a_1+a_2)(l+1)})^{n-1}\\
&\quad \cdot
 \prod_{j=2}^n\Bigg(\sum_{m_j=4}^{N(\beta_2)}2^{4m_j}c_0^{\frac{m_j}{2}}
 M^{\frac{a_1(l+1)}{2}m_j}\sum_{\delta=1}^2\|J_{m_j}^{l+1}(\beta_{\delta})\|_{l,0}\Bigg)\\
&\quad\cdot \sum_{m_1=4}^{N(\beta_2)}2^{4m_1}c_0^{\frac{m_1}{2}} M^{\frac{a_1(l+1)}{2}m_1}\\
&\quad\cdot\Bigg(|J_{m_1}^{l+1}(\beta_1)-J_{m_1}^{l+1}(\beta_2)|_l +
\beta_1^{-\frac{1}{2}}\sum_{r=0}^1M^{a_3(r-1)(l+1)}\sum_{\delta=1}^2\|J_{m_1}^{l+1}(\beta_{\delta})\|_{l,r}\Bigg)\\
&\quad\cdot 1_{\sum_{j=1}^nm_j-2n+2\ge m}
\end{align*}
Moreover, by \eqref{eq_input_upper_bound_infrared_general_weight},
 \eqref{eq_input_upper_bound_infrared_difference_general_weight} and the
 condition $\alpha \ge c M^{a_1/2}$, 
\begin{align*}
&M^{-(a_1+a_2+a_4)l+a_3l}\sum_{m=2}^{N(\beta_2)}c_0^{\frac{m}{2}}
 M^{\frac{a_1 l}{2}m}\alpha^m|T_m^{l,(n)}(\beta_1)-T_m^{l,(n)}(\beta_2)|_l\\
&\le c^n M^{-(a_1+a_2+a_4)l+a_3l}(M^{-(a_1+a_2)(l+1)})^{n-1}\\
&\quad \cdot
 \prod_{j=2}^n\Bigg(\sum_{m_j=4}^{N(\beta_2)}2^{4m_j}c_0^{\frac{m_j}{2}} M^{\frac{a_1(l+1)}{2}m_j}\sum_{\delta=1}^2\|J_{m_j}^{l+1}(\beta_{\delta})\|_{l,0}\Bigg)\\&\quad\cdot \sum_{m_1=4}^{N(\beta_2)}2^{4m_1}c_0^{\frac{m_1}{2}} M^{\frac{a_1(l+1)}{2}m_1}\\
&\quad\cdot\Bigg(|J_{m_1}^{l+1}(\beta_1)-J_{m_1}^{l+1}(\beta_2)|_l +
{\beta_1}^{-\frac{1}{2}}\sum_{r=0}^1M^{a_3(r-1)(l+1)}\sum_{\delta=1}^2\|J_{m_1}^{l+1}(\beta_{\delta})\|_{l,r}\Bigg)\\
&\quad\cdot
 2^{-2(\sum_{j=1}^nm_j-2n+2)}M^{-\frac{a_1}{2}(\sum_{j=1}^nm_j-2n+2)}\alpha^{\sum_{j=1}^nm_j-2n+2}\\
&\le c^n M^{-(a_1+a_2+a_4)l+a_3l}(M^{-a_1l-a_2(l+1)}\alpha^{-2})^{n-1}\\
&\quad \cdot
 \prod_{j=2}^n\Bigg(\sum_{m_j=4}^{N(\beta_2)}2^{2m_j}c_0^{\frac{m_j}{2}} M^{\frac{a_1l}{2}m_j}\alpha^{m_j}\sum_{\delta=1}^2\|J_{m_j}^{l+1}(\beta_{\delta})\|_{l,0}\Bigg)\\
&\quad\cdot \sum_{m_1=4}^{N(\beta_2)}2^{2m_1}c_0^{\frac{m_1}{2}} M^{\frac{a_1l}{2}m_1}\alpha^{m_1}\\
&\quad\cdot\Bigg(|J_{m_1}^{l+1}(\beta_1)-J_{m_1}^{l+1}(\beta_2)|_l +
\beta_1^{-\frac{1}{2}}\sum_{r=0}^1M^{a_3(r-1)(l+1)}\sum_{\delta=1}^2\|J_{m_1}^{l+1}(\beta_{\delta})\|_{l,r}\Bigg)\\
&\le c^n M^{-(a_1+a_2+a_4)l+a_3l}(M^{-a_1 l -
 a_2(l+1)}\alpha^{-2})^{n-1}(M^{(a_1+a_2+a_4)(l+1)-2a_1})^{n-1}\\
&\quad\cdot \beta_{1}^{-\frac{1}{2}}M^{(a_1+a_2+a_4)(l+1)-a_3(l+1)-2a_1}\\
&\le \beta_1^{-\frac{1}{2}}M^{a_2-a_3-a_4l}\alpha^{2}(c M^{-a_1+a_4(l+1)}\alpha^{-2})^n.\end{align*}
Then, by the assumption $\alpha \ge c$ we have 
\begin{align}
&M^{-(a_1+a_2+a_4)l+a_3l}\sum_{m=2}^{N(\beta_2)}c_0^{\frac{m}{2}}
 M^{\frac{a_1l}{2}m}\alpha^m\sum_{n=2}^{\infty}|T_m^{l,(n)}(\beta_1)-T_m^{l,(n)}(\beta_2)|_l\label{eq_application_tree_formula_bound_sum_infrared_difference}
\\
&\le c
 \beta_1^{-\frac{1}{2}}M^{-2a_1+a_2-a_3+2a_4+a_4l}\alpha^{-2}\le c
 \beta_1^{-\frac{1}{2}}M^{-a_1+a_2+a_4}.\notag
\end{align}

By coupling
 \eqref{eq_application_tree_formula_bound_sum_infrared_difference} with
 \eqref{eq_application_free_bound_sum_infrared_difference} and using the
 condition $M^{a_1-a_2-a_4}\ge c$ we reach the
 inequality \eqref{eq_infrared_general_result_difference}.
\end{proof}

\section{The Matsubara ultra-violet integration}\label{sec_UV}

The results summarized in Subsection \ref{subsec_UV_general} and
Subsection \ref{subsec_UV_general_difference} have practical
applications in the multi-scale integration over the Matsubara
frequency, which we are going to present in this section. The purpose
of the Matsubara UV integration in this paper is to find analytic
continuations of the Grassmann polynomials $R^+(\psi)$, $R^-(\psi)$,
which were defined in Lemma \ref{lem_logarithm_appl_bound}, into a $(\beta,L,h)$-independent domain of the multi-variables
$(U_1,U_2,\cdots,U_b)$ around the origin. This will enable us to consider
$(R^+(\psi)+R^-(\psi))/2$ as appropriate initial data for
the forthcoming infrared integration. What we need to achieve our
purpose is to show that the covariances used in the definition of
 $R^+(\psi)$, $R^-(\psi)$ can be decomposed into a sum of
covariances satisfying the conditions required in
Proposition \ref{prop_UV_integration_general} and Proposition
\ref{prop_UV_integration_difference_general}. Then we can prove the
existence of desired analytic continuations of $R^+(\psi)$, $R^-(\psi)$ by applying these propositions. 
The construction of this
section is based on the assumption that the matrix-valued function
$E:\R^d\to \Mat(b,\C)$ satisfies $E\in C^{\infty}(\R^d;\Mat(b,\C))$, the properties
\eqref{eq_dispersion_hermite}, \eqref{eq_dispersion_periodic} and 
\begin{align}
&\sup_{j\in\{1,2,\cdots,d\}}\sup_{(p_1,p_2,\cdots,p_d)\in\R^d}\left\|\left(\frac{\partial}{\partial
 p_j }\right)^n E\Bigg(\sum_{r=1}^dp_r\bv_r\Bigg)\right\|_{b\times b}\le E_1\cdot E_2^nn!,\label{eq_real_analytic_dispersion_relation}\\
&(\forall n\in\N\cup \{0\})\notag
\end{align}
with constants $E_1$, $E_2\in\R_{\ge 0}$. 

In order to shorten formulas, from this section we let the symbol
$c(\alpha_1,\alpha_2,\cdots,\alpha_n)$ denote a real positive constant
depending only on parameters $\alpha_1,\alpha_2,\cdots,\alpha_n$. For
example, $c(\beta,L)$ denotes a positive constant depending only on
$\beta$, $L$. 

\subsection{The covariance matrices with the Matsubara ultra-violet cut-off}
\label{subsec_covariance_UV}
First we have to specifically define a cut-off function on the
Matsubara frequency. Motivated by \cite[\mbox{Appendix A}]{P}, we
construct the cut-off function from a suitable
Gevrey-class function.
\begin{lemma}\label{lem_gevrey_class} There exists a function $\phi\in
 C^{\infty}(\R)$ satisfying the following properties.
\begin{align*}
&\phi(x)=1,\ (\forall x\in (-\infty,\pi^2/6]),\\
&\phi(x)=0,\ (\forall x\in [\pi^2/3,\infty)),\\
&\frac{d}{dx}\phi(x)\le 0,\ (\forall x\in \R),
\end{align*}
and 
\begin{align}
\left|\left(\frac{d}{dx}\right)^k\phi(x)\right|\le 2^k (k!)^2,\ (\forall x\in\R,k\in\N\cup
 \{0\}).
\label{eq_derivative_gevrey_cut_off}
\end{align}
\end{lemma}
\begin{proof}
Let us take the sequence $(a_j)_{j=0}^{\infty}$ in \cite[\mbox{Theorem
 1.3.5}]{H} to be $((j+1)^{-2})_{j=0}^{\infty}$. Since $\sum_{j=0}^{\infty}a_j=\pi^2/6$, the theorem reads
 that there exists a function $u\in C_0^{\infty}(\R)$ satisfying  
\begin{align*}
&u(x)\ge 0,\ (\forall x\in \R),\\
&u(x)=0,\ (\forall x\in \R\backslash [0,\pi^2/6]),\\
&\left|\left(\frac{d}{dx}\right)^ku(x)\right|\le 2^k ((k+1)!)^2,\ (\forall x\in\R,k\in\N\cup
 \{0\}),\\
&\int_{-\infty}^{\infty}u(x)dx=1.
\end{align*}
Set
$$
\phi(x):=\int_0^{-x+\frac{\pi^2}{3}}u(y)dy.
$$
One can check that the function $\phi$ satisfies the claimed properties.
\end{proof}

Take $M\in \R_{>1}$ and set 
\begin{align*}
&M_{UV}:=\frac{2\sqrt{6}}{\pi}(E_1+1),\\
&N_h:=\max\left\{\left[\frac{\log\big(2h(\frac{\pi^2}{6})^{-1/2}M_{UV}^{-1}\big)}{\log M}\right]+1,1\right\},
\end{align*}
where the symbol $[x]$ denotes the largest integer
less than or equal to $x$ for any $x\in \R$. We see that 
\begin{align*}
h|1-e^{i\frac{\o}{h}}|\le 2h\le
 \left(\frac{\pi^2}{6}\right)^{\frac{1}{2}}{M_{UV}}M^{N_h},\ (\forall \o\in\R).
\end{align*}
Thus,
$$
\phi(M_{UV}^{-2}M^{-2N_h}h^2|1-e^{i\frac{\o}{h}}|^2)=1,\ (\forall \o\in\R),
$$
where $\phi(\cdot)$ is the function introduced in Lemma \ref{lem_gevrey_class}.

For any $\o\in\R$ set
\begin{align*}
&\chi_{h,0}(\o):=\phi(M_{UV}^{-2}h^2|1-e^{i\frac{\o}{h}}|^2),\\
&\chi_{h,l}(\o):=\phi(M_{UV}^{-2}M^{-2l}h^2|1-e^{i\frac{\o}{h}}|^2)-
\phi(M_{UV}^{-2}M^{-2(l-1)}h^2|1-e^{i\frac{\o}{h}}|^2),\\
&(l\in\{1,2,\cdots,N_h\}).
\end{align*}
Then, we have that 
\begin{align}
\chi_{h,0}(\o)+\sum_{l=1}^{N_h}\chi_{h,l}(\o)=1,\ (\forall \o\in
 \R).\label{eq_cut_off_function_equality_UV}
\end{align}
The values of $\chi_{h,0}(\cdot)$, $\chi_{h,l}(\cdot)$ are
described as follows. 
\begin{align}
&\chi_{h,0}(\o)\left\{\begin{array}{ll}=1, &\text{if
	       }h|1-e^{i\frac{\o}{h}}|\le\frac{\pi}{\sqrt{6}}M_{UV},\\
\in [0,1], &\text{if }\frac{\pi}{\sqrt{6}}M_{UV}< h|1-e^{i\frac{\o}{h}}|<\frac{\pi}{\sqrt{3}}M_{UV},\\
=0, &\text{if
	       }\frac{\pi}{\sqrt{3}}M_{UV}\le
		     h|1-e^{i\frac{\o}{h}}|,\end{array}
\right.\label{eq_support_description_UV_cut_off}\\
&\chi_{h,l}(\o)\left\{\begin{array}{ll}=0, &\text{if
	       }h|1-e^{i\frac{\o}{h}}|\le\frac{\pi}{\sqrt{6}}M_{UV}M^{l-1},\\
\in [0,1], &\text{if }\frac{\pi}{\sqrt{6}}M_{UV}M^{l-1}< h|1-e^{i\frac{\o}{h}}|<\frac{\pi}{\sqrt{3}}M_{UV}M^{l},\\
=0, &\text{if
	       }\frac{\pi}{\sqrt{3}}M_{UV}M^{l}\le
		     h|1-e^{i\frac{\o}{h}}|,\end{array}
\right.\notag\\
&(\forall l\in\{1,2,\cdots,N_h\},\o\in\R).\notag
\end{align}

Using these cut-off functions, we define the covariances $C_l^+$,
$C_l^-:I_0^2\to \C$ $(l=0,1,\cdots,N_h)$ as follows. For any
$(\rho,\bx,\s,x)$, $(\eta,\by,\tau,y)\in I_0$, 
\begin{align}
&C_{l}^+(\rho\bx\s x,\eta\by\tau y)\label{eq_introduction_UV_covariance}\\
&:=\frac{\delta_{\s,\tau}}{\beta
 L^d}\sum_{\bk\in \G^*}\sum_{\o\in
 \cM_h}e^{-i\<\bx-\by,\bk\>}e^{i(x-y)\o}\chi_{h,l}(\o)h^{-1}(I_b-e^{-i\frac{\o}{h}I_b+\frac{1}{h}\overline{E(\bk)}})^{-1}(\rho,\eta),\notag\\
&C_{l}^-(\rho\bx\s x,\eta\by\tau y)\notag\\
&:=\frac{\delta_{\s,\tau}}{\beta
 L^d}\sum_{\bk\in \G^*}\sum_{\o\in
 \cM_h}e^{-i\<\bx-\by,\bk\>}e^{i(x-y)\o}\chi_{h,l}(\o)h^{-1}(e^{i\frac{\o}{h}I_b-\frac{1}{h}\overline{E(\bk)}}-I_b)^{-1}(\rho,\eta).\notag
\end{align}
We also define the covariances $C_{\le 0}^+$, $C_{> 0}^+$, $C_{> 0}^-$
by \eqref{eq_covariance_negative_index},
\eqref{eq_covariance_non_negative_index},
\eqref{eq_covariance_non_negative_index_another} respectively by employing
$\chi_{h,0}(\cdot)$ in place of $\chi(h|1-e^{i\cdot /h}|)$.
It follows from \eqref{eq_cut_off_function_equality_UV} that 
\begin{align*}
C_{>0}^+=\sum_{l=1}^{N_h}C_l^+,\ C_{>0}^-=\sum_{l=1}^{N_h}C_l^-,\
 C=\sum_{l=0}^{N_h}C_l^+.
\end{align*}

We show in the next lemma that the covariances $C_l^+$,
$C_l^-:I_0^2\to \C$ $(l=1,2,\cdots,N_h)$ satisfy the bound properties required
 in Subsection \ref{subsec_UV_general}. For this purpose let us
introduce finite-difference operators. 
For any function $f:\R\times \R^d\to \C$, 
set
\begin{align*}
&\cD_0f(\o,\bk):=\frac{\beta}{2\pi}\left(f\left(\o+\frac{2\pi}{\beta},\bk\right)-f(\o,\bk)\right),\\
&\cD_jf(\o,\bk):=
 \frac{L}{2\pi}\left(f\left(\o,\bk+\frac{2\pi}{L}\bv_j\right)-f(\o,\bk)\right),\\ 
&(\forall j\in \{1,2,\cdots,d\},(\o,\bk)\in\R\times\R^d).
\end{align*}

\begin{lemma}\label{lem_covariance_UV_bound}
Assume that $h\ge e^{2E_1}$. Then, there exist constants $c_0$,
 $c_0'\in\R_{\ge 1}$, which depend only on $b$, $d$, $M_{UV}$, $M$, $E_2$, and
 a constant $c_w\in (0,1]$ independent of any parameter such that the
 covariances $C_l^+$, $C_l^-$ $(l=1,2,\cdots,N_h)$ satisfy
 \eqref{eq_determinant_bound_UV_general},
\eqref{eq_decay_bound_UV_general} with $c_0$, $N_+=N_h$, the weight 
$$
\fw(0)=c_w(d+1)^{-2}\min\{M_{UV},(E_2+1)^{-1}\}M^{-2}
$$
and the exponent $\fr=1/2$, and 
 \eqref{eq_tadpole_bound_UV_general} with $c_0'$, $N_+=N_h$.
\end{lemma}
\begin{proof}
We prove the claims on $C_l^+$. The boundedness of $C_l^-$ can be proved
 in the same way.   Since
 $(2/\pi)|\o|\le h|1-e^{i\o/h}|\le |\o|$ for any $\o\in\R$ with $|\o|\le
 \pi h$, the conditions $\chi_{h,l}(\o)\neq 0$ and $|\o|\le \pi h$ implies that \begin{align}
c M_{UV}M^{l-1}\le |\o|\le c
 M_{UV}M^{l},\ (\forall l\in \{1,2,\cdots,N_h\}).\label{eq_matsubara_lower_upper_UV}
\end{align}
Therefore,
\begin{align}
&\frac{1}{\beta}\sum_{\o\in \cM_h} 1_{\chi_{h,l}(\o)\neq 0}\le c
 M_{UV}M^{l},\label{eq_support_cut_off_UV_measure}\\
&\int_{-\pi h}^{\pi h}d\o 1_{\chi_{h,l}(\o)\neq 0}\le c
 M_{UV}M^l,\ (\forall l\in \{1,2,\cdots,N_h\}).\notag
\end{align}
 
(Proof for \eqref{eq_determinant_bound_UV_general}): Note that 
\begin{align}
&h^{-1}(I_b-e^{-i\frac{\o}{h}I_b+\frac{1}{h}\overline{E(\bk)}})^{-1}\label{eq_expression_inverse_matrix_UV}\\
&=h^{-1}(1-e^{-i\frac{\o}{h}})^{-1}\big(I_b-h^{-1}(e^{i\frac{\o}{h}}-1)^{-1}h(e^{\frac{1}{h}\overline{E(\bk)}}-I_b)\big)^{-1}.\notag
\end{align}
If $h\ge e^{E_1}$, $\|h(e^{\frac{1}{h}\overline{E(\bk)}}-I_b)\|_{b\times
 b}\le E_1+1$. Thus, we can see from the definition of $M_{UV}$ and
 \eqref{eq_support_description_UV_cut_off} that 
\begin{align}
&\|h^{-1}(e^{i\frac{\o}{h}}-1)^{-1}h(e^{\frac{1}{h}\overline{E(\bk)}}-I_b)\|_{b\times
 b}\le\frac{\sqrt{6}}{\pi}M_{UV}^{-1}M^{-l+1}(E_1+1)=
 \frac{1}{2}M^{-l+1},\label{eq_denominator_denominator_dominant_UV}\\
&(\forall \o\in \R\text{ with }\chi_{h,l}(\o)\neq 0,\bk\in\R^d).\notag
\end{align}
It follows from  
 \eqref{eq_support_description_UV_cut_off},
 \eqref{eq_expression_inverse_matrix_UV} and \eqref{eq_denominator_denominator_dominant_UV} that
\begin{align}
&\|h^{-1}(I_b-e^{-i\frac{\o}{h}I_b+\frac{1}{h}\overline{E(\bk)}})^{-1}\|_{b\times
 b}\le c M_{UV}^{-1}M^{-l+1},\label{eq_bound_inverse_matrix_UV}\\
&(\forall \o\in \R\text{ with }\chi_{h,l}(\o)\neq 0,\bk\in \R^d).\notag
\end{align}

Recall the definition of the Hilbert space $\cH$ introduced in the proof
 of Lemma \ref{lem_h_independent_determinant_bound}. For any
 $(\rho,\bx,\s,x)\in I_0$ we define $f_{\rho\bx\s x}^l$,  $g_{\rho\bx\s
 x}^l\in \cH$ by
\begin{align*}
&f_{\rho\bx\s
 x}^l(\eta,\bk,\tau,\o):=\delta_{\rho,\eta}\delta_{\s,\tau}e^{i\<\bx,\bk\>}e^{-ix\o}\chi_{h,l}(\o)^{\frac{1}{2}},\notag\\
& g_{\rho\bx\s
 x}^l(\eta,\bk,\tau,\o):=\delta_{\s,\tau}e^{i\<\bx,\bk\>}e^{-ix\o}\chi_{h,l}(\o)^{\frac{1}{2}}h^{-1}(I_b-e^{-i\frac{\o}{h}I_b+\frac{1}{h}\overline{E(\bk)}})^{-1}(\eta,\rho).
\end{align*}
Then, $C_l^+(X,Y)=\<f_X^l,g_Y^l\>_{\cH}$, $(\forall X,Y\in
 I_0)$. Moreover, by \eqref{eq_support_cut_off_UV_measure} and
 \eqref{eq_bound_inverse_matrix_UV}, 
\begin{align}
\|f_{\rho\bx\s x}^l\|_{\cH}&\le
 c(M_{UV})M^{\frac{l}{2}},\label{eq_each_vector_bound_UV}\\
\|g_{\rho\bx\s x}^l\|_{\cH}&\le \Bigg(\frac{1}{\beta L^d}\sum_{\bk\in
 \G^*}\sum_{\o\in\cM_h}\chi_{h,l}(\o)\|h^{-1}(I_b-e^{-i\frac{\o}{h}I_b+\frac{1}{h}\overline{E(\bk)}})^{-1}\|_{b\times
 b}^2\Bigg)^{\frac{1}{2}}\notag\\
&\le  c(M_{UV},M)M^{-\frac{l}{2}},\ (\forall (\rho,\bx,\s,x)\in I_0).
\notag 
\end{align}
For any $r\in \N$ let $\C^r\otimes \cH$ denote the tensor product of the
 Hilbert spaces $\C^r$, $\cH$. Since 
\begin{align*}
\<\bp,\bq\>_{\C^r}C_l^+(X,Y)=\<\bp\otimes f_X^l, \bq\otimes
 g_Y^l\>_{\C^r\otimes \cH},\ (\forall \bp,\bq\in\C^r,X,Y\in I_0),
\end{align*}
Gram's inequality in the Hilbert space $\C^r\otimes \cH$ and
 \eqref{eq_each_vector_bound_UV} ensure that
\begin{align*}
&|\det(\<\bp_i,\bq_j\>_{\C^r}C_{l}^+(X_i,Y_j))_{1\le i,j\le n}|\le\prod_{i=1}^n\|\bp_i\otimes f_{X_i}^l\|_{\C^r\otimes \cH}\|\bq_i\otimes
 g_{Y_i}^l\|_{\C^r\otimes \cH}\\
&\le \prod_{i=1}^n\|f_{X_i}^l\|_{\cH}\|g_{Y_i}^l\|_{\cH}\le c(M_{UV},M)^n,\\
&(\forall r,n\in\N,\bp_i,\bq_i\in\C^r\text{ with }
\|\bp_i\|_{\C^r},\|\bq_i\|_{\C^r}\le 1,\notag\\
&\ X_i,Y_i\in I_0\
 (i=1,2,\cdots,n)).
\end{align*}

(Proof for \eqref{eq_decay_bound_UV_general}): By the assumption $h\ge
 e^{2E_1}$,
$$
N_h=\left[\frac{\log\big(2h(\frac{\pi^2}{6})^{-1/2}M_{UV}^{-1}\big)}{\log M}\right]+1.
$$
Thus, 
\begin{align}
M_{UV}M^{N_h-1}\le ch.\label{eq_h_M_inequality}
\end{align}
Using \eqref{eq_h_M_inequality} and the inequality $E_1/h\le 1$, 
we can check that
\begin{align}
&\left\|\left(\frac{\partial}{\partial
 \o}\right)^nh(I_b-e^{-i\frac{\o}{h}I_b+\frac{1}{h}\overline{E(\bk)}})\right\|_{b\times
 b}\label{eq_bound_inside_matrix_0th_derivative_UV}\\
&\le c h^{1-n}\le M_{UV}M^{l-1}(c M_{UV}^{-1}M^{1-l})^nn!,\ (\forall n\in
 \N).\notag
\end{align}
We can apply Lemma \ref{lem_gevrey_matrix}
 \eqref{item_gevrey_matrix_power} proved in Appendix \ref{app_gevrey}
together with  \eqref{eq_real_analytic_dispersion_relation} and the assumption
 $h\ge e^{2E_1}$ to derive that for any $j\in\{1,2,\cdots,d\}$,
\begin{align}
&\left\|\left(\frac{\partial}{\partial
 p_j}\right)^nh(I_b-e^{-i\frac{\o}{h}I_b+\frac{1}{h}\overline{E(\sum_{r=1}^d
 p_r\bv_r)}})\right\|_{b\times
 b}\label{eq_bound_inside_matrix_jth_derivative_UV}\\
&\le \sum_{m=1}^{\infty}\frah^{m-1}\frac{1}{m!}\left\|\left(\frac{\partial}{\partial
 p_j}\right)^n{\overline{E\left(\sum_{r=1}^dp_r\bv_r\right)}}^m\right\|_{b\times b}\notag\\
&\le
 \sum_{m=1}^{\infty}\frah^{m-1}\frac{1}{m!}(2E_1)^m(2E_2)^nn!\le (2E_1+1)(2E_2)^nn!,\ (\forall n\in \N).\notag
\end{align}
Taking into account \eqref{eq_bound_inverse_matrix_UV},
 \eqref{eq_bound_inside_matrix_0th_derivative_UV}, we can substitute 
$s=c M_{UV}^{-1}M^{-l+1}$, $q={M_{UV}}M^{l-1}$, $r=c
 M_{UV}^{-1}M^{-l+1}$, $t=1$ into the inequality in Lemma
 \ref{lem_gevrey_matrix} \eqref{item_gevrey_inverse_matrix} to
 obtain 
\begin{align}
&\left\|\left(\frac{\partial}{\partial
 \o}\right)^nh^{-1}(I_b-e^{-i\frac{\o}{h}I_b+\frac{1}{h}\overline{E(\bk)}})^{-1}\right\|_{b\times
 b}\label{eq_bound_inverse_matrix_derivative_0th_UV}\\ 
&\le c M_{UV}^{-1}M^{-l+1}(c M_{UV}^{-1}M^{-l+1})^nn!,\notag\\
&(\forall n\in \N\cup\{0\}, \o\in\R\text{ with
 }\chi_{h,l}(\o)\neq 0, \bk\in \R^d).\notag
\end{align}
Here we also used \eqref{eq_bound_inverse_matrix_UV} to claim
 \eqref{eq_bound_inverse_matrix_derivative_0th_UV} for $n=0$.
By \eqref{eq_bound_inverse_matrix_UV},
 \eqref{eq_bound_inside_matrix_jth_derivative_UV} we can apply Lemma
 \ref{lem_gevrey_matrix} \eqref{item_gevrey_inverse_matrix}
with $s=c M_{UV}^{-1}M^{-l+1}$, $q=2E_1+1$, $r=2E_2$, $t=1$ to deduce
 that
\begin{align}
&\left\|\left(\frac{\partial}{\partial
 p_j}\right)^nh^{-1}(I_b-e^{-i\frac{\o}{h}I_b+\frac{1}{h}\overline{E(\sum_{r=1}^dp_r\bv_r)}})^{-1}\right\|_{b\times
 b}\label{eq_bound_inverse_matrix_derivative_jth_UV}\\
&\le
 c M_{UV}^{-2}M^{-2l+2}(E_1+1)(1+c(E_1+1) M_{UV}^{-1}M^{-l+1})^{-1}\notag\\
&\quad\cdot(cE_2(1+c(E_1+1)M_{UV}^{-1}M^{-l+1}))^nn!\notag\\
&\le
 c M_{UV}^{-1}M^{-l+1}(cE_2)^nn!,\notag\\
&(\forall n\in \N,j\in\{1,2,\cdots,d\},\o\in\R\text{ with
 }\chi_{h,l}(\o)\neq 0,\notag\\
&\quad (p_1,p_2,\cdots,p_d)\in \R^d ),\notag
\end{align}
where we used the inequality $M_{UV}\ge E_1+1$ as well.

For any $\o\in\R$ let $p_h(\o)$ denote a number belonging to $[-\pi h,\pi h)$
 and satisfying $\o=p_h(\o)$ in $\R/2\pi h \Z$. By using
 \eqref{eq_matsubara_lower_upper_UV} we have for any $n\in \N$ and
 $\o\in\R$ with $\chi_{h,l}(\o)\neq 0$ that
\begin{align}
&\left|\left(\frac{d}{d\o}\right)^n(M_{UV}^{-2}M^{-2(l-1)}h^2|1-e^{i\frac{\o}{h}}|^2)\right|\label{eq_bound_inside_cut_off_derivative_UV}\\
&\le c M_{UV}^{-2}M^{-2(l-1)}c^{n+1}|p_h(\o)|^{2-n}n!\le c
 M^2(c M_{UV}^{-1}M^{-l+1})^nn!.\notag
\end{align}
By \eqref{eq_derivative_gevrey_cut_off},
 \eqref{eq_bound_inside_cut_off_derivative_UV} we can substitute
 $q_1=cM^2$, $r_1=cM_{UV}^{-1}M^{-l+1}$, $q_2=1$, $r_2=2$, $t=2$ 
into the result of Lemma
 \ref{lem_gevrey_composition} to derive that 
\begin{align}
\left|\left(\frac{d}{d\o}\right)^n\chi_{h,l}(\o)\right|
&\le cM^2(1+cM^2)^{-1}(c M_{UV}^{-1}M^{-l+1}(1+cM^2))^n(n!)^2
\label{eq_bound_cut_off_derivative_0th_UV}\\ 
&\le c(c M_{UV}^{-1}M^{-l+3})^n(n!)^2,\ (\forall n\in \N\cup \{0\},\o\in \R).\notag
\end{align}

By the periodicity with the variable $\o$, 
\begin{align}
&\left(\frac{\beta}{2\pi}\right)^n
 (e^{-i\frac{2\pi}{\beta}(x-y)}-1)^nC_{l}^+(\cdot\bx\s x,\cdot\by\tau y)\label{eq_finite_difference_matsubara_UV}\\
&=\frac{\delta_{\s,\tau}}{\beta
 L^d}\sum_{\bk\in \G^*}\sum_{\o\in
 \cM_h}e^{-i\<\bx-\by,\bk\>}e^{i(x-y)\o}\cD_0^n\big(\chi_{h,l}(\o)h^{-1}(I_b-e^{-i\frac{\o}{h}I_b+\frac{1}{h}\overline{E(\bk)}})^{-1}\big)\notag\\
&=\frac{\delta_{\s,\tau}}{\beta L^d}\sum_{\bk\in \G^*}\sum_{\o\in
 \cM_h}e^{-i\<\bx-\by,\bk\>}e^{i(x-y)\o}\prod_{r=1}^n\left(\frac{\beta}{2\pi}\int_0^{2\pi/\beta}d\o_r\right)\left(\frac{\partial}{\partial
 \eta}\right)^n\notag\\
&\quad\cdot
 \big(\chi_{h,l}(\eta)h^{-1}(I_b-e^{-i\frac{\eta}{h}I_b+\frac{1}{h}\overline{E(\bk)}})^{-1}\big)\Big|_{\eta
 = \o+\sum_{r=1}^n\o_r}\notag.
\end{align}
Note that by \eqref{eq_bound_inverse_matrix_derivative_0th_UV},
 \eqref{eq_bound_cut_off_derivative_0th_UV}, 
\begin{align}
&\left\|\left(\frac{\partial}{\partial
 \o}\right)^n\big(\chi_{h,l}(\o)h^{-1}(I_b-e^{-i\frac{\o}{h}I_b+\frac{1}{h}\overline{E(\bk)}})^{-1}\big)\right\|_{b\times
 b}\label{eq_bound_integrant_derivative_0th_UV}\\
&\le \sum_{j=0}^n\left(\begin{array}{c}n \\
		       j\end{array}\right)
\left|\left(\frac{d}{d \o}\right)^j\chi_{h,l}(\o)\right|\left\|\left(\frac{\partial}{\partial
 \o}\right)^{n-j}h^{-1}(I_b-e^{-i\frac{\o}{h}I_b+\frac{1}{h}\overline{E(\bk)}})^{-1}\right\|_{b\times
 b}\notag\\
&\le \sum_{j=0}^n\left(\begin{array}{c}n \\
		       j\end{array}\right)(cM_{UV}^{-1}M^{-l+3})^j(j!)^2\notag\\
&\quad\cdot c M_{UV}^{-1}M^{-l+1}(cM_{UV}^{-1}M^{-l+1})^{n-j}(n-j)!\notag\\
&\le c M_{UV}^{-1}M^{-l+1}(c M_{UV}^{-1}M^{-l+3})^{n}(n!)^2,\ (\forall n\in \N\cup \{0\},\o \in \R,\bk\in \R^d).\notag
\end{align}
By combining \eqref{eq_support_cut_off_UV_measure},
 \eqref{eq_bound_integrant_derivative_0th_UV} with
 \eqref{eq_finite_difference_matsubara_UV} we obtain that
\begin{align}
&\left\|\left(\frac{\beta}{2\pi}\right)^n
 (e^{-i\frac{2\pi}{\beta}(x-y)}-1)^nC_{l}^+(\cdot\bx\s x,\cdot\by\tau
 y)\right\|_{b\times b}\label{eq_time_direction_derivative_bound}\\
&\le cM (c M_{UV}^{-1}M^{-l+3})^{n}(n!)^2,\ (\forall n\in  \N\cup\{0\}).\notag
\end{align}
This implies that 
\begin{align}
&\|C_{l}^+(\cdot\bx\s x,\cdot\by\tau
 y)\|_{b\times b}\le cM e^{-(c^{-1}M_{UV}
 M^{l-3}\frac{\beta}{2\pi}|e^{i(2\pi/\beta)(x-y)}-1|)^{1/2}},\label{eq_weight_bound_0th_UV}\\
&(\forall (\bx,\s,x), (\by,\tau,y)\in\G\times\spin\times [0,\beta)_h). \notag
\end{align}
By the periodic condition \eqref{eq_dispersion_periodic} we similarly
 have for any $j\in \{1,2,\cdots,$ $d\}$ that
\begin{align*}
&\left(\frac{L}{2\pi}\right)^n
 (e^{i\frac{2\pi}{L}\<\bx-\by,\bv_j\>}-1)^nC_{l}^+(\cdot\bx\s x,\cdot\by\tau y)\\
&=\frac{\delta_{\s,\tau}}{\beta
 L^d}\sum_{\bk\in \G^*}\sum_{\o\in
 \cM_h}e^{-i\<\bx-\by,\bk\>}e^{i(x-y)\o}\chi_{h,l}(\o)\cD_j^n(h^{-1}(I_b-e^{-i\frac{\o}{h}I_b+\frac{1}{h}\overline{E(\bk)}})^{-1})\\
&=\frac{\delta_{\s,\tau}}{\beta L^d}\sum_{\bk\in \G^*}\sum_{\o\in
 \cM_h}e^{-i\<\bx-\by,\bk\>}e^{i(x-y)\o}\chi_{h,l}(\o)\prod_{r=1}^n\left(\frac{L}{2\pi}\int_0^{2\pi/L}dq_r\right)\left(\frac{\partial}{\partial
 p_j}\right)^n\\
&\quad\cdot
 \big(h^{-1}(I_b-e^{-i\frac{\o}{h}I_b+\frac{1}{h}\overline{E(\bk+p_j\bv_j)}})^{-1}\big)\Big|_{p_j=\sum_{r=1}^nq_r}.
\end{align*}
Insertion of \eqref{eq_support_cut_off_UV_measure}, \eqref{eq_bound_inverse_matrix_derivative_jth_UV} yields 
\begin{align*}
&\left\|\left(\frac{L}{2\pi}\right)^n
 (e^{i\frac{2\pi}{L}\<\bx-\by,\bv_j\>}-1)^nC_{l}^+(\cdot\bx\s x,\cdot\by\tau
 y)\right\|_{b\times b}\\
&\le cM (c E_2)^{n}n!,\ (\forall n\in  \N\cup\{0\}),
\end{align*}
where we also used \eqref{eq_time_direction_derivative_bound} to claim
 this equality for $n=0$.
This leads to 
\begin{align}
&\|C_{l}^+(\cdot\bx\s x,\cdot\by\tau
 y)\|_{b\times b}\le cM e^{-(c^{-1}(E_2+1)^{-1}\frac{L}{2\pi}|e^{i(2\pi/L)\<\bx-\by,\bv_j\>}-1|)^{1/2}},\label{eq_weight_bound_jth_UV}\\
&(\forall (\bx,\s,x), (\by,\tau,y)\in\G\times\spin\times
 [0,\beta)_h, j\in \{1,2,\cdots,d\}). \notag
\end{align}

One can derive from \eqref{eq_weight_bound_0th_UV},
 \eqref{eq_weight_bound_jth_UV} that
\begin{align*}
&\|C_{l}^+(\cdot\bx\s x,\cdot\by\tau
 y)\|_{b\times b}\le cM
 e^{-(c^{-1}(d+1)^{-2}M_{UV}M^{l-3}\frac{\beta}{2\pi}|e^{i(2\pi/\beta)(x-y)}-1|)^{1/2}}\\
 &\qquad\qquad\qquad\qquad\qquad\cdot e^{-\sum_{j=1}^d(c^{-1}(d+1)^{-2}(E_2+1)^{-1}\frac{L}{2\pi}|e^{i(2\pi/L)\<\bx-\by,\bv_j\>}-1|)^{1/2}},\\
&(\forall (\bx,\s,x), (\by,\tau,y)\in\G\times\spin\times
 [0,\beta)_h).
\end{align*}
Here we may assume that $c\ge 1$ by taking a larger number if necessary.
Set
\begin{align*} 
&c_w:= \frac{1}{9}c^{-1},\\
&\fw(0):=c_w(d+1)^{-2}\min\{M_{UV},(E_2+1)^{-1}\}M^{-2},
\end{align*}
with the constant $c\in \R_{\ge 1}$ appearing in the above inequality.
Then, we have
\begin{align}
&|C_{l}^+(\rho\bx\s x,\eta\by\tau
 y)|\label{eq_UV_covariance_gevrey_decay}\\
&\le c(M) e^{-3(\fw(0)M^{l-1}\frac{\beta}{2\pi}|e^{i(2\pi/\beta)(x-y)}-1|)^{1/2}}
 \cdot e^{-3\sum_{j=1}^d(\fw(0)\frac{L}{2\pi}|e^{i(2\pi/L)\<\bx-\by,\bv_j\>}-1|)^{1/2}},\notag\\
&(\forall (\rho,\bx,\s,x), (\eta,\by,\tau,y)\in I_0).\notag
\end{align}

With this weight $\fw(0)$ and the exponent $\fr=1/2$ we define the norm
 $\|\cdot\|_{0,0}$ and the semi-norm $\|\cdot\|_{0,1}$ by
 \eqref{eq_definition_norm_semi_norm}. We can check that
\begin{align*}
\|\widetilde{C_l^+}\|_{0,0}&\le c(M)\sup_{X\in I}\frac{1}{h}\sum_{Y\in
 I}e^{-(\fw(0)M^{l-1}d_0(X,Y))^{1/2}}\cdot
 e^{-\sum_{j=1}^d(\fw(0)d_j(X,Y))^{1/2}}\\
&\le c(M,b,d,\fw(0))M^{-l},\\
\|\widetilde{C_l^+}\|_{0,1}&\le c(M)\sup_{i\in\{0,1,\cdots,d\}}\sup_{X\in I}\frac{1}{h}\sum_{Y\in
 I}d_i(X,Y)\\
&\quad\cdot e^{-(\fw(0)M^{l-1}d_0(X,Y))^{1/2}}\cdot
 e^{-\sum_{j=1}^d(\fw(0)d_j(X,Y))^{1/2}}\\
&\le c(M,b,d,\fw(0))M^{-l}.
\end{align*}
These upper bounds can be derived even if we define $c_w$ as
 $(1/4)c^{-1}$ so that the right-hand side of 
\eqref{eq_UV_covariance_gevrey_decay} contains the factor 2 in place of
 3. However, we choose to define $c_w$ as
 $(1/9)c^{-1}$ in order that \eqref{eq_UV_covariance_gevrey_decay} can
 be used in the next lemma, too.

(Proof for \eqref{eq_tadpole_bound_UV_general}): It follows from
 \eqref{eq_expression_inverse_matrix_UV} that for any $\o\in \R$ with
 $\chi_{h,l}(\o)\neq 0$,
\begin{align*}
&h^{-1}(I_b-e^{-i\frac{\o}{h}I_b+\frac{1}{h}\overline{E(\bk)}})^{-1}\\
&=h^{-1}(1-e^{-i\frac{\o}{h}})^{-1}I_b\\
&\quad +h^{-1}(1-e^{-i\frac{\o}{h}})^{-1}\sum_{n=1}^{\infty}\big(h^{-1}(e^{i\frac{\o}{h}}-1)^{-1}h(e^{\frac{1}{h}\overline{E(\bk)}}-I_b)\big)^n.
\end{align*}
By substituting this equality we obtain 
\begin{align*}
&C_l^+(\cdot \b0\s 0, \cdot \b0\s 0)\\
&=\frac{1}{2\beta h}\sum_{\o\in \cM_h}\chi_{h,l}(\o)I_b+\frac{1}{\beta L^d}\sum_{\bk\in \G^*}\sum_{\o\in \cM_h}\chi_{h,l}(\o)
h^{-1}(1-e^{-i\frac{\o}{h}})^{-1}\\
&\qquad\qquad\qquad\qquad\qquad\qquad\cdot\sum_{n=1}^{\infty}\big(h^{-1}(e^{i\frac{\o}{h}}-1)^{-1}h(e^{\frac{1}{h}\overline{E(\bk)}}-I_b)\big)^n.
\end{align*}
Moreover, by \eqref{eq_support_description_UV_cut_off},
\eqref{eq_support_cut_off_UV_measure},
 \eqref{eq_denominator_denominator_dominant_UV} and
 \eqref{eq_h_M_inequality},
\begin{align*}
\|C_l^+(\cdot \b0 \s 0,\cdot \b0 \s 0)\|_{b\times b}&\le
 c(M)M^{l-N_h}+c(M)\sum_{n=1}^{\infty}\left(\frac{1}{2}M^{-l+1}\right)^n\\
&\le c(M)(M^{l-N_h}+M^{-l}),
\end{align*}
which implies \eqref{eq_tadpole_bound_UV_general} with $N_+=N_h$.
\end{proof}

Next let us find upper bounds on the differences between the covariances
defined at 2 different temperatures. These bounds were required in
Subsection \ref{subsec_UV_general_difference}.
\begin{lemma}\label{lem_covariance_UV_bound_difference}
Assume that \eqref{eq_beta_h_assumption} holds and $h\ge e^{2E_1}$. Let $\fw(0)$ be the weight introduced in
 Lemma \ref{lem_covariance_UV_bound}. Then, there exist constants $c_0$,
 $c_0'\in\R_{\ge 1}$, which depend only on $b$, $d$, $M_{UV}$, $M$,
 $E_2$, such
 that  the
 covariances $C_l^{+}(\beta_a)$, $C_l^{-}(\beta_a)$ $(l=1,2,\cdots,N_h, a=1,2)$ satisfy
 \eqref{eq_determinant_bound_UV_difference_general},
\eqref{eq_decay_bound_UV_difference_general} with $c_0$, $N_+=N_h$, the weight 
$\fw(0)$ and the exponent $\fr=1/2$, and 
 \eqref{eq_tadpole_bound_UV_difference_general} with $c_0'$, $N_+=N_h$.
\end{lemma}

\begin{proof}
We give the proof for $C_l^+$. The claims for $C_l^-$ can be proved in
 the same way. For any $(\rho,\bx,\s,x)$, $(\eta,\by,\tau,y)\in \hat{I}_0$, $a\in
 \{1,2\}$, set 
\begin{align*}
&C_{ont,l}(\beta_a)(\rho\bx\s x,\eta\by\tau y)\\
&:=(-1)^{n_{\beta_a}(x)+n_{\beta_a}(y)}
\frac{\delta_{\s,\tau}}{2\pi L^d}\sum_{\bk\in \G^*}\int_{-\pi h}^{\pi
 h}d\o e^{-i\<\bx-\by,\bk\>}e^{i(x-y)\o}\chi_{h,l}(\o)\\
&\quad\cdot h^{-1}(I_b-e^{-i\frac{\o}{h}I_b+\frac{1}{h}\overline{E(\bk)}})^{-1}(\rho,\eta).
\end{align*}
Since $n_{\beta_1}(x)=n_{\beta_2}(x)$ $(\forall x\in
 [-\beta_1/4,\beta_1/4)_h)$,
 $C_{ont,l}(\beta_1)=C_{ont,l}(\beta_2)$. Note that 
for any $a\in \{1,2\}$,
\begin{align*}
&C_{ont,l}(\beta_a)(\cdot \bx \s x,\cdot \by\tau
 y)-C_{l}^+(\beta_a)(\cdot \bx \s r_{\beta_a}(x),\cdot \by\tau
 r_{\beta_a}(y))\\
&=(-1)^{n_{\beta_a}(x)+n_{\beta_a}(y)}
\frac{\delta_{\s,\tau}}{2\pi L^d}\sum_{\bk\in \G^*}
 e^{-i\<\bx-\by,\bk\>}\\
&\quad \cdot \sum_{m=0}^{\frac{\beta_a
 h}{2}-1}\Bigg(\int_{\frac{2\pi}{\beta_a}m+\frac{\pi}{\beta_a}}^{\frac{2\pi}{\beta_a}(m+1)+\frac{\pi}{\beta_a}}d\o\int_{\frac{2\pi}{\beta_a}m+\frac{\pi}{\beta_a}}^{\o}du
 +
 \int_{-\frac{2\pi}{\beta_a}(m+1)-\frac{\pi}{\beta_a}}^{-\frac{2\pi}{\beta_a}m-\frac{\pi}{\beta_a}}d\o\int_{-\frac{2\pi}{\beta_a}m-\frac{\pi}{\beta_a}}^{\o}du
 \Bigg)\\
&\quad\cdot \frac{\partial}{\partial u}\big(e^{i(x-y)u}\chi_{h,l}(u)
h^{-1}(I_b-e^{-i\frac{u}{h}I_b+\frac{1}{h}\overline{E(\bk)}})^{-1}\big)\\
&\quad + (-1)^{n_{\beta_a}(x)+n_{\beta_a}(y)}
\frac{\delta_{\s,\tau}}{2\pi L^d}\sum_{\bk\in \G^*}
 e^{-i\<\bx-\by,\bk\>}\\
&\qquad\cdot\Bigg(\int_{-\frac{\pi}{\beta_a}}^{\frac{\pi}{\beta_a}}d\o
-\int_{\pi h}^{\pi h+\frac{\pi}{\beta_a}}d\o- \int^{-\pi h}_{-\pi h-\frac{\pi}{\beta_a}}d\o\Bigg)\\
&\qquad \cdot e^{i(x-y)\o}\chi_{h,l}(\o)
h^{-1}(I_b-e^{-i\frac{\o}{h}I_b+\frac{1}{h}\overline{E(\bk)}})^{-1}.
\end{align*}
On the assumption $h\ge e^{2E_1}$ we can apply
 \eqref{eq_support_cut_off_UV_measure},
 \eqref{eq_bound_integrant_derivative_0th_UV} to deduce that 
\begin{align}
&\|C_{l}^+(\beta_1)(\cdot \bx \s r_{\beta_1}(x),\cdot \by\tau
 r_{\beta_1}(y))-C_{l}^+(\beta_2)(\cdot \bx \s r_{\beta_2}(x),\cdot \by\tau
 r_{\beta_2}(y))\|_{b\times
 b}\label{eq_covariance_difference_UV_temperature}\\
&\le \sum_{a=1}^2\|C_{ont,l}(\beta_a)(\cdot \bx \s x,\cdot \by\tau
 y)-C_{l}^+(\beta_a)(\cdot \bx \s r_{\beta_a}(x),\cdot \by\tau
 r_{\beta_a}(y))\|_{b\times b}\notag\\
&\le \sum_{a=1}^2\frac{1}{\beta_a L^d}\sum_{\bk\in
 \G^*}\Bigg(\int_{\frac{\pi}{\beta_a}}^{\pi h+\frac{\pi}{\beta_a}}d\o
+\int^{-\frac{\pi}{\beta_a}}_{-\pi
 h-\frac{\pi}{\beta_a}}d\o\Bigg)\notag\\
&\quad\cdot
 \Bigg(|x-y|\chi_{h,l}(\o)\|h^{-1}(I_b-e^{-i\frac{\o}{h}I_b+\frac{1}{h}\overline{E(\bk)}})^{-1}\|_{b\times
 b}\notag\\
&\qquad+\left\|\frac{\partial}{\partial
 \o}(\chi_{h,l}(\o)h^{-1}(I_b-e^{-i\frac{\o}{h}I_b+\frac{1}{h}\overline{E(\bk)}})^{-1})\right\|_{b\times
 b}\Bigg)\notag\\
&\quad + \sum_{a=1}^2\frac{1}{2\pi L^d}\sum_{\bk\in
 \G^*}
\Bigg(\int_{-\frac{\pi}{\beta_a}}^{\frac{\pi}{\beta_a}}d\o
+\int_{\pi h}^{\pi h+\frac{\pi}{\beta_a}}d\o+ \int^{-\pi h}_{-\pi
 h-\frac{\pi}{\beta_a}}d\o\Bigg)\notag\\
&\qquad\cdot \chi_{h,l}(\o)\|h^{-1}(I_b-e^{-i\frac{\o}{h}I_b+\frac{1}{h}\overline{E(\bk)}})^{-1}\|_{b\times
 b}\notag\\
&\le c(M_{UV},M)\beta_1^{-1}(|x-y|+M^{-l}).\notag
\end{align}

On the other hand, the inequalities \eqref{eq_basic_temperature_bound},
 \eqref{eq_UV_covariance_gevrey_decay} imply that
\begin{align}
&|C_{l}^+(\beta_1)(\rho \bx \s r_{\beta_1}(x),\eta \by\tau
 r_{\beta_1}(y))-C_{l}^+(\beta_2)(\rho \bx \s r_{\beta_2}(x),\eta \by\tau
 r_{\beta_2}(y))|\label{eq_covariance_difference_UV_pre}\\
&\le
 c(M)e^{-3(\fw(0)M^{l-1}\frac{1}{\pi}|x-y|)^{1/2}-3\sum_{j=1}^d(\fw(0)\frac{L}{2\pi}|e^{i(2\pi/L)\<\bx-\by,\bv_j\>}-1|)^{1/2}}.\notag
\end{align}
By combining \eqref{eq_covariance_difference_UV_pre} with \eqref{eq_covariance_difference_UV_temperature} we have
\begin{align}
&|C_{l}^+(\beta_1)(\rho \bx \s r_{\beta_1}(x),\eta \by\tau
 r_{\beta_1}(y))-C_{l}^+(\beta_2)(\rho \bx \s r_{\beta_2}(x),\eta \by\tau
 r_{\beta_2}(y))|\label{eq_covariance_difference_UV_decay}\\
&\le
 c(M_{UV},M)\beta_1^{-\frac{1}{2}}M^{-\frac{l}{2}}(M^l|x-y|+1)^{\frac{1}{2}}\notag\\
&\quad\cdot e^{-\frac{3}{2}(\fw(0)M^{l-1}\frac{1}{\pi}|x-y|)^{1/2}-\frac{3}{2}\sum_{j=1}^d(\fw(0)\frac{L}{2\pi}|e^{i(2\pi/L)\<\bx-\by,\bv_j\>}-1|)^{1/2}},\notag\\
&(\forall (\rho,\bx,\s,x),(\eta,\by,\tau,y)\in \hat{I}_0).\notag
\end{align}
The inequalities \eqref{eq_decay_bound_UV_difference_general},
 \eqref{eq_tadpole_bound_UV_difference_general} follow from
 \eqref{eq_covariance_difference_UV_decay}.

To prove \eqref{eq_determinant_bound_UV_difference_general}, take any $X_i$,
 $Y_i\in \hat{I}_0$ and $\bp_i$, $\bq_i\in \C^r$ satisfying
 $\|\bp_i\|_{\C^r}$, $\|\bq_i\|_{\C^r}\le 1$ $(i=1,2,\cdots,n)$. Expanding the
 determinant along the 1st column and 
 using \eqref{eq_covariance_difference_UV_decay}, we observe that 
\begin{align*}
&|\det(\<\bp_i,\bq_j\>_{\C^r}(C_l^+(\beta_1)(R_{\beta_1}(X_i,Y_j))-C_l^+(\beta_2)(R_{\beta_2}(X_i,Y_j))))_{1\le
 i,j\le n}|\\
&\le c(M_{UV},M)\beta_1^{-\frac{1}{2}}\sum_{s=1}^n\\
&\quad\cdot\big|\det(\<\bp_i,\bq_j\>_{\C^r}(C_l^+(\beta_1)(R_{\beta_1}(X_i,Y_j))-C_l^+(\beta_2)(R_{\beta_2}(X_i,Y_j))))_{1\le
 i,j\le n\atop i\neq s,j\neq 1}\big|.
\end{align*}
Then, expanding the remaining determinant as in
 \eqref{eq_application_cauchy_binet} and substituting the determinant bound
 \eqref{eq_determinant_bound_UV_general} yield the result.
\end{proof}
 
\subsection{Application of the generalized ultra-violet integration}
\label{subsec_application_UV}
 \ \\ 
Since we have checked that the covariances $C_l^{\delta}$
$(l=1,2,\cdots,N_h,\delta = + ,-)$ satisfy the desired bound properties,
we can readily apply the propositions proved in Subsection
\ref{subsec_UV_general} and Subsection
\ref{subsec_UV_general_difference} to complete the Matsubara UV
integration.  With $V^{\delta}(\psi)$ $(\in \bigwedge \cV)$
($\delta=+,-$) defined in \eqref{eq_signed_interaction_polynomial}, set
\begin{align*}
&F^{\delta,N_h}(\psi):=-V^{\delta}(\psi),\\ 
&T^{\delta,N_h,(n)}(\psi):=0,\
 (\forall n\in \N_{\ge 2}),\ T^{\delta,N_h}(\psi):=0,\\
&J^{\delta,N_h}(\psi):= F^{\delta,N_h}(\psi)+T^{\delta,N_h}(\psi),\
 (\forall \delta\in\{+,-\}).
\end{align*}
Then, we inductively define $F^{\delta,l}(\psi)$,
$T^{\delta,l,(n)}(\psi)$ $(n\in \N_{\ge 2})$, $T^{\delta,l}(\psi)$,
$J^{\delta,l}(\psi)$ $\in
\bigwedge \cV$ $(l\in \{0,1,\cdots,N_h-1\})$  by
\eqref{eq_inductive_polynomial_UV_general} with the covariances 
$\{C_{l}^{\delta}\}_{l=1}^{N_h}$ for $\delta = +,-$ respectively.

\begin{proposition}\label{prop_UV_integration} 
Let the weight $\fw(0)$ be the same as in Lemma
 \ref{lem_covariance_UV_bound}, Lemma
 \ref{lem_covariance_UV_bound_difference} and the exponent $\fr$ be
 $1/2$. Assume that $h\ge e^{2E_1}$. Then, there exist constants $c_0$, $c_0'\in\R_{\ge 1}$,  which depend only on
 $b$, $d$, $M_{UV}$, $M$, $E_2$, and a constant $c\in\R_{>0}$ independent of
 any parameter such that if the parameters $M,\alpha\in \R_{\ge 1}$,
 $U_{\rho}\in \C$ $(\rho\in \cB)$ satisfy
\begin{align}
M\ge c,\ \alpha^2\ge cM,\ \sup_{\rho\in\cB}|U_{\rho}|\le\frac{1}{c ({c_0}+{c_0'})^2\alpha^4},\label{eq_UV_practice_assumption}
\end{align}
the following statements hold true.
\begin{enumerate}
\item\label{item_UV_bound}
For any $\delta \in \{+,-\}$, $r\in \{0,1\}$ and $l\in \{0,1,\cdots,N_h\}$,
\begin{align*}
&\frac{h}{N}\Bigg(|F^{\delta,l}_0|+\sum_{n=2}^{\infty}|T_0^{\delta,l,(n)}|\Bigg)\le
 \alpha^{-4},\\
&c_0\alpha^2\Bigg(\|F_2^{\delta,l}\|_{0,r}+\sum_{n=2}^{\infty}\|T_2^{\delta,l,(n)}\|_{0,r}\Bigg)\le
 1,\\
&M^{-2l}\sum_{m=2}^Nc_0^{\frac{m}{2}}M^{\frac{l}{2}m}\alpha^m\Bigg(\|F_m^{\delta,l}\|_{0,r}+\sum_{n=2}^{\infty}\|T_m^{\delta,l,(n)}\|_{0,r}\Bigg)\le
 1.
\end{align*}
\item\label{item_UV_analyticity} 
For any
     $\delta \in \{+,-\}$, $r\in \{0,1\}$ and $l\in \{0,1,\cdots,N_h\}$,
     $J^{\delta,l}(\psi)$ is continuous with $(U_1,U_2,\cdots,U_b)$ in 
$$
\{(U_1,U_2,\cdots,U_b)\in \C^b\ |\ |U_{\rho}|\le
     (c({c_0}+{c_0'})^2\alpha^4)^{-1},\ (\forall \rho\in \cB)\}
$$
and analytic with $(U_1,U_2,\cdots,U_b)$ in 
$$
\{(U_1,U_2,\cdots,U_b)\in \C^b\ |\ |U_{\rho}|<
     (c({c_0}+{c_0'})^2\alpha^4)^{-1},\ (\forall \rho\in \cB)\}.
$$
\item\label{item_UV_analytic_continuation} 
There exists a $(\beta,L)$-dependent, $h$-independent constant
     $c'\in\R_{>0}$ such that if the inequality $\sup_{\rho\in
     \cB}|U_{\rho}|\le c'$ additionally holds,
$$\Re \int e^{-V^{\delta}(\psi)}d\mu_{C_{>0}^{\delta}}(\psi)>0$$
and
$$
J^{\delta,0}(\psi)=\log\left(\int e^{-V^{\delta
     (\psi+\psi^1)}}d\mu_{C_{>0}^{\delta}}(\psi^1)
\right)
$$
for any $\delta \in \{+,-\}$. 
\item\label{item_UV_difference}
Assume that \eqref{eq_beta_h_assumption} holds. 
For any $\delta \in \{+,-\}$, $r\in \{0,1\}$ and $l\in
     \{0,1,\cdots,N_h\}$,
\begin{align*}
&\left|\frac{h}{N(\beta_1)}F^{\delta,l}_0(\beta_1)-\frac{h}{N(\beta_2)}F^{\delta,l}_0(\beta_2)\right|\\
&\quad+\sum_{n=2}^{\infty}\left|\frac{h}{N(\beta_1)}T_0^{\delta,l,(n)}(\beta_1)-
\frac{h}{N(\beta_2)}T_0^{\delta,l,(n)}(\beta_2)\right|\le
 \beta_1^{-\frac{1}{2}}\alpha^{-4},\\
&c_0\alpha^2\Bigg(|F_2^{\delta,l}(\beta_1)-F_2^{\delta,l}(\beta_2)|_{0}+\sum_{n=2}^{\infty}|T_2^{\delta,l,(n)}(\beta_1)-T_2^{\delta,l,(n)}(\beta_2)|_{0}\Bigg)\le \beta_1^{-\frac{1}{2}},\\
&M^{-2l}\sum_{m=2}^{N(\beta_2)}c_0^{\frac{m}{2}}M^{\frac{l}{2}m}\alpha^m\Bigg(|F_m^{\delta,l}(\beta_1)-F_m^{\delta,l}(\beta_2)|_{0}\\
&\quad+\sum_{n=2}^{\infty}|T_m^{\delta,l,(n)}(\beta_1)-T_m^{\delta,l,(n)}(\beta_2)|_{0}\Bigg)\le
 \beta_1^{-\frac{1}{2}}.
\end{align*}
\end{enumerate}
\end{proposition}

\begin{remark}\label{rem_replacement_possible_UV_constant}
It will be shown in the proof below that 
the constant $c_0$ in Proposition \ref{prop_UV_integration} is equal to
 the maximum of $c_0$ appearing in Lemma \ref{lem_covariance_UV_bound} and Lemma
 \ref{lem_covariance_UV_bound_difference}. Since these lemmas hold for
 any larger constant and the generic constant $c$ is independent of $c_0$, we have the freedom to replace $c_0$ in
 Proposition \ref{prop_UV_integration} by any larger constant without
 changing the constant $c$. Such a
 replacement will be necessary when we connect the UV integration to
 the IR integration in Subsection
 \ref{subsec_application_IR}.
\end{remark}

\begin{remark}
The definition of $J^{\delta,l}(\psi)$ $(l=0,1,\cdots,N_h)$ depends on
 the parameter $M$. This means that we have to fix $M$ before introducing these
 polynomials. However, the definition of these polynomials does
 not depend on the parameter $\alpha$. For $J^{\delta,l}(\psi)$
 $(l=0,1,\cdots,N_h)$ the results of Proposition
 \ref{prop_UV_integration} hold for any $\alpha \in\R_{\ge 1}$
 satisfying \eqref{eq_UV_practice_assumption}. Bearing this fact in
 mind, we will use the
 results of \eqref{item_UV_bound}, \eqref{item_UV_analyticity} for a
large, $(\beta,L,h)$-dependent $\alpha$ to prove the claim
 \eqref{item_UV_analytic_continuation} during the proof of the
 proposition below.
\end{remark}

\begin{proof}[Proof of Proposition \ref{prop_UV_integration}]
Let $c_0$, $c_0'(\in\R_{\ge 1})$ be the maximum of $c_0$, $c_0'$
 appearing in Lemma \ref{lem_covariance_UV_bound} and Lemma
 \ref{lem_covariance_UV_bound_difference} respectively. Then, there
 exists a constant $c\in\R_{>0}$ independent of any parameter such that
 if \eqref{eq_UV_practice_assumption} holds with $c$, the results of
 Proposition \ref{prop_UV_integration_general}, Proposition \ref{prop_UV_integration_difference_general} hold true. In the following we
 assume \eqref{eq_UV_practice_assumption} with this $c$.

\eqref{item_UV_bound}, \eqref{item_UV_difference}: We can apply Proposition
 \ref{prop_UV_integration_general}, Proposition
 \ref{prop_UV_integration_difference_general} to justify
 \eqref{item_UV_bound}, \eqref{item_UV_difference} respectively.

\eqref{item_UV_analyticity}: Fix $\delta \in \{+,-\}$. Set
$$
D:=\{(U_1,U_2,\cdots,U_b)\in \C^b\ |\ |U_{\rho}|<
     (c({c_0}+{c_0'})^2\alpha^4)^{-1},\ (\forall \rho\in \cB)\}.
$$
Apparently $J^{\delta, N_h}(\psi)$ is continuous in $\overline{D}$ and
 analytic in $D$. Assume that $l\in \{0,1,\cdots,N_h-1\}$ and
 $J^{\delta,l+1}(\psi)$ is continuous in $\overline{D}$ and analytic in
 $D$. Then, so are $F^{\delta,l}(\psi)$, $T^{\delta,l,(n)}(\psi)$
 $(n\in \N_{\ge 2})$, since these consist of finite sums and products of
 $J^{\delta,l+1}(\psi)$. The claim \eqref{item_UV_bound} implies that
 $\sum_{n=2}^{\infty}T^{\delta,l,(n)}(\psi)$ converges uniformly with respect to $(U_1,U_2,\cdots,U_b)$
 in $\overline{D}$. Therefore, $T^{\delta,l}(\psi)$ is continuous in
 $\overline{D}$ and analytic in $D$, and thus so is
 $J^{\delta,l}(\psi)$. The induction with respect to $l$ verifies the claim.

\eqref{item_UV_analytic_continuation}: Fix $\delta \in
 \{+,-\}$. First let us note that by definition there exists a $(\beta,L,h)$-dependent constant
 $\tilde{c}\in\R_{>0}$ such that if $\sup_{\rho\in\cB}|U_{\rho}|\le
 \tilde{c}$,
\begin{align}
\Re \int
 e^{-V^{\delta}(\psi)}d\mu_{\sum_{j=l+1}^{N_h}C_j^{\delta}}(\psi)>0,\
 (\forall l\in\{0,1,\cdots,N_h-1\}).\label{eq_positive_real_part_UV}
\end{align}

It follows from \eqref{item_UV_bound} that
\begin{align}
&|J_0^{\delta,l}|\le\frac{N}{h}\alpha^{-4},\label{eq_UV_result_simplified}\\
&\|J_m^{\delta,l}\|_{0,0}\le
 c_0^{-\frac{m}{2}}M^{-\frac{l}{2}m+2l}\alpha^{-m},\notag\\
&(\forall l\in \{0,1,\cdots,N_h\},m\in \{2,3,\cdots,N\}).\notag
\end{align}
This implies that there exists a $(\beta,L,h)$-dependent constant
 $\tilde{c}'\in\R_{>0}$ such that if $\alpha\ge \tilde{c}'$,
\begin{align*}
&\Re \int
 e^{z J^{\delta,l+1}(\psi)}d\mu_{C_{l+1}^{\delta}}(\psi)>0,\\
 &(\forall l\in\{0,1,\cdots,N_h-1\},z\in\C\text{ with }|z|\le 2).
\end{align*}
Thus, the Grassmann polynomials
$$
\log\left(\int
 e^{zJ^{\delta,l+1}(\psi+\psi^1)}d\mu_{C_{l+1}^{\delta}}(\psi^1)\right)\ (l=0,1,\cdots,N_h-1)
$$
are analytic with $z$ in $\{z\in \C\ |\ |z|<2 \}$ if
\begin{align}
\sup_{\rho\in\cB}|U_{\rho}|\le\frac{1}{c
 ({c_0}+{c_0'})^2\tilde{c}'^4}.\label{eq_coupling_constant_temporary_small}
\end{align}
Therefore, if
 \eqref{eq_coupling_constant_temporary_small} holds,
\begin{align}
&\log\left(\int
 e^{J^{\delta,l+1}(\psi+\psi^1)}d\mu_{C_{l+1}^{\delta}}(\psi^1)\right)\label{eq_desired_equality_UV_pre}\\
&=\sum_{n=1}^{\infty}\frac{1}{n!}\left(\frac{d}{dz}\right)^n\log\left(\int
 e^{zJ^{\delta,l+1}(\psi+\psi^1)}d\mu_{C_{l+1}^{\delta}}(\psi^1)\right)\Bigg|_{z=0}\notag\\
&=F^{\delta,l}(\psi)+\sum_{n=2}^{\infty}T^{\delta,l,(n)}(\psi)\notag\\
&=J^{\delta,l}(\psi),\ (\forall l\in\{0,1,\cdots,N_h-1\}).\notag
\end{align}

Let us show that 
\begin{align}
&J^{\delta,l}(\psi)=\log\left(\int e^{-V^{\delta
     (\psi+\psi^1)}}d\mu_{\sum_{j=l+1}^{N_h}C_{j}^{\delta}}(\psi^1)\right),\label{eq_desired_equality_UV}\\
&(\forall l\in \{0,1,\cdots,N_h-1\})\notag
\end{align}
on the assumption 
\begin{align}
\sup_{\rho\in\cB}|U_{\rho}|\le\min\left\{\tilde{c},\frac{1}{c
 ({c_0}+{c_0'})^2\tilde{c}'^4}\right\}.\label{eq_coupling_constant_temporary_small_next}
\end{align}
The equality \eqref{eq_desired_equality_UV} for $l=N_h-1$ holds, since
 it is equivalent to \eqref{eq_desired_equality_UV_pre} for $l=N_h-1$. Assume that
 \eqref{eq_desired_equality_UV} holds for $l+1$. By the condition
 \eqref{eq_positive_real_part_UV} we can apply \cite[\mbox{Lemma
 C.2}]{K3} to justify the equality 
$$
e^{J^{\delta,l+1}(\psi)}=\int e^{-V^{\delta}(\psi+\psi^1)}d\mu_{\sum_{j=l+2}^{N_h}C_{j}^{\delta}}(\psi^1).
$$
Moreover, by \eqref{eq_desired_equality_UV_pre} and
 \cite[\mbox{Proposition I.21}]{FKT2}, 
\begin{align*}
J^{\delta,l}(\psi)&=\log\left(\int
 e^{J^{\delta,l+1}(\psi+\psi^1)}d\mu_{C_{l+1}^{\delta}}(\psi^1)\right)\\
&=\log\left(\int\int
 e^{-V^{\delta}(\psi+\psi^1+\psi^2)}d\mu_{\sum_{j=l+2}^{N_h}C_j^{\delta}}(\psi^2)d\mu_{C_{l+1}^{\delta}}(\psi^1)\right)\\
&=\log\left(\int e^{-V^{\delta
     (\psi+\psi^1)}}d\mu_{\sum_{j=l+1}^{N_h}C_{j}^{\delta}}(\psi^1)\right).
\end{align*}
Thus, the induction concludes that the equality
 \eqref{eq_desired_equality_UV} holds for all $l\in
 \{0,1,\cdots,N_h-1\}$ on the assumption
 \eqref{eq_coupling_constant_temporary_small_next}.

By Lemma \ref{lem_exponential_appl_bound}
 \eqref{item_exponential_appl_0th} there exists a $(\beta,L)$-dependent,
 $h$-independent constant $c'\in \R_{>0}$ such that if
 $\sup_{\rho\in\cB}|U_{\rho}|\le c'$, 
$$
\Re \int e^{-V^{\delta}(\psi)}d\mu_{C_{>0}^{\delta}}(\psi)>0,
$$
and thus the Grassmann polynomial 
$$
\log\left(\int e^{-V^{\delta}(\psi+\psi^1)}d\mu_{C_{>0}^{\delta}}(\psi^1)\right)
$$
is analytic with $(U_1,U_2,\cdots,U_b)$ in
 $\{(U_1,U_2,\cdots,U_b)\in\C^b\ |\ |U_{\rho}|<c',$ $(\forall \rho\in\cB)\}$. 
On the other hand, by the claim \eqref{item_UV_analyticity} for $l=0$
 and taking the $h$-independent constant $c'$ smaller if necessary we
 see that $J^{\delta,0}(\psi)$ is analytic
 in $\{(U_1,U_2,\cdots,U_b)\in\C^b\ |\ |U_{\rho}|<c',\
 (\forall \rho\in\cB)\}$.
Therefore, the identity theorem ensures that the equality
 \eqref{eq_desired_equality_UV} for $l=0$ holds for any
 $(U_1,U_2,\cdots,U_b)\in\C^b$ satisfying $|U_{\rho}|<c'$ $(\forall
 \rho\in\cB)$.
\end{proof}

\section{The infrared integration of the model}\label{sec_IR_model}

Here we start the infrared analysis of the free energy density defined
in Subsection \ref{subsec_model}. As we saw in Remark
\ref{rem_phase_freedom}, the free energy density is independent of how
to choose the argument $\theta_L(\cdot,\cdot):\Z^2\times \Z^2\to \R$
satisfying \eqref{eq_phase_condition} and \eqref{eq_flux_pi}. Therefore,
let us focus on the model Hamiltonian with the argument $\theta_L$
simply defined by \eqref{eq_phase_example}. The periodic properties of
the hopping amplitude and the magnitude of the on-site coupling enable us to redefine
the free energy density as that governed by a 4-band Hamiltonian, whose
hopping amplitude and coupling constants are no longer dependent on the
position vector but on the band index. This 4-band many-electron system can be analyzed by means
of the general estimations constructed so far. The essential ingredient of the IR
integration process considered in this section is an extension of
Giuliani-Mastropietro's RG method designed for the 2-band Hamiltonian
(\cite{GM}). As in Giuliani-Mastropietro's RG we make use of 
the symmetries of Grassmann polynomials to show that covariances for the IR integration have good
bound properties. Then, applying the framework developed in Subsection
\ref{subsec_IR_general} and Subsection
\ref{subsec_IR_general_difference}, we will move on to the proof of 
Theorem \ref{thm_main_theorem}.

\subsection{The four-band formulation}\label{subsec_4_band}
Let us set up a 4-band Hamiltonian whose free energy density is equal to
that considered in Theorem \ref{thm_main_theorem}. From now we assume that $d=2$, $\bu_1=\bv_1=\be_1(=(1,0))$, $\bu_2=\bv_2=\be_2(=(0,1))$ so that
\begin{align*}
&\G=\left\{\sum_{j=1}^2m_j\be_j\ \Big|\ m_j\in\{0,1,\cdots,L-1\}\
 (j=1,2)\right\},\\
&\G^*=\left\{\frac{2\pi}{L}\sum_{j=1}^2m_j\be_j\ \Big|\ m_j\in\{0,1,\cdots,L-1\}\
 (j=1,2)\right\}.
\end{align*}
Since we are going to define a 4-band model, we assume that $b=4$ and 
$\cB=\{1,2,3,4\}$. The crystal lattice in a box is identified
with $\cB\times \G$. In the case $L=2$, the lattice $\cB\times \G$ can
be pictured as in Figure \ref{fig_4_band}.
 \begin{figure}
\begin{center}
\begin{picture}(120,120)(0,0)
\put(15,15){$\bullet$}
\put(75,15){$\bullet$}
\put(15,75){$\bullet$}
\put(75,75){$\bullet$}

\put(45,15){$\circ$}
\put(105,15){$\circ$}
\put(45,75){$\circ$}
\put(105,75){$\circ$}

\put(15,45){$\diamond$}
\put(75,45){$\diamond$}
\put(15,105){$\diamond$}
\put(75,105){$\diamond$}

\put(45,45){$\star$}
\put(105,45){$\star$}
\put(45,105){$\star$}
\put(105,105){$\star$}
\end{picture}
 \caption{The lattice $\cB\times \G$ for $L=2$, where the symbol ``$\bullet$''
 denotes the sites of $\{1\}\times\G$, the symbol ``$\circ$''
 denotes the sites of $\{2\}\times\G$, the symbol ``$\diamond$''
 denotes the sites of $\{3\}\times\G$, the symbol ``$\star$''
 denotes the sites of $\{4\}\times\G$. }\label{fig_4_band}
\end{center}
\end{figure}
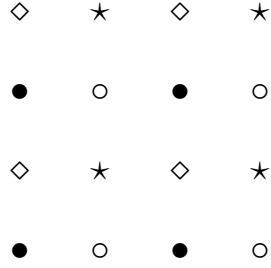
For any $\bx\in \G$ we assume that the site $(2,\bx)$ is right to the
site $(1,\bx)$, the site $(3,\bx)$ is above $(1,\bx)$, and the site
$(4,\bx)$ is right to the site $(3,\bx)$, as described in Figure \ref{fig_site_arrangement}.
\begin{figure}
\begin{center}
\begin{picture}(55,55)(0,0)
\put(5,15){$\bullet$}
\put(0,0){$(1,\bx)$}
\put(45,15){$\circ$}
\put(40,0){$(2,\bx)$}
\put(5,55){$\diamond$}
\put(0,40){$(3,\bx)$}
\put(45,55){$\star$}
\put(40,40){$(4,\bx)$}
\end{picture}
 \caption{The arrangement of 4 sites $(1,\bx),(2,\bx),(3,\bx),(4,\bx)$.}\label{fig_site_arrangement}
\end{center}
\end{figure}
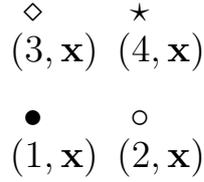

With the parameters $t_{h,e}$, $t_{h,o}$, $t_{v,e}$, $t_{v,o}\in\R_{>0}$
we define $E(\cdot):\R^2\to \Mat(4,\C)$ by
\begin{align}
&E(\bk)\label{eq_specific_dispersion_relation}\\
&:=\left(\begin{array}{cccc}0 & t_{h,e}(1+e^{-ik_1}) &
	 t_{v,e}(1+e^{-ik_2}) & 0 \\ 
 t_{h,e}(1+e^{ik_1}) & 0 & 0 &  -t_{v,o}(1+e^{-ik_2}) \\
t_{v,e}(1+e^{ik_2}) & 0 & 0 &  t_{h,o}(1+e^{-ik_1}) \\
 0 &  -t_{v,o}(1+e^{ik_2}) & t_{h,o}(1+e^{ik_1}) & 0 \end{array}
\right).\notag
\end{align}
We can see that $E(\cdot)$ satisfies \eqref{eq_dispersion_hermite} and
\eqref{eq_dispersion_periodic}. With this $E(\cdot)$ we define the free
Hamiltonian $H_0$ by \eqref{eq_free_hamiltonian_general}. 
For notational
convenience, set $U_1:=U_{e,e}$, $U_2:=U_{o,e}$, $U_3:=U_{e,o}$,
$U_4:=U_{o,o}$ with the parameters $U_{e,e}$, $U_{o,e}$, $U_{e,o}$,
$U_{o,o}\in\R$ introduced in Subsection \ref{subsec_model}. We define
the interacting part $V$ of the Hamiltonian by
\eqref{eq_interacting_hamiltonian_general}. Then, we define the Hamiltonian
$H:F_f(L^2(\cB\times \G\times \spin))\to F_f(L^2(\cB\times \G\times
\spin))$ by $H:=H_0+V$. 
\begin{lemma}\label{lem_free_energy_equivalence}
The quantity 
$$
-\frac{1}{\beta(2L)^2}\log(\Tr e^{-\beta H})
$$
derived from this Hamiltonian by the trace operation over
 $F_f(L^2(\cB\times \G\times \spin))$ is equal to the free energy
 density 
$$
-\frac{1}{\beta(2L)^2}\log(\Tr e^{-\beta \sH})
$$
considered in Theorem \ref{thm_main_theorem}.
\end{lemma}
\begin{proof}
For $\rho \in\cB$ set 
\begin{align}
\be(\rho):=\left\{\begin{array}{ll} \b0,&\text{if }\rho=1,\\
                                    \be_1,&\text{if }\rho=2,\\
                                    \be_2,&\text{if }\rho=3,\\
                                    \be_1+\be_2,&\text{if }\rho=4.
\end{array}               
\right.\label{eq_key_vectors}
\end{align}
Then, define the linear map $G:F_f(L^2(\cB\times \G\times \spin))\to 
F_f(L^2(\G(2L)\times \spin))$ by 
\begin{align*}
&G\O:=\O_{2L},\\
&G(\psi_{\rho_1\bx_1\s_1}^*\psi_{\rho_2\bx_2\s_2}^*\cdots
 \psi_{\rho_n\bx_n\s_n}^*\O):=
 \psi_{2\bx_1+\be(\rho_1)\s_1}^*\psi_{2\bx_2+\be(\rho_2)\s_2}^*\cdots
 \psi_{2\bx_n+\be(\rho_n)\s_n}^*\O_{2L},\\
&(\forall
 (\rho_j,\bx_j,\s_j)\in\cB\times \G\times \spin\ (j=1,2,\cdots,n)),
\end{align*}
and by linearity. We can check that the map $G$ is unitary. By
 definition,
\begin{align}\label{eq_free_hamiltonian_another_form}
H_0&=\sum_{(\bx,\s)\in\G\times\spin}(t_{h,e}\psi_{1,\bx\s}^*(\psi_{2,\bx\s}+\psi_{2,\bx-\be_1\s})
+t_{v,e}\psi_{1,\bx\s}^*(\psi_{3,\bx\s}+\psi_{3,\bx-\be_2\s})\\
&\qquad\quad-t_{v,o}\psi_{2,\bx\s}^*(\psi_{4,\bx\s}+\psi_{4,\bx-\be_2\s})
+t_{h,o}\psi_{3,\bx\s}^*(\psi_{4,\bx\s}+\psi_{4,\bx-\be_1\s}))
+\text{h.c}, \notag
\end{align}
where the notation `h.c'
 means that the adjoint operator of the operator in front is
 placed. Note that
\begin{align*}
&GH_0G^*\\
&=\sum_{(\bx,\s)\in\G\times\spin}(t_{h,e}\psi_{2\bx\s}^*(\psi_{2\bx+\be_1\s}+\psi_{2\bx-\be_1\s})
+t_{v,e}\psi_{2\bx\s}^*(\psi_{2\bx+\be_2\s}+\psi_{2\bx-\be_2\s})\\
&\qquad\qquad-t_{v,o}\psi_{2\bx+\be_1\s}^*(\psi_{2\bx+\be_1+\be_2\s}+\psi_{2\bx+\be_1-\be_2\s})\\
&\qquad\qquad +t_{h,o}\psi_{2\bx+\be_2\s}^*(\psi_{2\bx+\be_1+\be_2\s}+\psi_{2\bx-\be_1+\be_2\s}))
+\text{h.c}\\
&=\mathsf{H}_0,
\end{align*}
where $\sH_0$ is the operator defined in \eqref{eq_free_part_original}
 with the phase $\theta_L$ defined in \eqref{eq_phase_example}. Similarly we
 can confirm that $GVG^*=\sV$, which was defined in
 \eqref{eq_interacting_part_original}. Thus, we have that
$\Tr e^{-\beta H}=\Tr e^{-\beta GHG^*}=\Tr e^{-\beta \sH}$
with the Hamiltonian $\sH$ containing the phase \eqref{eq_phase_example}. Then,
 the claim follows from Remark \ref{rem_phase_freedom}.
\end{proof}

Lemma \ref{lem_free_energy_equivalence} tells us that it suffices to prove the
same statements as in Theorem \ref{thm_main_theorem} for 
$$
-\frac{1}{\beta L^2}\log(\Tr e^{-\beta H})
$$  
with the Hamiltonian $H$ defined above. Moreover, we will later confirm that
the claims of Theorem \ref{thm_main_theorem} follow from the theorem
proved under the assumption 
\begin{align}\label{eq_hopping_amplitude_normalized}
\max\{ t_{h,e}, t_{h,o}, t_{v,e}, t_{v,o}\}=1. 
\end{align}
Thus, from now we assume \eqref{eq_hopping_amplitude_normalized} unless
otherwise stated.

\subsection{The cut-off function for the infrared integration}
\label{subsec_cut_off_IR}

Here we define a cut-off function whose support covers the zero set of
the free dispersion relation. In order to choose such a cut-off function
correctly, let us study properties of $E$ first.
\begin{lemma}\label{lem_estimate_free_dispersion_relation}
The following inequalities hold.
\begin{enumerate}
\item\label{item_derivative_upper_bound_dispersion}
\begin{align*}
\left\|\left(\frac{\partial}{\partial
 k_j}\right)^nE(\bk)\right\|_{4\times 4}\le 4,\ (\forall
 \bk\in\R^2,n\in\N\cup \{0\},j\in\{1,2\}).
\end{align*}
\item\label{item_upper_bound_free_propagator}
\begin{align*}
\|(i\o I_4-E(\bk))^{-1}\|_{4\times 4}\le
 \Bigg(\o^2+f_{\bt}\sum_{j=1}^2(1+\cos k_j)\Bigg)^{-\frac{1}{2}},\ (\forall
 (\o,\bk)\in\R^3),
\end{align*}
where $f_{\bt}$ is the quantity defined in \eqref{eq_hopping_amplitude_function}.
\end{enumerate}
\end{lemma}
\begin{proof}
For any $\bk\in \R^2$, $p,q\in \{1,-1\}$ set
\begin{align}
&X_{p,q}(\bk):=p\left(A(\bk)+q\sqrt{A(\bk)^2-4B(\bk)^2}\right)^{\frac{1}{2}},\label{eq_explicit_free_dispersion_relation}\\
&A(\bk):=(t_{h,e}^2+t_{h,o}^2)(1+\cos
 k_1)+(t_{v,e}^2+t_{v,o}^2)(1+\cos k_2),\notag\\
&B(\bk):=t_{h,e}t_{h,o}(1+\cos
 k_1)+t_{v,e}t_{v,o}(1+\cos k_2).\notag
\end{align}
A calculation shows that the eigen values of $E(\bk)$ are
 $X_{p,q}(\bk)$ $(p,q\in\{1,-1\})$. One can also check that the eigen
 values of $(\partial/\partial k_1)^nE(\bk)$ are $t_{h,e}$, $-t_{h,e}$,
$t_{h,o}$, $-t_{h,o}$ and the eigen
 values of $(\partial/\partial k_2)^nE(\bk)$ are $t_{v,e}$, $-t_{v,e}$,
$t_{v,o}$, $-t_{v,o}$ for any $n\in\N$.

\eqref{item_derivative_upper_bound_dispersion}:
Using the assumption \eqref{eq_hopping_amplitude_normalized}, we have
 that 
\begin{align*}
&\|E(\bk)\|_{4\times 4}\le \max_{p,q\in \{1,-1\}}|X_{p,q}(\bk)|\le
 \sqrt{2}|A(\bk)|^{\frac{1}{2}}\le 4,\\
&\left\|\left(\frac{\partial}{\partial
 k_j}\right)^nE(\bk)\right\|_{4\times 4}\le 1,\ (\forall j\in\{1,2\},n\in\N).
\end{align*}

\eqref{item_upper_bound_free_propagator}:
Set 
$$
s:=\max\left\{\frac{t_{h,o}}{t_{h,e}}+\frac{t_{h,e}}{t_{h,o}},
             \frac{t_{v,o}}{t_{v,e}}+\frac{t_{v,e}}{t_{v,o}}\right\}.
$$
Since $A(\bk)\le s B(\bk)$, for any $p,q\in \{1,-1\}$,
\begin{align*}
|X_{p,q}(\bk)|&\ge \left|A(\bk)-\sqrt{A(\bk)^2-4B(\bk)^2}\right|^{\frac{1}{2}}
\\
&\ge
 \left|sB(\bk)-\sqrt{s^2B(\bk)^2-4B(\bk)^2}\right|^{\frac{1}{2}}=(s-\sqrt{s^2-4})^{\frac{1}{2}}B(\bk)^{\frac{1}{2}}\\
&\ge
 (s-\sqrt{s^2-4})^{\frac{1}{2}}(\min\{t_{h,e}t_{h,o},t_{v,e}t_{v,o}\})^{\frac{1}{2}}\Bigg(\sum_{j=1}^2(1+\cos k_j)\Bigg)^{\frac{1}{2}}.
\end{align*}
Since 
$$
s-\sqrt{s^2-4}\ge
 \min\left\{\frac{t_{h,o}}{t_{h,e}},\frac{t_{h,e}}{t_{h,o}},
\frac{t_{v,o}}{t_{v,e}},\frac{t_{v,e}}{t_{v,o}}\right\},
$$
$$
|X_{p,q}(\bk)|\ge f_{\bt}^{\frac{1}{2}}\left(\sum_{j=1}^2(1+\cos k_j)
\right)^{\frac{1}{2}},\ (\forall p,q\in \{1,-1\}).
$$
Thus,
\begin{align*}
\|(i\o I_4-E(\bk))^{-1}\|_{4\times 4}&\le \max_{p,q\in
 \{1,-1\}}|i\o-X_{p,q}(\bk)|^{-1}\\
&\le \left(\o^2+f_{\bt}\sum_{j=1}^2(1+\cos
 k_j)\right)^{-\frac{1}{2}}.
\end{align*}
\end{proof}

Lemma \ref{lem_estimate_free_dispersion_relation}
\eqref{item_derivative_upper_bound_dispersion} implies that the
inequality \eqref{eq_real_analytic_dispersion_relation} holds with
$E_1=4$, $E_2=1$. In this section we will apply the results of Section
\ref{sec_UV} for $E_1=4$, $E_2=1$.
It follows that 
$$
M_{UV}=\frac{10\sqrt{6}}{\pi}
$$
and the weight $\fw(0)$ originally set in Lemma
\ref{lem_covariance_UV_bound} satisfies 
\begin{align}
\fw(0)=\frac{c_w}{18}M^{-2}.\label{eq_initial_weight}
\end{align}

In order to adjust the support of cut-off functions for the IR
integration, from now we assume that 
$$
M>\sqrt{2}.
$$
Let us set
$$
M_{IR}:=\frac{\sqrt{6}}{\pi}\left(\frac{\pi^2}{3}M_{UV}^2+4\right)^{\frac{1}{2}}
$$
and 
$$
N_{\beta}:=\min\left\{\left[\frac{\log\big(\frac{\pi}{\beta}\big(\frac{\pi}{\sqrt{3}}M_{IR}\big)^{-1}\big)}{\log
M}\right],0\right\}.
$$
Since $(\pi/\sqrt{3})M_{IR}M^{N_{\beta}}\le \pi/\beta$, 
\begin{align}
&\phi\left(M_{IR}^{-2}M^{-2N_{\beta}}\left(\o^2+f_{\bt}\sum_{j=1}^2(1+\cos
 k_j)\right)\right)=0,\label{eq_infrared_cut_off_smallest}\\
& (\forall (\o,\bk)\in\R^3\text{ with }|\o|\ge \pi/\beta),\notag
\end{align}
where $\phi$ is the smooth function introduced in Lemma
\ref{lem_gevrey_class}. We
define the functions $\chi_l:\R^3\to\R$ $(l\in \{0,-1,\cdots,N_{\beta}\})$ by 
\begin{align*}
&\chi_l(\o,\bk)\\
&:=\phi(M_{UV}^{-2}\o^2)\Bigg(\phi\left(M_{IR}^{-2}M^{-2(l+1)}\left(\o^2+f_{\bt}\sum_{j=1}^2(1+\cos
 k_j)\right)\right)\\
&\qquad-\phi\left(M_{IR}^{-2}M^{-2l}\left(\o^2+f_{\bt}\sum_{j=1}^2(1+\cos
 k_j)\right)\right)\Bigg),\ ((\o,\bk)\in\R^3).\\
\end{align*}
If $\phi(M_{UV}^{-2}\o^2)\neq 0$,
$$
\o^2+f_{\bt}\sum_{j=1}^2(1+\cos k_j)\le \frac{\pi^2}{3}M_{UV}^{2}+4f_{\bt}\le  \frac{\pi^2}{3}M_{UV}^{2}+4,
$$
and thus
\begin{align}
&\phi\left(M_{IR}^{-2}M^{-2}\left(\o^2+f_{\bt}\sum_{j=1}^2(1+\cos
 k_j)\right)\right)=1,\label{eq_infrared_cut_off_largest}\\
&(\forall \o\in
 \R\text{ with }\phi(M_{UV}^{-2}\o^2)\neq
 0,\bk\in\R^2).\notag
\end{align}
It follows from \eqref{eq_infrared_cut_off_smallest},
\eqref{eq_infrared_cut_off_largest} that
\begin{align}
\sum_{l=0}^{N_{\beta}}\chi_l(\o,\bk)=\phi(M_{UV}^{-2}\o^2),\ (\forall
\o\in \cM,\bk\in \R^2).\label{eq_cut_off_function_basic_sum}
\end{align}
The value of $\chi_l(\o,\bk)$ is described as follows.
\begin{align}
&\chi_l(\o,\bk)\left\{\begin{array}{ll}=0,&\text{if
	       }\left(\o^2+f_{\bt}\sum_{j=1}^2(1+\cos
		 k_j)\right)^{\frac{1}{2}}\le
	       \frac{\pi}{\sqrt{6}}M_{IR}M^l,\\
\in [0,1],&\text{if
	       } \frac{\pi}{\sqrt{6}}M_{IR}M^l<
	       \left(\o^2+f_{\bt}\sum_{j=1}^2(1+\cos
		k_j)\right)^{\frac{1}{2}}\\
&\qquad\qquad\quad\ < \frac{\pi}{\sqrt{3}}M_{IR}M^{l+1},\\
=0,&\text{if
	       }\left(\o^2+f_{\bt}\sum_{j=1}^2(1+\cos
		 k_j)\right)^{\frac{1}{2}}\ge
		     \frac{\pi}{\sqrt{3}}M_{IR}M^{l+1},\end{array}
\right.\label{eq_support_properties_infrared}\\
&(\forall l\in \{0,-1,\cdots,N_{\beta}\}, (\o,\bk)\in\R^3).\notag
\end{align}
Since $M>\sqrt{2}$,
$(\pi/\sqrt{3})M_{IR}M^{l-1}<(\pi/\sqrt{6})M_{IR}M^{l}$.
This inequality implies that
\begin{align}
&\{(\o,\bk)\in\R^3\ |\ \chi_l(\o,\bk)\neq 0\}\cap\{(\o,\bk)\in\R^3\ |\
 \chi_j(\o,\bk)\neq
 0\}=\emptyset,\label{eq_infrared_cut_off_intersection}\\
&(\forall j,l\in \{0,-1,\cdots,N_{\beta}\}\text{ with }|j-l|\ge 2).\notag
\end{align}
We use $\chi_l:\R^3\to\R$ $(l=0,-1,\cdots,N_{\beta})$ as the cut-off
functions in the IR integration. For any $l\in
\{0,-1,\cdots,N_{\beta}\}$ set 
\begin{align*}
&\chi_{\le l}(\o,\bk):=\sum_{j=l}^{N_{\beta}}\chi_j(\o,\bk),\\
&\hat{\chi}_{\le
 l}(\o,\bk):=\phi(M_{UV}^{-2}\o^2)\phi\left(M_{IR}^{-2}M^{-2(l+1)}\left(\o^2+f_{\bt}\sum_{j=1}^2(1+\cos k_j)\right)\right),\\
&(\forall (\o,\bk)\in\R^3).
\end{align*}
Note that $\supp \chi_l(\cdot)\subset \supp \chi_{\le
l}(\cdot)\subset\supp\hat{\chi}_{\le l}(\cdot)$. Concerning the support of
these cut-off functions, we will frequently use
the following lemma.
\begin{lemma}\label{lem_infrared_cut_off_measure}
Let $l\in \{0,1,\cdots,N_{\beta}\}$. If $(\o,\bk)\in\R\times [0,2\pi]^2$
 satisfies $\hat{\chi}_{\le l}(\o,\bk)\neq 0$, then,
$$
|\o|\le \frac{\pi}{\sqrt{3}}M_{IR}M^{l+1},\ |k_j-\pi|\le
 \frac{\pi^2}{\sqrt{6}}f_{\bt}^{-\frac{1}{2}}M_{IR}M^{l+1},\ (j=1,2).
$$
\end{lemma}
\begin{proof}
If $\hat{\chi}_{\le l}(\o,\bk)\neq 0$, 
$$
\left(\o^2+f_{\bt}\sum_{j=1}^2(1+\cos k_j)\right)^{\frac{1}{2}}\le \frac{\pi}{\sqrt{3}}M_{IR}M^{l+1}.
$$
Then, by using the inequality $\sqrt{1-\cos\theta}\ge
 (\sqrt{2}/\pi)|\theta|$ $(\forall \theta\in[-\pi,\pi])$ we can derive
 the claimed inequalities.
\end{proof}

In order to indicate the dependency on $\beta$, we will sometimes write
$\chi_{\le l}(\beta)$ instead of $\chi_{\le l}$. By
\eqref{eq_infrared_cut_off_smallest} we see that 
\begin{align}
&\chi_{\le l}(\beta)(\o,\bk)=\hat{\chi}_{\le
 l}(\o,\bk),\label{eq_equivalence_IR_cut_off}\\ 
&(\forall
 (\o,\bk)\in\R^3\text{ with }|\o|\ge \pi/\beta,l\in
 \{0,-1,\cdots,N_{\beta}\}),\notag\\
&\chi_{\le l}(\beta_1)(\o,\bk)=\chi_{\le l}(\beta_2)(\o,\bk),\notag\\
& (\forall
 (\o,\bk)\in\R^3\text{ with }|\o|\ge \pi/\beta_1,l\in
 \{0,-1,\cdots,N_{\beta_1}\}),\notag
\end{align}
if $\beta_1\le \beta_2$.

We define the weights $\fw(l)$ $(l\in\Z_{\le 0})$ by
$$
\fw(l):=\fw(0)M^l,\ (\forall l\in \Z_{\le 0}),
$$
with the weight $\fw(0)$ characterized in \eqref{eq_initial_weight}. 
Moreover, we take the exponent $\fr$ inside $\|\cdot\|_{l,0}$,
$\|\cdot\|_{l,1}$, $|\cdot-\cdot|_{l}$ to be $1/2$ throughout this section.
To organize formulas systematically, for any differentiable function $f$ with the
variable $(w,k_1,k_2)$ let $(\partial/\partial k_0 )f$ denote $(\partial/\partial \o)f$.  

\begin{lemma}\label{lem_infrared_cut_off_derivative}
There exists a constant $c\in \R_{>0}$ independent of any parameter such
 that 
\begin{align*}
&\left|\left(\frac{\partial}{\partial k_j}\right)^n\hat{\chi}_{\le
 l}(\o,\bk)\right|,\ \left|\left(\frac{\partial}{\partial
 k_j}\right)^n\chi_{ l}(\o,\bk)\right|\le (c \fw(l)^{-1})^n(n!)^2,\\
&(\forall (\o,\bk)\in\R^3,n\in \N\cup \{0\},j\in\{0,1,2\},l\in
 \{0,-1,\cdots,N_{\beta}\}).
\end{align*}
\end{lemma}
\begin{proof}
Using the inequality $f_{\bt}\le
 1$, we see that for any $n\in\N$,
\begin{align*}
&\left|\left(\frac{d}{d\o}\right)^nM_{UV}^{-2}\o^2\right|\le c n!, \
 (\forall \o\in\R\text{ with }\phi(M_{UV}^{-2}\o^2)\neq 0),\\
&\left|\left(\frac{\partial}{\partial\o}\right)^nM_{IR}^{-2}M^{-2l}\left(\o^2+f_{\bt}\sum_{i=1}^2(1+\cos
 k_i)\right)\right|\\
&\le \left\{\begin{array}{ll} c M^{-l+1},&\text{if
}n=1,\\
 c M^{-2l},&\text{if }n=2,\\
0,&\text{if }n\ge 3,\end{array}
\right.
\le cM\cdot M^{-ln}n!,\\
&\left|\left(\frac{\partial}{\partial k_j}\right)^nM_{IR}^{-2}M^{-2l}\left(\o^2+f_{\bt}\sum_{i=1}^2(1+\cos
 k_i)\right)\right|\\
&\le \left\{\begin{array}{ll} c M^{-l+1},&\text{if
}n=1,\\
 c M^{-2l},&\text{if }n\ge 2,\end{array}
\right.
\le cM\cdot M^{-ln}n!,\\
&(\forall (\o,\bk)\in\R^3\text{ satisfying }\\
&\qquad\qquad\phi\left(M_{IR}^{-2}M^{-2(l+1)}\left(\o^2+f_{\bt}\sum_{i=1}^2(1+\cos
 k_i)\right)\right)\neq 0,\\
&\ \forall j\in\{1,2\}).
\end{align*}
Thus, by \eqref{eq_derivative_gevrey_cut_off}, Lemma
 \ref{lem_gevrey_composition} and the inequality $M\ge 1$,
\begin{align*}
&\left|\left(\frac{d}{d\o}\right)^n\phi(M_{UV}^{-2}\o^2)\right|\le c^n
 (n!)^2, \\
&\left|\left(\frac{\partial}{\partial k_j}\right)^n\phi\left(M_{IR}^{-2}M^{-2l}\left(\o^2+f_{\bt}\sum_{i=1}^2(1+\cos
 k_i)\right)\right)\right|,\\
&\left|\left(\frac{\partial}{\partial k_j}\right)^n\phi\left(M_{IR}^{-2}M^{-2(l+1)}\left(\o^2+f_{\bt}\sum_{i=1}^2(1+\cos
 k_i)\right)\right)\right|\\
&\le (c M^{-l+1})^n(n!)^2,\ (\forall (\o,\bk)\in\R^3,j\in
 \{0,1,2\},n\in\N\cup \{0\}).
\end{align*}
By the condition $c_w\in (0,1]$, $M^{-l+1}\le \fw(l)^{-1}$. Using these inequalities, we can deduce that for any
 $(\o,\bk)\in\R^3$,
\begin{align*}
&\left|\left(\frac{\partial}{\partial k_j}\right)^n\hat{\chi}_{\le
 l}(\o,\bk)\right|\\
&\le \sum_{m=0}^n\left(\begin{array}{c}n\\
		       m\end{array}\right)\left|\left(\frac{\partial}{\partial k_j}\right)^m\phi(M_{UV}^{-2}\o^2)\right|\\
&\qquad\cdot
\left|\left(\frac{\partial}{\partial
 k_j}\right)^{n-m}\phi\left(M_{IR}^{-2}M^{-2(l+1)}\left(\o^2+f_{\bt}\sum_{j=1}^2(1+\cos
 k_j)\right)\right)\right|\\
&\le \sum_{m=0}^n\left(\begin{array}{c}n\\
		       m\end{array}\right) c^m (m!)^2(c
 M^{-l+1})^{n-m}((n-m)!)^2\le (c\fw(l)^{-1})^n(n!)^2.
\end{align*}
The upper bound on $|(\partial/\partial k_j)^n\chi_l(\o,\bk)|$ can be
 derived similarly.
\end{proof}

By definition we have $(\pi/\sqrt{3})M_{IR}M^{N_{\beta}}\le
\pi/\beta$. The next lemma suggests that an opposite inequality is also
available when we deal with covariances for the IR integration. 
\begin{lemma}\label{lem_beta_inverse_upper_bound}
If $\chi_{\le 0}(\o,\bk)\neq 0$ for some $(\o,\bk)\in\R^3$ with $|\o|\ge
 \pi/\beta$, 
$$
\frac{1}{\beta}\le M_{IR}M^{N_{\beta}+1}.
$$
\end{lemma}
\begin{proof}
By assumption $\phi(M_{UV}^{-2}(\pi/\beta)^2)\neq 0$, which
 implies that $1/\beta\le M_{UV}/\sqrt{3}\le M_{IR}/\sqrt{6}$. Thus, if
 $N_{\beta}=0$, the claimed inequality holds. If $N_{\beta}< 0$,
$$
N_{\beta}\ge
 \frac{\log\big(\frac{\pi}{\beta}\big(\frac{\pi}{\sqrt{3}}M_{IR}\big)^{-1}\big)}{\log M}-1,
$$
which implies that $1/\beta\le (1/\sqrt{3})M_{IR}M^{N_{\beta}+1}$.
\end{proof}

\subsection{The covariance matrices in the infrared integration}
\label{subsec_covariance_IR}
Following the infrared integration scheme proposed in \cite{P}, 
\cite{BGM}, \cite{GM}, we update the covariance by inserting the kernel
of the quadratic Grassmann polynomial produced by the previous integration
 at every integration step. To simulate this procedure, we introduce a family of subsets of $\bigwedge \cV$ consisting of
polynomials
 satisfying certain bound properties and invariant properties. Then, we
 define a prototypical covariance
by substituting the kernel of a polynomial belonging to one of
these subsets and study its properties.

Let $c_{IR}$, $\alpha \in \R_{> 0}$ and $D(\subset \C^4)$
be a domain satisfying that $\overline{\bU}\in \overline{D}$ $(\forall
\bU\in\overline{D})$, where 
$\overline{\bU}$ is the complex conjugate of $\bU$ and $\overline{D}$ is
the closure of $D$.  For any $l\in \Z_{\le 0}$ we
define the subset $\cS(l)$ of $\bigwedge \cV$ as follows. A Grassmann
polynomial $J(\psi)$ $(\in \bigwedge \cV)$ belongs to $\cS(l)$ if and
only if $J(\psi)$ is parameterized by $\bU\in\overline{D}$ and satisfies
the conditions \eqref{item_IR_set_analytic}, \eqref{item_IR_set_bound},
\eqref{item_IR_set_invariance}, \eqref{item_IR_set_complex_invariance}.  
\begin{enumerate}[(i)]
\item\label{item_IR_set_analytic}
$J(\bU)(\psi)$ is continuous in $\overline{D}$ and analytic
     in $D$ with $\bU$.
\item\label{item_IR_set_bound}
\begin{align}
&\frac{h}{N}|J_0|\le
 M^{\frac{7}{2}l}\alpha^{-3},\ (\forall \bU\in\overline{D}),\label{eq_IR_0th_bound}\\
&M^{-\frac{7}{2}l+rl}\sum_{m=2}^Nc_{IR}^{\frac{m}{2}}M^{ml}\alpha^m\|J_m\|_{l,r}\le 1,\ (\forall \bU\in\overline{D},r\in \{0,1\}).\label{eq_IR_higher_order_bound}
\end{align}
\item\label{item_IR_set_invariance}
With the notation introduced in Subsection 
\ref{subsec_invariance_general}, 
$$
J(\bU)(\psi)=J(\bU)(\cR\psi),\ (\forall \bU\in \overline{D}),
$$
for all $S:I\to I$ and $Q:I\to \R$ defined as follows.
\begin{align}
&S((\rho,\bx,\s,x,\theta)):=(\rho,\bx,\s,x,\theta),\label{eq_particle_hole_maps}\\
&Q((\rho,\bx,\s,x,\theta)):=\frac{\pi}{2}\theta,\ (\forall
 (\rho,\bx,\s,x,\theta)\in I).\notag
\end{align}
\begin{align}
&S((\rho,\bx,\s,x,\theta)):=(\rho,\bx,\s,x,\theta),\label{eq_spin_maps}\\
&Q((\rho,\bx,\s,x,\theta)):= \pi 1_{\s=\ua },\ (\forall
 (\rho,\bx,\s,x,\theta)\in I).\notag
\end{align}
\begin{align}
&S((\rho,\bx,\s,x,\theta)):=(\rho,\bx,-\s,x,\theta),\label{eq_spin_reflection_maps}\\
&Q((\rho,\bx,\s,x,\theta)):= 0,\ (\forall
 (\rho,\bx,\s,x,\theta)\in I).\notag
\end{align}
\begin{align}
&S((\rho,\bx,\s,x,\theta)):=(\rho,r_L(\bx+\bz),\s,r_{\beta}(x+s),\theta),\label{eq_translation_maps}\\
&Q((\rho,\bx,\s,x,\theta)):= \pi n_{\beta}(r_{\beta}(x-s)+s),\ (\forall
 (\rho,\bx,\s,x,\theta)\in I),\notag
\end{align}
where $\bz\in \Z^2$ and $s\in (1/h)\Z$ are arbitrarily taken and fixed.
\begin{align}
&S((\rho,\bx,\s,x,\theta)):=(\rho,r_L(-\bx-\be(\rho)),\s,x,\theta),\label{eq_rotation_maps}\\
&Q((\rho,\bx,\s,x,\theta)):= 0,\ (\forall
 (\rho,\bx,\s,x,\theta)\in I),\notag
\end{align}
where $\be(\rho)$ $(\rho\in \cB)$ are the vectors defined in \eqref{eq_key_vectors}.
\item\label{item_IR_set_complex_invariance}
$$
J(\bU)(\psi)=\overline{J(\overline{\bU})}(\cR\psi),\ (\forall \bU\in \overline{D}),
$$
for all $S:I\to I$ and $Q:I\to \R$ defined as follows. 
\begin{align}
&S((\rho,\bx,\s,x,\theta)):=(\rho,\bx,\s,r_{\beta}(-x),-\theta),\label{eq_hermitian_maps}\\
&Q((\rho,\bx,\s,x,\theta)):= \pi (1_{\theta=1}+1_{x\neq 0}),\ (\forall
 (\rho,\bx,\s,x,\theta)\in I).\notag
\end{align}
\begin{align}
&S((\rho,\bx,\s,x,\theta)):=(\rho,\bx,\s,x,-\theta),\label{eq_half_filled_maps}\\
&Q((\rho,\bx,\s,x,\theta)):= \pi 1_{\rho\in \{1,4\}},\ (\forall
 (\rho,\bx,\s,x,\theta)\in I).\notag
\end{align}
\end{enumerate}

We write $\cS(l)(\beta)$ in place of $\cS(l)$ when we want to indicate
the dependency on $\beta$. On the assumption
\eqref{eq_beta_h_assumption} we define the subset $\tilde{\cS}(l)$ of
$\cS(l)(\beta_1) \times \cS(l)(\beta_2)$ as follows. A pair of Grassmann polynomials 
$(J(\beta_1)(\psi), J(\beta_2)(\psi))\in \cS(l)(\beta_1)\times\cS(l)(\beta_2)$
belongs to $\tilde{\cS}(l)$ if and only if 
\begin{align}
&\left|\frac{h}{N(\beta_1)}J_0(\beta_1)-\frac{h}{N(\beta_2)}J_0(\beta_2)\right|
\le \beta_1^{-\frac{1}{2}}M^{\frac{5}{2}l}\alpha^{-3},\ (\forall
 \bU\in\overline{D}).\label{eq_IR_0th_difference}\\
&M^{-\frac{5}{2}l}\sum_{m=2}^{N(\beta_2)}c_{IR}^{\frac{m}{2}}M^{ml}\alpha^m|J_m(\beta_1)-J_m(\beta_2)|_l\le
 \beta_1^{-\frac{1}{2}},\ (\forall
 \bU\in\overline{D}).\label{eq_IR_higher_order_difference}
\end{align}

Later in Subsection \ref{subsec_application_IR} we will see that the
output of the
infrared integration at scale $l+1$ belongs to $\cS(l)$
and a pair of the output at $\beta_1$ and
$\beta_2$ belongs to $\tilde{\cS}(l)$. The bound properties
assumed in $\cS(l)$ and $\tilde{\cS}(l)$ correspond to the resulting
inequalities in Proposition \ref{prop_infrared_integration_general} and
Proposition \ref{prop_infrared_integration_difference_general} with
$a_1=2$, $a_2=1$, $a_3=1$, $a_4=1/2$. In fact we will apply these
propositions with these exponents in the
forthcoming IR analysis. The invariant properties listed in
\eqref{item_IR_set_invariance}, \eqref{item_IR_set_complex_invariance}
are especially needed in order that for $J(\psi)\in\cS(l)$ the
kernel of $J_2(\psi)$ has
desirable symmetries for updating the covariance without changing the
original infrared singularity. 

As a preliminary, let us characterize the quadratic part of a polynomial
belonging to this class in the momentum space. Let $l\in \Z_{\le 0}$ and $J^l(\psi)\in
\cS(l)$. Using the kernel $J_2^l(\cdot):I^2\to\C$ of the quadratic part
$J_2^l(\psi)$, we define the map $W^l(\cdot,\cdot):\cM\times (2\pi/L)\Z^2\to \Mat(4,\C)$ by
\begin{align}
&W^l(\o,\bk)(\rho,\eta)\label{eq_partial_dispersion_relation}\\
&:=\frac{2}{h}\sum_{\bx\in \G}\sum_{x\in
 [0,\beta)_h}e^{-i\<\bk,\bx\>}e^{-i\o
 x}J_2^l((\rho,\bx,\ua,x,-1),(\eta,\b0,\ua,0,1)),\notag\\
&\ (\forall (\o,\bk)\in \cM\times (2\pi/L)\Z^2,
\rho,\eta\in\cB).\notag  \end{align}
\begin{lemma}\label{lem_properties_one_scale_self_energy}
The following statements hold true.
\begin{enumerate}
\item\label{item_quadratic_term_characterization}
\begin{align*}
&J_2^l(\psi)\\
&=\frac{1}{h^2}\sum_{(\rho,\bx,\s,x),\atop (\eta,\by,\tau,y)\in
 I_0}\frac{\delta_{\s,\tau}}{\beta
 L^2}\sum_{(\o,\bk)\in\cM_h\times\G^*}e^{i\<\bk,\bx-\by\>}e^{i\o(x-y)}W^l(\o,\bk)(\rho,\eta)\psi_{\rho\bx\s
 x}\opsi_{\eta\by\tau y}.
\end{align*}
\item\label{item_symmetry_one_scale_self_energy}
For any $(\o,\bk)\in \cM\times (2\pi/L)\Z^2$, $\bU\in \overline{D}$,
     $\rho,\eta\in\cB$, 
\begin{align}
&W^l(\bU)(\o,\bk)(\rho,\eta)=\overline{W^l(\overline{\bU})(-\o,\bk)(\eta,\rho)},\label{eq_hermitian_momentum}\\
&W^l(\bU)(\o,\bk)(\rho,\eta)=(-1)^{1_{\rho\in\{1,4\}}+1_{\eta\in\{1,4\}}+1}\overline{W^l(\overline{\bU})(\o,\bk)(\eta,\rho)},\label{eq_half_filled_momentum}\\
&W^l(\bU)(\o,\bk)(\rho,\eta)=e^{i\<\be(\rho)-\be(\eta),\bk\>}W^l(\bU)(\o,-\bk)(\rho,\eta).\label{eq_rotation_momentum}
\end{align}
\item\label{item_bound_one_scale_self_energy}
There exists a constant $c\in\R_{>0}$ independent of any parameter such
 that 
\begin{align*}
&|W^l(\o,\bk)(\rho,\eta)|\le c\cdot c_{IR}^{-1}\left(|\o|+\sum_{j=1}^2|k_j-\pi|+\frac{1}{\beta}+\frac{1}{L}\right)M^{\frac{1}{2}l}\alpha^{-2},\\
&\ (\forall (\o,\bk)\in \cM\times (2\pi/L)\Z^2, \rho,\eta\in
 \cB).
\end{align*}
\end{enumerate}
\end{lemma}
\begin{remark}
The inequality in
 \eqref{item_bound_one_scale_self_energy} suggests that
 $W^l(\o,\bk)$ becomes negligibly small as $(\o,\bk)$ approaches
 $(0,\pi,\pi)$ and thus the point
 $(0,\pi,\pi)$ is essentially a zero-point of the perturbed
 matrix $i\o I_4-E(\bk)-W^l(\o,\bk)$. This
 is the crucial reason why the multi-scale IR integration around the
 point $(0,\pi,\pi)$ converges. We prove the inequality in
 \eqref{item_bound_one_scale_self_energy} by making use of the
 invariant properties summarized in
 \eqref{item_symmetry_one_scale_self_energy}. Our argument based on the
 preserved symmetries is motivated by the preceding work \cite{GM} by
 Giuliani and Mastropietro and should be regarded as an extension of 
 Giuliani-Mastropietro's RG method designed for the 2-band Hubbard model on the honeycomb lattice.
\end{remark}
\begin{proof}[Proof of Lemma \ref{lem_properties_one_scale_self_energy}]
\eqref{item_quadratic_term_characterization}: By the invariance with
 $S$, $Q$ defined in \eqref{eq_particle_hole_maps} we obtain that 
\begin{align}
&J_2^l((\rho,\bx,\s,x,\theta),(\eta,\by,\tau,y,\xi))=e^{i\frac{\pi}{2}(\theta+\xi)}
J_2^l((\rho,\bx,\s,x,\theta),(\eta,\by,\tau,y,\xi)),\label{eq_particle_hole_appl}\\
&(\forall (\rho,\bx,\s,x,\theta),(\eta,\by,\tau,y,\xi)\in I).\notag
\end{align}
By the invariance with $S$, $Q$ defined in \eqref{eq_spin_maps},
\begin{align}
&J_2^l((\rho,\bx,\s,x,\theta),(\eta,\by,\tau,y,\xi))\label{eq_spin_appl}\\
&=(-1)^{1_{\s=\ua}+1_{\tau=\ua}}J_2^l((\rho,\bx,\s,x,\theta),(\eta,\by,\tau,y,\xi)),\notag\\
&(\forall (\rho,\bx,\s,x,\theta),(\eta,\by,\tau,y,\xi)\in I).\notag
\end{align}
By the invariance with $S$, $Q$ defined in \eqref{eq_spin_reflection_maps},
\begin{align}
&J_2^l((\rho,\bx,\s,x,\theta),(\eta,\by,\tau,y,\xi))=J_2^l((\rho,\bx,-\s,x,\theta),(\eta,\by,-\tau,y,\xi)),\label{eq_spin_reflection_appl}\\
&(\forall (\rho,\bx,\s,x,\theta),(\eta,\by,\tau,y,\xi)\in I).\notag
\end{align}
Note that for any $x\in [0,\beta)_h$, $s\in (1/h)\Z$,
$$
n_{\beta}(r_{\beta}(r_{\beta}(x+s)-s)+s)=n_{\beta}(x+s).
$$
Using this equality and the uniqueness of the anti-symmetric kernel,
 we can deduce from the invariance with $S$, $Q$ defined in
 \eqref{eq_translation_maps} that
\begin{align}
&J_2^l((\rho,\bx,\s,x,\theta),(\eta,\by,\tau,y,\xi))\label{eq_translation_appl}\\
&=(-1)^{n_{\beta}(x+s)+n_{\beta}(y+s)}\notag\\
&\quad\cdot J_2^l((\rho,r_L(\bx+\bz),\s,r_{\beta}(x+s),\theta),(\eta,r_L(\by+\bz),\tau,r_{\beta}(y+s),\xi)),\notag\\
&(\forall (\rho,\bx,\s,x,\theta),(\eta,\by,\tau,y,\xi)\in I,\bz\in
 (2\pi/L)\Z^2,s\in (1/h)\Z).\notag
\end{align}
Using the equalities \eqref{eq_particle_hole_appl},
 \eqref{eq_spin_appl}, \eqref{eq_spin_reflection_appl},
 \eqref{eq_translation_appl} in this order, we observe that
\begin{align}
&\frac{1}{h^2}\sum_{(\rho,\bx,\s,x,\theta),\atop (\eta,\by,\tau,y,\xi)\in
 I}J_2^l((\rho,\bx,\s,x,\theta),(\eta,\by,\tau,y,\xi))\psi_{\rho\bx \s x\theta}\psi_{\eta\by \tau
 y\xi}\label{eq_quadratic_term_characterization_pre}\\
&=\frac{2}{h^2}\sum_{(\rho,\bx,\s,x),\atop (\eta,\by,\tau,y)\in
 I_0}J_2^l((\rho,\bx,\s,x,-1),(\eta,\by,\tau,y,1))\psi_{\rho\bx \s
 x}\opsi_{\eta\by \tau y}\notag\\
&=\frac{2}{h^2}\sum_{(\rho,\bx,\s,x),\atop (\eta,\by,\tau,y)\in
 I_0}\delta_{\s,\tau}J_2^l((\rho,\bx,\s,x,-1),(\eta,\by,\tau,y,1))\psi_{\rho\bx
 \s x}\opsi_{\eta\by \tau y}\notag\\
&=\frac{2}{h^2}\sum_{(\rho,\bx,\s,x),\atop (\eta,\by,\tau,y)\in
 I_0}\delta_{\s,\tau}(-1)^{n_{\beta}(x-y)}\notag\\
&\quad \cdot J_2^l((\rho,r_L(\bx-\by),\ua,r_{\beta}(x-y),-1),(\eta,\b0,\ua,0,1))\psi_{\rho\bx
 \s x}\opsi_{\eta\by \tau y}.\notag
\end{align}
It follows from \eqref{eq_partial_dispersion_relation} that 
\begin{align*}
&\frac{1}{\beta L^2}\sum_{(\o,\bk)\in
 \cM_h\times\G^*}e^{i\<\bk,\bx-\by\>}e^{i\o(x-y)}W^l(\o,\bk)(\rho,\eta)\\
&=(-1)^{n_{\beta}(x-y)}2J_2^l((\rho,r_L(\bx-\by),\ua,r_{\beta}(x-y),-1),
(\eta,\b0,\ua,0,1)).
\end{align*}
By combining this equality with
 \eqref{eq_quadratic_term_characterization_pre}
we obtain the claimed equality.

\eqref{item_symmetry_one_scale_self_energy}: By the invariance with $S$,
 $Q$ defined in \eqref{eq_hermitian_maps},
\begin{align*}
&\overline{J^l(\overline{\bU})_2((\rho,\bx,\s,x,\theta),(\eta,\by,\tau,y,\xi))}\\
&=(-1)^{1_{\theta=-1}+1_{\xi=-1}+1_{x\neq 0}+1_{y\neq 0}}\\
&\quad\cdot J^l(\bU)_2((\rho,\bx,\s,r_{\beta}(-x),-\theta),(\eta,\by,\tau,r_{\beta}(-y),-\xi)),\\
&(\forall (\rho,\bx,\s,x,\theta),(\eta,\by,\tau,y,\xi)\in I).
\end{align*}
By using this equality, \eqref{eq_translation_appl} and the
 anti-symmetry of the kernel we have that
\begin{align*}
&\overline{W^l(\overline{\bU})(-\o,\bk)(\eta,\rho)}\\
&=\frac{2}{h}\sum_{(\bx,x)\in \G\times
 [0,\beta)_h}e^{i\<\bk,\bx\>}e^{-i\o
 x}\overline{J^l(\overline{\bU})_2((\eta,\bx,\ua,x,-1),(\rho,\b0,\ua,0,1))}\\
&=\frac{2}{h}\sum_{(\bx,x)\in \G\times
 [0,\beta)_h}e^{i\<\bk,\bx\>}e^{-i\o
 x}(-1)^{1+1_{x\neq 0}}\\
&\quad\cdot J^l(\bU)_2((\eta,\bx,\ua,r_{\beta}(-x),1),(\rho,\b0,\ua,0,-1))\\
&=\frac{2}{h}\sum_{(\bx,x)\in \G\times
 [0,\beta)_h}e^{i\<\bk,\bx\>}e^{-i\o
 x}(-1)\\
&\quad\cdot J^l(\bU)_2((\eta,\b0,\ua,0,1),(\rho,r_L(-\bx),\ua,x,-1))\\
&=W^l(\bU)(\o,\bk)(\rho,\eta),
\end{align*}
which is \eqref{eq_hermitian_momentum}.

By the invariance with $S$, $Q$ defined in \eqref{eq_half_filled_maps}, 
\begin{align*}
&\overline{J^l(\overline{\bU})_2((\rho,\bx,\s,x,\theta),(\eta,\by,\tau,y,\xi))}\\
&=(-1)^{1_{\rho\in \{1,4\}}+1_{\eta\in \{1,4\}}}J^l(\bU)_2((\rho,\bx,\s,x,-\theta),(\eta,\by,\tau,y,-\xi)),\\
&(\forall (\rho,\bx,\s,x,\theta),(\eta,\by,\tau,y,\xi)\in I).
\end{align*}
It follows from this equality, \eqref{eq_translation_appl} and anti-symmetry that 
\begin{align*}
&(-1)^{1_{\rho\in \{1,4\}}+1_{\eta\in \{1,4\}}+1}\overline{W^l(\overline{\bU})(\o,\bk)(\eta,\rho)}\\
&=\frac{2}{h}\sum_{(\bx,x)\in \G\times
 [0,\beta)_h}e^{i\<\bk,\bx\>}e^{i\o
 x}(-1)^{1_{\rho\in \{1,4\}}+1_{\eta\in \{1,4\}}+1}\\
&\quad\cdot\overline{J^l(\overline{\bU})_2((\eta,\bx,\ua,x,-1),(\rho,\b0,\ua,0,1))}\\
&=\frac{2}{h}\sum_{(\bx,x)\in \G\times
 [0,\beta)_h}e^{i\<\bk,\bx\>}e^{i\o
 x}(-1) J^l(\bU)_2((\eta,\bx,\ua,x,1),(\rho,\b0,\ua,0,-1))\\
&=\frac{2}{h}\sum_{(\bx,x)\in \G\times
 [0,\beta)_h}e^{i\<\bk,\bx\>}e^{i\o
 x}(-1)^{n_{\beta}(-x)}\\
&\quad\cdot J^l(\bU)_2((\rho,r_L(-\bx),\ua,r_{\beta}(-x),-1),(\eta,\b0,\ua,0,1))\\
&=W^l(\bU)(\o,\bk)(\rho,\eta),
\end{align*}
which is \eqref{eq_half_filled_momentum}.

By the invariance with $S$, $Q$ defined in \eqref{eq_rotation_maps}, 
\begin{align*}
&J_2^l((\rho,\bx,\s,x,\theta),(\eta,\by,\tau,y,\xi))\\
&=J_2^l((\rho,r_L(-\bx-\be(\rho)),\s,x,\theta),(\eta,r_L(-\by-\be(\eta)),\tau,y,\xi)),\\
&(\forall (\rho,\bx,\s,x,\theta),(\eta,\by,\tau,y,\xi)\in I).
\end{align*}
By using this equality and \eqref{eq_translation_appl} we can derive
 that 
\begin{align*}
&e^{i\<\be(\rho)-\be(\eta),\bk\>}W^l(\o,-\bk)(\rho,\eta)\\
&=\frac{2}{h}\sum_{(\bx,x)\in \G\times
 [0,\beta)_h}e^{-i\<\bk,-\bx-\be(\rho)+\be(\eta)\>}e^{-i\o
 x}\\
&\quad\cdot J_2^l((\rho,r_L(-\bx-\be(\rho)),\ua,x,-1),(\eta,r_L(-\be(\eta)),\ua,0,1))\\
&=W^l(\o,\bk)(\rho,\eta),
\end{align*}
which is \eqref{eq_rotation_momentum}.

\eqref{item_bound_one_scale_self_energy}:
Take any $\rho,\eta\in \cB$ satisfying $\rho,\eta\in \{1,4\}$ or
 $\rho,\eta\in \{2,3\}$. Since 
$$
(-1)^{1_{\rho\in \{1,4\}}+1_{\eta\in \{1,4\}}+1}=-1,
$$ 
the equalities \eqref{eq_hermitian_momentum},
 \eqref{eq_half_filled_momentum} yield that
\begin{align*}
W^l(\o,\bk)(\rho,\eta)+W^l(-\o,\bk)(\rho,\eta)=0,\ (\forall (\o,\bk)\in
 \cM\times (2\pi/L)\Z^2).
\end{align*}
Especially, 
\begin{align*}
W^l\left(\frac{\pi}{\beta},\bk\right)(\rho,\eta)+W^l\left(-\frac{\pi}{\beta},\bk\right)(\rho,\eta)=0,\ (\forall \bk\in
 (2\pi/L)\Z^2).
\end{align*}
Using this equality, we have that
\begin{align}
|W^l(\o,\bk)(\rho,\eta)|
&\le
 \frac{1}{2}\left|W^l(\o,\bk)(\rho,\eta)-W^l\left(\frac{\pi}{\beta},\bk\right)(\rho,\eta)\right|\label{eq_bound_0th_direction}\\
&\quad +\frac{1}{2}\left|W^l(\o,\bk)(\rho,\eta)-W^l\left(-\frac{\pi}{\beta},\bk\right)(\rho,\eta)\right|\notag\\
&\le \left(|\o|+\frac{c}{\beta}\right)\sup_{(\o,\bk)\in \cM\times \frac{2\pi}{L}\Z^2}|\cD_0W^l(\o,\bk)(\rho,\eta)|.\notag
\end{align}

On the other hand, let us fix $\rho,\eta\in\cB$ satisfying $\rho\in
 \{1,4\}$ and $\eta\in \{2,3\}$, or $\rho\in
 \{2,3\}$ and $\eta\in \{1,4\}$. First, consider the case that $L\in
 2\N$. Since $(\pi,\pi)\in(2\pi/L)\Z^2$ in this case, the equality
 \eqref{eq_rotation_momentum} ensures that
 $W^l(\o,(\pi,\pi))(\rho,\eta)=0$. We deduce from this equality that
\begin{align}
|W^l(\o,\bk)(\rho,\eta)|
&\le
 |W^l(\o,\bk)(\rho,\eta)-W^l(\o,(\pi,k_2))(\rho,\eta)|\label{eq_bound_momentum_direction_even}\\
&\quad +
 |W^l(\o,(\pi,k_2))(\rho,\eta)-W^l(\o,(\pi,\pi))(\rho,\eta)|\notag\\
&\le \sum_{j=1}^2|k_j-\pi|\sup_{(\o,\bk)\in \cM\times
 \frac{2\pi}{L}\Z^2}|\cD_j W^l(\o,\bk)(\rho,\eta)|.\notag
\end{align}
Next, let us assume that $L\notin 2\N$. In this case, $\pi-\pi/L$,
 $\pi+\pi/L\in (2\pi/L)\Z$. Therefore, it follows from the equality
 \eqref{eq_rotation_momentum} that
\begin{align*}
&\left|W^l\left(\o,\left(\pi+\frac{\pi}{L},\pi+\frac{\pi}{L}\right)\right)(\rho,\eta)+W^l\left(\o,\left(\pi-\frac{\pi}{L},\pi-\frac{\pi}{L}\right)\right)(\rho,\eta)\right|\\
&\le |e^{i\frac{\pi}{L}}-1|\left|
 W^l\left(\o,\left(\pi-\frac{\pi}{L},\pi-\frac{\pi}{L}\right)\right)(\rho,\eta)\right|.\notag
\end{align*}
Substituting this inequality, we see that  
\begin{align}
&|W^l(\o,\bk)(\rho,\eta)|\label{eq_bound_momentum_direction_odd}\\
&\le
 \frac{1}{2}\sum_{\delta\in \{1,-1\}}\Bigg(
\left|W^l(\o,\bk)(\rho,\eta)-W^l\left(\o,\left(\pi+\frac{\delta\pi}{L},k_2\right)\right)(\rho,\eta)\right|\notag\\
&\qquad\quad +\Bigg|W^l\left(\o,\left(\pi+\frac{\delta\pi}{L},k_2\right)\right)(\rho,\eta)\notag\\
&\qquad\qquad\quad -W^l\left(\o,\left(\pi+\frac{\delta\pi}{L},\pi+\frac{\delta\pi}{L}\right)\right)(\rho,\eta)\Bigg|\Bigg)\notag\\
&\quad +\frac{1}{2}\left|\sum_{\delta\in
 \{1,-1\}}W^l\left(\o,\left(\pi+\frac{\delta\pi}{L},\pi+\frac{\delta\pi}{L}\right)\right)(\rho,\eta)\right|\notag\\
&\le
 \sum_{j=1}^2\left(|k_j-\pi|+\frac{c}{L}\right)\sup_{(\o,\bk)\in\cM\times\frac{2\pi}{L}\Z^2}|\cD_jW^l(\o,\bk)(\rho,\eta)|\notag\\
&\quad + \frac{c}{L}\sup_{(\o,\bk)\in\cM\times\frac{2\pi}{L}\Z^2}|W^l(\o,\bk)(\rho,\eta)|.\notag
\end{align}

The inequalities \eqref{eq_bound_0th_direction},
 \eqref{eq_bound_momentum_direction_even},
 \eqref{eq_bound_momentum_direction_odd} lead to 
\begin{align}
|W^l(\o,\bk)(\rho,\eta)|
&\le\left(|\o|+\sum_{m=1}^2|k_m-\pi|+\frac{c}{\beta}+\frac{c}{L}\right)\label{eq_bound_one_scale_self_energy_pre}\\
&\quad\cdot\Bigg(\sup_{j\in\{0,1,2\}}\sup_{(\o,\bk)\in\cM\times
 \frac{2\pi}{L}\Z^2}|\cD_jW^l(\o,\bk)(\rho,\eta)|\notag\\
&\qquad\quad+\sup_{(\o,\bk)\in
 \cM\times \frac{2\pi}{L}\Z^2}|W^l(\o,\bk)(\rho,\eta)|\Bigg),\notag\\
&(\forall (\o,\bk)\in \cM\times (2\pi/L)\Z^2,\rho,\eta\in\cB).\notag
\end{align}
We can see from the definition of $W^l(\cdot)$ and
 \eqref{eq_IR_higher_order_bound} that 
\begin{align*}
&|W^l(\o,\bk)(\rho,\eta)|\le 2\|J_2^l\|_{l,0}\le 2
 c_{IR}^{-1}M^{\frac{3}{2}l}\alpha^{-2},\\
&|\cD_j W^l(\o,\bk)(\rho,\eta)|\le 2\|J_2^l\|_{l,1}
\le 2
 c_{IR}^{-1}M^{\frac{1}{2}l}\alpha^{-2},\ (\forall j\in \{0,1,2\}).
\end{align*}
By combining these inequalities with
 \eqref{eq_bound_one_scale_self_energy_pre} we obtain the inequality
 claimed in \eqref{item_bound_one_scale_self_energy}.
\end{proof}

For later use we define an extension of the function $W^l(\cdot)$ with
continuous variables. For this purpose we need
a few more notations. For any $x\in [0,\beta)$ let 
$$r_{\beta}'(x):=\left\{\begin{array}{ll}x&\text{if }x\in [0,\beta/2),\\
x-\beta&\text{if }x\in [\beta/2,\beta).\end{array}
\right.$$
For any $x\in [0,L)$ let 
$$
s_L(x):=\left\{\begin{array}{ll}x&\text{if }x\in [0,L/2),\\
x-L&\text{if }x\in [L/2,L).\end{array}
\right.
$$
Then, for any $(x_1,x_2)\in [0,L)^2$ let
$r_L'((x_1,x_2)):=(s_L(x_1),s_L(x_2))$.
With these notations, set
\begin{align}
&\widehat{W}^l(\o,\bk)(\rho,\eta)
:=\frac{2}{h}\sum_{(\bx,x)\in \G\times
 [0,\beta)_h}e^{-i\<\bk,r_L'(\bx)\>}e^{-i\o r'_{\beta}(x)}(-1)^{n_{\beta}(r_{\beta}'(x))}\label{eq_effective_partial_dispersion_relation}\\
&\qquad\qquad\qquad\qquad\qquad\cdot J_2^l((\rho,\bx,\ua,x,-1),(\eta,\b0,\ua,0,1)),\notag\\
&(\forall (\o,\bk)\in \R^3,
\rho,\eta\in\cB).\notag  
\end{align}
Note that 
\begin{align} 
\widehat{W}^l(\o,\bk)=W^l(\o,\bk),\ (\forall (\o,\bk)\in \cM\times (2\pi/L)\Z^2),\label{eq_extension_normal_self_energy}
\end{align}
Let us establish various inequalities involving 
this function in 
Lemma \ref{lem_bound_extended_one_scale_self_energy},
Lemma \ref{lem_bound_extended_one_scale_self_energy_difference},
Lemma \ref{lem_bound_self_energy} and 
Lemma \ref{lem_bound_self_energy_difference} below, step by step.

\begin{lemma}\label{lem_bound_extended_one_scale_self_energy}
Assume that 
\begin{align}
\frac{1}{L}\le \frac{1}{\beta}\le
 M_{IR}M^{N_{\beta}+1}.\label{eq_beta_L_temporal_assumption}
\end{align}
 Then, there exists a constant $c\in\R_{>0}$ independent
 of any parameter such that the following inequalities hold
for any $l\in \{0,-1,\cdots,N_{\beta}\}$, $j\in
 \{0,-1,\cdots,l\}$.
\begin{enumerate}
\item\label{item_bound_extended_one_scale_self_energy}
\begin{align*}
&\|\widehat{W}^j(\o,\bk)\|_{4\times 4}\le c\cdot
 c_{IR}^{-1}f_{\bt}^{-\frac{1}{2}}M^{\frac{1}{2}j+l+1}\alpha^{-2},\\
&(\forall (\o,\bk)\in\R^3\text{ satisfying }\chi_l(\o,\bk)\neq 0).
\end{align*}
\item\label{item_derivative_extended_one_scale_self_energy}
\begin{align*}
&\left\|\left(\frac{\partial}{\partial k_i}\right)^n\widehat{W}^j(\o,\bk)\right\|_{4\times 4}\le c\cdot
 c_{IR}^{-1}M^{\frac{3}{2}j}\alpha^{-2}(c\fw(j)^{-1})^n(n!)^2,\\
&(\forall (\o,\bk)\in\R^3,i\in \{0,1,2\},n\in \N\cup \{0\}).
\end{align*}
\end{enumerate}
\end{lemma}
\begin{proof}
\eqref{item_bound_extended_one_scale_self_energy}: 
Take any $(\o,\bk)\in\R^3$ satisfying $\chi_l(\o,\bk)\neq 0$. By
 periodicity we may assume that $\bk\in [0,2\pi)^2$ without losing
 generality. Let $\hat{\o}\in\cM$,
 $\hat{\bk}=(\hat{k}_1,\hat{k}_2)\in (2\pi/L)\Z^2$ be such that
 $\o\in [\hat{\o},\hat{\o}+2\pi/\beta)$, $\bk\in
 [\hat{k}_1,\hat{k}_1+2\pi/L)\times  [\hat{k}_2,\hat{k}_2+2\pi/L)$. It
 follows from \eqref{eq_IR_higher_order_bound} that 
\begin{align}
&|\widehat{W}^j(\o,\bk)(\rho,\eta)-W^j(\hat{\o},\hat{\bk})(\rho,\eta)|\label{eq_linear_constant_piecewise}\\
&\le \frac{2}{h}\sum_{(\bx,x)\in\G\times [0,\beta)_h}(|e^{i\o
 r_{\beta}'(x)}-e^{i\hat{\o} r_{\beta}'(x)}|+
|e^{i\<\bk,r_{L}'(\bx)\>}-e^{i\<\hat{\bk},r_{L}'(\bx)\>}|)\notag\\
&\qquad\qquad\qquad\cdot|J_2^j((\rho,\bx,\ua,x,-1),(\eta,\b0,\ua,0,1))|\notag\\
&\le \frac{c}{h}\sum_{(\bx,x)\in\G\times
 [0,\beta)_h}\Bigg(|\o-\hat{\o}|(1_{x< \frac{\beta}{2}}|x|+1_{x\ge 
 \frac{\beta}{2}}|x-\beta|)\notag\\
&\qquad\qquad\qquad\qquad+\sum_{m=1}^2|k_m-\hat{k}_m|(1_{x_m<
 \frac{L}{2}}|x_m|+1_{x_m\ge
 \frac{L}{2}}|x_m-L|)\Bigg)\notag\\
&\qquad\qquad\qquad\cdot|J_2^j((\rho,\bx,\ua,x,-1),(\eta,\b0,\ua,0,1))|\notag\\
&\le  \frac{c}{h}\sum_{(\bx,x)\in\G\times
 [0,\beta)_h}\Bigg(\frac{1}{\beta}d_0((\rho,\bx,\ua,x,-1),(\eta,\b0,\ua,0,1))\notag\\
&\qquad\qquad\qquad\qquad+\frac{1}{L}\sum_{m=1}^2d_m((\rho,\bx,\ua,x,-1),(\eta,\b0,\ua,0,1))\Bigg)\notag\\
&\qquad\qquad\qquad\cdot|J_2^j((\rho,\bx,\ua,x,-1),(\eta,\b0,\ua,0,1))|\notag\\
&\le c\left(\frac{1}{\beta}+\frac{1}{L}\right)\|J_2^j\|_{j,1}
\le c\left(\frac{1}{\beta}+\frac{1}{L}\right) c_{IR}^{-1}M^{\frac{1}{2}j}\alpha^{-2}.\notag\end{align}

By Lemma \ref{lem_infrared_cut_off_measure} and the inequality
 $f_{\bt}\le 1$,
\begin{align*}
|\hat{\o}|+\sum_{m=1}^2|\hat{k}_m-\pi|\le
 c\left(\frac{1}{\beta}+\frac{1}{L}+f_{\bt}^{-\frac{1}{2}}M^{l+1}\right).
\end{align*}
We can combine this inequality with Lemma
 \ref{lem_properties_one_scale_self_energy}
 \eqref{item_bound_one_scale_self_energy},
 \eqref{eq_beta_L_temporal_assumption} and
 \eqref{eq_linear_constant_piecewise} to deduce that
\begin{align*}
\|\widehat{W}^j(\o,\bk)\|_{4\times 4}&\le
 \|\widehat{W}^j(\o,\bk)-W^j(\hat{\o},\hat{\bk})\|_{4\times 4}+
 \|W^j(\hat{\o},\hat{\bk})\|_{4\times 4}\\
&\le c\cdot
 c_{IR}^{-1}\Bigg(|\hat{\o}|+\sum_{m=1}^2|\hat{k}_m-\pi|+\frac{1}{\beta}+\frac{1}{L}\Bigg)M^{\frac{1}{2}j}\alpha^{-2}\\
&\le c\cdot
 c_{IR}^{-1}f_{\bt}^{-\frac{1}{2}}
M^{\frac{1}{2}j+l+1}\alpha^{-2}.
\end{align*}

\eqref{item_derivative_extended_one_scale_self_energy}:
By \eqref{eq_IR_higher_order_bound} and the inequality 
$$
|r_{\beta}'(x)|\le
 \frac{\pi}{2}|d_0((\rho,\bx,\ua,x,-1),(\eta,\b0,\ua,0,1))|,\ (\forall
 x\in [0,\beta)_h),
$$
we have that
\begin{align*}
\left\|\left(\frac{\partial}{\partial \o}\right)^n\widehat{W}^j(\o,\bk)
\right\|_{4\times 4}&\le
 2\left(\frac{\pi}{2}\right)^n\fw(j)^{-n}(2n)!\|J_2^j\|_{j,0}\\
&\le c^{n+1}\fw(j)^{-n}(n!)^2 c_{IR}^{-1}M^{\frac{3}{2}j}\alpha^{-2},\
 (\forall n\in \N\cup \{0\}).
\end{align*}
Since 
\begin{align*}
|s_{L}(x_i)|\le
 \frac{\pi}{2}|d_i((\rho,\bx,\ua,x,-1),(\eta,\b0,\ua,0,1))|,\ (\forall
 i\in\{1,2\},(x_1,x_2)\in \G),
\end{align*}
the upper bound on $\|(\partial/\partial
 k_i)^n\widehat{W}^j(\o,\bk)\|_{4\times 4}$ $(i\in \{1,2\})$ can be obtained
 in the same way.
\end{proof}

\begin{lemma}\label{lem_bound_extended_one_scale_self_energy_difference}
Assume that \eqref{eq_beta_h_assumption} holds. 
Then, there exists a constant $c\in\R_{>0}$ independent of any parameter
 such that the following inequality holds true for any 
$j\in  \{0,-1,\cdots,N_{\beta_1}\}$,
 $(J^j(\beta_1)(\psi),J^j(\beta_2)(\psi))\in \tilde{\cS}(j)$.
\begin{align*}
&\left\|\left(\frac{\partial}{\partial
 k_i}\right)^n(\widehat{W}^j(\beta_1)(\o,\bk)
-\widehat{W}^j(\beta_2)(\o,\bk))\right\|_{4\times 4}\\
&\le
 c\beta_1^{-\frac{1}{2}}c_{IR}^{-1}M^{\frac{1}{2}j}\alpha^{-2}
(c\fw(j)^{-1})^n(n!)^2,\\
&(\forall (\o,\bk)\in\R^3,i\in \{0,1,2\},n\in \N\cup\{0\}).
\end{align*}
\end{lemma}
\begin{proof}
Note that for any $a\in \{1,2\}$,
\begin{align*}
&\widehat{W}^j(\beta_a)(\o,\bk)(\rho,\eta)\\
&=\frac{2}{h}\sum_{(\bx,x)\in \G\times
 [0,\beta_a)_h}e^{-i\<\bk,r_L'(\bx)\>}\big(1_{x\in
 [0,\frac{\beta_1}{4})}e^{-i\o x}+1_{x\in
 [\frac{\beta_1}{4},\frac{\beta_a}{2})}e^{-i\o x}\\
&\qquad\qquad\qquad-1_{x\in
 [\frac{\beta_a}{2},\beta_a-\frac{\beta_1}{4})}e^{-i\o (x-\beta_a)}
-1_{x\in
 [\beta_a-\frac{\beta_1}{4},\beta_a)}e^{-i\o (x-\beta_a)}\big)\\
&\qquad\qquad \cdot J_2^j(\beta_a)((\rho,\bx,\ua,x,-1),(\eta,\b0,\ua,0,1))\\
&=\frac{2}{h}\sum_{\bx\in\G}\sum_{x\in
 [-\frac{\beta_1}{4},\frac{\beta_1}{4})_h}e^{-i\<\bk,r_L'(\bx)\>}\\
&\qquad\qquad\cdot\big(1_{x\in
 [0,\frac{\beta_1}{4})}e^{-i\o
 x}J_2^j(\beta_a)((\rho,\bx,\ua,x,-1),(\eta,\b0,\ua,0,1)\big)\\
&\qquad\qquad\qquad-1_{x\in
 [-\frac{\beta_1}{4},0)}e^{-i\o
 x}J_2^j(\beta_a)((\rho,\bx,\ua,x+\beta_a,-1),(\eta,\b0,\ua,0,1)))\\
&\quad + \frac{2}{h}\sum_{(\bx,x)\in\G\times [0,\beta_a)_h}
e^{-i\<\bk,r_L'(\bx)\>}\big(1_{x\in
 [\frac{\beta_1}{4},\frac{\beta_a}{2})}e^{-i\o
 x}-1_{x\in
 [\frac{\beta_a}{2},\beta_a-\frac{\beta_1}{4})}e^{-i\o
 (x-\beta_a)}\big)\\
&\qquad\qquad\cdot J_2^j(\beta_a)((\rho,\bx,\ua,x,-1),(\eta,\b0,\ua,0,1)).
\end{align*}
Thus, 
\begin{align*}
&\left|\prod_{i=0}^2\left(\frac{\partial}{\partial k_i}\right)^{n_i}
(\widehat{W}^j(\beta_1)(\o,\bk)(\rho,\eta)
-\widehat{W}^j(\beta_2)(\o,\bk)(\rho,\eta))\right|\\
&\le \frac{2}{h}\sum_{\bx\in\G}\sum_{x\in
 [-\frac{\beta_1}{4},\frac{\beta_1}{4})_h}|x|^{n_0}|s_L(x_1)|^{n_1}|s_L(x_2)|^{n_2}\\
&\quad\cdot
 |J_2^j(\beta_1)(R_{\beta_1}((\rho,\bx,\ua,x,-1),(\eta,\b0,\ua,0,1)))\\
&\qquad -J_2^j(\beta_2)(R_{\beta_2}((\rho,\bx,\ua,x,-1),(\eta,\b0,\ua,0,1)))|\\
&\quad + \frac{2}{h}\sum_{a=1}^2
\sum_{(\bx,x)\in\G\times [0,\beta_a)_h}\big(1_{x\in
 [\frac{\beta_1}{4},\frac{\beta_a}{2})}|x|^{n_0}
+1_{x\in
 [\frac{\beta_a}{2},\beta_a-\frac{\beta_1}{4})}|x-\beta_a|^{n_0}\big)\\
&\qquad\cdot |s_L(x_1)|^{n_1}|s_L(x_2)|^{n_2}|J_2^j(\beta_a)((\rho,\bx,\ua,x,-1),(\eta,\b0,\ua,0,1))|\\
&\le \frac{2}{h}\sum_{\bx\in\G}\sum_{x\in
 [-\frac{\beta_1}{4},\frac{\beta_1}{4})_h}\left(\frac{\pi}{2}\right)^{n_1+n_2}\prod_{i=0}^2|\hat{d}_i((\rho,\bx,\ua,x,-1),(\eta,\b0,\ua,0,1))|^{n_i}\\
&\quad\cdot
 |J_2^j(\beta_1)(R_{\beta_1}((\rho,\bx,\ua,x,-1),(\eta,\b0,\ua,0,1)))\\
&\qquad -J_2^j(\beta_2)(R_{\beta_2}((\rho,\bx,\ua,x,-1),(\eta,\b0,\ua,0,1)))|\\
&\quad + \frac{2}{h}\sum_{a=1}^2
\sum_{(\bx,x)\in\G\times [0,\beta_a)_h}\left(\frac{\pi}{2}\right)^{n_1+n_2}\\
&\qquad\cdot \Bigg(1_{x\in
 [\frac{\beta_1}{4},\frac{\beta_a}{2})}\frac{|x|^{n_0}}{|d_0(\beta_a)((\rho,\bx,\ua,x,-1),(\eta,\b0,\ua,0,1))|^{n_0+1}}\\
&\qquad\quad+1_{x\in
 [\frac{\beta_a}{2},\beta_a-\frac{\beta_1}{4})}\frac{|x-\beta_a|^{n_0}}{|d_0(\beta_a)((\rho,\bx,\ua,x,-1),(\eta,\b0,\ua,0,1))|^{n_0+1}}\Bigg)\\
&\qquad\cdot
 |d_0(\beta_a)((\rho,\bx,\ua,x,-1),(\eta,\b0,\ua,0,1))|\\
&\qquad\cdot
 \prod_{i=0}^2|d_i((\rho,\bx,\ua,x,-1),(\eta,\b0,\ua,0,1))|^{n_i}\\
&\qquad\cdot |J_2^j(\beta_a)((\rho,\bx,\ua,x,-1),(\eta,\b0,\ua,0,1))|.
\end{align*}
Moreover, substitution of \eqref{eq_IR_higher_order_bound},
 \eqref{eq_IR_higher_order_difference} gives
\begin{align*}
&\left\|\left(\frac{\partial}{\partial
 k_i}\right)^n(\widehat{W}^j(\beta_1)(\o,\bk)
-\widehat{W}^j(\beta_2)(\o,\bk))\right\|_{4\times 4}\\
&\le
 c^{n+1}\fw(j)^{-n}(2n)!|J_2^j(\beta_1)-J_2^j(\beta_2)|_j\\
&\quad +c^{n+1}\beta_1^{-1}\fw(j)^{-n}(2n)!\sum_{a=1}^2\|J_2^j(\beta_a)\|_{j,1}\\
&\le c \beta_1^{-\frac{1}{2}}c_{IR}^{-1}M^{\frac{1}{2}j}\alpha^{-2}
(c\fw(j)^{-1})^n(n!)^2,\\
& (\forall (\o,\bk)\in\R^3,i\in\{0,1,2\},n\in \N\cup\{0\}).
\end{align*}
\end{proof}

\begin{lemma}\label{lem_bound_self_energy}
Assume that 
 \eqref{eq_beta_L_temporal_assumption} holds.
Then, there exists a constant $c\in\R_{>0}$ independent of any parameter
 such that the following inequalities hold true for any 
$l\in \{0,-1,\cdots,N_{\beta}\}$.
\begin{enumerate}
\item\label{item_bound_extended_self_energy}
\begin{align*}
&\left\|\sum_{j=0}^{l'}\hat{\chi}_{\le
 j}(\o,\bk)\widehat{W}^j(\o,\bk)\right\|_{4\times 4}\le c\cdot
 c_{IR}^{-1}f_{\bt}^{-\frac{1}{2}}M^{l+1}\alpha^{-2},\\
&(\forall (\o,\bk)\in\R^3\text{ satisfying }\chi_l(\o,\bk)\neq
 0,l'\in \{0,-1,\cdots,l\}).
\end{align*}
\item\label{item_neumann_justification_extended_self_energy}
\begin{align*}
&\left\|(i\o I_4-E(\bk))^{-1}\sum_{j=0}^{l'}\hat{\chi}_{\le
 j}(\o,\bk)\widehat{W}^j(\o,\bk)\right\|_{4\times 4}\le c\cdot
 c_{IR}^{-1}f_{\bt}^{-\frac{1}{2}}M\alpha^{-2},\\
&(\forall (\o,\bk)\in\R^3\text{ satisfying }\chi_l(\o,\bk)\neq
 0,l'\in \{0,-1,\cdots,l\}).
\end{align*}
\item\label{item_bound_extended_self_energy_derivative}
\begin{align*}
&\left\|\left(\frac{\partial}{\partial k_i}\right)^n\left(\sum_{j=0}^{l'}\hat{\chi}_{\le
 j}(\o,\bk)\widehat{W}^j(\o,\bk)\right)\right\|_{4\times 4}\le c\cdot
 c_{IR}^{-1}M^{l}\alpha^{-2}(c\fw(l)^{-1})^n(n!)^2,\\
&(\forall (\o,\bk)\in\R^3,i\in
 \{0,1,2\},n\in \N, l'\in\{0,-1,\cdots,l\}).
\end{align*}
\end{enumerate}
\end{lemma}
\begin{proof}
\eqref{item_bound_extended_self_energy}: It follows from Lemma
 \ref{lem_bound_extended_one_scale_self_energy}
 \eqref{item_bound_extended_one_scale_self_energy} and the assumption
 $M> \sqrt{2}$ that 
\begin{align*}
&\left\|\sum_{j=0}^{l'}\hat{\chi}_{\le
 j}(\o,\bk)\widehat{W}^j(\o,\bk)\right\|_{4\times 4}\\
&\le c\cdot
 c_{IR}^{-1}f_{\bt}^{-\frac{1}{2}}M^{l+1}\alpha^{-2}\sum_{j=0}^lM^{\frac{1}{2}j}\le c\cdot
 c_{IR}^{-1}f_{\bt}^{-\frac{1}{2}}M^{l+1}\alpha^{-2},\\
&(\forall (\o,\bk)\in\R^3\text{ satisfying }\chi_l(\o,\bk)\neq
 0,l'\in \{0,-1,\cdots,l\}).
\end{align*}

\eqref{item_neumann_justification_extended_self_energy}:
By Lemma \ref{lem_estimate_free_dispersion_relation}
 \eqref{item_upper_bound_free_propagator},
 \eqref{eq_support_properties_infrared} and the inequality in
 \eqref{item_bound_extended_self_energy},
\begin{align*}
&\left\|(i\o I_4-E(\bk))^{-1}\sum_{j=0}^{l'}\hat{\chi}_{\le
 j}(\o,\bk)\widehat{W}^j(\o,\bk)\right\|_{4\times 4}\\
&\le \|(i\o I_4-E(\bk))^{-1}\|_{4\times 4}\left\|\sum_{j=0}^{l'}\hat{\chi}_{\le
 j}(\o,\bk)\widehat{W}^j(\o,\bk)\right\|_{4\times 4}\\
&\le c\cdot
 c_{IR}^{-1}f_{\bt}^{-\frac{1}{2}}M\alpha^{-2},\\
&(\forall (\o,\bk)\in\R^3\text{ satisfying }\chi_l(\o,\bk)\neq
 0,l'\in \{0,-1,\cdots,l\}).
\end{align*}

\eqref{item_bound_extended_self_energy_derivative}:
One can see from Lemma \ref{lem_infrared_cut_off_derivative}, Lemma
 \ref{lem_bound_extended_one_scale_self_energy}
 \eqref{item_derivative_extended_one_scale_self_energy}, $\fw(j)=\fw(0)M^j$
 and $M > \sqrt{2}$ that
\begin{align*}
&\left\|\left(\frac{\partial}{\partial k_i}\right)^n\left(\sum_{j=0}^{l'}\hat{\chi}_{\le
 j}(\o,\bk)\widehat{W}^j(\o,\bk)\right)\right\|_{4\times 4}\\
&\le \sum_{j=0}^{l'}\sum_{m=0}^n\left(\begin{array}{c} n\\ m\end{array}\right)
\left|\left(\frac{\partial}{\partial k_i}\right)^m\hat{\chi}_{\le
 j}(\o,\bk)\right|\left\|\left(\frac{\partial}{\partial
 k_i}\right)^{n-m}\widehat{W}^j(\o,\bk)\right\|_{4\times 4}\\
&\le \sum_{j=0}^{l'}\sum_{m=0}^n\left(\begin{array}{c} n\\
				      m\end{array}\right)\\
&\quad\cdot (c \fw(j)^{-1})^m(m!)^2 c\cdot
 c_{IR}^{-1}M^{\frac{3}{2}j}\alpha^{-2}(c\fw(j)^{-1})^{n-m}((n-m)!)^2\\
&\le c\cdot
 c_{IR}^{-1}\alpha^{-2}(c \fw(0)^{-1})^n
 \sum_{j=0}^lM^{\frac{3}{2}j-jn}
\sum_{m=0}^n\left(\begin{array}{c} n\\ m\end{array}\right)
(m!)^2 ((n-m)!)^2\\
&\le  c\cdot  
 c_{IR}^{-1}\alpha^{-2}(c \fw(0)^{-1})^n(n!)^2M^{(1-n)l}
 \sum_{j=0}^lM^{\frac{1}{2}j}\\
&\le  c\cdot
 c_{IR}^{-1}M^l\alpha^{-2}(c \fw(l)^{-1})^n(n!)^2,\\
&(\forall (\o,\bk)\in\R^3,i\in
 \{0,1,2\},n\in \N,l'\in \{0,-1,\cdots,l\}).
\end{align*}
\end{proof}

\begin{lemma}\label{lem_bound_self_energy_difference}
Assume that \eqref{eq_beta_h_assumption} holds. 
Then, there exists a constant $c\in\R_{>0}$ independent of any parameter
 such that the following inequality holds true for any 
 $l\in \{0,-1,\cdots,N_{\beta_1}\}$ and
 $(J^j(\beta_1)(\psi),J^j(\beta_2)(\psi))\in\tilde{\cS}(j)$ $(j
=0,-1,\cdots,l)$. 
\begin{align*}
&\Bigg\|\left(\frac{\partial}{\partial k_i}\right)^n\Bigg(
\sum_{j=0}^{l}\hat{\chi}_{\le
 j}(\o,\bk)\widehat{W}^j(\beta_1)(\o,\bk)\\
&\qquad\qquad\quad-\sum_{j=0}^{l}\hat{\chi}_{\le
 j}(\o,\bk)\widehat{W}^j(\beta_2)(\o,\bk)
\Bigg)\Bigg\|_{4\times 4}\\
&\le  c\beta_1^{-\frac{1}{2}}
 c_{IR}^{-1}\alpha^{-2}(c\fw(l)^{-1})^n(n!)^2,\\
&(\forall (\o,\bk)\in\R^3,i\in
 \{0,1,2\},n\in \N\cup\{0\}).
\end{align*}
\end{lemma}
\begin{proof}
By Lemma \ref{lem_infrared_cut_off_derivative}, Lemma
 \ref{lem_bound_extended_one_scale_self_energy_difference} and the
 assumption $M> \sqrt{2}$,
\begin{align*}
&\Bigg\|\left(\frac{\partial}{\partial k_i}\right)^n\Bigg(
\sum_{j=0}^{l}\hat{\chi}_{\le
 j}(\o,\bk)\widehat{W}^j(\beta_1)(\o,\bk)\\
&\qquad\qquad\quad-\sum_{j=0}^{l}\hat{\chi}_{\le
 j}(\o,\bk)\widehat{W}^j(\beta_2)(\o,\bk)
\Bigg)\Bigg\|_{4\times 4}\\
&\le \sum_{j=0}^l\sum_{m=0}^n\left(\begin{array}{c}n\\ m\end{array}
\right)\left|\left(\frac{\partial}{\partial k_i}\right)^m\hat{\chi}_{\le
 j}(\o,\bk)\right|\\
&\qquad\qquad\quad\cdot\left\|\left(\frac{\partial}{\partial k_i}\right)^{n-m}
(\widehat{W}^j(\beta_1)(\o,\bk)-\widehat{W}^j(\beta_2)(\o,\bk))
\right\|_{4\times 4}\\
&\le \sum_{j=0}^l\sum_{m=0}^n\left(\begin{array}{c}n\\ m\end{array}
\right)(c\fw(j)^{-1})^m(m!)^2 \\
&\qquad\qquad\quad\cdot c\beta_1^{-\frac{1}{2}}
 c_{IR}^{-1}M^{\frac{1}{2}j}
\alpha^{-2}(c\fw(j)^{-1})^{n-m}((n-m)!)^2\\
&\le c\beta_1^{-\frac{1}{2}}
 c_{IR}^{-1}\alpha^{-2}(c\fw(l)^{-1})^{n}(n!)^2\sum_{j=0}^lM^{\frac{1}{2}j}\\
&\le c\beta_1^{-\frac{1}{2}}
 c_{IR}^{-1}\alpha^{-2}(c\fw(l)^{-1})^{n}(n!)^2,\\
&(\forall (\o,\bk)\in\R^3,i\in
 \{0,1,2\},n\in \N\cup\{0\}).
\end{align*}
\end{proof}

Assuming that $l\in \{0,-1,\cdots,N_{\beta}\}$ and
$J^j(\psi)\in\cS(j)$ $(j=0,-1,\cdots,$ $l)$, we can introduce a
covariance which mimics a real covariance in the IR integration at the
$l$th scale. 
Set 
\begin{align}
E_l(\o,\bk):=\sum_{j=0}^l\hat{\chi}_{\le j}(\o,\bk)\widehat{W}^j(\o,\bk),\
 ((\o,\bk)\in\R^3).\label{eq_effective_dispersion_relation}
\end{align}
By \eqref{eq_equivalence_IR_cut_off} and \eqref{eq_extension_normal_self_energy},
if $(\o,\bk)\in \cM\times \G^*$, 
\begin{align}
E_l(\o,\bk)=\sum_{j=0}^l\chi_{\le j}(\o,\bk)W^j(\o,\bk).\label{eq_extended_dispersion_discrete}
\end{align}
Define the covariance $C_l:I_0^2\to\C$ by
\begin{align}
C_l(\rho\bx \s x,\eta \by \tau y)
:=& \frac{\delta_{\s,\tau}}{\beta L^2}\sum_{(\o,\bk)\in \cM_h\times
 \G^*}e^{i\<\bx-\by,\bk\>}e^{i(x-y)\o}\label{eq_effective_covariance}\\
&\cdot\chi_l(\o,\bk)(i\o
 I_4-E(\bk)-E_l(\o,\bk))^{-1}(\rho,\eta).\notag
\end{align}
In the rest of this subsection we mainly study properties of $C_l$.
To this end let us measure the support of the cut-off functions. 
Since $\supp\chi_l(\cdot)\subset \supp\chi_{\le l}(\cdot)\subset
\supp\hat{\chi}_{\le l}(\cdot)$, it is sufficient to measure
$\supp$ $\hat{\chi}_{\le l}(\cdot)$.

\begin{lemma}\label{lem_IR_cut_off_support_measure}
Assume that \eqref{eq_beta_L_temporal_assumption} holds. Then,
there exists a constant $c\in\R_{>0}$ independent of any parameter such
 that the following statements hold true for any $l\in
 \{0,-1,\cdots,N_{\beta}\}$ and $(\o',\bk')\in\R^3$.
\begin{enumerate}
\item\label{item_measure_support_full_discrete}
$$
\frac{1}{\beta L^2}\sum_{\o\in \cM}\sum_{\bk\in\G^*}1_{\hat{\chi}_{\le
     l}(\o+\o',\bk+\bk')\neq 0}\le c f_{\bt}^{-1}M_{IR}^3M^{3l+3}.
$$
\item\label{item_measure_support_semi_discrete}
$$
\int_{-\infty}^{\infty}d\o
\frac{1}{L^2}\sum_{\bk\in\G^*}1_{\hat{\chi}_{\le
     l}(\o,\bk+\bk')\neq 0}\le c f_{\bt}^{-1}M_{IR}^3M^{3l+3}.
$$
\item\label{item_measure_support_no_matsubara}
$$
\frac{1}{L^2}\sum_{\bk\in\G^*}1_{\hat{\chi}_{\le
     l}(\o',\bk+\bk')\neq 0}\le c f_{\bt}^{-1}M_{IR}^2M^{2l+2}.
$$
\end{enumerate}
\end{lemma}
\begin{proof}
All the claims are
 verified by Lemma \ref{lem_infrared_cut_off_measure},
 \eqref{eq_beta_L_temporal_assumption} and the inequality $f_{\bt}\le
 1$.
\end{proof}

\begin{lemma}\label{lem_covariance_infrared_bound}
Assume that 
$$
h\ge \frac{1}{\sqrt{3}}M_{UV},\ L\ge \beta.
$$
Then, there exists a constant
$c\in\R_{>0}$ independent of any parameter such that if 
$$M\ge c,\ \alpha^2\ge cM,$$
the following statements hold true for any
 $l\in\{0,-1,\cdots,N_{\beta}\}$ and $J^j(\psi)\in\cS(j)$ $(j=0,-1,\cdots,l)$.
\begin{enumerate}
\item\label{item_covariance_infrared_analyticity}  
If $c_{IR}\ge f_{\bt}^{-1}$ holds, $\bU\mapsto C_l(\bU)(\bX)$ is continuous in
     $\overline{D}$ and analytic in $D$ $(\forall \bX\in I_0^2)$.
\item\label{item_covariance_infrared_bound} There exists a constant
     $c(M,c_w)\in\R_{\ge 1}$ depending only on $M$ and $c_w$ such that 
if $c_{IR}\ge c(M,c_w)f_{\bt}^{-1}$ holds,
\begin{align*}
&|\det(\<\bp_i,\bq_j\>_{\C^r}C_{l}(\bU)(X_i,Y_j))_{1\le i,j\le n}|\le
 (c_{IR}M^{2l})^n,\\
&(\forall r,n\in\N,\bp_i,\bq_i\in\C^r\text{ with }
\|\bp_i\|_{\C^r},\|\bq_i\|_{\C^r}\le 1,\\
&\quad X_i,Y_i\in I_0\
 (i=1,2,\cdots,n),\bU\in\overline{D}),\\
&\|\widetilde{C_{l}}(\bU)\|_{l-1,r}\le c_{IR}
M^{-l-rl},\ (\forall r\in
 \{0,1\},\bU\in \overline{D}).
\end{align*}
\item\label{item_covariance_infrared_symmetries}
\begin{align*}
\widetilde{C_l}(\bU)(\bX)=e^{iQ_2(S_2(\bX))}\widetilde{C_l}(\bU)(S_2(\bX)),\
 (\forall \bX\in I^2,\bU\in \overline{D}),
\end{align*}
for $S:I\to I$, $Q:I\to \R$ defined by
     \eqref{eq_particle_hole_maps}, \eqref{eq_spin_maps},
     \eqref{eq_spin_reflection_maps}, \eqref{eq_translation_maps},
     \eqref{eq_rotation_maps} respectively.
\item\label{item_covariance_infrared_complex_symmetries}
\begin{align*}
\widetilde{C_l}(\bU)(\bX)=e^{-iQ_2(S_2(\bX))}\overline{\widetilde{C_l}(\overline{\bU})(S_2(\bX))},\
 (\forall \bX\in I^2,\bU\in \overline{D}),
\end{align*}
for $S:I\to I$, $Q:I\to \R$ defined by
     \eqref{eq_hermitian_maps}, \eqref{eq_half_filled_maps} respectively.
\end{enumerate}
In \eqref{item_covariance_infrared_bound}, \eqref{item_covariance_infrared_symmetries},
 \eqref{item_covariance_infrared_complex_symmetries},
 $\widetilde{C_l}:I^2\to\C$ is the anti-symmetric extension of $C_l$
 defined as in \eqref{eq_anti_symmetric_extension_covariance}.
\end{lemma}
\begin{proof}
Note that if $C_l(\bX)=0$ $(\forall \bX\in I_0^2)$, all the claims 
trivially hold
 true. If $C_l(\bX)\neq 0$ for some $\bX\in I_0^2$, there exists
 $(\o,\bk)\in\cM\times \G^*$ such that $\chi_{\le 0}(\o,\bk)\neq
 0$. Thus, by Lemma \ref{lem_beta_inverse_upper_bound} we can always
 assume that $1/\beta \le M_{IR}M^{N_{\beta}+1}$ during the
 proof. Combined with the assumption $L\ge \beta$, this
 means that \eqref{eq_beta_L_temporal_assumption} holds and thus we can refer to the results of Lemma
 \ref{lem_bound_self_energy} and Lemma
 \ref{lem_IR_cut_off_support_measure}. In the following we prove
 \eqref{item_covariance_infrared_analyticity}, 
\eqref{item_covariance_infrared_symmetries},
\eqref{item_covariance_infrared_complex_symmetries} first and 
\eqref{item_covariance_infrared_bound} in the end.

 \eqref{item_covariance_infrared_analyticity}: Assume that $\alpha^2\ge
 2cM$ with the constant $c$ appearing in the right-hand side of the
 inequality in Lemma \ref{lem_bound_self_energy}
 \eqref{item_neumann_justification_extended_self_energy}. Since
 $c_{IR}^{-1}\le f_{\bt}$ and $f_{\bt}\le 1$ by assumption, 
\begin{align*}
&\|(i\o I_4-E(\bk))^{-1}E_{l'}(\bU)(\o,\bk)\|_{4\times
 4}\le\frac{1}{2}f_{\bt}\cdot f_{\bt}^{-\frac{1}{2}}\le\frac{1}{2},\\
&(\forall (\o,\bk)\in\R^3\text{ satisfying }\chi_l(\o,\bk)\neq
 0,l'\in \{0,-1,\cdots,l\},\bU\in \overline{D}).
\end{align*}
Therefore, by Lemma \ref{lem_estimate_free_dispersion_relation}
 \eqref{item_upper_bound_free_propagator} and
 \eqref{eq_support_properties_infrared}, 
\begin{align}
&\|(i\o I_4-E(\bk)-E_{l'}(\bU)(\o,\bk))^{-1}\|_{4\times
 4}\label{eq_infrared_integrant_bound}\\
&\le \sum_{n=0}^{\infty}\|(i\o I_4-E(\bk))^{-1}E_{l'}(\bU)(\o,\bk)\|_{4\times
 4}^n\|(i\o I_4-E(\bk))^{-1}\|_{4\times 4}\notag\\
&\le 2\left(\frac{\pi}{\sqrt{6}}M_{IR}M^l\right)^{-1}\le M^{-l},\notag\\
&(\forall (\o,\bk)\in\R^3\text{ satisfying }\chi_l(\o,\bk)\neq
 0,l'\in \{0,-1,\cdots,l\},\bU\in \overline{D}).\notag
\end{align}
This implies the well-definedness of $C_l$. Since $\bU\mapsto
 E_l(\bU)(\o,\bk)$ is continuous in $\overline{D}$ and analytic in $D$
for any $(\o,\bk)\in\R^3$ by definition, so is the function $\bU\mapsto
 C_l(\bU)(\bX)$ for any $\bX\in I_0^2$.

\eqref{item_covariance_infrared_symmetries}: For any $S:I\to I$, $Q:I\to
 \R$ defined in \eqref{eq_particle_hole_maps}, \eqref{eq_spin_maps},
 \eqref{eq_spin_reflection_maps}, the claimed equality clearly holds.

For $S:I\to I$, $Q:I\to
 \R$ defined in \eqref{eq_translation_maps} with $\bz\in \Z^2$, $s\in
 (1/h)\Z$, 
\begin{align*}
&e^{iQ_2(S_2((\rho,\bx,\s,x,\theta),(\eta,\by,\tau,y,\xi)))}\widetilde{C_l}(S_2((\rho,\bx,\s,x,\theta),(\eta,\by,\tau,y,\xi)))\\
&=(-1)^{n_{\beta}(x+s)+n_{\beta}(y+s)}\\
&\quad\cdot\widetilde{C_l}((\rho,r_L(\bx+\bz),\s,r_{\beta}(x+s),\theta),(\eta,r_L(\by+\bz),\tau,r_{\beta}(y+s),\xi))\\
&=\widetilde{C_l}((\rho,\bx,\s,x,\theta),(\eta,\by,\tau,y,\xi)),\
 (\forall (\rho,\bx,\s,x,\theta),(\eta,\by,\tau,y,\xi)\in I).
\end{align*}

To prove the invariance with $S:I\to I$, $Q:I\to \R$ defined in
 \eqref{eq_rotation_maps}, let us define the map $Z:(2\pi/L)\Z^2\to \Mat(4,\C)$ by
 $Z(\bk)(\rho,\eta):=e^{i\<\be(\rho),\bk\>}\delta_{\rho,\eta}$,
 $(\forall \bk\in
 (2\pi/L)\Z^2,\rho,\eta\in \cB)$. By \eqref{eq_specific_dispersion_relation}
and \eqref{eq_rotation_momentum},
\begin{align*}
&Z(\bk)^*(E(\bk)+E_l(\o,\bk))Z(\bk)=E(-\bk)+E_l(\o,-\bk),\\
&\ (\forall (\o,\bk)\in\cM\times (2\pi/L)\Z^2).
\end{align*}
Therefore,
\begin{align*}
&C_l((\rho,r_L(-\bx-\be(\rho)),\s,x),(\eta,r_L(-\by-\be(\eta)),\tau,y))\\
&=\frac{\delta_{\s,\tau}}{\beta L^2}\sum_{(\o,\bk)\in
 \cM_h\times\G^*}e^{i\<-\bx-\be(\rho)-(-\by-\be(\eta)),\bk\>}e^{i(x-y)\o}\\
&\qquad\qquad\cdot\chi_l(\o,\bk)(i\o I_4-E(\bk)-E_l(\o,\bk))^{-1}(\rho,\eta)\\
&=\frac{\delta_{\s,\tau}}{\beta L^2}\sum_{(\o,\bk)\in
 \cM_h\times\G^*}e^{i\<-\bx+\by,\bk\>}e^{i(x-y)\o}\\
&\qquad\qquad \cdot \chi_l(\o,\bk)(Z(\bk)^{*}(i\o I_4-E(\bk)-E_l(\o,\bk))Z(\bk))^{-1}(\rho,\eta)\\
&=\frac{\delta_{\s,\tau}}{\beta L^2}\sum_{(\o,\bk)\in
 \cM_h\times\G^*}e^{i\<\bx-\by,-\bk\>}e^{i(x-y)\o}\\
&\qquad\qquad \cdot \chi_l(\o,-\bk)(i\o I_4-E(-\bk)-E_l(\o,-\bk))^{-1}(\rho,\eta)\\
&=C_l((\rho,\bx,\s,x),(\eta,\by,\tau,y)),\ (\forall
 (\rho,\bx,\s,x),(\eta,\by,\tau,y)\in I_0).
\end{align*}
This implies the claimed invariance.

\eqref{item_covariance_infrared_complex_symmetries}:
It follows from definition and \eqref{eq_hermitian_momentum} that
\begin{align}
&\overline{E(\bk)(\rho,\eta)}=E(\bk)(\eta,\rho),\ \overline{E_l(\overline{\bU})(-\o,\bk)(\rho,\eta)}=E_l(\bU)(\o,\bk)(\eta,\rho),\label{eq_hermitian_momentum_application}\\
&(\forall (\o,\bk)\in \cM\times (2\pi/L)\Z^2,\rho,\eta\in \cB,\bU\in
 \overline{D}).\notag
\end{align}
Recall that for any $(X,\theta)$, $(Y,\xi)\in I$,
\begin{align*}
&\widetilde{C_l}(\bU)((X,\theta),(Y,\xi))\\
&=\frac{1}{2}
(1_{(\theta,\xi)=(1,-1)}C_l(\bU)(X,Y)-1_{(\theta,\xi)=(-1,1)}C_l(\bU)(Y,X)).
\end{align*}
For $S:I\to I$, $Q:I\to \R$ defined in \eqref{eq_hermitian_maps},
\begin{align}
&e^{-iQ_2(S_2((\rho,\bx,\s,x,\theta),(\eta,\by,\tau,y,\xi)))}\overline{\widetilde{C_l}(\overline{\bU})(S_2((\rho,\bx,\s,x,\theta),(\eta,\by,\tau,y,\xi)))}\label{eq_infrared_hermitian_pre}\\
&=e^{-i\pi(1_{\theta=-1}+1_{\xi=-1}+1_{x\neq 0}+1_{y\neq 0})}\notag\\
&\quad\cdot\overline{\widetilde{C_l}(\overline{\bU})((\rho,\bx,\s,r_{\beta}(-x),-\theta),(\eta,\by,\tau,r_{\beta}(-y),-\xi))}\notag\\
&=\frac{1}{2}(-1)^{1_{x\neq 0}+1_{y\neq 0}}
\bigg(1_{(\theta,\xi)=(1,-1)}\overline{C_l(\overline{\bU})((\eta,\by,\tau,r_{\beta}(-y)),(\rho,\bx,\s,r_{\beta}(-x)))}\notag\\
&\qquad\qquad\qquad -1_{(\xi,\theta)=(1,-1)}\overline{C_l(\overline{\bU})((\rho,\bx,\s,r_{\beta}(-x)),(\eta,\by,\tau,r_{\beta}(-y)))}\bigg).\notag
\end{align}
Note that
\begin{align}
&(-1)^{1_{x\neq 0}+1_{y\neq
 0}}\overline{C_l(\overline{\bU})((\rho,\bx,\s,r_{\beta}(-x)),(\eta,\by,\tau,r_{\beta}(-y)))}\label{eq_infrared_hermitian_pre_another}\\
&=(-1)^{1_{x\neq 0}+1_{y\neq 0}}\frac{\delta_{\s,\tau}}{\beta
 L^2}\sum_{(\o,\bk)\in \cM_h\times
 \G^*}e^{-i\<\bx-\by,\bk\>}e^{-i(r_{\beta}(-x)-r_{\beta}(-y))\o}\notag\\
&\quad\cdot \chi_l(\o,\bk)\bigg(-i\o
 I_4-\overline{E(\bk)}-\overline{E_l(\overline{\bU})(\o,\bk)}\bigg)^{-1}(\rho,\eta)\notag\\
&=\frac{\delta_{\s,\tau}}{\beta
 L^2}\sum_{(\o,\bk)\in \cM_h\times
 \G^*}e^{i\<\by-\bx,\bk\>}e^{i(y-x)\o}\notag\\
&\quad\cdot\chi_l(\o,\bk)\bigg(i\o
 I_4-\overline{E(\bk)}-\overline{E_l(\overline{\bU})(-\o,\bk)}\bigg)^{-1}(\rho,\eta)\notag\\
&=C_l(\bU)((\eta,\by,\tau,y),(\rho,\bx,\s,x)),\notag
\end{align}
where we used \eqref{eq_hermitian_momentum_application}. By substituting
 \eqref{eq_infrared_hermitian_pre_another} into
 \eqref{eq_infrared_hermitian_pre} we obtain the claimed equality with
 $S:I\to I$, $Q:I\to \R$ defined in \eqref{eq_hermitian_maps}.

For $S:I\to I$, $Q:I\to \R$ defined in \eqref{eq_half_filled_maps},
\begin{align}
&e^{-iQ_2(S_2((\rho,\bx,\s,x,\theta),(\eta,\by,\tau,y,\xi)))}\overline{\widetilde{C_l}(\overline{\bU})(S_2((\rho,\bx,\s,x,\theta),(\eta,\by,\tau,y,\xi)))}\label{eq_infrared_half_filled_pre}\\
&=\frac{1}{2}(-1)^{1_{\rho\in \{1,4\}}+1_{\eta\in \{1,4\}}+1}
\bigg(1_{(\theta,\xi)=(1,-1)}\overline{C_l(\overline{\bU})((\eta,\by,\tau,y),(\rho,\bx,\s,x))}\notag\\
&\qquad\qquad\qquad -1_{(\xi,\theta)=(1,-1)}\overline{C_l(\overline{\bU})((\rho,\bx,\s,x),(\eta,\by,\tau,y))}\bigg).\notag
\end{align}
Let us define $Y\in \Mat(4,\C)$ by 
\begin{align}
Y(\rho,\eta):=(-1)^{1_{\rho\in\{1,4\}}}\delta_{\rho,\eta},\
 (\rho,\eta\in\cB).\label{eq_half_filled_transform_matrix}
\end{align}
We can see from definition and \eqref{eq_half_filled_momentum} that
\begin{align*}
&Y(E(\bk)^*+E_l(\overline{\bU})(\o,\bk)^*)Y=-(E(\bk)+E_l(\bU)(\o,\bk)),\\
&\ (\forall (\o,\bk)\in\cM\times (2\pi/L)\Z^2,\bU\in \overline{D}).
\end{align*}
Using this inequality, we observe that
\begin{align}
&(-1)^{1_{\rho\in \{1,4\}}+1_{\eta\in
 \{1,4\}}+1}\overline{C_l(\overline{\bU})((\rho,\bx,\s,x),(\eta,\by,\tau,y))}\label{eq_infrared_half_filled_pre_another}\\
&=(-1)^{1_{\rho\in \{1,4\}}+1_{\eta\in
 \{1,4\}}+1}\frac{\delta_{\s,\tau}}{\beta
 L^2}\sum_{(\o,\bk)\in \cM_h\times
 \G^*}e^{-i\<\bx-\by,\bk\>}e^{-i(x-y)\o}\notag\\
&\quad\cdot \chi_l(\o,\bk)\bigg(-i\o
 I_4-\overline{E(\bk)}-\overline{E_l(\overline{\bU})(\o,\bk)}\bigg)^{-1}(\rho,\eta)\notag\\
&=(-1)^{1_{\rho\in \{1,4\}}+1_{\eta\in
 \{1,4\}}+1}\frac{\delta_{\s,\tau}}{\beta
 L^2}\sum_{(\o,\bk)\in \cM_h\times
 \G^*}e^{i\<\by-\bx,\bk\>}e^{i(y-x)\o}\notag\\
&\quad\cdot \chi_l(\o,\bk)(-i\o
 I_4-E(\bk)^*-E_l(\overline{\bU})(\o,\bk)^*)^{-1}(\eta,\rho)\notag\\
&=\frac{\delta_{\s,\tau}}{\beta
 L^2}\sum_{(\o,\bk)\in \cM_h\times
 \G^*}e^{i\<\by-\bx,\bk\>}e^{i(y-x)\o}\notag\\
&\quad\cdot \chi_l(\o,\bk)(i\o
 I_4+Y(E(\bk)^*+E_l(\overline{\bU})(\o,\bk)^*)Y)^{-1}(\eta,\rho)\notag\\
&=C_l(\bU)((\eta,\by,\tau,y),(\rho,\bx,\s,x)).\notag
\end{align}
By combining \eqref{eq_infrared_half_filled_pre_another} with
 \eqref{eq_infrared_half_filled_pre} we obtain the claimed equality with 
$S:I\to I$, $Q:I\to \R$ defined in \eqref{eq_half_filled_maps}.

\eqref{item_covariance_infrared_bound}:
Recall the definition of the Hilbert space $\cH$ given in the proof of
 Lemma \ref{lem_h_independent_determinant_bound}. For any $(\rho,\bx,\s,x)\in
 I_0$ let us define the vectors $u_{\rho\bx\s x}^l$, $v_{\rho\bx\s
 x}^l\in \cH$ by 
\begin{align*}
&u_{\rho\bx\s
 x}^l(\eta,\bk,\tau,\o):=\delta_{\rho,\eta}\delta_{\s,\tau}e^{-i\<\bx,\bk\>}e^{-ix\o}\chi_l(\o,\bk)^{\frac{1}{2}},\\
&v_{\rho\bx\s
 x}^l(\eta,\bk,\tau,\o)\\
&:=\delta_{\s,\tau}e^{-i\<\bx,\bk\>}e^{-ix\o}\chi_l(\o,\bk)^{\frac{1}{2}}(i\o I_4-E(\bk)-E_l(\o,\bk))^{-1}(\eta,\rho)
,\\
&(\forall (\eta,\bk,\tau,\o)\in \cB\times \G^*\times \spin\times \cM_h).
\end{align*}
It follows that $C_l(X,Y)=\<u_X^l,v_Y^l\>_{\cH}$ $(\forall X,Y\in
 I_0)$. Using Lemma \ref{lem_IR_cut_off_support_measure} \eqref{item_measure_support_full_discrete} and
 \eqref{eq_infrared_integrant_bound}, we can derive that 
\begin{align*}
\|u_X^l\|_{\cH}\le c(M)f_{\bt}^{-\frac{1}{2}}M^{\frac{3}{2}l},\ \|v_X^l\|_{\cH}\le c(M)f_{\bt}^{-\frac{1}{2}}M^{\frac{1}{2}l},\ (\forall
 X\in I_0).
\end{align*}
Therefore, the standard argument based on Gram's inequality concludes that for any $r,n\in \N$,
 $\bp_j,\bq_j\in \C^r$ with $\|\bp_j\|_{\C^r}$, $\|\bq_j\|_{\C^r}\le 1$,
 $X_j,Y_j\in I_0$ $(j=1,2,\cdots,n)$,
\begin{align}
|\det(\<\bp_i,\bq_j\>_{\C^r}C_l(X_i,Y_j))_{1\le i,j\le n}|\le (c(M)f_{\bt}^{-1}M^{2l})^n.\label{eq_IR_determinant_bound_pre}
\end{align}
On the assumption $c_{IR}\ge c(M)f_{\bt}^{-1}$ we obtain the claimed
 determinant bound.

Next let us prove the claimed upper bound on $\|\widetilde{C_l}\|_{l-1,r}$. 
By assumption, $\pi h\ge (\pi/\sqrt{3})M_{UV}$.
Since $\chi_l(\o,\bk)=0$ if $|\o|\ge (\pi/\sqrt{3})M_{UV}$, we can replace $\cM_h$ by
 $\cM$ inside $C_l$. Then, by using the periodicity with $\bk$ we can
 justify the following transformation.
\begin{align}
&\left(\frac{\beta}{2\pi}\right)^{n_0}
 (e^{-i(x-y)\frac{2\pi}{\beta}}-1)^{n_0}
\prod_{j=1}^2\left(\left(\frac{L}{2\pi}\right)^{n_j}
 (e^{-i(x_j-y_j)\frac{2\pi}{L}}-1)^{n_j}\right)\label{eq_infrared_covariance_integration_by_parts}\\
&\quad\cdot C_{l}(\cdot\bx\s x,\cdot\by\tau y)\notag\\
&=\frac{\delta_{\s,\tau}}{\beta
 L^2}\sum_{(\o,\bk)\in \cM\times
 \G^*}e^{i\<\bx-\by,\bk\>}e^{i(x-y)\o}\notag\\
&\quad\cdot\prod_{j=0}^2\cD_j^{n_j}(\chi_{l}(\o,\bk)(i\o
 I_4-E(\bk)-E_l(\o,\bk))^{-1})\notag\\
&=\frac{\delta_{\s,\tau}}{\beta
 L^2}\sum_{(\o,\bk)\in \cM\times
 \G^*}e^{i\<\bx-\by,\bk\>}e^{i(x-y)\o}\notag\\
&\quad\cdot\prod_{j_0=1}^{n_0}\left(\frac{\beta}{2\pi}\int_0^{\frac{2\pi}{\beta}}d\o_{j_0}\right)
\prod_{j_1=1}^{n_1}\left(\frac{L}{2\pi}\int_0^{\frac{2\pi}{L}}dp_{1,j_1}\right)
\prod_{j_2=1}^{n_2}\left(\frac{L}{2\pi}\int_0^{\frac{2\pi}{L}}dp_{2,j_2}\right)
\notag\\
&\quad\cdot\left(\frac{\partial}{\partial \o'}\right)^{n_0}
\prod_{q=1}^2\left(\frac{\partial}{\partial k_q'}\right)^{n_q}(\chi_{l}(\o',\bk')(i\o'
 I_4-E(\bk')-E_l(\o',\bk'))^{-1})\notag\\
&\quad\cdot \Big|_{\o'=\o+\sum_{j_0=1}^{n_0}\o_{j_0},\bk'=\bk+\sum_{j_1=1}^{n_1}p_{1,j_1}\be_1+\sum_{j_2=1}^{n_2}p_{2,j_2}\be_2}.\notag
\end{align}

Note that by Lemma \ref{lem_estimate_free_dispersion_relation}
 \eqref{item_derivative_upper_bound_dispersion}, the definition of
 $\fw(l)$ and the fact $c_w\le 1$, 
\begin{align*}
&\left\|\left(\frac{\partial}{\partial k_j}\right)^n(i\o
 I_4-E(\bk))\right\|_{4\times 4}\le c M^{l-2}(\fw(l)^{-1})^n(n!)^2,\\
&(\forall (\o,\bk)\in\R^3,j\in \{0,1,2\},n\in \N).
\end{align*}
Thus, by Lemma \ref{lem_bound_self_energy}
 \eqref{item_bound_extended_self_energy_derivative} and the inequality
 $c_{IR}\ge 1$,
\begin{align}
&\left\|\left(\frac{\partial}{\partial k_j}\right)^n(i\o
 I_4-E(\bk)-E_{l'}(\o,\bk))\right\|_{4\times 4}\label{eq_IR_integrant_derivative_bound_pre}\\
&\le (cM^{-2}+c\alpha^{-2})
 M^{l}(c \fw(l)^{-1})^n(n!)^2,\notag\\
&(\forall (\o,\bk)\in\R^3,j\in \{0,1,2\},n\in \N,l'\in\{0,-1,\cdots,l\}).\notag
\end{align}
By \eqref{eq_infrared_integrant_bound},
 \eqref{eq_IR_integrant_derivative_bound_pre} and the assumption $\alpha,M\ge c$ we can apply Lemma
 \ref{lem_gevrey_matrix} to deduce that
\begin{align}
&\left\|\left(\frac{\partial}{\partial k_j}\right)^n(i\o
 I_4-E(\bk)-E_{l'}(\o,\bk))^{-1}\right\|_{4\times 4}
\le c M^{-l}(c \fw(l)^{-1})^n(n!)^2,\label{eq_infrared_integrant_bound_derivative}\\
&(\forall (\o,\bk)\in\R^3\text{ satisfying }\chi_l(\o,\bk)\neq
 0,j\in \{0,1,2\},n\in \N\cup \{0\},\notag\\
&\ l'\in \{0,-1,\cdots,l\}).\notag
\end{align}
Lemma \ref{lem_infrared_cut_off_derivative} and
 \eqref{eq_infrared_integrant_bound_derivative} yield 
\begin{align}
&\left\|\left(\frac{\partial}{\partial k_j}\right)^n(\chi_l(\o,\bk)(i\o
 I_4-E(\bk)-E_l(\o,\bk))^{-1})\right\|_{4\times 4}\label{eq_infrared_integrant_bound_derivative_cut_off}\\
&\le c M^{-l}(c \fw(l)^{-1})^n(n!)^2,\ (\forall (\o,\bk)\in\R^3,j\in \{0,1,2\},n\in \N\cup \{0\}).\notag
\end{align}

It follows from Lemma \ref{lem_IR_cut_off_support_measure} \eqref{item_measure_support_full_discrete},
 \eqref{eq_infrared_covariance_integration_by_parts} and
 \eqref{eq_infrared_integrant_bound_derivative_cut_off} that
$$
|d_j(\bX)^n\widetilde{C_l}(\bX)|\le c(M)f_{\bt}^{-1}M^{2l}(c\fw(l)^{-1})^n(n!)^2,
$$
which is also true for $n=0$ by \eqref{eq_IR_determinant_bound_pre}.
Therefore, we reach 
\begin{align}
|\widetilde{C_l}(\bX)|\le c(M)f_{\bt}^{-1}M^{2l}e^{-\sum_{j=0}^2(c
 \fw(l)d_j(\bX))^{1/2}},\ (\forall \bX\in
 I^2).\label{eq_infrared_covariance_decay_bound}
\end{align}
By the equality $\fw(l-1)=\fw(l)M^{-1}$ and the assumption $M\ge c$ we have
\begin{align*}
|e^{\sum_{j=0}^2(
 \fw(l-1)d_j(\bX))^{1/2}}
\widetilde{C_l}(\bX)|\le c(M)f_{\bt}^{-1}M^{2l}e^{-\sum_{j=0}^2(c
 \fw(l)d_j(\bX))^{1/2}},
\end{align*}
which implies that
\begin{align*}
\|\widetilde{C_l}\|_{l-1,r}\le c(M)f_{\bt}^{-1}M^{2l}\fw(l)^{-3-r}\le
 c(M,c_w)f_{\bt}^{-1}M^{-l-rl}.
\end{align*}
Thus, if $c_{IR}\ge c(M,c_w)f_{\bt}^{-1}$, we obtain the claimed upper bound.
\end{proof}

\begin{lemma}\label{lem_covariance_infrared_bound_difference}
Assume that \eqref{eq_beta_h_assumption} holds and 
\begin{align*}
h\ge \frac{1}{\sqrt{3}}M_{UV},\ L\ge \beta_2.
\end{align*} 
Then, there exists a constant $c\in\R_{>0}$ independent of any parameter such
 that 
if $$M\ge c,\ \alpha^2\ge cM,$$
the following statement holds true for any 
$l\in \{0,-1,\cdots,N_{\beta_1}\}$ and
$(J^j(\beta_1)(\psi),J^j(\beta_2)(\psi))\in \tilde{\cS}(j)$
 $(j=0,-1,\cdots,l)$.

There exists a constant $c(M,c_w)\in\R_{\ge 1}$ depending only on $M$,
 $c_w$ such that if $c_{IR}\ge c(M,c_w)f_{\bt}^{-1}$ holds,
\begin{align}
&|\det(\<\bp_i,\bq_j\>_{\C^r}C_{l}(\bU)(\beta_1)(R_{\beta_1}(X_i,Y_j)))_{1\le
  i,j\le n}\label{eq_determinant_bound_difference_IR_real}\\
&\quad-
\det(\<\bp_i,\bq_j\>_{\C^r}C_{l}(\bU)(\beta_2)(R_{\beta_2}(X_i,Y_j)))_{1\le
  i,j\le n}|\notag\\
&\le \beta_1^{-\frac{1}{2}}M^{-l}(c_{IR} M^{2l})^n,\notag\\
&(\forall r,n\in\N,\bp_i,\bq_i\in\C^r\text{ with }
\|\bp_i\|_{\C^r},\|\bq_i\|_{\C^r}\le 1,\notag\\
&\quad X_i,Y_i\in \hat{I}_0\
 (i=1,2,\cdots,n),\bU\in\overline{D}),\notag\\
&|\widetilde{C_l}(\bU)(\beta_1)-\widetilde{C_l}(\bU)(\beta_2)|_{l-1}\le
 \beta_1^{-\frac{1}{2}}c_{IR}M^{-2l},\ (\forall \bU\in\overline{D}).\label{eq_decay_bound_difference_IR_real}
\end{align}
\end{lemma}
\begin{proof}
For any $(\rho,\bx,\s,x),(\eta,\by,\tau,y)\in \hat{I}_0$, $a\in
 \{1,2\}$, set 
\begin{align*}
&C_{ont,l}(\beta_a)(\rho\bx\s x,\eta\by\tau y)\\
&:=(-1)^{n_{\beta_a}(x)+n_{\beta_a}(y)}\frac{\delta_{\s,\tau}}{2\pi
 L^2}\sum_{\bk\in \G^*}\int_{-\pi h}^{\pi h}d\o
 e^{i\<\bx-\by,\bk\>}e^{i(x-y)\o}\\
&\quad\cdot \chi_l(\o,\bk)(iw
 I_4-E(\bk)-E_l(\beta_a)(\o,\bk))^{-1}(\rho,\eta),
\end{align*}
Since $L\ge \beta_2\ge \beta_1\ge 1> \sqrt{3}M_{IR}^{-1}$, it follows
 from the definition of $N_{\beta}$ that $1/L\le 1/\beta_a\le
 M_{IR}M^{N_{\beta_a}+1}$ $(\forall a\in\{1,2\})$.
This means that the condition \eqref{eq_beta_L_temporal_assumption}
 holds for $\beta_1$ and $\beta_2$. Thus, we can use the results of
 Lemma \ref{lem_IR_cut_off_support_measure} for $\beta_1$ and $\beta_2$. 
Note that 
\begin{align*}
&C_{ont,l}(\beta_a)(\cdot \bx \s x,\cdot \by\tau
 y)-C_{l}(\beta_a)(\cdot \bx \s r_{\beta_a}(x),\cdot \by\tau
 r_{\beta_a}(y))\\
&=(-1)^{n_{\beta_a}(x)+n_{\beta_a}(y)}
\frac{\delta_{\s,\tau}}{2\pi L^2}\sum_{\bk\in \G^*}
 e^{i\<\bx-\by,\bk\>}\\
&\quad \cdot \sum_{m=0}^{\frac{\beta_a
 h}{2}-1}\Bigg(\int_{\frac{2\pi}{\beta_a}m+\frac{\pi}{\beta_a}}^{\frac{2\pi}{\beta_a}(m+1)+\frac{\pi}{\beta_a}}d\o\int_{\frac{2\pi}{\beta_a}m+\frac{\pi}{\beta_a}}^{\o}du
 +
 \int_{-\frac{2\pi}{\beta_a}(m+1)-\frac{\pi}{\beta_a}}^{-\frac{2\pi}{\beta_a}m-\frac{\pi}{\beta_a}}d\o\int_{-\frac{2\pi}{\beta_a}m-\frac{\pi}{\beta_a}}^{\o}du
 \Bigg)\\
&\quad\cdot \frac{\partial}{\partial u}\big(e^{i(x-y)u}\chi_{l}(u,\bk)
(iu I_4-E(\bk)-E_l(\beta_a)(u,\bk))^{-1}\big)\\
&\quad + (-1)^{n_{\beta_a}(x)+n_{\beta_a}(y)}
\frac{\delta_{\s,\tau}}{2\pi L^2}\sum_{\bk\in \G^*}
 e^{i\<\bx-\by,\bk\>}\\
&\qquad\cdot\Bigg(\int_{-\frac{\pi}{\beta_a}}^{\frac{\pi}{\beta_a}}d\o
-\int_{\pi h}^{\pi h+\frac{\pi}{\beta_a}}d\o- \int^{-\pi h}_{-\pi h-\frac{\pi}{\beta_a}}d\o\Bigg)\\
&\qquad \cdot e^{i(x-y)\o}\chi_{l}(\o,\bk)
(i\o I_4-E(\bk)-E_l(\beta_a)(\o,\bk))^{-1}.
\end{align*}
By Lemma \ref{lem_IR_cut_off_support_measure}
 \eqref{item_measure_support_semi_discrete}, 
\eqref{item_measure_support_no_matsubara} and 
 \eqref{eq_infrared_integrant_bound_derivative_cut_off} 
we see that
\begin{align}
&\|C_{ont,l}(\beta_a)(\cdot \bx \s x,\cdot \by\tau
 y)-C_{l}(\beta_a)(\cdot \bx \s r_{\beta_a}(x),\cdot \by\tau
 r_{\beta_a}(y))\|_{4\times 4}\label{eq_covariance_difference_IR_temperature}\\
&\le \frac{1}{\beta_a L^2}\sum_{\bk\in
 \G^*}\Bigg(\int_{\frac{\pi}{\beta_a}}^{\pi h+\frac{\pi}{\beta_a}}d\o+
\int^{-\frac{\pi}{\beta_a}}_{-\pi h-\frac{\pi}{\beta_a}}d\o
\Bigg)\notag\\
&\quad \cdot \Bigg(|x-y|\chi_l(\o,\bk)\|(i\o
 I_4-E(\bk)-E_l(\beta_a)(\o,\bk))^{-1}\|_{4\times 4}\notag\\
&\qquad\quad +\left\|\frac{\partial}{\partial \o}
(\chi_l(\o,\bk)(i\o
 I_4-E(\bk)-E_l(\beta_a)(\o,\bk))^{-1})\right\|_{4\times
 4}\Bigg)\notag\\
&\quad + \frac{1}{L^2}\sum_{\bk\in
 \G^*}\left(\int_{-\frac{\pi}{\beta_a}}^{\frac{\pi}{\beta_a}}d\o
+\int_{\pi h}^{\pi h+\frac{\pi}{\beta_a}}d\o
+\int^{-\pi h}_{-\pi h-\frac{\pi}{\beta_a}}d\o\right)\notag\\ 
&\qquad\cdot \chi_l(\o,\bk)\|(i\o
 I_4-E(\bk)-E_l(\beta_a)(\o,\bk))^{-1}\|_{4\times 4}\notag\\
&\le \frac{1}{\beta_1}c(M,c_w)f_{\bt}^{-1}(|x-y|M^{2l}+M^l).\notag
\end{align}
Since $\pi h\ge (\pi/\sqrt{3})M_{UV}$, we can replace the integral over
 $[-\pi h,\pi h]$ inside $C_{ont,l}(\beta_a)$ by the
 integral over $(-\infty,\infty)$. Then, the integration by parts with
 $\o$ and the periodicity with $\bk$ yield 
\begin{align*}
&(x-y)^{n_0}\prod_{j=1}^2\left(\left(\frac{L}{2\pi}\right)^{n_j}(e^{-i(x_j-y_j)\frac{2\pi}{L}}-1)^{n_j}\right)C_{ont,l}(\beta_a)(\cdot\bx\s
 x,\cdot\by\tau y)\\
&=(-1)^{n_{\beta_a}(x)+n_{\beta_a}(y)}\frac{\delta_{\s,\tau}}{2\pi
 L^2}\sum_{\bk\in\G^*}\int_{-\infty}^{\infty}d\o
 e^{i\<\bx-\by,\bk\>}e^{i(x-y)\o}\\
&\quad\cdot
 i^{n_0}\prod_{j_1=1}^{n_1}\left(\frac{L}{2\pi}\int_{0}^{\frac{2\pi}{L}}dp_{1,j_1}\right)\prod_{j_2=1}^{n_2}\left(\frac{L}{2\pi}\int_{0}^{\frac{2\pi}{L}}dp_{2,j_2}\right)\\
&\quad\cdot \left(\frac{\partial}{\partial
 \o}\right)^{n_0}\prod_{q=1}^2\left(\frac{\partial}{\partial
 k_q'}\right)^{n_q}
\big(\chi_l(\o,\bk')(i\o I_4-E(\bk')-E_l(\beta_a)(\o,\bk'))^{-1}\big)\\
&\quad\cdot\Bigg|_{\bk'=\bk+\sum_{j_1=1}^{n_1}p_{1,j_1}\be_1+\sum_{j_2=1}^{n_2}p_{2,j_2}\be_2}.
\end{align*}
It follows from Lemma \ref{lem_IR_cut_off_support_measure}
 \eqref{item_measure_support_semi_discrete}, 
\eqref{eq_infrared_integrant_bound_derivative_cut_off} and this equality that
\begin{align*}
&\left|\hat{d}_j(\bX)^n\widetilde{C_{ont,l}}(\beta_a)(\bX)\right|\le
 c(M)f_{\bt}^{-1}M^{2l}(c\fw(l)^{-1})^n(n!)^2,\\ 
&(\forall \bX\in \hat{I}^2,j\in\{0,1,2\},n\in \N\cup\{0\}), 
\end{align*}
where 
\begin{align*}
&\widetilde{C_{ont,l}}(\beta_a)(X\theta,Y\xi)\\
&:=\frac{1}{2}\big(1_{(\theta,\xi)=(1,-1)}C_{ont,l}(\beta_a)(X,Y)-1_{(\xi,\theta)=(1,-1)}C_{ont,l}(\beta_a)(Y,X)\big),\\
&\quad (\forall X,Y\in \hat{I}_0,\theta,\xi\in \{1,-1\}).
\end{align*}
This leads to 
\begin{align*}
&\left|\widetilde{C_{ont,l}}(\beta_a)(\bX)\right|\le
 c(M)f_{\bt}^{-1}M^{2l}e^{-\sum_{j=0}^2(c\fw(l)\hat{d}_j(\bX))^{1/2}},\\ 
&(\forall \bX\in \hat{I}^2,a\in\{1,2\}).
\end{align*}
By combining this inequality with
 \eqref{eq_infrared_covariance_decay_bound} and using the inequality 
$$d_j(R_{\beta_a}(\bX))\ge \frac{2}{\pi}\hat{d}_j(\bX),\ (\forall \bX\in
 \hat{I}^2,j\in \{0,1,2\}),$$
we obtain
\begin{align}
&\left|\widetilde{C_{ont,l}}(\beta_a)(\bX)-
 \widetilde{C_{l}}(\beta_a)(R_{\beta_a}(\bX))\right|\le
 c(M)f_{\bt}^{-1}M^{2l}e^{-\sum_{j=0}^2(c\fw(l)\hat{d}_j(\bX))^{1/2}},\label{eq_covariance_difference_IR_gevrey_decay}\\
&(\forall \bX\in \hat{I}^2,a\in \{1,2\}).\notag
\end{align}
It follows from \eqref{eq_covariance_difference_IR_temperature},
 \eqref{eq_covariance_difference_IR_gevrey_decay} that
\begin{align}
&\left|\widetilde{C_{ont,l}}(\beta_a)(\bX)-
 \widetilde{C_{l}}(\beta_a)(R_{\beta_a}(\bX))\right|\label{eq_covariance_difference_IR_gevrey_temperature}\\
&\le \beta_1^{-\frac{1}{2}} c(M,c_w)f_{\bt}^{-1}(M^{\frac{1}{2}l}\fw(l)^{\frac{1}{2}}\hat{d}_0(\bX)^{\frac{1}{2}}+M^{\frac{1}{2}l})
M^{l}e^{-\sum_{j=0}^2(c\fw(l)\hat{d}_j(\bX))^{1/2}}\notag\\
&\le \beta_1^{-\frac{1}{2}} c(M,c_w)
 f_{\bt}^{-1}M^{\frac{3}{2}l}e^{-\sum_{j=0}^2(c\fw(l)\hat{d}_j(\bX))^{1/2}},\
 (\forall \bX\in \hat{I}^2,a\in \{1,2\}).\notag
\end{align}

On the other hand, note that
\begin{align}
&(x-y)^{n_0}\prod_{j=1}^2\left(\frac{L}{2\pi}(e^{-i(x_j-y_j)\frac{2\pi}{L}}-1)\right)^{n_j}\label{eq_integration_by_parts_IR_difference}\\
&\quad\cdot (C_{ont,l}(\beta_1)(\cdot \bx \s x,\cdot \by\tau
 y)-C_{ont,l}(\beta_2)(\cdot \bx \s x,\cdot \by\tau
 y))\notag\\
&=(-1)^{n_{\beta_1}(x)+n_{\beta_1}(y)}
\frac{\delta_{\s,\tau}}{2\pi L^2}\sum_{\bk\in \G^*}
 e^{i\<\bx-\by,\bk\>}\Bigg(\int_{\frac{\pi}{\beta_1}}^{\infty}d\o+\int^{-\frac{\pi}{\beta_1}}_{-\infty}d\o\Bigg)e^{i(x-y)\o}i^{n_0}\notag\\
&\quad \cdot
 \prod_{j_1=1}^{n_1}\left(\frac{L}{2\pi}\int_0^{\frac{2\pi}{L}}dp_{1,j_1}\right)
\prod_{j_2=1}^{n_2}\left(\frac{L}{2\pi}\int_0^{\frac{2\pi}{L}}dp_{2,j_2}\right)
\left(\frac{\partial}{\partial \o}\right)^{n_0}
\prod_{q=1}^2\left(\frac{\partial}{\partial k_q'}\right)^{n_q}\notag\\
&\quad\cdot \chi_l(\o,\bk')(i\o
 I_4-E(\bk')-E_l(\beta_1)(\o,\bk'))^{-1}\notag\\
&\quad \cdot (E_l(\beta_1)(\o,\bk')-E_l(\beta_2)(\o,\bk'))(i\o
 I_4-E(\bk')-E_l(\beta_2)(\o,\bk'))^{-1}\notag\\
&\quad\cdot\Bigg|_{\bk'=\bk+\sum_{j_1=1}^{n_1}p_{1,j_1}\be_1+\sum_{j_2=1}^{n_2}p_{2,j_2}\be_2}\notag\\
&+(-1)^{n_{\beta_1}(x)+n_{\beta_1}(y)}
\frac{\delta_{\s,\tau}}{2\pi L^2}\sum_{\bk\in
 \G^*}\int_{-\frac{\pi}{\beta_1}}^{\frac{\pi}{\beta_1}}d\o
 e^{i\<\bx-\by,\bk\>}e^{i(x-y)\o}i^{n_0}\notag\\
&\quad\cdot
\prod_{j_1=1}^{n_1}\left(\frac{L}{2\pi}\int_0^{\frac{2\pi}{L}}dp_{1,j_1}\right)
\prod_{j_2=1}^{n_2}\left(\frac{L}{2\pi}\int_0^{\frac{2\pi}{L}}dp_{2,j_2}\right)
\left(\frac{\partial}{\partial \o}\right)^{n_0}
\prod_{q=1}^2\left(\frac{\partial}{\partial k_q'}\right)^{n_q}\notag\\
&\quad\cdot \chi_l(\o,\bk')((i\o
 I_4-E(\bk')-E_l(\beta_1)(\o,\bk'))^{-1}\notag\\
&\qquad-(i\o
 I_4-E(\bk')-E_l(\beta_2)(\o,\bk'))^{-1})\Bigg|_{\bk'=\bk+\sum_{j_1=1}^{n_1}p_{1,j_1}\be_1+\sum_{j_2=1}^{n_2}p_{2,j_2}\be_2}.\notag
\end{align}
The inequalities
 \eqref{eq_infrared_integrant_bound_derivative},
 \eqref{eq_infrared_integrant_bound_derivative_cut_off}, $c_{IR}\ge 1$,
 $\alpha \ge c$ and Lemma \ref{lem_bound_self_energy_difference} ensure
 that
\begin{align*}
&\Bigg\|\left(\frac{\partial}{\partial k_j}\right)^n(\chi_l(\o,\bk)(i\o
 I_4-E(\bk)-E_l(\beta_1)(\o,\bk))^{-1}\\
&\quad\cdot (E_l(\beta_1)(\o,\bk)-E_l(\beta_2)(\o,\bk))
(i\o I_4-E(\bk)-E_l(\beta_2)(\o,\bk))^{-1})\Bigg\|_{4\times 4}\\
&\le \sum_{m_1=0}^n\left(\begin{array}{c} n \\
			 m_1\end{array}\right)
\Bigg\|\left(\frac{\partial}{\partial
 k_j}\right)^{m_1}\chi_l(\o,\bk)(i\o
 I_4-E(\bk)-E_l(\beta_1)(\o,\bk))^{-1}\Bigg\|_{4\times 4}\\
&\quad\cdot  \sum_{m_2=0}^{n-m_1}\left(\begin{array}{c} n-m_1 \\
			 m_2\end{array}\right)
\Bigg\|\left(\frac{\partial}{\partial
 k_j}\right)^{m_2}(E_l(\beta_1)(\o,\bk)-E_l(\beta_2)(\o,\bk))\Bigg\|_{4\times
 4}\\
&\quad\cdot \Bigg\|\left(\frac{\partial}{\partial
 k_j}\right)^{n-m_1-m_2}(i\o
 I_4-E(\bk)-E_l(\beta_2)(\o,\bk))^{-1}\Bigg\|_{4\times 4}\\
&\le c\beta_1^{-\frac{1}{2}}c_{IR}^{-1}\alpha^{-2}M^{-2l}(c\fw(l)^{-1})^n\sum_{m_1=0}^n\left(\begin{array}{c} n \\
			 m_1\end{array}\right)(m_1!)^2\\
&\quad\cdot \sum_{m_2=0}^{n-m_1}\left(\begin{array}{c} n-m_1 \\
			 m_2\end{array}\right)(m_2!)^2((n-m_1-m_2)!)^2\\
&\le c \beta_1^{-\frac{1}{2}}M^{-2l}(c\fw(l)^{-1})^n(n!)^2,\ (\forall
 j\in\{0,1,2\},n\in \N\cup\{0\}).
\end{align*}
Using this inequality, 
Lemma \ref{lem_IR_cut_off_support_measure}
\eqref{item_measure_support_semi_discrete},\eqref{item_measure_support_no_matsubara}
 and
 \eqref{eq_infrared_integrant_bound_derivative_cut_off}, we can derive
 from \eqref{eq_integration_by_parts_IR_difference} that
\begin{align*}
&\left|\hat{d}_j(\bX)^n\left(\widetilde{C_{ont,l}}(\beta_1)(\bX)-\widetilde{C_{ont,l}}(\beta_2)(\bX)\right)\right|\\
&\le \beta_1^{-\frac{1}{2}} c(M)f_{\bt}^{-1}M^l(c\fw(l)^{-1})^n(n!)^2,\
 (\forall \bX\in \hat{I}^2,j\in\{0,1,2\},n\in \N\cup\{0\}),
\end{align*}
which leads to 
\begin{align}
&\left|\widetilde{C_{ont,l}}(\beta_1)(\bX)-\widetilde{C_{ont,l}}(\beta_2)(\bX)\right|\label{eq_covariance_difference_IR_gevrey_continuous}\\
&\le \beta_1^{-\frac{1}{2}}
 c(M)f_{\bt}^{-1}M^le^{-\sum_{j=0}^2(c\fw(l)\hat{d}_j(\bX))^{1/2}},\
 (\forall \bX\in
 \hat{I}^2).\notag
\end{align}

By combining \eqref{eq_covariance_difference_IR_gevrey_continuous} with \eqref{eq_covariance_difference_IR_gevrey_temperature} we have
\begin{align}
&\left|\widetilde{C_{l}}(\beta_1)(R_{\beta_1}(\bX))-\widetilde{C_{l}}(\beta_2)(R_{\beta_2}(\bX))\right|\label{eq_covariance_difference_IR_decay}\\
&\le \beta_1^{-\frac{1}{2}}
 c(M,c_w)f_{\bt}^{-1}M^le^{-\sum_{j=0}^2(c\fw(l)\hat{d}_j(\bX))^{1/2}},\
 (\forall \bX\in
 \hat{I}^2).\notag
\end{align}
By the equality $\fw(l-1)=\fw(l)M^{-1}$ and the assumption $M\ge c$ we
 can derive from \eqref{eq_covariance_difference_IR_decay} that
\begin{align*}
\left|\widetilde{C_{l}}(\beta_1)-\widetilde{C_{l}}(\beta_2)\right|_{l-1}
\le \beta_1^{-\frac{1}{2}}
 c(M,c_w)f_{\bt}^{-1}M^{-2l}.
\end{align*}
On the assumption $c_{IR}\ge c(M,c_w)f_{\bt}^{-1}$, the above inequality
  gives
 \eqref{eq_decay_bound_difference_IR_real}.

To prove the determinant bound, let us take any $r,n\in \N$,
$\bp_i$, $\bq_i\in \C^r$ satisfying $\|\bp_i\|_{\C^r}$,  $\|\bq_i\|_{\C^r}\le 1$
 and $X_i,Y_i\in \hat{I}_0$ $(i=1,2,\cdots,n)$. By
 expanding along the 1st column and using
 \eqref{eq_covariance_difference_IR_decay},
\begin{align*}
&|\det(\<\bp_i,\bq_j\>_{\C^r}(C_l(\beta_1)(R_{\beta_1}(X_i,Y_j))-C_l(\beta_2)(R_{\beta_2}(X_i,Y_j))))_{1\le
 i,j\le n}|\\
&\le \beta_1^{-\frac{1}{2}} c(M,c_w)f_{\bt}^{-1}M^l\sum_{s=1}^n\\
&\quad \cdot\left|\det(\<\bp_i,\bq_j\>_{\C^r}(C_l(\beta_1)(R_{\beta_1}(X_i,Y_j))-C_l(\beta_2)(R_{\beta_2}(X_i,Y_j))))_{1\le
 i,j\le n\atop i\neq s, j\neq 1}\right|.
\end{align*}
By expanding the remaining determinants by means of the
 Cauchy-Binet formula as in \eqref{eq_application_cauchy_binet} and
using \eqref{eq_IR_determinant_bound_pre} we obtain that
\begin{align*}
&|\det(\<\bp_i,\bq_j\>_{\C^r}(C_l(\beta_1)(R_{\beta_1}(X_i,Y_j))-C_l(\beta_2)(R_{\beta_2}(X_i,Y_j))))_{1\le
 i,j\le n}|\\
&\le \beta_1^{-\frac{1}{2}} M^{-l} (c(M,c_w)f_{\bt}^{-1}M^{2l})^n.
\end{align*}
On the assumption $c_{IR}\ge c(M,c_w)f_{\bt}^{-1}$, this inequality yields
\eqref{eq_determinant_bound_difference_IR_real}. 
\end{proof}

Though the next lemma is not used in our proof of Theorem
\ref{thm_main_theorem}, it enlightens us about why the infrared
integration is necessary to achieve our goal. The discussion in Remark
\ref{rem_no_IR_analysis} was based on this lemma. Here we temporarily
lift the imposition of \eqref{eq_hopping_amplitude_normalized}.

\begin{lemma}\label{lem_temperature_dependency_covariance_bound}
Set $t_{max}:=\max\{t_{h,e},t_{h,o},t_{v,e},t_{v,o}\}$. Assume that
$t_{max}\beta\ge 1$. Then, there exist constants $c_1,c_2,c_3>0$
 independent of any physical parameter such that the following inequality holds
 for any $L\in \N$ satisfying $L\ge c_1 t_{max}\beta$ and $\s\in \spin$.
\begin{align*}
c_2\beta \le \int_0^{\beta}dx \sum_{\bx\in\G}\|C(\cdot\bx\s x,\cdot
 \b0\s 0 )\|_{4\times 4}\le c_3
 t_{max}^{\frac{1}{2}}f_{\bt}^{-1}\beta,
\end{align*}
where $C:(\cB\times\G\times\spin\times[0,\beta))^2\to\C$ is the free
 covariance \eqref{eq_covariance_2_point_function} with the
 free Hamiltonian $H_0$ given by \eqref{eq_free_hamiltonian_another_form}.
\end{lemma}
\begin{proof}
First we prove the claim under the assumption
 \eqref{eq_hopping_amplitude_normalized}. Let $C_{\le 0}^+:I_0^2\to\C$
 be defined by \eqref{eq_covariance_negative_index} with
 $\phi(M_{UV}^{-2}h^2|1-e^{i\frac{\o}{h}}|^2)$ in place of
 $\chi(h|1-e^{i\frac{\o}{h}}|)$ and $E(\cdot)$ defined by
 \eqref{eq_specific_dispersion_relation}. 
Since now the inequality
 \eqref{eq_real_analytic_dispersion_relation} holds with $E_1=4$,
 $E_2=1$, the results of Lemma \ref{lem_covariance_UV_bound} for
 $E_1=4$, $E_2=1$ are available. Then, by recalling Lemma
 \ref{lem_characterization_covariance} and
 \eqref{eq_cut_off_function_equality_UV} we have that
\begin{align}\label{eq_essentially_UV_covariance_bound}
\sup_{Y\in I_0}\frac{1}{h}\sum_{X\in I_0}|C(X,Y)-C_{\le 0}^+(X,Y)|
\le c(M,c_w)\sum_{l=1}^{N_h}M^{-l}\le c(M,c_w),
\end{align}
if $h\ge c$. Define the covariances $C_l':I_0^2\to\C$
 $(l=0,-1,\cdots,N_{\beta})$ by
\begin{align*}
&C_l'(\rho\bx\s x,\eta\by\tau y)\\
&:=\frac{\delta_{\s,\tau}}{\beta
 L^2}\sum_{(\o,\bk)\in \cM_h\times
 \G^*}e^{i\<\bx-\by,\bk\>}e^{i(x-y)\o}\chi_l(\o,\bk)(i \o I_4-E(\bk))^{-1}(\rho,\eta).
\end{align*}
The covariance $C_l'$ is equal to $C_l$ with $J^j(\psi)=0\in\cS(j)$
 $(j=0,-1,\cdots,$ $l)$.
Thus, Lemma \ref{lem_covariance_infrared_bound}
 \eqref{item_covariance_infrared_bound} ensures that
$$
\sup_{Y\in I_0}\frac{1}{h}\sum_{X\in I_0}|C_l'(X,Y)|\le c(M,c_w)f_{\bt}^{-1}M^{-l}
$$
on the assumption that $h\ge c$, $L\ge \beta$ and $M\ge c$. The
 assumption $\beta \ge 1$ implies that $M^{-N_{\beta}}\le
 c(M)\beta$. Therefore, if $h\ge c$, $1\le \beta \le L$, 
\begin{align}\label{eq_essentially_IR_covariance_bound}
\sup_{Y\in I_0}\frac{1}{h}\sum_{X\in I_0}\left|\sum_{l=0}^{N_{\beta}}C_l'(X,Y)\right|
\le c(M,c_w)f_{\bt}^{-1}\beta.
\end{align}
By \eqref{eq_cut_off_function_basic_sum} we can deduce that
for any $(\bx,\s,x)\in\G\times\spin\times[0,\beta)_h$,
\begin{align}\label{eq_difference_IR_covariance_with_without_h}
&\Bigg\|C_{\le 0}^+(\cdot\bx\s x,\cdot\b0\s
 0)-\sum_{l=0}^{N_{\beta}}C_l'(\cdot\bx\s x,\cdot\b0\s
 0)\Bigg\|_{4\times4}\\
&\le \frac{1}{\beta
 L^2}\sum_{(\o,\bk)\in\cM_h\times\G^*}\notag\\
&\qquad\cdot\big(|\phi(M_{UV}^{-2}h^2|1-e^{i\frac{\o}{h}}|^2)-\phi(M_{UV}^{-2}\o^2)|\|h^{-1}(I_4-e^{-i\frac{\o}{h}I_4+\frac{1}{h}E(\bk)})^{-1}\|_{4\times 4}\notag\\
&\qquad\quad
 +\phi(M_{UV}^{-2}\o^2)\|h^{-1}(I_4-e^{-i\frac{\o}{h}I_4+\frac{1}{h}E(\bk)})^{-1}\|_{4\times 4}\|(i\o I_4-E(\bk))^{-1}\|_{4\times 4}\notag\\
&\qquad\qquad\cdot\|h(I_4-e^{-i\frac{\o}{h}I_4+\frac{1}{h}E(\bk)})-
(i\o I_4-E(\bk))\|_{4\times 4}\big)\notag\\
&\le \frac{1}{\beta L^2}\sum_{(\o,\bk)\in \cM_h\times
 \G^*}1_{\phi(M_{UV}^{-2}h^2|1-e^{i\frac{\o}{h}}|^2)\neq 0 \bigvee
 \phi(M_{UV}^{-2}\o^2)\neq 0}(c\beta h^{-2}+c\beta^2 h^{-1})\notag\\
&\le c\beta h^{-2}+c\beta^2 h^{-1}.\notag
\end{align}
Combination of
 \eqref{eq_essentially_UV_covariance_bound},
 \eqref{eq_essentially_IR_covariance_bound}, 
\eqref{eq_difference_IR_covariance_with_without_h} yields that
\begin{align*}
&\frac{1}{h}\sum_{(\bx,x)\in\G\times [0,\beta)_h}\|C(\cdot\bx\s x,
 \cdot\b0\s 0)\|_{4\times 4}\\
&\le
 c(M,c_w)f_{\bt}^{-1}\beta+c(M,c_w)+cL^2\beta^2h^{-2}
+cL^2\beta^3h^{-1}.
\end{align*}
Then, sending $h\to \infty$ and using the inequality
 $f_{\bt}^{-1}\beta\ge 1$ we reach the inequality
\begin{align}\label{eq_full_covariance_temperature_upper_bound}
 \int_0^{\beta}dx \sum_{\bx\in\G}\|C(\cdot\bx\s x,\cdot
 \b0\s 0 )\|_{4\times 4}\le c(M,c_w)
 f_{\bt}^{-1}\beta
\end{align}
on the assumption $1\le \beta \le L$.

Let us prove the lower bound. By
 \eqref{eq_full_covariance_continuous_formula} in Appendix
 \ref{app_h_L_limit}, 
\begin{align*}
C(\cdot\bx\s x,\cdot\b0\s
 0)=\frac{1}{L^2}\sum_{\bk\in\G^*}e^{i\<\bx,\bk\>}e^{xE(\bk)}(I_4+e^{\beta E(\bk)})^{-1}
\end{align*}
for any $(\bx,x)\in\G\times [0,\beta)$.
Consider the case that $L\in 2\N$. Since
 $(\pi,\pi)\in\G^*$ and $E(\pi,\pi)=O$, we have that for any $x\in
 [0,\beta)$
\begin{align*}
\sum_{\bx\in\G}e^{-i\<\bx,(\pi,\pi)\>}C(\cdot\bx\s x,\cdot\b0\s
 0)=\frac{1}{2}I_4.
\end{align*}
Therefore, 
\begin{align}\label{eq_full_covariance_lower_bound_even}
\int_0^{\beta}dx \sum_{\bx\in\G}\|C(\cdot\bx\s x,\cdot
 \b0\s 0 )\|_{4\times 4}
&\ge \int_0^{\beta}dx \left\|\sum_{\bx\in\G}e^{-i\<\bx,(\pi,\pi)\>}C(\cdot\bx\s x,\cdot
 \b0\s 0 )\right\|_{4\times 4}\\
&=\frac{\beta}{2}.\notag
\end{align}
On the other hand, if $L\in 2\N+1$,
 $(\pi-\frac{\pi}{L},\pi-\frac{\pi}{L})\in\G^*$. For any $x\in
 [0,\beta)$,
\begin{align*}
&\sum_{\bx\in\G}e^{-i\<\bx,(\pi-\frac{\pi}{L},\pi-\frac{\pi}{L})\>}C(\cdot\bx\s x,\cdot
 \b0\s 0 )\\
&=e^{xE(\pi-\frac{\pi}{L},\pi-\frac{\pi}{L})}(I_4+e^{\beta
 E(\pi-\frac{\pi}{L},\pi-\frac{\pi}{L})})^{-1}.
\end{align*}
Moreover, by using Lemma \ref{lem_estimate_free_dispersion_relation}
 \eqref{item_derivative_upper_bound_dispersion} we can prove that
\begin{align*}
&\left\|\sum_{\bx\in\G}e^{-i\<\bx,(\pi-\frac{\pi}{L},\pi-\frac{\pi}{L})\>}C(\cdot\bx\s
 x,\cdot \b0\s 0 )\right\|_{4\times 4}\\
&\ge \frac{1}{2}-\left\|e^{xE(\pi-\frac{\pi}{L},\pi-\frac{\pi}{L})}(I_4+e^{\beta
 E(\pi-\frac{\pi}{L},\pi-\frac{\pi}{L})})^{-1}-\frac{1}{2}I_4\right\|_{4\times
 4}\\
&\ge \frac{1}{2}-\int_{-\frac{\pi}{L}}^0dr\left\|\frac{d}{dr}e^{xE(\pi+r,\pi+r)}(I_4+e^{\beta
 E(\pi+r,\pi+r)})^{-1}\right\|_{4\times 4}\\
&\ge \frac{1}{2}-c\beta L^{-1}.
\end{align*}
This implies that 
\begin{align}\label{eq_full_covariance_lower_bound_odd}
\int_0^{\beta}dx\sum_{\bx\in\G}\|C(\cdot\bx\s
 x,\cdot \b0\s 0 )\|_{4\times 4}\ge
 \frac{\beta}{2}-c\beta^2L^{-1}.
\end{align}
By assuming that $1\le \beta$ and $c\beta \le L$ we obtain
 from \eqref{eq_full_covariance_temperature_upper_bound},
 \eqref{eq_full_covariance_lower_bound_even}, 
\eqref{eq_full_covariance_lower_bound_odd} that
\begin{align*}
c\beta \le \int_0^{\beta}dx\sum_{\bx\in\G} \|C(\cdot\bx\s
 x,\cdot \b0\s 0 )\|_{4\times 4}\le c(M,c_w)f_{\bt}^{-1}\beta.
\end{align*}
Thus, the claim has been proved under the assumption
 \eqref{eq_hopping_amplitude_normalized}.

Let us admit that the claim is true under the assumption
 \eqref{eq_hopping_amplitude_normalized} and derive the result in the
 general case. Let
 $\hat{C}(\beta):(\cB\times\G\times\spin\times[0,\beta))^2\to\C$ be the
 covariance defined by \eqref{eq_covariance_2_point_function} with the
 free Hamiltonian $\frac{1}{t_{max}}H_0$. It follows that if $\beta\ge
 1$ and $L\ge c_1\beta$, 
\begin{align*}
c_2\beta \le \int_0^{\beta}dx \sum_{\bx\in\G}\|\hat{C}(\beta)(\cdot\bx\s x,\cdot
 \b0\s 0 )\|_{4\times 4}\le c_3
 f_{\bt/t_{max}}^{-1}\beta.
\end{align*}
We can see from the definition that 
\begin{align*}
\hat{C}(\beta)(\cdot\bx\s x,\cdot\b0\s
 0)=C(\beta/t_{max})(\cdot\bx\s(x/t_{max}),\cdot\b0\s0)
\end{align*}
for any $(\bx,\s,x)\in\G\times \spin\times [0,\beta)$. Therefore, if
 $t_{max}\beta\ge 1$ and $L\ge c_1t_{max}\beta$,
\begin{align*}
c_2t_{max}\beta \le \int_0^{t_{max}\beta}dx \sum_{\bx\in\G}\|C(\beta)(\cdot\bx\s (x/t_{max}),\cdot
 \b0\s 0 )\|_{4\times 4}\le c_3
 f_{\bt/t_{max}}^{-1}t_{max}\beta,
\end{align*}
or
\begin{align*}
c_2\beta \le \int_0^{\beta}dx \sum_{\bx\in\G}\|C(\beta)(\cdot\bx\s x,\cdot
 \b0\s 0 )\|_{4\times 4}\le c_3
 f_{\bt/t_{max}}^{-1}\beta=c_3t_{max}^{\frac{1}{2}}f_{\bt}^{-1}\beta.
\end{align*}
\end{proof}

\subsection{Application of the generalized infrared integration}
\label{subsec_application_IR}
Since we have studied the properties of the prototypical covariance for the infrared
integration in the previous subsection, we can apply the general
infrared analysis summarized in Subsection \ref{subsec_IR_general} and
Subsection \ref{subsec_IR_general_difference} to prove Theorem
\ref{thm_main_theorem}. We follow several steps until we reach the proof
of the main theorem. Since some technicalities arise in how to choose 
the constant $c_{IR}$ and the domain $D$ of $\C^4$, let us write
$\cS(c_{IR},D)(l)$, $\tilde{\cS}(c_{IR},D)(l)$ in place of $\cS(l)$,
$\tilde{\cS}(l)$ respectively. We also write $\cS(\beta)(c_{IR},$ $D)(l)$
instead of $\cS(c_{IR},D)(l)$ when we want to indicate the dependency on
$\beta$ as well. Throughout this subsection we assume that
\begin{align}\label{eq_final_parameter_conditions}
L\ge \beta,\ h\ge e^{8}.
\end{align}
These are the conditions required in 
Proposition \ref{prop_UV_integration},
Lemma \ref{lem_covariance_infrared_bound} 
and Lemma \ref{lem_covariance_infrared_bound_difference},
since now $E_1=4$, $M_{UV}=10\sqrt{6}/\pi$. 

We use the output of the UV integration as the
input of the infrared integration. Thus, we need to confirm the next
statement as the first step.
\begin{lemma}\label{lem_input_IR_ok} 
Let $c\ (\in\R_{>0})$ be the generic positive constant and  $c_0$,
 $c_0'\ (\in\R_{\ge 1})$ be the $M$-dependent constants appearing in Proposition
 \ref{prop_UV_integration}.
Assume that 
$$
M\ge \max\{M_{UV},c\},\ \alpha^2\ge cM
$$
and set
\begin{align*}
D_{UV}:=\{(U_1,U_2,U_3,U_4)\in \C^4\ |\ |U_{\rho}|<
 (c(c_0+{c_0'})^2\alpha^4)^{-1},\ (\forall \rho\in\cB)\}.
\end{align*}
Let $J^{+,0}(\psi)$, $J^{-,0}(\psi)\in\bigwedge \cV$ be the Grassmann
 polynomials defined in the beginning of Subsection \ref{subsec_application_UV}. Set
$$
J^0(\psi):=\frac{1}{2}(J^{+,0}(\psi)+J^{-,0}(\psi)).
$$ 
Then, 
$$J^0(\psi)\in \cS(c_0,D_{UV})(0).$$
Moreover, on the assumption \eqref{eq_beta_h_assumption},
$$
(J^0(\beta_1)(\psi),J^0(\beta_2)(\psi))\in \tilde{\cS}(c_0,D_{UV})(0).
$$
\end{lemma}

\begin{remark} One necessary condition for a Grassmann polynomial to
 belong to $\cS(c_0,D_{UV})(0)$ is the invariance with $S:I\to I$,
 $Q:I\to \R$ defined in \eqref{eq_half_filled_maps}. It is shown below
 that the polynomial $J^0(\psi)$ satisfies this invariance, while
 $J^{+,0}(\psi)$ or $J^{-,0}(\psi)$ cannot be proved to have this
 invariance by itself. For this reason it is more convenient to deal with
 $J^0(\psi)$ than $J^{+,0}(\psi)$, $J^{-,0}(\psi)$ as the input to the
 infrared integration. The adoption of
 $J^0(\psi)$ in place of $J^{+,0}(\psi)$, $J^{-,0}(\psi)$ is justified
 by Lemma \ref{lem_symmetric_formulation}.
\end{remark}

\begin{proof}[Proof of Lemma \ref{lem_input_IR_ok}]
By Proposition
 \ref{prop_UV_integration} \eqref{item_UV_bound},
 \eqref{item_UV_analyticity} we see that $J^0(\psi)$ satisfies the
 conditions \eqref{item_IR_set_analytic}, \eqref{item_IR_set_bound} of
 $\cS(c_0,D_{UV})(0)$. Moreover, on the assumption
 \eqref{eq_beta_h_assumption}, Proposition \ref{prop_UV_integration}
 \eqref{item_UV_difference} implies that
 $(J^0(\beta_1)(\psi),J^0(\beta_2)(\psi))\in \tilde{\cS}(c_0,D_{UV})(0)$. 

It
 remains to check that $J^0(\psi)$ satisfies the invariant properties
 \eqref{item_IR_set_invariance}, \eqref{item_IR_set_complex_invariance}.
 Let $J^{\delta,l}(\psi)$ 
 $(l\in\{0,1,\cdots,N_h\},\delta\in\{+,-\})$ be the Grassmann
 polynomials defined in the beginning of Subsection
 \ref{subsec_application_UV}.
 Since $J^{\delta,N_h}(\psi)=-V^{\delta}(\psi)$, by recalling \eqref{eq_signed_interaction_polynomial} we can check that
$J^{\delta, N_h}(\psi)$ satisfies the invariant properties
 \eqref{item_IR_set_invariance}, \eqref{item_IR_set_complex_invariance}
 except the invariance with $S:I\to I$, $Q:I\to \R$ defined in
 \eqref{eq_half_filled_maps}. In the same way as in the proof of Lemma
 \ref{lem_covariance_infrared_bound}
 \eqref{item_covariance_infrared_symmetries},
 \eqref{item_covariance_infrared_complex_symmetries} we can show that
 for any $l\in\{0,1,\cdots,N_h\}$, $\delta \in \{+,-\}$,
\begin{align*}
\widetilde{C_l^{\delta}}(\bU)(\bX)=e^{iQ_2(S_2(\bX))}\widetilde{C_l^{\delta}}(\bU)(S_2(\bX)),\
 (\forall \bX\in I^2,\bU\in \overline{D_{UV}}),
\end{align*}
with  $S:I\to I$, $Q:I\to \R$ defined in \eqref{eq_particle_hole_maps},
 \eqref{eq_spin_maps}, \eqref{eq_spin_reflection_maps},
 \eqref{eq_translation_maps}, \eqref{eq_rotation_maps} respectively, and 
\begin{align*}
\widetilde{C_l^{\delta}}(\bU)(\bX)=e^{-iQ_2(S_2(\bX))}
\overline{\widetilde{C_l^{\delta}}(\overline{\bU})(S_2(\bX))},\
 (\forall \bX\in I^2,\bU\in \overline{D_{UV}}),
\end{align*}
with  $S:I\to I$, $Q:I\to \R$ defined in \eqref{eq_hermitian_maps}.
 Therefore, we can
 inductively apply Lemma \ref{lem_free_tree_invariance_general} to
 conclude that $J^{\delta,0}(\psi)$ satisfies the invariant properties
 \eqref{item_IR_set_invariance}, \eqref{item_IR_set_complex_invariance}
except the invariance with $S:I\to I$, $Q:I\to \R$ defined in 
 \eqref{eq_half_filled_maps}, and so does $J^0(\psi)$ by definition.

In the following let $S:I\to I$, $Q:I\to \R$ be those defined in
 \eqref{eq_half_filled_maps}. We see that
\begin{align}
V^{\delta}(\cR \psi)=V^{-\delta}(\psi),\ (\forall \delta \in
 \{+,-\}).\label{eq_original_interaction_shift}
\end{align}
Moreover, for any $(\rho,\bx,\s,x,\theta),(\eta,\by,\tau,y,\xi)\in I$,
\begin{align}
&e^{-iQ_2(S_2((\rho,\bx,\s,x,\theta),(\eta,\by,\tau,y,\xi)))}\overline{\widetilde{C_{>0}^\delta}(S_2((\rho,\bx,\s,x,\theta),(\eta,\by,\tau,y,\xi)))}\label{eq_half_filled_preliminary_UV}\\
&=\frac{1}{2}(-1)^{1_{\rho\in \{1,4\}}+1_{\eta\in \{1,4\}}+1}
\bigg(1_{(\theta,\xi)=(1,-1)}\overline{C_{>0}^{\delta}((\eta,\by,\tau,y),(\rho,\bx,\s,x))}\notag\\
&\qquad\qquad\qquad -1_{(\xi,\theta)=(1,-1)}\overline{C_{>0}^{\delta}((\rho,\bx,\s,x),(\eta,\by,\tau,y))}\bigg).\notag
\end{align}
For the matrix $Y\in \Mat(4,\C)$ defined in
 \eqref{eq_half_filled_transform_matrix}, $YE(\bk)Y=-E(\bk)$. Using
 this equality and $E(\bk)^*=E(\bk)$ we can derive that
\begin{align*}
&(-1)^{1_{\rho\in \{1,4\}}+1_{\eta\in
 \{1,4\}}+1}\overline{C_{>0}^{+}((\rho,\bx,\s,x),(\eta,\by,\tau,y))}\\
&=\frac{\delta_{\s,\tau}}{\beta
 L^2}\sum_{(\o,\bk)\in \cM_h\times
 \G^*}e^{i\<\bx-\by,\bk\>}e^{-i(x-y)\o}\sum_{l=1}^{N_h}\chi_{h,l}(\o)\\
&\quad\cdot h^{-1}(-I_4+e^{i\frac{\o}{h}I_4+\frac{1}{h}YE(\bk)Y})^{-1}(\rho,\eta)\\
&=\frac{\delta_{\s,\tau}}{\beta
 L^2}\sum_{(\o,\bk)\in \cM_h\times
 \G^*}e^{i\<\bx-\by,\bk\>}e^{-i(x-y)\o}\sum_{l=1}^{N_h}\chi_{h,l}(\o)\\
&\quad\cdot h^{-1}(e^{i\frac{\o}{h}I_4-\frac{1}{h}E(\bk)^*}-I_4)^{-1}(\rho,\eta)\\
&=\frac{\delta_{\s,\tau}}{\beta
 L^2}\sum_{(\o,\bk)\in \cM_h\times
 \G^*}e^{-i\<\by-\bx,\bk\>}e^{i(y-x)\o}\sum_{l=1}^{N_h}\chi_{h,l}(\o)
 h^{-1}(e^{i\frac{\o}{h}I_4-\frac{1}{h}\overline{E(\bk)}}-I_4)^{-1}(\eta,\rho)\\
&=C_{>0}^-((\eta,\by,\tau,y),(\rho,\bx,\s,x)),
\end{align*}
which also implies that
\begin{align*}
&(-1)^{1_{\rho\in \{1,4\}}+1_{\eta\in
 \{1,4\}}+1}\overline{C_{>0}^{-}((\rho,\bx,\s,x),(\eta,\by,\tau,y))}\\
&=C_{>0}^+((\eta,\by,\tau,y),(\rho,\bx,\s,x)).
\end{align*}
Substituting these equalities into
 \eqref{eq_half_filled_preliminary_UV} yields
\begin{align}
e^{-iQ_2(S_2(\bX))}
\overline{\widetilde{C_{>0}^{\delta}}(S_2(\bX))}=\widetilde{C_{>0}^{-\delta}}(\bX),\
 (\forall \bX\in I^2,\delta \in \{+,-\}).\label{eq_covariance_UV_half_filled_shift}
\end{align}
We see from the
 definition of Grassmann Gaussian integral and \eqref{eq_covariance_UV_half_filled_shift} that for any $\bX\in I^n$,
\begin{align*}
\int\psi_{\bX}d\mu_{C_{>0}^{\delta}}(\psi)&=e^{-\sum_{\bY\in
 I^2}\widetilde{C_{>0}^{\delta}}(\bY)\frac{\partial}{\partial
 \psi_{\bY}}}\psi_{\bX}\Big|_{\psi=0}\\
&=e^{-\sum_{\bY\in
 I^2}e^{iQ_2(\bY)}\widetilde{C_{>0}^{\delta}}(S_2^{-1}(\bY))\frac{\partial}{\partial
 \psi_{\bY}}}e^{-iQ_n(S_n(\bX))}
\psi_{S_n(\bX)}\Big|_{\psi=0}
\\
&=e^{-\sum_{\bY\in
 I^2}\overline{\widetilde{C_{>0}^{-\delta}}(\bY)}\frac{\partial}{\partial
 \psi_{\bY}}}e^{-iQ_n(S_n(\bX))}\psi_{S_n(\bX)}\Big|_{\psi=0}\\
&=\int
 e^{-iQ_n(S_n(\bX))}\psi_{S_n(\bX)}d\mu_{\overline{C_{>0}^{-\delta}}}(\psi),
\end{align*}
or 
\begin{align}
\int \psi_{\bX}d\mu_{\overline{C_{>0}^{\delta}}}(\psi)=\int e^{iQ_n(S_n(\bX))}\psi_{S_n(\bX)}d\mu_{C_{>0}^{-\delta}}(\psi)
=\int(\cR\psi)_{\bX}d\mu_{C_{>0}^{-\delta}}(\psi).\label{eq_shift_free_preparation}
\end{align}
Proposition \ref{prop_UV_integration}
 \eqref{item_UV_analytic_continuation} states that if $\sup_{\rho\in\cB}|U_{\rho}|$ is
 sufficiently small,
\begin{align*}
J^{\delta,0}(\psi)=\log\left(\int
 e^{-V^{\delta}(\psi+\psi^1)}d\mu_{C_{>0}^{\delta}}(\psi^1)\right).
\end{align*}
By \eqref{eq_shift_free_preparation}, 
\begin{align*}
\overline{J^{\delta,0}(\overline{\bU})}(\psi)&=\log\left(\int
 e^{-V^{\delta}(\psi+\psi^1)}d\mu_{\overline{C_{>0}^{\delta}}}(\psi^1)\right)\\
&=\log\left(\int
 e^{-V^{\delta}(\psi+\cR\psi^1)}d\mu_{C_{>0}^{-\delta}}(\psi^1)\right).
\end{align*}
Then, by \eqref{eq_original_interaction_shift} we obtain 
$$
\overline{J^{\delta,0}(\overline{\bU})}(\cR\psi)=J^{-\delta,0}(\bU)(\psi).
$$
Since $J^{-\delta,0}(\bU)$, $\overline{J^{\delta,0}(\overline{\bU})}(\cR\psi)$ are continuous in
 $\overline{D_{UV}}$ and analytic in $D_{UV}$ with $\bU$, 
the identity theorem and the continuity guarantee this equality for all
 $\bU\in\overline{D_{UV}}$. Therefore, we have
$$
\overline{J^0(\overline{\bU})}(\cR\psi)=J^0(\bU)(\psi),\ (\forall \bU\in\overline{D_{UV}}).
$$
\end{proof}

The aim of our IR integration is to find an analytic
continuation of 
$$
-\frac{1}{\beta L^2}\log\left(\int e^{J^0(\psi)}d\mu_{C_{\le 0}^{\infty}}(\psi)
\right)
$$
into a $(\beta,L,h)$-independent domain of $\bU$ around the origin. 
Here the covariance $C_{\le 0}^{\infty}:I_0^2\to\C$ is defined by
\eqref{eq_covariance_non_positive_index_infinity} with
$\phi(M_{UV}^{-2}\o^2)$ in place of $\chi(|\o|)$.
The next lemma
explains how this aim is achieved.
\begin{lemma}\label{lem_IR_analytic_continuation_expansion}
There exists a constant $c\in\R_{>0}$ independent of any parameter such that if
$M\ge c$, we can choose $J^l(\psi)\in \bigwedge \cV$
 $(l\in\{-1,-2,$ $\cdots, N_{\beta}-1\})$ so that 
the following statements hold true.
\begin{enumerate}
\item\label{item_input_IR_consistency}
There exist constants $c(M,c_w),c'(M,c_w)\in\R_{\ge 1}$ depending
     only on $M$, $c_w$ such that
\begin{align*}
&J^l(\psi)\in \cS(c_{IR},D_{IR})(l),\\
&(\forall l\in\{0,-1,\cdots,N_{\beta}-1\},\alpha\in\R_{>0}\text{ with
 }\alpha^2\ge cM^7),
\end{align*}
where $J^0(\psi)$ is the polynomial set in Lemma \ref{lem_input_IR_ok},
\begin{align*}
&c_{IR}:=c(M,c_w)f_{\bt}^{-1},\\
&D_{IR}:=\left\{(U_1,U_2,U_3,U_4)\in\C^4\ \Big|\ 
|U_{\rho}|< \frac{f_{\bt}^2}{c'(M,c_w)\alpha^4},\ (\forall \rho\in
 \cB)\right\}.
\end{align*}
\item\label{item_IR_necessary_positive_reasoning}
\begin{align*}
&\Re \{\det(I_4-(i\o I_4-E(\bk)-E_{l+1}(\o,\bk))^{-1}\chi_{\le
 l}(\o,\bk)W^l(\o,\bk))\}>0,\\
&(\forall (\o,\bk)\in \cM\times \G^*,l\in
 \{0,-1,\cdots,N_{\beta}\},\alpha\in\R_{>0}\text{ with }\alpha^2\ge cM^7,\\
&\quad\bU\in \overline{D_{IR}}),
\end{align*}
where $E_1(\o,\bk):=0$, $W^l(\o,\bk)$ and $E_{l+1}(\o,\bk)$
     $(l\in\{0,-1,\cdots,$ $N_{\beta}\})$ are derived
     from $J^l(\psi)$ by \eqref{eq_partial_dispersion_relation} and
     \eqref{eq_effective_dispersion_relation} respectively.
\item\label{item_IR_expansion}
There exists a constant $c(\beta,L,M)\in\R_{>0}$ depending only on $\beta$, $L$,
     $M$ such that if $\alpha \ge c(\beta,L,M)$ additionally holds,
$$
\int e^{J^0(\psi)}d\mu_{C_{\le
 0}^{\infty}}(\psi)\in\C\backslash \R_{\le 0}
$$
and
\begin{align}
&-\frac{1}{\beta L^2}\log\left(\int e^{J^0(\psi)}d\mu_{C_{\le
 0}^{\infty}}(\psi)\right)\label{eq_IR_expansion}\\
&=-\frac{1}{\beta
 L^2}\sum_{l=0}^{N_{\beta}-1}J_0^l\notag\\
&\quad -\sum_{l=0}^{N_{\beta}}\frac{2}{\beta
 L^2}\sum_{(\o,\bk)\in \cM\times \G^*}\notag\\
&\qquad\cdot\log\big(\det\big(I_4-(i\o
 I_4-E(\bk)-E_{l+1}(\o,\bk))^{-1}\chi_{\le l}(\o,\bk)W^l(\o,\bk)\big)\big),\notag\\
&(\forall \bU\in \overline{D_{IR}}).\notag
\end{align}
\item\label{item_IR_polynomial_difference}
Assume that \eqref{eq_beta_h_assumption} holds and $L\ge \beta_2$.
Then, 
\begin{align*}
&(J^l(\beta_1)(\psi),J^l(\beta_2)(\psi))\in
 \tilde{\cS}(c_{IR},D_{IR})(l),\\
&(\forall l\in
 \{0,-1,\cdots,N_{\beta_1}\},\alpha\in\R_{>0}\text{ with }\alpha^2\ge cM^7).
\end{align*}
\end{enumerate}
\end{lemma}
\begin{proof}
In the following we will apply Proposition
 \ref{prop_infrared_integration_general},
 Proposition \ref{prop_infrared_integration_difference_general}
for $a_1=2$, $a_2=1$, $a_3=1$, $a_4=1/2$, 
Lemma \ref{lem_covariance_infrared_bound},
Lemma \ref{lem_covariance_infrared_bound_difference} 
and Lemma \ref{lem_input_IR_ok}. Note that there exists a constant
 $c\in\R_{>0}$ independent of any parameter such that for any
 $M,\alpha\in\R_{\ge 1}$ satisfying 
\begin{align}
M^{\frac{1}{2}}\ge c,\ \alpha\ge
 cM^{\frac{7}{2}},\label{eq_condition_for_many_lemmas} 
\end{align}
the claims of these propositions and lemmas hold.
We assume the condition \eqref{eq_condition_for_many_lemmas} during the
 proof. We can replace $c$ in \eqref{eq_condition_for_many_lemmas} by a
larger generic constant without altering the statements of this
 lemma. Such a replacement will be necessary in the proof of the
 claim \eqref{item_IR_necessary_positive_reasoning}.

In order to organize the argument, we introduce the sets $\cR(l)$ $(l\in
 \Z_{\le 0})$ of covariances
  as follows. For a constant $c_1\in\R_{\ge 1}$ and
 a domain $D_o(\subset \C^4)$ satisfying that $\overline{\bU}\in
 \overline{D_o}$ $(\forall \bU\in \overline{D_o})$, where $\overline{D_o}$
 is the closure of $D_o$, a function $C_o:I_0^2\to \C$ belongs to
 $\cR(c_1,D_o)(l)$ if and only if $C_o$ is parameterized by $\bU\in
 \overline{D_o}$ and satisfies the following conditions.
\begin{enumerate}[(i)]
\item $\bU\mapsto C_o(\bU)(\bX)$ is continuous in $\overline{D_o}$ and
      analytic in $D_o$ $(\forall \bX\in I_0^2)$.
\item  
\begin{align*}
&|\det(\<\bp_i,\bq_j\>_{\C^r}C_o(\bU)(X_i,Y_j))_{1\le i,j\le n}|\le
 (c_1 M^{2l})^n,\\
&(\forall r,n\in\N,\bp_i,\bq_i\in\C^r\text{ with }
\|\bp_i\|_{\C^r},\|\bq_i\|_{\C^r}\le 1,\\
&\ X_i,Y_i\in I_0\
 (i=1,2,\cdots,n),\bU\in \overline{D_o}).
\end{align*}
\item
\begin{align*}
\|\widetilde{C_o}(\bU)\|_{l-1,r}\le c_1M^{-l-rl},\ (\forall
      r\in\{0,1\},\bU\in \overline{D_o}).
\end{align*}
\item
Under the notation introduced in Subsection \ref{subsec_invariance_general},
\begin{align*}
\widetilde{C_o}(\bU)(\bX)=e^{iQ_2(S_2(\bX))}\widetilde{C_o}(\bU)(S_2(\bX)),\
 (\forall \bX\in I^2,\bU\in \overline{D_o}),
\end{align*}
for $S:I\to I$, $Q:I\to \R$ defined by \eqref{eq_particle_hole_maps},
      \eqref{eq_spin_maps}, \eqref{eq_spin_reflection_maps},
      \eqref{eq_translation_maps}, \eqref{eq_rotation_maps}
      respectively.
\item
\begin{align*}
\widetilde{C_o}(\bU)(\bX)=e^{-iQ_2(S_2(\bX))}\overline{\widetilde{C_o}(\overline{\bU})(S_2(\bX))},\
 (\forall \bX\in I^2,\bU\in \overline{D_o}),
\end{align*}
for $S:I\to I$, $Q:I\to \R$ defined by \eqref{eq_hermitian_maps},
      \eqref{eq_half_filled_maps} respectively.
\end{enumerate}
Here $\widetilde{C_o}:I^2\to \C$ is the anti-symmetric extension of
 $C_o$ defined as in \eqref{eq_anti_symmetric_extension_covariance}.
 We will also write $\cR(\beta)(c_1,D_o)(l)$ in place of $\cR(c_1,D_o)(l)$
when we want to indicate its $\beta$-dependency.

Let us inductively construct $J^l(\psi)\in\cS(l)$
 $(l=0,-1,\cdots,N_{\beta}-1)$, $C_j\in\cR(l)$
 $(l=0,-1,\cdots,N_{\beta})$. Let $c(M,c_w)$ be the maximum of the
 constants with the same notation appearing in Lemma
 \ref{lem_covariance_infrared_bound} and Lemma
 \ref{lem_covariance_infrared_bound_difference}. 
Recall that $c(M,c_w)$ depends only on $M$,
 $c_w$ and that the
 constant $c_0(\in\R_{\ge 1})$ appearing in Lemma \ref{lem_input_IR_ok} stems from
 Proposition \ref{prop_UV_integration} and depends only on $M$, since
 now $b$, $d$, $M_{UV}$, $E_2$ are fixed constants. As remarked in Remark
 \ref{rem_replacement_possible_UV_constant}, we can replace $c_0$ by
 $\max\{c(M,c_w),c_0\}f_{\bt}^{-1}$ in Proposition
 \ref{prop_UV_integration}. Accordingly Lemma \ref{lem_input_IR_ok}
 ensures that
\begin{align*}
J^0(\psi)\in\cS\big(\max\{c(M,c_w),c_0\}f_{\bt}^{-1},D\big)(0),
\end{align*}
and on the assumption \eqref{eq_beta_h_assumption},
\begin{align*}
(J^0(\beta_1)(\psi),J^0(\beta_2)(\psi))\in
 \tilde{\cS}\big(\max\{c(M,c_w),c_0\}f_{\bt}^{-1},D\big)(0),
\end{align*}
where 
\begin{align*}
D:=\Bigg\{(U_1,U_2,&U_3,U_4)\in\C^4\ \Bigg| \\
&  |U_{\rho}|<\frac{1}{c(\max\{c(M,c_w),c_0\}f_{\bt}^{-1}+c_0')^2\alpha^4},\ (\forall \rho\in\cB)\Bigg\}
\end{align*}
with the constant $c$ appearing in \eqref{eq_condition_for_many_lemmas}.
Now, set 
\begin{align*}
&c_{IR}:=\max\{c(M,c_w),c_0\}f_{\bt}^{-1},\\
&D_{IR}:=\Bigg\{(U_1,U_2,U_3,U_4)\in\C^4\ \Bigg| \\
&\qquad\qquad\qquad
 |U_{\rho}|<\frac{f_{\bt}^2}{c(\max\{c(M,c_w),c_0\}+c_0')^2\alpha^4},\ (\forall \rho\in\cB)\Bigg\}.
\end{align*}
Since $D_{IR}\subset D$, we have $J^0(\psi)\in\cS(c_{IR},D_{IR})(0)$, 
$(J^0(\beta_1)(\psi),J^0(\beta_2)(\psi))$ $\in \tilde{\cS}(c_{IR},D_{IR})(0)$.

Assume that $l\in \{0,-1,\cdots,N_{\beta}\}$ and  
\begin{align*}
J^j(\psi)\in \cS(c_{IR},D_{IR})(j),\ (\forall j\in\{0,-1,\cdots,l\}).
\end{align*}
Using $J^j(\psi)$ $(j=0,-1,\cdots,l)$, we define the covariance
 $C_l:I_0^2\to \C$ by \eqref{eq_effective_covariance}. It follows from
 Lemma \ref{lem_covariance_infrared_bound} and the inequality $c_{IR}\ge
 c(M,c_w)f_{\bt}^{-1}$ that $C_l\in\cR(c_{IR},D_{IR})(l)$.
Then, we define the Grassmann polynomials $F^{l-1}(\psi)$,
 $T^{l-1,(n)}(\psi)$ $(n\in\N_{\ge 2})$, $T^{l-1}(\psi)$,
 $J^{l-1}(\psi)\in\bigwedge \cV$ by
 \eqref{eq_inductive_polynomial_UV_general} with the covariance $C_l$
 and the input $J^l(\psi)-J_0^l-J_2^l(\psi)$. We can apply Proposition
 \ref{prop_infrared_integration_general} for $a_1=2$, $a_2=1$, $a_3=1$,
 $a_4=1/2$ and Lemma \ref{lem_free_tree_invariance_general} to prove
 that $J^{l-1}(\psi)\in \cS(c_{IR},D_{IR})(l-1)$.
The continuity and the analyticity of $J^{l-1}(\psi)$ with $\bU$ can be
 proved by the same argument as in the proof of Proposition
 \ref{prop_UV_integration} \eqref{item_UV_analyticity}.
Thus, we have inductively constructed 
\begin{align*}
&J^l(\psi)\in\cS(c_{IR},D_{IR})(l),\ (l=0,-1,\cdots,N_{\beta}-1),\\
&C_l\in\cR(c_{IR},D_{IR})(l),\ (l=0,-1,\cdots,N_{\beta}).
\end{align*}
Thus, the claim \eqref{item_input_IR_consistency} holds with
 $c'(M,c_w):=c(\max\{c(M,c_w),c_0\}+c_0')^2$. 
 Before proving the claims
 \eqref{item_IR_necessary_positive_reasoning},
 \eqref{item_IR_expansion}, let us show that the claim
 \eqref{item_IR_polynomial_difference} holds true.

\eqref{item_IR_polynomial_difference}: Define the subset
 $\tilde{\cR}(c_{IR},D_{IR})(l)$ of
 $\cR(\beta_1)(c_{IR},D_{IR})(l)\times$\\
$\cR(\beta_2)(c_{IR},D_{IR})(l)$
 $(l=0,-1,\cdots,N_{\beta_1})$ as
 follows. $(C_o(\beta_1),C_o(\beta_2))\in
 \cR(\beta_1)(c_{IR},D_{IR})(l)\times \cR(\beta_2)(c_{IR},D_{IR})(l)$
belongs to $\tilde{\cR}(c_{IR},D_{IR})(l)$ if and only if 
\begin{enumerate}[(i)]
\item
 \begin{align*}
&|\det(\<\bp_i,\bq_j\>_{\C^r}C_{o}(\bU)(\beta_1)(R_{\beta_1}(X_i,Y_j)))_{1\le
  i,j\le n}\\
&\quad-
\det(\<\bp_i,\bq_j\>_{\C^r}C_{o}(\bU)(\beta_2)(R_{\beta_2}(X_i,Y_j)))_{1\le
  i,j\le n}|\notag\\
&\le \beta_1^{-\frac{1}{2}}M^{-l}(c_{IR} M^{2l})^n,\\
&(\forall r,n\in\N,\bp_i,\bq_i\in\C^r\text{ with }
\|\bp_i\|_{\C^r},\|\bq_i\|_{\C^r}\le 1,\\
&\quad X_i,Y_i\in \hat{I}_0\
 (i=1,2,\cdots,n),\bU\in\overline{D_{IR}}).
\end{align*}
\item
\begin{align*}
|\widetilde{C_{o}}(\bU)(\beta_1)-\widetilde{C_{o}}(\bU)(\beta_2)|_l\le
 \beta_1^{-\frac{1}{2}}{c_{IR}}M^{-2l},\ (\forall \bU\in \overline{D_{IR}}),
\end{align*}
where $\widetilde{C_{o}}(\beta_j):I(\beta_j)^2\to \C$ is the
anti-symmetric extension of $C_{o}(\beta_j)$ defined as in
\eqref{eq_anti_symmetric_extension_covariance} for $j=1,2$.
\end{enumerate}

We have already seen that
 $(J^0(\beta_1)(\psi),J^0(\beta_2)(\psi))\in\tilde{\cS}(c_{IR},D_{IR})(0)$.
By Lemma \ref{lem_covariance_infrared_bound_difference} and the
 inequality $c_{IR}\ge c(M,c_w)f_{\bt}^{-1}$, 
$(C_0(\beta_1),C_0(\beta_2))\in \tilde{\cR}(c_{IR},D_{IR})(0)$.

Assume that $l\in \{0,-1,\cdots,N_{\beta_1}\}$ and 
\begin{align*}
&(J^j(\beta_1)(\psi),J^j(\beta_2)(\psi))\in
 \tilde{\cS}(c_{IR},D_{IR})(j),\\
&(C_j(\beta_1),C_j(\beta_2))\in\tilde{\cR}(c_{IR},D_{IR})(j),\ (\forall j\in \{0,-1,\cdots,l\}).
\end{align*}
Then, we can apply Proposition
 \ref{prop_infrared_integration_difference_general} for $a_1=2$,
 $a_2=1$, $a_3=1$, $a_4=1/2$ to conclude that
$$
(J^{l-1}(\beta_1)(\psi),J^{l-1}(\beta_2)(\psi))\in\tilde{\cS}(c_{IR},D_{IR})(l-1).
$$
Moreover, if $l\ge N_{\beta_1}+1$, we can again apply Lemma
 \ref{lem_covariance_infrared_bound_difference} to prove that 
$$
(C_{l-1}(\beta_1),C_{l-1}(\beta_2))\in \tilde{\cR}(c_{IR},D_{IR})(l-1).
$$
Therefore, the claim \eqref{item_IR_polynomial_difference} has been
 proved by induction with $l$.

\eqref{item_IR_necessary_positive_reasoning}:
Take any $(\o,\bk)\in\R^3$ satisfying $|\o|\ge \pi/\beta$ and $\chi_{\le
 l}(\o,\bk)\neq 0$. By the assumption $L\ge \beta$ and Lemma \ref{lem_beta_inverse_upper_bound}, the condition
 \eqref{eq_beta_L_temporal_assumption} holds.
By Lemma \ref{lem_bound_extended_one_scale_self_energy}
 \eqref{item_bound_extended_one_scale_self_energy},
 \eqref{eq_infrared_integrant_bound} and the inequalities
 $c_{IR}^{-1}\le f_{\bt}$, $f_{\bt}\le  1$,
\begin{align*}
&\|(i\o I_4-E(\bk)-E_{l+1}(\o,\bk))^{-1}\chi_{\le
 l}(\o,\bk)\widehat{W}^l(\o,\bk)\|_{4\times
 4}\\
&\le \sum_{j=l}^{N_{\beta}}\chi_j(\o,\bk)\|(i\o
 I_4-E(\bk)-E_{l+1}(\o,\bk))^{-1}\|_{4\times
 4}\|\widehat{W}^l(\o,\bk)\|_{4\times 4}\notag\\
&\le \sum_{j=l}^{N_{\beta}}\chi_j(\o,\bk)M^{-j}c\cdot
 c_{IR}^{-1}f_{\bt}^{-\frac{1}{2}}M^{\frac{1}{2}l+j+1}\alpha^{-2}\notag\\
&\le
 c f_{\bt}\cdot f_{\bt}^{-\frac{1}{2}}
M^{\frac{1}{2}l+1}\alpha^{-2}\le cM\alpha^{-2}.
\end{align*}
Since the left-hand side of the above inequality vanishes if 
$\chi_{\le l}(\o,\bk)=0$, we eventually have
\begin{align}
&\|(i\o I_4-E(\bk)-E_{l+1}(\o,\bk))^{-1}\chi_{\le
 l}(\o,\bk)\widehat{W}^l(\o,\bk)\|_{4\times
 4}\label{eq_IR_determinant_pre_bound}\\
&\le cM\alpha^{-2},\ (\forall (\o,\bk)\in\R^3\text{ with }|\o|\ge
 \pi/\beta).\notag
\end{align}
By the condition \eqref{eq_condition_for_many_lemmas},
\begin{align}
&\left|\det\big(I_4-(i\o I_4-E(\bk)-E_{l+1}(\o,\bk))^{-1}\chi_{\le
 l}(\o,\bk)\widehat{W}^l(\o,\bk)\big)-1\right|
 \label{eq_IR_determinant_pre_another}\\
&\le c M\alpha^{-2},\ (\forall (\o,\bk)\in\R^3\text{ with }|\o|\ge
 \pi/\beta).\notag
\end{align}
The claim follows from this inequality and the replacement of the
 constant $c$ in \eqref{eq_condition_for_many_lemmas} by a
 larger constant if necessary. 

\eqref{item_IR_expansion}: By \eqref{eq_IR_0th_bound} and
 \eqref{eq_IR_higher_order_bound}, 
\begin{align*}
&|J_0^l|\le \frac{N}{h}\alpha^{-3},\\
&\frah^m\sum_{\bX\in I^m}|J_m^l(\bX)|\le
 \frac{N}{h}c_{IR}^{-\frac{m}{2}}M^{-lm}\alpha^{-m},\ (\forall m\in \N).
\end{align*}
Using these inequalities, we have that
\begin{align*}
&|e^{J_0^l}-1|\le e^{\frac{N}{h}\alpha^{-3}}-1.\\
&\left|\int e^{z\sum_{m=4}^NJ_m^l(\psi)}d\mu_{C_l}(\psi)-1\right|\\
&\le
 \sum_{n=1}^{\infty}\frac{1}{n!}\prod_{j=1}^n\Bigg(2\sum_{m_j=4}^N\frah^{m_j}\sum_{\bX\in I^{m_j}}|J_{m_j}^l(\bX)|\Bigg)(c_{IR}M^{2l})^{\frac{1}{2}\sum_{j=1}^nm_j}\\
&\le \sum_{n=1}^{\infty}\frac{1}{n!}\Bigg(2\frac{N}{h}\sum_{m=4}^N\alpha^{-m}
\Bigg)^n\le e^{2\frac{N}{h}\frac{\alpha^{-4}}{1-\alpha^{-1}}}-1,\ (\forall z\in
 \C\text{ with }|z|<2).
\end{align*}
The above inequalities imply that
\begin{align}
&\Re e^{J_0^l}>0,\ (\forall l\in\{0,-1,\cdots,N_{\beta}-1\}),\label{eq_IR_real_part_positive}\\&\Re\int e^{z
 \sum_{m=4}^NJ_m^l(\psi)}d\mu_{C_l}(\psi)>0,\notag\\
&(\forall l\in \{0,-1,\cdots,N_{\beta}\},z\in\C\text{ with
 }|z|<2),\notag
\end{align}
if $\alpha$ is larger than or equal to a constant $c(\beta,L)\in\R_{>0}$ depending only on $\beta$ and $L$. We can especially see from
 \eqref{eq_IR_real_part_positive} that
$$
\log\Bigg(\int e^{z\sum_{m=4}^NJ_m^l(\psi+\psi^1)}d\mu_{C_l}(\psi^1)\Bigg)
$$
is analytic with $z$ in $\{z\in \C\ |\ |z|< 2\}$, and thus
\begin{align}
&\log\Bigg(\int
 e^{\sum_{m=4}^NJ_m^l(\psi+\psi^1)}d\mu_{C_l}(\psi^1)\Bigg)\label{eq_infrared_RG_justification}\\
&=\sum_{n=1}^{\infty}\frac{1}{n!}\left(\frac{d}{dz}\right)^n\log\Bigg(\int
 e^{z\sum_{m=4}^NJ_m^l(\psi+\psi^1)}d\mu_{C_l}(\psi^1)\Bigg)\Bigg|_{z=0}\notag\\
&=J^{l-1}(\psi),\ (\forall l\in \{0,-1,\cdots,N_{\beta}\}),\notag
\end{align}
if $\alpha\ge c(\beta,L)$.

Define the covariances $C_{\le l}$, $D_{\le l}:I_0^2\to \C$
 $(l=0,-1,\cdots,N_{\beta})$ by 
\begin{align*}
&C_{\le l}(\cdot \bx\s x,\cdot \by\tau y):=\frac{\delta_{\s,\tau}}{\beta L^2}\sum_{(\o,\bk)\in \cM_h\times
 \G^*}e^{i\<\bk,\bx-\by\>}e^{i\o(x-y)}\\
&\qquad\qquad\qquad\qquad\qquad\cdot \chi_{\le l}(\o,\bk)(i\o I_4-E(\bk)-E_{l}(\o,\bk))^{-1},\\&D_{\le l}(\cdot \bx\s x,\cdot \by\tau y):=\frac{\delta_{\s,\tau}}{\beta L^2}\sum_{(\o,\bk)\in \cM_h\times
 \G^*}e^{i\<\bk,\bx-\by\>}e^{i\o(x-y)}\\
&\qquad\qquad\qquad\qquad\qquad\cdot \chi_{\le l}(\o,\bk)(i\o I_4-E(\bk)-E_{l+1}(\o,\bk))^{-1},\end{align*}
where $E_1(\o,\bk):=0$. By \eqref{eq_infrared_integrant_bound}, $C_{\le
 l}$, $D_{\le l}$ are well-defined.
Note that by \eqref{eq_cut_off_function_basic_sum},
$D_{\le 0}$ is equal to $C_{\le
 0}^{\infty}$. Moreover, 
\begin{align}
C_l(\bX)+D_{\le l-1}(\bX)=C_{\le l}(\bX),\ (\forall \bX\in
 I_0^2,l\in\{0,-1,\cdots,N_{\beta}+1\}).\label{eq_infrared_covariance_shift}
\end{align}
Introduce the functions $J_{0,2}^l:I_0^2\to \C$ $(l\in
 \{0,-1,\cdots,N_{\beta}\})$ by
\begin{align*}
&J_{0,2}^l(\rho\bx\s x,\eta\by\tau y)\\
&:=\frac{\delta_{\s,\tau}}{\beta
 L^2}\sum_{(\o,\bk)\in\cM_h\times\G^*}1_{\chi_{\le l}(\o,\bk)=0}e^{i\<\bk,\bx-\by\>}e^{i\o(x-y)}W^l(\o,\bk)(\rho,\eta).
\end{align*}
Then, let us define $\tilde{J}_{2}^l(\psi)\in\bigwedge \cV$ by
$$
\tilde{J}_{2}^l(\psi):=\frah^2\sum_{X,Y\in I_0}J_{0,2}^l(X,Y)\psi_{X}\opsi_{Y}.
$$
Since 
\begin{align*}
&\sum_{X\in I_0}J_{0,2}^l(X,Y)D_{\le l}(Z,X)=\sum_{X\in
 I_0}J_{0,2}^l(X,Y)D_{\le l}(X,Z)=0,\\
&(\forall Y,Z\in I_0,l\in \{0,-1,\cdots,N_{\beta}\}),
\end{align*}
we can derive from the definition of the Grassmann Gaussian integral that
\begin{align}
\int e^{J^l(\psi)}d\mu_{D_{\le l}}(\psi)=\int
 e^{J^l(\psi)-\tilde{J}_{2}^l(\psi)}d\mu_{D_{\le
 l}}(\psi),\ (\forall l\in \{0,-1,\cdots,N_{\beta}\}).\label{eq_irrelevant_subtraction}
\end{align}
To approximate $C_{\le l}$, $D_{\le l}$ by invertible matrices, let us
 take $\eps\in\R_{>0}$. Then, define the covariances $C_{\le
 l}^{\eps}$, $D_{\le l}^{\eps}:I_0^2\to \C$
 $(l\in\{0,-1,\cdots,N_{\beta}\})$ by
\begin{align*}
&C_{\le l}^{\eps}(\cdot\bx\s x,\cdot \by\tau y)\\
&:=C_{\le l}(\cdot\bx\s x,\cdot \by\tau y)+
\frac{\delta_{\s,\tau}}{\beta L^2}\sum_{(\o,\bk)\in \cM_h\times
 \G^*}e^{i\<\bk,\bx-\by\>}e^{i\o(x-y)}1_{\chi_{\le l}(\o,\bk)= 0}\eps I_4,\\
&D_{\le l}^{\eps}(\cdot\bx\s x,\cdot \by\tau y)\\
&:=D_{\le l}(\cdot\bx\s x,\cdot \by\tau y)+\frac{\delta_{\s,\tau}}{\beta L^2}\sum_{(\o,\bk)\in \cM_h\times
 \G^*}e^{i\<\bk,\bx-\by\>}e^{i\o(x-y)}1_{\chi_{\le l}(\o,\bk)= 0}\eps I_4.
\end{align*}
Let $C_{\le l}^{\eps,-1}$, $D_{\le l}^{\eps,-1}:I_0^2\to\C$ be the
 inverse matrix of $C_{\le l}^{\eps}$, $D_{\le l}^{\eps}$ respectively.
We can see that
\begin{align*}
&C_{\le l}^{\eps,-1}(\cdot\bx\s x,\cdot \by\tau y)\\
&=\frac{\delta_{\s,\tau}}{\beta L^2h^2}\sum_{(\o,\bk)\in \cM_h\times
 \G^*}e^{i\<\bk,\bx-\by\>}e^{i\o(x-y)}\\
&\quad\cdot\big(1_{\chi_{\le l}(\o,\bk)\neq 0} \chi_{\le l}(\o,\bk)^{-1}(i\o
 I_4-E(\bk)-E_{l}(\o,\bk))+1_{\chi_{\le l}(\o,\bk)= 0}\eps^{-1} I_4\big),\\
&D_{\le l}^{\eps,-1}(\cdot\bx\s x,\cdot \by\tau y)\\
&=\frac{\delta_{\s,\tau}}{\beta L^2h^2}\sum_{(\o,\bk)\in \cM_h\times
 \G^*}e^{i\<\bk,\bx-\by\>}e^{i\o(x-y)}\\
&\quad\cdot\big(1_{\chi_{\le l}(\o,\bk)\neq 0} \chi_{\le l}(\o,\bk)^{-1}(i\o
 I_4-E(\bk)-E_{l+1}(\o,\bk))+1_{\chi_{\le l}(\o,\bk)= 0}\eps^{-1} I_4\big).
\end{align*}
By Lemma
 \ref{lem_properties_one_scale_self_energy}
 \eqref{item_quadratic_term_characterization} and
 \eqref{eq_extended_dispersion_discrete} we have
\begin{align}
&J_2^l(\psi)-\tilde{J}_{2}^l(\psi)-\sum_{X,Y\in I_0}D_{\le
 l}^{\eps,-1}(X,Y)\psi_X\opsi_Y=-\sum_{X,Y\in I_0}C_{\le
 l}^{\eps,-1}(X,Y)\psi_X\opsi_Y,
\label{eq_essence_renormalization}\\
&(\forall l\in\{0,-1,\cdots,N_{\beta}\}).\notag
\end{align}

Here let us recall some basics of Grassmann integration. We can number
 each element of $I_0$ so that $I_0=\{X_j\}_{j=1}^{8L^2\beta h}$. We 
 define the linear map 
$\int\cdot d\psi d\opsi:\bigwedge(\cV\oplus \cV^1)\to \bigwedge
 \cV^1$
as follows. For any $f(\psi^1)\in\bigwedge \cV^1$,
\begin{align*}
&\int f(\psi^1)\opsi_{X_1}\opsi_{X_2}\cdots \opsi_{X_{8L^2\beta
 h}}\psi_{X_1}\psi_{X_2}\cdots \psi_{X_{8L^2\beta h}}d\psi
 d\opsi:=f(\psi^1),\\
&\int f(\psi^1)\psi_{\bX}d\psi d\opsi:=0,\ (\forall \bX\in I^m,m\neq
 16L^2\beta h).
\end{align*}
Then, for any $g\in \bigwedge(\cV\oplus \cV^1)$, $\int g d\psi d\opsi(\in
 \bigwedge \cV^1)$ is uniquely determined by linearity and anti-symmetry.
Let $C_o:I_0^2\to \C$ be an invertible covariance matrix. By taking the
 Grassmann variables $a_1,a_2,\cdots,a_D$ to be
 $\opsi_{X_1},\opsi_{X_2},\cdots,\opsi_{X_{8L^2\beta
 h}}$,$\psi_{X_1},\psi_{X_2},$ $\cdots,\psi_{X_{8L^2\beta h}}$ and the skew
 symmetric matrix $S$ to be
$$
\left(\begin{array}{cc}0 & C_o \\ -C_o^t & 0\end{array}\right)
$$
in \cite[\mbox{Lemma I.10}]{FKT2} we obtain the equality that 
\begin{align*}
&\int e^{\sum_{X\in
 I_0}(\opsi_X^1\opsi_X+\psi_X^1\psi_X)}e^{-\sum_{X,Y\in
 I_0}C_o^{-1}(X,Y)\psi_X\opsi_Y}d\psi d\opsi\\
&\quad \cdot\Big\slash \int e^{-\sum_{X,Y\in
 I_0}C_o^{-1}(X,Y)\psi_X\opsi_Y}d\psi d\opsi\\
&=e^{-\sum_{X,Y\in
 I_0}C_o(X,Y)\opsi_X^1\psi_Y^1}.
\end{align*}
By applying the left derivative 
$$
\frac{\partial}{\partial
 \opsi_{X_{i_1}}^1}\cdots\frac{\partial}{\partial
 \opsi_{X_{i_l}}^1}\frac{\partial}{\partial \psi_{X_{j_m}}^1}\cdots \frac{\partial}{\partial \psi_{X_{j_1}}^1}
$$
to both sides of the above equality and comparing the constant terms we
 obtain that
\begin{align*}
&\int \opsi_{X_{i_1}}\cdots \opsi_{X_{i_l}}\psi_{X_{j_m}}\cdots \psi_{X_{j_1}}
e^{-\sum_{X,Y\in
 I_0}C_o^{-1}(X,Y)\psi_X\opsi_Y}d\psi d\opsi\\
&\quad \cdot\Big\slash \int e^{-\sum_{X,Y\in
 I_0}C_o^{-1}(X,Y)\psi_X\opsi_Y}d\psi d\opsi\\
&=\left\{\begin{array}{ll}\det(C_o(X_{i_p},X_{j_q}))_{1\le p,q\le
   l}&\text{if }l=m,\\
0&\text{if }l\neq m
\end{array}
\right.\\
&=\int \opsi_{X_{i_1}}\cdots \opsi_{X_{i_l}}\psi_{X_{j_m}}\cdots \psi_{X_{j_1}}d\mu_{C_o}(\psi).
\end{align*}
By linearity,
\begin{align}
&\int f(\psi)d\mu_{C_o}(\psi)=\int f(\psi)
e^{-\sum_{X,Y\in
 I_0}C_o^{-1}(X,Y)\psi_X\opsi_Y}d\psi d\opsi\label{eq_grassmann_gaussian_integral_detailed}\\
&\qquad\qquad\qquad\qquad\cdot\Big\slash \int e^{-\sum_{X,Y\in
 I_0}C_o^{-1}(X,Y)\psi_X\opsi_Y}d\psi d\opsi,\notag\\
&(\forall
 f(\psi)\in\bigwedge \cV).\notag
\end{align}
Note that 
\begin{align} 
\int e^{-\sum_{X,Y\in
 I_0}C_o^{-1}(X,Y)\psi_X\opsi_Y}d\psi d\opsi=(-1)^{8L^2 \beta h (8
 L^2\beta h-1)/2}\det C_o^{-1}=\det
 C_o^{-1}.\label{eq_grassmann_gaussian_determinant}
\end{align}

By combining \eqref{eq_irrelevant_subtraction},
\eqref{eq_essence_renormalization}
  with the general equalities
 \eqref{eq_grassmann_gaussian_integral_detailed},
 \eqref{eq_grassmann_gaussian_determinant} we observe that
\begin{align}
&\int e^{J^0(\psi)}d\mu_{C_{\le 0}^{\infty}}(\psi)=\int
 e^{J^0(\psi)}d\mu_{D_{\le 0}}(\psi)
=\int e^{J^0(\psi)-\tilde{J}^0_2(\psi)}d\mu_{D_{\le 0}}(\psi)
\label{eq_IR_scheme_initial}\\
&=\lim_{\eps\searrow 0}
\int e^{J^0(\psi)-\tilde{J}_2^0(\psi)}
e^{-\sum_{X,Y\in
 I_0}D_{\le 0}^{\eps,-1}(X,Y)\psi_X\opsi_Y}d\psi d\opsi\notag\\
&\qquad\qquad \cdot\Big\slash \int e^{-\sum_{X,Y\in
 I_0}D_{\le 0}^{\eps,-1}(X,Y)\psi_X\opsi_Y}d\psi d\opsi\notag\\
&=e^{J_0^0}\lim_{\eps\searrow 0}\int e^{\sum_{m=4}^NJ^0_m(\psi)}
e^{-\sum_{X,Y\in
 I_0}C_{\le 0}^{\eps,-1}(X,Y)\psi_X\opsi_Y}d\psi d\opsi\notag\\
&\qquad\qquad \cdot\Big\slash \int e^{-\sum_{X,Y\in
 I_0}D_{\le 0}^{\eps,-1}(X,Y)\psi_X\opsi_Y}d\psi d\opsi\notag\\
&=e^{J_0^0}\lim_{\eps\searrow 0}\int 
e^{-\sum_{X,Y\in
 I_0}C_{\le 0}^{\eps,-1}(X,Y)\psi_X\opsi_Y}d\psi d\opsi\notag\\
&\qquad\qquad \cdot\Big\slash \int e^{-\sum_{X,Y\in
 I_0}D_{\le 0}^{\eps,-1}(X,Y)\psi_X\opsi_Y}d\psi d\opsi\notag\\
&\quad\cdot \int  e^{\sum_{m=4}^NJ^0_m(\psi)}d\mu_{C_{\le
 0}^{\eps}}(\psi)\notag\\
&=e^{J_0^0}\lim_{\eps\searrow 0}\notag\\
&\quad\cdot\prod_{(\o,\bk)\in\cM_h\times \G^*}
\big(\det\big(1_{\chi_{\le 0}(\o,\bk)\neq 0} \chi_{\le 0}(\o,\bk)^{-1}(i\o
 I_4-E(\bk)-E_{0}(\o,\bk))\notag\\
&\qquad\qquad\qquad\qquad\quad+1_{\chi_{\le 0}(\o,\bk)= 0}\eps^{-1} I_4\big)\big)^2\notag\\
&\qquad\qquad \cdot\Big\slash \prod_{(\o,\bk)\in\cM_h\times \G^*}\big(\det\big(1_{\chi_{\le 0}(\o,\bk)\neq 0} \chi_{\le 0}(\o,\bk)^{-1}(i\o
 I_4-E(\bk))\notag\\
&\qquad\qquad\qquad\qquad\qquad\qquad\quad+1_{\chi_{\le 0}(\o,\bk)= 0}\eps^{-1} I_4\big)\big)^2\notag\\
&\quad\cdot \int  e^{\sum_{m=4}^NJ^0_m(\psi)}d\mu_{C_{\le
 0}^{\eps}}(\psi)\notag\\
&=e^{J_0^0}\prod_{(\o,\bk)\in\cM_h\times \G^*\atop\text{with }\chi_{\le
 0}(\o,\bk)\neq 0}
\big(\det\big(I_4-(i\o I_4-E(\bk))^{-1}\chi_{\le
 0}(\o,\bk)W^0(\o,\bk)\big)\big)^2\notag\\
&\quad\cdot \int  e^{\sum_{m=4}^NJ^0_m(\psi)}d\mu_{C_{\le
 0}}(\psi)\notag\\
&=e^{J_0^0}\prod_{(\o,\bk)\in\cM_h\times \G^*\atop\text{ with }\chi_{\le
 0}(\o,\bk)\neq 0}
\big(\det\big(I_4-(i\o I_4-E(\bk))^{-1}\chi_{\le
 0}(\o,\bk)W^0(\o,\bk)\big)\big)^2\notag\\
&\quad\cdot \int  e^{J^{-1}(\psi)}d\mu_{D_{\le
 -1}}(\psi).\notag
\end{align}
In the last line of \eqref{eq_IR_scheme_initial} we did the following
 transformation, which can be justified by
 \eqref{eq_IR_real_part_positive}, \eqref{eq_infrared_RG_justification},
\eqref{eq_infrared_covariance_shift}, \cite[\mbox{Proposition I.21}]{FKT2}
 and \cite[\mbox{Lemma C.2}]{K3}.
\begin{align*}
\int e^{\sum_{m=4}^NJ_m^0(\psi)}d\mu_{C_{\le 0}}(\psi)&=\int\int
 e^{\sum_{m=4}^NJ_m^0(\psi+\psi^1)}d\mu_{C_0}(\psi^1)d\mu_{D_{\le
 -1}}(\psi)\\
&=\int
 e^{\log(\int e^{\sum_{m=4}^NJ_m^0(\psi+\psi^1)}d\mu_{C_0}(\psi^1))}d\mu_{D_{\le -1}}(\psi)\\
&=\int e^{J^{-1}(\psi)}d\mu_{D_{\le -1}}(\psi).
\end{align*}
By repeating this procedure and using the fact that $C_{\le
 N_{\beta}}=C_{N_{\beta}}$ we obtain the following.
\begin{align}
&\int e^{J^0(\psi)}d\mu_{C_{\le 0}^{\infty}}(\psi)\label{eq_IR_expansion_inside}\\
&= e^{J_0^0+J_0^{-1}+\cdots+J_0^{N_{\beta}+1}}\prod_{l=0}^{N_{\beta}+1}
\Bigg(\prod_{(\o,\bk)\in\cM_h\times \G^*\atop\text{ with }\chi_{\le
 l}(\o,\bk)\neq 0}\notag\\
&\qquad\cdot\det\big(I_4-(i\o
 I_4-E(\bk)-E_{l+1}(\o,\bk))^{-1}\chi_{\le
 l}(\o,\bk)W^l(\o,\bk)\big)^2\Bigg)\notag\\
&\quad \cdot \int e^{J^{N_{\beta}}(\psi)}d\mu_{D_{\le N_{\beta}}}(\psi)\notag\\
&=
 e^{J_0^0+J_0^{-1}+\cdots+J_0^{N_{\beta}}}\prod_{l=0}^{N_{\beta}}\Bigg(\prod_{(\o,\bk)\in\cM_h\times \G^*\atop\text{ with }\chi_{\le
 l}(\o,\bk)\neq 0}\notag\\
&\qquad\cdot\det\big(I_4-(i\o
 I_4-E(\bk)-E_{l+1}(\o,\bk))^{-1}\chi_{\le
 l}(\o,\bk)W^l(\o,\bk)\big)^2\Bigg)\notag\\
&\quad \cdot \int e^{\sum_{m=4}^NJ^{N_{\beta}}_m(\psi)}d\mu_{C_{\le
 N_{\beta}}}(\psi)\notag\\
&=
 e^{J_0^0+J_0^{-1}+\cdots+J_0^{N_{\beta}-1}}\prod_{l=0}^{N_{\beta}}\Bigg(\prod_{(\o,\bk)\in\cM_h\times \G^*\atop\text{ with }\chi_{\le
 l}(\o,\bk)\neq 0}\notag\\
&\qquad\cdot\det\big(I_4-(i\o
 I_4-E(\bk)-E_{l+1}(\o,\bk))^{-1}\chi_{\le
 l}(\o,\bk)W^l(\o,\bk)\big)^2\Bigg).\notag
\end{align}

It follows from \eqref{eq_IR_0th_bound} that
\begin{align}
\left|\sum_{l=0}^{N_{\beta}-1}J_0^l\right|\le\frac{N}{h}c\alpha^{-3},\ 
\left|e^{\sum_{l=0}^{N_{\beta}-1}J_0^l}-1\right|\le
 e^{\frac{N}{h}c \alpha^{-3}}-1.\label{eq_IR_exponential_part_real_positive_pre}
\end{align}
Set 
\begin{align*}
&K:=\{(l,\o,\bk,\s)\in \{0,-1,\cdots,N_{\beta}\}\times \cM_h\times
 \G^*\times \spin\ |\\
&\qquad\qquad \ 
 \chi_{\le l}(\o,\bk)\neq 0\},\\
&K':=\{(l,\o,\bk,\s)\in \{0,-1,\cdots,N_{\beta}\}\times \cM\times
 \G^*\times \spin\ |\\
&\qquad\qquad \ \phi(M_{UV}^{-2}\o^2)\neq 0\}.
\end{align*}
Using \eqref{eq_IR_determinant_pre_bound} and \eqref{eq_IR_determinant_pre_another}, we see that for any
 $Q\subset K$ with $Q\neq \emptyset$,
\begin{align}
&\Bigg|\prod_{(l,\o,\bk,\s)\in Q}\det\big(I_4-(i\o
 I_4-E(\bk)-E_{l+1}(\o,\bk))^{-1}\chi_{\le
 l}(\o,\bk)W^l(\o,\bk)\big)-1\Bigg|\label{eq_IR_determinant_part_real_positive_pre}\\
&=\Bigg|\prod_{(l,\o,\bk,\s)\in Q}\det\big(I_4-(i\o
 I_4-E(\bk)-E_{l+1}(\o,\bk))^{-1}\chi_{\le
 l}(\o,\bk)W^l(\o,\bk)\big)\notag\\
&\qquad -\prod_{(l,\o,\bk,\s)\in Q}1\Bigg|\notag\\
&\le cM\alpha^{-2}\sharp Q(1+cM\alpha^{-2})^{\sharp Q-1}\le cM\alpha^{-2}\sharp K'(1+cM\alpha^{-2})^{\sharp K'-1}.\notag
\end{align}
By \eqref{eq_IR_exponential_part_real_positive_pre} and
 \eqref{eq_IR_determinant_part_real_positive_pre},
\begin{align}
&\Im \sum_{l=0}^{N_{\beta}-1}J_0^l\in (-\pi,\pi),\label{eq_IR_determinant_part_real_positive}\\
&\Re
 e^{\sum_{l=0}^{N_{\beta}-1}J_0^l}>0,\notag\\
&\Re \prod_{(l,\o,\bk,\s)\in Q}\det\big(I_4-(i\o
 I_4-E(\bk)-E_{l+1}(\o,\bk))^{-1}\chi_{\le
 l}(\o,\bk)W^l(\o,\bk)\big)\notag\\
&\quad >0,\ (\forall Q\subset K\text{ with }Q\neq \emptyset),\notag
\end{align}
if $\alpha$ is larger than a constant $c(\beta,L,M)\in\R_{>0}$ depending only on
 $\beta$, $L$, $M$. Since 
$\log e^z=z$ $(\forall z\in\C$ with $\Im z\in(-\pi,\pi))$,
 $z_1z_2\in\C\backslash \R_{\le 0}$, 
$\log (z_1z_2)=\log z_1+\log z_2$ $(\forall
 z_1,z_2\in\C$ with $\Re z_1>0$, $\Re z_2>0$), the condition
 \eqref{eq_IR_determinant_part_real_positive} enables us to deduce the
 claim \eqref{item_IR_expansion} from \eqref{eq_IR_expansion_inside}.
\end{proof}

On the assumption that $M\ge c$, $\alpha^2\ge cM^7$ with the constant $c\in\R_{>0}$ appearing in
Lemma \ref{lem_IR_analytic_continuation_expansion}, we define the
function $J_{end}(\cdot):\overline{D_{IR}}\to \C$ by the right-hand side
of \eqref{eq_IR_expansion}. In the next two lemmas we study properties of
$J_{end}$.
\begin{lemma}\label{lem_infrared_integration_final_bounds}
There exists a constant $c\in\R_{>0}$ independent of any parameter such that if
$$M\ge c,\ \alpha^2\ge cM^7, 
$$
the following statements hold true. 
\begin{enumerate}
\item\label{item_IR_final_analyticity}
$\bU\mapsto J_{end}(\bU)$ is continuous in $\overline{D_{IR}}$ and
     analytic in $D_{IR}$.
\item\label{item_IR_final_bound}
There exists a constant $c(M)\in\R_{>0}$ depending only on $M$ such that
\begin{align*}
|J_{end}(\bU)|\le c(M) f_{\bt}^{-1}\alpha^{-2},\ (\forall
 \bU\in\overline{D_{IR}}).
\end{align*}
\item\label{item_IR_final_difference}
Additionally assume that \eqref{eq_beta_h_assumption} holds and $L\ge
     \beta_2$. Then, there exists a constant $c''(M,c_w)\in\R_{>0}$ depending only
     on $M$, $c_w$ such that
\begin{align*}
|J_{end}(\beta_1)(\bU)-J_{end}(\beta_2)(\bU)|\le
 c''(M,c_w)\beta_1^{-\frac{1}{2}}f_{\bt}^{-1}\alpha^{-2},\ (\forall \bU\in
 \overline{D_{IR}}).
\end{align*}
\end{enumerate}
\end{lemma}
\begin{proof}
\eqref{item_IR_final_analyticity}: Since $J^l(\psi)\in
 \cS(c_{IR},D_{IR})(l)$ $(\forall l\in \{0,-1,\cdots,N_{\beta}-1\})$, 
$$
\bU\mapsto -\frac{1}{\beta L^2}\sum_{l=0}^{N_{\beta}-1}J_0^l(\bU)
$$
is continuous in $\overline{D_{IR}}$ and analytic in $D_{IR}$. By Lemma
 \ref{lem_IR_analytic_continuation_expansion}
 \eqref{item_IR_necessary_positive_reasoning} and the definitions
 \eqref{eq_partial_dispersion_relation},
 \eqref{eq_effective_dispersion_relation}, the second term of $J_{end}$
is also seen to be continuous in $\overline{D_{IR}}$ and analytic in
 $D_{IR}$.

\eqref{item_IR_final_bound}: To estimate the second term of
 $J_{end}$, we can apply Lemma \ref{lem_IR_cut_off_support_measure}
 \eqref{item_measure_support_full_discrete}, since the condition
 \eqref{eq_beta_L_temporal_assumption} is ensured by Lemma
 \ref{lem_beta_inverse_upper_bound}. 
Using \eqref{eq_IR_0th_bound},
 \eqref{eq_IR_determinant_pre_another},
 Lemma \ref{lem_IR_cut_off_support_measure}
 \eqref{item_measure_support_full_discrete} and the inequality
 $f_{\bt}\le 1$,  
we have that
\begin{align*}
|J_{end}|&\le c\sum_{l=0}^{N_{\beta}-1}M^{\frac{7}{2}l}\alpha^{-3}
 +c\sum_{l=0}^{N_{\beta}}f_{\bt}^{-1}M^{3l+3}\sum_{n=1}^{\infty}\frac{1}{n}(cM\alpha^{-2})^n\\
&\le c\alpha^{-3}+cM^4 f_{\bt}^{-1}\alpha^{-2}\le c(M)f_{\bt}^{-1}\alpha^{-2}.
\end{align*}

\eqref{item_IR_final_difference}: Note that by the condition on $L$,
 $\beta_1$, $\beta_2$, the inequality
 \eqref{eq_beta_L_temporal_assumption} holds for $\beta_1$ and
 $\beta_2$. This means that we can use Lemma
 \ref{lem_bound_extended_one_scale_self_energy},
Lemma \ref{lem_bound_self_energy} and Lemma
 \ref{lem_IR_cut_off_support_measure}
\eqref{item_measure_support_semi_discrete},\eqref{item_measure_support_no_matsubara},
in which \eqref{eq_beta_L_temporal_assumption} is assumed, in the following. 

We can decompose $J_{end}(\beta_2)$ as
 follows. 
$$
J_{end}(\beta_2)=J_{end}^1(\beta_2)+J_{end}^2(\beta_2),
$$
where
\begin{align*}
J_{end}^1(\beta_2)&:=-\frac{1}{\beta_2
 L^2}\sum_{l=0}^{N_{\beta_1}-1}J_0^l(\beta_2)-\sum_{l=0}^{N_{\beta_1}}\frac{2}{\beta_2
 L^2}\sum_{(\o,\bk)\in \cM(\beta_2)\times \G^*}\\
&\quad\cdot\log\big(\det\big(I_4-(i\o
 I_4-E(\bk)-E_{l+1}(\beta_2)(\o,\bk))^{-1}\\
&\qquad\qquad\qquad\qquad\cdot\chi_{\le
 l}(\beta_2)(\o,\bk){W}^l(\beta_2)(\o,\bk)\big)\big),\\
J_{end}^2(\beta_2)&:=-\frac{1}{\beta_2
 L^2}\sum_{l=N_{\beta_1}-2}^{N_{\beta_2}-1}J_0^l(\beta_2)-\sum_{l=N_{\beta_1}-1}^{N_{\beta_2}}\frac{2}{\beta_2
 L^2}\sum_{(\o,\bk)\in \cM(\beta_2)\times \G^*}\\
&\quad\cdot\log\big(\det\big(I_4-(i\o
 I_4-E(\bk)-E_{l+1}(\beta_2)(\o,\bk))^{-1}\\
&\qquad\qquad\qquad\qquad\cdot \chi_{\le
 l}(\beta_2)(\o,\bk){W}^l(\beta_2)(\o,\bk)\big)\big).
\end{align*}
By using \eqref{eq_IR_0th_bound}, \eqref{eq_IR_determinant_pre_another} 
 and Lemma \ref{lem_IR_cut_off_support_measure} \eqref{item_measure_support_full_discrete} and recalling the
 definition of $N_{\beta_1}$ we can deduce that
\begin{align}
|J_{end}^2(\beta_2)|&\le c\sum_{l=N_{\beta_1}-2}^{N_{\beta_2}-1}M^{\frac{7}{2}l}\alpha^{-3}
 +c\sum_{l=N_{\beta_1}-1}^{N_{\beta_2}}f_{\bt}^{-1}M^{3l+3}\sum_{n=1}^{\infty}\frac{1}{n}(cM\alpha^{-2})^n\label{eq_IR_final_large_part_bound}
\\
&\le c(M)\beta_1^{-3}f_{\bt}^{-1}\alpha^{-2}.\notag
\end{align}

Next let us find an upper bound on
 $|J_{end}(\beta_1)-J_{end}^1(\beta_2)|$. Note that 
\begin{align}
&|J_{end}(\beta_1)-J_{end}^1(\beta_2)|\label{eq_IR_final_difference_decomposition}\\
&\le \sum_{l=0}^{N_{\beta_1}-1}\left|\frac{1}{\beta_1
 L^2}J_0^l(\beta_1)-\frac{1}{\beta_2 L^2}J_0^l(\beta_2)\right|
 +2\sum_{l=0}^{N_{\beta_1}}\sum_{a=1}^2\notag\\
&\quad\cdot\Bigg|\frac{1}{\beta_a
 L^2}\sum_{(\o,\bk)\in \cM(\beta_a)\times \G^*}\log\big(\det\big(I_4-(i\o
 I_4-E(\bk)-E_{l+1}(\beta_a)(\o,\bk))^{-1}\notag\\
&\qquad\qquad\qquad\qquad\qquad\qquad\qquad\qquad\cdot\chi_{\le
 l}(\beta_a)(\o,\bk)W^l(\beta_a)(\o,\bk)\big)\big)\notag\\
&\qquad -\frac{1}{2\pi L^2}\sum_{\bk\in\G^*}\Bigg(\int_{\frac{\pi}{\beta_a}}^{\pi
 h}du+ \int^{-\frac{\pi}{\beta_a}}_{-\pi h}du\Bigg)\notag\\
&\quad\qquad\qquad\qquad\cdot\log\big(\det\big(I_4-(iu
 I_4-E(\bk)-E_{l+1}(\beta_a)(u,\bk))^{-1}\notag\\
&\qquad\qquad\qquad\qquad\qquad\qquad\qquad\qquad\cdot \chi_{\le
 l}(\beta_a)(u,\bk)\widehat{W}^l(\beta_a)(u,\bk)\big)\big)\Bigg|\notag\\
&\quad +\sum_{l=0}^{N_{\beta_1}}\frac{1}{\pi L^2}\sum_{\bk\in
 \G^*}\Bigg(\int_{\frac{\pi}{\beta_1}}^{\pi h}du+\int_{-\pi h}^{-\frac{\pi}{\beta_1}}du
 \Bigg)\notag\\
&\qquad\cdot\Big|\log\big(\det\big(I_4-(iu
 I_4-E(\bk)-E_{l+1}(\beta_1)(u,\bk))^{-1}\notag\\
&\qquad\qquad\qquad\qquad\qquad\cdot \chi_{\le
 l}(\beta_1)(u,\bk)\widehat{W}^l(\beta_1)(u,\bk)\big)\big)\notag\\
&\quad\qquad -\log\big(\det\big(I_4-(iu
 I_4-E(\bk)-E_{l+1}(\beta_2)(u,\bk))^{-1}\notag\\
&\quad\qquad\qquad\qquad\qquad\qquad\cdot \chi_{\le
 l}(\beta_2)(u,\bk)\widehat{W}^l(\beta_2)(u,\bk)\big)\big)\Big|\notag\\
&\quad +\sum_{l=0}^{N_{\beta_1}}\frac{1}{\pi L^2}\sum_{\bk\in
 \G^*}\Bigg(\int_{\frac{\pi}{\beta_2}}^{\frac{\pi}{\beta_1}}du+\int^{-\frac{\pi}{\beta_2}}_{-\frac{\pi}{\beta_1}}du\Bigg)\notag\\
&\qquad\cdot\Big|\log\big(\det\big(I_4-(iu
 I_4-E(\bk)-E_{l+1}(\beta_2)(u,\bk))^{-1}\notag\\
&\qquad\qquad\qquad\qquad\qquad\cdot \chi_{\le
 l}(\beta_2)(u,\bk)\widehat{W}^l(\beta_2)(u,\bk)\big)\big)\Big|.\notag
\end{align}
The assumption $h\ge e^{8}$ implies that $\pi h\ge (\pi/\sqrt{3})M_{UV}$,
 or $\phi(M_{UV}^{-2}u^2)=0$ $(\forall u\in\R$ with $|u|\ge \pi h)$.
This explains why the integrals with $u$ over a domain outside $[-\pi h,\pi h]$
 vanish in the following calculations. By using
 \eqref{eq_IR_determinant_pre_another}, \eqref{eq_IR_determinant_pre_bound},
Lemma
 \ref{lem_IR_cut_off_support_measure} \eqref{item_measure_support_semi_discrete},
\eqref{eq_infrared_integrant_bound_derivative}, 
Lemma \ref{lem_infrared_cut_off_derivative}, 
Lemma \ref{lem_bound_extended_one_scale_self_energy}
\eqref{item_derivative_extended_one_scale_self_energy} and the
 inequality $c_{IR}\ge f_{\bt}^{-1}$ in this order,
\begin{align}
&\Bigg|\frac{1}{\beta_a
 L^2}\sum_{(\o,\bk)\in \cM(\beta_a)\times \G^*}\log\big(\det\big(I_4-(i\o
 I_4-E(\bk)-E_{l+1}(\beta_a)(\o,\bk))^{-1}\label{eq_IR_final_decomposition_discrete_cont}\\
&\quad\qquad\qquad\qquad\qquad\qquad\qquad\qquad\cdot\chi_{\le
 l}(\beta_a)(\o,\bk)W^l(\beta_a)(\o,\bk)\big)\big)\notag\\
&\quad -\frac{1}{2\pi L^2}\sum_{\bk\in\G^*}
\Bigg(\int^{\pi h}_{\frac{\pi}{\beta_a}}du+\int_{-\pi
											 h}^{-\frac{\pi}{\beta_a}}du\Bigg)\notag\\
&\qquad\quad\cdot\log\big(\det\big(I_4-(iu
 I_4-E(\bk)-E_{l+1}(\beta_a)(u,\bk))^{-1}\notag\\
&\qquad\qquad\qquad\qquad\qquad\cdot \chi_{\le
 l}(\beta_a)(u,\bk)\widehat{W}^l(\beta_a)(u,\bk)\big)\big)\Bigg|\notag\\
&\le \frac{1}{2\pi L^2}\sum_{\bk\in \G^*}\sum_{m=0}^{\frac{\beta_a
 h}{2}-1}\notag\\
&\quad\cdot\Bigg(\int_{\frac{2\pi}{\beta_a}m+\frac{\pi}{\beta_a}}^{\frac{2\pi}{\beta_a}(m+1)+\frac{\pi}{\beta_a}}d\o\int_{\frac{2\pi}{\beta_a}m+\frac{\pi}{\beta_a}}^{\o}du+\int_{-\frac{2\pi}{\beta_a}(m+1)-\frac{\pi}{\beta_a}}^{-\frac{2\pi}{\beta_a}m-\frac{\pi}{\beta_a}}d\o\int_{-\frac{2\pi}{\beta_a}m-\frac{\pi}{\beta_a}}^{\o}du
\Bigg)\notag\\
&\quad \cdot\Bigg| \frac{\partial}{\partial u}\log\big(\det\big(I_4-(iu
 I_4-E(\bk)-E_{l+1}(\beta_a)(u,\bk))^{-1}\notag\\
&\qquad\qquad\qquad\qquad\qquad\cdot \chi_{\le
 l}(\beta_a)(u,\bk)\widehat{W}^l(\beta_a)(u,\bk)\big)\big)\Bigg|\notag\\
&\quad+ \frac{1}{2\pi L^2}\sum_{\bk\in\G^*}\Bigg(\int_{\pi h}^{\pi
 h+\frac{\pi}{\beta_a}}du+\int^{-\pi h}_{-\pi
 h-\frac{\pi}{\beta_a}}du\Bigg)\notag\\
&\qquad\cdot\Big|\log\big(\det\big(I_4-(iu
 I_4-E(\bk)-E_{l+1}(\beta_a)(u,\bk))^{-1}\notag\\
&\qquad\qquad\qquad\qquad\qquad\cdot {\chi}_{\le
 l}(\beta_a)(u,\bk)\widehat{W}^l(\beta_a)(u,\bk)\big)\big)\Big|\notag\\
&\le \frac{c}{\beta_a L^2}\sum_{\bk\in \G^*}\Bigg(\int_{\frac{\pi}{\beta_a}}^{\pi h}du+\int^{-\frac{\pi}{\beta_a}}_{-\pi
 h}du\Bigg)\notag\\
&\quad \cdot \Bigg|\frac{\partial}{\partial u}\det\big(I_4-(iu
 I_4-E(\bk)-E_{l+1}(\beta_a)(u,\bk))^{-1}\notag\\
&\qquad\qquad\qquad\qquad\cdot \chi_{\le
 l}(\beta_a)(u,\bk)\widehat{W}^l(\beta_a)(u,\bk)\big)\Bigg|\notag\\
&\le \frac{c}{\beta_1 L^2}\sum_{\bk\in \G^*}\Bigg(\int_{\frac{\pi}{\beta_a}}^{\pi h}du+\int^{-\frac{\pi}{\beta_a}}_{-\pi
 h}du\Bigg)\sum_{j=l}^{N_{\beta_a}}\notag\\
&\quad \cdot 
\Bigg\|\frac{\partial}{\partial u}\big((iu
 I_4-E(\bk)-E_{l+1}(\beta_a)(u,\bk))^{-1}
\chi_{j}(u,\bk)\widehat{W}^l(\beta_a)(u,\bk)\big)\Bigg\|_{4\times 4}\notag\\
&\le  \frac{c}{\beta_1 L^2}\sum_{\bk\in \G^*}\Bigg(\int_{\frac{\pi}{\beta_a}}^{\pi h}du+\int^{-\frac{\pi}{\beta_a}}_{-\pi
 h}du\Bigg)\sum_{j=l}^{N_{\beta_a}}\notag\\
&\quad \cdot \sum_{n=0}^1\Bigg\|\left(\frac{\partial}{\partial u}\right)^n\big((iu
 I_4-E(\bk)-E_{l+1}(\beta_a)(u,\bk))^{-1}\Bigg\|_{4\times 4}\notag\\
&\quad\cdot\sum_{m=0}^{1-n}\Bigg|\left(\frac{\partial}{\partial											 u}\right)^m\chi_j(u,\bk)\Bigg|
\Bigg\|\left(\frac{\partial}{\partial u}\right)^{1-n-m}\widehat{W}^l(\beta_a)(u,\bk)
\Bigg\|_{4\times 4}\notag\\
&\le c(M)\beta_1^{-1} \sum_{j=l}^{N_{\beta_a}}f_{\bt}^{-1}M^{3j}
\sum_{n=0}^1M^{-j}(c\fw(j)^{-1})^n\notag\\
&\quad\cdot
 \sum_{m=0}^{1-n}(c\fw(j)^{-1})^mc_{IR}^{-1}M^{\frac{3}{2}l}\alpha^{-2}(c\fw(l)^{-1})^{1-n-m}\notag\\
&\le c''(M,c_w)\beta_1^{-1}M^{\frac{3}{2}l}\alpha^{-2}\sum_{j=l}^{N_{\beta_a}}M^{j}\notag\\
&\le c''(M,c_w)\beta_1^{-1}M^{\frac{5}{2}l}\alpha^{-2}.\notag
\end{align}
Also, by recalling \eqref{eq_equivalence_IR_cut_off} we see that
\begin{align}
&\frac{1}{ L^2}\sum_{\bk\in
 \G^*}\Bigg(\int_{\frac{\pi}{\beta_1}}^{\pi h}du+\int_{-\pi h}^{-\frac{\pi}{\beta_1}}du
 \Bigg)\label{eq_IR_final_decomposition_cont_cont_pre}\\
&\quad\cdot\Big|\log\big(\det\big(I_4-(iu
 I_4-E(\bk)-E_{l+1}(\beta_1)(u,\bk))^{-1}\notag\\
&\quad\qquad\qquad\qquad\qquad\cdot \chi_{\le
 l}(\beta_1)(u,\bk)\widehat{W}^l(\beta_1)(u,\bk)\big)\big)\notag\\
&\qquad -\log\big(\det\big(I_4-(iu
 I_4-E(\bk)-E_{l+1}(\beta_2)(u,\bk))^{-1}\notag\\
&\qquad\qquad\qquad\qquad\qquad\cdot \chi_{\le
 l}(\beta_2)(u,\bk)\widehat{W}^l(\beta_2)(u,\bk)\big)\big)\Big|\notag\\
&\le \frac{1}{L^2}\sum_{\bk\in
 \G^*}\Bigg(
 \int_{\frac{\pi}{\beta_1}}^{\pi h}du+\int_{-\pi h}^{-\frac{\pi}{\beta_1}}du\Bigg)1_{\chi_{\le
 l}(\beta_1)(u,\bk)\neq 0}\notag\\
&\quad \cdot
 \Big|\log\Big(\det\big(I_4+\big(I_4-(iuI_4-E(\bk)-E_{l+1}(\beta_2)(u,\bk))^{-1}\notag\\&\qquad\qquad\qquad\qquad\qquad\qquad\cdot \chi_{\le
 l}(\beta_2)(u,\bk)\widehat{W}^l(\beta_2)(u,\bk)\big)^{-1}\notag\\
&\qquad\cdot \big((iuI_4-E(\bk)-E_{l+1}(\beta_2)(u,\bk))^{-1}\chi_{\le
 l}(\beta_2)(u,\bk)\widehat{W}^l(\beta_2)(u,\bk)\notag\\
&\quad\qquad-(iuI_4-E(\bk)-E_{l+1}(\beta_1)(u,\bk))^{-1}\chi_{\le
 l}(\beta_1)(u,\bk)\widehat{W}^l(\beta_1)(u,\bk)\big)\big)\Big)\Big|\notag\\
&=\frac{1}{L^2}\sum_{\bk\in
 \G^*}\Bigg(
 \int_{\frac{\pi}{\beta_1}}^{\pi h}du+\int_{-\pi h}^{-\frac{\pi}{\beta_1}}du\Bigg)1_{\chi_{\le
 l}(\beta_1)(u,\bk)\neq 0}\notag\\
&\quad \cdot
 \Big|\log\Big(\det\big(I_4+\big(I_4-(iuI_4-E(\bk)-E_{l+1}(\beta_2)(u,\bk))^{-1}\notag\\&\qquad\qquad\qquad\qquad\qquad\qquad\cdot \chi_{\le
 l}(\beta_2)(u,\bk)\widehat{W}^l(\beta_2)(u,\bk)\big)^{-1}\notag\\
&\qquad\cdot \big((iuI_4-E(\bk)-E_{l+1}(\beta_2)(u,\bk))^{-1}\chi_{\le
 l}(\beta_1)(u,\bk)\notag\\
&\qquad\quad\cdot
 (\widehat{W}^l(\beta_2)(u,\bk)-\widehat{W}^l(\beta_1)(u,\bk))\notag\\
&\qquad\quad
 +(iuI_4-E(\bk)-E_{l+1}(\beta_2)(u,\bk))^{-1}\notag\\
&\qquad\qquad\cdot(E_{l+1}(\beta_2)(u,\bk)-E_{l+1}(\beta_1)(u,\bk))\notag\\
&\qquad\qquad\cdot(iuI_4-E(\bk)-E_{l+1}(\beta_1)(u,\bk))^{-1}\notag\\
&\qquad\qquad\cdot \chi_{\le
 l}(\beta_1)(u,\bk)\widehat{W}^l(\beta_1)(u,\bk)\big)\big)\Big)\Big|,\notag
\end{align}
where we used the equalities of the form
\begin{align*}
&\log(\det(I_4-A))-\log(\det(I_4-B))\\
&\quad=\log(\det(I_4+(I_4-B)^{-1}(B-A))),\\
&(C-A)^{-1}-(C-B)^{-1}=(C-A)^{-1}(A-B)(C-B)^{-1},
\end{align*}
for $A,B,C\in \Mat(4,\C)$. Moreover, by
 \eqref{eq_infrared_integrant_bound}, the inequality $c_{IR}\ge
 f_{\bt}^{-1}$,  Lemma \ref{lem_bound_extended_one_scale_self_energy}
 \eqref{item_bound_extended_one_scale_self_energy}, Lemma
 \ref{lem_bound_extended_one_scale_self_energy_difference}, Lemma
 \ref{lem_bound_self_energy} \eqref{item_bound_extended_self_energy} and 
 Lemma \ref{lem_bound_self_energy_difference}  
 we have for any $j\in \{l,l-1,\cdots,$ $N_{\beta_1}\}$, $(u,\bk)\in\R^3$ with
 $\chi_j(u,\bk)\neq 0$ that
\begin{align}
&\Big\|\big(I_4-(iu I_4-E(\bk)-E_{l+1}(\beta_1)(u,\bk))^{-1}\chi_{\le
 l}(\beta_2)(u,\bk)\widehat{W}^l(\beta_2)(u,\bk)\big)^{-1}\label{eq_IR_final_decomposition_cont_cont_pre_inside}\\
&\quad\cdot \big((iu I_4-E(\bk)-E_{l+1}(\beta_2)(u,\bk))^{-1}\chi_{\le
 l}(\beta_1)(u,\bk)\notag\\
&\qquad\cdot
 (\widehat{W}^l(\beta_2)(u,\bk)-\widehat{W}^l(\beta_1)(u,\bk))\notag\\
&\qquad
 +(iu I_4-E(\bk)-E_{l+1}(\beta_2)(u,\bk))^{-1}(E_{l+1}(\beta_2)(u,\bk)-E_{l+1}(\beta_1)(u,\bk))\notag\\
&\quad\qquad\cdot(iu I_4-E(\bk)-E_{l+1}(\beta_1)(u,\bk))^{-1}
\chi_{\le
 l}(\beta_1)(u,\bk)\widehat{W}^l(\beta_1)(u,\bk)\big)\Big\|_{4\times
 4}\notag\\
&\le c\|(iu I_4-E(\bk)-E_{l+1}(\beta_2)(u,\bk))^{-1}\chi_{\le
 l}(\beta_1)(u,\bk)\notag\\
&\qquad\cdot(\widehat{W}^l(\beta_2)(u,\bk)-\widehat{W}^l(\beta_1)(u,\bk))\|_{4\times
 4}\notag\\
&\quad + c\|(iu I_4-E(\bk)-E_{l+1}(\beta_2)(u,\bk))^{-1}\notag\\
&\qquad\quad\cdot(E_{l+1}(\beta_2)(u,\bk)-E_{l+1}(\beta_1)(u,\bk))\|_{4\times 4}\notag\\
&\le c
 M^{-j}\|\widehat{W}^l(\beta_2)(u,\bk)-\widehat{W}^l(\beta_1)(u,\bk)\|_{4\times
 4}\notag\\
&\quad +c
 M^{-j}\|E_{l+1}(\beta_2)(u,\bk)-E_{l+1}(\beta_1)(u,\bk)\|_{4\times
 4}\notag\\
&\le \min\{cM\alpha^{-2},c\beta_1^{-\frac{1}{2}}M^{-j}\alpha^{-2}\}<1.\notag
\end{align}
By taking into account  Lemma \ref{lem_IR_cut_off_support_measure} \eqref{item_measure_support_semi_discrete} and 
\eqref{eq_IR_final_decomposition_cont_cont_pre_inside} we observe that 
\begin{align}
&(\text{the right-hand side of
 \eqref{eq_IR_final_decomposition_cont_cont_pre}})\label{eq_IR_final_decomposition_cont_cont}\\
&\le \sum_{j=l}^{N_{\beta_1}}\frac{1}{L^2}\sum_{\bk\in
 \G^*}\Bigg(\int^{\pi
 h}_{\frac{\pi}{\beta_1}}du+\int_{-\pi h}^{-\frac{\pi}{\beta_1}}du\Bigg)1_{\chi_j(u,\bk)\neq
 0}\notag\\
&\quad\cdot\sum_{n=1}^{\infty}\frac{1}{n}(\min\{cM\alpha^{-2},c\beta_1^{-\frac{1}{2}}M^{-j}\alpha^{-2}\})^n\notag\\
&\le
 c(M)\sum_{j=l}^{N_{\beta_1}}f_{\bt}^{-1}M^{3j}\min\{cM\alpha^{-2},c\beta_1^{-\frac{1}{2}}M^{-j}\alpha^{-2}\}\le c(M)\beta_1^{-\frac{1}{2}}f_{\bt}^{-1}M^{2l}\alpha^{-2}.\notag
\end{align}

Similarly, by using Lemma \ref{lem_IR_cut_off_support_measure}
 \eqref{item_measure_support_no_matsubara} and
 \eqref{eq_IR_determinant_pre_another} we can derive that
\begin{align}
&\frac{1}{L^2}\sum_{\bk\in
 \G^*}\Bigg(\int_{\frac{\pi}{\beta_2}}^{\frac{\pi}{\beta_1}}du+\int^{-\frac{\pi}{\beta_2}}_{-\frac{\pi}{\beta_1}}du\Bigg)\label{eq_IR_final_decomposition_cont_rest}\\
&\quad\cdot\Big|\log\big(\det\big(I_4-(iu
 I_4-E(\bk)-E_{l+1}(\beta_2)(u,\bk))^{-1}\notag\\
&\qquad\qquad\qquad\qquad\quad\cdot \chi_{\le
 l}(\beta_2)(u,\bk)\widehat{W}^l(\beta_2)(u,\bk)\big)\big)\Big|\notag\\
&\le \frac{1}{L^2}\sum_{\bk\in\G^*}\Bigg(\int_{\frac{\pi}{\beta_2}}^{\frac{\pi}{\beta_1}}du+\int^{-\frac{\pi}{\beta_2}}_{-\frac{\pi}{\beta_1}}du\Bigg)
1_{\chi_{\le
 l}(\beta_2)(u,\bk)\neq
 0}\sum_{n=1}^{\infty}\frac{1}{n}(cM\alpha^{-2})^{n}\notag\\
&\le c(M)\beta_1^{-1}f_{\bt}^{-1}M^{2l}\alpha^{-2}.\notag
\end{align}

Combining \eqref{eq_IR_0th_difference},
 \eqref{eq_IR_final_decomposition_discrete_cont},
 \eqref{eq_IR_final_decomposition_cont_cont}, 
\eqref{eq_IR_final_decomposition_cont_rest} with
 \eqref{eq_IR_final_difference_decomposition} yields
\begin{align}
|J_{end}(\beta_1)-J_{end}^1(\beta_2)|
&\le
 c\sum_{l=0}^{N_{\beta_1}-1}\beta_1^{-\frac{1}{2}}M^{\frac{5}{2}l}\alpha^{-3}
 +\sum_{l=0}^{N_{\beta_1}}c''(M,c_w)\beta_1^{-\frac{1}{2}}f_{\bt}^{-1}M^{2l}\alpha^{-2}\label{eq_IR_final_normal_part_bound}\\
&\le c''(M,c_w)\beta^{-\frac{1}{2}}f_{\bt}^{-1}\alpha^{-2}.\notag
\end{align}
Finally, by coupling \eqref{eq_IR_final_normal_part_bound} with \eqref{eq_IR_final_large_part_bound} we reach the claimed inequality. 
\end{proof}

In order to indicate the dependency on the parameters $\beta$, $L$, $h$,
let us write $J_{end}(\beta,L,h)(\bU)$ in place of $J_{end}(\bU)$ in the
following. For any compact set $K$ of $\C^4$ let $C(K;\C)$ denote the
Banach space of all complex-valued continuous functions on $K$, equipped
with the uniform norm. To prove the next lemma, we need Lemma
\ref{lem_h_L_limit} proved in Appendix \ref{app_h_L_limit}. 

\begin{lemma}\label{lem_infrared_integration_convergence}
For any non-empty compact set $K$ of $\C^4$ satisfying $K\subset D_{IR}$ the following statements hold true.
\begin{enumerate}
\item\label{item_IR_h_limit} For any $\beta\in \R_{>0}$, $L\in \N$ with
     $L\ge \beta$,
$J_{end}(\beta,L,h)(\cdot)$ converges in $C(K;\C)$ as $h\to \infty$
     $(h\in 2\N/\beta)$.
\item\label{item_IR_L_limit} Set
 $J(\beta,L):=\lim_{h\to \infty,h\in
     2\N/\beta}J_{end}(\beta,L,h)$.
For any $\beta\in \R_{>0}$, $J(\beta,L)(\cdot)$ converges in $C(K;\C)$ as $L\to \infty$
     $(L\in \N)$.
\item\label{item_IR_beta_limit}
Set $J(\beta):=\lim_{L\to \infty, L\in
     \N}J(\beta,L)$. 
$J(\beta)(\cdot)$ converges in $C(K;\C)$ as $\beta \to \infty$
     $(\beta \in \N)$.
\end{enumerate}
\end{lemma}
\begin{proof}
\eqref{item_IR_h_limit},\eqref{item_IR_L_limit}:
Take any $U_0\in (0,f_{\bt}^2(c'(M,c_w)\alpha^4)^{-1})$ and small
 $\eps\in\R_{>0}$, where $c'(M,c_w)$ is the constant appearing in the
 definition of $D_{IR}$ in Lemma
 \ref{lem_IR_analytic_continuation_expansion} \eqref{item_input_IR_consistency}.
 Set 
$$
D_{\eps}:=\{(U_1,U_2,U_3,U_4)\in\C^4\ |\ |U_{\rho}|<U_0-\eps,\ (\forall
 \rho\in \cB)\}.
$$
Note that $z\bU\in D_{IR}$ for any $\bU\in \overline{D_{\eps}}$ and $z\in \C$ with
 $|z|\le U_0/(U_0-\eps)$. By Lemma
 \ref{lem_infrared_integration_final_bounds}
 \eqref{item_IR_final_analyticity} and Cauchy's integral formula we can
 justify the following equalities.
\begin{align}
J_{end}(\beta,L,h)(\bU)&=\sum_{n=0}^{\infty}\frac{1}{n!}\left.\left(\frac{\partial}{\partial
 z}\right)^nJ_{end}(\beta,L,h)(z\bU)\right|_{z=0}\label{eq_IR_perturbation_series}\\
&=\sum_{n=0}^{\infty}a_n(\beta,L,h)(\bU),\ (\forall
 \bU\in\overline{D_{\eps}}),\notag
\end{align}
where
\begin{align*}
a_n(\beta,L,h)(\bU):=\frac{1}{2\pi
 i}\oint_{|z|=\frac{U_0}{U_0-\eps}}dz\frac{J_{end}(\beta,L,h)(z\bU)}{z^{n+1}}.
\end{align*}
By Lemma \ref{lem_infrared_integration_final_bounds}
 \eqref{item_IR_final_bound},
\begin{align}
&|a_n(\beta,L,h)(\bU)|\le
 c(M)f_{\bt}^{-1}\alpha^{-2}\left(\frac{U_0-\eps}{U_0}\right)^n,\label{eq_IR_dominated_uniform_bound}\\
&(\forall \bU\in\overline{D_{\eps}}, n\in \N\cup \{0\}).\notag
\end{align}
 
By Lemma \ref{lem_IR_analytic_continuation_expansion}
 \eqref{item_IR_expansion} there exists a constant $c(\beta,L,M)\in\R_{>0}$
 depending only on $\beta$, $L$, $M$ such that
\begin{align*}
&-\frac{1}{\beta L^2}\log\left(\int e^{J^0(\bU)(\psi)}d\mu_{C_{\le
 0}^{\infty}}(\psi)\right)=J_{end}(\beta,L,h)(\bU),\\
&\ (\forall \bU\in\C^4\text{ with }|U_{\rho}|\le
 f_{\bt}^2(c'(M,c_w)c(\beta,L,M)^4)^{-1},\ (\forall \rho\in \cB)).
\end{align*}
By this equality and Lemma \ref{lem_symmetric_formulation}
 \eqref{item_real_part_symmetric} we can choose constants
 $c_1,h_0\in\R_{>0}$, which may depend on $\beta$, $L$ but are independent of $h$,
 so that for any $\bU\in\overline{D_{\eps}}$, $h\in 2\N/\beta$ with
 $h\ge h_0$, 
\begin{align}
&a_n(\beta,L,h)(\bU)\label{eq_IR_perturbation_series_decomposition}\\
&=\frac{1}{2\pi
 i}\oint_{|z|=c_1}dz\frac{1}{z^{n+1}}\Bigg(-\frac{1}{\beta
 L^2}\log\Bigg(\int e^{J^0(z\bU)(\psi)}d\mu_{C_{\le 0}^{\infty}}(\psi)\Bigg)\notag\\
&\qquad\qquad\qquad\qquad\qquad+\frac{1}{\beta L^2}\log\Bigg(\int
 e^{-V(z\bU)(\psi)}d\mu_{C}(\psi)\Bigg)\Bigg)\notag\\
&\quad -\frac{1}{\beta L^2}\frac{1}{n!}\left.\left(\frac{\partial}{\partial z}\right)^n\log\Bigg(\int
 e^{-V(z\bU)(\psi)}d\mu_{C}(\psi)\Bigg)\right|_{z=0}.\notag
\end{align}
By Lemma \ref{lem_h_L_limit} proved in Appendix \ref{app_h_L_limit} the last term in the right-hand side of
 \eqref{eq_IR_perturbation_series_decomposition} uniformly converges
 with respect to $\bU\in \overline{D_{\eps}}$ as 
we send $h$ to infinity first, then $L$ to infinity next. Moreover, Lemma \ref{lem_symmetric_formulation}
 \eqref{item_logarithm_final_h_estimate} and Proposition
 \ref{prop_UV_integration} \eqref{item_UV_analytic_continuation} imply that
\begin{align*}
&\lim_{h\to \infty\atop h\in
 2\N/\beta}\sup_{\bU\in\overline{D_{\eps}}}\Bigg|
\frac{1}{2\pi
 i}\oint_{|z|=c_1}dz\frac{1}{z^{n+1}}\Bigg(\frac{1}{\beta
 L^2}\log\Bigg(\int e^{J^0(z\bU)(\psi)}d\mu_{C_{\le 0}^{\infty}}(\psi)\Bigg)\\
&\qquad\qquad\qquad\qquad\qquad\qquad\quad-\frac{1}{\beta L^2}\log\Bigg(\int
 e^{-V(z\bU)(\psi)}d\mu_{C}(\psi)\Bigg)\Bigg)\Bigg|\notag\\
&=0.\notag
\end{align*}
Therefore, there exist
 $\{a_n(\beta,L)\}_{n=0}^{\infty}$,
 $\{a_n(\beta)\}_{n=0}^{\infty}\subset C(\overline{D_{\eps}};\C)$
 such that  
\begin{align}
&\lim_{h\to \infty\atop h\in 2\N/\beta}\sup_{\bU\in
 \overline{D_{\eps}}}|a_n(\beta,L,h)(\bU)-a_n(\beta,L)(\bU)|=0,\label{eq_IR_h_convergence}\\
&\lim_{L\to \infty\atop L\in\N}\sup_{\bU\in
 \overline{D_{\eps}}}|a_n(\beta,L)(\bU)-a_n(\beta)(\bU)|=0,\ (\forall n\in\N\cup\{0\}).\label{eq_IR_L_convergence}
\end{align}
Since the right-hand side of \eqref{eq_IR_dominated_uniform_bound} is
 summable with $n$ over $\N\cup\{0\}$, we can apply the dominated convergence
 theorem for $l^1(\N\cup \{0\})$ together with \eqref{eq_IR_h_convergence},
 \eqref{eq_IR_L_convergence} to deduce from
 \eqref{eq_IR_perturbation_series} that 
\begin{align*}
&\lim_{h\to \infty\atop h\in 2\N/\beta}\sup_{\bU\in
 \overline{D_{\eps}}}\left|J_{end}(\beta,L,h)(\bU)-\sum_{n=0}^{\infty}
a_n(\beta,L)(\bU)\right|=0,\\
&\lim_{L\to \infty\atop L\in\N}\sup_{\bU\in
 \overline{D_{\eps}}}\left|\sum_{n=0}^{\infty}a_n(\beta,L)(\bU)-\sum_{n=0}^{\infty}a_n(\beta)(\bU)\right|=0.
\end{align*}
For any compact set $K$ of $\C^4$ with $K\subset D_{IR}$ we can choose $U_0\in
 (0,f_{\bt}^2(c'(M,c_w)\alpha^4)^{-1})$ and $\eps\in\R_{>0}$ so that
 $K\subset \overline{D_{\eps}}$. Thus, the claims
 \eqref{item_IR_h_limit}, \eqref{item_IR_L_limit} have been proved.

\eqref{item_IR_beta_limit}: By sending $h$ and $L$ to infinity in Lemma
 \ref{lem_infrared_integration_final_bounds}
 \eqref{item_IR_final_difference} we obtain 
\begin{align*}
\sup_{\bU\in K}|J(\beta_1)(\bU)-J(\beta_2)(\bU)|\le
 c''(M,c_w)\beta_1^{-\frac{1}{2}}f_{\bt}^{-1}\alpha^{-2},
\end{align*}
for any $\beta_1,\beta_2\in\N$ with $\beta_2\ge \beta_1$, which implies
 that $(J(\beta))_{\beta\in\N}$ is a Cauchy sequence in the
 Banach space $C(K;\C)$. Therefore, the claim \eqref{item_IR_beta_limit}
 holds true.
\end{proof}

Here we can give the proof of Theorem \ref{thm_main_theorem},
admitting lemmas proved in Appendix \ref{app_partition_function}.
\begin{proof}[Proof of Theorem \ref{thm_main_theorem}]
 First of all let
 us assume that Theorem \ref{thm_main_theorem} is true if
 $\max\{t_{h,e},t_{h,o},t_{v,e},t_{v,o}\}=1$. 
Set $t_{max}:=\max\{t_{h,e},t_{h,o},t_{v,e},t_{v,o}\}$.
By the theorem for the Hamiltonian
 $\sH_0/t_{max}+\sV$ there exist a generic constant $c\in\R_{>0}$ and
 continuous functions $F_{\beta,L}(\cdot)$, $F_{\beta}(\cdot)$, $F(\cdot)$  
  $:\overline{D_{\bt/t_{max}}(c)}^4\to\C$ such that $F_{\beta,L}(\cdot)$
is analytic in $D_{\bt/t_{max}}(c)^4$ and 
\begin{align*}
&F_{\beta,L}(\bU)=-\frac{1}{\beta (2L)^2}\log(\Tr e^{-\beta
 (\sH_0/t_{max}+\sV)}),\\
&(\forall \bU\in \overline{D_{\bt/t_{max}}(c)}^4\cap
 \R^4,\beta\in\R_{>0},L\in\N\text{ with }L\ge \beta),\\
&\lim_{L\to \infty\atop
 L\in\N}\sup_{\bz\in\overline{D_{\bt/t_{max}}(c)}^4}|F_{\beta,L}(\bz)-F_{\beta}(\bz)|=0,\
 (\forall \beta\in \R_{>0}),\notag\\
&\lim_{\beta\to \infty\atop
 \beta \in\R_{>0}}
 \sup_{\bz\in\overline{D_{\bt/t_{max}}(c)}^4}|F_{\beta}(\bz)-F(\bz)|=0.
\end{align*}
By replacing $\beta$ by $t_{max}\beta$ and taking into account the
 equality $D_{\bt/t_{max}}(c)=(1/t_{max})D_{\bt}(c)$ we have that
\begin{align*}
&t_{max}F_{t_{max}\beta,L}\left(\frac{1}{t_{max}}\bU\right)=-\frac{1}{\beta
 (2L)^2}\log(\Tr e^{-\beta \sH}),\\
&(\forall \bU\in \overline{D_{\bt}(c)}^4\cap
 \R^4,\beta\in\R_{>0},L\in\N\text{ with }L\ge t_{max}\beta),\\
&\lim_{L\to \infty\atop
 L\in\N}\sup_{\bz\in\overline{D_{\bt}(c)}^4}\left|t_{max}F_{t_{max}\beta,L}\left(\frac{1}{t_{max}}\bz\right)-t_{max} F_{t_{max}\beta}\left(\frac{1}{t_{max}}\bz\right)\right|=0,\\
&(\forall \beta\in \R_{>0}),\notag\\
&\lim_{\beta\to \infty\atop
 \beta \in\R_{>0}}\sup_{\bz\in\overline{D_{\bt}(c)}^4}\left|t_{max}F_{t_{max}\beta}\left(\frac{1}{t_{max}}\bz\right)-t_{max} F\left(\frac{1}{t_{max}}\bz\right)\right|=0.
\end{align*}
Since the functions $F_{t_{max}\beta,L}(\cdot/t_{max})$,
 $F_{t_{max}\beta}(\cdot/t_{max})$, $F(\cdot/t_{max})$ 
are continuous in $\overline{D_{\bt}(c)}^4$ and
 $F_{t_{max}\beta,L}(\cdot/t_{max})$ is analytic in $D_{\bt}(c)^4$,
the claims for the Hamiltonian $\sH$ hold true. Therefore,
 it suffices to prove the theorem on the assumption that $t_{max}=1$.

In the following we assume that $\beta\in\R_{>0}$, $L\in\N$ satisfies
 $L\ge \beta$ and the parameters $M$, $\alpha\in\R_{>0}$ satisfy the
 conditions required in Lemma
 \ref{lem_IR_analytic_continuation_expansion} and Lemma
 \ref{lem_infrared_integration_final_bounds}. By Lemma
 \ref{lem_infrared_integration_convergence} there exist functions
 $J(\beta,L)(\cdot)$, $J(\beta)(\cdot)$, $J(\cdot):D_{IR}\to\C$ such
 that
\begin{align*}
&J(\beta,L)(\bU)=\lim_{h\to\infty\atop h\in
 2\N/\beta}J_{end}(\beta,L,h)(\bU),\\
&J(\beta)(\bU)=\lim_{L\to \infty\atop L\in\N}\lim_{h\to\infty\atop h\in 2\N/\beta}J_{end}(\beta,L,h)(\bU),\\
&J(\bU)=\lim_{\beta\to\infty\atop\beta\in\N}\lim_{L\to \infty\atop
 L\in\N}\lim_{h\to\infty\atop h\in 2\N/\beta}J_{end}(\beta,L,h)(\bU),\
 (\forall \bU\in D_{IR}).
 \end{align*}

\eqref{item_theorem_analytic_continuation}: By Lemma
 \ref{lem_formulation_covariance}, Lemma
 \ref{lem_symmetric_formulation},
 Proposition \ref{prop_UV_integration} \eqref{item_UV_analytic_continuation}
 and Lemma
 \ref{lem_IR_analytic_continuation_expansion} \eqref{item_IR_expansion},
 there exists a constant $c_1\in\R_{>0}$ which may
 depend on $\beta$, $L$ but not on $h$ such that for any $\bU\in \R^4$
 with $|U_{\rho}|\le c_1$ $(\forall \rho\in\cB)$,
\begin{align}
&J(\beta,L)(\bU)\label{eq_h_limit_beta_L_fixed}\\
&=\lim_{h\to \infty\atop h\in 2\N/\beta}\notag\\
&\quad\cdot \Bigg(-\frac{1}{\beta
 L^2}\log\left(\int e^{J^0(\psi)}d\mu_{C_{\le 0}^{\infty}}(\psi)\right)
+\frac{1}{\beta L^2}\log\left(\int
 e^{-V(\psi)}d\mu_{C}(\psi)\right)\Bigg)\notag\\
&\quad +\lim_{h\to \infty\atop h\in 2\N/\beta}\left(-\frac{1}{\beta
 L^2}\log\left(\int e^{-V(\psi)}d\mu_{C}(\psi)\right)\right)\notag\\
&=-\frac{1}{\beta L^2}\log\left(\frac{\Tr e^{-\beta H}}{\Tr e^{-\beta
 H_0}}\right).\notag
\end{align}
By Lemma \ref{lem_infrared_integration_final_bounds}
 \eqref{item_IR_final_analyticity} and Lemma
 \ref{lem_infrared_integration_convergence} \eqref{item_IR_h_limit},
 $\bU\mapsto J(\beta,L)(\bU)$ is analytic in $D_{IR}$. On the other hand,
 by Lemma \ref{lem_analyticity_domain_partition_function} proved in
 Appendix \ref{app_partition_function} there exists a domain $O\subset
 \C^4$ such that $D_{IR}\cap \R^4\subset O$
and $\bU\mapsto \log(\Tr e^{-\beta H})$ is analytic in $O$. Therefore,
 by the identity theorem we obtain that
\begin{align}
J(\beta,L)(\bU)=-\frac{1}{\beta L^2}\log\left(\frac{\Tr e^{-\beta
 H}}{\Tr e^{-\beta H_0}}\right),\ (\forall \bU\in D_{IR}\cap
 \R^4),\label{eq_h_limit_beta_L_fixed_enlarged}
\end{align}
or by Lemma \ref{lem_free_energy_equivalence} and 
Lemma \ref{lem_free_free_energy} proved in Appendix \ref{app_partition_function},
\begin{align}
&-\frac{1}{\beta (2L)^2}\log(\Tr e^{-\beta
 \sH})\label{eq_h_limit_beta_L_fixed_decomposed}\\
&=\frac{1}{4}J(\beta,L)(\bU)-\frac{1}{2\beta L^2}\sum_{\bk\in \G^*}\sum_{p,q\in
 \{1,-1\}}\log(1+e^{-\beta
 X_{p,q}(\bk)}),\notag\\
&(\forall \bU\in D_{IR}\cap
 \R^4),\notag
\end{align}
where $X_{p,q}(\bk)$ $(p,q\in \{1,-1\})$ are the eigen values of $E(\bk)$
 given in \eqref{eq_explicit_free_dispersion_relation}. 
The equality \eqref{eq_h_limit_beta_L_fixed_decomposed} implies that the
 claim \eqref{item_theorem_analytic_continuation} holds for
 $D_{\bt}(c')$ with any $c'\in(0,(c'(M,c_w)\alpha^4)^{-1})$, where 
$c'(M,c_w)(\in\R_{>0})$ is the $M,c_w$-dependent constant appearing in
 the definition of $D_{IR}$ in Lemma
 \ref{lem_IR_analytic_continuation_expansion}
\eqref{item_input_IR_consistency}. 
In the following we fix $c'$ to be
 $(2c'(M,c_w)\alpha^4)^{-1}$ so that $\overline{D_{\bt}(c')}^4\subset
 D_{IR}$.

\eqref{item_theorem_L_limit}: Set
\begin{align*}
&F_{\beta,L}(\bU):=\frac{1}{4}J(\beta,L)(\bU)-\frac{1}{2\beta L^2}\sum_{\bk\in\G^*}\sum_{p,q\in
 \{1,-1\}}\log(1+e^{-\beta
 X_{p,q}(\bk)}),\\
&F_{\beta}(\bU):=\frac{1}{4}J(\beta)(\bU)-\frac{1}{2\beta (2\pi)^2}\sum_{p,q\in
 \{1,-1\}}\int_{[-\pi,\pi]^2}d\bk \log(1+e^{-\beta
 X_{p,q}(\bk)}),\\
 &(\forall \bU\in D_{IR}).
\end{align*}
We can see from Lemma \ref{lem_infrared_integration_convergence}
 \eqref{item_IR_L_limit} and the continuity of the function $\bk\mapsto 
X_{p,q}(\bk)$ that 
\begin{align*}
\lim_{L\to \infty\atop L\in\N}\sup_{\bU\in
 \overline{D_{\bt}(c')}^4}|F_{\beta,L}(\bU)
-F_{\beta}(\bU)|=0,\ (\forall \beta \in\R_{>0}).
\end{align*}
Thus, the claim \eqref{item_theorem_L_limit} is true.

\eqref{item_theorem_beta_limit}: 
By Lemma \ref{lem_beta_integer_limit_partition} proved in Appendix \ref{app_partition_function}, Lemma
 \ref{lem_infrared_integration_final_bounds}
 \eqref{item_IR_final_bound},
\eqref{eq_h_limit_beta_L_fixed_enlarged} and continuity,
\begin{align}
&\left|\frac{1}{\beta L^2}\log\left(\frac{\Tr e^{-\beta
 H}}{\Tr e^{-\beta H_0}}\right)-\frac{1}{[\beta] L^2}\log\left(\frac{\Tr e^{-[\beta]
 H}}{\Tr e^{-[\beta]
 H_0}}\right)\right|\label{eq_beta_gauss_symbol_decomposition}\\
&\le \int_{[\beta]}^{\beta}d\gamma\left|\frac{1}{\gamma^2L^2}\left(\frac{\Tr e^{-\gamma
 H}}{\Tr e^{-\gamma
 H_0}}\right)\right|\notag\\
&\quad+c\left(1+\sup_{\bU\in \overline{D_{IR}}}\|\bU\|_{\C^4}\right)\log\left(\frac{\beta}{[\beta]}\right)\notag\\
&\le
 c\left(c(M)f_{\bt}^{-1}\alpha^{-2}+1+\sup_{\bU\in
 \overline{D_{IR}}}\|\bU\|_{\C^4}\right)\log\left(\frac{\beta}{[\beta]}\right),\notag\\
&(\forall \bU\in \overline{D_{IR}}\cap \R^4),\notag
\end{align}
where $c(M)\in\R_{>0}$ is the $M$-dependent constant appearing in Lemma
 \ref{lem_infrared_integration_final_bounds} \eqref{item_IR_final_bound}.
Set 
\begin{align*}
&F(\bU):=\frac{1}{4}J(\bU)+\frac{1}{2(2\pi)^2}\sum_{p,q\in\{1,-1\}}\int_{[-\pi,\pi]^2}d\bk
 1_{X_{p,q}(\bk)<0}X_{p,q}(\bk),\\
&(\forall \bU\in D_{IR}).
\end{align*}
Lemma \ref{lem_infrared_integration_convergence} implies that the
 function $\bU\mapsto F(\bU)$ is analytic
 in $D_{IR}$. We can derive from \eqref{eq_h_limit_beta_L_fixed_enlarged}
 and
 \eqref{eq_beta_gauss_symbol_decomposition} that for any $\beta\in\R_{\ge 1}$,
\begin{align*}
&\sup_{\bU\in \overline{D_{\bt}(c')}^4\cap \R^4}|F_{\beta}(\bU)-F(\bU)|\\
&\le
 c\left(c(M)f_{\bt}^{-1}\alpha^{-2}+1+\sup_{\bU\in \overline{D_{IR}}}\|\bU\|_{\C^4}
\right)\log\left(\frac{\beta}{[\beta]}\right)\\
&\quad +\sup_{\bU\in \overline{D_{\bt}(c')}^4}|J([\beta])(\bU)-J(\bU)|\\
&\quad +\Bigg|-\frac{1}{2\beta
 (2\pi)^2}\sum_{p,q\in\{1,-1\}}\int_{[-\pi,\pi]^2}d\bk\log(1+e^{-\beta
 X_{p,q}(\bk)})\\
&\qquad\quad-\frac{1}{2
 (2\pi)^2}\sum_{p,q\in\{1,-1\}}\int_{[-\pi,\pi]^2}d\bk1_{X_{p,q}(\bk)<0}X_{p,q}(\bk)\Bigg|.
\end{align*}
Thus, by Lemma \ref{lem_infrared_integration_convergence}
 \eqref{item_IR_beta_limit},
\begin{align}
\lim_{\beta\to \infty\atop \beta \in \R_{>0}}\sup_{\bU\in
 \overline{D_{\bt}(c')}^4\cap
 \R^4}|F_{\beta}(\bU)-F(\bU)|=0.\label{eq_theorem_proof_convergence_real_axis}
\end{align}

To complete the proof, we need to show that 
\begin{align}
\lim_{\beta\to \infty\atop \beta \in \R_{>0}}\sup_{\bU\in
 \overline{D_{\bt}(c')}^4}|F_{\beta}(\bU)-F(\bU)|=0.\label{eq_theorem_proof_convergence_disk}
\end{align}
The convergence property \eqref{eq_theorem_proof_convergence_disk} can
 be proved by a basic argument. However, we present the proof for
 completeness. Note that by Lemma
 \ref{lem_infrared_integration_final_bounds}
 \eqref{item_IR_final_bound},
\begin{align}
\sup_{\bU\in D_{IR}}|F_{\beta}(\bU)|\le \tilde{c},\ (\forall
 \beta\in\R_{\ge 1}),\ \sup_{\bU\in D_{IR}}|F(\bU)|\le
 \tilde{c},\label{eq_theorem_proof_uniform_bounds}
\end{align}
with some constant $\tilde{c}(\in\R_{>0})$ independent of $\beta$.
For any $j\in\{1,2,3,4\}$, $n\in\N\cup\{0\}$,
 $\bU\in\overline{D_{\bt}(c')}^4$, set
\begin{align*}
a_{\beta,j,n}(\bU):=\frac{1}{n!}\left(\frac{\partial}{\partial
 U_j}\right)^{n}F_{\beta}(\bU),\ a_{j,n}(\bU):=\frac{1}{n!}\left(\frac{\partial}{\partial
 U_j}\right)^{n}F(\bU).
\end{align*}
Since $F_{\beta}(\cdot)$, $F(\cdot)$ are analytic in $D_{IR}$,
\begin{align}
&F_{\beta}(\bU)=\sum_{n=0}^{\infty}a_{\beta,j,n}(U_1,\cdots,U_{j-1},0,U_{j+1},\cdots,U_4)U_j^n,\label{eq_theorem_proof_expansion}\\
&F(\bU)=\sum_{n=0}^{\infty}a_{j,n}(U_1,\cdots,U_{j-1},0,U_{j+1},\cdots,U_4)U_j^n,\notag\\
&(\forall j\in\{1,2,3,4\}, \bU\in
 \overline{D_{\bt}(c')}^4,\beta\in\R_{\ge 1}).\notag
\end{align}
Moreover, by \eqref{eq_theorem_proof_uniform_bounds} Cauchy's integral
 formula gives that
\begin{align}
&|a_{\beta,j,n}(U_1,\cdots,U_{j-1},0,U_{j+1},\cdots,U_4)|\le \tilde{c}
\left(\frac{3}{2}c'f_{\bt}^2\right)^{-n},\label{eq_theorem_proof_each_term_bound}\\
&|a_{j,n}(U_1,\cdots,U_{j-1},0,U_{j+1},\cdots,U_4)|\le  \tilde{c}
\left(\frac{3}{2}c'f_{\bt}^2\right)^{-n},\notag\\
&(\forall j\in\{1,2,3,4\},n\in \N\cup
 \{0\},(U_1,\cdots,U_{j-1},U_{j+1},\cdots,U_4)\in\overline{D_{\bt}(c')}^3,\notag\\
&\quad\beta \in \R_{\ge 1}).\notag
\end{align}

Let us prove that
\begin{align}
\lim_{\beta\to \infty\atop
 \beta\in\R_{>0}}\sup_{\bU\in\overline{D_{\bt}(c')}^3\cap\R^3}
|a_{\beta,1,n}(0,\bU)-a_{1,n}(0,\bU)|=0,\ (\forall n\in\N\cup \{0\})
\label{eq_theorem_proof_each_term_convergence_1}
\end{align}
by induction with $n$. By \eqref{eq_theorem_proof_convergence_real_axis},
\begin{align*}
&\lim_{\beta\to \infty\atop
 \beta\in\R_{>0}}\sup_{\bU\in\overline{D_{\bt}(c')}^3\cap\R^3}
|a_{\beta,1,0}(0,\bU)-a_{1,0}(0,\bU)|\\
&=\lim_{\beta\to \infty\atop
 \beta\in\R_{>0}}\sup_{\bU\in\overline{D_{\bt}(c')}^3\cap\R^3}
|F_{\beta}(0,\bU)-F(0,\bU)|=0.
\end{align*}
Next, let us assume that there exists $m\in\N\cup\{0\}$
 such that \eqref{eq_theorem_proof_each_term_convergence_1} holds for
 all $n\in\N\cup\{0\}$ with $n\le m$. Suppose that there exist $\delta
 \in\R_{>0}$ and a sequence $(\beta_l)_{l=1}^{\infty}$ such that $\beta_l\to \infty$ as
 $l\to \infty$ and 
$$\sup_{\bU\in\overline{D_{\bt}(c')}^3\cap\R^3}
|a_{\beta_l,1,m+1}(0,\bU)-a_{1,m+1}(0,\bU)|\ge \delta,\
 (\forall l\in \N).
$$
Then, we can see from \eqref{eq_theorem_proof_expansion},
 \eqref{eq_theorem_proof_each_term_bound} that
\begin{align*}
\delta|U|^{m+1}
&\le \sup_{\bU\in\overline{D_{\bt}(c')}^3\cap\R^3}
|a_{\beta_l,1,m+1}(0,\bU)-a_{1,m+1}(0,\bU)||U|^{m+1}\\
&\le \sup_{\bU\in\overline{D_{\bt}(c')}^3\cap\R^3}
|F_{\beta_l}(|U|,\bU)-F(|U|,\bU)|\\
&\quad +\sum_{n=0}^m\sup_{\bU\in\overline{D_{\bt}(c')}^3\cap\R^3}
|a_{\beta_l,1,n}(0,\bU)-a_{1,n}(0,\bU)||U|^{n}\\
&\quad +
 2\tilde{c}\sum_{n=m+2}^{\infty}\left(\frac{3}{2}c'f_{\bt}^2\right)^{-n}|U|^{n},\ (\forall U\in\overline{D_{\bt}(c')}).
\end{align*}
By \eqref{eq_theorem_proof_convergence_real_axis} and the induction
 hypothesis the first term and the second term in the right-hand side of
 the above inequality converge to 0 as $l\to \infty$. Then, by dividing both sides by
 $|U|^{m+1}$ we obtain that
\begin{align*}
\delta &\le 2\tilde{c}
\sum_{n=m+2}^{\infty}\left(\frac{3}{2}c'f_{\bt}^2\right)^{-n}|U|^{n-m-1}\\
&\le
 2\tilde{c}|U|\left(\frac{3}{2}c'f_{\bt}^2\right)^{-m-2}\sum_{n=m+2}^{\infty}\left(\frac{3}{2}\right)^{-n+m+2},\ (\forall U\in \overline{D_{\bt}(c')}\backslash\{0\}).
\end{align*}
Sending $U$ to $0$ yields $\delta \le 0$, which is a
 contradiction. Thus,
\begin{align*}
\lim_{\beta\to \infty\atop
 \beta\in\R_{>0}}\sup_{\bU\in\overline{D_{\bt}(c')}^3\cap\R^3}
|a_{\beta,1,m+1}(0,\bU)-a_{1,m+1}(0,\bU)|=0. 
\end{align*}
By induction, the convergence property
 \eqref{eq_theorem_proof_each_term_convergence_1} holds true.

It follows from \eqref{eq_theorem_proof_expansion}, 
\eqref{eq_theorem_proof_each_term_bound},
 \eqref{eq_theorem_proof_each_term_convergence_1} and the dominated
 convergence theorem that 
\begin{align}
\lim_{\beta\to \infty\atop \beta\in\R_{>0}}\sup_{U\in
 \overline{D_{\bt}(c')}\atop
\bU\in \overline{D_{\bt}(c')}^3\cap\R^3}|F_{\beta}(U,\bU)-F(U,\bU)|=0.\label{eq_theorem_proof_convergence_real_axis_next}
\end{align}
By using \eqref{eq_theorem_proof_convergence_real_axis_next} in place of 
\eqref{eq_theorem_proof_convergence_real_axis} in an inductive argument
 parallel to that above we can prove that 
\begin{align}
\lim_{\beta\to \infty\atop
 \beta\in\R_{>0}}\sup_{U\in \overline{D_{\bt}(c')}\atop
\bU\in\overline{D_{\bt}(c')}^2\cap\R^2}
|a_{\beta,2,n}(U,0,\bU)-a_{2,n}(U,0,\bU)|=0,\ (\forall n\in\N\cup\{0\}).
\label{eq_theorem_proof_each_term_convergence_2} 
\end{align}
Then, combination of \eqref{eq_theorem_proof_expansion},
 \eqref{eq_theorem_proof_each_term_bound},
 \eqref{eq_theorem_proof_each_term_convergence_2} and the dominated
 convergence theorem concludes that
\begin{align*}
\lim_{\beta\to \infty\atop \beta\in\R_{>0}}\sup_{\bU\in
 \overline{D_{\bt}(c')}^2\atop
\bU'\in \overline{D_{\bt}(c')}^2\cap\R^2}|F_{\beta}(\bU,\bU')-F(\bU,\bU')|=0.
\end{align*}
By repeating this argument twice more we reach
 \eqref{eq_theorem_proof_convergence_disk}.
The proof of the theorem is complete.
\end{proof}

\appendix

\section{The flux phase problem}\label{app_flux_phase}

A sufficient condition to be a minimizer of the flux phase problem for
the half-filled Hubbard model on a square lattice was essentially 
given by Lieb in
\cite{L}. In order to support readers' verification of Corollary
\ref{cor_application_flux_phase}, here we state Lieb's theorem with some
 supplementary arguments concerning the repeated reflection, which are not explicit in
the short letter \cite{L}. Since the proof below is based on the proved
lemmas \cite[\mbox{Lemma}]{L} and \cite[\mbox{Lemma 2.1}]{LL}, it will
present no more than a review to readers who are already familiar with
this subject. Apart from the condition on the flux per plaquette, we
need a condition on the flux through the large circles around the
periodic lattice in order to define a Hamiltonian minimizing the free
energy under the periodic boundary condition. The sufficiency of these
conditions was discussed by Macris and Nachtergaele in
\cite{MN}. The statement of the theorem below is also implied by
\cite[\mbox{Theorem 1.4, Remarks (a)}]{MN}, which provides a
generalized version of Lieb's theorem with an alternative proof. The
proof below merely uses the original key lemmas  \cite[\mbox{Lemma}]{L} and \cite[\mbox{Lemma 2.1}]{LL}.

First let us consider the problem in a general setting.
Assume that the hopping amplitude
$t(\cdot,\cdot):\Z^2\times\Z^2\to \R_{\ge 0}$ and the magnitude of
on-site interaction $U(\cdot):\Z^2\to \R$ satisfy the following.
\begin{align*}
&t(\bx,\by)=t(\by,\bx)=t(\bx+2mL\be_1+2nL\be_2,\by),\\
&t(\bx,\by)=0\text{ if }\bx-\by\neq \be_1,-\be_1,\be_2,-\be_2\text{ in
 }(\Z/2L\Z)^2,\\
&U(\bx)=U(\bx+2mL\be_1+2nL\be_2),\ (\forall \bx,\by\in\Z^2,m,n\in\Z).
\end{align*}
We minimize the free energy with respect to the argument
of the hopping
matrix elements. The argument is represented by a function
$\phi(\cdot,\cdot):\Z^2\times\Z^2\to\R$ satisfying that
\begin{align}
&\phi(\bx,\by)=-\phi(\by,\bx)\text{ in
 }\R/2\pi\Z,\label{eq_phase_condition_appendix}\\
&\phi(\bx+2mL\be_1+2nL\be_2,\by)=\phi(\bx,\by)\text{ in
 }\R/2\pi\Z,\notag\\
&(\forall\bx,\by\in\Z^2,m,n\in\Z).\notag
\end{align}
With $\phi$ satisfying \eqref{eq_phase_condition_appendix}, define the
Hamiltonian $H(\phi)$ on $F_f(L^2(\G(2L)\times\spin))$ by 
\begin{align*}
H(\phi):=&\sum_{\s\in\spin}\sum_{\bx,\by\in\G(2L)}t(\bx,\by)e^{i\phi(\bx,\by)}\psi_{\bx\s}^{*}\psi_{\by\s}\\
&+\sum_{\bx\in\G(2L)}U(\bx)\Bigg(\psi_{\bx\ua}^*\psi_{\bx\da}^*\psi_{\bx\da}\psi_{\bx\ua}-\frac{1}{2}\sum_{\s\in\spin}\psi_{\bx\s}^*\psi_{\bx\s}\Bigg).
\end{align*}
The flux phase problem is to find a phase $\phi$ satisfying
\eqref{eq_phase_condition_appendix} such that 
\begin{align*}
-\frac{1}{\beta}\log(\Tr e^{-\beta H(\phi)}) =\min_{\eta:\Z^2\times\Z^2\to\R\atop\text{ satisfying
 }\eqref{eq_phase_condition_appendix}}\left(-\frac{1}{\beta}\log(\Tr e^{-\beta H(\eta)})
\right).
\end{align*}
For a phase $\phi$, the flux per plaquette
$f_p(\phi)(\cdot):\Z^2\to \R$ is defined by
\begin{align*}
f_p(\phi)(\bx):=&\phi(\bx+\be_1,\bx)+\phi(\bx+\be_1+\be_2,\bx+\be_1)\\
&+\phi(\bx+\be_2,\bx+\be_1+\be_2)+\phi(\bx,\bx+\be_2),\
 (\forall \bx\in\Z^2).
\end{align*}
The flux through horizontal circle $f_h(\phi)(\cdot):\Z\to \R$ 
and the flux through vertical circle $f_v(\phi)(\cdot):\Z\to \R$ 
 are defined by
\begin{align*}
&f_h(\phi)(x):=\sum_{j=0}^{2L-1}\phi((j+1,x),(j,x)),\\
&f_v(\phi)(x):=\sum_{j=0}^{2L-1}\phi((x,j+1),(x,j)),\ (\forall x\in\Z).
\end{align*}

Before stating the theorem, let us confirm the fact that the free energy
depends on a phase $\phi$ only through the flux $f_p(\phi)(\cdot)$,
$f_h(\phi)(\cdot)$, $f_v(\phi)(\cdot)$. Assume that phases
$\phi_1(\cdot,\cdot)$, $\phi_2(\cdot,\cdot):\Z^2\times\Z^2\to \R$
satisfy \eqref{eq_phase_condition_appendix} and 
\begin{align*}
&f_p(\phi_1)(\bx)=f_p(\phi_2)(\bx)\text{ in }\R/2\pi \Z,\ (\forall
 \bx\in\Z^2),\\
&f_h(\phi_1)(x)=f_h(\phi_2)(x)\text{ in }\R/2\pi \Z,\\
&f_v(\phi_1)(x)=f_v(\phi_2)(x)\text{ in }\R/2\pi \Z,\ (\forall
 x\in\Z).
\end{align*}
Our aim is to prove that $\Tr e^{-\beta H(\phi_1)}=\Tr e^{-\beta
H(\phi_2)}$. To reach the conclusion, let us follow a few lemmas.
 In the following let $\|\cdot\|_{\R^2}$ denote the euclidean norm of
 $\R^2$ and $(\phi_1-\phi_2)(\bx,\by)$ denote $\phi_1(\bx,\by)-\phi_2(\bx,\by)$.

\begin{lemma}\label{lem_circuit_to_boundary}
Assume that $n\ge 2$, $\bx_1,\bx_2,\cdots,\bx_n\in\G(2L)$,
 $\|\bx_j-\bx_{j+1}\|_{\R^2}$ $=1$ $(j=1,2,\cdots,n-2)$, 
$\|\bx_{n-1}-\bx_{n}\|_{\R^2}=2L-1$ and $\bx_{n-1}-\bx_n=\pm \be_1,\pm
 \be_2$ in $(\Z/2L\Z)^2$. Then, there exist
 $\by_1,\by_2,\cdots,\by_m\in\G(2L)$ such that $\by_1=\bx_1$,
 $\by_m=\bx_n$,  $\|\by_j-\by_{j+1}\|_{\R^2}=1$ $(j=1,2,\cdots,m-1)$
 and 
\begin{align*}
\sum_{j=1}^{m-1}(\phi_1-\phi_2)(\by_{j+1},\by_j)=\sum_{j=1}^{n-1}(\phi_1-\phi_2)(\bx_{j+1},\bx_j)\text{
 in }\R/2\pi \Z.
\end{align*}
\end{lemma}
\begin{proof}
We can choose $\by_1,\by_2,\cdots,\by_m\in\G(2L)$ satisfying that $\by_1=\bx_1$,
 $\by_m=\bx_n$,  $\|\by_j-\by_{j+1}\|_{\R^2}=1$
 $(j=1,2,\cdots,m-1)$. Then, by using the equality
 $f_p(\phi_1-\phi_2)(\bx)=0$ $(\bx\in\Z^2)$ repeatedly and either
 $f_h(\phi_1-\phi_2)(x)=0$ $(x\in\Z)$ or $f_v(\phi_1-\phi_2)(x)=0$ $(x\in\Z)$
 only once we can deduce that
\begin{align*}
&(\phi_1-\phi_2)(\by_{m-1},\by_m)+(\phi_1-\phi_2)(\by_{m-2},\by_{m-1})+\cdots
 +(\phi_1-\phi_2)(\by_1,\by_2)\\
&+\sum_{j=1}^{n-1}(\phi_1-\phi_2)(\bx_{j+1},\bx_j)=0\text{
 in }\R/2\pi\Z.
\end{align*}
\end{proof}

\begin{lemma}\label{lem_any_circuit}
Assume that $n\ge 2$, $\bx_1,\bx_2,\cdots,\bx_n\in\G(2L)$
and  $\bx_j-\bx_{j+1}=\pm \be_1,\pm \be_2$ in $(\Z/2L\Z)^2$
 $(j=1,2,\cdots,n)$ where $\bx_{n+1}:=\bx_1$.
Then,
$$
\sum_{j=1}^n(\phi_1-\phi_2)(\bx_{j+1},\bx_j)=0\text{ in }\R/2\pi\Z.
$$
\end{lemma}
\begin{proof}
If $\|\bx_j-\bx_{j+1}\|_{\R^2}=1$ $(j=1,2,\cdots,n)$, we can prove the
 claimed equality only by using that $f_p(\phi_1-\phi_2)(\bx)=0$
 $(\bx\in\Z^2)$. Let us consider the case that there are
 $j_1,j_2,\cdots,j_l\in\{1,2,\cdots,n\}$ such that $j_1<j_2<\cdots<j_l$, 
$\|\bx_{j_i}-\bx_{j_i+1}\|_{\R^2}=2L-1$ $(i=1,2,\cdots,l)$ and 
$\|\bx_{j}-\bx_{j+1}\|_{\R^2}=1$ $(\forall
 j\in\{1,2,\cdots,n\}\backslash\{j_1,j_2,\cdots,j_l\})$. By Lemma
 \ref{lem_circuit_to_boundary} there exist
 $\by_1,\by_2,\cdots,\by_{m_1}\in\G(2L)$ such that $\by_1=\bx_1$,
 $\by_{m_1}=\bx_{j_1+1}$, $\|\by_j-\by_{j+1}\|_{\R^2}=1$
 ($j=1,2,\cdots,m_1-1$) and 
\begin{align*}
\sum_{j=1}^{m_1-1}(\phi_1-\phi_2)(\by_{j+1},\by_j)=\sum_{j=1}^{j_1}(\phi_1-\phi_2)(\bx_{j+1},\bx_j)\text{
 in }\R/2\pi \Z.
\end{align*}
Then, we can apply Lemma \ref{lem_circuit_to_boundary} to the sequence
 $\by_1,\cdots,\by_{m_1},\bx_{j_1+2},\cdots,$ $\bx_{j_2},\bx_{j_2+1}$. By
 repeating this procedure we conclude that there are
 $\bz_1,\bz_2,$ $\cdots,\bz_m\in\G(2L)$ such that 
 $\|\bz_j-\bz_{j+1}\|_{\R^2}=1$
 ($j=1,2,\cdots,m$), $\bz_{m+1}=\bz_1$ and  
\begin{align*}
\sum_{j=1}^{m}(\phi_1-\phi_2)(\bz_{j+1},\bz_j)=\sum_{j=1}^{n}(\phi_1-\phi_2)(\bx_{j+1},\bx_j)\text{
 in }\R/2\pi \Z.
\end{align*}
We have already seen that the left-hand side of the above equality is
 $0$ (mod $2\pi$) in the beginning. 
\end{proof}
The idea of the next lemma is essentially the same as 
\cite[\mbox{Lemma 2.1}]{LL}.
\begin{lemma}\label{lem_gauge_existence}(\cite[\mbox{Lemma 2.1}]{LL})
There exists a function $\theta(\cdot):\G(2L)\to \R$ such that for any
 $\bx,\by\in\G(2L)$ with $\bx-\by=\pm\be_1,\pm\be_2$ in $(\Z/2L\Z)^2$,
$$
\phi_1(\bx,\by)=\phi_2(\bx,\by)+\theta(\bx)-\theta(\by)\text{
 in }\R/2\pi \Z.
$$
\end{lemma}
\begin{proof}
For any $(x,y)\in\G(2L)$ set 
\begin{align*}
\theta((x,y)):=&1_{x\ge 1}\sum_{j=0}^{x-1}(\phi_1-\phi_2)((j+1,0),(j,0))\\
&+1_{y\ge 1}
\sum_{j=0}^{y-1}(\phi_1-\phi_2)((x,j+1),(x,j)).
\end{align*}
Then, it follows from Lemma \ref{lem_any_circuit} that 
\begin{align*}
\theta(\bx)+(\phi_1-\phi_2)(\by,\bx)-\theta(\by)=0\text{ in }\R/2\pi \Z.
\end{align*}
\end{proof}
\begin{lemma}\label{lem_gauge_invariance}
$$
\Tr e^{-\beta H(\phi_1)}= \Tr e^{-\beta H(\phi_2)}.
$$
\end{lemma}
\begin{proof}
Let $\theta(\cdot):\G(2L)\to \R$ be the function introduced in Lemma
 \ref{lem_gauge_existence}. We can extend the domain of $\theta(\cdot)$ to $\Z^2$ so that 
\begin{align*}
&\theta(\bx+2mL\be_1+2nL\be_2)=\theta(\bx),\ (\forall
 \bx\in\Z^2,m,n\in\Z),\\
&e^{i\phi_1(\bx,\by)}=e^{i(\phi_2(\bx,\by)+\theta(\bx)-\theta(\by))},\\
&  (\forall \bx,\by\in\Z^2\text{ with
 }\bx-\by=\pm\be_1,\pm\be_2\text{ in }(\Z/2L\Z)^2).
\end{align*}
Let us define the transform $B$ on $F_f(L^2(\G(2L)\times\spin))$ by
\begin{align*}
&B\O_{2L}:= \O_{2L},\\
&B(\psi_{\bx_1\s_1}^*\psi_{\bx_2\s_2}^*\cdots\psi_{\bx_n\s_n}^*\O_{2L})=
e^{i\sum_{j=1}^n\theta(\bx_j)}\psi_{\bx_1\s_1}^*\psi_{\bx_2\s_2}^*\cdots\psi_{\bx_n\s_n}^*\O_{2L},\\
&(\forall (\bx_j,\s_j)\in \G(2L)\times\spin\
 (j=1,2,\cdots,n) )
\end{align*}
and by linearity. The transform $B$ is unitary and satisfies the
 equality $B\sH(\phi_2)B^*=\sH(\phi_1)$, which yields that $\Tr
 e^{-\beta \sH(\phi_1)}=\Tr
 e^{-\beta \sH(\phi_2)}$.
\end{proof}

The following theorem was essentially proved in \cite{L}.
\begin{theorem}\label{thm_Lieb_theorem}(\cite{L})
Assume that $U\in\R$, $t_h,t_v\in\R_{> 0}$ and 
\begin{align*}
&t(\bx,\by)=\left\{\begin{array}{ll} t_h &\text{ if }\bx-\by =
	     \be_1,-\be_1\text{ in }(\Z/2L\Z)^2,\\
t_v &\text{ if }\bx-\by = \be_2,-\be_2\text{ in }(\Z/2L\Z)^2,\\
0 &\text{ otherwise, }
\end{array}\right.\\
&U(\bx)=U,\ (\forall \bx,\by\in\Z^2).
\end{align*}
Moreover, assume that $\phi:\Z^2\times \Z^2\to
 \R$ satisfies \eqref{eq_phase_condition_appendix} and 
\begin{equation}\label{eq_flux_condition_appendix}
\begin{split}
&f_p(\phi)(\bx)=\pi\text{ in }\R/2\pi\Z,\ (\forall
 \bx\in\Z^2),\\
&f_h(\phi)(x)=f_v(\phi)(x)=\pi(L-1)\text{ in }\R/2\pi \Z,\ (\forall x\in\Z).
 \end{split}
\end{equation}
Then,
\begin{align*}
-\frac{1}{\beta}\log(\Tr e^{-\beta H(\phi)}) =\min_{\eta:\Z^2\times\Z^2\to\R\atop\text{ satisfying
 }\eqref{eq_phase_condition_appendix}}\left(-\frac{1}{\beta}\log(\Tr e^{-\beta H(\eta)})
\right).
\end{align*}
\end{theorem}
\begin{proof}
Lemma \ref{lem_gauge_invariance} implies that
 if there exists a minimizer
 $\phi:\Z^2\times\Z^2\to\R$ satisfying
 \eqref{eq_phase_condition_appendix} and \eqref{eq_flux_condition_appendix}, then
 any $\eta:\Z^2\times\Z^2\to\R$ satisfying
 \eqref{eq_phase_condition_appendix} and
 \eqref{eq_flux_condition_appendix} is also a
 minimizer.  Thus, it is sufficient to prove
 the existence of a minimizer satisfying
 \eqref{eq_flux_condition_appendix}.

Since this is a minimization problem of a continuous function defined on
 $[0,2\pi]^{2(2L)^2}$, a minimizer exists. Assume that
 $\phi(\cdot,\cdot):\Z^2\times \Z^2\to \R$ satisfies
 \eqref{eq_phase_condition_appendix} and gives the minimum. 
Set 
\begin{align*}
&\G(2L)_a:=\{(x_1,x_2)\in\G(2L)\ |\ 1\le x_2\le L\},\\
&\G(2L)_b:=\G(2L)\backslash \G(2L)_a.
\end{align*}
Define $\theta(\cdot):\Z^2\to \R$ by
\begin{align*}
&\theta((x,y)):=\left\{\begin{array}{ll}-\phi((x,y),(x,y)+\be_2)&\text{if
		  }y=0\text{ or }L\text{ in }\Z/2L\Z,\\
0&\text{otherwise,}\end{array}
\right.\\
&(\forall (x,y)\in\Z^2).
\end{align*}
Then, define the transform $\tau$ on $F_f(L^2(\G(2L)\times\spin))$ by 
\begin{align*}
&\tau(\psi_{\bx_1\s_1}^*\psi_{\bx_2\s_2}^*\cdots
 \psi_{\bx_n\s_n}^*\O_{2L})\\
&:= e^{i(\theta(\bx_1)+\theta(\bx_2)+\cdots+\theta(\bx_n))}\psi_{\bx_1\s_1}^*\psi_{\bx_2\s_2}^*\cdots
 \psi_{\bx_n\s_n}^*\O_{2L},\\
&(\forall (\bx_j,\s_j)\in\G(2L)\times \spin\
 (j=1,2,\cdots,n)),
\end{align*}
and by linearity. We can see that $\tau$ is unitary,
 $\tau H(\phi)\tau^*=H(\phi')$ with $\phi':\Z^2\times\Z^2\to \R$
 satisfying \eqref{eq_phase_condition_appendix} and that if $y=0$ or $L$
 in $\Z/2L\Z$,
\begin{align}
\phi'((x,y),(x,y)+\be_2)=0,\ (\forall
x\in\Z).\label{eq_horizontal_reflection_condition}
\end{align}
For any $\bx\in \Z^2$ let $\Ref(\bx)$ denote the point of $\Z^2$
 obtained from $\bx$ by reflection with respect to the horizontal line
 $\{(x,1/2)\in\R^2\ |\ x\in\R\}$ in $\R^2$. For conciseness let
 $V(\bx)$ denote the operator 
$$
U\Bigg(\psi_{\bx\ua}^*\psi_{\bx\da}^*\psi_{\bx\da}\psi_{\bx\ua}-\frac{1}{2}\sum_{\s\in\spin}\psi_{\bx\s}^*\psi_{\bx\s}\Bigg).
$$
We decompose $H(\phi')$ as follows.
$$
H(\phi')=H_a+H_b+H_{int},
$$
where 
\begin{align*}
&H_a:=\sum_{\s\in\spin}\sum_{\bx,\by\in\G(2L)_a}t(\bx,\by)e^{i\phi'(\bx,\by)}\psi_{\bx\s}^*\psi_{\by\s}+\sum_{\bx\in\G(2L)_a}V(\bx),\\
&H_b:=\sum_{\s\in\spin}\sum_{\bx,\by\in\G(2L)_b}t(\bx,\by)e^{i\phi'(\bx,\by)}\psi_{\bx\s}^*\psi_{\by\s}+\sum_{\bx\in\G(2L)_b}V(\bx),\\
&H_{int}:=H(\phi')-H_a-H_b.
\end{align*}
Moreover, set
\begin{align*}
&\Xi(H_a):=\sum_{\s\in\spin}\sum_{\bx,\by\in\G(2L)_b}t(\bx,\by)e^{i(\phi'(\Ref(\bx),\Ref(\by))+\pi)}\psi_{\bx\s}^*\psi_{\by\s}+\sum_{\bx\in\G(2L)_b}V(\bx),\\
&\Xi(H_b):=\sum_{\s\in\spin}\sum_{\bx,\by\in\G(2L)_a}t(\bx,\by)e^{i(\phi'(\Ref(\bx),\Ref(\by))+\pi)}\psi_{\bx\s}^*\psi_{\by\s}+\sum_{\bx\in\G(2L)_a}V(\bx).
\end{align*}
Since the property \eqref{eq_horizontal_reflection_condition} holds, we can apply
 \cite[\mbox{Lemma}]{L} concerning the reflection with respect to the
 horizontal line $\{(x,1/2)\in\R^2\ |\ x\in\R\}$.  Since
 $t(\cdot,\cdot)$ is invariant under this reflection, the 
transformation of $H_a$,
 $H_b$ in \cite[\mbox{Lemma}]{L} yields $\Xi(H_a)$, $\Xi(H_b)$
 respectively. The result is that
\begin{align*}
(\Tr e^{-\beta H(\phi')})^2\le \Tr e^{-\beta
 (H_a+\Xi(H_a)+H_{int})}
\Tr e^{-\beta (H_b+\Xi(H_b)+H_{int})}.
\end{align*}
Since $\Tr e^{-\beta H(\phi)}$ is maximum and $\Tr e^{-\beta
 H(\phi)}=\Tr e^{-\beta H(\phi')}$, we can derive from the above
 inequality that 
$$
\Tr e^{-\beta H(\phi)}=\Tr e^{-\beta (H_b+\Xi(H_b)+H_{int})}. 
$$
There exists a phase $\eta:\Z^2\times\Z^2\to\R$ satisfying
 \eqref{eq_phase_condition_appendix} such that for any
 $\bx,\by\in\G(2L)$,
\begin{align}
\eta(\bx,\by)=\left\{\begin{array}{ll} \phi'(\Ref(\bx),\Ref(\by))+\pi
	       &\text{if }\bx,\by\in\G(2L)_a,\\
\phi'(\bx,\by) &\text{otherwise,}
\end{array}
\right.
\label{eq_horizontal_reflection_phase}
\end{align}
and $H_b+\Xi(H_b)+H_{int}=H(\eta)$. By
 \eqref{eq_horizontal_reflection_condition} and
 \eqref{eq_horizontal_reflection_phase} we observe that for any $x\in \Z$
\begin{align*}
&f_v(\eta)(x)=
 \phi'((x,1),(x,0))+\sum_{j=1}^{L-1}(\phi'(\Ref((x,j+1)),\Ref((x,j)))+\pi)\\
&\qquad\qquad\quad+\sum_{j=L}^{2L-1}\phi'((x,j+1),(x,j))\\
&=\sum_{j=1}^{L-1}(\phi'((x,2L-j),(x,2L+1-j))+\pi)
+\sum_{j=L+1}^{2L-1}\phi'((x,j+1),(x,j))\\
&=\pi(L-1)\text{ in }\R/2\pi \Z.
\end{align*}

In the following we repeat the reflection with vertical lines until we
 obtain a phase minimizing the free energy and satisfy the conditions 
\eqref{eq_phase_condition_appendix}, \eqref{eq_flux_condition_appendix}.
For $s\in
 \{0,1,\cdots,L-1\}$ set 
\begin{align*}
&\G(2L)_r^s:=\{(x_1,x_2)\in\G(2L)\ |\ s+1\le x_1\le s+L\},\\
&\G(2L)_l^s:=\G(2L)\backslash \G(2L)_r^s.
\end{align*}
Define $\theta_s(\cdot):\Z^2\to \R$ by
\begin{align*}
&\theta_s((x,y)):=\left\{\begin{array}{ll}-\eta((x,y),(x,y)+\be_1)&\text{if
		  }x=s\text{ in }\Z/2L\Z,\\
0&\text{otherwise,}\end{array}
\right.\\
&(\forall (x,y)\in\Z^2).
\end{align*}
Then, we define the transform $\tau_s(\eta)$ on $F_f(L^2(\G(2L)\times\spin))$ by 
\begin{align*}
&\tau_s(\eta)(\psi_{\bx_1\s_1}^*\psi_{\bx_2\s_2}^*\cdots
 \psi_{\bx_n\s_n}^*\O_{2L})\\
&:= e^{i(\theta_s(\bx_1)+\theta_s(\bx_2)+\cdots+\theta_s(\bx_n))}\psi_{\bx_1\s_1}^*\psi_{\bx_2\s_2}^*\cdots
 \psi_{\bx_n\s_n}^*\O_{2L},\\
&(\forall (\bx_j,\s_j)\in\G(2L)\times \spin\
 (j=1,2,\cdots,n)),
\end{align*}
and by linearity. Remark that $\tau_s(\eta)$ is unitary,
 $\tau_s(\eta)H(\eta)\tau_s(\eta)^*=H(\eta')$ with $\eta':\Z^2\times\Z^2\to \R$
 satisfying \eqref{eq_phase_condition_appendix},
\begin{align}\label{eq_vertical_flux_condition}
f_v(\eta')(x)=\pi(L-1)\text{ in }\R/2\pi \Z,\ (\forall x\in\Z)
\end{align}
 and that if $x=s$ in
 $\Z/2L\Z$,
\begin{align}
\eta'((x,y),(x,y)+\be_1)=0,\ (\forall
y\in\Z).\label{eq_positivity_convention_0}
\end{align}
For any $\bx\in \Z^2$ let $\Ref_{s}(\bx)$ be the point of $\Z^2$
 obtained from $\bx$ by reflection with respect to the line
 $\{(s+1/2,y)\in\R^2\ |\ y\in\R\}$ in $\R^2$. 

When $s=0$, let us decompose $H(\eta')$ as follows.
$$
H(\eta')=H_r^0+H_l^0+H_{int}^0,
$$
where 
\begin{align*}
&H_r^0:=\sum_{\s\in\spin}\sum_{\bx,\by\in\G(2L)_r^0}t(\bx,\by)e^{i\eta'(\bx,\by)}\psi_{\bx\s}^*\psi_{\by\s}+\sum_{\bx\in\G(2L)_r^0}V(\bx),\\
&H_l^0:=\sum_{\s\in\spin}\sum_{\bx,\by\in\G(2L)_l^0}t(\bx,\by)e^{i\eta'(\bx,\by)}\psi_{\bx\s}^*\psi_{\by\s}+\sum_{\bx\in\G(2L)_l^0}V(\bx),\\
&H_{int}^0:=H(\eta')-H_r^0-H_l^0.
\end{align*}
Moreover, set
\begin{align*}
&\Xi(H_r^0):=\sum_{\s\in\spin}\sum_{\bx,\by\in\G(2L)_l^0}t(\bx,\by)e^{i(\eta'(\Ref_0(\bx),\Ref_0(\by))+\pi)}\psi_{\bx\s}^*\psi_{\by\s}+\sum_{\bx\in\G(2L)_l^0}V(\bx),\\
&\Xi(H_l^0):=\sum_{\s\in\spin}\sum_{\bx,\by\in\G(2L)_r^0}t(\bx,\by)e^{i(\eta'(\Ref_0(\bx),\Ref_0(\by))+\pi)}\psi_{\bx\s}^*\psi_{\by\s}+\sum_{\bx\in\G(2L)_r^0}V(\bx).
\end{align*}
The property \eqref{eq_positivity_convention_0} for $s=0$ enables us to apply
 \cite[\mbox{Lemma}]{L} concerning the reflection with respect to the
 line $\{(1/2,y)\in\R^2\ |\ y\in\R\}$.  By the invariant property of $t(\cdot,\cdot)$ the transformation of $H_l^0$,
 $H_r^0$ in \cite[\mbox{Lemma}]{L} gives $\Xi(H_l^0)$, $\Xi(H_r^0)$
 respectively. As the result,
\begin{align*}
(\Tr e^{-\beta H(\eta')})^2\le \Tr e^{-\beta
 (H_l^0+\Xi(H_l^0)+H_{int}^0)}
\Tr e^{-\beta (H_r^0+\Xi(H_r^0)+H_{int}^0)}.
\end{align*}
Since $\Tr e^{-\beta H(\phi)}$ is maximum and $\Tr e^{-\beta
 H(\phi)}=\Tr e^{-\beta H(\eta')}$, the above
 inequality implies that 
$$
\Tr e^{-\beta H(\phi)}=\Tr e^{-\beta (H_l^0+\Xi(H_l^0)+H_{int}^0)}. 
$$
There exists a phase $\phi_0:\Z^2\times\Z^2\to\R$ satisfying
 \eqref{eq_phase_condition_appendix}, \eqref{eq_vertical_flux_condition} such that
 $H_l^0+\Xi(H_l^0)+H_{int}^0=H(\phi_0)$. Moreover, by
 \eqref{eq_positivity_convention_0} for $s=0$, if $x=0$ in $\Z/2L\Z$, 
\begin{align}
f_p(\phi_0)((x,y))\label{eq_flux_condition_0}
&=\eta'((x,y)+\be_1,(x,y))\\
&\quad+\eta'(\Ref_0((x,y)+\be_1+\be_2),\Ref_0((x,y)+\be_1))+\pi\notag\\
&\quad+\eta'((x,y)+\be_2,(x,y)+\be_1+\be_2)+\eta'((x,y),(x,y)+\be_2)\notag\\
&= \eta'((0,y)+\be_2,(0,y))+\pi + \eta'((0,y),(0,y)+\be_2)\notag\\
&=\pi\text{ in }\R/2\pi\Z,\ (\forall y\in\Z).\notag
\end{align}
If $L=1$, the equalities \eqref{eq_flux_condition_0}, \eqref{eq_phase_condition_appendix} imply
 \eqref{eq_flux_condition_appendix}. Thus, $\phi_0$ is the desired minimizer.

Let us assume that $L\ge 2$, $s\in \{0,1,\cdots,L-2\}$ and $\Tr
 e^{-\beta H(\phi)}=\Tr e^{-\beta H(\phi_s)}$ with
 $\phi_s:\Z^2\times\Z^2\to \R$ satisfying
 \eqref{eq_phase_condition_appendix}, \eqref{eq_vertical_flux_condition} and
 that if $x=j$ in $\Z/2L\Z$ for some $j\in \{0,1,\cdots,s\}$, 
\begin{align}
f_p(\phi_s)((x,y))=\pi\text{ in }\R/2\pi\Z,\ (\forall
 y\in\Z).\label{eq_flux_condition_s}
\end{align}
We can write that  
$$
\tau_{s+1+L}(\phi_s)\tau_{s+1}(\phi_s)H(\phi_s)\tau_{s+1}(\phi_s)^*\tau_{s+1+L}(\phi_s)^*=H(\phi_s')
$$
with a phase $\phi_s':\Z^2\times\Z^2\to\R$ satisfying
 \eqref{eq_phase_condition_appendix},
 \eqref{eq_vertical_flux_condition}, \eqref{eq_flux_condition_s} and that if $x=s+1$ or $s+1+L$ in
 $\Z/2L\Z$, 
\begin{align}
\phi_s'((x,y),(x,y)+\be_1)=0,\ (\forall
 y\in\Z).\label{eq_positivity_convention_s}
\end{align}
We can decompose $H(\phi_s')$ as follows.
$$
H(\phi_s')=H_r^{s+1}+H_l^{s+1}+H_{int}^{s+1},
$$
where 
\begin{align*}
&H_r^{s+1}:=\sum_{\s\in\spin}\sum_{\bx,\by\in\G(2L)_r^{s+1}}t(\bx,\by)e^{i\phi_s'(\bx,\by)}\psi_{\bx\s}^*\psi_{\by\s}+\sum_{\bx\in\G(2L)_r^{s+1}}V(\bx),\\
&H_l^{s+1}:=\sum_{\s\in\spin}\sum_{\bx,\by\in\G(2L)_l^{s+1}}t(\bx,\by)e^{i\phi_s'(\bx,\by)}\psi_{\bx\s}^*\psi_{\by\s}+\sum_{\bx\in\G(2L)_l^{s+1}}V(\bx),\\
&H_{int}^{s+1}:=H(\phi_s')-H_r^{s+1}-H_l^{s+1}.
\end{align*}
Define the operators $\Xi(H_r^{s+1})$, $\Xi(H_l^{s+1})$ by 
\begin{align*}
&\Xi(H_r^{s+1})\\
&:=\sum_{\s\in\spin}\sum_{\bx,\by\in\G(2L)_l^{s+1}}t(\bx,\by)e^{i(\phi_s'(\Ref_{s+1}(\bx),\Ref_{s+1}(\by))+\pi)}\psi_{\bx\s}^*\psi_{\by\s}+\sum_{\bx\in\G(2L)_l^{s+1}}V(\bx),\\
&\Xi(H_l^{s+1})\\
&:=\sum_{\s\in\spin}\sum_{\bx,\by\in\G(2L)_r^{s+1}}t(\bx,\by)e^{i(\phi_s'(\Ref_{s+1}(\bx),\Ref_{s+1}(\by))+\pi)}\psi_{\bx\s}^*\psi_{\by\s}+\sum_{\bx\in\G(2L)_r^{s+1}}V(\bx).
\end{align*}
Again the property \eqref{eq_positivity_convention_s} enables us 
to apply \cite[\mbox{Lemma}]{L} concerning the reflection with respect to the
 line $\{(s+3/2,y)\in\R^2\ |\ y\in\R \}$. Since $t(\cdot,\cdot)$ is
 invariant under the reflection, we have that
\begin{align*}
(\Tr e^{-\beta H(\phi_s')})^2\le \Tr e^{-\beta
 (H_l^{s+1}+\Xi(H_l^{s+1})+H_{int}^{s+1})}
\Tr e^{-\beta (H_r^{s+1}+\Xi(H_r^{s+1})+H_{int}^{s+1})}.
\end{align*}
This inequality implies that
$$
\Tr e^{-\beta H(\phi)}=\Tr e^{-\beta (H_l^{s+1}+\Xi(H_l^{s+1})+H_{int}^{s+1})},
$$
because $\Tr e^{-\beta H(\phi)}$ is maximum and $\Tr e^{-\beta
 H(\phi)}=\Tr e^{-\beta H(\phi_s')}$. There exists a phase $\phi_{s+1}:\Z^2\times\Z^2\to\R$ satisfying
 \eqref{eq_phase_condition_appendix} such that
for any $\bx,\by\in\G(2L)$, 
\begin{align}
&\phi_{s+1}(\bx,\by)=\left\{\begin{array}{ll}
\phi_s'(\Ref_{s+1}(\bx),\Ref_{s+1}(\by))+\pi&\text{if
 }\bx,\by\in\G(2L)_r^{s+1},\\
\phi_s'(\bx,\by)&\text{otherwise, }
\end{array}
\right.\label{eq_connection_s_s+1}
\end{align}
and $H(\phi_{s+1})=H_l^{s+1}+\Xi(H_l^{s+1})+H_{int}^{s+1}$. Using
 \eqref{eq_vertical_flux_condition}, 
 \eqref{eq_flux_condition_s}, \eqref{eq_positivity_convention_s} for
 $\phi_s'$ and considering \eqref{eq_connection_s_s+1}, we can check
 that $\phi_{s+1}$ satisfies \eqref{eq_vertical_flux_condition} and
if $x=j$ in $\Z/2L\Z$
 for some $j\in\{0,1,\cdots,s\}$, 
\begin{align*}
f_p(\phi_{s+1})((x,y))=f_p(\phi_{s}')((x,y))=\pi\text{ in
 }\R/2\pi\Z,\ (\forall y\in\Z),
\end{align*}
if $x=s+1$ in $\Z/2L\Z$,
\begin{align*}
&f_p(\phi_{s+1})((x,y))\\
&=\phi_{s}'((x,y)+\be_1,(x,y))\\
&\quad
 +\phi_{s}'(\Ref_{s+1}((x,y)+\be_1+\be_2),\Ref_{s+1}((x,y)+\be_1))+\pi\\
&\quad+\phi_{s}'((x,y)+\be_2,(x,y)+\be_1+\be_2)+\phi_{s}'((x,y),(x,y)+\be_2)\\
&=\phi_{s}'((s+1,y)+\be_2,(s+1,y))+\pi+\phi_{s}'((s+1,y),(s+1,y)+\be_2)\\
&=\pi\text{ in }\R/2\pi\Z,\ (\forall y\in\Z),
\end{align*}
if $x=j$ in $\Z/2L\Z$ for some $j\in
 \{s+2,s+3,\cdots,2s+2\}$, 
\begin{align*}
f_p(\phi_{s+1})((x,y))
&=\phi_{s}'(\Ref_{s+1}((x,y)+\be_1),\Ref_{s+1}((x,y)))\\
&\quad+\phi_{s}'(\Ref_{s+1}((x,y)+\be_1+\be_2),\Ref_{s+1}((x,y)+\be_1))\\
&\quad+\phi_{s}'(\Ref_{s+1}((x,y)+\be_2),\Ref_{s+1}((x,y)+\be_1+\be_2))\\
&\quad+\phi_{s}'(\Ref_{s+1}((x,y)),\Ref_{s+1}((x,y)+\be_2))\\
&=-f_p(\phi_{s}')(\Ref_{s+1}((x,y))-\be_1)\\
&=\pi\text{ in }\R/2\pi\Z,\ (\forall y\in\Z).
\end{align*}
In summary, if $x=j$ in $\Z/2L\Z$ for some $j\in\{0,1,\cdots,2s+2\}$,
$$
f_p(\phi_{s+1})((x,y))=\pi\text{ in }\R/2\pi\Z,\ (\forall y\in\Z).
$$

By induction with $s$ we have that $\Tr e^{-\beta
 H(\phi)}=\Tr e^{-\beta H(\phi_{L-1})}$ with
 a phase $\phi_{L-1}:\Z^2\times\Z^2\to\R$ satisfying
 \eqref{eq_phase_condition_appendix}, \eqref{eq_vertical_flux_condition} and that if $x=j$ in $\Z/2L\Z$ for some
 $j\in \{0,1,\cdots,2L-2\}$,
$$
f_p(\phi_{L-1})((x,y))=\pi\text{ in }\R/2\pi\Z,\ (\forall y\in\Z).
$$
Moreover, by \eqref{eq_positivity_convention_s} and
 \eqref{eq_connection_s_s+1} for $s=L-2$,
if $x=2L-1$ in $\Z/2L\Z$,
\begin{align*}
f_p(\phi_{L-1})((x,y))
&=\phi_{L-2}'((x,y)+\be_1,(x,y))\\
&\quad+\phi_{L-2}'((x,y)+\be_1+\be_2,(x,y)+\be_1)\\
&\quad+\phi_{L-2}'((x,y)+\be_2,(x,y)+\be_1+\be_2)\\
&\quad+\phi_{L-2}'(\Ref_{L-1}((x,y)),\Ref_{L-1}((x,y)+\be_2))+\pi\\
&=\phi_{L-2}'((0,y)+\be_2,(0,y))+\phi_{L-2}'((0,y),(0,y)+\be_2)+\pi\\
&=\pi\text{ in }\R/2\pi\Z,\ (\forall y\in\Z).
\end{align*}
Thus, the phase $\phi_{L-1}$ has the flux $\pi$ (mod $2\pi$) per plaquette.
Furthermore, by \eqref{eq_positivity_convention_s} and
 \eqref{eq_connection_s_s+1} for $s=L-2$,
\begin{align*}
f_h(\phi_{L-1})(x)&=
 \sum_{j=0}^{L-1}\phi_{L-2}'((j+1,x),(j,x))\\
&\quad+\sum_{j=L}^{2L-2}(\phi_{L-2}'(\Ref_{L-1}((j+1,x)),\Ref_{L-1}((j,x)))+\pi)\\
&\quad+\phi_{L-2}'((0,x),(2L-1,x))\\
&= \sum_{j=0}^{L-2}\phi_{L-2}'((j+1,x),(j,x))\\
&\quad+\sum_{j=L}^{2L-2}(\phi_{L-2}'((2L-2-j,x),(2L-1-j,x))+\pi)\\
&=\pi(L-1)\text{ in }\R/2\pi \Z,\ (\forall x\in\Z).
\end{align*}
Therefore, the phase $\phi_{L-1}$ is a minimizer satisfying 
\eqref{eq_flux_condition_appendix}.
\end{proof}

\section{$L^1$-estimates of kernels of Grassmann polynomials}\label{app_grassmann_L1_theory}

Here we prove several lemmas used in Subsection
\ref{subsec_symmetric_formulation}. These lemmas concern estimations of
Grassmann polynomials with respect to the $L^1$-norm $\|\cdot\|_{L^1}$
on their anti-symmetric kernels. Though we know the unique existence of
anti-symmetric kernels, it is not always trivial to characterize the
kernels explicitly. First of all we confirm that we can estimate 
anti-symmetric kernels without characterizing them.
\begin{lemma}\label{lem_L1_norm_comparison}
Assume that $W_m(\psi)\in\cP_m(\bigwedge \cV)$ is written as 
$$
W_m(\psi)=\frah^m\sum_{\bX\in I^m}\hat{W}_m(\bX)\psi_{\bX},
$$
where the function $\hat{W}_m:I^m\to \C$ is not necessarily
 anti-symmetric. Then,
$$
\|W_m\|_{L^1}\le \|\hat{W}_m\|_{L^1}.
$$
\end{lemma}
\begin{proof}
By the uniqueness of anti-symmetric kernels we have that 
$$
W_m(\bX)=\frac{1}{m!}\sum_{\s\in\S_m}\sgn(\s)\hat{W}_m(\bX_{\s}),\
 (\forall \bX\in I^m).
$$
Thus,
$$
\|W_m\|_{L^1}\le \frac{1}{m!}\sum_{\s\in\S_m}\|\hat{W}_m\|_{L^1}=\|\hat{W}_m\|_{L^1}.
$$
\end{proof}

We summarize necessary bounds on polynomials produced by Grassmann Gaussian
integrals in the next lemma.
\begin{lemma}\label{lem_exponential_general_bound}
Assume that a covariance $A:I_0^2\to \C$ and a covariance $A^{(\eps)}:I_0^2\to \C$
 parameterized by $\eps \in [0,1)$ satisfy that
\begin{align*}
&|\det(A(X_i,Y_j))_{1\le i,j\le n}|\le c_A^n,\\
&|\det(A^{(\eps)}(X_i,Y_j))_{1\le i,j\le n}|\le \eps\cdot c_A^n,\\
&(\forall n\in\N,X_j,Y_j\in I_0\ (j=1,2,\cdots,n),\eps\in [0,1)),
\end{align*}
with a constant $c_A\in\R_{\ge 0}$. With $W(\psi),W^{(1)}(\psi),W^{(2)}(\psi)\in\bigwedge\cV$ set
\begin{align*}
&S(\psi):=\int e^{W(\psi+\psi^1)}d\mu_{A}(\psi^1),\\
&S^{(j)}(\psi):=\int e^{W^{(j)}(\psi+\psi^1)}d\mu_{A}(\psi^1),\ (j=1,2),\\
&S^{(\eps)}(\psi):=\int e^{W(\psi+\psi^1)}d\mu_{A+A^{(\eps)}}(\psi^1),\
 (\eps\in [0,1)).
\end{align*}
Then, the following inequalities hold true.
\begin{enumerate}
\item\label{item_exponential_general_0th}
\begin{align*}
|S_0-e^{W_0}|\le
 e^{|W_0|}\left(e^{\sum_{m=1}^Nc_A^{\frac{m}{2}}\|W_m\|_{L^1}}-1\right).
\end{align*}
\item\label{item_exponential_general_sum}
For any $\alpha \in\R_{\ge 0}$, 
\begin{align*}
\sum_{m=0}^N\alpha^mc_A^{\frac{m}{2}}\|S_m\|_{L^1}\le
 e^{\sum_{m=0}^N(\alpha+1)^mc_A^{\frac{m}{2}}\|W_m\|_{L^1}}.
\end{align*}
\item\label{item_exponential_general_difference}
For any $\alpha \in\R_{\ge 0}$, 
\begin{align*}
&\sum_{m=0}^N\alpha^mc_A^{\frac{m}{2}}\|S_m^{(1)}-S_m^{(2)}\|_{L^1}\notag\\
&\le
 \left(|e^{W_0^{(1)}}-e^{W_0^{(2)}}|+e^{|W_0^{(2)}|}\sum_{m=1}^N(\alpha+1)^mc_A^{\frac{m}{2}}\|W_m^{(1)}-W_m^{(2)}\|_{L^1}\right)\\
&\quad\cdot
 e^{\sup_{j\in\{1,2\}}\sum_{m=1}^N(\alpha+1)^mc_A^{\frac{m}{2}}\|W_m^{(j)}\|_{L^1}}.
\end{align*}
\item\label{item_exponential_general_covariance_difference}
For any $\alpha \in\R_{\ge 0}$, 
\begin{align*}
\sum_{m=0}^N\alpha^mc_A^{\frac{m}{2}}\|S_m-S_m^{(\eps)}\|_{L^1}\le\eps e^{|W_0|} \left(e^{\sum_{m=1}^N(\alpha+2)^mc_A^{\frac{m}{2}}\|W_m\|_{L^1}}-1\right).
\end{align*}
\end{enumerate}
\end{lemma}
\begin{proof}
By anti-symmetry we have that
\begin{align}
S_m(\psi)=&
 e^{W_0}\Bigg(1_{m=0}+\sum_{n=1}^N\frac{1}{n!}\label{eq_exponential_general_expansion}\\
&\qquad\cdot\prod_{l=1}^n\Bigg(\sum_{m_l=1}^N\frah^{m_l}\sum_{k_l=0}^{m_l}\left(\begin{array}{c}m_l\\ k_l\end{array}\right)
\sum_{\bX_l\in I^{m_l-k_l}}\sum_{\bY_l\in
 I^{k_l}}W_{m_l}(\bX_l,\bY_l)\Bigg)\notag\\
&\qquad \cdot \eps_{\pm}1_{\sum_{l=1}^nk_l=m}\int
 \psi_{\bX_1}^1\psi_{\bX_2}^1\cdots
 \psi_{\bX_n}^1d\mu_{A}(\psi^1)\psi_{\bY_1}
\psi_{\bY_2}\cdots \psi_{\bY_n}\Bigg),\notag
\end{align}
where the factor $\eps_{\pm}\in \{1,-1\}$ depends only on
 $(m_l)_{l=1}^n$, $(k_l)_{l=1}^n$.

\eqref{item_exponential_general_0th}: We can derive from
 \eqref{eq_exponential_general_expansion} that
\begin{align}
|S_0-e^{W_0}|&\le
 e^{|W_0|}\sum_{n=1}^N\frac{1}{n!}\Bigg(\sum_{m=1}^Nc_A^{\frac{m}{2}}\|W_m\|_{L^1}\Bigg)^n\label{eq_general_0th_proof}\\
&=
 e^{|W_0|}\left(e^{\sum_{m=1}^Nc_A^{\frac{m}{2}}\|W_m\|_{L^1}}-1\right).\notag
\end{align}

\eqref{item_exponential_general_sum}: It follows from Lemma
 \ref{lem_L1_norm_comparison} and
 \eqref{eq_exponential_general_expansion} 
that 
\begin{align*}
&\alpha^m c_A^{\frac{m}{2}}\|S_m\|_{L^1}\notag\\
&\le
 e^{|W_0|}\\
&\quad\cdot\Bigg(1_{m=0}+\sum_{n=1}^N\frac{1}{n!}\prod_{l=1}^n\Bigg(\sum_{m_l=1}^Nc_A^{\frac{m_l}{2}}\|W_{m_l}\|_{L^1}\sum_{k_l=0}^{m_l}\left(\begin{array}{c}m_l\\ k_l\end{array}\right)\alpha^{k_l}\Bigg)1_{\sum_{l=1}^nk_l=m}\Bigg).
\end{align*}
Thus,
 \begin{align*}
&\sum_{m=0}^N\alpha^m c_A^{\frac{m}{2}}\|S_m\|_{L^1}\\
&\le
  e^{|W_0|}\Bigg(1+\sum_{n=1}^N\frac{1}{n!}\prod_{l=1}^n\Bigg(\sum_{m_l=1}^Nc_A^{\frac{m_l}{2}}\|W_{m_l}\|_{L^1}\sum_{k_l=0}^{m_l}\left(\begin{array}{c}m_l\\ k_l\end{array}\right)\alpha^{k_l}\Bigg)\Bigg)\\
&\le e^{\sum_{m=0}^N(\alpha+1)^mc_A^{\frac{m}{2}}\|W_m\|_{L^1}}.
\end{align*}

\eqref{item_exponential_general_difference}: From
 \eqref{eq_exponential_general_expansion} we deduce that
\begin{align*}
&\alpha^m c_A^{\frac{m}{2}}\|S_m^{(1)}-S_m^{(2)}\|_{L^1}\\
&\le |e^{W_0^{(1)}}-e^{W_0^{(2)}}|\\
&\quad\cdot \Bigg(1_{m=0} +
\sum_{n=1}^N\frac{1}{n!}
\prod_{l=1}^n\Bigg(\sum_{m_l=1}^Nc_A^{\frac{m_l}{2}}\|W_{m_l}^{(1)}\|_{L^1}
\sum_{k_l=0}^{m_l}\left(\begin{array}{c}m_l\\ k_l\end{array}\right)\alpha^{k_l}
\Bigg)1_{\sum_{l=1}^nk_l= m}\Bigg)\notag\\
&\quad +
 e^{|W_0^{(2)}|}\sum_{n=1}^N\frac{1}{n!}\prod_{l=1}^n\Bigg(\sum_{m_l=1}^Nc_A^{\frac{m_l}{2}}\sum_{k_l=0}^{m_l}\left(\begin{array}{c}m_l\\ k_l\end{array}\right)\alpha^{k_l}\frah^{m_l}\sum_{\bX_{l}\in I^{m_l}}\Bigg)1_{\sum_{l=1}^n k_l= m}\notag\\
&\qquad\qquad\cdot \Bigg|
\prod_{l=1}^nW_{m_l}^{(1)}(\bX_l)-\prod_{l=1}^nW_{m_l}^{(2)}(\bX_l)\Bigg|.\notag\end{align*}
Therefore,
\begin{align*}
&\sum_{m=0}^N\alpha^m c_A^{\frac{m}{2}}\|S_m^{(1)}-S_m^{(2)}\|_{L^1}\\
&\le |e^{W_0^{(1)}}-e^{W_0^{(2)}}|
 e^{\sum_{m=1}^N(\alpha+1)^mc_A^{\frac{m}{2}}\|W_m^{(1)}\|_{L^1}}\\
&\quad
 +e^{|W_0^{(2)}|}\sum_{m=1}^N(\alpha+1)^mc_A^{\frac{m}{2}}\|W_m^{(1)}-W_m^{(2)}\|_{L^1}\\
&\qquad\cdot \sum_{n=1}^N\frac{1}{(n-1)!}\Bigg(\sup_{j\in\{1,2\}}\sum_{m=1}^N(\alpha+1)^mc_A^{\frac{m}{2}}\|W_m^{(j)}\|_{L^1}\Bigg)^{n-1}\\
&\le \Bigg(|e^{W_0^{(1)}}-
 e^{W_0^{(2)}}|+e^{|W_0^{(2)}|}\sum_{m=1}^N(\alpha+1)^mc_A^{\frac{m}{2}}\|W_m^{(1)}-W_m^{(2)}\|_{L^1}\Bigg)\\
&\quad\cdot e^{\sup_{j\in\{1,2\}}\sum_{m=1}^N(\alpha+1)^mc_A^{\frac{m}{2}}\|W_m^{(j)}\|_{L^1}}.
\end{align*}

\eqref{item_exponential_general_covariance_difference}: By applying the
 Cauchy-Binet formula in the same way as in
 \eqref{eq_application_cauchy_binet} and substituting the determinant
 bounds on $A$, $A^{(\eps)}$ we observe that for any $n\in\N$,
 $X_j,Y_j\in I_0$ $(j=1,2,\cdots,n)$, 
\begin{align*}
&|\det((A+A^{(\eps)})(X_i,Y_j))_{1\le i,j\le n}-\det(A(X_i,Y_j))_{1\le
 i,j\le n}|\\
&\le \sum_{\phi:\{1,2,\cdots,n\}\to \{1,2,\cdots,2n\}\atop\text{with } \phi(1)<\phi(2)<\cdots<\phi(n)}1_{\phi(n)>n}\\
&\quad\cdot|\det((A_{(n)},I_n)(i,\phi(j)))_{1\le i,j\le n}|\left|
\det\left(\left(\begin{array}{c}I_n\\A_{(n)}^{(\eps)}\end{array}
\right)(\phi(i),j)\right)_{1\le i,j\le n}\right|\\
&\le \sum_{m=0}^{n-1}\sum_{\phi:\{1,2,\cdots,n\}\to
 \{1,2,\cdots,2n\}\atop \text{with }\phi(1)<\phi(2)<\cdots<\phi(n)}1_{\phi(m)\le
 n<\phi(m+1)}c_A^m\eps c_A^{n-m}\le \eps(2^2c_A)^n,
\end{align*}
where 
\begin{align*}
&A_{(n)}(i,j):=A(X_i,Y_j),\
 A_{(n)}^{(\eps)}(i,j):=A^{(\eps)}(X_i,Y_j),\ (\forall i,j\in \{1,2,\cdots,n\}),\\
&\phi(0):=n.
\end{align*}
By using this inequality and Lemma \ref{lem_L1_norm_comparison} 
we can derive from
 \eqref{eq_exponential_general_expansion} that for any $m\in
 \{0,1,\cdots,N\}$,
\begin{align*}
&\alpha^mc_A^{\frac{m}{2}}\|S_m-S_m^{(\eps)}\|_{L^1}\\
&\le \alpha^mc_A^{\frac{m}{2}} e^{|W_0|}\sum_{n=1}^N\frac{1}{n!}\prod_{l=1}^n\\
&\quad\cdot\Bigg(\sum_{m_l=1}^N
\frah^{m_l}\sum_{k_l=0}^{m_l}\left(\begin{array}{c}m_l\\ k_l\end{array}\right)
\sum_{\bX_{l}\in I^{m_l-k_l}}\sum_{\bY_l\in
 I^{k_l}}|W_{m_l}(\bX_l,\bY_l)|\Bigg)1_{\sum_{l=1}^nk_l= m}\\
&\quad \cdot \Bigg|\int \psi_{\bX_1}^1\psi_{\bX_2}^1\cdots
 \psi_{\bX_n}^1d\mu_A(\psi^1)- \int \psi_{\bX_1}^1\psi_{\bX_2}^1\cdots
 \psi_{\bX_n}^1d\mu_{A+A^{(\eps)}}(\psi^1)\Bigg|\\
&\le \eps
 e^{|W_0|}\sum_{n=1}^N\frac{1}{n!}\prod_{l=1}^n\Bigg(
\sum_{m_l=1}^N(2^2c_A)^{\frac{m_l}{2}}\|W_{m_l}\|_{L^1}
\sum_{k_l=0}^{m_l}\left(\begin{array}{c}m_l\\ k_l\end{array}\right)2^{-k_l}\alpha^{k_l}\Bigg) 1_{\sum_{l=1}^nk_l= m},
\end{align*}
which implies that
\begin{align*}
\sum_{m=0}^N\alpha^mc_A^{\frac{m}{2}}\|S_m-S_m^{(\eps)}\|_{L^1}
&\le \eps e^{|W_0|}\sum_{n=1}^N\frac{1}{n!}
\left(\sum_{m=1}^N(\alpha+2)^mc_A^{\frac{m}{2}}\|W_m\|_{L^1}\right)^n\\
&\le \eps
 e^{|W_0|}\left(e^{\sum_{m=1}^N(\alpha+2)^mc_A^{\frac{m}{2}}\|W_m\|_{L^1}}-1\right).
\end{align*}
\end{proof}

We also need upper bounds on logarithm of Grassmann polynomials.
\begin{lemma}\label{lem_logarithm_general_bound}
With $W(\psi),W^{(1)}(\psi),W^{(2)}(\psi)\in\bigwedge \cV$ satisfying
 $|W_0-1|$, $|W_0^{(j)}-1|<1$ $(j=1,2)$ set $Q(\psi):=\log W(\psi)$,
 $Q^{(j)}(\psi):=\log W^{(j)}(\psi)$ $(j=1,2)$. Then, the following
 inequalities hold.
\begin{enumerate}
\item\label{item_logarithm_general_0th}
$$
|Q_0|\le - \log (1-|W_0-1|).
$$
\item\label{item_logarithm_general_sum}
For any $\alpha\in\R_{\ge 0}$ satisfying
\begin{align} 
|W_0|^{-1}\sum_{m=1}^N\alpha^m\|W_m\|_{L^1}<1,\label{eq_one_assumption_logarithm}
\end{align}
\begin{align*}
\sum_{m=1}^N\alpha^m\|Q_m\|_{L^1}\le -\log
 \Bigg(1-|W_0|^{-1}\sum_{m=1}^N\alpha^m\|W_m\|_{L^1}\Bigg).
\end{align*}
\item\label{item_logarithm_general_0th_difference}
$$
|Q_0^{(1)}-Q_0^{(2)}|\le |\log(W_0^{(1)})-\log(W_0^{(2)})|.
$$
\item\label{item_logarithm_general_difference}
For any $\alpha\in\R_{\ge 0}$ satisfying 
\begin{align}
\sup_{j\in\{1,2\}}\sum_{m=1}^N\alpha^m\|W_m^{(j)}\|_{L^1}<\inf_{j\in\{1,2\}}|W_0^{(j)}|,\label{eq_logarithm_inf_sup_condition}
\end{align}
\begin{align*}
&\sum_{m=1}^N\alpha^m\|Q_m^{(1)}-Q_m^{(2)}\|_{L^1}\\
&\le \Bigg(1-\left(\inf_{j\in\{1,2\}}|W_0^{(j)}|\right)^{-1}\sup_{j\in
 \{1,2\}}\sum_{m=1}^N\alpha^m\|W_m^{(j)}\|_{L^1}\Bigg)^{-1}\\
&\quad\cdot |W_0^{(1)}W_0^{(2)}|^{-1}\sum_{m=0}^N\alpha^m\|W_m^{(1)}\|_{L^1}
\sum_{n=0}^N\alpha^n\|W_n^{(1)}-W_n^{(2)}\|_{L^1}.
\end{align*}
\end{enumerate}
\end{lemma}
\begin{proof}

\eqref{item_logarithm_general_0th},\eqref{item_logarithm_general_0th_difference}: Since $Q_0=\log W_0$,
 $Q_0^{(j)}=\log W_0^{(j)}$ $(j=1,2)$, the claimed inequalities are true.

\eqref{item_logarithm_general_sum}: Note that for any $m\in
 \{1,2,\cdots,N\}$,
\begin{align}
Q_m(\psi)=&\sum_{n=1}^m\frac{(-1)^{n-1}}{n}W_0^{-n}\prod_{l=1}^n
\Bigg(\sum_{m_l=1}^N\frah^{m_l}\sum_{\bX_l\in I^{m_l}}W_{m_l}(\bX_l)\Bigg)\label{eq_logarithm_general_expansion}\\
&\quad\cdot\psi_{\bX_1}\psi_{\bX_2}\cdots\psi_{\bX_n}1_{\sum_{l=1}^nm_l=m}.\notag
\end{align}
Thus, it follows from Lemma \ref{lem_L1_norm_comparison} and the
 assumption \eqref{eq_one_assumption_logarithm} that
\begin{align*}
\sum_{m=1}^N\alpha^m\|Q_m\|_{L^1}
&\le \sum_{m=1}^N\sum_{n=1}^m\frac{1}{n}|W_0|^{-n}\prod_{l=1}^n
\Bigg(\sum_{m_l=1}^N\alpha^{m_l}\|W_{m_l}\|_{L^1}\Bigg)1_{\sum_{l=1}^nm_l=m}\\
&\le
 \sum_{n=1}^N\frac{1}{n}|W_0|^{-n}\Bigg(\sum_{m=1}^N\alpha^{m}\|W_{m}\|_{L^1}\Bigg)^n\\
&\le -\log\Bigg(1-|W_0|^{-1}\sum_{m=1}^N\alpha^m\|W_m\|_{L^1}\Bigg).
\end{align*}

\eqref{item_logarithm_general_difference}: By
 \eqref{eq_logarithm_general_expansion} and Lemma
 \ref{lem_L1_norm_comparison} we have that 
\begin{align*}
&\|Q_m^{(1)}-Q_m^{(2)}\|_{L^1}\\
&\le \sum_{n=1}^m\frac{1}{n}|W_0^{(1)-n}-W_0^{(2)-n}|\prod_{l=1}^n
\Bigg(\sum_{m_l=1}^N\|W_{m_l}^{(1)}\|_{L^1}\Bigg)1_{\sum_{l=1}^nm_l=m}\\
&\quad+ \sum_{n=1}^m\frac{1}{n}|W_0^{(2)}|^{-n}
\prod_{l=1}^n
\Bigg(\sum_{m_l=1}^N\frah^{m_l}\sum_{\bX_l\in I^{m_l}}\Bigg)\\
&\qquad\qquad\cdot \Bigg|\prod_{j=1}^nW_{m_j}^{(1)}(\bX_j)-\prod_{j=1}^nW_{m_j}^{(2)}(\bX_j)
\Bigg|1_{\sum_{l=1}^nm_l=m}.
\end{align*}
Therefore, on the assumption \eqref{eq_logarithm_inf_sup_condition}, 
\begin{align*}
&\sum_{m=1}^N\alpha^m\|Q_m^{(1)}-Q_m^{(2)}\|_{L^1}\\
&\le \sum_{n=1}^N\sum_{m=n}^N\frac{1}{n}|W_0^{(1)-n}-W_0^{(2)-n}|\prod_{l=1}^n
\Bigg(\sum_{m_l=1}^N\alpha^{m_l}\|W_{m_l}^{(1)}\|_{L^1}\Bigg)1_{\sum_{l=1}^nm_l=m}\\
&\quad+ \sum_{n=1}^N\sum_{m=n}^N\frac{1}{n}|W_0^{(2)}|^{-n}
\prod_{l=1}^n
\Bigg(\sum_{m_l=1}^N\alpha^{m_l}\frah^{m_l}\sum_{\bX_l\in I^{m_l}}\Bigg)\\
&\qquad\qquad\cdot \Bigg|\prod_{j=1}^nW_{m_j}^{(1)}(\bX_j)-\prod_{j=1}^nW_{m_j}^{(2)}(\bX_j)
\Bigg|1_{\sum_{l=1}^nm_l=m}\\
&\le \sum_{n=1}^N\frac{1}{n}|W_0^{(1)-n}-W_0^{(2)-n}|
\Bigg(\sum_{m=1}^N\alpha^{m}\|W_{m}^{(1)}\|_{L^1}\Bigg)^n\\
&\quad+ \sum_{n=1}^N\frac{1}{n}|W_0^{(2)}|^{-n}
\prod_{l=1}^n
\Bigg(\sum_{m_l=1}^N\alpha^{m_l}\frah^{m_l}\sum_{\bX_l\in I^{m_l}}\Bigg)\\
&\qquad\qquad\cdot \Bigg|\prod_{j=1}^nW_{m_j}^{(1)}(\bX_j)-\prod_{j=1}^nW_{m_j}^{(2)}(\bX_j)
\Bigg|\\
&\le |W_0^{(1)-1}-W_0^{(2)-1}|\sum_{n=1}^N\left(\inf_{j\in\{1,2\}}|W_0^{(j)}|\right)^{-n+1}\Bigg(\sum_{m=1}^N\alpha^{m}\|W_{m}^{(1)}\|_{L^1}\Bigg)^n\\
&\quad
 +\sum_{m=1}^N\alpha^m\|W_{m}^{(1)}-W_{m}^{(2)}\|_{L^1}\sum_{n=1}^N|W_0^{(2)}|^{-n}\Bigg(\sup_{j\in\{1,2\}}\sum_{k=1}^N\alpha^{k}\|W_{k}^{(j)}\|_{L^1}\Bigg)^{n-1}\\
&\le  \Bigg(|W_0^{(1)-1}-W_0^{(2)-1}|\sum_{m=1}^N\alpha^m\|W_m^{(1)}\|_{L^1}\\
&\qquad+\sum_{m=1}^N\alpha^m\|W_m^{(1)}-W_m^{(2)}\|_{L^1}|W_0^{(2)}|^{-1}\Bigg)\\
&\quad\cdot
 \Bigg(1-\left(\inf_{j\in\{1,2\}}|W_0^{(j)}|\right)^{-1}\sup_{j\in
 \{1,2\}}\sum_{m=1}^N\alpha^m\|W_m^{(j)}\|_{L^1}\Bigg)^{-1},
\end{align*}
which leads to the claimed inequality.
\end{proof}

\section{Estimation of Gevrey-class functions}\label{app_gevrey}

Here we establish some estimates on functions and matrix-valued
functions whose local regularity
is that of Gevrey-class. We use these estimates to derive decay bounds on
covariance matrices containing a Gevrey-class cut-off function. We
intend not to expand our analysis more than what we need for our
purposes. More general calculus of Gevrey-class functions are found in,
e.g., \cite{Gev}, \cite{R}.

\begin{lemma}\label{lem_gevrey_composition}
Assume that $O_1,O_2\ (\subset \R)$ are open intervals, $f_j\in$ \\
$C^{\infty}(O_j;\R)$ $(j=1,2)$ and $f_1(O_1)\subset O_2$. Moreover,
 assume that $x_0\in O_1$ and 
\begin{align*}
&\left|\left(\frac{d}{dx}\right)^nf_1(x)\Big|_{x=x_0}\right|\le
 q_1r_1^nn!,\\
&\left|\left(\frac{d}{dx}\right)^nf_2(x)\Big|_{x=f_1(x_0)}\right|\le
 q_2r_2^n(n!)^t,\ (\forall n\in \N),
\end{align*}
with constants $q_j,r_j\in\R_{\ge 0}$ $(j=1,2)$, $t\in\R_{\ge 1}$. Then,
\begin{align*}
\left|\left(\frac{d}{dx}\right)^nf_2(f_1(x))\Big|_{x=x_0}\right|\le
 \frac{q_1q_2r_2}{1+q_1r_2}(r_1(1+q_1r_2))^n(n!)^t,\ (\forall n\in\N).
\end{align*}
\end{lemma}
\begin{remark}
A systematic estimation of the composition of Gevrey-class functions was
 presented in \cite[\mbox{Section I}]{Gev}. Here we provide another
 basic estimation motivated by \cite[\mbox{Proposition 1.4.6}]{R}.
\end{remark}
\begin{proof}[Proof of Lemma \ref{lem_gevrey_composition}]
Fix $n\in \N$. By Taylor's theorem, for any $x\in O_1$, 
\begin{align*}
f_1(x)&=\sum_{l=0}^n\frac{1}{l!}\left(\frac{d}{d
 y}\right)^lf_1(y)\Big|_{y=x_0}(x-x_0)^l\\
&\quad +\frac{1}{n!}\int_{x_0}^x(x-y)^n\left(\frac{d}{d
 y}\right)^{n+1}f_1(y)dy.
\end{align*}
Thus, for any $m\in\N$,
\begin{align*}
&\left(\frac{d}{d
 x}\right)^{n}(f_1(x)-f_1(x_0))^m\Big|_{x=x_0}\\
&=\left(\frac{d}{d
 x}\right)^{n}\Bigg(\sum_{l=1}^n\frac{1}{l!}\left(\frac{d}{d
 y}\right)^{l}f_1(y)\Big|_{y=x_0}(x-x_0)^l\Bigg)^m\Bigg|_{x=x_0}\\
&=n!\prod_{j=1}^m\Bigg(\sum_{l_j=1}^n\frac{1}{l_j!}\left(\frac{d}{d
 y}\right)^{l_j}f_1(y)\Big|_{y=x_0}\Bigg)1_{\sum_{j=1}^ml_j=n}.
\end{align*}
Moreover, by Taylor's theorem, for any $x\in O_1$, 
\begin{align*}
&f_2(f_1(x))\\
&=\sum_{m=0}^n\frac{1}{m!}\left(\frac{d}{d
 y}\right)^{m}f_2(y)\Big|_{y=f_1(x_0)}(f_1(x)-f_1(x_0))^m\\
&\quad +\frac{1}{n!}\int_{f_1(x_0)}^{f_1(x)}(f_1(x)-y)^n
\left(\frac{d}{d
 y}\right)^{n+1}f_2(y)dy.
\end{align*}
Therefore, 
\begin{align*}
&\left(\frac{d}{d x}\right)^{n}f_2(f_1(x))\Big|_{x=x_0}\\
&=\sum_{m=1}^n\frac{1}{m!}\left(\frac{d}{d
 y}\right)^{m}f_2(y)\Big|_{y=f_1(x_0)}\left(\frac{d}{d
 x}\right)^{n}(f_1(x)-f_1(x_0))^m\Big|_{x=x_0}\\
&=\sum_{m=1}^n\frac{1}{m!}\left(\frac{d}{d
 y}\right)^{m}f_2(y)\Big|_{y=f_1(x_0)}\\
&\quad\cdot n!\prod_{j=1}^m\Bigg(\sum_{l_j=1}^n\frac{1}{l_j!}\left(\frac{d}{d
 x}\right)^{l_j}f_1(y)\Big|_{y=x_0}\Bigg)1_{\sum_{j=1}^ml_j=n}.
\end{align*}
By substituting the assumed upper bounds we have
\begin{align*}
\left|\left(\frac{d}{d x}\right)^{n}f_2(f_1(x))\Big|_{x=x_0}
\right|
&\le
 \sum_{m=1}^n\frac{1}{m!}q_2r_2^m(m!)^tn!\prod_{j=1}^m\Bigg(\sum_{l_j=1}^nq_1r_1^{l_j}\Bigg)1_{\sum_{j=1}^ml_j=n}\\
&\le
 q_2r_1^n(n!)^t\sum_{m=1}^n(q_1r_2)^m\prod_{j=1}^m\Bigg(\sum_{l_j=1}^n\Bigg)1_{\sum_{j=1}^ml_j=n}\\
&= q_2 r_1^n(n!)^t\frac{1}{n!}\left(\frac{d}{d
 z}\right)^{n}\Big|_{z=0}\sum_{m=1}^n(q_1r_2)^m\left(\frac{z}{1-z}\right)^m\\
&=q_2r_1^n(n!)^t\frac{1}{n!}\left(\frac{d}{d
 z}\right)^{n}\Big|_{z=0}\frac{\frac{q_1r_2z}{1-z}}{1-\frac{q_1r_2z}{1-z}}\\
&=q_1q_2r_2(1+q_1r_2)^{-1}(r_1(1+q_1r_2))^n(n!)^t.
\end{align*}
\end{proof}

In the rest of this section we find upper bounds on matrix-valued
functions.
\begin{lemma}\label{lem_gevrey_matrix}
Assume that $O(\subset \R)$ is an open interval, $x_0\in O$,
$A\in C^{\infty}(O;\Mat(b,\C))$ and 
\begin{align*}
\left\|\left(\frac{d}{dx}\right)^nA(x)\Big|_{x=x_0}\right\|_{b\times
 b}\le qr^n(n!)^t,\ (\forall n\in\N),
\end{align*}
with constants $q,r\in\R_{\ge 0}$, $t\in\R_{\ge 1}$. Then, the following statements hold true.
\begin{enumerate}
\item\label{item_gevrey_matrix_power}
If $t=1$,
\begin{align*}
\left\|\left(\frac{d}{dx}\right)^nA(x)^m\Big|_{x=x_0}\right\|_{b\times
 b}\le (2q)^m(2r)^nn!,\ (\forall m,n\in\N).
\end{align*}
\item\label{item_gevrey_inverse_matrix}
If $A(x)$ is invertible for any $x\in O$ and $\|A(x_0)^{-1}\|_{b\times
     b}\le s$ with a constant $s\in\R_{\ge 0}$, 
\begin{align*}
&\left\|\left(\frac{d}{dx}\right)^nA(x)^{-1}\Big|_{x=x_0}\right\|_{b\times
 b}\le \frac{s^2q}{(1+(sq)^{\frac{1}{t}})^{t}}\big(r(1+(sq)^{\frac{1}{t}})^t\big)^n(n!)^t
,\\
&(\forall n\in\N).
\end{align*}
\end{enumerate}
\end{lemma}
\begin{proof}
\eqref{item_gevrey_matrix_power}: For any $x\in O$, 
\begin{align*}
A(x)&=\sum_{l=0}^n\frac{1}{l!}\left(\frac{d}{dy}\right)^lA(y)\Big|_{y=y_0}(x-x_0)^l\\
&\quad
 +\frac{1}{n!}\int_{x_0}^x(x-y)^n\left(\frac{d}{dy}\right)^{n+1}A(y)dy.
\end{align*}
Thus, for any $m\in\N$,
\begin{align*}
\left(\frac{d}{dx}\right)^nA(x)^m\Big|_{x=x_0}
&=\left(\frac{d}{dx}\right)^n\Bigg(\sum_{l=0}^n\frac{1}{l!}\left(\frac{d}{dy}\right)^lA(y)\Big|_{y=x_0}(x-x_0)^l\Bigg)^m\Bigg|_{x=x_0}\\
&=n!\prod_{j=1\atop
 order}^m\Bigg(\sum_{l_j=0}^n\frac{1}{l_j!}\left(\frac{d}{dy}\right)^{l_j}A(y)\Big|_{y=x_0}\Bigg)1_{\sum_{j=1}^ml_j=n}.
\end{align*}
Then, by using the assumed upper bounds and Cauchy's integral formula we
 observe that
\begin{align*}
&\left\|\left(\frac{d}{dx}\right)^nA(x)^m\Big|_{x=x_0}\right\|_{b\times
 b}\le
 n!\prod_{j=1}^m\Bigg(\sum_{l_j=0}^nqr^{l_j}\Bigg)1_{\sum_{j=1}^ml_j=n}\\
&=q^mr^n\left(\frac{d}{dz}\right)^n\Big|_{z=0}\left(\frac{1}{1-z}\right)^m=q^mr^n\frac{n!}{2\pi i}\oint_{|z|=\frac{1}{2}}dz\frac{1}{z^{n+1}(1-z)^m}\\
&\le (2q)^m(2r)^nn!.
\end{align*}

\eqref{item_gevrey_inverse_matrix}: First let us prove the equality that for any
 $n\in \N$, $x\in O$,
\begin{align}
\left(\frac{d}{dx}\right)^nA(x)^{-1}\label{eq_matrix_derivative_formula}
&=\sum_{l=1}^n\prod_{j=1}^l\Bigg(\sum_{m_j=1}^n\Bigg)1_{\sum_{j=1}^lm_j=n}c^{(n)}(l,m_1,m_2,\cdots,m_l)\\
&\quad\cdot(-1)^l
\prod_{k=1\atop
 order}^l\Bigg(A(x)^{-1}\left(\frac{d}{dx}\right)^{m_k}A(x)\Bigg)A(x)^{-1},\notag 
\end{align}
where the coefficients $c^{(n)}(l,m_1,m_2,\cdots,m_l)\in\N$ ($\forall l\in
 \{1,2,\cdots,n\}$,\\
$m_j\in\{1,2,\cdots,n\}$ $(j=1,2,\cdots,l)$ with
 $\sum_{j=1}^lm_j=n$) are inductively defined as follows.
\begin{align*}
&c^{(1)}(1,1):=1,\\
&c^{(n)}(l,m_1,m_2,\cdots,m_l)\notag\\
&:=\sum_{i=1}^l\big(1_{l\ge
 2}1_{m_i=1}c^{(n-1)}(l-1,m_1,\cdots,m_{i-1},\overbrace{m_i},m_{i+1},\cdots,m_l)\notag\\
&\qquad+1_{l\le n-1}1_{m_i\neq
 1}c^{(n-1)}(l,m_1,\cdots,m_{i-1},m_i-1,m_{i+1},\cdots,m_l)\big).\notag
\end{align*}
Here $(m_1,\cdots,m_{i-1},\overbrace{m_i},m_{i+1},\cdots,m_l)$ denotes
 $(m_1,\cdots,m_{i-1},m_{i+1},$\\ $\cdots,m_l)$.

Since 
$$
\frac{d}{dx}A(x)^{-1}=-A(x)^{-1}\frac{d}{dx}A(x)\cdot A(x)^{-1},
$$
the equality \eqref{eq_matrix_derivative_formula} holds for $n=1$. 
Assume that \eqref{eq_matrix_derivative_formula}
 holds for $n-1$. Then,
\begin{align*}
&\left(\frac{d}{dx}\right)^nA(x)^{-1}\\
&=\sum_{l=1}^{n-1}\prod_{j=1}^l\Bigg(\sum_{m_j=1}^{n-1}\Bigg)1_{\sum_{j=1}^lm_j=n-1}c^{(n-1)}(l,m_1,m_2,\cdots,m_l)(-1)^l\\
&\quad\cdot\Bigg(
\sum_{i=1}^l\prod_{k=1\atop
 order}^{i-1}\left(A(x)^{-1}\left(\frac{d}{dx}\right)^{m_k}A(x)\right)\\
&\qquad\cdot\Bigg(-A(x)^{-1}\frac{d}{dx}A(x)\cdot
 A(x)^{-1}\left(\frac{d}{dx}\right)^{m_i}A(x)\\
&\qquad\qquad+A(x)^{-1}\left(\frac{d}{dx}\right)^{m_i+1}A(x)\Bigg)\\
&\qquad\cdot \prod_{p=i+1\atop
 order}^l\left(A(x)^{-1}\left(\frac{d}{dx}\right)^{m_p}A(x)\right)A(x)^{-1}\\
&\qquad -\prod_{k=1\atop
 order}^l\left(A(x)^{-1}\left(\frac{d}{dx}\right)^{m_k}A(x)\right)A(x)^{-1}\frac{d}{dx}A(x)\cdot A(x)^{-1}\Bigg)\\
&=\sum_{l=1}^{n-1}\prod_{j=1}^{l+1}\Bigg(\sum_{m_j=1}^n\Bigg)1_{\sum_{j=1}^{l+1}m_j=n}(-1)^{l+1}\\
&\quad\cdot
 \sum_{i=1}^{l+1}1_{m_i=1}c^{(n-1)}(l,m_1,\cdots,m_{i-1},\overbrace{m_i},m_{i+1},\cdots,m_{l+1})\\
&\quad \cdot \prod_{k=1\atop
 order}^{l+1}\left(A(x)^{-1}\left(\frac{d}{dx}\right)^{m_k}A(x)\right)A(x)^{-1}\\
&\quad + \sum_{l=1}^{n-1}\prod_{j=1}^{l}\Bigg(\sum_{m_j=1}^n\Bigg)1_{\sum_{j=1}^{l}m_j=n}(-1)^{l}\\
&\qquad\cdot
 \sum_{i=1}^{l}1_{m_i\neq 1}c^{(n-1)}(l,m_1,\cdots,m_{i-1},m_i-1,m_{i+1},\cdots,m_{l})\\
&\qquad \cdot \prod_{k=1\atop
 order}^{l}\left(A(x)^{-1}\left(\frac{d}{dx}\right)^{m_k}A(x)\right)A(x)^{-1}\\
&=\sum_{l=1}^{n}\prod_{j=1}^{l}\Bigg(\sum_{m_j=1}^n\Bigg)1_{\sum_{j=1}^{l}m_j=n}(-1)^{l}\sum_{i=1}^l\\
&\quad\cdot\big(1_{l\ge
 2}1_{m_i=1}c^{(n-1)}(l-1,m_1,\cdots,m_{i-1},\overbrace{m_i},m_{i+1},\cdots,m_{l})\\
&\qquad\quad +1_{l\le
 n-1}1_{m_i\neq 1}c^{(n-1)}(l,m_1,\cdots,m_{i-1},m_i-1,m_{i+1},\cdots,m_{l})\big)\\
&\quad \cdot \prod_{k=1\atop
 order}^{l}\left(A(x)^{-1}\left(\frac{d}{dx}\right)^{m_k}A(x)\right)A(x)^{-1},
\end{align*}
which is equal to the right-hand side of
 \eqref{eq_matrix_derivative_formula} for $n$. Thus, by induction
 the equality \eqref{eq_matrix_derivative_formula} holds for all $n\in\N$.

It follows from \eqref{eq_matrix_derivative_formula} and the assumed
 upper bounds that 
\begin{align}
&\left\|\left(\frac{d}{dx}\right)^{n}A(x)^{-1}\Big|_{x=x_0}\right\|_{b\times
 b}\label{eq_inverse_matrix_gevrey_pre}\\
&\le
 \sum_{l=1}^{n}\prod_{j=1}^l\Bigg(\sum_{m_j=1}^{n}\Bigg)1_{\sum_{j=1}^lm_j=n}c^{(n)}(l,m_1,m_2,\cdots,m_l)s^{l+1}\prod_{k=1}^l(qr^{m_k}(m_k!)^t)
\notag\\
&=sr^n \sum_{l=1}^n(sq)^l\prod_{j=1}^l\Bigg(\sum_{m_j=1}^n(m_j!)^t
\Bigg)1_{\sum_{j=1}^lm_j=n}c^{(n)}(l,m_1,m_2,\cdots,m_l)\notag\\
&\le s r^n \Bigg(\sum_{l=1}^n(sq)^{\frac{l}{t}}\prod_{j=1}^l\Bigg(\sum_{m_j=1}^nm_j!\Bigg)1_{\sum_{j=1}^lm_j=n}c^{(n)}(l,m_1,m_2,\cdots,m_l)\Bigg)^t.\notag
\end{align}

For any $X\in\R_{\ge 0}$ let us compute the sum
\begin{align}
\sum_{l=1}^nX^l\prod_{j=1}^l\Bigg(\sum_{m_j=1}^nm_j!\Bigg)1_{\sum_{j=1}^lm_j=n}c^{(n)}(l,m_1,m_2,\cdots,m_l).\label{eq_generalized_order_1}
\end{align}
Set 
$$
f(x):=\frac{X}{1+X}\sum_{n=0}^{\infty}x^n.
$$
We see that
\begin{align*}
&(\text{the sum }\eqref{eq_generalized_order_1})\\
&=\frac{1}{X+1}\sum_{l=1}^n\prod_{j=1}^l\Bigg(\sum_{m_j=1}^n\Bigg)1_{\sum_{j=1}^lm_j=n}c^{(n)}(l,m_1,m_2,\cdots,m_l)(-1)^l\\
&\quad\cdot
 \prod_{k=1}^l\Bigg(\frac{1}{1-f(0)}\left(\frac{d}{dx}\right)^{m_k}(1-f(x))\Big|_{x=0}\Bigg)\frac{1}{1-f(0)}.
\end{align*}
Then, by applying the formula \eqref{eq_matrix_derivative_formula} with
 $A(x)=1-f(x)$ we obtain
\begin{align*}
(\text{the sum }\eqref{eq_generalized_order_1})=\frac{1}{1+X}\left(\frac{d}{dx}\right)^{n}\frac{1}{1-f(x)}\Big|_{x=0}=X(1+X)^{n-1}n!.
\end{align*}
Substitution of this equality with $X=(sq)^{\frac{1}{t}}$ into \eqref{eq_inverse_matrix_gevrey_pre}
 gives
\begin{align*}
&\left\|\left(\frac{d}{dx}\right)^{n}A(x)^{-1}\Big|_{x=x_0}\right\|_{b\times
 b}\\
&\le s r^n
 \big((sq)^{\frac{1}{t}}(1+(sq)^{\frac{1}{t}})^{n-1}n!\big)^t=s^2q(1+(sq)^{\frac{1}{t}})^{-t}\big(r(1+(sq)^{\frac{1}{t}})^{t}\big)^n(n!)^t.
\end{align*}
\end{proof}

\section{The time-continuum, infinite-volume limit of the truncated Grassmann integral formulation}\label{app_h_L_limit}

In this section we prove that for any $n\in \N$, 
$$
-\frac{1}{\beta L^d}\left(\frac{d}{dz}\right)^n\log\Bigg(\int e^{-z V(\psi)}d\mu_{C}(\psi)\Bigg)\Bigg|_{z=0}
$$
converges uniformly with respect to the coupling constants as $h\to
\infty$, $L\to \infty$. This
convergence property itself does not imply the convergence of the
full formulation
$$
-\frac{1}{\beta L^d}\log\Bigg(\int e^{-V(\psi)}d\mu_{C}(\psi)\Bigg).
$$
It only guarantees the convergence of any finite
truncation of the Taylor series of the function 
$$
z\mapsto -\frac{1}{\beta L^d}\log\Bigg(\int e^{-z V(\psi)}d\mu_{C}(\psi)\Bigg)
$$
around $z=0$ as $h,L\to\infty$. However, once we know the analyticity of the Grassmann integral
formulation with the coupling constants on an $(h,L)$-independent domain
containing the origin, we can use the result of this section to prove
the uniform convergence of the full formulation as
$h,L\to\infty$. 

We need this type of convergence result only in Subsection \ref{subsec_application_IR},
where the model Hamiltonian is specifically analyzed. However, we set up
the problem in a general setting without
specifying the kinetic term of the Hamiltonian. We assume that 
$$
E\in C^{d+1}(\R^d;\Mat(b,\C))
$$
and \eqref{eq_dispersion_hermite}, \eqref{eq_dispersion_periodic} are
valid.  For any $n\in\N$, $\bU\in\C^b$ set
\begin{align*}
a_n(\beta,L,h)(\bU)
&:=-\frac{1}{\beta
 L^d}\frac{1}{n!}\left(\frac{d}{dz}\right)^n\log\Bigg(\int e^{-z
 V(\psi)}d\mu_{C}(\psi)\Bigg)\Bigg|_{z=0}
\end{align*}
with the covariance $C$ defined by \eqref{eq_covariance_2_point_function}
and $V(\psi)(\in\bigwedge \cV)$ defined by \eqref{eq_interaction_counterpart}.

\begin{lemma}\label{lem_h_L_limit}
For any non-empty compact set $K$ of $\C^b$ and $n\in\N$ the following
 statements hold true.
\begin{enumerate}
\item\label{item_perturbative_h_limit}
$a_n(\beta,L,h)(\cdot)$ converges in $C(K;\C)$ as $h\to \infty$ ($h\in
     2\N/\beta$).
\item\label{item_perturbative_L_limit}
Set $a_n(\beta,L):=\lim_{h\to \infty,h\in
     2\N/\beta}a_n(\beta,L,h)$. $a_n(\beta,L)(\cdot)$ converges in
     $C(K;\C)$ as $L\to \infty$ $(L\in \N)$.
\end{enumerate}
\end{lemma}
\begin{proof} The claims can be proved by
following the same idea as in \cite[\mbox{Appendix B}]{K2}, \cite[\mbox{Appendix
D}]{K3}. However, we present the proof in a self-contained style.  
Here we use the notations introduced in the proof of Lemma
 \ref{lem_characterization_covariance}. Let us define the matrix-valued
 function $A:\R^d\to \Mat(b,\C)$ by
 $A(\bk):=(\alpha_{\rho}(\bk)\delta_{\rho,\eta})_{1\le \rho,\eta\le
 b}$. By using \eqref{eq_diagonalization_kinetic} we can deduce from
 \eqref{eq_time_characterization_covariance_full} that
\begin{align}\label{eq_full_covariance_continuous_formula}
&C(\cdot\bx\s x,\cdot\by\tau y)\\
&=\frac{\delta_{\s,\tau}}{L^d}\sum_{\bk\in\G^*}e^{-i\<\bx-\by,\bk\>}\overline{U(\bk)}e^{(x-y)A(\bk)}\notag\\
&\qquad\cdot (1_{x\ge
 y}(I_b+e^{\beta A(\bk)})^{-1}-1_{x<y}(I_b+e^{-\beta
 A(\bk)})^{-1})U(\bk)^t\notag\\
&=\frac{\delta_{\s,\tau}}{L^d}\sum_{\bk\in\G^*}e^{-i\<\bx-\by,\bk\>}e^{(x-y)\overline{E(\bk)}}\big(1_{x\ge
 y}\big(I_b+e^{\beta \overline{E(\bk)}}\big)^{-1}-1_{x<y}\big(I_b+e^{-\beta
 \overline{E(\bk)}}\big)^{-1}\big),\notag\\
&(\forall (\bx,\s,x),(\by,\tau, y)\in \G\times
 \spin \times [0,\beta)).\notag
\end{align}
Set 
\begin{align*}
&C_{\infty}(\cdot\bx\s x,\cdot\by\tau y)\\
&:=\frac{\delta_{\s,\tau}}{(2\pi)^d}\int_{[0,2\pi)^d}d\bp
 e^{-i\<\bx-\by,\sum_{j=1}^d\bp_j\bv_j\>}e^{(x-y)\overline{E(\sum_{j=1}^d\bp_j\bv_j)}}\\
&\qquad\cdot\big(1_{x\ge
 y}\big(I_b+e^{\beta \overline{E(\sum_{j=1}^d\bp_j\bv_j)}}\big)^{-1}-1_{x<y}\big(I_b+e^{-\beta
 \overline{E(\sum_{j=1}^d\bp_j\bv_j)}}\big)^{-1}\big),\\
&(\forall (\bx,\s,x),(\by,\tau, y)\in \G_{\infty}\times
 \spin \times [0,\beta)).
\end{align*}
It follows from the continuity of the function $\bk\mapsto E(\bk)$ that 
\begin{align}
&\lim_{L\to \infty\atop L\in\N}C(\cdot\bx\s x,\cdot \by\tau
 y)=C_{\infty}(\cdot\bx\s x,\cdot \by\tau
 y),\label{eq_full_covariance_L_limit}\\
&(\forall (\bx,\s,x),(\by,\tau, y)\in \G_{\infty}\times
 \spin \times [0,\beta)).\notag
\end{align}
Using the periodicity with the variable $\bk$, we observe
 that for any $\bx,\by\in\G_{\infty}$ and $n\in\{1,2,\cdots,d+1\}$,
\begin{align*}
&\left(\frac{L}{2\pi}\big(e^{i\frac{2\pi}{L}\<\bx-\by,\bv_j\>}-1\big)\right)^n
C(\cdot\bx\s x,\cdot\by\tau y)\\
&=\frac{\delta_{\s,\tau}}{L^d}\sum_{\bk\in\G^*}e^{-i\<\bx-\by,\bk\>}
\prod_{m=1}^n\Bigg(\frac{L}{2\pi}\int_0^{\frac{2\pi}{L}}dq_m\Bigg)\\
&\quad \cdot\left(\frac{\partial}{\partial p_j}\right)^n
\Big(e^{(x-y)\overline{E(\bk+p_j\bv_j)}}\big(1_{x\ge
 y}\big(I_b+e^{\beta \overline{E(\bk+p_j\bv_j)}}\big)^{-1}\\
&\qquad\qquad\qquad\quad-1_{x<y}\big(I_b+e^{-\beta
 \overline{E(\bk+p_j\bv_j)}}\big)^{-1}\big)\Big)\Big|_{p_j=\sum_{m=1}^nq_m}.
\end{align*}
Set
\begin{align*}
E_{max}:=\sup_{l\in \{1,2,\cdots,d\}\atop
 m\in\{0,1,\cdots,d+1\}}\sup_{\bp\in\R^d}\left\|\left(\frac{\partial}{\partial
 p_l}\right)^mE\Bigg(\sum_{j=1}^dp_j\bv_j\Bigg)\right\|_{b\times b}.
\end{align*}
From the above equality we can derive that
\begin{align}
&\|C(\cdot\bx\s x,\cdot\by\tau y)\|_{b\times b}\le
 \frac{c(\beta,d,E_{max})}{1+\sum_{j=1}^d\left|\frac{L}{2\pi}
\big(e^{i\frac{2\pi}{L}\<\bx-\by,\bv_j\>}-1\big)\right|^{d+1}},
\label{eq_sufficient_decay_L_limit_pre}\\
&(\forall (\bx,\s,x),(\by,\tau, y)\in \G_{\infty}\times
 \spin \times [0,\beta)),\notag
\end{align}
where the constant $c(\beta,d,E_{max})(\in \R_{>0})$ depends only on $\beta,d,E_{max}$. 
Especially we have
\begin{align}
&\|C(\cdot\bx\s x,\cdot\by\tau y)\|_{b\times b}\le
 \frac{c(\beta,d,E_{max})}{1+\left(\frac{2}{\pi}\right)^{d+1}\sum_{j=1}^d|\<\bx-\by,\bv_j\>|^{d+1}},\label{eq_sufficient_decay_L_limit}\\
&(\forall (\bx,\s,x),(\by,\tau, y)\in \G_{\infty}\times
 \spin \times [0,\beta)\notag\\
&\quad \text{ with }|\<\bx-\by,\bv_j\>|\le L/2\ (\forall
 j\in \{1,2,\cdots,d\})).\notag
\end{align}

Note that 
\begin{align*}
a_1(\beta,L,h)(\bU)&=\frac{1}{\beta L^d}\int V(\psi)d\mu_{C}(\psi)\\
&=\sum_{\rho\in\cB}U_{\rho}(C(\rho\b0\ua 0,\rho\b0\ua 0)^2-C(\rho\b0\ua
 0,\rho\b0\ua 0)).
\end{align*}
By \eqref{eq_full_covariance_L_limit} the claims
 \eqref{item_perturbative_h_limit}, \eqref{item_perturbative_L_limit}
 for $n=1$ hold true.

To prove the claims for $n\in\N_{\ge 2}$, we need to introduce a few more notations.
For any $T\in \T_n$, $j\in \{1,2,\cdots,n\}$ set
\begin{align*}
&G_j^1(T):=\{v\in\{1,2,\cdots,n\}\ |\\
&\qquad\qquad\qquad v\text{ is on the shortest path
 connecting 1 to $j$ in $T$}\},\\
&\tilde{G}_{j}^1(T):=G_j^1(T)\backslash\{1\}.
\end{align*}
Note that $1,j\in G_j^1(T)$.
In the following we use the notations introduced in
 Subsection \ref{subsec_general_tree_expansion} plus
 \eqref{eq_truncated_interaction_polynomial}. For any $T\in \T_n$,
 $((\s_l,\theta_l))_{l\in T}\in\prod_{l\in T}(\spin\times\{1,-1\})$, 
$(\rho_j, \bx_j,x_j)\in \cB\times\G_{\infty}\times[0,\beta)$ $(j=1,2,\cdots,n)$, set
\begin{align*}
&F_{T,((\s_l,\theta_l))_{l\in
 T}}(\rho_1\bx_1x_1,\rho_2\bx_2x_2,\cdots,\rho_n\bx_nx_n)\\
&:=\prod_{j\in\{1,2,\cdots,n\}\atop\text{with }L_j^1(T)\neq
 \emptyset}\prod_{\{j,s\}\in
 L_j^1(T)}(-2\widetilde{C}(\rho_j\bx_j\s_{\{j,s\}}x_j\theta_{\{j,s\}}, \rho_s\bx_s(-\s_{\{j,s\}})x_s(-\theta_{\{j,s\}}))),\\
&F_{T,((\s_l,\theta_l))_{l\in
 T}}'((\rho_1,\rho_2,\cdots,\rho_n),(\bx_2,\bx_3,\cdots,\bx_n),
(x_1,x_2,\cdots,x_n))\\
&:=\prod_{j\in\{1,2,\cdots,n\}\atop\text{with }L_j^1(T)\neq
 \emptyset}\prod_{\{j,s\}\in
 L_j^1(T)}(-2\widetilde{C}(\rho_j\b0\s_{\{j,s\}}x_j\theta_{\{j,s\}}, \rho_s\bx_s(-\s_{\{j,s\}})x_s(-\theta_{\{j,s\}}))),
\end{align*}
where
\begin{align*}
&\widetilde{C}(\rho\bx\s x\theta,\eta\by\tau y\xi)\\
&:=\frac{1}{2}(1_{(\theta,\xi)=(1,-1)}C(\rho\bx\s x,\eta\by\tau
 y)-1_{(\theta,\xi)=(-1,1)}C(\eta\by\tau y,\rho\bx\s x)),\\
&(\forall
 (\rho,\bx,\s,x,\theta),(\eta,\by,\tau,y,\xi)\in\cB\times\G_{\infty}\times\spin\times[0,\beta)\times\{1,-1\}).
\end{align*}
Moreover, for any $(\rho_j,\bx_j,x_j)\in \cB\times\G\times[0,\beta)_h$ $(j=1,2,\cdots,n)$ set 
\begin{align*}
&H_{T,((\s_l,\theta_l))_{l\in
 T}}(\rho_1\bx_1x_1,\rho_2\bx_2x_2,\cdots,\rho_n\bx_nx_n)\\
&:=\prod_{j\in\{1,2,\cdots,n\}\atop\text{with }L_j^1(T)\neq
 \emptyset}\prod_{\{j,s\}\in
 L_j^1(T)}\frac{\partial}{\partial \psi^j_{\rho_j\bx_j\s_{\{j,s\}}x_j\theta_{\{j,s\}}}}\frac{\partial}{\partial \psi^s_{\rho_s\bx_s(-\s_{\{j,s\}})x_s(-\theta_{\{j,s\}})}}.
\end{align*}
Recall the definition \eqref{eq_truncated_interaction_polynomial} of the
 polynomial $V_{\rho\bx x}^+(\psi)\in\bigwedge \cV$. By the invariant
 property
\begin{align}
&C((\rho,\bx+\bz,\s,x),(\eta,\by+\bz,\tau,y))=C((\rho,\bx,\s,x),(\eta,\by,\tau,y)),\label{eq_spatial_translation_invariance_infinity}\\
&(\forall \bx,\by,\bz\in\G_{\infty})\notag
\end{align}
we see that
\begin{align}
&ope(T,C) H_{T,((\s_l,\theta_l))_{l\in
 T}}(\rho_1\bx_1x_1,\rho_2\bx_2x_2,\cdots,\rho_n\bx_nx_n)\label{eq_left_derivative_translation_invariance}\\
&\cdot \prod_{j=1}^nV_{\rho_j\bx_jx_j}^+(\psi^j)\Bigg|_{\psi^j=0\atop (\forall
 j\in
 \{1,2,\cdots,n\})}\notag\\
&= ope(T,C)\notag\\
&\quad\cdot  H_{T,((\s_l,\theta_l))_{l\in
 T}}(\rho_1r_L(\bx_1+\by)x_1,\rho_2r_L(\bx_2+\by)x_2,\cdots,\rho_nr_L(\bx_n+\by)x_n)\notag\\
&\quad\cdot \prod_{j=1}^nV_{\rho_jr_L(\bx_j+\by)x_j}^+(\psi^j)\Bigg|_{\psi^j=0\atop (\forall
 j\in
 \{1,2,\cdots,n\})},\ (\forall \by\in\G_{\infty}).\notag
\end{align}
Using
 \eqref{eq_laplacian_anti_symmetrization},
 \eqref{eq_spatial_translation_invariance_infinity},
 \eqref{eq_left_derivative_translation_invariance}, we have that 
\begin{align}
&a_n(\beta,L,h)(\bU)=\frac{(-1)^{n+1}}{n!\beta
 L^d}\sum_{T\in\T_n}Ope(T,C)\prod_{j=1}^nV(\psi^j)\Bigg|_{\psi^j=0\atop
 (\forall j\in\{1,2,\cdots,n\})}\label{eq_transform_truncation_translation}\\
&=\frac{(-1)^{n+1}}{n!\beta L^d}\sum_{T\in \T_n}\prod_{l\in
 T}\Bigg(\sum_{(\s_l,\theta_l)\atop \in\spin\times \{1,-1\}}\Bigg)\prod_{i=1}^n\Bigg(
\frac{1}{h}\sum_{(\rho_i,\bx_i,x_i)\atop \in
 \cB\times\G\times[0,\beta)_h}U_{\rho_i}\Bigg)ope(T,C)\notag\\
&\quad \cdot F_{T,((\s_l,\theta_l))_{l\in
 T}}(\rho_1\bx_1x_1,\rho_2\bx_2x_2,\cdots,\rho_n\bx_nx_n)\notag\\
&\quad \cdot H_{T,((\s_l,\theta_l))_{l\in
 T}}(\rho_1\bx_1x_1,\rho_2\bx_2x_2,\cdots,\rho_n\bx_nx_n)\notag\\
&\quad\cdot
 \prod_{j=1}^nV_{\rho_j\bx_jx_j}^+(\psi^j)\Bigg|_{\psi^j=0\atop (\forall
 j\in \{1,2,\cdots,n\})}\notag\\
&=\frac{(-1)^{n+1}}{n!\beta L^d}\sum_{T\in \T_n}\prod_{l\in
 T}\Bigg(\sum_{(\s_l,\theta_l)\atop \in\spin\times \{1,-1\}}\Bigg)\prod_{i=1}^n\Bigg(
\frac{1}{h}\sum_{(\rho_i,\bx_i,x_i)\atop\in
 \cB\times\G\times[0,\beta)_h}U_{\rho_i}\Bigg)ope(T,C)\notag\\
&\quad \cdot F_{T,((\s_l,\theta_l))_{l\in
 T}}\Bigg(\rho_1\bx_1x_1,\rho_2r_L\Bigg(\sum_{v\in G_2^1(T)}\bx_v\Bigg)x_2,\cdots,\rho_nr_L\Bigg(\sum_{v\in G_n^1(T)}\bx_v\Bigg)x_n\Bigg)\notag\\
&\quad \cdot H_{T,((\s_l,\theta_l))_{l\in
 T}}\Bigg(\rho_1\bx_1x_1,\rho_2r_L\Bigg(\sum_{v\in G_2^1(T)}\bx_v\Bigg)x_2,\cdots,\rho_nr_L\Bigg(\sum_{v\in G_n^1(T)}\bx_v\Bigg)x_n\Bigg)\notag\\
&\quad\cdot
 \prod_{j=1}^nV_{\rho_jr_L(\sum_{v\in G_j^1(T)}\bx_v)x_j}^+(\psi^j)\Bigg|_{\psi^j=0\atop (\forall
 j\in \{1,2,\cdots,n\})}\notag\\
&=\frac{(-1)^{n+1}}{n!\beta L^d}\sum_{T\in \T_n}\prod_{l\in
 T}\Bigg(\sum_{(\s_l,\theta_l)\atop\in\spin\times \{1,-1\}}\Bigg)\prod_{i=1}^n\Bigg(
\frac{1}{h}\sum_{(\rho_i,\bx_i,x_i)\atop\in
 \cB\times\G\times[0,\beta)_h}U_{\rho_i}\Bigg)ope(T,C)\notag\\
&\quad \cdot F_{T,((\s_l,\theta_l))_{l\in
 T}}'((\rho_1,\rho_2,\cdots,\rho_n),(\bx_2,\bx_3,\cdots,\bx_n),(x_1,x_2,\cdots,x_n))\notag\\
&\quad \cdot H_{T,((\s_l,\theta_l))_{l\in
 T}}\Bigg(\rho_1\bx_1x_1,\rho_2r_L\Bigg(\sum_{v\in G_2^1(T)}\bx_v\Bigg)x_2,\cdots,\rho_nr_L\Bigg(\sum_{v\in G_n^1(T)}\bx_v\Bigg)x_n\Bigg)\notag\\
&\quad\cdot
 \prod_{j=1}^nV_{\rho_jr_L(\sum_{v\in G_j^1(T)}\bx_v)x_j}^+(\psi^j)\Bigg|_{\psi^j=0\atop (\forall
 j\in \{1,2,\cdots,n\})}\notag\\
&=\frac{(-1)^{n+1}}{n!\beta}\sum_{T\in \T_n}
\prod_{l\in
 T}\Bigg(\sum_{(\s_l,\theta_l)\atop\in\spin\times \{1,-1\}}\Bigg)
\Bigg(
\frac{1}{h}\sum_{(\rho_1,x_1)\atop\in
 \cB\times[0,\beta)_h}U_{\rho_1}\Bigg)
\prod_{i=2}^n\Bigg(
\frac{1}{h}\sum_{(\rho_i,\bx_i,x_i)\atop\in
 \cB\times\G\times[0,\beta)_h}U_{\rho_i}\Bigg)\notag\\
&\quad \cdot ope(T,C)\notag\\
&\quad\cdot F_{T,((\s_l,\theta_l))_{l\in
 T}}'((\rho_1,\rho_2,\cdots,\rho_n),(\bx_2,\bx_3,\cdots,\bx_n),(x_1,x_2,\cdots,x_n))\notag\\
&\quad \cdot H_{T,((\s_l,\theta_l))_{l\in
 T}}\Bigg(\rho_1\b0 x_1,\rho_2r_L\Bigg(\sum_{v\in
 \tilde{G}_2^1(T)}\bx_v\Bigg)x_2,\cdots,\rho_nr_L\Bigg(\sum_{v\in
 \tilde{G}_n^1(T)}\bx_v\Bigg)x_n\Bigg)\notag\\
&\quad\cdot
V_{\rho_1\b0 x_1}^+(\psi^1)
 \prod_{j=2}^nV_{\rho_jr_L(\sum_{v\in \tilde{G}_j^1(T)}\bx_v)x_j}^+(\psi^j)\Bigg|_{\psi^j=0\atop (\forall
 j\in \{1,2,\cdots,n\})}.\notag
\end{align}
For any $(\bx_1,\bx_2,\cdots,\bx_{n-1})\in\G_{\infty}^{n-1}$ set 
\begin{align*}
&\chi_L(\bx_1,\bx_2,\cdots,\bx_{n-1})\\
&:=1_{L\in 2\N}1_{\<\bx_i,\bv_j\>\in
 \{-\frac{L}{2},-\frac{L}{2}+1,\cdots,\frac{L}{2}-1\},(\forall
 i\in\{1,2,\cdots,n-1\},j\in \{1,2,\cdots,d\})}\\
&\quad + 1_{L\notin 2\N}1_{\<\bx_i,\bv_j\>\in
 \{-\frac{L-1}{2},-\frac{L-1}{2}+1,\cdots,\frac{L-1}{2}\},(\forall
 i\in\{1,2,\cdots,n-1\},j\in \{1,2,\cdots,d\})}.
\end{align*}
For any $(\rho_1,\rho_2,\cdots,\rho_n)\in\cB^n$,
 $(\bx_2,\bx_3,\cdots,\bx_n)\in\G_{\infty}^{n-1}$,
 $(x_1,x_2,\cdots,x_n)\in[0,\beta)_h^n$, set
\begin{align*}
&H_{T,((\s_l,\theta_l))_{l\in
 T}}'((\rho_1,\rho_2,\cdots,\rho_n),(\bx_2,\bx_3,\cdots,\bx_n),
(x_1,x_2,\cdots,x_n))\\
&:=ope(T,C)\\
&\quad\cdot H_{T,((\s_l,\theta_l))_{l\in
 T}}\Bigg(\rho_1\b0 x_1,\rho_2r_L\Bigg(\sum_{v\in
 \tilde{G}_2^1(T)}\bx_v\Bigg)x_2,\cdots,\rho_nr_L\Bigg(\sum_{v\in
 \tilde{G}_n^1(T)}\bx_v\Bigg)x_n\Bigg)\\
&\quad\cdot
V_{\rho_1\b0 x_1}^+(\psi^1)
 \prod_{j=2}^nV_{\rho_jr_L(\sum_{v\in \tilde{G}_j^1(T)}\bx_v)x_j}^+(\psi^j)\Bigg|_{\psi^j=0\atop (\forall
 j\in \{1,2,\cdots,n\})}.
\end{align*}
Moreover, for any $x\in[0,\beta)$ let $\hat{x}$ denote an element of
 $[0,\beta)_h$ satisfying $x\in [\hat{x},\hat{x}+1/h)$. With these
 notations we obtain from \eqref{eq_transform_truncation_translation}
 that
\begin{align*}
&a_n(\beta,L,h)(\bU)\\
&=\frac{(-1)^{n+1}}{n!\beta}\sum_{T\in \T_n}
\prod_{l\in
 T}\Bigg(\sum_{(\s_l,\theta_l)\atop\in\spin\times \{1,-1\}}\Bigg)
\sum_{\rho_1\in\cB}U_{\rho_1}\int_0^{\beta}dx_1\\
&\quad\cdot\prod_{i=2}^n\Bigg(\sum_{(\rho_i,\bx_i)\atop\in
 \cB\times\G_{\infty}}U_{\rho_i}\int_0^{\beta}dx_i\Bigg)\chi_{L}(\bx_2,\bx_3,\cdots,\bx_n) \\
&\quad\cdot F_{T,((\s_l,\theta_l))_{l\in
 T}}'((\rho_1,\rho_2,\cdots,\rho_n),(\bx_2,\bx_3,\cdots,\bx_n),(\hat{x_1},\hat{x_2},\cdots,\hat{x_n}))\notag\\
&\quad \cdot H_{T,((\s_l,\theta_l))_{l\in
 T}}'((\rho_1,\rho_2,\cdots,\rho_n),(\bx_2,\bx_3,\cdots,\bx_n),(\hat{x_1},\hat{x_2},\cdots,\hat{x_n})).
\end{align*}
Using \eqref{eq_sufficient_decay_L_limit_pre}, \eqref{eq_property_irrelevant_function} and the
 properties of the matrix $M_{at}(T,\xi,\bs)$ inside $ope(T,C)$, we can
 derive that 
\begin{align}
&| H_{T,((\s_l,\theta_l))_{l\in
 T}}'((\rho_1,\rho_2,\cdots,\rho_n),(\bx_2,\bx_3,\cdots,\bx_n),(\hat{x_1},\hat{x_2},\cdots,\hat{x_n}))|\label{eq_determinant_part_uniform_bound}\\
&\le \prod_{i=1}^n\Bigg(\sum_{m_i\in\{2,4\}}\Bigg)\Bigg|
ope(T,C)\notag\\
&\quad\cdot H_{T,((\s_l,\theta_l))_{l\in
 T}}\Bigg(\rho_1\b0 \hat{x_1},\rho_2r_L\Bigg(\sum_{v\in
 \tilde{G}_2^1(T)}\bx_v\Bigg)\hat{x_2},\cdots,\rho_nr_L\Bigg(\sum_{v\in
 \tilde{G}_n^1(T)}\bx_v\Bigg)\hat{x_n}\Bigg)\notag\\
&\quad\cdot
\cP_{m_1}V_{\rho_1\b0 \hat{x}_1}^+(\psi^1)
 \prod_{j=2}^n\cP_{m_j}V_{\rho_jr_L(\sum_{v\in \tilde{G}_j^1(T)}\bx_v)\hat{x}_j}^+(\psi^j)\Bigg|_{\psi^j=0\atop (\forall
 j\in \{1,2,\cdots,n\})}\Bigg|\notag\\
&\le \prod_{i=1}^n\Bigg(\sum_{m_i\in\{2,4\}}1_{n_i(T)\le
 m_i}n_i(T)!\left(\begin{array}{c}m_i\\ n_i(T)\end{array}
\right)\Bigg)\notag\\
&\quad\cdot \sup_{\bp_j,\bq_j\in\C^n \text{ with
 }\|\bp_j\|_{\C^n},\|\bq_j\|_{\C^n}\le 1\atop
 (j=1,2,\cdots,\frac{1}{2}\sum_{k=1}^nm_k-n+1)}\sup_{X_j,Y_j\in I_0\atop
 (j=1,2,\cdots,\frac{1}{2}\sum_{k=1}^nm_k-n+1)}\notag\\
&\qquad\cdot |\det(\<\bp_i,\bq_j\>_{\C^n}C(X_i,Y_j))_{1\le
 i,j\le \frac{1}{2}\sum_{k=1}^nm_k-n+1}|\notag\\
&\le \prod_{i=1}^n\Bigg(\sum_{m_i\in\{2,4\}}1_{n_i(T)\le
 m_i}n_i(T)!\left(\begin{array}{c}m_i\\ n_i(T)\end{array}
\right)\Bigg)\notag\\
&\quad \cdot \Bigg(\frac{1}{2}\sum_{k=1}^nm_k-n+1\Bigg)!c(\beta,d,E_{max})^{\frac{1}{2}\sum_{k=1}^nm_k-n+1}.\notag
\end{align}
With these preparations we can prove the claims
 \eqref{item_perturbative_h_limit}, \eqref{item_perturbative_L_limit}
 for $n\in \N_{\ge 2}$. Let us take any non-empty compact set $K$ of $\C^b$.

\eqref{item_perturbative_h_limit}: Since it consists of finite sums and
 products of the covariance $C:(\cB\times \G_{\infty}\times
 \spin\times [0,\beta))^2\to \C$, the domain of the function 
$$
H_{T,((\s_l,\theta_l))_{l\in
 T}}'((\rho_1,\cdots,\rho_n),\cdot,\cdot)
$$
can be naturally extended to $(\G_{\infty})^{n-1}\times [0,\beta)^n$. Note that 
$(x,y)\mapsto C(\rho\bx\s x,\eta\by \tau y)$ is continuous a.e. in
 $[0,\beta)^2$, and thus
\begin{align*}
&(x_1,x_2,\cdots,x_n)\mapsto\\
& H_{T,((\s_l,\theta_l))_{l\in
 T}}'((\rho_1,\cdots,\rho_n),(\bx_2,\cdots,\bx_n),({x_1},{x_2},\cdots,{x_n}))
\end{align*}
is continuous a.e. in $[0,\beta)^n$. These imply that 
\begin{align*}
&\lim_{h\to \infty\atop h\in
 2\N/\beta}F_{T,((\s_l,\theta_l))_{l\in
 T}}'((\rho_1,\cdots,\rho_n),(\bx_2,\cdots,\bx_n),(\hat{x_1},\hat{x_2},\cdots,\hat{x_n}))\\
&=F_{T,((\s_l,\theta_l))_{l\in
 T}}'((\rho_1,\cdots,\rho_n),(\bx_2,\cdots,\bx_n),({x_1},{x_2},\cdots,{x_n})),\\
&\lim_{h\to \infty\atop h\in
 2\N/\beta}H_{T,((\s_l,\theta_l))_{l\in
 T}}'((\rho_1,\cdots,\rho_n),(\bx_2,\cdots,\bx_n),(\hat{x_1},\hat{x_2},\cdots,\hat{x_n}))\\
&=H_{T,((\s_l,\theta_l))_{l\in
 T}}'((\rho_1,\cdots,\rho_n),(\bx_2,\cdots,\bx_n),({x_1},{x_2},\cdots,{x_n}))\\
&\text{ for a.e. }(x_1,x_2,\cdots,x_n)\in [0,\beta)^n.
\end{align*}
By these convergence properties and the uniform bounds
 \eqref{eq_sufficient_decay_L_limit_pre},
 \eqref{eq_determinant_part_uniform_bound} we can apply the dominated
 convergence theorem for $L^1([0,\beta)^n)$ to conclude that 
$$
\lim_{h\to \infty\atop h\in 2\N/\beta}a_n(\beta,L,h)=a_n(\beta,L)\text{
 in }C(K;\C),
$$ 
where 
\begin{align*}
&a_n(\beta,L)(\bU)\\
&:=\frac{(-1)^{n+1}}{n!\beta}\sum_{T\in \T_n}
\prod_{l\in
 T}\Bigg(\sum_{(\s_l,\theta_l)\atop\in\spin\times \{1,-1\}}\Bigg)
\sum_{\rho_1\in\cB}U_{\rho_1}\int_0^{\beta}dx_1\\
&\quad\cdot\prod_{i=2}^n\Bigg(\sum_{(\rho_i,\bx_i)\atop\in
 \cB\times\G_{\infty}}U_{\rho_i}\int_0^{\beta}dx_i\Bigg)\chi_{L}(\bx_2,\bx_3,\cdots,\bx_n) \\
&\quad\cdot F_{T,((\s_l,\theta_l))_{l\in
 T}}'((\rho_1,\rho_2,\cdots,\rho_n),(\bx_2,\bx_3,\cdots,\bx_n),({x_1},{x_2},\cdots,{x_n}))\notag\\
&\quad \cdot H_{T,((\s_l,\theta_l))_{l\in
 T}}'((\rho_1,\rho_2,\cdots,\rho_n),(\bx_2,\bx_3,\cdots,\bx_n),({x_1},{x_2},\cdots,{x_n})).
\end{align*}

\eqref{item_perturbative_L_limit}: Substitution of
 \eqref{eq_sufficient_decay_L_limit} and
 \eqref{eq_determinant_part_uniform_bound} yields that
\begin{align}
&\Bigg|\prod_{i=1}^nU_{\rho_i}\chi_L(\bx_2,\bx_3,\cdots,\bx_{n})\label{eq_L_limit_final_uniform_upper_bound}\\
&\quad\cdot F_{T,((\s_l,\theta_l))_{l\in
 T}}'((\rho_1,\rho_2,\cdots,\rho_n),(\bx_2,\bx_3,\cdots,\bx_n),({x_1},{x_2},\cdots,{x_n}))\notag\\
&\quad \cdot H_{T,((\s_l,\theta_l))_{l\in
 T}}'((\rho_1,\rho_2,\cdots,\rho_n),(\bx_2,\bx_3,\cdots,\bx_n),({x_1},{x_2},\cdots,{x_n}))\Bigg|\notag\\
&\le
 \prod_{i=1}^n|U_{\rho_i}|\prod_{j=2}^n\Bigg(
\frac{c(\beta,d,E_{max})}{1+\left(\frac{2}{\pi}\right)^{d+1}\sum_{k=1}^d|\<\bx_j,\bv_k\>|^{d+1}}\Bigg)\notag\\
&\quad \cdot\prod_{k=1}^n\Bigg(\sum_{m_k\in\{2,4\}}1_{n_k(T)\le
 m_k}n_k(T)!\left(\begin{array}{c}m_k\\ n_k(T)\end{array}\right)
\Bigg)\notag\\
&\quad\cdot
 \Bigg(\frac{1}{2}\sum_{k=1}^nm_k-n+1\Bigg)!c(\beta,d,E_{max})^{\frac{1}{2}\sum_{k=1}^nm_k-n+1},\notag\\
&(\forall (\bx_2,\bx_3,\cdots,\bx_n)\in\G_{\infty}^{n-1},(x_1,x_2,\cdots,x_n)\in[0,\beta)^n).\notag
\end{align}
The right-hand side of the inequality above is integrable over
 $\G_{\infty}^{n-1}\times [0,\beta)^n$. Note that
$$
\lim_{L\to \infty\atop L\in \N}\chi_l(\bx_2,\bx_3,\cdots,\bx_n)=1,\ (\forall
 (\bx_2,\bx_3,\cdots,\bx_n)\in \G_{\infty}^{n-1}).
$$
Moreover, we can see from \eqref{eq_full_covariance_L_limit} that 
\begin{align*}
\lim_{L\to \infty\atop L\in \N}
 &F_{T,((\s_l,\theta_l))_{l\in
 T}}'((\rho_1,\rho_2,\cdots,\rho_n),(\bx_2,\bx_3,\cdots,\bx_n),({x_1},{x_2},\cdots,{x_n}))\notag\\
&\cdot H_{T,((\s_l,\theta_l))_{l\in
 T}}'((\rho_1,\rho_2,\cdots,\rho_n),(\bx_2,\bx_3,\cdots,\bx_n),({x_1},{x_2},\cdots,{x_n}))
\end{align*}
exists for any
 $((\bx_2,\bx_3,\cdots,\bx_n),(x_1,x_2,\cdots,x_n))\in\G_{\infty}^{n-1}\times
 [0,\beta)^n$. By these convergence properties and the inequality \eqref{eq_L_limit_final_uniform_upper_bound} we can apply the dominated convergence theorem
 for $L^1(\G_{\infty}^{n-1}\times[0,\beta)^n)$ to conclude that
 $a_n(\beta,L)$ converges in $C(K;\C)$ as $L\to \infty$ $(L\in \N)$.
\end{proof}

\section{Direct treatment of the free energy density}\label{app_partition_function}

In this part of Appendix we prove some lemmas concerning the free energy
density in direct ways without going through the Grassmann integral
formulation. Though the results of this section are needed only in the
model-dependent analysis in Subsection \ref{subsec_application_IR}, here
we consider the problem in a general configuration.
Let $H_0$, $V$, $H$ be the operators on
$F_f(L^2(\cB\times \G\times\spin))$ defined in Subsection \ref{subsec_Hamiltonian}
with $E\in C(\R^d;\Mat(b,\C))$ satisfying \eqref{eq_dispersion_hermite},
\eqref{eq_dispersion_periodic}. In the following $id$ denotes the
identity map on $F_f(L^2(\cB\times \G\times\spin))$.

\begin{lemma}\label{lem_free_free_energy}
For any $\bk\in\G^*$ let $\alpha_{\rho}(\bk)$ $(\rho\in\cB)$ be the eigen
 values of $E(\bk)$. Then, 
\begin{align*}
-\frac{1}{\beta L^d}\log(\Tr e^{-\beta H_0})=-\frac{2}{\beta
 L^d}\sum_{\rho\in \cB}\sum_{\bk\in\G^*}\log(1+e^{-\beta
 \alpha_{\rho}(\bk)}).
\end{align*}
\end{lemma}
\begin{proof}
With the unitary matrix $U(\bk)\in \Mat(b,\C)$ satisfying
 \eqref{eq_diagonalization_kinetic}, set 
\begin{align*}
\psi_{\rho\bk\s}:=\frac{1}{L^{\frac{d}{2}}}\sum_{\bx\in\G}e^{-i\<\bk,\bx\>}\sum_{\eta\in\cB}U(\bk)^*(\rho,\eta)\psi_{\eta\bx\s},\
 ((\rho,\bk,\s)\in\cB\times \G^*\times\spin).
\end{align*}
We can number each element of $\cB\times \G^*\times\spin$ so that
$\cB\times \G^*\times \spin=\{K_j\}_{j=1}^{2bL^d}$. The
 anti-commutation relation holds as follows. 
\begin{align*}
&\psi_{K_i}\psi_{K_j}+\psi_{K_j}\psi_{K_i}=0,\\
&\psi_{K_i}^*\psi_{K_j}+\psi_{K_j}\psi_{K_i}^*=\delta_{i,j}id,\ (\forall
 i,j\in \{1,2,\cdots,2bL^d\}).
\end{align*}
By \eqref{eq_diagonalization_kinetic},
$$
H_0=\sum_{(\rho,\bk,\s)\in\cB\times \G^*\times\spin}\alpha_{\rho}(\bk)\psi_{\rho\bk\s}^*\psi_{\rho\bk\s},
$$
and thus,
\begin{align*}
&H_0\psi_{\rho_1\bk_1\s_1}^*\psi_{\rho_2\bk_2\s_2}^*\cdots \psi_{\rho_n\bk_n\s_n}^*\O=\sum_{j=1}^n\alpha_{\rho_j}(\bk_j)\psi_{\rho_1\bk_1\s_1}^*\psi_{\rho_2\bk_2\s_2}^*\cdots \psi_{\rho_n\bk_n\s_n}^*\O,\\
&(\forall (\rho_j,\bk_j,\s_j)\in\cB\times \G^*\times\spin\
 (j=1,2,\cdots,n)).
\end{align*}
This implies that 
\begin{align*}
\Tr e^{-\beta H_0}=1+\sum_{S\subset \cB\times \G^*\times \spin\atop{with }S\neq
 \emptyset}e^{-\beta \sum_{(\rho,\bk,\s)\in S}\alpha_{\rho}(\bk)}=\prod_{\rho\in\cB}\prod_{\bk\in\G^*}(1+e^{-\beta
 \alpha_{\rho}(\bk)})^2.
\end{align*}
\end{proof}

We use the following lemma to approximate the normalized free energy
density at $\beta$ by that at $[\beta](\in\N)$.

\begin{lemma}\label{lem_beta_integer_limit_partition}
For any $\beta \in\R_{\ge 1}$,
\begin{align*}
&\left|\frac{1}{\beta L^d}\log\left(\frac{\Tr e^{-\beta H}}{\Tr
 e^{-\beta H_0}}\right)- \frac{1}{[\beta] L^d}\log\left(\frac{\Tr e^{-[\beta] H}}{\Tr
 e^{-[\beta] H_0}}\right)\right|\\
&\le \int_{[\beta]}^{\beta}d\gamma\frac{1}{\gamma^2L^d}\left|\log\left(\frac{\Tr e^{-\gamma H}}{\Tr
 e^{-\gamma H_0}}\right)\right|\\
&\quad +2b\left(2\sup_{\bk\in\R^d}\|E(\bk)\|_{b\times
 b}+\sup_{\rho\in\cB}|U_{\rho}|\right)\log\left(\frac{\beta}{[\beta]}\right).
\end{align*}
\end{lemma}
\begin{proof}
Let $\s(H)$ denote the set of all eigen values of $H$. We can take an
 orthonormal basis $B$ of $F_f(L^2(\cB\times\G\times\spin))$
 consisting of eigen vectors of $H$. Then,
\begin{align}
|\Tr(He^{-\beta H})|\le \max_{\alpha\in \s(H)}|\alpha|\sum_{v\in
 B}\<v,e^{-\beta H}v\>_{F_f}=\|H\|_{\mathfrak{B}(F_f)}\Tr e^{-\beta
 H}.\label{eq_application_schwarz_ineq}
\end{align}
Since the eigen values of $H_0$ are 
$$
\Bigg\{\sum_{(\rho,\bk,\s)\in S}\alpha_{\rho}(\bk)\ \Big|\ S\subset
 \cB\times \G^*\times\spin,\ S\neq \emptyset\Bigg\}\cup \{0\},
$$
\begin{align}
\|H_0\|_{\mathfrak{B}(F_f)}\le \sum_{(\rho,\bk,\s)\in\cB\times
 \G^*\times \spin}|\alpha_{\rho}(\bk)|\le
 2bL^d\sup_{\bk\in\R^d}\|E(\bk)\|_{b\times
 b}.\label{eq_H_0_uniform_bound}
\end{align}
Since 
$$
\|\psi_{\rho\bx\s}^*\psi_{\rho\bx\s}\|_{\mathfrak{B}(F_f)}=1,\ (\forall
 (\rho,\bx,\s)\in\cB\times \G\times \spin),
$$
\begin{align}
\|V\|_{\mathfrak{B}(F_f)}&\le \sum_{(\rho,\bx)\in\cB\times
 \G}|U_{\rho}|\|\psi_{\rho\bx\ua}^*\psi_{\rho\bx\da}^*\psi_{\rho\bx\da}\psi_{\rho\bx\ua}\|_{\mathfrak{B}(F_f)}\label{eq_V_norm_bound}\\
&\quad+\frac{1}{2}
 \sum_{(\rho,\bx,\s)\in\cB\times \G\times
 \spin}|U_{\rho}|\|\psi_{\rho\bx\s}^*\psi_{\rho\bx\s}\|_{\mathfrak{B}(F_f)}\notag\\
&\le 2bL^d\sup_{\rho\in \cB}|U_{\rho}|.\notag
\end{align}
By combining \eqref{eq_H_0_uniform_bound}, \eqref{eq_V_norm_bound} with
 \eqref{eq_application_schwarz_ineq} we obtain
\begin{align}
\left|\frac{\Tr(e^{-\beta H}H)}{\Tr e^{-\beta H}}\right|\le
 2bL^d\left(\sup_{\bk\in \R^d}\|E(\bk)\|_{b\times b}+\sup_{\rho\in\cB}|U_{\rho}|\right).\label{eq_hamiltonian_expectation}
\end{align}
Note that
\begin{align}
&\frac{1}{\beta L^d}\log\left(\frac{\Tr e^{-\beta H}}{\Tr
 e^{-\beta H_0}}\right)- \frac{1}{[\beta] L^d}\log\left(\frac{\Tr e^{-[\beta] H}}{\Tr
 e^{-[\beta] H_0}}\right)\label{eq_free_energy_difference}\\
&=\int_{[\beta]}^{\beta}d\gamma\frac{d}{d\gamma}\left(\frac{1}{\gamma L^d}\log\left(\frac{\Tr e^{-\gamma H}}{\Tr
 e^{-\gamma H_0}}\right)\right)\notag\\
&=-\int_{[\beta]}^{\beta}d\gamma\frac{1}{\gamma^2 L^d}\log\left(\frac{\Tr e^{-\gamma H}}{\Tr
 e^{-\gamma H_0}}\right)-\int_{[\beta]}^{\beta}d\gamma\frac{1}{\gamma L^d}\frac{\Tr (e^{-\gamma H}H)}{\Tr
 e^{-\gamma H}}\notag\\
&\quad +\int_{[\beta]}^{\beta}d\gamma\frac{1}{\gamma L^d}\frac{\Tr (e^{-\gamma H_0}H_0)}{\Tr
 e^{-\gamma H_0}}.\notag
\end{align}
Using \eqref{eq_hamiltonian_expectation},
 we can derive the claimed inequality from \eqref{eq_free_energy_difference}.
\end{proof}

We use the next lemma to relate the output of the infrared integration
to the free energy density by means of the identity theorem. 

\begin{lemma}\label{lem_analyticity_domain_partition_function}
For any $r\in\R_{>0}$ there exists a domain $O(\subset \C)$ such that $(-r,r)\subset O$ and the function $\bU\mapsto \log(\Tr
 e^{-\beta H})$ is analytic in
 $O^b(=O\times O\times\cdots \times O)$.
\end{lemma}
\begin{proof}
Take any $\delta_{\rho}\in[-1,1]$ $(\rho\in\cB)$. 
Define the operator $V_0$ on $F_f(L^2(\cB\times \G\times\spin))$ by
\begin{align*}
V_0:= \sum_{(\rho,\bx)\in\cB\times
 \G}\delta_{\rho}\psi_{\rho\bx\ua}^*\psi_{\rho\bx\da}^*\psi_{\rho\bx\da}\psi_{\rho\bx\ua}-\frac{1}{2}
 \sum_{(\rho,\bx,\s)\in\cB\times \G\times
 \spin}\delta_{\rho}\psi_{\rho\bx\s}^*\psi_{\rho\bx\s}.
\end{align*}
Take any $r\in\R_{>0}$ and assume that $\bU\in (-r,r)^b$. Then, for any $\delta \in[0,1]$,
\begin{align*}
|\Tr e^{-\beta (H+i\delta V_0)}-\Tr e^{-\beta H}|&\le
 \int_0^{\delta}d\eps
\left|\frac{d}{d\eps}\Tr e^{-\beta (H+i\eps
 V_0)}\right|\\
&\le \delta \beta 2^{2bL^d}\|V_0\|_{\mathfrak{B}(F_f)}e^{\beta
 (\|H\|_{\mathfrak{B}(F_f)}+\|V_0\|_{\mathfrak{B}(F_f)})},
\end{align*}
where we used the equality $(d/d\eps)\Tr e^{-\beta (H+i\eps
 V_0)}=-i\beta \Tr(e^{-\beta (H+i\eps
 V_0)}V_0)$. This equality can be justified by, e.g., \cite[\mbox{Lemma
 2.3}]{K1}. Therefore,
\begin{align*}
&\Re \Tr e^{-\beta (H+i\delta V_0)}\\
&\ge \Tr
 e^{-\beta H}- \delta \beta 2^{2bL^d}\|V_0\|_{\mathfrak{B}(F_f)}e^{\beta
 (\|H\|_{\mathfrak{B}(F_f)}+\|V_0\|_{\mathfrak{B}(F_f)})}\\
&\ge 1-\delta \beta 2^{2bL^d}\sup_{\bU\in[-r,r]^b,\atop \delta_{\rho}\in[-1,1](\rho\in\cB)}\left\{\|V_0\|_{\mathfrak{B}(F_f)}e^{\beta
 (\|H\|_{\mathfrak{B}(F_f)}+\|V_0\|_{\mathfrak{B}(F_f)})}\right\}.
\end{align*}
We can conclude from the above inequality that there
 exists $\eps\in\R_{>0}$ such that 
\begin{align*}
\Re \Tr e^{-\beta (H+iV_0)}>0,\ (\forall \bU\in (-r,r)^b,\delta_{\rho}
 \in (-\eps,\eps)\ (\rho\in\cB)).
\end{align*}
This implies that the function $\bU\mapsto \log(\Tr e^{-\beta H})$ is
 analytic in the domain 
$$
\{x+iy\ |\ x\in(-r,r),y\in (-\eps,\eps)\}^b\ (\subset \C^b).
$$
\end{proof}

\section*{Acknowledgments}
The author wishes to thank Fumihiko Nakano for interesting conversations.
This work was supported by JSPS KAKENHI Grant Number 26870110.

\section*{Notation}
\subsection*{Parameters and constants}
\begin{center}
\begin{longtable}{p{3cm}|p{8.5cm}|p{3.5cm}}
Notation & Description & Reference \\
\hline
$L$ & size of the spatial lattice & Subsection \ref{subsec_model}\\
$t_{h,e}$, $t_{h,o}$,\par $t_{v,e}$, $t_{v,o}$ & magnitude of the hopping matrix
 elements & Subsection \ref{subsec_model}\\
$U_{e,e}$, $U_{o,e}$,\par $U_{e,o}$, $U_{o,o}$ & coupling constants & Subsection
 \ref{subsec_model}\\
$\beta$ & inverse temperature & Subsection
 \ref{subsec_model}\\
$f_{\bt}$ & parameter depending only on $t_{h,e}$, $t_{h,o}$,\par
 $t_{v,e}$, $t_{v,o}$ & Subsection
 \ref{subsec_model}\\
$d$ & spatial dimension & Subsection
 \ref{subsec_Hamiltonian}\\
$b$ & number of sites in a primitive unit cell & Subsection
 \ref{subsec_Hamiltonian}\\
$h$ & element of $(2/\beta)\N$, step size of the discretization of $[0,\beta)$ & Subsection
 \ref{subsec_grassmann_integral}\\
$N$ & $4b\beta h L^d$, cardinality of $I$ & Subsection
 \ref{subsec_grassmann_integral}\\
$\fw(l)$ & scale-dependent weight & beginning of\par Section
 \ref{sec_general_estimation}\\
$\fr$ & number belonging to $(0,1]$, 
exponent inside $\|\cdot\|_{l,0}$, $\|\cdot\|_{l,1}$, $|\cdot-\cdot|_l$ & beginning of\par Section
 \ref{sec_general_estimation}\\
$c$ & real positive constant independent of any parameter & beginning of\par Section
 \ref{sec_multi_scale_analysis}\\
$M$ & parameter to control the upper bounds of covariances & Subsection
 \ref{subsec_UV_general}\\
$c(\alpha_1,\cdots,\alpha_n)$ & real positive constant
 depending only on parameters $\alpha_1,\cdots,\alpha_n$ & beginning of\par Section
 \ref{sec_UV}\\
$M_{UV}$ & parameter to control the size of support of UV cut-off
 functions & Subsection
 \ref{subsec_covariance_UV}\\
$N_h$ & largest scale in the UV integration  & Subsection
 \ref{subsec_covariance_UV}\\
$c_w$ & constant ($\in (0,1]$) inside $\fw(0)$ independent of any parameter & Subsection
 \ref{subsec_covariance_UV}\\
$M_{IR}$ & parameter to control the size of support of IR cut-off
 functions & Subsection
 \ref{subsec_cut_off_IR}\\
$N_{\beta}$ & smallest scale in the IR integration & Subsection
 \ref{subsec_cut_off_IR}
\end{longtable}
\end{center}

\subsection*{Sets and spaces}
\begin{center}
\begin{longtable}{p{3cm}|p{8.5cm}|p{3.5cm}}
Notation & Description & Reference \\
\hline
$\G(2L)$ & $\{0,1,\cdots,2L-1\}^2$ & Subsection \ref{subsec_model}\\
$F_f(L^2(\G(2L)\times\spin))$ & Fermionic Fock space & Subsection \ref{subsec_model}\\
$D_{\bt}(c)$ & subset of $\C$ which depends on $t_{h,e}$, $t_{h,o}$,
 $t_{v,e}$, $t_{v,o}$  & Subsection \ref{subsec_model}\\ 
$\G$ & spatial lattice for a generalized system & Subsection
 \ref{subsec_Hamiltonian}\\
$\G^*$ & momentum lattice for a generalized system & Subsection \ref{subsec_Hamiltonian}\\
$\cB$ & $\{1,2,\cdots,b\}$ & Subsection \ref{subsec_Hamiltonian}\\
$F_f(L^2(\cB\times\G\times\spin))$ & Fermionic Fock space & Subsection \ref{subsec_Hamiltonian}\\
$\Mat(n,\C)$ & set of all $n\times n$ matrices & Subsection
 \ref{subsec_Hamiltonian}\\
$[0,\beta)_h$ & $\{0,1/h,\cdots,\beta-1/h\}$ & Subsection
 \ref{subsec_grassmann_integral}\\
$I_0$ & $\cB\times\G\times\spin\times [0,\beta)_h$ & Subsection
 \ref{subsec_grassmann_integral}\\
$I$ & $I_0\times\{1,-1\}$ & Subsection
 \ref{subsec_grassmann_integral}\\
$\cV$ & complex vector space spanned by\par the basis $\{\psi_X\}_{X\in I}$ & Subsection
 \ref{subsec_grassmann_integral}\\
$\cV_p$ & complex vector space spanned by\par the basis $\{\psi_X^p\}_{X\in I}$ & Subsection
 \ref{subsec_grassmann_integral}\\
$\bigwedge\cV$ & Grassmann algebra generated by\par $\{\psi_X\}_{X\in I}$ & Subsection
 \ref{subsec_grassmann_integral}\\
$\S_n$ & set of all permutations over $\{1,2,$ $\cdots,$ $n\}$ & Subsection
 \ref{subsec_grassmann_integral}\\
$\cM$ & $(\pi/\beta)(2\Z+1)$ & Subsection
 \ref{subsec_covariance}\\
$\cM_h$ & $\{\o\in\cM\ |\ |\o|<\pi h\}$ & Subsection
 \ref{subsec_covariance}\\
$\cH$ & Hilbert space $L^2(\cB\times \G^*\times \spin \times \cM_h)$ & Subsection
 \ref{subsec_symmetric_formulation}\\
$\T_n$ & set of all trees over $\{1,2,\cdots,n\}$ & Subsection
 \ref{subsec_general_tree_expansion}\\
$L_q^p(T)$ & subgraph of tree $T$  & Subsection \ref{subsec_general_tree_expansion}\\
$I_{0,\infty}$ & $\cB\times \G\times \spin\times (1/h)\Z$ & beginning of\par
 Section \ref{sec_general_estimation_temperature}\\
$I_{\infty}$ & $I_{0,\infty}\times \{1,-1\}$ & beginning of\par
 Section \ref{sec_general_estimation_temperature}\\
$\G_{\infty}$ & $\{\sum_{j=1}^dm_j\bu_j\ |\ m_j\in\Z\ (j=1,2,\cdots,d)\}$ & beginning of\par Section \ref{sec_general_estimation_temperature}\\
$[-\beta_1/4,\beta_1/4)_h$ & $\{-\beta_1/4,-\beta_1/4+1/h,\cdots,\beta_1/4-1/h\}$ & beginning of\par
 Section \ref{sec_general_estimation_temperature}\\
$[\beta_1/4,$\par$\beta_a-\beta_1/4)_h$ & $\{\beta_1/4,\beta_1/4+1/h,\cdots,\beta_a-\beta_1/4-1/h\}$ & beginning of\par
 Section \ref{sec_general_estimation_temperature}\\
$\hat{I}_0$ & $\cB\times \G \times \spin \times [-\beta_1/4,\beta_1/4)_h$ & beginning of\par
 Section \ref{sec_general_estimation_temperature}\\
$\hat{I}$ & $\hat{I}_0\times \{1,-1\}$ & beginning of\par
 Section \ref{sec_general_estimation_temperature}\\
$I_0^0$ & $\cB\times \G \times \spin \times \{0\}$ & beginning of\par
 Section \ref{sec_general_estimation_temperature}\\
$I^0$ & $I_0^0\times \{1,-1\}$ & beginning of\par
 Section \ref{sec_general_estimation_temperature}\\
$\cS(l)$ & subset of $\bigwedge\cV$ & Subsection
 \ref{subsec_covariance_IR}\\
$\tilde{\cS}(l)$ & subset of $\cS(l)(\beta_1)\times\cS(l)(\beta_2)$ & Subsection \ref{subsec_covariance_IR}
\end{longtable}
\end{center}

\subsection*{Functions and maps}
\begin{center}
\begin{longtable}{p{3cm}|p{8.5cm}|p{3.5cm}}
Notation & Description & Reference \\
\hline
$\sH$ & 1-band Hamiltonian on $F_f(L^2(\G(2L)\times\spin))$  &
 Subsection \ref{subsec_model} \\
$\sH_0$ & kinetic part of $\sH$  &
 Subsection \ref{subsec_model} \\
$\sV$ & interacting part of $\sH$  &
 Subsection \ref{subsec_model} \\
$H$ & $b$-band Hamiltonian on $F_f(L^2(\cB\times\G\times\spin))$  &
 Subsection \ref{subsec_Hamiltonian} \\
$H_0$ & kinetic part of $H$  &
 Subsection \ref{subsec_Hamiltonian} \\
$V$ & interacting part of $H$  &
 Subsection \ref{subsec_Hamiltonian} \\
$E(\cdot)$ & the generalized hopping matrix in the momentum space & Subsection \ref{subsec_Hamiltonian}\\
$\cP_n$  & projection from $\bigwedge\cV$ to $\bigwedge^n\cV$ 
 & Subsection \ref{subsec_grassmann_integral}\\
$\partial /\partial \psi_X$  & Grassmann left derivative  & Subsection \ref{subsec_grassmann_integral}\\
$C(\cdot)$  & full covariance & Subsection \ref{subsec_covariance}\\
$I_n$ & $n\times n$ unit matrix  & Subsection \ref{subsec_covariance}\\
$C_{\le 0}^{\infty}(\cdot)$ & $h$-independent covariance matrix with
 Matsubara UV cut-off &  Subsection \ref{subsec_symmetric_formulation}\\
$d_j(\cdot)$ & function to measure the difference between 2 elements of
 $I$ &  beginning of\par Section \ref{sec_general_estimation}\\
$ope(T,C_o)$ & operator made of Grassmann left-derivatives &  Subsection
 \ref{subsec_general_tree_expansion}\\
$Ope(T,C_o)$ & $ope(T,C_o)\prod_{\{p,q\}\in T}(\D_{p,q}(C_0)+\D_{q,p}(C_0))$ 
&  Subsection
 \ref{subsec_general_tree_expansion}\\
$r_{\beta}(\cdot)$ & map from $(1/h)\Z$ to $[0,\beta)_h$ satisfying
 $x=n_{\beta}(x)\beta+r_{\beta}(x)$, $(\forall x\in (1/h)\Z)$ 
&  beginning of\par Section
 \ref{sec_general_estimation_temperature}\\
$n_{\beta}(\cdot)$ & map from $(1/h)\Z$ to $\Z$ satisfying
 $x=n_{\beta}(x)\beta+r_{\beta}(x)$, $(\forall x\in (1/h)\Z)$ 
&  beginning of\par Section
 \ref{sec_general_estimation_temperature}\\
$R_{\beta}(\cdot)$ & map from $I_{0,\infty}^n$ to $I_0^n$, or from $I_{\infty}^n$ to $I^n$ &  beginning of\par Section
 \ref{sec_general_estimation_temperature}\\
$N_{\beta}(\cdot)$ & map from $I_{0,\infty}^n$ to $\Z$, or from $I_{\infty}^n$ to $\Z$ &  beginning of\par Section
 \ref{sec_general_estimation_temperature}\\
$r_L(\cdot)$ & map from $\G_{\infty}$ to $\G$ &  beginning of\par Section
 \ref{sec_general_estimation_temperature}\\
$\hat{d}_j(\cdot)$ & function to measure the difference between 2 elements of
 $\hat{I}$  &  beginning of\par Section
 \ref{sec_general_estimation_temperature}\\
$\phi(\cdot)$ & Gevrey-class function used to construct cut-off functions  &  Subsection
 \ref{subsec_covariance_UV}\\
$\chi_{h,l}(\cdot)$ & UV cut-off function  &  Subsection
 \ref{subsec_covariance_UV}\\
$C_{l}^+(\cdot)$, $C_{l}^-(\cdot)$  &
 covariance matrices for the UV integration  &  Subsection
 \ref{subsec_covariance_UV}\\
$\cD_j$  & finite difference operator &  Subsection
 \ref{subsec_covariance_UV}\\
$\chi_{l}(\cdot)$ & IR cut-off function  &  Subsection
 \ref{subsec_cut_off_IR}\\
$\chi_{\le l}(\cdot)$ & $\sum_{j=l}^{N_{\beta}}\chi_j(\cdot)$  &  Subsection
 \ref{subsec_cut_off_IR}\\
$\hat{\chi}_{\le l}(\cdot)$ & variant of $\chi_{\le l}(\cdot)$  &  Subsection
 \ref{subsec_cut_off_IR}\\
$C_l(\cdot)$  & covariance matrix for the IR integration &  Subsection
 \ref{subsec_covariance_IR}
\end{longtable}
\end{center}

\subsection*{Inner products, norms and semi-norms}
\begin{center}
\begin{longtable}{p{3cm}|p{8.5cm}|p{3.5cm}}
Notation & Description & Reference \\
\hline
$\<\cdot,\cdot\>$ & standard inner product of $\R^d$ & Subsection
 \ref{subsec_Hamiltonian}\\
$\|\cdot\|_{b\times b}$ & operator norm for $b\times b$-matrices &
 Subsection \ref{subsec_Hamiltonian}\\
$\<\cdot,\cdot\>_{\C^n}$ & standard inner product of $\C^n$ & Subsection
 \ref{subsec_Hamiltonian}\\
$\|\cdot\|_{\C^n}$ & norm of $\C^n$ induced by $\<\cdot,\cdot\>_{\C^n}$
& Subsection
 \ref{subsec_Hamiltonian}\\
$\<\cdot,\cdot\>_{F_f}$ & inner product of  $F_f(L^2(\cB\times\G\times
 \spin))$ & Subsection \ref{subsec_grassmann_gaussian_integral}\\
$\|\cdot\|_{F_f}$ & norm of $F_f(L^2(\cB\times\G\times
 \spin))$ induced by $\<\cdot,\cdot\>_{F_f}$ & Subsection \ref{subsec_grassmann_gaussian_integral}\\
$\|\cdot\|_{\mathfrak{B}(F_f)}$ & operator norm for linear transforms on
 $F_f(L^2(\cB\times\G\times \spin))$ & Subsection \ref{subsec_grassmann_gaussian_integral}\\
$\<\cdot,\cdot\>_{\cH}$ & inner product of the Hilbert space $\cH$ & Subsection
 \ref{subsec_symmetric_formulation}\\
$\|\cdot\|_{\cH}$ & norm of $\cH$ induced by $\<\cdot,\cdot\>_{\cH}$ & Subsection
 \ref{subsec_symmetric_formulation}\\
$\|\cdot\|_{L^1}$ & $L^1$-norm for functions on $I^n$ & Subsection
 \ref{subsec_symmetric_formulation}\\
$\|\cdot\|_{l,0}$ & scale-dependent norm\par for anti-symmetric functions
& beginning
 of\par Section \ref{sec_general_estimation}\\
$\|\cdot\|_{l,1}$ & scale-dependent semi-norm\par for anti-symmetric functions
& beginning
 of\par Section \ref{sec_general_estimation}\\
$|\cdot-\cdot|_{l}$ & scale-dependent measurement of the difference
 between two anti-symmetric functions defined at $\beta_1$ and $\beta_2$
& beginning
 of\par Section \ref{sec_general_estimation_temperature}
\end{longtable}
\end{center}

\subsection*{Other notations}
\begin{center}
\begin{longtable}{p{3cm}|p{8.5cm}|p{3.5cm}}
Notation & Description & Reference \\
\hline
$\be_1,\be_2$ & $\be_1=(1,0)$, $\be_2=(0,1)$ & Subsection \ref{subsec_model} \\
$\O_{2L}$ & vacuum of $F_f(L^2(\G(2L)\times
 \spin))$ & Subsection \ref{subsec_model} \\
$\O$ & vacuum of $F_f(L^2(\cB\times\G\times
 \spin))$ & Subsection \ref{subsec_covariance}\\
$\be(\rho)$\par $(\rho=1,2,3,4)$  & $\be(1)=(0,0)$,  $\be(2)=(1,0)$,\par  $\be(3)=(0,1)$,  $\be(4)=(1,1)$
 &  Subsection
 \ref{subsec_4_band}
\end{longtable}
\end{center}

\end{document}